\documentclass[11pt,a4paper]{article}

\usepackage[cp1251]{inputenc}
\usepackage[T2A]{fontenc}
\usepackage[russian]{babel}

\usepackage[dvips]{graphicx}

\usepackage{amsmath}
\usepackage{amssymb}
\usepackage{amsxtra}
\usepackage{amsthm}
\usepackage{latexsym}
\usepackage{longtable}
\usepackage{array}
\usepackage{multirow}
\usepackage{hhline}
\usepackage{ccaption}
\captiondelim{. }

\usepackage{placeins}
\usepackage[top=2cm,bottom=2cm,left=25mm,right=2cm]{geometry}

\newtheorem{theorem}{Теорема}
\newtheorem{defin}{Определение}
\newtheorem{remark}{Замечание}
\newtheorem{proposition}{Предложение}
\newtheorem{lemma}{Лемма}

\numberwithin{equation}{section}

\newcounter{myta}
\numberwithin{myta}{section}
\newcommand{\myt}{\refstepcounter{myta}\themyta}

\newcommand{\fns}[1]{{\footnotesize #1}}
\newcommand{\fts}[1] {{\small #1}}
\newcommand{\hf}[1]{\hfil #1 \hfil}
\newcommand{\mbf}[1] {\mathbf{#1}}
\newcommand{\pir}[1]{\displaystyle{\frac{#1}{2}\left[-r+\frac{1}{r-\lambda}\D\right]}}
\newcommand{\ru}[1]{\rule{0pt}{#1pt}}
\newcommand{\tfs}[1]{\{#1\}}
\newcommand{\ts}[1]{\textsl{#1}}
\newcommand{\mwide}[1]{\Lambda(#1)}

\def\myrul{\vphantom{\begin{tabular}{c}aa\\bb\end{tabular}}}
\def\myrulm{\vphantom{\begin{tabular}{c}aa\\bb\\cc\end{tabular}}}
\def\rz{\xi}
\def\aaa{a}
\def\bbb{b}
\def\ccc{c}
\def\gmm{\gamma}
\def\mct{\mathcal{R}}
\def\mP{P}
\def\mSell{\Sigma_{HK}(\ell,\ld)}
\def\mSash{\Sigma_{LK}(h,\ld)}
\def\A{\mathfrak{a}}
\def\ba{{\boldsymbol\alpha}}
\def\bbI{\bR^3{(\ell,h,k)}}
\def\bR{\mathbb{R}}
\def\be{\mathbf{e}}
\def\bio{\operatorname{symm}}
\def\bl{{\boldsymbol\lambda}}
\def\bm{\mathbf{M}}
\def\bo{{\boldsymbol\omega}}
\def\br{\mathbf{c}}
\def\bs{\boldsymbol}
\def\bT{\mathbf{T}}
\def\cons{\mathop{\rm const}\nolimits}
\def\D{d}
\def\diag{\mathop{\rm diag}\nolimits}
\def\ds{\displaystyle}
\def\Fun{\mathcal{G}}
\def\gaa{\Delta_1}
\def\gac{\Delta_3}
\def\gan{\Delta_0}
\def\grad{\mathop{\rm grad}\nolimits}
\def\gs{\geqslant}
\def\htan{h_{\rm tan}}

\def\ie{\mathbf{I}\mathbf{e}}
\def\iso{Q_{\ell,h}^3}
\def\isom{Q_{-\ell,h}^3}
\def\jpriv{\mathcal{J}}
\def\ld{\lambda}
\def\Linc{\mathfrak{b}}
\def\ls{\leqslant}
\def\mA{\mathfrak{A}}
\def\mct{\mathcal{C}}
\def\mD{\mathcal{D}}
\def\mi{\mathcal{M}}
\def\mK{\mathcal{K}}
\def\mk{\mathcal{L}}
\def\mm{\mathcal{M}_1}
\def\mn{\mathcal{M}_2}
\def\mno{\mathcal{M}_{2,3}}
\def\mo{\mathcal{M}_3}
\def\mP{{P}}
\def\mPel{{\mP^4_\ell}}
\def\mPr{\widetilde{P}}
\def\mstrut{\vphantom{\bigl(}}
\def\mtA{\mathbb{A}}
\def\mtB{\mathbb{B}}
\def\mtC{\mathbb{C}}
\def\mtD{\mathbb{D}}
\def\mtE{\mathbb{E}}
\def\mtF{\mathbb{F}}
\def\mtG{\mathbb{G}}
\def\mtH{\mathbb{H}}
\def\po{p}
\def\qo{q}
\def\ri{\mathrm{i}\,}
\def\rk{\mathop{\rm rank}\nolimits}
\def\ro{\rho}
\def\sgn{\mathop{\rm sgn}\nolimits}
\def\sgrad{\mathop{\rm sgrad}\nolimits}
\def\smale{\mathcal{S}_{LH}}
\def\smales{\mathcal{S}'_{LH}}
\def\vk{\varkappa}
\def\vpi{\pi}
\def\vq{\varepsilon}
\def\wsa{\Pi_1}
\def\wsb{\Pi_2}
\def\wsc{\Pi_3}
\def\wsi{\Pi}
\def\x{\xi}

\usepackage[unicode]{hyperref}

\begin{document}

\begin{center}

\Large{{\bf Топологический атлас гиростата Ковалевской -- Яхья:\\ аналитические результаты и топологический анализ}}

\vspace{5mm}

\normalsize

{\bf М.П.\,Харламов, П.Е.\,Рябов, И.И.\,Харламова, \\
А.Ю.\,Савушкин, Е.Г.\,Шведов}

\vspace{4mm}

\small

Волгоградская академия государственной службы, Россия, Волгоград

Финансовый университет при Правительстве РФ, Россия, Москва

Волгоградский государственный технический университет, Россия, Волгоград

E-mail: mharlamov@vags.ru, orelryabov@mail.ru
\end{center}

\begin{flushright}
{\it 23 ноября 2014 года}
\end{flushright}

\vspace{3mm}

{\footnotesize Представлен обзор результатов, полученных за последние полвека в задаче о движении тяжелого гиростата при условиях С.В.\,Ковалевской. Работа выполнена при финансовой поддержке грантов РФФИ № 10-01-00043, 10-01-97001, 13-01-97025, 14-01-00119.}

\normalsize

\vspace{3mm} Ключевые слова: гиростат, условия Ковалевской, интеграл Яхья, точные решения, фазовая топология

\vspace{6mm}


\Large

\tableofcontents

\normalsize

\clearpage

\section{Введение}\label{sec1}
Сообщение о том, что случай С.В.\,Ковалевской в динамике твердого тела обобщается на гиростат, Х.М.\,Яхья сделал на семинарах В.Г.\,Демина и В.В.\,Козлова в МГУ в 1985 году. Тогда же была представлена заметка в ``J. de Mecanique Theor. Appl.'', которая по каким-то причинам так и не была опубликована. В связи с этим официальной датой открытия этого случая интегрируемости следует считать 1986 год, когда вышла статья \cite{Yeh1}. Статья на русском языке, представленная в апреле 1986 года и содержащая в том числе и этот результат, вышла значительно позже \cite{YehRus}. На самом деле в \cite{Yeh1} интеграл Ковалевской был обобщен сразу в двух направлениях -- на гиростат и на двойное поле, моделирующее действие суперпозиции поля силы тяжести и постоянного магнитного поля. Ранее аналог интеграла Ковалевской для двойного поля был найден О.И.\,Богоявленским \cite{BogRus2}, однако, обобщение Яхья оказалось принципиальным\,-- введение гиростатического момента нарушило классическую структуру интеграла (сумма квадратов), а также его однородность -- новое слагаемое, пропорциональное гиростатическому моменту имеет третью степень по угловым скоростям подобно интегралу Горячева--Чаплыгина. В 1987 году появились сразу две публикации с обобщением интеграла Ковалевской на гиростат в однородном поле \cite{Komar,Gavr}, но в этом отношении их уже нельзя считать оригинальными.

Поскольку введение двойного поля в общем случае ликвидирует симметрию задачи, уничтожая интеграл площадей, Х.М.\,Яхья отмечает два случая полной интегрируемости -- гиростат типа Ковалевской в поле силы тяжести и гиростат в двойном поле особой структуры, допускающей сингулярную симметрию.
В представленной работе обсуждается первая задача. До настоящего времени она не сведена к квадратурам. Однако удивительным образом оказалось, что все движения на критических многообразиях отображения момента, играющие сегодня ключевую роль в топологическом анализе задачи, были выявлены и проинтегрированы задолго до публикации \cite{Yeh1} и возникновения современного подхода к исследованию интегрируемых гамильтоновых систем. Мы предпринимаем попытку систематического изложения (с современной точки зрения) имеющихся аналитических и качественных результатов.

Настоящая работа была задумана М.П.\,Харламовым как обзорная статья 2011 года к 25-летнему юбилею открытия этого интегрируемого случая. Однако обнаружилось, что в многочисленных работах на эту тему имеются разночтения и явные ошибки, пробелы в доказательствах. Поэтому была продолжена в первую очередь аналитическая работа над этим решением, с целью дать исчерпывающее описание фазовой топологии со всеми необходимыми обоснованиями, что привело еще к ряду публикаций коллектива  исполнителей грантов РФФИ, указанных в списке авторов настоящей работы. Все эти результаты приведены в предлагаемом ниже тексте, за который полностью отвечает первый автор.

Мы не включаем сюда подробное изложение результатов, касающихся двух задач, в определенном смысле вырожденных по отношению к общей случаю. Это, во-первых, классическая задача С.В.\,Ковалевской (нулевой гиростатический момент) и, во-вторых, гиростат Ковалевской\,--\,Яхья при условии равенства нулю постоянной циклического интеграла (постоянной площадей). Грубая топология первой задачи полностью описана в работах \cite{KhDan83,KhPMM83} и, со всеми деталями, в книге \cite{KhBook88}. В работе \cite{Oshem} последовательность бифуркаций на изоэнергетических уровнях объединена в графы Фоменко и найден полный список этих графов (дана классификация графов в плоскости энергии--момента). Метки на графах (тонкий топологический инвариант или инвариант Фоменко\,--\,Цишанга) вычислены в работе \cite{BRF}, в которой также содержится последовательное изложение метода круговых молекул в тонком топологическом анализе. Топологический анализ семейства систем, возникающих при нулевой постоянной площадей (свободный параметр -- величина гиростатического момента), выполнен в цикле работ \cite{Ryab2,Ryab1,Ryab7,Ryab6,Ryab5,RyabDis}. Метки на графах (тонкий топологический инвариант или инвариант Фоменко\,--\,Цишанга) вычислены в работах \cite{Mor,Mo2008}, где также доказан ряд утверждений общего характера, важных в приложениях метода круговых молекул. К этим задачам мы обращаемся лишь по мере необходимости, если возникает возможность перенести какие-либо утверждения или факты на общий случай.

\clearpage

\section{Аналитические результаты}\label{sec2}
\subsection{Уравнения и интегралы}
Уравнения Эйлера\,--\,Пуассона движения гиростата в поле силы тяжести в общем случае имеют вид
\begin{equation}\label{eq2_1}
\dot {\bm}=\bm \times \bo + \br \times \ba,\quad
\dot {\ba}= \ba \times \bo,
\end{equation}
где $\bo$ -- угловая скорость, $\ba$ -- единичный вектор направления силы тяжести (называемого вертикалью), $\br$ -- вектор, направленный из неподвижной точки $O$ в центр масс, равный по длине произведению веса системы ``тело--ротор'' на расстояние от $O$ до центра масс. Все объекты отнесены к подвижным осям. Вектор кинетического момента $\bm$ связан с угловой скоростью зависимостью
\begin{equation}\label{eq2_2}
\bm = \mbf{I} \bo + \bl,
\end{equation}
где $\mbf{I}$, $\bl$ -- тензор инерции в точке $O$ и вектор гиростатического момента, постоянные в подвижной системе отсчета. Для краткости объекты, не изменяющиеся по отношению к подвижной системе отсчета, будем называть ``постоянными в теле''. Система \eqref{eq2_1} получена из механической системы с конфигурационным пространством $SO(3)$ (гамильтоновой системы с тремя степенями свободы) факторизацией по действию группы вращений вокруг постоянного в инерциальном пространстве направления вертикали, и ее фазовым пространством является пятимерное многообразие $\mP^5=\bR^3(\bo){\times}S^2(\ba)$.

Здесь и далее через $\bR^n(u_1,\ldots,u_n)$ обозначаем пространство $\bR^n$ с выделенной системой координатных функций $u_1,\ldots,u_n$. Обозначение $S^2(\ba)$ относится к единичной сфере в $\bR^3(\ba)$.

Вытекающее из определений соотношение
\begin{equation}\label{eq2_3}
\Gamma=1,
\end{equation}
где
\begin{equation*}
\Gamma=\alpha_1^2+ \alpha_2^2+\alpha_3^2,
\end{equation*}
в механике называется геометрическим интегралом. Другими общими интегралами являются гамильтониан (полная энергия)
\begin{equation}\label{eq2_4}
H=\frac{1}{2}\mbf{I}\bo \cdot \bo - \br \cdot \ba
\end{equation}
и циклический интеграл (интеграл площадей)
\begin{equation}\label{eq2_5}
L= \frac{1}{2} \bm  \cdot \ba.
\end{equation}
Коэффициент $1/2$ вводим, следуя С.В.\,Ковалевской. Скобка Пуассона рассматриваемой механической системы на пространстве $\{(\bm,\ba)\}$ может быть представлена в виде своих значений на парах координатных функций (см., например, \cite{BogRus2})
\begin{equation}\label{eq2_6}
\{M_i,M_j\}=\varepsilon_{ijk} M_k, \quad \{M_i,\alpha_j\}=\varepsilon_{ijk} \alpha_k, \quad \{\alpha_i,\alpha_j\}=0.
\end{equation}
На любом фиксированном уровне интеграла площадей
\begin{equation}\label{eq2_7}
\mPel = \{L=\ell\} \subset \mP^5
\end{equation}
относительно этой скобки индуцируется гамильтонова система с двумя степенями свободы (понижение порядка по Раусу). Для полной интегрируемости необходимо указать еще один общий интеграл, независимый с $H,L$ почти всюду на $\mP^5$. Под независимостью функций в точке понимается линейная независимость их дифференциалов. Независимость двух функций почти всюду иногда называют функциональной независимостью \cite{BolFom}. Случаи, когда можно указать подмногообразия в $\mP^5$ положительной коразмерности, отличные от $\mPel$, инвариантные для системы \eqref{eq2_1}, на которых эта система интегрируется в квадратурах, называются частными случаями интегрируемости.

Следующие предположения будем называть условиями Ковалевской\,--\,Яхья{\footnote{В соответствии с правилами русского языка арабская фамилия Яхья склоняется. Однако нам представляется более уважительным сохранить ее неизменной.}}. Пусть в качестве подвижной системы отсчета выбран ортонормированный триэдр $O\mbf{e}_1 \mbf{e}_2 \mbf{e}_3$ главных осей тензора инерции. Предположим, что соответствующие главные моменты инерции находятся в отношении $2{:}2{:}1$, центр масс лежит в экваториальной плоскости $\br \cdot \mbf{e}_3=0$, а гиростатический момент направлен по оси динамической симметрии
\begin{equation}\notag
\bl = \ld \mbf{e}_3.
\end{equation}
Тогда оставшимся произволом в выборе подвижных осей и введением подходящих единиц измерения можно добиться того, чтобы было
\begin{equation}\label{eq2_8}
\br = \{1,0,0\}, \qquad \mbf{I}=\diag\{2,2,1\}, \qquad \ld > 0,
\end{equation}
и уравнения \eqref{eq2_1} примут вид
\begin{equation}\label{eq2_9}
\begin{array}{lll}
2\dot\omega _1   = \omega _2 (\omega _3- \ld)  , &
2\dot\omega _2 =  - \omega _1 (\omega _3-\ld)  - \alpha _3 , &
\dot\omega _3   = \alpha _2, \\
\dot\alpha _1   = \alpha _2 \omega _3  - \alpha _3 \omega _2 , &
\dot\alpha _2   = \alpha _3 \omega _1  - \alpha _1 \omega _3 , &
\dot\alpha_3   = \alpha_1 \omega_2  - \alpha_2 \omega_1.
\end{array}
\end{equation}
Эта система обладает дополнительным интегралом Ковалевской\,--\,Яхья, обозначаемым далее через $K$. Таким образом, имеется полный инволютивный набор интегралов
\begin{equation}\label{eq2_10}
\begin{array}{l}
H = \omega _1^2  + \omega _2^2 + \ds{\frac{1}{2}}\omega _3^2 -
\alpha _1, \qquad
L = \omega _1 \alpha _1  + \omega _2 \alpha _2  +\frac{1}{2} (\omega _3+\ld) \alpha _3, \\
K=(\omega_1^2-\omega^2_2+\alpha_1)^2+(2\omega_1\omega_2+\alpha_2)^2 + 2\ld[(\omega_3-\ld)(\omega_1^2+\omega^2_2)+2\omega_1 \alpha_3].
\end{array}
\end{equation}

Здесь следует отметить один важный аспект терминологии, иногда приводящий к ложному ощущению наличия и исследования более общих задач, чем описанная последними формулами. Поскольку центр масс лежит в экваториальной плоскости эллипсоида инерции, любой ортонормированный базис в ней является базисом главных осей инерции. Учитывая еще и произвол в выборе единиц измерения, запишем условия \eqref{eq2_8} в виде
\begin{equation*}
\br = \{c_1,c_2,0\}, \qquad \mbf{I}=\diag\{2C,2C,C\}, \qquad \ld \in \bR.
\end{equation*}
В современной литературе по отношению к этим условиям появился термин ``семейство систем типа Ковалевской\,--\,Яхья''. Подчеркнем, что любая такая система ``типа'' Ковалевской\,--\,Яхья полностью эквивалентна рассматриваемой здесь системе с условиями \eqref{eq2_8}, то есть получена из нее \textit{глобальной невырожденной линейной} заменой переменных, поэтому применение указанного термина нам представляется некорректным.

\subsection{Точные решения}
До работы \cite{Yeh1} были известны три общих случая интегрируемости при $\bl \ne 0$ -- обобщение случая Эйлера, данное Н.Е.\,Жуковским ($\br=0$), очевидное обобщение случая Лагранжа ($\bl \times \br=0$) и обобщение случая Горячева\,--\,Чаплыгина, данное Л.Н.\,Сретенским. Ряд частных случаев интегрируемости тяжелого гиростата ($\br \ne 0$) указан в работах П.В.\,Харламова и Е.И.\,Харламовой. Два из них оказались принципиально важны для исследования системы \eqref{eq2_9}. Остановимся на них подробнее.

\subsubsection{Решение П.В.\,Харламова}
Первое точное решение\footnote{Точным решением обычно называют общий или частный случай интегрируемости при наличии его явного сведения к квадратурам.} опубликовано в \cite{PVLect}. Предположим, что в системе \eqref{eq2_1}, \eqref{eq2_2}
\begin{equation}\label{eq2_11}
    \bo = \bo_0 + r(t) \be  \qquad (r \ne {\rm const}),
\end{equation}
где $\bo_0$ и $\be $ постоянны в теле. Без ограничения общности полагаем $|\be |=1$ и $\bo_0 \cdot \be =0$. Представление \eqref{eq2_11} является наиболее общим для решений, допускающих два линейных по угловым скоростям частных интеграла. Хорошо известно, что все решения системы \eqref{eq2_1} определены для всех $t$ и ограничены (например, в силу компактности любого уровня интеграла энергии). Отсюда следует, что если функция $r(t)$ удовлетворяет уравнению вида
\begin{equation}\label{eq2_12}
    b \, \dot r = \sum_{i=0}^k a_i \, r^i
\end{equation}
с постоянными коэффициентами $b,a_i$, то все эти коэффициенты равны нулю.

\begin{theorem}[П.В.\,Харламов \cite{PVLect}]\label{th1}
Для любого движения вида \eqref{eq2_11} выполнены следующие условия:

$(1)$ ось $O\be $ является главной осью инерции в закрепленной точке;

$(2)$ центр масс лежит в плоскости, проходящей через закрепленную точку перпендикулярно $O\be $, то есть
\begin{equation}\notag
    \br  \cdot \be =0;
\end{equation}

$(3)$ полный кинетический момент системы компланарен векторам $\br, \be $.
\end{theorem}

\begin{proof}
Домножим первое векторное уравнение \eqref{eq2_1} (уравнение Эйлера) скалярно на постоянный отличный от нуля вектор $\br$. Получим уравнение вида \eqref{eq2_12}, в котором $k=2$ и при этом $b=\br \cdot \ie$, $a_2=\br \cdot (\ie \times \be)$. Таким образом,
\begin{equation}\label{eq2_13}
    \br \cdot \ie =0, \qquad \br \cdot (\ie \times \be)=0.
\end{equation}

Запишем геометрическое тождество
\begin{equation}\label{eq2_14}
    \ba = \frac{1}{\br^2} [(\br \cdot \ba) \br +(\br \times \ba) \times \br].
\end{equation}
Обозначая через $h$ постоянную интеграла \eqref{eq2_4}, из \eqref{eq2_1}, \eqref{eq2_4} найдем выражения
\begin{equation}\notag
\br \times \ba=\dot {\mbf{I} \bo}- (\mbf{I} \bo + \bl) \times \bo,\quad
\br \cdot \ba= h-\frac{1}{2} {\mbf{I} \bo \cdot \bo}.
\end{equation}
Подставим их в правую часть \eqref{eq2_14}, после чего соотношение $L=\ell$ с интегралом площадей \eqref{eq2_5} даст уравнение вида \eqref{eq2_12}, в котором $k=3$,
$b=(\mbf{I}\bo + \bl) \cdot (\ie \times \br)$, $a_3=2(\br \cdot \be) (\ie \cdot \ie)-(\br \cdot \ie)(\ie \times \be)$. Поэтому, с учетом \eqref{eq2_13}, имеем
\begin{equation}\label{eq2_15}
    \br \cdot \be =0, \qquad \bm \cdot (\ie \times \br)=0.
\end{equation}
При ненулевом $\br$ первое из этих равенств вместе с \eqref{eq2_13} дает $\ie \times \be=0$, поэтому ось $O \be$ -- главная ось тензора $\mbf{I}$. Тогда равенства \eqref{eq2_15} выражают оставшиеся утверждения теоремы.
\end{proof}

Отметим, что здесь не вводились требования на постоянный в теле вектор $\bo_0$. В частности, при $\bo_0=0$ получаем, что теорема верна для всех маятниковых движений гиростата в поле силы тяжести.

В соответствии с доказанным, введем подвижную систему отсчета так, что $\mbf{e}_3=\be$ и обозначим компоненты постоянных векторов и тензора инерции
\begin{equation}\notag
    \bo_0 = \po  \, \be_1 + \qo  \, \be_2, \qquad \bl = \ld_1 \be_1 + \ld_2 \be_2 + \ld_3 \be_3, \qquad \mbf{I} = \diag \{A,B,C\}.
\end{equation}
Пусть $\bm_0$ -- проекция вектора $\bm$ на плоскость $O \be_1 \be_2$. По теореме \ref{th1}
\begin{equation}\label{eq2_16}
    \bm_0 = m \br, \qquad m = \cons.
\end{equation}
Проекция уравнения Эйлера на плоскость $O \be_1 \be_2$ примет вид
\begin{equation}\notag
    r\, ( m \br - C \bo_0) + \alpha_3 \br - \ld_3 \bo_0 =0.
\end{equation}
Так как по предположению $\br \ne 0, r \ne \cons$, отсюда следует, что
\begin{eqnarray}
& & \bo_0 = \vq \br \label{eq2_17}, \\
& & (m- \vq C) r + \alpha_3 - \ld_3 \vq  =0. \label{eq2_18}
\end{eqnarray}
Из \eqref{eq2_16}, \eqref{eq2_17} получим $\bm_0 \times \bo_0=0$, поэтому уравнение Эйлера в проекции на ось $O\be_3$ имеет вид
\begin{equation}\label{eq2_19}
    C \dot r + c_2 \alpha_1 - c_1 \alpha_2 =0.
\end{equation}
В то же время, дифференцируя \eqref{eq2_18}, с учетом \eqref{eq2_17} получим
\begin{equation}\label{eq2_20}
    (m- \vq C) \dot r+ \vq (c_2 \alpha_1 - c_1 \alpha_2)=0.
\end{equation}
Условие совместности \eqref{eq2_19}, \eqref{eq2_20} дает
\begin{equation}\label{eq2_21}
    m = 2 C \vq.
\end{equation}
В результате, выбирая произвольно константы
\begin{equation}\notag
    A,B,C,c_1, c_2, \ld_3, \vq, h,
\end{equation}
получим решение в виде
\begin{equation}\label{eq2_22}
    \begin{array}{l}
      \omega_1 = \vq c_1, \qquad  \omega_2 = \vq c_2, \qquad \omega_3 = r,\\
      \ds{\alpha_1=\frac{1}{c_1^2+c_2^2}\left[ (\frac{1}{2} C r^2+h_*)c_1- C\sqrt{R(r)} c_2\right], }\\[3mm]
      \ds{\alpha_2=\frac{1}{c_1^2+c_2^2}\left[ (\frac{1}{2} C r^2+h_*)c_2+ C\sqrt{R(r)} c_1\right], }\\[3mm]
      \alpha_3 = \vq (\ld_3 - C r),
    \end{array}
\end{equation}
где обозначено
\begin{eqnarray}
  & & h_* = \frac{\vq^2}{2}(A c_1^2+B c_2^2) - h, \nonumber \\
  & & R(r)= \frac{1}{C} \left\{ (c_1^2+c_2^2)\left[ 1-\vq ^2 (\ld_3- C r)^2\right] - (\frac{1}{2} C r^2+h_*)^2\right\}.\nonumber
\end{eqnarray}
При этом две компоненты гиростатического момента, согласно \eqref{eq2_16}, \eqref{eq2_21}, заданы равенствами
\begin{equation}\label{eq2_23}
    \ld_1 = (2C-A) \vq c_1, \qquad \ld_2 = (2C-B) \vq c_2,
\end{equation}
а изменение переменной $r$, выбранной в качестве основной, задается уравнением
\begin{equation}\label{eq2_24}
    \dot r = \sqrt{R(r)}.
\end{equation}
В связи с произволом выбора единиц измерения, среди пяти физических параметров системы $A,B,C,c_1,c_2,\ld_3$ существенны лишь три. Параметры $\vq, h$, связанные с выбором начальной точки, произвольны. Поэтому формулы \eqref{eq2_22} и \eqref{eq2_24}, при фиксированных физических параметрах, описывают трехмерное инвариантное подмногообразие фазового пространства, расслоенное на периодические траектории (с возможными бифуркациями). Это решение обобщает классический случай Бобылева\,--\,Стеклова.

Рассмотрим гиростат при условиях Ковалевской ($A=B=2C$, центр масс в экваториальной плоскости) и запишем решения \eqref{eq2_22}, где в качестве оси $O\be$, главной по теореме~\ref{th1}, выбрана ось динамической симметрии. Из \eqref{eq2_23} вытекает, что $\ld_1=\ld_2=0$. Любая ось в экваториальной плоскости является главной, поэтому без ограничения общности считаем $c_2=0$, тогда $\omega_2 \equiv 0$. Выберем единицы измерения так, что $C=1, c_1 =1$, а за независимый параметр вместо $\vq$ выберем первую компоненту угловой скорости $\po $. Опуская индекс у единственной ненулевой компоненты $\ld_3$, из \eqref{eq2_22} получим параметрические уравнения трехмерного многообразия, обозначаемого далее через $\mm$:
\begin{equation}\label{eq2_25}
\mm: \left\{    \begin{array}{l}
      \omega_1 = \po , \qquad  \omega_2 = 0, \qquad \omega_3 = r,\\
      \ds{\alpha_1=\frac{1}{2} r^2+\po ^2 - h, }   \qquad      \ds{\alpha_2=\sqrt{R(r)}, } \qquad \alpha_3 = - \po  (r - \ld).
    \end{array} \right.
\end{equation}
Здесь
\begin{equation}\label{eq2_26}
    \ds{R=-\frac{1}{4}r^4-(2\po ^2-h)r^2+2 \ld \po ^2 r+1-(\po ^2-h)^2-\po ^2\ld^2.}
\end{equation}
Динамика на $\mm$ задана уравнением \eqref{eq2_24}. Точное решение \eqref{eq2_24}, \eqref{eq2_25}, \eqref{eq2_26} обобщает на гиростат семейство особо замечательных движений 4-го класса Аппельрота классического волчка Ковалевской. Отметим, что обобщение 4-го класса Аппельрота на волчок и гиростат в двойном поле получены в \cite{KhRCD05,KhND07}.

\subsubsection{Решение П.В.\,Харламова и Е.И.\,Харламовой}
Второе точное решение для гиростата в поле силы тяжести, имеющее непосредственное отношение к рассматриваемой задаче, построено в работах \cite{EIPVHDan,PVMtt71}. Исходными предположениями являются условие принадлежности центра масс и гиростатического момента одной из главных плоскостей тензора инерции и наличие частных алгебраических интегралов, гарантирующих возможность выражения всех фазовых переменных через одну вспомогательную. При условиях \eqref{eq2_8} наиболее общая форма соответствующих решений представлена в \cite{PVMtt71}, где они объединены в одно семейство.
Запишем его в удобных для дальнейшего обозначениях, как предложено в \cite{mtt40}.
Фиксируем постоянную $\ell$ интеграла площадей и пусть $s$ -- некоторая отличная от нуля константа. Положим
\begin{equation}\label{eq2_27}
\begin{array}{c}
   \vk ^2 = \ell^2+\ld^2 s^2, \quad \ro^2=1-\ds{\frac{2 \vk^2}{s}},  \quad
   \Fun^2=\ds{\frac{1}{2}}\left[\bigl(X+\ds{\frac{\ld}{\vk}}\bigr)^2+\bigl(\ro Y +\ds{\frac{\ell}{s \vk}}\bigr)^2-1\right],\\
   (X,Y)=\left\{ \begin{array}{ll}
   (\cos \sigma, \sin \sigma), & \ro^2 \gs 0 \\
   (\cosh \sigma, \ri \sinh \sigma), & \ro^2 < 0
   \end{array}\right.
\end{array}
\end{equation}
Здесь $\ri$ -- мнимая единица, $\sigma$ -- вспомогательная переменная. Инвариантные многообразия $\mn$ и $\mo$ определены соответственно при $s<0$ и $s>0$ следующей системой параметрических уравнений
\begin{equation}\label{eq2_28}
\mno: \left\{
\begin{array}{lll}
  \ds{\omega_1=-\frac{\ell} {s}- \vk \ro Y,} &
  \ds{\omega_2=-\ro \sqrt{s}\, \Fun,} &
  \ds{\omega_3 = \ld+2 \vk X,}\\[2mm]
  \ds{\alpha_1=\frac{\ld s X+\ell\ro Y}{\vk} -2\vk^2Y^2 ,} &
  \ds{\alpha_2=-2 \vk Y \sqrt{s} \, \Fun, } &
  \ds{\alpha_3 = \frac{\ell X-\ld s \ro Y }{\vk}.}
\end{array}\right.
\end{equation}
Динамика, индуцированная системой \eqref{eq2_9}, описывается уравнением
\begin{equation}\label{eq2_29}
    \dot \sigma^2  = \sgn(\ro^2)\, s \, \Fun^2.
\end{equation}
Поскольку $\ell$ и $s$ -- свободные параметры, то при заданных физических параметрах задачи вновь имеем трехмерное инвариантное подмногообразие в фазовом пространстве, расслоенное на периодические траектории (как и в первом решении на этом многообразии могут происходить и бифуркации периодических траекторий).

Непосредственно проверяется, что при $\ld=0$ на семействе траекторий \eqref{eq2_28}, \eqref{eq2_29} имеется следующая связь между интегральными константами
\begin{equation}\notag
\begin{array}{c}
(h-2 \ell^2)^2-k = 0,
\end{array}
\end{equation}
что соответствует особо замечательным движениям 2-го ($s<0$) и 3-го ($s>0$) классов Аппельрота, а константа $s$ оказывается кратным корнем многочлена в дифференциальных уравнениях Ковалевской. Имеются и соответствующие обобщения на двойное поле \cite{KhRCD05,KhND07}, но в случае $\ld \ne 0$ они не сведены к квадратурам.

\subsubsection{Решение И.Н.\,Гашененко}
Попытки построить явное разделение переменных в приведенной системе на $\mPel$ для гиростата Ковалевской\,--\,Яхья хотя бы при $\ell=0$ успехом не увенчались. В работах \cite{Kuzn,BorMam2,Tsi3} предложены замены переменных на $\mP^4_0$, приводящие к уравнениям типа Абеля--Якоби с комплексными переменными. Их овеществление произвести пока не удалось, в связи с чем для исследования системы (в частности, для ответа на вопрос об условиях существования вещественных решений, который К.\,Якоби отмечал как основную цель перехода к уравнениям с радикалами) они не пригодны.

Сведение к квадратурам решений системы \eqref{eq2_9} на четырехмерном подмногообразии в $\mP^5$, отличном от уровня интеграла площадей, предложил И.Н.\,Гашененко \cite{Gash1,Gash3}. В этих работах также выполнена полная классификация найденного семейства решений по типу движений (периодическое, асимптотическое к периодическому, двоякопериодическое). Приведем основные результаты работ \cite{Gash1,Gash3}. Рассмотрим соотношения на постоянные $h,\ell,k$ первых интегралов \eqref{eq2_10}, порожденные семейством решений \eqref{eq2_25}:
\begin{equation}\label{eq2_30}
    \begin{array}{l}
            \ell=\ds{-\frac{1}{2}\po (2h-\ld^2)+\po ^3}, \qquad  k=1-\po ^2(2h-\ld^2)+3 \po ^4.
    \end{array}
\end{equation}
Полным прообразом множества таких значений ($\po $ и $h$ произвольны) в $\mP^5$ служит стратифицированное многообразие $G^4=N^4 \cup \mm$, где $\dim N^4=4$ и $\partial N^4=\mm$. Перейдем на $N^4$ от $({\bs \omega},{\bs \alpha})$ к новым координатам $(x,y,z,\alpha,\beta,\gamma)$, полагая \cite{Gash1}
\begin{equation}\label{eq2_31}
    \begin{array}{l}
      \omega_1=\po -x M^{-1},\qquad  \omega_2=-y M^{-1},\\[2mm]
      \omega_3=2z +\ld +4x \gamma M^{-1}, \qquad  M=x^2+y^2 \ne 0,\\[2mm]
      \alpha_1=-2\alpha+4(x^2-y^2)\gamma^2 M^{-2}+2\po  x(1-2\po  x)M^{-1}+4\ld x \gamma M^{-1}, \\[2mm]
      \alpha_2=-2\beta+8 x y \gamma^2 M^{-2}+2\po  y(1-2\po  x)M^{-1}+4\ld y \gamma M^{-1}, \\[2mm]
      \alpha_3=2(1-2 \po  x)\gamma M^{-1}-2 \po  z.
    \end{array}
\end{equation}
В силу \eqref{eq2_30} из уравнений движения и первых интегралов выделяется замкнутая подсистема
\begin{equation}\label{eq2_32}
    \begin{array}{c}
      \ds \dot z= \sqrt{f(z)}, \quad \dot \alpha = (z-\frac{1}{2}\ld)^2+ \frac{1}{2}(\po ^2-h), \quad \beta = \sqrt{f(z)},\\[2mm]
      \ds f(z)=\ds{\frac{1}{4}-\po ^2 z^2-[(z-\frac{1}{2}\ld)^2+ \frac{1}{2}(\po ^2-h)]^2},
    \end{array}
\end{equation}
решение которой $z(t),\alpha(t),\beta(t)$ находится в эллиптических функциях времени.
Вводя обозначения независимых параметров
\begin{equation}\label{eq2_33}
    \begin{array}{l}
      \ds{L_1=\frac{1}{16}(2h-2\po ^2-\ld^2)}, \\[2mm]
      \ds{L_2=\frac{1}{4}(2h-6\po ^2-\ld^2)}, \\[2mm]
      \ds{L_3=4(\po ^2+\ld^2)L_1-\frac{1}{4}},
    \end{array}
\end{equation}
функцию $\Psi$ и вспомогательную переменную $\eta$
\begin{equation}\label{eq2_34}
    \begin{array}{l}
    \Phi(z)=(z-\ld)^2-L_2, \\
    \ds \eta = y \, \Phi^{-\frac{1}{2}}(z),
    \end{array}
\end{equation}
получим эллиптическую квадратуру для $\eta$ и конечные выражения для оставшихся переменных
\begin{equation}\label{eq2_35}
    \begin{array}{c}
      \ds{\int \frac{d\eta}{\sqrt{\mstrut L_1-L_2 L_3 \eta^2}}}+ \ds{\int \frac{dz }{\Phi(z)\sqrt{\mstrut f(z)}}}=\cons,\\[3mm]
      y=\eta \Phi^{\frac{1}{2}}(z), \quad x=-\left(\ds{\frac{1}{2}} \po  L_2^{-1}+\eta \beta \Phi^{-\frac{1}{2}}(z)+\dot \eta \Phi^{\frac{1}{2}}(z)(z-\ld)L_2^{-1} \right), \\[3mm]
      \gamma = \ds{\frac{1}{2}} \po  z L_2^{-1}+\ld \eta \beta \Phi^{-\frac{1}{2}}(z)+\dot \eta \Phi^{\frac{1}{2}}(z) (\alpha+\po ^2)L_2^{-1}.
    \end{array}
\end{equation}

Классификация движений в найденном решении проводится по параметрам \eqref{eq2_33}. Выделяются следующие случаи \cite{Gash1}:
\begin{equation}\label{eq2_36}
\begin{array}{ll}
\mathrm{(I)}& L_1>0,L_2>0,L_3>0;\\
\mathrm{(II)}& L_1>0,L_2<0,L_3<0;\\
\mathrm{(III)}& L_1>0,L_2<0,L_3>0;\\
\mathrm{(IV)}& L_1>0,L_2>0,L_3<0;\\
\mathrm{(V)}& L_1<0;\\
\mathrm{(VI)}& L_1=0;\\
\mathrm{(VII)}& L_2=0;\\
\mathrm{(VIII)}& L_3=0.
\end{array}
\end{equation}

Как отмечает автор \cite{Gash1}, в случаях I, II возможны двоякопериодические движения, в случаях III, IV, VII, VIII движения, отличные от \eqref{eq2_25}, асимптотически стремятся к траекториям на $\mm$, а в случаях V, VI движения, отличные от \eqref{eq2_25}, невозможны. Ниже мы будем возвращаться к обсуждению этих классов и укажем их связь с типами особых точек отображения момента.

\clearpage

\section{Критическое множество отображения момента}\label{sec3}
\subsection{Представление Лакса}
Несмотря на то, что в работе \cite{Yeh1} интеграл Ковалевской был обобщен и на двойное поле, такая задача не рассматривалась как интегрируемая, поскольку двойное поле препятствует существованию интеграла площадей. Однако позже в работе \cite{ReySemRus} было указано представление Лакса со спектральным параметром для гиростата типа Ковалевской в двойном поле, откуда следовало существование еще одного первого интеграла, который при исчезновении второго поля превращается в квадрат интеграла площадей. Приведем соответствующее представление Лакса и вытекающие из него результаты при условиях Ковалевской\,--\,Яхья.

Следуя С.В.\,Ковалевской введем комплексные переменные ($\ri^2=-1$):
\begin{equation}\label{eq3_1}
\begin{array}{c}
x_1 = \alpha_1 + \ri \alpha_2, \quad x_2 = \alpha_1  - \ri \alpha_2, \\
w_1 = \omega_1 + \ri\omega_2 , \quad w_2 = \omega_1  - \ri \omega_2, \\
z = \alpha_3, \quad \quad w_3 = \omega_3
\end{array}
\end{equation}
(последняя строка переобозначений введена для удобства и единообразия). Обозначая штрихом дифференцирование $d/d(it)$, запишем систему \eqref{eq2_9} так:
\begin{equation}\label{eq3_2}
\begin{array}{rcl}
2w'_1  =  - (w_1 w_3  + z ),& 2w'_2  = w_2 w_3  + z, &
2w'_3 = x_1  - x_2, \\
{x'_1  =  - x_1 w_3  + z w_1,} & {x'_2  = x_2 w_3  - z w_2,} & {2z'  = x_1 w_2  - x_2 w_1.}
\end{array}
\end{equation}

Представление Лакса для этой системы, как частный случай результата работы \cite{ReySemRus}, имеет вид
\begin{equation} \label{eq3_3}
\displaystyle{B' = B A - A B},
\end{equation}
где
\begin{equation} \label{eq3_4}
\begin{array}{l}
\displaystyle{B=\begin{Vmatrix} \displaystyle{2 \ld} &
\displaystyle{\frac{x_2}{\kappa}} & \displaystyle{-2 w_1} &
\displaystyle{\frac{z}{\kappa}}
\\[2mm]
\displaystyle{-\frac{x_1}{\kappa}} & \displaystyle{-2 \ld} &
\displaystyle{-\frac{z}{\kappa}} &
\displaystyle{2 w_1 } \\[2mm]
\displaystyle{-2 w_1} & \displaystyle{\frac{z}{\kappa}} &
\displaystyle{-2 w_3} &
\displaystyle{-\frac{x_1}{\kappa} - 4 \kappa} \\[2mm]
\displaystyle{-\frac{z}{\kappa}} & \displaystyle{2 w_2} &
\displaystyle{\frac{x_2}{\kappa}+ 4\kappa} & \displaystyle{2 w_3}
\end{Vmatrix}},\quad
\displaystyle{A=\begin{Vmatrix} \displaystyle{-\frac{w_3}{2}} &
\displaystyle{0} & \displaystyle{\frac{w_2}{2}} &
\displaystyle{0}\\[2mm]
\displaystyle{0} & \displaystyle{\frac{w_3}{2}} & \displaystyle{0} &
\displaystyle{-\frac{w_1}{2}}\\[2mm]
\displaystyle{\frac{w_1}{2}} & \displaystyle{0} &
\displaystyle{\frac{w_2}{2}} &
\displaystyle{\kappa}\\[2mm]
\displaystyle{0} & \displaystyle{-\frac{w_2}{2}} &
\displaystyle{-\kappa} & \displaystyle{-\frac{w_3}{2}}
\end{Vmatrix},}
\end{array}
\end{equation}
через $\kappa$ обозначен спектральный параметр, производная в \eqref{eq3_3} вычисляется в силу \eqref{eq3_2}. Уравнение для собственных значений $\mu$ матрицы $B$ определяет ассоциированную с данным представлением алгебраическую кривую \cite{Bob}. Положим $s=2\kappa^2$. Уравнение алгебраической кривой примет вид
\begin{equation} \label{eq3_5}
\begin{array}{l}
\displaystyle{ \mu ^4 -  4\left[ \frac{1}{s}  - (2  h +
\ld ^2) + 2  s \right] \mu ^2 + 4   \left[\frac{1}{s^2} + \frac{2}{s}(4  \ell^2 -
2 h-\ld ^2)+(4 k + 8\ld ^2 h)
 - 8\ld ^2 s \right] = 0.}
\end{array}
\end{equation}

Введем порожденное интегралами \eqref{eq2_10} интегральное отображение (называемое в современной литературе отображением момента)
\begin{equation}\label{eq3_6}
    J=L{\times}H{\times}K: \mP^5 \to \bR^3.
\end{equation}
Множество $\Sigma$ его критических значений называется бифуркационной диаграммой и является классифицирующим множеством при исследовании фазовой топологии системы. Как показывает опыт такого исследования, при наличии представления Лакса бифуркационная диаграмма отображения $J$ содержится в множестве тех значений $(\ell,h,k)$, при которых кривая \eqref{eq3_5} перестает быть неособенной, то есть либо является приводимой -- левая часть уравнения \eqref{eq3_5} распадается в произведение рациональных выражений, либо имеет особую точку в стандартном смысле. Таким путем можно предугадать результат следующего утверждения. Однако для его строгого доказательства необходимы непосредственные вычисления на множестве критических точек $J$. Это множество будет указано ниже.

\begin{theorem}\label{th2}
Бифуркационная диаграмма интегрального отображения $L{\times}H{\times}K$ содержится в объединении следующих $($пересекающихся$)$ поверхностей
в пространстве $\bbI$:
\begin{eqnarray}
& \wsa : & \left\{ \begin{array}{l} \displaystyle{k =
1+(h-\frac{\ld^2}{2})^2-4(h-\frac{\ld^2}{2})s+3s^2,} \\[3mm]
\displaystyle{\ell^2=(h-\frac{\ld^2}{2}-s)s^2,
\qquad s \in {\bR};}
\end{array} \right. \label{eq3_7} \\
& \wsi_{2,3}: & \left\{ \begin{array}{l} \displaystyle{k = -
2\ld^2(h-\frac{\ld^2}{2}-2s)-\ld^4+ \frac{1}
{4s^2},} \\[3mm]
\displaystyle{\ell^2=\frac{1}{2}(h+\frac{\ld^2}{2})-\ld^2 s^2-\frac{1}{4s},
\qquad s \in {\bR}\backslash\{0\}.}
\end{array} \right. \label{eq3_8}
\end{eqnarray}
Здесь $s<0$ для $\wsb$ и $s>0$ для $\wsc$.
\end{theorem}

\begin{remark}\label{rem1}
В представленном виде уравнения бифуркационных диаграмм получены из \eqref{eq3_5}. В качестве независимых параметров на поверхностях выбраны $h,s$. Поскольку $h$ и $\ell^2$ при этом связаны линейно, то можно выразить $k,h$ через $s,\ell$. На $\wsb, \wsc$ проблем не возникает. На $\wsa$ возникает искусственная особенность, связанная с возможностью $s=0$. Геометрически это соответствует линии самопересечения поверхности $\wsa$, которая в параметрической записи \eqref{eq3_7} особенностью не является. Однако при рассмотрении бифуркационных диаграмм приведенных систем на $\mPel$ важно иметь и выражение через $s,\ell$, которое запишем в виде
\begin{equation}\label{eq3_9}
\begin{array}{l}
\wsa =
\left\{ \displaystyle{h= \frac{\ell^2}{s^2}+\frac{\ld^2}{2} + s,}\;
\displaystyle{k=\frac{\ell^4}{s^4} - \frac{2 \ell^2}{s}+1}, \; \ell s \neq 0
\right\} \cup \\[3mm]
\phantom{\wsa = } \quad \cup \left\{k=1, \; \ell=0 \right\} \cup \left\{\displaystyle{k=1+(h-\frac{\ld^2}{2})^2},\;  \ell=0 \right\}.
\end{array}
\end{equation}
\end{remark}

\subsection{Описание критического множества}
Множество критических точек отображения момента впервые описано в работе \cite{Ryab2}, где приведены и параметрические уравнения особых поверхностей $\wsi_j$.  Более детальное изложение представлено в \cite{Ryab1}. Параллельно такие же результаты получены в \cite{Gash4}. Показано, что поверхности $\wsi_j$ естественным образом возникают как дискриминантные множества некоторых многочленов. Уравнения поверхности $\wsa$ также следуют из результатов работы \cite{Gash1}, но в ее контексте соответствующие условия на постоянные интегралов не связывались с понятием критических точек.

Множество $\mct$ критических точек отображения момента стратифицировано рангом этого отображения. В силу того, что интеграл $L$ всюду регулярен и расслаивает $P^5$ на гладкие симплектические листы \eqref{eq2_7}, естественно принять следующую терминологию.

\begin{defin}\label{def1}
Рангом точки $x \in \mPel \subset \mP^5$ будем называть ранг в этой точке ото\-браже\-ния-огра\-ни\-че\-ния
\begin{equation}\label{eq3_10}
    \jpriv_\ell=H{\times}K|_{\mPel}:\mPel \to \bR^2.
\end{equation}
\end{defin}
Таким образом, ранг точки на единицу меньше, чем ранг в этой точке отображения \eqref{eq3_6}. В соответствии с этим имеем
\begin{equation}\label{eq3_11}
    \mct = \mct^0 \cup \mct^1,  \qquad \mct^i=\{x \in \mPel | \rk \jpriv_\ell(x)=i \}.
\end{equation}
Определенная выше бифуркационная диаграмма есть $J$-образ множества $\mct$.

Наличие особых поверхностей \eqref{eq3_7}, \eqref{eq3_8} порождает другое разбиение критического множества.
\begin{theorem}[\cite{Ryab2,Ryab1,Gash4}]\label{th3}
Множество критических точек отображения момента имеет вид
\begin{equation}\label{eq3_12}
    \mct = \mm \cup \mn \cup \mo,
\end{equation}
где $\mi_j$ есть замыкание множества $J^{-1}(\wsi_j) \cap \mct^1$ $(j=1,2,3)$.
\end{theorem}

Равенство \eqref{eq3_12} следует из аналитического описания множества $\mct$, полученного в \cite{Ryab2} с помощью исследования миноров матрицы Якоби
\begin{equation*}
    \frac{\partial(L,H,K,\Gamma)}{\partial({\bs \omega},{\bs \alpha})}.
\end{equation*}
Элегантное доказательство с помощью преобразования интегральных уравнений к симметричной системе шести комплексных координат на $\bR^6{({\bs \omega},{\bs \alpha})}$ имеется в \cite{Gash4}. Тот факт, что образ каждого из многообразий $\mi_j$ содержится в соответствующей поверхности $\wsi_j$, проверяется непосредственным вычислением. При этом на $\mm$ нужно в формулах \eqref{eq2_30} положить
\begin{equation}\label{eq3_13}
    \ds s= h -\frac{\ld^2}{2}-\po ^2.
\end{equation}

\begin{defin}\label{def2}
Многообразия $\mm$, $\mn$, $\mo$ с индуцированной на них динамикой  будем называть, соответственно, первой, второй и третьей критическими подсистемами случая Ко\-ва\-лев\-ской\,--\,Яхья.
\end{defin}

\begin{remark}\label{rem2}
Мы называем множества $\mi_j$ многообразиями, имея в виду их инвариантность {\rm (``}инвариантные многообразия{\rm '').} В действительности может оказаться, что они являются гладкими многообразиями лишь почти всюду.
\end{remark}

Представление $\mct$ в виде множества решений систем инвариантных соотношений получим как частный случай соответствующего результата работы \cite{KhND07}.

\begin{theorem} \label{th4} Множество критических точек интегрального отображения
$J$ состоит из следующих подмножеств в $\mP^5:$

$1)$ множества, определяемого системой
\begin{equation}\label{eq3_14}
R_1  = 0,\quad R_2  = 0,
\end{equation}
где
\begin{equation}\label{eq3_15}
\begin{array}{l}
\displaystyle{R_1= \omega_2, \qquad R_2= (\omega_3-\ld)\omega_1+\alpha_3;}
 \end{array}
\end{equation}

$2)$ множества, определяемого системой
\begin{equation}\label{eq3_16}
F_1  = 0,\quad F_2  = 0,
\end{equation}
где
\begin{equation}\label{eq3_17}
\begin{array}{l}
F_1=\left[\alpha_2 \omega_2+\omega_1 \left(\alpha_1+\omega_1^2+\omega_2^2 \right)\right] \left[-2\alpha_3\left(\alpha_1\omega_1+\alpha_2\omega_2\right)+\left(\alpha_1^2+\alpha_2^2\right)\omega_3\right]+ \\[3mm]
\phantom{F_1=}+\left\{-2\alpha_3^2\omega_1\left(\omega_1^2+\omega_2^2\right)-\alpha_2\alpha_3\omega_2\left(3\alpha_3+2\omega_1\omega_3\right)+
\alpha_1\left[-3\alpha_3^2\omega_1+ \right. \right.\\[3mm]
\phantom{F_1=}+\left.2\alpha_2\omega_2\left(\omega_1^2+ \omega_2^2\right)+\alpha_3\left(-\omega_1^2+\omega_2^2\right)\omega_3\right]+
\alpha_2^2 \left[\alpha_3\omega_3-\omega_1\left(\omega_1^2+\omega_2^2-\omega_3^2\right)\right]+
 \\[3mm]
\phantom{F_1=}+\left. \alpha_1^2\left[\omega_1^3+ \alpha_3\omega_3+\omega_1\left(\omega_2^2+\omega_3^2 \right)\right]
\right\}\ld+ \\[3mm]
\phantom{F_1=}+\left\{-\alpha_3^3-2\alpha_3^2\omega_1\omega_3+\left[\left(\alpha_1^2-\alpha_2^2\right)\omega_1+2\alpha_1\alpha_2\omega_2 \right]\omega_3+
\alpha_1\alpha_3\left(\omega_1^2+\omega_2^2+\omega_3^2\right)\right\}\ld^2+ \\[3mm]
\phantom{F_1=}+\alpha_1\alpha_3\omega_3\ld^3, \\[3mm]
F_2=
\left[\alpha _2 \omega _1+\omega _2 \left(-\alpha _1+\omega _1^2+\omega _2^2\right)\right] \left[-2 \alpha _3 \left(\alpha _1 \omega _1+\alpha _2 \omega _2\right)+\left(\alpha _1^2+\alpha _2^2\right) \omega _3\right]+\\[3mm]
\phantom{F_1=}+
\left\{\alpha _2 \left[\omega _1 \left(-\alpha _3^2+2 \alpha _1 \left(\omega _1^2+\omega _2^2\right)\right)+\alpha _3 \left(\omega _1^2-\omega _2^2\right) \omega _3\right]+\alpha _2^2 \omega _2 \left(\omega _1^2+\omega _2^2+\omega _3^2\right)- \right.\\[3mm]
\phantom{F_1=}-\left. \omega _2 \left[2 \alpha _3^2 \left(\omega _1^2+\omega _2^2\right)-\alpha _1 \alpha _3 \left(\alpha _3-2 \omega _1 \omega _3\right)+\alpha _1^2 \left(\omega _1^2+\omega _2^2-\omega _3^2\right)\right]\right\}\ld+\\[3mm]
\phantom{F_1=}+\left\{\alpha _2^2 \omega _2 \omega _3-\left(\alpha _1^2+2 \alpha _3^2\right) \omega _2 \omega _3+\alpha _2 \left[2 \alpha _1 \omega _1 \omega _3+\alpha _3 \left(\omega _1^2+\omega _2^2+\omega _3^2\right)\right]\right\}\ld^2+\\[3mm]
\phantom{F_1=}+\alpha _2 \alpha _3 \omega _3\ld^3.
\end{array}
\end{equation}
\end{theorem}

В инвариантности систем \eqref{eq3_14}, \eqref{eq3_16} можно убедиться дифференцированием их в силу системы \eqref{eq2_9}. Непосредственно проверяется, что точки \eqref{eq2_25} удовлетворяют системе \eqref{eq3_14}, а точки \eqref{eq2_28} --- системе \eqref{eq3_16}, что в силу размерностей соответствующих подмножеств вновь доказывает равенство \eqref{eq3_12}.

Из результатов работы \cite{KhND07} следует, как частный случай, что системы \eqref{eq3_14}, \eqref{eq3_16} получаются как уравнения множества критических точек на $\mP^5$ функции с неопределенными множителями Лагранжа
\begin{equation}\label{eq3_18}
  2 L^2 + (\tau-1) H+s K.
\end{equation}
Дифференциал этой функции инвариантен, поэтому $s$ и $\tau$, как функции фазовых переменных, становятся частными интегралами критических подсистем $\mi_j$:
\begin{eqnarray}
\mm:  & & s=-\frac{\ell}{\omega_1}, \qquad \tau = 1+2 \ell \omega_1,           \label{eq3_19}\\
\mno: & & s=\frac{\omega_3-2 \ell \alpha_3}{2(\omega_1^2+\omega_2^2+\ld \omega_3)\ld}, \qquad \tau = 2 \ld^2 s. \label{eq3_20}
\end{eqnarray}
Естественно, именно эти значения $s$ определяют соответствующие точки особых поверхностей $\wsi_j$ в записи \eqref{eq3_7}, \eqref{eq3_8}. Заметим, что, как видно из \eqref{eq3_19}, на множестве $\mm$ функцию \eqref{eq3_18} можно ``сократить'' на $L$ без потери критических точек. На $\mno$ это уже не так, поскольку имеется множество точек глобальной зависимости на $\mP^5$ функций $K,H$, на котором $L$ обращается в нуль, следовательно, неопределенный множитель при $dL$ в нулевой нетривиальной комбинации дифференциалов пропорционален $L$.

Отметим, что из результатов \cite{KhND07} следует также, что скобки Пуассона пар функций \eqref{eq3_15}, \eqref{eq3_17} имеют вид
\begin{eqnarray}
\{R_1,R_2\} & = & \frac{3}{2} s-(h-\frac{\ld^2}{2})\label{eq3_21},\\
\{F_1,F_2\} & = & \frac{\ld F_0}{ s \sqrt{2}}(1-8\ld^2s^3)\sqrt{2s^2-2 (h+\frac{\ld^2}{2})s+1},\label{eq3_22}
\end{eqnarray}
где
\begin{equation}\label{eq3_23}
F^2_0=\left(\omega_1^2+\omega_2^2+\ld \omega_3 \right)^3 \left[(\alpha_1\omega_1+\alpha_2\omega_2+\ld \alpha_3)^2+ (\alpha_2\omega_1-\alpha_1\omega_2)^2 \right].
\end{equation}
Обращение в нуль скобок \eqref{eq3_21} и \eqref{eq3_22} соответствует случаям вырождения симплектической структуры, индуцированной на двумерных многообразиях $\mi_j \cap \mPel$, являющихся фазовыми пространствами гамильтоновых систем с одной степенью свободы. Ниже будет ясна связь этого явления с типами критических точек. Забегая вперед, без строгих определений, можно сказать, что, поскольку невырожденные критические точки организованы в симплектические подмногообразия (см. соответствующее утверждение в \cite{BolFom}), точки вырождения индуцированной симплектической структуры на критическом множестве должны быть вырождены и как критические точки отображения момента.

\clearpage

\section{Относительные равновесия -- критические точки ранга 0}\label{sec4}
\subsection{Зависимость интегралов энергии и площадей}
Неподвижные точки уравнений Эйлера\,--\,Пуассона --- это проекции на $\mP^5$ движений тела, при которых траектория в $SO(3)$ совпадает с орбитой группы симметрий, то есть вращений вокруг вертикали с постоянной угловой скоростью. В теории понижения порядка по Раусу такие неподвижные точки называют {\it относительными равновесиями}. В теории С.Смейла \cite{Smale} эти точки являются критическими точками отображения {\it энергии\,--\,момента}, что в применении к динамике твердого тела означает зависимость интегралов энергии и площадей. Бифуркационные диаграммы отображения
\begin{equation}\label{eq4_1}
    L{\times}H: \mP^5 \to \bR^2
\end{equation}
называют диаграммами Смейла.

В свою очередь, на симплектическом листе \eqref{eq2_7} критические точки отображения \eqref{eq4_1} являются критическими точками ``приведенного'' гамильтониана
\begin{equation}\label{eq4_2}
    H_{\ell}= H|_{\mPel}: \mPel \to \bR.
\end{equation}
При необходимости подчеркнуть зависимость от $\ld$ пишем $\mPel(\ld)$ и $H_{\ell,\ld}$. Как следует из результатов \cite{RyabUdgu}, почти все такие точки невырождены в смысле Морса (ниже мы рассмотрим классификацию этих точек более подробно). Но невырожденная критическая точка гамильтониана будет критической для любого другого первого интеграла. Поэтому рассматриваемое множество совпадает с $\mct^0$. Итак, задача исследования критических точек отображения \eqref{eq4_1} совпадает с задачей исследования множества критических точек ранга $0$ отображения \eqref{eq3_6} в смысле определения~\ref{def1}.

Здесь и в дальнейшем для нахождения критических точек функций на $\mP^5$ без введения дополнительного неопределенного множителя, отвечающего ограничению \eqref{eq2_3}, удобно воспользоваться следующим утверждением.

\begin{lemma}\label{lem1}
Пусть $f$ -- гладкая функция на $\bR^6{(\bo,\ba)}$. Множество критических точек ограничения $f$ на подмногообразие $\mP^5$ определяется системой уравнений
\begin{equation}\label{eq4_3}
    \frac{\partial f}{\partial \bo}=0, \qquad \ba \times \frac{\partial f}{\partial \ba}=0.
\end{equation}
В комплексных переменных \eqref{eq3_1} эта система имеет вид
\begin{eqnarray}
& \ds \frac{\partial f}{\partial w_1}=\frac{\partial f}{\partial w_2}=\frac{\partial f}{\partial w_3}=0, \nonumber\\
&  \ds{2z \frac{\partial f}{\partial x_1}-x_1 \frac{\partial f}{\partial z}=0}, \quad \ds{2z \frac{\partial f}{\partial x_2}-x_2 \frac{\partial f}{\partial z}=0},\label{eq4_4} \\
& \ds{x_1 \frac{\partial f}{\partial x_2}-x_2 \frac{\partial f}{\partial x_1}=0}.\nonumber
\end{eqnarray}
\end{lemma}

Доказательство очевидно. Заметим, что три последних уравнения как в системе \eqref{eq4_3}, так и в системе \eqref{eq4_4}, линейно зависимы. Выбор независимой пары в случае общего положения диктуется соображениями удобства.

Запишем уравнения \eqref{eq4_3} для функции $f_H=H-2\sigma L$.
Получим
\begin{eqnarray}
  & & \omega_1 =\sigma \alpha_1, \quad \omega_2 =\sigma \alpha_2, \quad \omega_3 =\sigma \alpha_3, \nonumber\\
  & &  \alpha_3-[\alpha_1(\omega_3+\ld)-2 \alpha_3 \omega_1]\sigma =0,\quad \alpha_2-2(\alpha_1 \omega_2-\alpha_2 \omega_1)\sigma=0, \label{eq4_5}\\
  & & [2 \alpha_3 \omega_2- \alpha_2(\omega_3+\ld)]\sigma =0. \nonumber
\end{eqnarray}
Эта система, имеющая ранг 5, вместе с уравнением \eqref{eq2_3}, позволяет выразить фазовые переменные и неопределенный множитель $\sigma$ через какую-либо одну переменную, выбранную в качестве независимой. В работе \cite{RyabDis} представлена параметризация множества $\mct^0$ переменной~$\alpha_1$.
\begin{proposition}[П.Е.\,Рябов \cite{Ryab2,RyabDis}]\label{propos1}
Множество критических точек ранга $0$ описывается следующей системой уравнений
\begin{equation}\label{eq4_6}
    \begin{array}{lll}
      \omega_1 = \ds{\frac{x}{2d_1}(\ld-d_2)}, & \omega_2 = 0, & \omega_3=\ds{\frac{1}{2}(\ld-d_2)},\\
      \alpha_1 = x, & \alpha_2 = 0, & \alpha_3 = d_1,
    \end{array}
\end{equation}
где
\begin{equation}\notag
    d_1 = \pm \sqrt{1-x^2}, \qquad d_2 = \pm \sqrt{\ld^2-\ds{\frac{4}{x}(1-x^2)}},
\end{equation}
знаки $d_1,d_2$ произвольны, а параметр $x$ удовлетворяет условиям
\begin{equation}\notag
x \in \left\{ \begin{array}{ll}
[-1,0) \cup [c_\ld,1], & d_2 >0\\
(-1,0) \cup [c_\ld,1), & d_2 <0
\end{array} \right. , \quad \ds c_\ld=\frac{1}{8}\left(\sqrt{\ld^4+64}-\ld^2\right).
\end{equation}
При этом значения первых интегралов таковы
\begin{equation}\label{eq4_7}
    \ell = \ds{\frac{1}{d_1}[(3-x^2)\ld-(1+x^2)d_2]}, \quad \ds h=\frac{\ld(1-x^2)-d_2(1+3x^2)}{2x(\ld+d_2)},
\end{equation}
а неопределенный множитель
\begin{equation}\notag
    \sigma = \ds{\frac{1}{2 d_1}(\ld - d_2)}.
\end{equation}
\end{proposition}

Заметим, что система \eqref{eq4_5} не имеет решений в $\mP^5$ с $\alpha_1 = 0$, поэтому особенность $x=0$ неустранима.

Уравнения \eqref{eq4_7} -- это параметрические уравнения диаграммы Смейла. При исследовании бифуркационных диаграмм отображений \eqref{eq3_10} первое уравнение \eqref{eq4_7} позволяет решить задачу нахождения всех критических точек ранга $0$, попадающих на заданный уровень $P^4_\ell$.

Другую параметризацию множества $\mct^0$ предложил И.Н.\,Гашененко \cite{Gash4}, отметив, что такие точки принадлежат пересечению $\mm$ с объединением множеств $\mn,\mo$. Следовательно, удобно выбрать в качестве независимого параметра величину $\omega_3=r$ в равенствах \eqref{eq2_25}. При этом многочлен \eqref{eq2_26} должен иметь кратный корень, что и позволяет легко выразить все неизвестные через $r$. Естественно, те же выражения дает и система \eqref{eq4_5}. С учетом результатов \cite{Gash4,GashDis} приходим к следующему утверждению.
\begin{proposition}[И.Н.\,Гашененко \cite{Gash4}]\label{propos2}
Множество критических точек ранга $0$ описывается следующей системой уравнений
\begin{equation}\label{eq4_8}
    \begin{array}{lll}
      \omega_1 = \pm \sqrt{\pir{r}}, & \omega_2 = 0, & \omega_3=r ,\\
      \alpha_1 = -\pir{r-\ld}, & \alpha_2 = 0, & \alpha_3 = \mp (r-\ld)\sqrt{\pir{r}},
    \end{array}
\end{equation}
где
\begin{equation}\label{eq4_9}
    \D = \pm \sqrt{r^2(r-\ld)^2+4},
\end{equation}
параметр $r$ пробегает множество
\begin{equation}\label{eq4_10}
    r \in (-\infty,0] \cup [0,\ld) \cup (\ld,+\infty),
\end{equation}
знак $\D$ совпадает со знаком $r (r-\ld)$ при $r \ne 0$ и произволен при $r=0$.
При этом значения первых интегралов таковы
\begin{equation}\label{eq4_11}
\begin{array}{l}
    \ell = \mp \ds{\frac{1}{2}[\ld(r-\ld)+\D]}\sqrt{\pir{r}},\quad
    h= -\ds{\frac{1}{2}r(r-\ld)+\frac{2r-\ld}{2(r-\ld)}\D},
\end{array}
\end{equation}
а неопределенный множитель
\begin{equation}\notag
    \sigma = \mp \ds{\frac{1}{\sqrt{2}}\sqrt{\frac{r}{r-\ld}\bigl[r(r-\ld)+\D\bigr] } }.
\end{equation}
Знаки $\omega_1, \alpha_3, \ell, \sigma$ согласованы $($все верхние или все нижние$)$.
\end{proposition}

Сравнивая с уравнениями \eqref{eq2_25} видим, что точки ранга 0 получаются в подсистеме $\mm$, если положить
\begin{equation}\label{eq4_12}
  p^2=\pir{r}.
\end{equation}

\begin{remark}\label{rem3}
Для того чтобы убедиться непосредственно, что во всех найденных точках равен нулю и дифференциал $d(K|_{P^4_\ell})$, достаточно записать уравнения \eqref{eq4_3} для функции
\begin{equation}\notag
    f_K=K-4 \gamma L, \qquad \gamma=\mp \ds{\frac{\ld}{r-\ld}}\sqrt{\pir{r}}.
\end{equation}
Для диаграмм Смейла значение интеграла $K$ несущественно, однако ниже нам понадобится также и это значение в точках \eqref{eq4_8}. Согласно {\rm \cite{Gash4}} его можно представить в виде
\begin{equation}\label{eq4_13}
    k  = \ds{\frac{\ld}{4(r-\ld)^2} \bigl[r(r-\ld)-\D\bigr]\bigl[r(r-\ld)(4 r-3\ld)-\ld \D\bigr]}.
\end{equation}
Вычислим также в точках \eqref{eq4_8} значения частных интегралов $s$ в системах $\mi_j$, то есть значения параметра $s$ на поверхностях $\wsi_j$ в точках ранга $0$. Используя \eqref{eq3_19}, \eqref{eq3_20}, находим
\begin{eqnarray}
      \mm: & & s = \ds{\frac{1}{2}[\ld(r-\ld)+\D]}, \label{eq4_14}\\
      \mno: &  & s= \ds{\frac{r-\ld}{4\ld}[r(r-\ld)-\D].} \label{eq4_15}
\end{eqnarray}
\end{remark}

\begin{figure}[!ht]
\centering
\includegraphics[width=80mm,keepaspectratio]{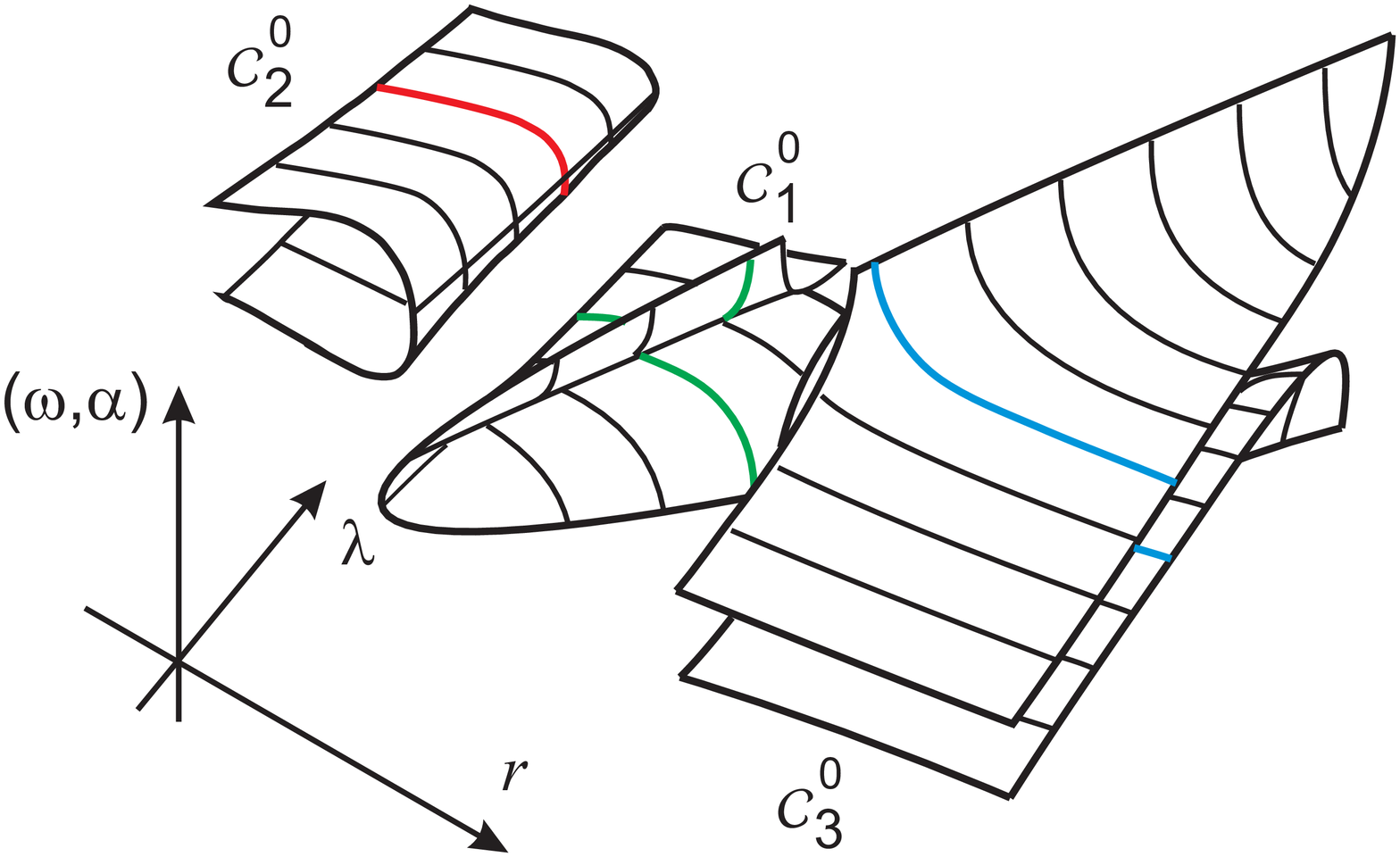}\\
\caption{Поверхности относительных равновесий.}\label{fig_kh_fig01}
\end{figure}

Из \eqref{eq4_8}, \eqref{eq4_10} сразу же следует, что множество $\mct^0$ имеет ровно четыре связных компоненты, гомеоморфных $\bR$. В соответствии с областью изменения $r$ введем обозначения для подмножеств в $\mct^0$, определяемых формулами \eqref{eq4_8}:
\begin{equation}\label{eq4_16}
    \begin{array}{llll}
      \mct^0_1: & r \in [0,\ld), & \D<0, & \ds{\lim_{r\to +0}} \D = -2, \\
      \mct^0_2: & r \in (-\infty,0] , & \D > 0, & \ds{\lim_{r\to -0}} \D = 2, \\
      \mct^0_3: & r \in (\ld, +\infty), & \D > 0.&
    \end{array}
\end{equation}
Первые два множества связны, последнее состоит из двух компонент, отличающихся знаком $\omega_1$. В $\mct^0_1$ и $\mct^0_2$ каждому значению $r\ne 0$ отвечает ровно две точки, в $\mct^0_1$ значению $r=0$ отвечает точка $\bo=0$, $\ba=\{1,0,0\}$ нижнего положения равновесия тела (для дальнейшего обозначим ее в соответствии со стремлением $r$ к нулю справа через $c_+$), а в $\mct^0_2$ нулевое значение $r$ приводит к точке $\bo=0$, $\ba=\{-1,0,0\}$ верхнего положения равновесия (обозначим ее через $c_-$). На множествах \eqref{eq4_16} определена очевидная симметрия
\begin{equation}\label{eq4_17}
    \bio: (\omega_1,\alpha_3) \mapsto (-\omega_1,-\alpha_3),
\end{equation}
которая меняет знак постоянной площадей $\ell$, связные множества $\mct^0_1$, $\mct^0_2$ переводит в себя, а в множестве $\mct^0_3$ меняет местами связные компоненты. Устройство семейства множеств $\mct^0(\ld)$ проиллюстрировано на рис.~\ref{fig_kh_fig01}.

\subsection{Классификация критических точек ранга $0$}

В этом разделе мы приводим результаты, касающиеся аналитической классификации типов критических точек ранга 0 в соответствии с работой \cite{KhMTT42}.

Напомним некоторые понятия, связанные с понятием типа критической точки в интегрируемой системе \cite{BolFom}.

Пусть $M$ --- симплектическое многообразие. Любой гладкой функции $F:M\to \bR$ сопоставляется гамильтоново векторное поле на $M$, обозначаемое $\sgrad F$.
Скобки Пуассона, порожденные симплектической структурой, обозначаем через $\{\cdot,\cdot\}$, так что дифференциальные уравнения системы $\sgrad F$ имеют вид
\begin{equation}\notag
\dot x =\{F,x\}
\end{equation}
(именно такой порядок аргументов в скобке нужен для того, чтобы из определений \eqref{eq2_6}, \eqref{eq2_4} следовали уравнения \eqref{eq2_1}).

Пусть $\dim M =2n$ и гамильтонова система $\sgrad H$ имеет $n$ функционально независимых первых интегралов
\begin{equation}\notag
F_1,\ldots,F_n
\end{equation}
в инволюции ($\{F_i,F_j\}\equiv 0$). Точка $\x \in M$ называется критической ранга $k<n$, если ранг системы векторов $\sgrad F_i$ в точке $\x$ равен $k$. Введем понятие типа критической точки, следуя \cite{BolFom}.

Рассмотрим критическую точку $\x$ ранга $n-m$ ($m>0$). Пусть далее в этом разделе индекс $i$ пробегает множество $1,\ldots, m$. Линейной заменой c постоянными коэффициентами системы функций $F_1,\ldots,F_n$ можно добиться того, чтобы точка $\x$ была критической для каждой из функций $F_i$ и регулярной для всех остальных. Тогда $\sgrad F_i(\x)=0$ и линеаризация этого поля в точке $\x$ есть симплектический оператор $\A_{i}:T_\x M \to T_\x M$. Линейная оболочка $\mA(\x)$ таких операторов есть подалгебра в алгебре всех симплектических операторов на $T_\x M$.
\begin{defin}[\cite{BolFom}]\label{def3}
Точка $\x$ называется невырожденной критической точкой ранга $n-m$, если $\mA(\x)$ есть подалгебра Картана, что равносильно следующим требованиям:

$(1)$ симплектические операторы $\A_{i}$ линейно независимы $(\dim \mA (\x) = m)$;

$(2)$ существует оператор $\A \in \mA(\x)$, у которого все собственные числа различны.
\end{defin}

Напомним, что собственные числа симплектического оператора разбиваются на группы: пары чисто мнимых $\pm \ri a$, пары вещественных $\pm b$ и четверки комплексных $\pm b \pm \ri a$ ($ ab \ne 0$). Для оператора $\A_{i}$ в проекции на корневые подпространства таких групп поле $\sgrad F_i$ имеет соответственно центр, седло или фокус. Оператор $\A \in \mA(\x)$ с  различными собственными числами называется регулярным элементом. Выбрав в невырожденной точке регулярный элемент, обозначим через $m_1,m_2,m_3$ соответственно количество центров, седел и фокусов $(m=m_1+m_2+2m_3)$. От выбора регулярного элемента эти целые неотрицательные числа не зависят.
\begin{defin}[\cite{BolFom}]\label{def4}
Четверка $(n-m,m_1,m_2,m_3)$ называется типом невырожденной критической точки $\x$.
\end{defin}

Знание типа невырожденной критической точки в значительной мере (но не полностью) определяет слоение Лиувилля (слоение совместных уровней интегралов $F_1,\ldots,F_n$) в окрестности этой точки. Для построения грубого топологического описания слоения Лиувилля (например, в виде круговой молекулы без меток в системе с двумя степенями свободы) достаточно знать еще количество связных компонент регулярных и критических интегральных поверхностей. Для критических уровней малой сложности (то есть когда на одну связную компоненту попадает мало критических орбит), например, для критических точек ранга 0 с одной точкой на слое, этого оказывается достаточно и для нахождения тонкого топологического инварианта (меченой круговой молекулы) в силу наличия описания всех имеющихся возможностей \cite{BolFom}. Ниже мы представим такое описание в рассматриваемой задаче.

\begin{defin}\label{def5}
Критическую точку $\x$ назовем \emph{сильно} вырожденной, если для нее нарушается условие $(1)$ определения~{\rm \ref{def3}}, и \emph{слабо} вырожденной в противном случае.
\end{defin}

Заметим, что эта терминология не является общепринятой и используется здесь для удобства. Следующее определение является дискуссионным и здесь не используется.

\begin{defin}\label{def6}
Тип вырожденной критической точки $\x$ -- это четверка $(k,m_1,m_2,m_3)$, где $k=n-m$ по-прежнему есть ранг, а $m_1, m_2, m_3$ --- \emph{максимальные} количества центров, седел и фокусов, отвечающих \emph{различным} группам собственных чисел операторов из $\mA(\x)$.
\end{defin}

Для уравнений Эйлера\,--\,Пуассона 2-форма, индуцированная на $\mP^5$ симплектической структурой многообразия $TSO(3)$, вырождена. Введенные понятия необходимо рассматривать с точки зрения систем на $\mPel$. Явный переход к этим системам делает вычисления необозримыми. Однако в этом нет необходимости. Корректно определена скобка Пуассона \eqref{eq2_6}, хотя и вырожденная, поэтому сопоставленное функции $F$ поле $\sgrad F$ определено уравнениями
\begin{equation}\label{eq4_18}
\dot {\bm}=\bm \times \ds{\frac{\partial F}{\partial \bm}} + \ba \times \ds{\frac{\partial F}{\partial \ba}},\quad
\dot {\ba}= \ba \times \ds{\frac{\partial F}{\partial \bm}}.
\end{equation}
\begin{remark}\label{rem4}
Если $F$ есть функция Казимира, то есть тождественно зависима с $L$ и $\ba^2$, то правые части \eqref{eq4_18} --- тождественный ноль. В связи с этим при вычислении линеаризации полей $\sgrad F$ и собственных чисел соответствующих операторов в $\bR^6$ нет необходимости учитывать неопределенные множители Лагранжа для этих функций --- необходимо лишь учесть определение~$\ref{def1}$ и отбросить два нулевых собственных числа, которые здесь обязательно существуют.\end{remark}

Физическая модель гиростата -- это система с четырьмя степенями свободы (тело плюс ротор), а уравнения Эйлера\,--\,Пуассона получены понижением порядка в такой системе, для которой $\ld$ есть константа циклического интеграла. Поэтому все рассуждения о постоянстве каких-либо свойств в пространстве интегральных или иных параметров естественно рассматривать в расширенном пространстве этих параметров, включающих ось $\bR=\bR(\ld)$. В связи с этим будем использовать следующее обозначение. Пусть $A$ --- какое-либо множество, а $B(\ld)$ --- семейство его подмножеств, зависящих от параметра $\ld$. Обозначим
\begin{equation}\label{eq4_19}
    \mwide{A} = A \times \bR, \qquad \mwide{B} = \bigcup_{\ld} B(\ld)\times \{\ld\} \subset \mwide{A}.
\end{equation}

В соответствии с этим рассмотрим расширенное множество $\mwide{\mct^0} \subset \mwide{\mP^5}=\mP^5 \times \bR$. Напомним, что согласно \eqref{eq2_8} мы считаем $\ld>0$, а случай $\ld=0$ рассматриваем лишь как предельную возможность там, где это явно оговорено. По предложению~\ref{propos2} множество $\mwide{\mct^0}$ непрерывно дважды накрывает область
\begin{equation}\notag
\mD^0=\{(r,\ld)\in \bR^2: \ld >0, r \ne \ld \}.
\end{equation}

\begin{remark}\label{rem5}
Для дальнейшего условимся образы множеств $\mwide{\mct^0_i}$ в различных пространствах параметров $($постоянных общих и частных интегралов, физического параметра $\ld)$ обозначать через $\delta_i$ $(i=1,2,3)$. В частности, как подмножества в $\mD^0$ они имеют вид
\begin{equation}\label{eq4_20}
    \begin{array}{ll}
      \delta_1: \{(r,\ld): 0 \ls r < \ld, \; \ld >0 \}, \\
      \delta_2: \{(r,\ld): r \ls 0, \; \ld >0 \}, \\
      \delta_3: \{(r,\ld): r > \ld, \; \ld >0 \}.
    \end{array}
\end{equation}
Подмножества этих множеств, полученные в результате дальнейшей детализации, будут снабжаться двойными индексами.
\end{remark}

\begin{defin}\label{def7}
Будем говорить, что точки $\x_1,\x_2 \in \mwide{\mct^0}$ принадлежат к одному классу, если существует непрерывный путь в $\mwide{\mct^0}$, соединяющий эти точки или точку $\x_1$ с точкой $\bio(\x_2)$, вдоль которого не меняется тип критических точек.
\end{defin}

Пусть $(r,\ld) \in \mD^0$. Обозначим через $\x_\pm(r,\ld)$ точку \eqref{eq4_8} при $r\ne 0$ и, в соответствии с предложением~\ref{propos2}, при выбранном знаке величины \eqref{eq4_9} $\sgn \D=\sgn[r(r-\ld)]$. Согласно принятым ранее обозначениям имеем
\begin{equation}\notag
\lim _{r \to +0} \x_\pm(r,\ld) = c_+ \in \mct^0_1, \qquad \lim _{r \to -0} \x_\pm(r,\ld) = c_- \in \mct^0_2.
\end{equation}

\begin{defin}\label{def8}
Точку $(r,\ld)\in \mD^0$ назовем разделяющей, если в любой ее окрестности найдутся образы точек из $\mwide{\mct^0}$ разных классов.
\end{defin}

Разделяющим является луч запрещенных точек $r=\ld, \ld>0$. Обозначим его через $\overline{\vpi}$.
Поскольку $\x_1\in \mct^0_1$ и $\x_2\in \mct^0_2$ не могут принадлежать одному классу, точки вида $(0,\ld)$ всегда являются разделяющими. Обозначим полуось $r=0,\ld>0$ через $\vpi_0$.
При $r\ne 0$ типы критических точек $\x_\pm(r,\ld)$ всегда одинаковы, поэтому точка $(r,\ld)$ является разделяющей тогда и только тогда, когда обе критические точки $\x_\pm(r,\ld)$ вырождены.


Пусть $\A_H,\A_K$ --- симплектические операторы линеаризации полей $\sgrad H,\sgrad K$ в точке \eqref{eq4_8}.
\begin{proposition}\label{propos3}
Множество сильно вырожденных критических точек $\x_\pm(r,\ld)$ соответствует кривой в области $\delta_2$
\begin{equation}\label{eq4_21}
    \vpi_{21}: \quad r=\ld-\frac{1}{\ld^{1/3}}, \quad 0< \ld \ls 1.
\end{equation}
\end{proposition}

\begin{proof}
В точках сильного вырождения составим комбинацию
\begin{equation}\label{eq4_22}
    \Linc = \nu_1 \A_H + \nu_2 \A_K = 0.
\end{equation}
Располагая переменные и, соответственно, элементы матриц в порядке $\omega_1,\omega_2,\omega_3,\alpha_1, \alpha_2,\alpha_3$, имеем
\begin{equation}\label{eq4_23}
    \Linc_{12} = \ds{\frac{1}{2}}(r-\ld)\left[ \nu_1-2 \nu_2 \ld (Q^2+\ld+r)\right],
\end{equation}
где обозначено
\begin{equation*}
    Q=\sqrt{\pir{1}}.
\end{equation*}
Очевидно, $Q \ne 0$. Из \eqref{eq4_22}, \eqref{eq4_23} получим
\begin{equation}\notag
    \nu_1 = 2 \nu_2 \ld (Q^2+\ld+r), \quad \nu_2 \ne 0.
\end{equation}
Подставив в \eqref{eq4_22}, получим
\begin{equation*}
    \Linc = 2 \nu_2 (\ld-r)Q (Q^2+\ld)\left(
\begin{array}{cccccc}
0 & 0 & 0 & 0 & 0 & 0 \\
0 & 0 & 0 & 0 & 0 & 0 \\
0 & -2\sqrt{r} & 0 & 0 & \ds{\frac{\ld+r}{Q(\ld-r)}} & 0 \\
0 & -2r\sqrt{r} & 0 & 0 & \ds{\frac{r(\ld+r)}{Q(\ld-r)}} & 0 \\
2r\sqrt{r} & 0 & -\ld Q & -\ds{\frac{r(\ld+r)}{Q(\ld-r)}} & 0 & \ds{\frac{2\ld \sqrt{r}}{\ld-r}} \\
0 & 2 r Q & 0 & 0 & -\ds{\frac{\sqrt{r}(\ld+r)}{\ld-r}} & 0
\end{array}
    \right) =0.
\end{equation*}
Матрица здесь ненулевая, поэтому $Q^2+\ld=0$, что равносильно уравнению
\begin{equation}\label{eq4_24}
    1+\ld (r-\ld)^3 =0
\end{equation}
с условием $Q^2<0$. Но $\sgn Q^2=\sgn [(r-\ld)\D]$ совпадает с $\sgn r$ (см. предложение~\ref{propos2}). Поэтому из решений \eqref{eq4_24} нужно взять только лежащие в $\delta_2$, что и дает кривую \eqref{eq4_21}.
\end{proof}

\begin{proposition}\label{propos4}
Пусть точка $(r,\ld) \in \mD^0$ не лежит на кривой \eqref{eq4_21} и не удовлетворяет ни одному из уравнений
\begin{eqnarray}
  & & r+\ld=0, \label{eq4_25}\\
  & & (2r-\ld)(r-\ld)+\D =0, \label{eq4_26}\\
  & & (2r-\ld)(r-\ld)-\D =0, \label{eq4_27}\\
  & & (2r-\ld)(r-\ld)r + \ld \D =0. \label{eq4_28}
\end{eqnarray}
Тогда критические точки $\x_\pm(r,\ld)$ невырождены.
\end{proposition}
\begin{proof}
Характеристический многочлен оператора $\A_H$ в точке $\x_\pm(r,\ld)$, сокращенный на $\mu^2$ в соответствии с замечанием~\ref{rem4}, имеет вид
\begin{equation}\label{eq4_29}
    \chi_H(\mu)= \mu^4 - 2 a \mu^2+ b,
\end{equation}
где
\begin{equation*}
    \begin{array}{l}
       a=\ds{\frac{1}{8(r-\ld)}}[-(3r-\ld)(2r-\ld)(r-\ld)+(r-3\ld)\D], \\
       b=\ds{\frac{1}{8(r-\ld)}}[(r-\ld)^3(4r-\ld)r-4\ld-(2r-\ld)(r-\ld)^2\D].
     \end{array}
\end{equation*}
Дискриминант многочлена \eqref{eq4_29}
\begin{equation*}
    a^2-b =\ds{\frac{(r+\ld)^2}{64(r-\ld)^2}}[(2r-\ld)(r-\ld)+\D]^2
\end{equation*}
обращается в нуль только при условиях \eqref{eq4_25} и \eqref{eq4_26}. Корни $\chi_H$ по $\mu^2$ находятся явно
\begin{eqnarray}
  \mu_1^2 & = & -\ds{\frac{1}{4}}[(2r-\ld)(r-\ld)-\D], \label{eq4_30}\\
  \mu_2^2 & = & -\ds{\frac{1}{2(r-\ld)}}[(2r-\ld)(r-\ld)r+\ld\D]. \label{eq4_31}
\end{eqnarray}
Эти величины, фактически и определяющие типы точек ранга 0, впервые вычислены в работе \cite{RyabUdgu}. Таким образом, за пределами множества, определенного уравнениями \eqref{eq4_25}--\eqref{eq4_28}, все корни $\chi_H$ различны. При этом вне кривой \eqref{eq4_21} алгебра, порожденная операторами $\A_H,\A_K$, двумерна. Предложение доказано.
\end{proof}

\begin{proposition}\label{propos5}
При условии \eqref{eq4_26} критические точки невырождены.
\end{proposition}
\begin{proof}
Пусть выполнено условие \eqref{eq4_26}. Если допустить, что при этом $\D>0$, то $r \in (-\infty, 0] \cup (\ld, +\infty)$. Но на этом промежутке $(2r-\ld)(r-\ld)>0$, что противоречит \eqref{eq4_26}. Следовательно, должно быть $\D<0$, $r \in [0,\ld)$ и
\begin{equation*}
    \frac{1}{2} |\D|=(r-\frac{\ld}{2})(r-\ld).
\end{equation*}
Следовательно, $r\in [0,\ld/2]$. Возводя в квадрат, получим уравнение
\begin{equation}\label{eq4_32}
    (r-\ld)^3(3r-\ld)-4=0.
\end{equation}
Его решение на нужном интервале представим в параметрической форме, для чего введем переменную $x$, полагая \cite{mtt40}
\begin{equation}\label{eq4_33}
    x=\ld-r.
\end{equation}
Из \eqref{eq4_32}, \eqref{eq4_33} найдем $\D=-(x^4+4)/(2 x^2)$ и
\begin{equation}\label{eq4_34}
    r=\ds{\frac{x^4-4}{2x^3}}, \qquad \ld=\ds{\frac{3x^4-4}{2x^3}}.
\end{equation}
При этом условие $r\in [0,\ld)$ выполнено, если
\begin{equation}\label{eq4_35}
    x \gs \sqrt{2}.
\end{equation}
Характеристический многочлен оператора $\A_H$ в подстановке \eqref{eq4_34} принимает вид
$$
\chi_H(\mu)=(\mu^2+\ds{\frac{\sqrt{x^4+4}}{2x}})^2,
$$
поэтому задачу о невырожденности он не решает. Вычислим, однако, характеристический многочлен оператора $\A_K$:
$$
\chi_K(\mu)=\left[\mu^2 + \ds{\frac{(x^4-4)^2 (4 + x^4)}{x^{14}}}\right] \left[\mu^2 + \ds{\frac{(x^4+4) (3 x^8 - 7 x^4 + 4)^2}{x^{14}}}\right].
$$
При условии \eqref{eq4_35} все его корни различны, поэтому он представляет собой искомый регулярный элемент алгебры.
\end{proof}

\begin{proposition}\label{propos6}
На кривой
\begin{equation}\label{eq4_36}
    \vpi_{22}: \quad r=-\ld, \quad \ld >0,
\end{equation}
то есть при условии \eqref{eq4_25}, все критические точки вырождены.
\end{proposition}
\begin{proof}
Вычислим характеристический многочлен комбинации $\Linc=\nu_1 \A_H + \nu_2 \A_K$ в точках \eqref{eq4_8} при условии \eqref{eq4_25}. Получим
\begin{equation*}
    \chi(\mu)=\left[\mu^2+\ds{\frac{(Z^2-2)(\nu_1 Z+\nu_2)}{2Z^3}}\right]^2, \qquad Z=\ld^2+\sqrt{\ld^4+1}.
\end{equation*}
Поэтому в линейной оболочке операторов $\A_H ,\A_K$ регулярного элемента нет.\end{proof}
Заметим, что при $\nu_1 Z+\nu_2=0$ все собственные числа оператора $\Linc$ равны нулю, однако, он остается ненулевым, за исключением значения $\ld=1/2^{3/4}$, которое соответствует точке пересечения $\vpi_{22}$ с кривой \eqref{eq4_21} сильного вырождения.

\begin{proposition}\label{propos7}
Условие \eqref{eq4_27} реализуется в области $\mD^0$ на следующих кривых
\begin{eqnarray}
& &    \vpi_{23}: \quad r=\ds{\frac{x^4-4}{2x^3}}, \qquad \ld=\ds{\frac{3x^4-4}{2x^3}}, \qquad x \in (\sqrt[4]{4/3},\sqrt{2}],\label{eq4_37}\\
& &    \vpi_{31}: \quad r=\ds{\frac{x^4-4}{2x^3}}, \qquad \ld=\ds{\frac{3x^4-4}{2x^3}}, \qquad x \in (-\sqrt[4]{4/3},0) \label{eq4_38}.
\end{eqnarray}
Все соответствующие критические точки вырождены.
\end{proposition}
\begin{proof}
Следствием уравнения \eqref{eq4_27} является уравнение \eqref{eq4_32}, решения которого представлены в виде \eqref{eq4_34}. Если допустить, что $\D<0$, то тогда $r\in [0,\ld)$, следовательно, $r\in [0,\ld/3)$ и $(2r-\ld)(r-\ld)>0$, то есть таких решений \eqref{eq4_27} не имеет. Если же $\D>0$, то в областях $\delta_2,\delta_3$ имеем соответствующие пределы изменения $x$ в \eqref{eq4_37}, \eqref{eq4_38}. В таких точках характеристический многочлен комбинации $\nu_1 \A_H + \nu_2 \A_K$ имеет вид
\begin{equation}\label{eq4_39}
    \chi(\mu)= \mu^4 -\ds{\frac{(x^8+2x^4-8) [-2 \nu_1 x^2 + \nu_2 (x^4-4 )]^2}{8 x^{10}}}\mu^2,
\end{equation}
и всегда имеет два нулевых корня. Поэтому регулярного элемента в алгебре операторов нет.
\end{proof}

\begin{proposition}\label{propos8}
Условие \eqref{eq4_28} реализуется в области $\mD^0$ на кривой
\begin{equation}\label{eq4_40}
\vpi_{24}: \quad r=\ds{\frac{1}{2}}\left(\ld-\sqrt{\ld^2+4\ld^{2/3}}\right), \qquad \ld> 0,
\end{equation}
Все соответствующие критические точки вырождены.
\end{proposition}
\begin{proof}
Следствием уравнения \eqref{eq4_28} является уравнение
\begin{equation}\label{eq4_41}
    r(r-\ld)=\ld^{2/3}.
\end{equation}
Очевидно, оно имеет ровно два решения, по одному в областях $\delta_3$ ($r>\ld$) и $\delta_2$ ($r<0$). При этом следует выбирать $\D>0$. Но в $\delta_3$ имеем $(2r-\ld)(r-\ld)r>0$, поэтому \eqref{eq4_28} не выполняется. Следовательно, решением является только нижний корень уравнения \eqref{eq4_41}, что и приводит к точкам кривой \eqref{eq4_40}.
В этих точках характеристический многочлен комбинации $\nu_1 \A_H + \nu_2 \A_K$ имеет вид
\begin{equation}\label{eq4_42}
\chi(\mu)= \mu^4 + \ds{\frac{(Z^2-8) (4 + Z^2) \{
    \nu_2 [Z^2 (Z^2-8)^2-64]-4 \nu_1 Z^3\}^2}{512 Z^7}} \mu^2, \quad  Z=\ld^{2/3}+\sqrt{4+\ld^{4/3}}.
\end{equation}
и всегда имеет два нулевых корня. Поэтому регулярного элемента в алгебре операторов нет.
\end{proof}

Область $\mD^0$ с разделяющими кривыми указана на рис.~\ref{fig_RazdCrit0}. Здесь же введены обозначения для классов невырожденных точек ранга $0$. Как видим, в множестве $\mwide{\mct^0_1}$ всего один класс, и он обозначен, как и подобласть в $\mD^0$, через $\delta_1$. Он включает точки, соответствующие значениям с $r=0$, так что этот класс состоит из одной связной компоненты. В множестве $\mwide{\mct^0_3}$ два класса, они обозначены через $\delta_{31},\delta_{32}$. Они не включают точек с $r=0$, поэтому каждый такой класс состоит из двух связных компонент. В множестве $\mwide{\mct^0_2}$ восемь классов $\delta_{21},\ldots, \delta_{28}$. Три из них ($\delta_{21}, \delta_{26},\delta_{27}$) содержат точки с $r=0$ и имеют поэтому одну связную компоненту, остальные состоят из двух связных компонент. На этом рисунке также введена кривая $\ell_0$, порожденная ситуацией, специфической только для гиростата, -- наличие равномерных вращений вокруг вертикали ($r \ne 0$), на которых постоянная площадей все же равна нулю. Из \eqref{eq4_11} найдем, что тогда $\ld(r-\ld)+d=0$. В соответствии с определением знака $d$, это возможно лишь на $\delta_2$, поэтому кривая $\ell_0$ определяется условиями
\begin{equation}\label{eq4_43}
    \ell_0: \quad (r-\ld)^3(r+\ld)-4=0, \qquad r \ls 0.
\end{equation}
Здесь уместно еще раз напомнить, что по договоренности $\ld \gs 0$. При пересечении кривой $\ell_0$ тип критических точек не меняется, но, как будет показано ниже, она вызывает перестройку диаграммы Смейла, а также меняет топологию совместного уровня первых интегралов в целом. Как обозначено на рисунке, подобласти, на которые кривая $\ell_0$ разбивает классы $\delta_{27}, \delta_{28}$, будем снабжать штрихами.

\begin{figure}[ht]
\centering
\includegraphics[width=0.7\textwidth,keepaspectratio]{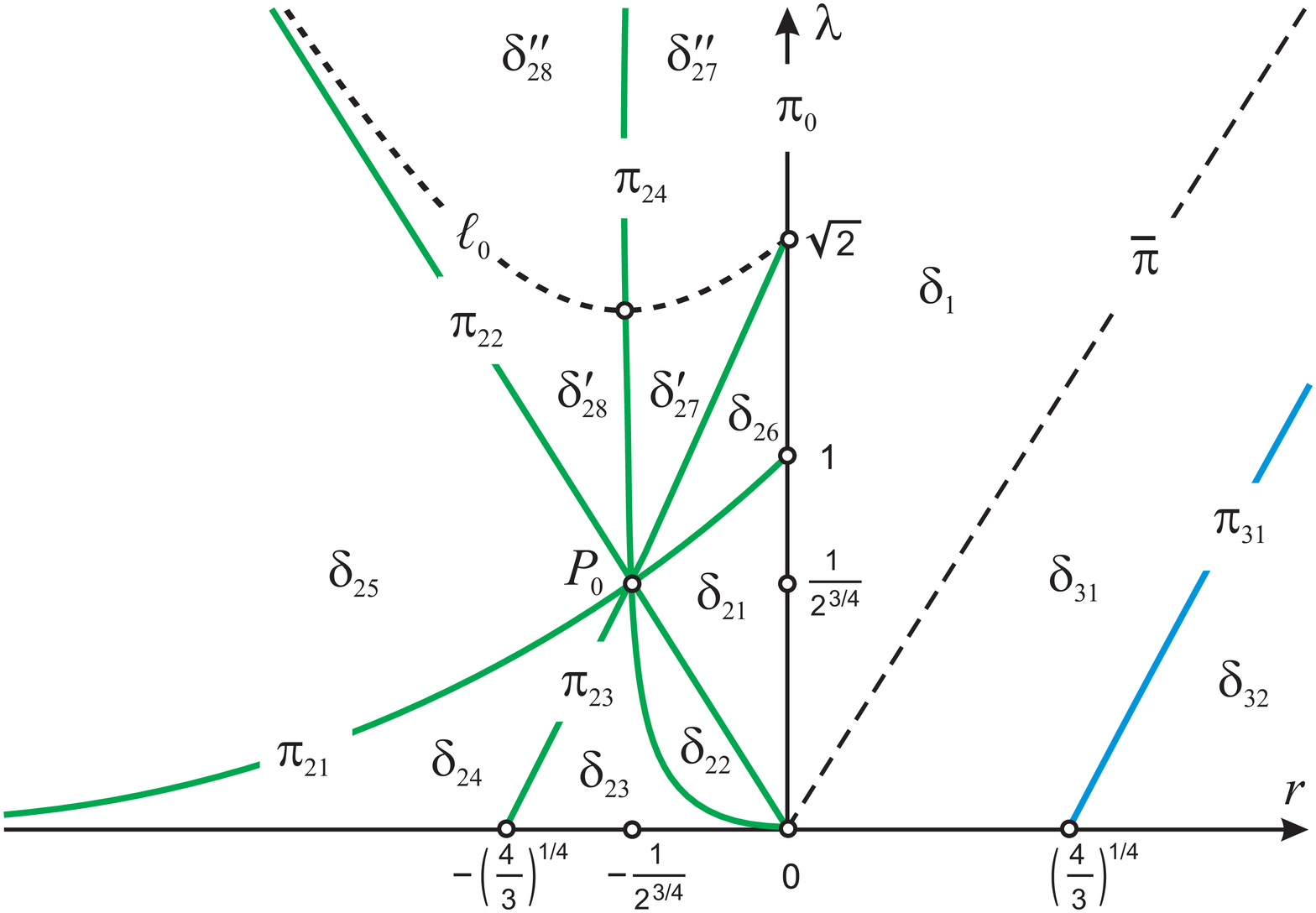}
\caption{Разделяющее множество и классы критических точек ранга $0$}\label{fig_RazdCrit0}
\end{figure}

Поскольку все случаи обращения в ноль величин \eqref{eq4_30}, \eqref{eq4_31} аналитически установлены, то их знаки в порожденных подобластях $(r,\ld)$-плоскости приводят к следующей классификации точек ранга $0$.

\begin{theorem}\label{th5}
В расширенном фазовом пространстве $\mwide{\mP^5}$ критические точки ранга $0$ случая Ко\-ва\-лев\-ской\,--\,Яхья  формируют четыре многообразия, диффеоморфных $\bR^2$.
Вырожденным критическим точкам соответствуют пять разделяющих кривых на плоскости параметров~-- кривые $\vpi_{21} - \vpi_{24}, \vpi_{31}$. Невырожденные критические точки в соответствии с введенным отношением эквивалентности $($наличие непрерывного пути с постоянным типом критической точки из первой точки во вторую или в симметричную второй относительно заданного ото\-бра\-же\-ния-ин\-вер\-сии$)$ разбиты на $11$ классов. Четыре класса, содержащие в себе одно из положений равновесия тела, связны. Остальные, содержащие только образы равномерных вращений $($относительные равновесия$)$, состоят из двух связных компонент, симметричных друг другу относительно ото\-бра\-же\-ния-ин\-вер\-сии. В соответствии с обозначениями классов на рис.~$\ref{fig_RazdCrit0}$ критические точки имеют следующий тип $($указывается в порядке следования пар корней $\mu_1^2,\mu_2^2):$

$\delta_{21}, \delta_{22}, \delta_{26}$ --- ``седло-седло''$;$

$\delta_{23}, \delta_{31}$ --- ``седло-центр''$;$

$\delta_{27}$ --- ``центр-седло''$;$

$\delta_{24}, \delta_{25}, \delta_{28}, \delta_{1}, \delta_{32}$ --- ``центр-центр''.
\end{theorem}

Проведем сравнение классификации относительных равновесий по типам и классификации их в решении И.Н.\,Гашененко \eqref{eq2_31}, \eqref{eq2_32}, \eqref{eq2_35} по параметрам \eqref{eq2_33}.
На рис.~\ref{fig_RazdCritL} нанесены знаки троек чисел \eqref{eq2_33}, а на разделяющих кривых указаны и нулевые значения параметров. Оказалось, что кривая $\vpi_{24}$ на знаки троек не влияет, на кривой $\vpi_{21}$ обращается в ноль $L_3$, но в примыкающих к ней областях знаки троек одинаковы. На кривой $\vpi_{22}$ обращаются в ноль два параметра $L_2$ и $L_3$. В итоге имеем следующее соответствие классам \eqref{eq2_36}:
\begin{equation}\label{eq4_44}
\begin{array}{ll}
\mathrm{(I)}& \delta_{24},\delta_{25},\delta_{32};\\
\mathrm{(II)}& \delta'_{27},\delta'_{28};\\
\mathrm{(III)}& \delta_{22},\delta_{23},\delta_{31};\\
\mathrm{(IV)}& \delta_{21},\delta_{26};\\
\mathrm{(V)}& \delta_{1},\delta''_{27},\delta''_{28};\\
\mathrm{(VI)}& \ell_0;\\
\mathrm{(VII)}& \vpi_{22},\vpi_{23},\vpi_{31};\\
\mathrm{(VIII)}& \vpi_{21},\vpi_{22}.
\end{array}
\end{equation}

\begin{figure}[ht]
\centering
\includegraphics[width=0.7\textwidth,keepaspectratio]{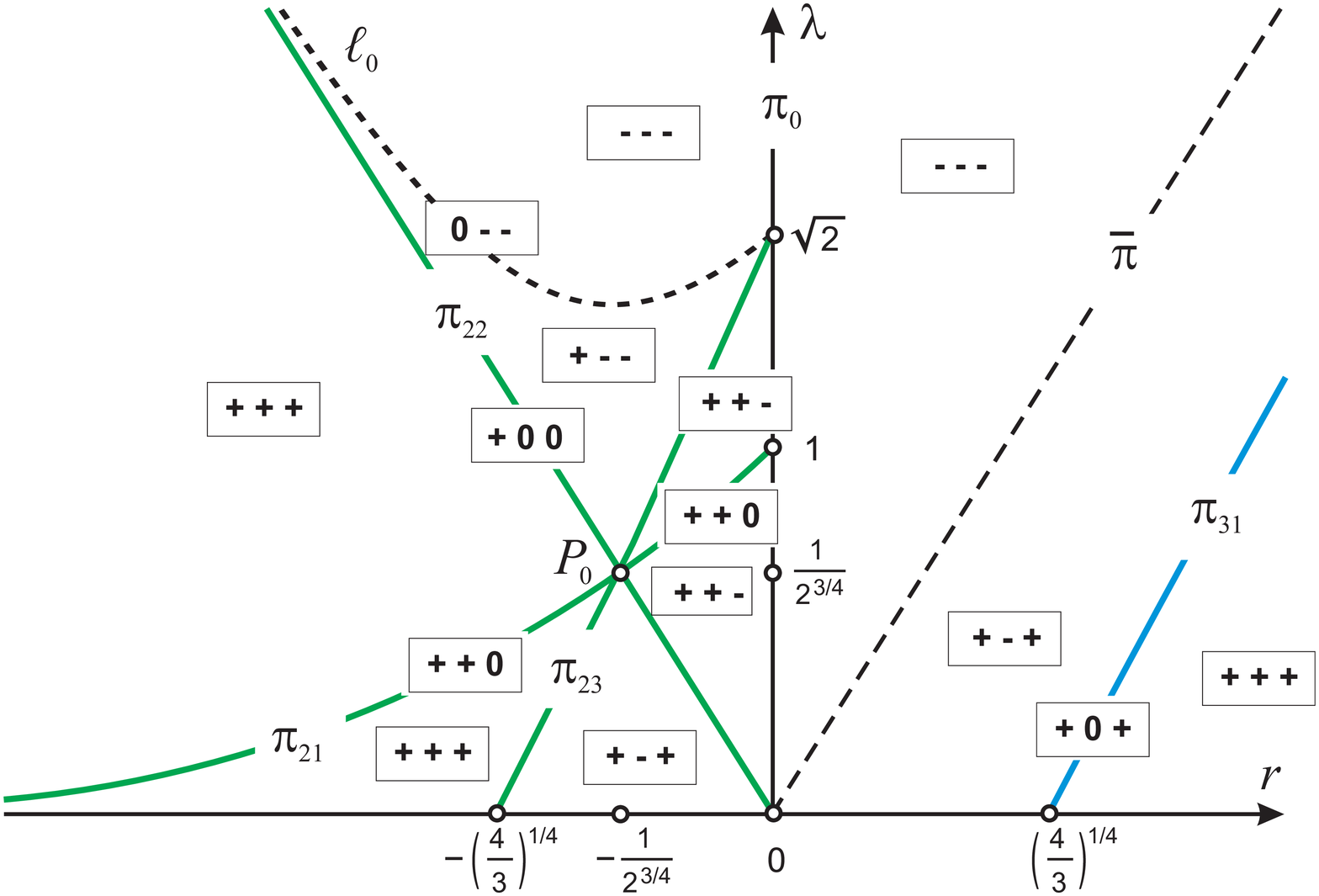}
\caption{Знаки троек $(L_1,L_2,L_3)$.}\label{fig_RazdCritL}
\end{figure}


Отметим следующее свойство, которое неявно использовалось в исследованиях П.Е.\,Рябова и И.Н.\,Гашененко при классификации бифуркационных диаграмм и грубых инвариантов Фоменко (более детально и с уточнениями эти вопросы будут обсуждаться ниже). Предложение~\ref{propos9} утверждает, что все особенности ранга~$0$ имеют сложность~$1$ (понятие сложности введено в \cite{BolFom}). Предложение~\ref{propos10} утверждает, что, более того, даже на один уровень первых интегралов при ненулевой постоянной площадей две таких точки попасть не могут. Возможность совпадения в двух разных точках ранга $0$ значений всех первых интегралов важна при изучении бифуркационных диаграмм различных отображений момента, возникающих в этой задаче. Удивительно то, что явного и четкого доказательства этих двух предложений так нигде и не было предъявлено.

\begin{proposition}\label{propos9}
При всех значениях первых интегралов связная компонента интегрального многообразия не может содержать более одной критической точки ранга $0$.
\end{proposition}

\begin{proposition}\label{propos10}
На один совместный уровень первых интегралов {\rm(}в разные компоненты{\rm)} попадают две точки ранга $0$, отвечающие кривой $\ell_0$ на плоскости $(r,\ld)$. На всех остальных совместных уровнях первых интегралов, содержащих точку ранга $0$, такая точка единственна.
\end{proposition}

\begin{proof}\label{prop10page}
Все точки ранга~0 принадлежат подсистеме $\mm$. Фиксируем $\ld \ne 0$. Согласно \eqref{eq2_25}, точка ранга 0 однозначно определяется значениями $h,p,r$, где $r$ -- кратный корень многочлена \eqref{eq2_26}. Если предположить, что кратных корней два, то есть, что $4R(r)=-(r-r_1)^2(r-r_2)^2$, где $r_1 \ne r_2$, то сразу же приходим к несовместной системе
$$
r_2=-r_1, \quad p=0, \quad r_2^2=2h, \quad r_2^4-4 h^2 + 4=0.
$$
Итак, при заданных $\ld,h,p$ критическая точка ранга 0 единственна (если существует). Теперь фиксируем $\ld, h$ и допустим, что одна и та же пара $(\ell,k)$ определяется разными $p_1 \ne p_2$. Из \eqref{eq2_30} получим две возможности. Первая возможность $p_1=-p_2$ и $h-\ld^2/2=p_1^2$ дает $\ell=0$, и это фигурирующая в утверждении кривая $\ell_0$, в прообразе которой две точки с одним $r$, но с противоположными $p$. Если же $p_1+p_2 \ne 0$, то получаем систему
$$
p_1^2+p_1 p_2+p_2^2 = h-\frac{\ld^2}{2}, \qquad p_1^2+p_2^2 = \frac{1}{6}(h-\frac{\ld^2}{2}),
$$
которая, очевидно, несовместна. Итак, при наличии кратного корня у $R(r)$, этот корень значениями $h,p$ определен однозначно, а, в свою очередь, эти же значения однозначно определены значениями $\ell,k$. Предложение \ref{propos10} \textit{полностью доказано}.\end{proof}

Заметим, что при $\ell \ne 0$ из этого следует и утверждение предложения \ref{propos9}, а при $\ell=0$ доказательство предложения \ref{propos9} из формул приведенного выше решения И.Н.\,Гашененко получается совсем просто.

Таким образом, в случае Ковалевской\,--\,Яхья \textit{все} точки ранга 0 имеют сложность один.

В работе \cite{LogRus} без доказательства сформулировано более слабое утверждение, в котором изначально отбрасываются случаи, разделяющие в плоскости $(\ell,\ld)$ различные виды бифуркационных диаграмм отображения \eqref{eq3_10}. Это разделяющее множество $\Theta_L$ найдено в работах П.Е.\,Рябова \cite{Ryab1,Ryab7,RyabDis} и указано в работе \cite{RyabRCD}. В \cite{LogRus} неразделяющие значения пары $(\ell,\ld)$ называются \textit{небифуркационными}. Как видно из приведенного доказательства предложения \ref{propos10}, никакого отношения к вопросу о сложности точки ранга 0 бифуркационность пары $(\ell,\ld)$ не имеет.

Уравнения кривых в составе упомянутого разделяющего множества $\Theta_L$ П.Е.\,Ря\-бо\-ва можно найти, например, в \cite{RyabRCD}. Это множество классифицирует бифуркационные диаграммы отображений $\jpriv_\ell(\ld)$. Выясним, с чем оно действительно связано. Вычислим множество $\widehat\Theta$ --- образ в октанте $\{(\ell,\ld):\ell \gs 0, \ld \gs 0\}$ кривых $\vpi_{ij}$, служащих разделяющими при классификации точек ранга 0, вместе с их предельными точками при $\ld=0$. Сохраняя для кривых-образов те же обозначения, что и у кривых-прообразов, получим:
\begin{equation}\label{eq4_48}
\begin{array}{lll}
  \vpi_{21}: & \ds \ell = \frac{1}{2\ld^{1/3}}\sqrt{1-\ld^{4/3}}, & 0\ls \ld\ls 1;\\[4mm]
  \vpi_{22}: & \ds \ell = {\frac{1}{\sqrt{2}}(\sqrt{1+\ld^4}-\ld^2)^{3/2}}, & \ld \gs 0;\\[4mm]
  \vpi_{23}: & \left\{
  \begin{array}{l}
  \ell=\ds \frac{(4-x^4)^{3/2}}{4x^3}, \\[2mm]
  \ld=\ds{\frac{3x^4-4}{2x^3}}
  \end{array}
  \right., & x \in [\sqrt[4]{4/3},\sqrt{2}]; \\[4mm]
  \vpi_{24}: &
  \ell=\ds \frac{|\sqrt{4 + \ld^{4/3}}-2\ld^{2/3}|}{\sqrt{2}(\sqrt{4 + \ld^{4/3}}-\ld^{2/3})^{1/2}}, & \ld \gs 0;\\[4mm]
  \vpi_{31}: & \left\{
  \begin{array}{l}
  \ell=\ds \frac{(4-x^4)^{3/2}}{4x^3}, \\[2mm]
  \ld=\ds{\frac{3x^4-4}{2x^3}}
  \end{array}
  \right., & x \in [-\sqrt[4]{4/3},0).
\end{array}
\end{equation}
Напомним, что кривая $\vpi_0=\{\ell=0, \ld \gs 0\}$ не является разделяющей внутри класса $\delta_2$ в смысле введенного ранее отношения эквивалентности. Из нее в разделяющее множество входят лишь точки, в которых заканчиваются кривые $\vpi_{2j}$: $\ld=0,1,\sqrt{2}$. Итак, мы видим, что $\Theta_L \setminus \widehat \Theta$ состоит из одной кривой $\ell=(4\ld)^{-1}$, смысл которой здесь будет выявлен ниже. В работе \cite{RyabHarlUdgu2} показано, что эта кривая соответствует экстремальному значению интеграла $L$ на семействе вырожденных критических точек ранга 1. Таким образом, подавляющее число перестроек бифуркационных диаграмм отображений $\jpriv_\ell(\ld)$ происходит в случаях, когда в приведенной системе имеется \textit{вырожденная} критическая точка ранга 0, а множество $\Theta_L$ бифуркационных пар $(\ell,\ld)$ --- это, за исключением одной кривой, образ вырожденных точек ранга~0, и множество $\Theta_L$ никак не связано с возможностью попадания нескольких критических точек ранга 0 на один интегральный уровень.

Заметим, что теорема 9 работы \cite{LogRus} (равно как и тождественная ей теорема 8 работы \cite{Slav2Rus}) в определенном смысле является центральным утверждением в том, что касается вопроса невырожденности точек ранга 0, отвечающих \textit{небифуркационным} значениям $(\ell,\ld)$. На этом строится их дальнейшая классификация, меченые круговые молекулы и, как следствие, ищется часть меток на графах Фоменко. Однако и эта часть утверждения, очевидная при анализе уравнений \eqref{eq4_48}, в работах \cite{LogRus,Slav2Rus} не доказана. Указано, что для всех небифуркационных значений $(\ell,\ld)$ проверка условий невырожденности может быть выполнена с помощью компьютера. Однако характеристические многочлены соответствующих симплектических операторов через $\ell$ и $\ld$ не выражены. То же самое относится и к проверке условия двумерности подалгебры, порожденной операторами $\A_H,\A_K$. Создается впечатление, что она всегда двумерна, что не так. Выше нами предъявлено множество значений параметров, где подалгебра одномерна.

В связи с этим, отметим также, что в \cite{Slav2Rus} имеется ссылка на работу \cite{AndrRus}, из которой, в частности, извлекается информация о количестве семейств торов в камерах, что позволяет строить круговые молекулы точек ранга 0 (в работе \cite{LogRus} количество семейств указано неверно). Однако работа \cite{AndrRus} в этом отношении также не может быть принята в качестве доказательств, так как использует некоторые численные эксперименты, основанные на том, что авторы имеют алгоритм, позволяющий для любого набора констант первых интегралов выбрать по одной начальной точке на каждом из соответствующих регулярных торов (в частности, тем самым и установить их количество). Такого алгоритма авторы работы \cite{AndrRus} не предъявили. Известно из сообщений на конференциях (например, \cite{AndrCon}), что использована гипотеза П.Е.\,Рябова, состоящая в том, что каждый регулярный тор в $\mP^5$ имеет выход на подмножество $\{\omega_2=0,\alpha_2=0\}$, но эта гипотеза на сегодня так никем и не доказана.

К моменту подготовки работы \cite{Slav2Rus} уже были опубликованы работы \cite{RyabUdgu,KhMTT42}. В первой даны явные точные аналитические выражения для собственных чисел симплектических операторов, приводящие к установлению типа критических точек. Во второй работе \cite{KhMTT42} содержится точная аналитическая классификация точек ранга 0, в том числе по вырожденности и невырожденности, по типам в случае невырожденности и по полному совместному уровню первых интегралов, в том числе первое (и единственное на сегодня) строгое доказательство невырожденности точек ранга 0 вне разделяющего множества П.Е.\,Рябова.


\subsection{Диаграммы Смейла и изоэнергетические поверхности}
В задаче Ковалевской ($\ld=0$) бифуркационную диаграмму интегралов энергии и площадей построил А.\,Якоб \cite{Jacob}. Он же с помощью конструкции Смейла (приведенное расслоение единичных сфер над областью возможности движения) определил топологический тип изоэнергетических многообразий -- трехмерных уровней ``приведенного гамильтониана'' \eqref{eq4_2}
\begin{equation}\label{eq4_49}
    \iso=\{\zeta \in \mPel: H_\ell(\zeta)=h\}.
\end{equation}
Диаграмма отображения
\begin{equation}\notag
    L{\times}H: \mP^5 \to \bR^2
\end{equation}
(обозначим ее $\smale$) в этом случае состоит из двух парабол
\begin{equation}\notag
\delta_1^0:\; h=-1+\ell^2, \qquad \delta_2^0:\; h=1+\ell^2
\end{equation}
и пары симметричных относительно оси $Oh$ кривых, которые удобно записать в параметрической форме
\begin{equation}\label{eq4_50}
\delta_3^0:\;  \ell = \frac{x^2+4}{4\sqrt{2x}}, \qquad   h=\frac{3x}{4}+\frac{1}{x},  \qquad x \in (0,2].
\end{equation}
Эти значения достигаются на относительных равновесиях, фазовые координаты которых
\begin{equation}\notag
\begin{array}{lll}
\ds  \omega_1 =-\sqrt{\frac{x}{2}}, & \ds \omega_2=0, & \ds \omega_3 =\frac{\sqrt{4-x^2}}{2\sqrt{x}} \\[3mm]
\ds  \alpha_1 = -\frac{x}{2}, & \alpha_2=0, & \ds \alpha_3=\frac{\sqrt{4-x^2}}{2}
\end{array}
\end{equation}
(радикалы $\sqrt{\mathstrut x}, \sqrt{\mathstrut 4-x^2}$ -- алгебраические). Кривые \eqref{eq4_50} касаются верхней параболы в точках $(\pm 1,2)$ и трансверсально ее пересекают в точках $(\pm \sqrt{2(\sqrt{2}-1)},2\sqrt{2}-1)$. Точки возврата имеют координаты $(\pm 2/3^{3/4},\sqrt{3})$, достигаются при $x^2=4/3$. Отметим, что при этом $\omega_3^4=4/3$.

Изоэнергетические поверхности пусты в области $h<-1+\ell^2$. Для остальных областей, на которые $\smale$ делит плоскость $O\ell h$, они диффеоморфны (см. рис.~\ref{fig_smalekowa}) следующим многообразиям
\begin{equation}\label{eq4_51}
S^3, \; K^3=(S^2{\times}S^1)\#(S^2{\times}S^1), \; S^2{\times}S^1, \; {\bR P^3}.
\end{equation}
Гладкий тип $\iso$ в любой точке $(\ell,h)$ можно определить, зная индекс Морса функции $H_\ell$ в ее критических точках, лежащих в прообразах бифуркационных кривых, и приходя в точку $(\ell,h)$ вдоль вертикальной прямой из достаточно низко лежащей точки с заведомо недопустимым значением $h$ (то есть из такой точки, где $\iso=\varnothing$).

\begin{figure}[ht]
\centering
\includegraphics[width=0.8\textwidth,keepaspectratio]{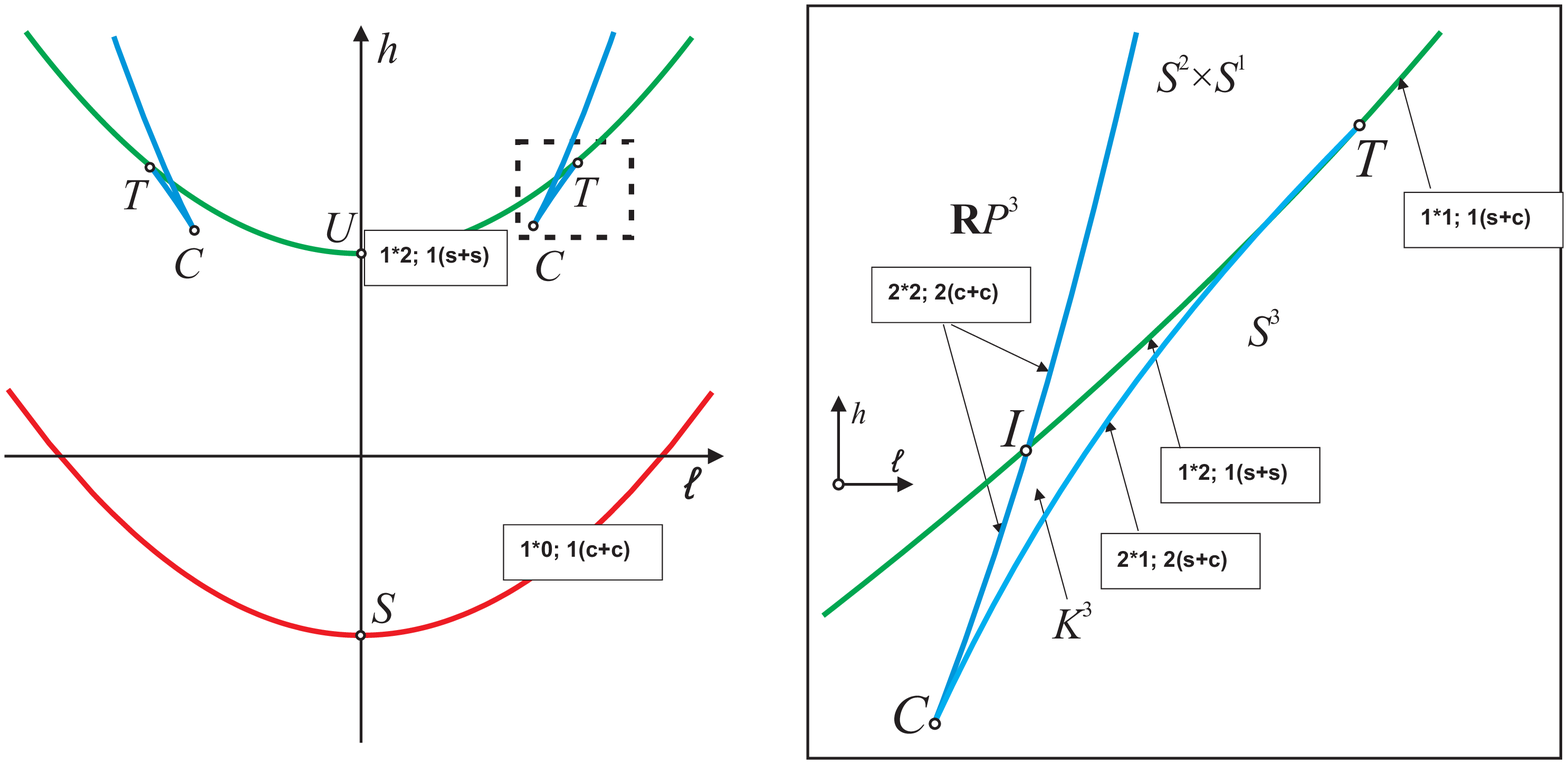}
\caption{Диаграмма Смейла классической задачи.}\label{fig_smalekowa}
\end{figure}

Рассмотрим, что получается в классической задаче (конечно, эти результаты известны \cite{KhPMM83,KhDan83,Richt,BolFom}, но явные вычисления индексов никогда не предъявлялись). Характеристические многочлены оператора $\A_H$ получим предельным переходом из \eqref{eq4_29}
\begin{equation}\notag
    \begin{array}{ll}
      \delta_1^0: &  \chi_H(\mu)=\ds{\left(\mu^2+\frac{1}{2}\right)\left[\mu^2+(1+\ell^2)\right]}, \\[3mm]
      \delta_2^0: &  \chi_H(\mu)=\ds{\left(\mu^2-\frac{1}{2}\right)\left[\mu^2-(1-\ell^2)\right]}, \\
      \delta_3^0: &  \chi_H(\mu)=\ds{(\mu^2+r^2)\left[\mu^2-\frac{1}{4}(\sqrt{4+r^4}-2r^2)\right]}.
    \end{array}
\end{equation}
Поэтому на нижней параболе все критические точки ранга $0$ в прообразе имеют тип ``центр-центр'', на верхней параболе -- тип ``седло-седло'' на ограниченном отрезке между двумя симметричными друг другу относительно оси $Oh$ точками $T$ касания с третьей кривой (в частности, вращения с центром масс в наивысшем положении при $|\ell|<1$ неустойчивы по всем переменным) и тип ``седло-центр'' на неограниченных участках за пределами точек касания (вращения с центром масс в наивысшем положении при $|\ell|>1$ по части переменных устойчивы). В прообразе каждой точки парабол такая критическая точка одна. На кривых \eqref{eq4_50} имеем в прообразе по две точки типа ``седло-центр'' при $r^4<4/3$ (на ограниченных участках между точками возврата $C$ и касания $T$) и по две точки типа ``центр-центр'' при $r^4>4/3$ (на неограниченных участках от точек возврата в бесконечность). Проход по гладкой ветви через точку трансверсального пересечения $I$ на тип не влияет.

В силу механического характера гамильтониана $H_\ell$ его индекс Морса равен индексу Морса ``эффективного потенциала'' -- функции на сфере Пуассона, субуровни которой есть области возможности движения (ОВД). Эффективный потенциал для случая Ковалевской\,--\,Яхья, вычисленный по схеме Смейла, имеет вид
\begin{equation}\notag
U_{\ell,\ld}= -\alpha_1 + \frac{(2 \ell - \ld\alpha_3 )^2}{2[2(\alpha_1^2 + \alpha_2^2) + \alpha_3^2]}.
\end{equation}

Чтобы легко вычислить индекс Морса ограничения функции трех переменных $f(\alpha_1,\alpha_2, \alpha_3)$ на сферу Пуассона \eqref{eq2_3} не вводя локальных координат, применим следующее утверждение.
\begin{lemma}\label{lem2}
Рассмотрим дифференциальный оператор
\begin{equation}\label{eq4_52}
\Xi = {\bs \alpha}\times \frac{\partial}{\partial {\bs \alpha}},
\end{equation}
порождающий вторую группу уравнений \eqref{eq4_3}. Пусть ${\bs \alpha}_0 \in S^2=\{{\bs \alpha}: |{\bs \alpha}|=1\}$ -- невырожденная в смысле Морса критическая точка ограничения функции $f({\bs \alpha})$ на $S^2$. Индекс Морса функции $f$ в точке ${\bs \alpha}_0$ равен количеству отрицательных корней многочлена
\begin{equation}\notag
\xi_f(\mu) = \frac{1}{\mu}\det \bigl[(\Xi^2 f)({\bs \alpha}_0) - \mu E\bigr].
\end{equation}
\end{lemma}

Применяя к функции $U_{\ell,0}$, получим
\begin{equation}\notag
    \begin{array}{ll}
      \delta_1^0: &  \xi_H(\mu)=(\mu-1)[\mu-(\ell^2+1)],\\[2mm]
      \delta_2^0: &  \xi_H(\mu)=(\mu+1)[\mu-(\ell^2-1)], \\[2mm]
      \delta_3^0: &  \xi_H(\mu)=\ds{\left[ \mu+ \frac{1}{2}(\sqrt{4+r^4}+r^2)\right]}\ds{\left[\mu-\frac{r}{\sqrt{4+r^4}}(\sqrt{4+r^4}-2r^2)\right]}.
    \end{array}
\end{equation}
Расстановку индексов Морса и типов вдоль бифуркационных кривых получим как показано на рис.~\ref{fig_smalekowa}. Здесь обозначение $n*m$ означает, что в прообразе лежит $n$ точек индекса $m$, обозначения $c+c,s+c,s+s$ указывают тип точки (``центр-центр'', ``седло-центр'', ``седло-седло'').

Известно, что при пересечении значением $h$ критического значения происходят следующие перестройки ОВД (проекции уровня энергии на конфигурационное пространство): индекс 0 -- добавление диска $D^2$, индекс 1 -- приклейка ручки (из одного диска делает диск с дыркой, то есть кольцо, а из двух дисков может сделать один), индекс 2 -- заклейка дырки диском.

Проведем на плоскости $\bR^2(\ell,h)$ вертикальную прямую $\ell=\cons$ между точками $I$ и $T$. Вдоль нее гамильтониан и эффективный потенциал имеют критические значения
$$
h_1=-1+\ell^2<h_2<h_3=1+\ell^2<h_4
$$
с количеством критических точек в прообразе соответственно $1,2,1,2$ с индексами $0,1,2,2$. В соответствии с этим ОВД на сфере таковы: диск, диск с двумя дырками (сфера с тремя дырками), кольцо (сфера с двумя дырками), сфера. Приведенные расслоения единичных окружностей над ними дают соответственно многообразия \eqref{eq4_51}.

В общем случае вычисления при $\ld>0$ дают следующие показатели Морса.
\begin{theorem}\label{th6}
Корни характеристического многочлена $\xi_H(\mu)$ в критических точках \eqref{eq4_8} в обозначениях предложения~{\rm \ref{propos2}} таковы
\begin{equation}\notag
    \begin{array}{l}
    \mu_1 = -\ds{\frac{1}{2}\left[ r(r-\ld)+d\right]}, \\
    \mu_2=-\ds{\frac{1}{2(r-\ld)d}\left[ (2r-\ld)(r-\ld)- d\right]\left[ (2r-\ld)(r-\ld)r + \ld d\right]}.
    \end{array}
\end{equation}
В частности, знак $\mu_1$ всегда противоположен знаку $d$, поэтому $\mu_1$ положительно в области $\delta_1$ и отрицательно в областях $\delta_2, \delta_3$. Знак $\mu_2$ определяется расположением точки $(r,\ld)$ относительно разделяющих кривых $\overline{\vpi},\vpi_0,\vpi_{23},\vpi_{24},\vpi_{31}$. В итоге, индекс Морса эффективного потенциала $U_{\ell,\ld}$ и гамильтониана $H_{\ell,\ld}$ равен

$0$ в области $\delta_1$,

$1$ в областях $\delta_{23},\delta_{27},\delta_{31}$,

$2$ в областях $\delta_{21},\delta_{22},\delta_{24},\delta_{25},\delta_{26},\delta_{28},\delta_{32}$.
\end{theorem}

\begin{figure}[ht]
\centering
\includegraphics[width=0.6\textwidth,keepaspectratio]{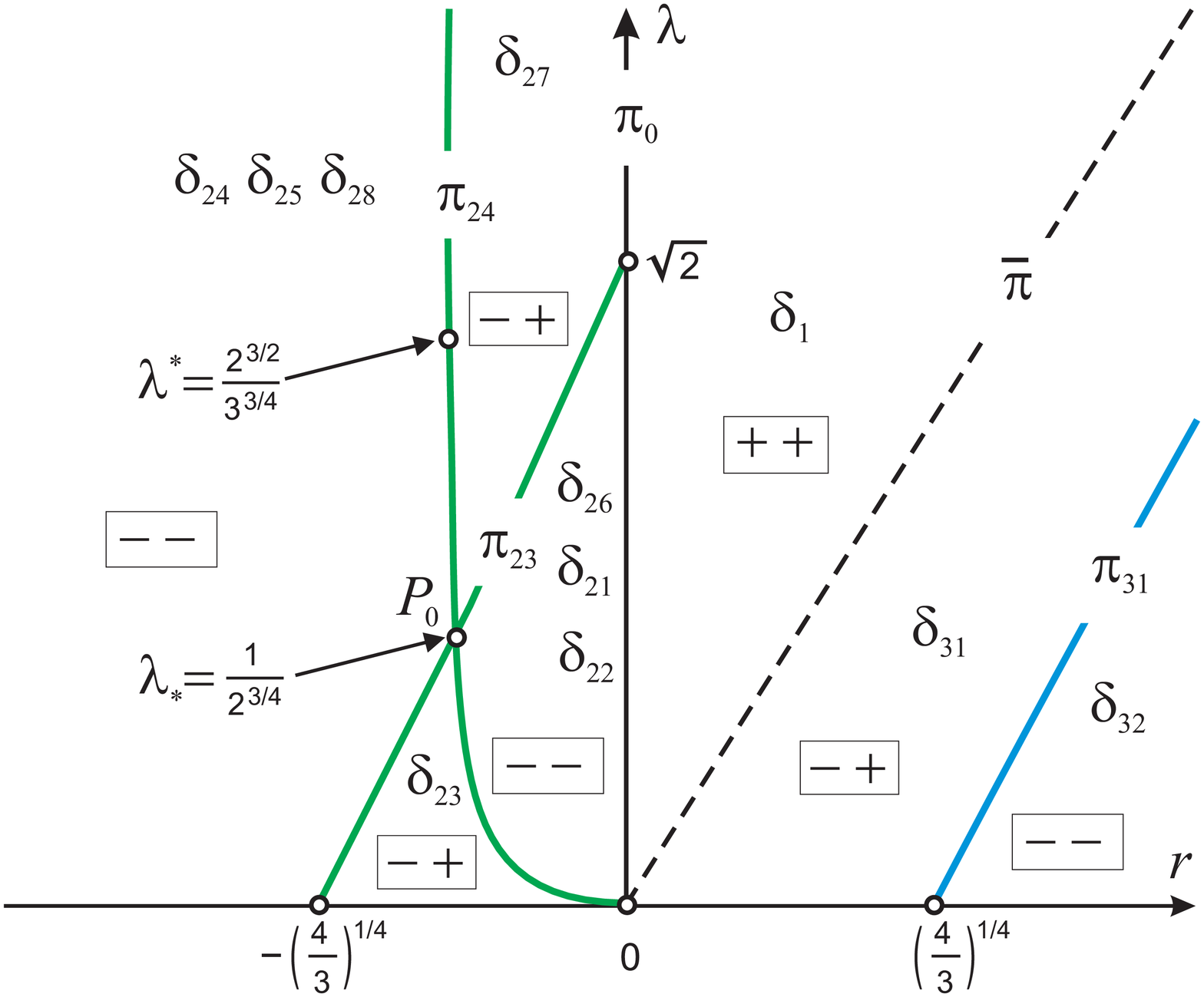}
\caption{Показатели Морса эффективного потенциала.}\label{fig_RazdCrit1}
\end{figure}

\FloatBarrier

\begin{figure}[ht]
\centering
\includegraphics[width=0.8\textwidth,keepaspectratio]{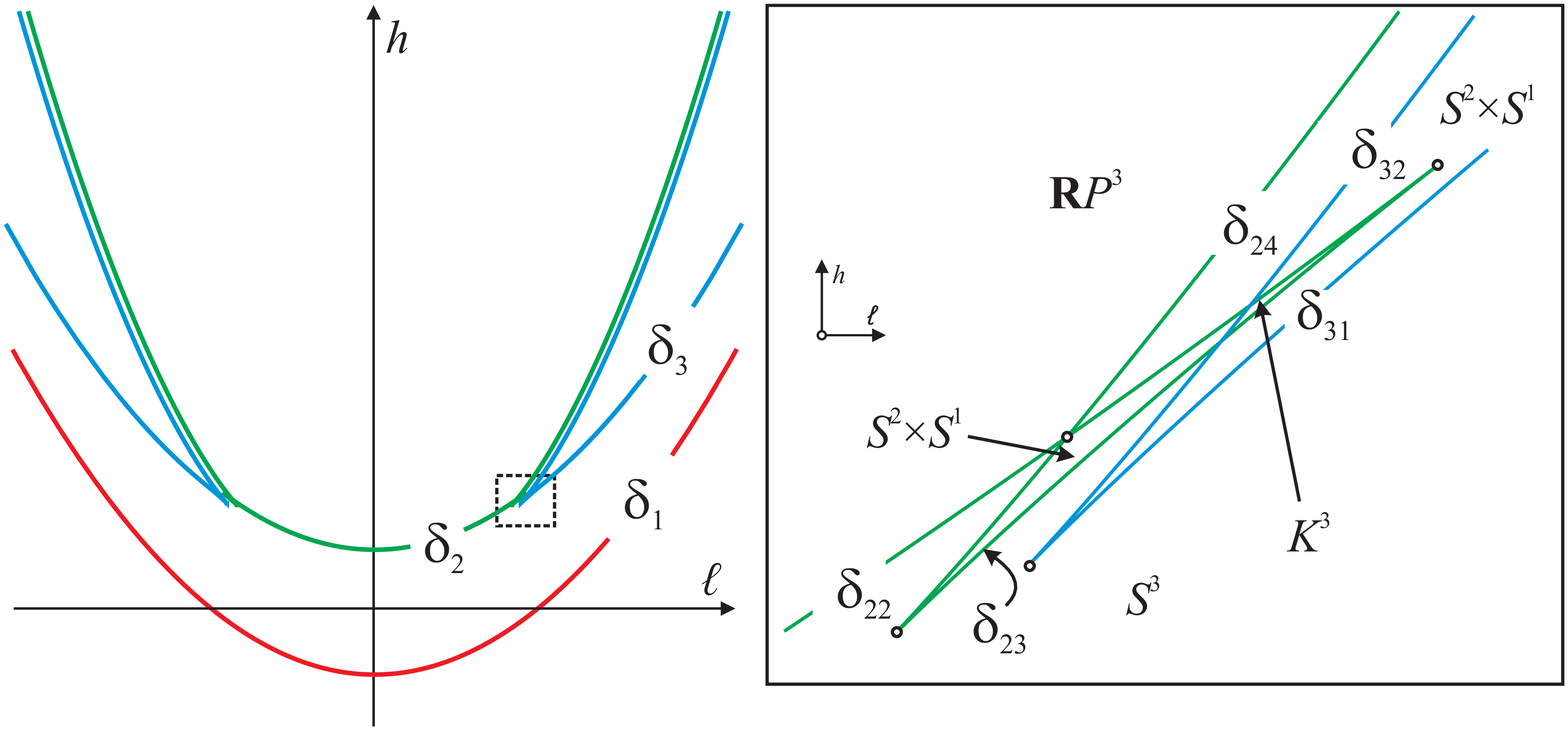}
\caption{Диаграмма $\smale$ при малых $\ld$.}\label{fig_sm1}
\end{figure}

На рис.~\ref{fig_RazdCrit1} приведены пары показателей Морса, а также показана некоторая информация (о которой будет сказано ниже), необходимая для анализа видоизменений диаграммы Смейла при ненулевых $\ld$. При возмущении этого значения от нулевого кривая $\delta_1^0$ трансформируется в $\delta_1$, перестает быть параболой, но никаких особых точек на ней не возникает. Каждая связная часть кривой $\delta_3^0$, за точками которой стоит по две критических точки, расщепляется и один из отрезков, полученных из ограниченного участка, объединяется с ветвью кривой $\delta_2^0$ выше точки $T$. Возникает $\delta_3$, у которой каждая из симметричных относительно $Oh$ компонент состоит из двух бесконечных ветвей, сходящихся в точке возврата, возмущенной из $C$. Из второй части $\delta_3^0$ и кривой $\delta_2^0$ склеивается кривая $\delta_2$, треугольник $CTI$ порождает на ней ласточкин хвост. С учетом уже известных индексов Морса получаем картину, показанную на рис.~\ref{fig_sm1}. Этот тип диаграммы сохраняется при всех $\ld \in (0,\ld_1)$, где $\ld_1$ такое значение параметра, при котором самая правая точка возврата на кривой $\delta_2$ (отвечающая разделяющей кривой $\vpi_{24}$) попадает на кривую $\delta_3$. Ранее существование такого значения было предсказано на основе численного моделирования \cite{Ryab5,Ryab7,RyabDis}, но вычислить его аналитически не удавалось. Современные САВ позволяют это сделать. Запишем систему уравнений, исходя из \eqref{eq4_11}, \eqref{eq4_40} (необходимо выбрать знак $d>0$):
\begin{equation}\label{eq4_53}
\ell(r,\ld)=\ell(r_0,\ld), \quad h(r,\ld)=h(r_0,\ld), \quad r_0= \ds{\frac{1}{2}}\left(\ld-\ld^{1/3}\sqrt{4+\ld^{4/3}}\right).
\end{equation}
Заметим, что она описывает не только все возможные случаи попадания точки $\vpi_{24}$ на кривую $\delta_3$, но и гипотетическую возможность ее попадания на другую ветвь кривой $\delta_2$.
Исключим из этой системы $r_0,r$, полагая $r_0 \ne r$. Последовательность действий такова. Избавляемся в первых двух уравнениях от радикала $d$, подставляем $r_0$ и делаем замену
$U=\ld^{1/3}$. Два полученных уравнения сокращаются на $(2r-U^3+U\sqrt{4+U^4})^2$ (что соответствует $r=r_0$). Приходим к системе
\begin{equation}\label{eq4_54}
\begin{array}{l}
  p_{13}(U,r)+(U^3-r)p_8(U,r)\sqrt{4+U^4}=0, \\
  p_{25}(U,r)+(U^3-r)p_{20}(U,r)\sqrt{4+U^4}=0,
\end{array}
\end{equation}
где $p_i(U,r)$ --- многочлен степени $i$ по $U$. Вычисляем результант левых частей по $r$ в подстановке $X=(U^2+\sqrt{4+U^4})^2$, что соответствует замене
\begin{equation}\label{eq4_55}
X=(\ld^{2/3}+\sqrt{4+\ld^{4/3}})^2.
\end{equation}
Получим уравнение
$$
(X-8)^2 (X^2-16)^4 (X+4)^2 P_1(X) P_2(X)=0,
$$
где
\begin{eqnarray}
& & P_1= X^4-24 X^3+720 X^2-2048 X-3072, \nonumber\\
& & P_2= 4 X^{15} - 293 X^{14} + 7864 X^{13} - 70320 X^{12} \nonumber \\
& & \phantom{P_2=}- 831232 X^{11} + 26316032 X^{10} - 263235584 X^9  + 1223192576 X^8 \nonumber \\
& & \phantom{P_2=}- 2241200128 X^7 + 323747840 X^6 + 1465909248 X^5 - 16521363456 X^4 \nonumber \\ & & \phantom{P_2=}- 25736249344 X^3 - 74155294720 X^2 - 75161927680 X -17179869184.\nonumber
\end{eqnarray}

При вытекающем из \eqref{eq4_55} очевидном условии $X\gs 4$ уравнение $P_2(X)=0$ имеет единственный вещественный корень $X\approx 11.2707$, при котором $\ld \approx 1.1268$ и единственный общий вещественный корень уравнений \eqref{eq4_54} $r\approx 0.002089 \in (0,\ld)$, то есть не дает точек из $\delta_{2,3}$. Если же предположить, что получена точка на $\delta_1$, то необходимо выбрать в уравнениях \eqref{eq4_53} знак $d<0$, но тогда найденные значения $(r,\ld)$ этим уравнениям не удовлетворяют. Итак, все корни $P_2(X)$ --- посторонние. Условие $X=8$ приводит к разделяющему значению $\ld_*=1/2^{3/4}$, при котором происходит слияние обеих точек возврата и точки самопересечения кривой $\delta_2$. Уравнение $P_1(X)=0$ имеет единственный допустимый корень $X\approx 4.3418$, при котором $\ld \approx 0,02349$ и единственный общий вещественный корень уравнений \eqref{eq4_54} $r\approx 1.47328 \in (\ld,+\infty)$. Поэтому такое значение $\ld$ отвечает искомому случаю $\vpi_{24} \in \delta_3$. Попутно доказано, что точка возврата $\vpi_{24}$ не может попасть на другую ветвь $\delta_3$. Найденное разделяющее значение $\ld$ обозначим через
$\ld_1$ и выпишем точно
\begin{equation}\label{eq4_56}
\begin{array}{l}
\ld_1=\ds{\left(-\frac{7}{6}-\frac{1}{12} p_2 +\frac{1}{2} \sqrt{-\frac{205}{9}+\frac{3823}{4} p_1^{-1/3}-\frac{1}{4} p_1^{1/3}+\frac{16393}{9 p_2}}\right)^{3/4}}\approx 0.0235,\\
p_1=1045767+183872 \sqrt{34}, \qquad p_2 =\sqrt{-410-34407 p_1^{-1/3}+9 p_1^{1/3}}.
\end{array}
\end{equation}
При $\ld>\ld_1$ кривые $\delta_2, \delta_3$ более общих точек не имеют (фрагмент показан на рис.~\ref{fig_sm2}). В частности, исчезает область (суперпозиция ``угла'' $\delta_3$ и ``хвоста'' кривой $\delta_2$), в которой $\iso=K^3$.

\begin{figure}[ht]
\centering
\includegraphics[width=0.5\textwidth,keepaspectratio]{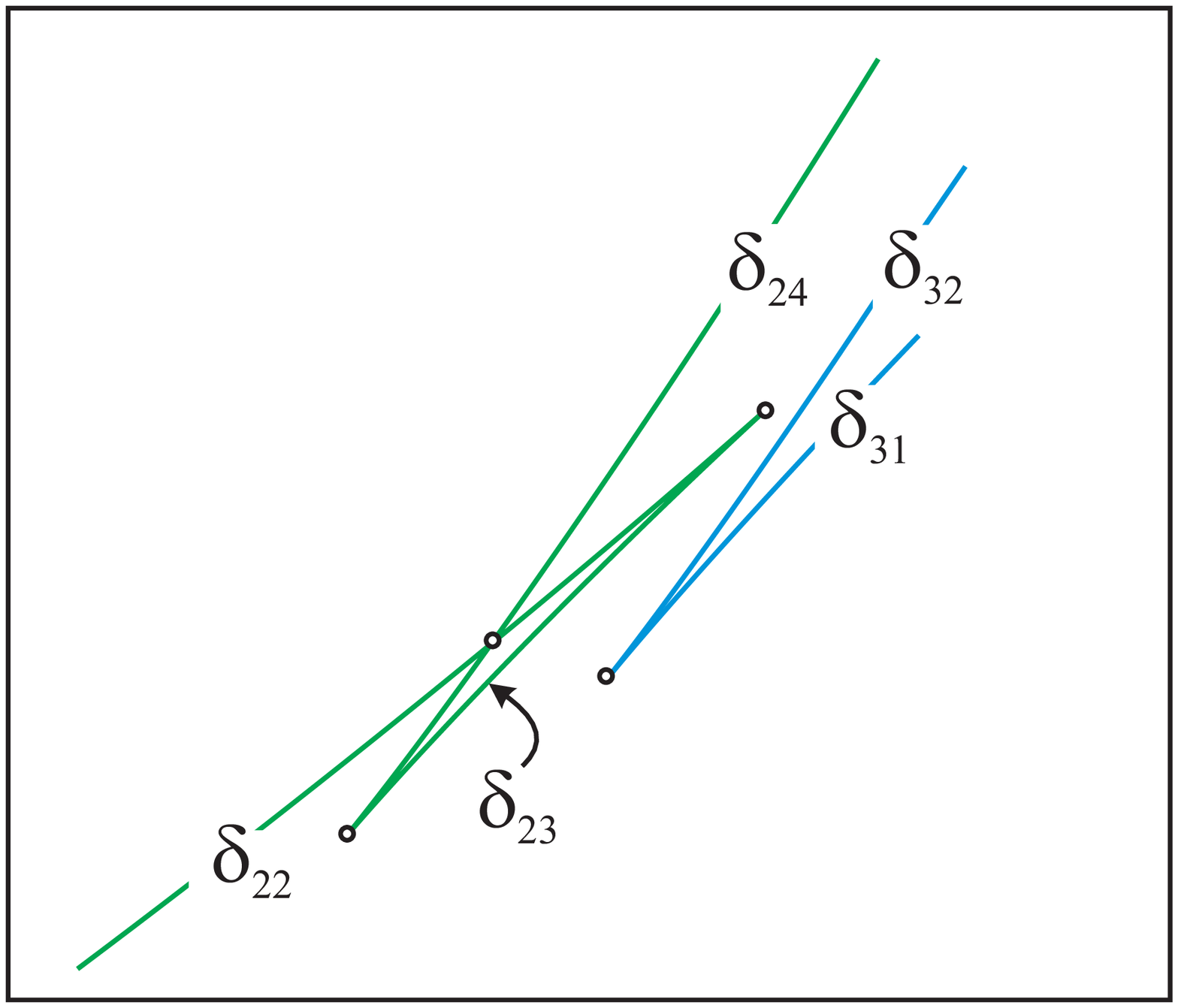}
\caption{Фрагмент диаграммы $\smale$ при $\ld>\ld_1$}\label{fig_sm2}
\end{figure}

Дальнейшие перестройки диаграммы Смейла связаны только с эволюцией кривой $\delta_2$. При этом вырожденные критические точки ранга $0$, порождающие разделяющие кривые $\vpi_{21},\vpi_{22}$, собственно на диаграмму не влияют, так как лежат на ее гладких участках. При переходе через $\ld_*=1/2^{3/4}$ меняются местами ``кончики хвоста'' (см. рис.~\ref{fig_sm3}), что легко заметить, если провести на рис.~\ref{fig_RazdCrit1} горизонтальную прямую ($\ld=\cons$) и двигать ее вверх.

\begin{figure}[ht]
\centering
\includegraphics[width=\textwidth,keepaspectratio]{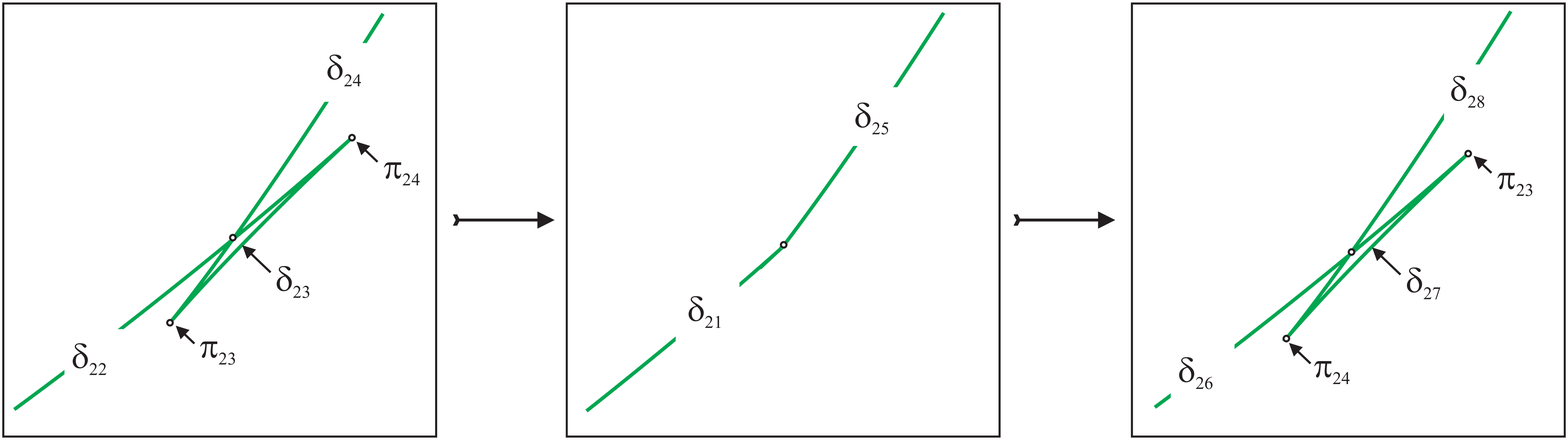}
\caption{Переход через $\ld_*$}\label{fig_sm3}
\end{figure}

Изучим возможность попадания точки возврата $\vpi_{23}$ на другую ветвь кривой $\delta_2$. Для этого случая из \eqref{eq4_37} запишем систему
\begin{equation}\notag
\begin{array}{c}
  \ell(r,\ld)=\ell(r_0,\ld), \quad h(r,\ld)=h(r_0,\ld), \\
  r_0=\ds{\frac{x^4-4}{2x^3}}, \quad \ld=\ds{\frac{3x^4-4}{2x^3}},\quad
  x^4-2 r x^3-4 \ne 0, \quad  x \in \left(\sqrt[4]{4/3},\sqrt{2}\right].
\end{array}
\end{equation}

\noindent Исключая из нее $r_0,r,\ld$, придем к уравнению для $X=x^4$
$$
(X-2)^2 Q_1(X) Q_2(X)=0,
$$
где
\begin{eqnarray}
& & \;Q_1= 3X^4+32 X^3-180 X^2+96 X-64, \nonumber\\
& & \begin{array}{l} Q_2= 6903 X^{11} -153216 X^{10} +1489200 X^9  - 9324352 X^8  \\ \phantom{P_2=}
+44169408 X^7 -160186880 X^6 + 425104384 X^5 - 806682624 X^4 \\ \phantom{P_2=}
+1108361216 X^3 - 1120534528 X^2 +784072704 X - 285212672. \end{array}\nonumber
\end{eqnarray}
Уравнение $Q_2(X)=0$ в промежутке $X\in (4/3,4]$ корней не имеет, а уравнение $Q_1(X)=0$ имеет в этом промежутке ровно один корень:
$$
X=\frac{1}{3} (\sqrt{308 + \frac{435}{q_1} + 3 q_1 + \frac{3400}{q_2}} -
     q_2-8), q_1=(2951 - 408 \sqrt{34})^{1/3}, q_2=(154 - \frac{435}{q_1} - 3 q_1)^{1/2}
$$
Соответственно, в терминах $\ld$ имеем уравнение
\begin{equation}\notag
    64 \ld^{16}+2784 \ld^{12}-274803 \ld^8+15476896 \ld^4-45349632 =0,
\end{equation}
с единственным вещественным решением $\ld_2 \approx 1.32631$. Его точное значение в радикалах
\begin{equation}\label{eq4_57}
\begin{array}{l}
  \ld_2 = \left(\ds{-\frac{87}{8}  - \frac{1}{16}p_2+\frac{1}{2} p_3}
  \right)^{1/4}, \\
  p_1=10467417865895 + 1809781698048 \sqrt{34}, \\
  p_2=\sqrt{213478 - 1093516839 p_1^{-1/3} + 9 p_1^{1/3}},\\
  p_3=\ds{\sqrt{\frac{106739}{16} + \frac{1093516839}{64} p_1^{-1/3} - \frac{9 }{64}p_1^{1/3} + \frac{88449457}{16 p_2}}}.
\end{array}
\end{equation}

Остальные изменения связаны с осью $Oh$ \cite{RyabDis}. Константа $\ell$ обращается в нуль не только в положениях равновесия тела ($r=0$), но и при вытекающем из \eqref{eq4_11} условии \eqref{eq4_43}. Соответствующая кривая, ранее обозначенная через $\ell_0$, показана на рис.~\ref{fig_RazdCrit0}. Минимум $\ld$ на этой кривой равен $\ld^*=2\sqrt{2}/3^{3/4}$. Поэтому кроме всегда существующих точек $h=\pm 1, \ell=0$ диаграмма имеет еще две точки пересечения с осью $Oh$ при $\ld^*<\ld<\sqrt{2}$ и одну точку при $\ld>\sqrt{2}$. Весь этот набор точек на оси $Oh$ перестраивается в момент совпадения двух из них. Легко посчитать, что это происходит при $\ld=2 \sqrt{\sqrt{2}-1}\approx 1.2871885058111654 \in (\ld^*,\sqrt{2})$. Таким образом, имеем типы диаграмм, существенные фрагменты которых показаны на рис.~\ref{fig_fig7}: $a)$~$\ld_*<\ld<\ld^*$; $b)$~$\ld^*<\ld<2 \sqrt{\sqrt{2}-1}$; $c)$~$2 \sqrt{\sqrt{2}-1}<\ld<\ld_2$; $d)$~$\ld_2<\ld<\sqrt{2}$; $e)$~$\ld>\sqrt{2}$. С учетом вычисленных выше индексов Морса имеем указанную на этом же рисунке расстановку изоэнергетических многообразий.

\begin{figure}[!ht]
\centering
\includegraphics[width=\textwidth,keepaspectratio]{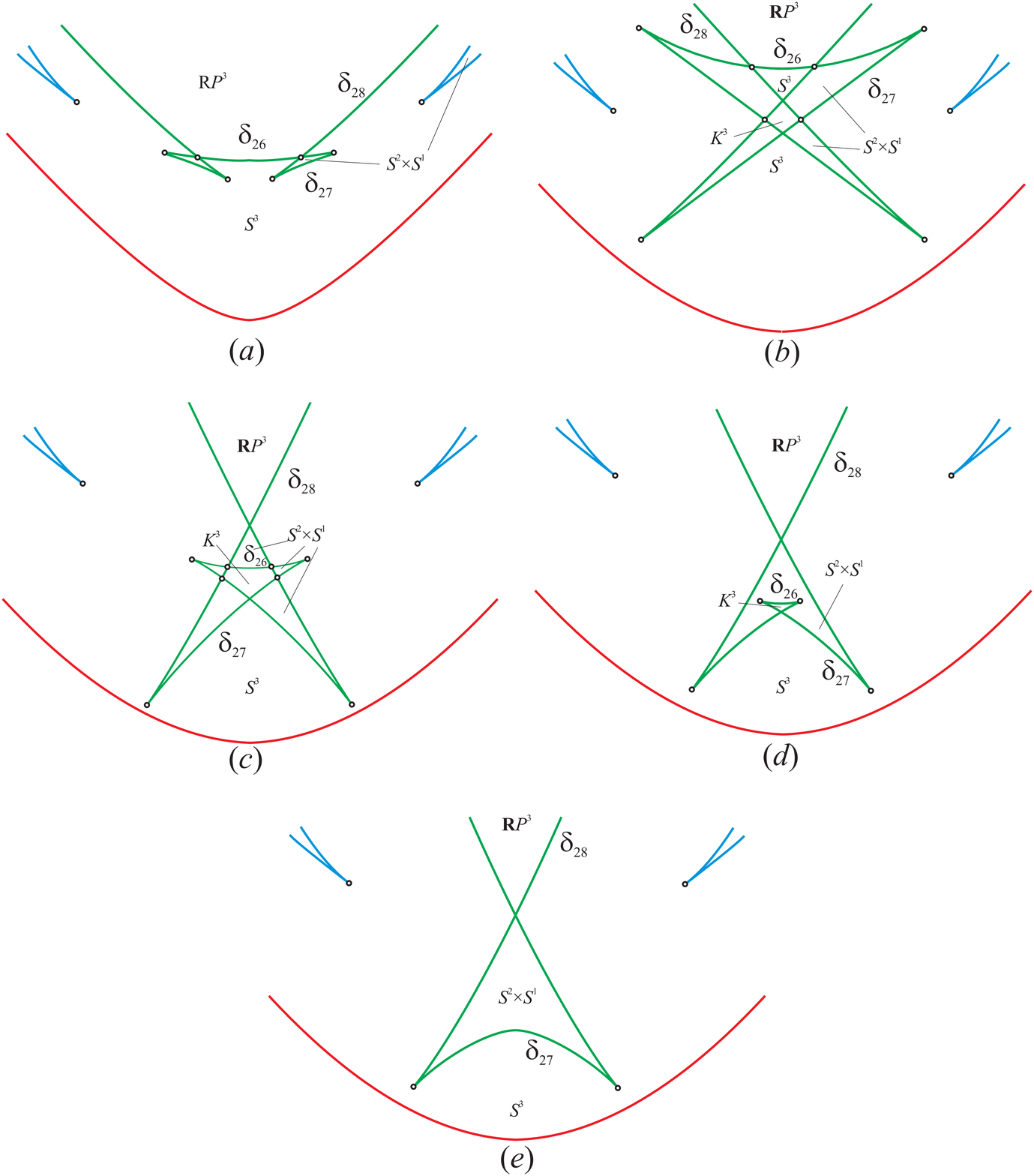}
\caption{Перестройки $\smale$ на оси $Oh$.}\label{fig_fig7}
\end{figure}

\FloatBarrier

Распространим на все фазовое пространство симметрию \eqref{eq4_17}
\begin{equation}\label{eq4_58}
    \bio: (\omega_1,\omega_2,\alpha_3) \mapsto (-\omega_1,-\omega_2,-\alpha_3).
\end{equation}
Она устанавливает изоморфизм фазовых потоков на $\iso$ и $\isom$.
Напомним, что в соответствии с договоренностью \eqref{eq4_19}, в пространстве $\mwide{\bR^2_{(\ell,h)}}=\bR^3_{(\ell,h,\ld)}$ возникает расширенная диаграмма Смейла
$$
\mwide{\smale}=\bigcup_\ld \smale(\ld){\times}\{\ld\}.
$$
Она делит $\bR^3_{(\ell,h,\ld)}$ на открытые связные компоненты, которые принято называть {\it камерами}. В силу симметрии $\bio$ объявим одной камерой также и объединение двух компонент, симметричных относительно плоскости $\ell=0$. Получим следующее утверждение.
\begin{theorem}\label{th7}
В случае Ковалевской\,--\,Яхья имеется семь структурно устойчивых по параметру $\ld$ диаграмм Смей\-ла $\smale(\ld)$. Разделяющими значениями параметра $\ld$ служат
\begin{equation}\notag
0,\quad \ld_1, \quad \ld_*=1/2^{3/4}, \quad \ld^*=2 \sqrt{2}/3^{3/4}, \quad 2\sqrt{\sqrt{2}-1}, \quad \ld_2, \quad \sqrt{2},
\end{equation}
где $\ld_1,\ld_2$ определены равенствами {\rm \eqref{eq4_56},\eqref{eq4_57}}. Расширенная диаграмма $\mwide{\smale}$ делит пространство $\bR^3{(\ell,h,\ld)}$ на восемь камер
$\mtA,\ldots, \mtH$ с непустыми многообразиями $\iso(\ld)$. Соответствующая информация представлена на рис.~{\rm \ref{fig_region_smale_ah})} и в табл.~{\rm \ref{table1}}.
\begin{center}
\small
\begin{tabular}{|c| l| c| c|}
\multicolumn{4}{r}{\fts{Таблица \myt\label{table1}}}\\
\hline
\begin{tabular}{c}\fts{Код}\\[-3pt]\fts{камеры}\end{tabular} &\begin{tabular}{c}\fts{Время}\\[-3pt]\fts{жизни по $\ld$}\end{tabular}
&\begin{tabular}{c}\fts{Компонент}\\[-3pt]\fts{в камере}\end{tabular}&\begin{tabular}{c}$\iso$\end{tabular} \\
\hline
$\mtA$ & $\ld\in[0,+\infty)$ & $1$ & $S^3$\\
\hline
$\mtB$ & $\ld\in[0,+\infty)$ & $2$ & $S^2{\times}S^1$\\
\hline
$\mtC$ & $\ld\in(0,\ld_*)$ & $2$ & $S^2{\times}S^1$\\
\hline
$\mtD$ & $\ld\in[0,\ld_1)$ & $2$ & $K^3$\\
\hline
$\mtE$ & $\ld\in[0,+\infty)$ & $1$ & ${\bR P^3}$\\
\hline
$\mtF$ & $\ld\in(\ld_*,\ld_2)$ & $2$ & $S^2{\times}S^1$\\
\hline
$\mtG$ & $\ld\in(\ld^*,\sqrt{2})$ & $1$ & $K^3$\\
\hline
$\mtH$ & $\ld\in(\ld^*,+\infty)$ & $2$ & $S^2{\times}S^1$\\
\hline
\end{tabular}
\end{center}

\normalsize

\end{theorem}

\begin{figure}[ht]
\centering
\includegraphics[width=0.7\textwidth,keepaspectratio]{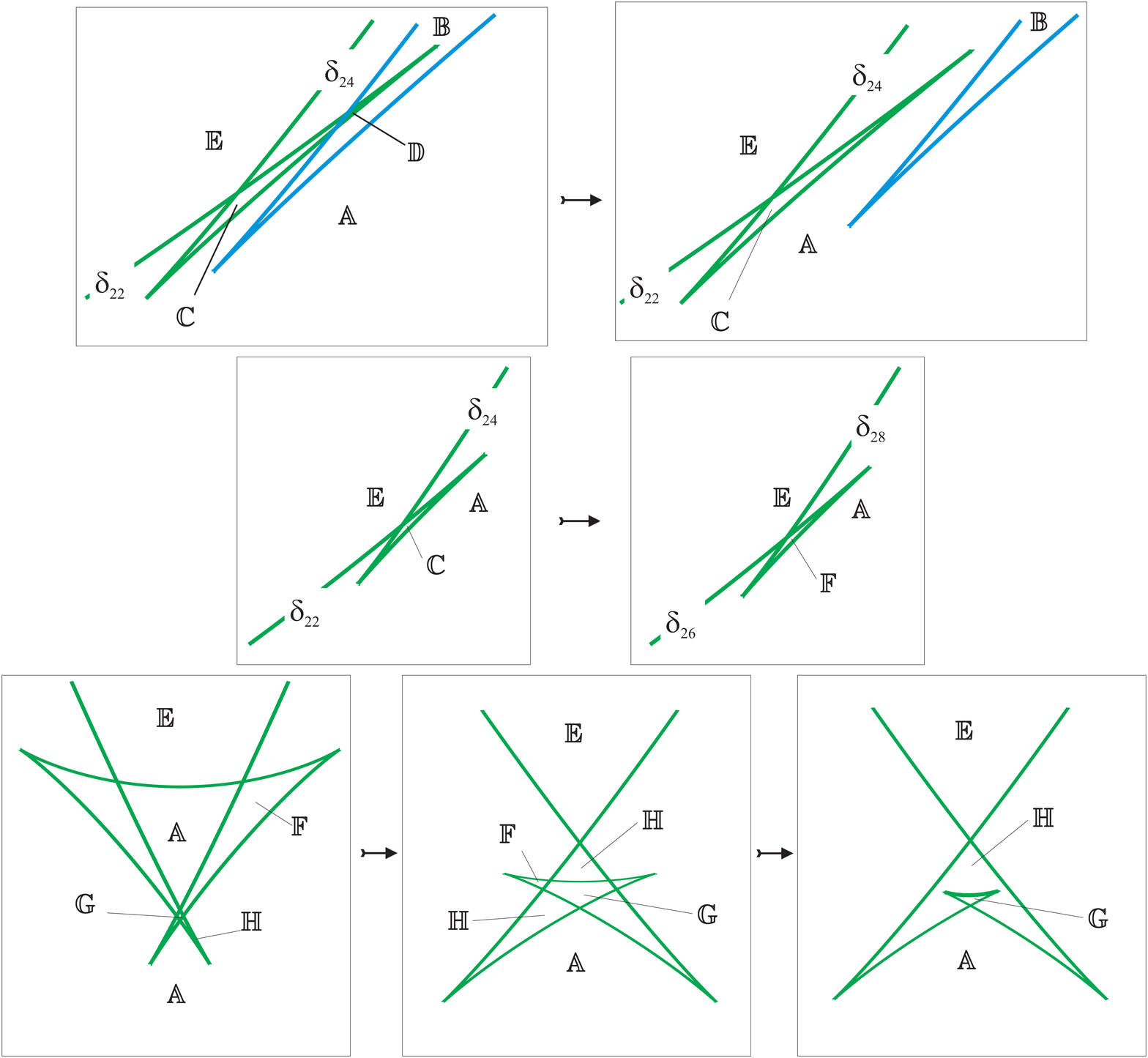}\\
\caption{Камеры диаграмм Смейла}\label{fig_region_smale_ah}
\end{figure}

\FloatBarrier

\vspace{5mm}
?????

Перейдем к изучению критических точек общего положения, имеющих в системе с двумя степенями свободы ранг 1.

\clearpage

\section{Классификация критических точек ранга $1$}\label{sec5}
\subsection{Формулы для вычисления типа}
Изучим тип критических точек ранга 1 в первой критической подсистеме, то есть точек множества $\mm\backslash \mct^0$.
\begin{proposition}[П.Е.\,Рябов \cite{RyabUdgu}]\label{propos11}
Тип критических точек ранга $1$ в первой критической подсистеме полностью определен собственными числами симплектического оператора $\A_{F_1}$, где
\begin{equation}\label{eq5_1}
    F_1=K-2\nu H,
\end{equation}
а неопределенный множитель в соответствующей точке равен $\nu = \po ^2$. Характеристический многочлен оператора $\A_F$ имеет вид
\begin{equation}\notag
\chi_{F_1}(\mu)=\mu^2 - 4 C_1,  \qquad C_1=\left[\frac{3}{2}s- (h - \frac{\ld^2}{2})\right] \left[2 s^2 - 2 (h + \frac{\ld^2}{2}) s+1 \right].
\end{equation}
Точки ранга $1$ вырождены тогда и только тогда, когда
\begin{equation}\label{eq5_2}
\left[\frac{3}{2}s - (h - \frac{\ld^2}{2})\right] \left[2 s^2 - 2 (h + \frac{\ld^2}{2}) s+1 \right]=0,
\end{equation}
и имеют тип ``центр'' при $C_1<0$ и тип ``седло'' при $C_1>0$.
\end{proposition}
\begin{proof}
В координатах $(\bo, \ba)$ порожденное функцией $F_1$ гамильтоново поле в точках \eqref{eq2_25} имеет вид
\begin{equation}\notag
\sgrad F_1 = 2(\po ^2-\nu) \left\{0,0, \sqrt{R(r)},r \sqrt{R(r)},-\frac{1}{2} R'(r),-\po   \right\}
\end{equation}
и обращается в нуль при $\nu =\po ^2$. Отметим, что при других значениях $\nu$ условие $\sgrad F_1=0$ влечет одновременное обращение в нуль $R(r),R'(r)$, что приводит к неподвижным точкам -- критическим точкам ранга 0.

Напомним, что согласно \eqref{eq3_13}
$$
    \ds s= h -\frac{\ld^2}{2}-\po ^2.
$$
Вычисление многочлена $\chi_{F_1}(\mu)$ в силу уравнений \eqref{eq2_25} многообразия $\mm$ труда не составляет. Таким образом, если $C_1 \ne 0$, то критические точки ранга 1 невырождены. Пусть $C_1=0$. Допустим, что в такой точке существует другая комбинация $G=\nu_1 K - 2\nu_2 H$, удовлетворяющая равенству $\sgrad G=0$. В силу предположения $\rk\{K,H\}|_{P_\ell^4}=1$ имеем пропорциональность неопределенных множителей в комбинациях
$$
dF=dK-2\nu dH |_{P_\ell^4} =0,\qquad dG=\nu_1 dK - 2\nu_2 dH |_{P_\ell^4}=0 \quad \Rightarrow \quad \nu_2 = \nu_1 \nu.
$$
Отсюда следует, что $\nu_1 \ne 0$, то есть характеристические числа оператора $\A_G$ пропорциональны (нулевым) характеристическим числам оператора $\A_{F_1}$ и потому также равны нулю. Следовательно, в точках \eqref{eq5_2} не существует интеграла, порождающего регулярный элемент в алгебре симплектических операторов.
\end{proof}

Отметим, что доказательство предложено потому, что здесь нельзя было воспользоваться непосредственно результатами работы \cite{KhND07}, где для построения системы $\mm$ использовалась функция \eqref{eq3_18}. Для нее корни характеристического многочлена обращаются в нуль также и на множестве $s=0$, что влечет равенство $\ell=0$. Появление такой ``посторонней'' особенности связано с тем, что в \cite{KhND07} фигурировал интеграл Реймана~--~Семенова-Тян-Шанского, вырождающийся в нашем случае в $L^2$.

\begin{proposition}[П.Е.\,Рябов \cite{RyabUdgu}]\label{propos12}
Тип критических точек ранга $1$ во второй и третьей критических подсистемах полностью определен собственными числами симплектического оператора $\A_{F_2}$, где
\begin{equation}\label{eq5_3}
    F_2=K + (2\ld^2-\frac{1}{s}) H,
\end{equation}
а $s$ -- значение, определяющее точку \eqref{eq2_28}. Характеристический многочлен оператора $\A_{F_2}$ имеет вид
\begin{equation}\notag
\begin{array}{l}
\chi_{F_2}(\mu)=\mu^2 -  C_2,  \\
\ds{C_2=-\frac{1}{s^3}\left(8\ld^2 s^3-1 \right) \left[2 s^2 - 2 (h + \frac{\ld^2}{2}) s+1 \right]=}
\\
\ds{\phantom{C_2} = \frac{2}{s^2}( 8 \ld^2 s^3-1) (2 \ld^2 s^2-s +2 \ell^2)}.
\end{array}
\end{equation}
Точки ранга $1$ в $\mn$ все невырождены и имеют тип ``центр''. Точки ранга $1$ в $\mo$ вырождены тогда и только тогда, когда
\begin{equation}\notag
\left(8\ld^2s^3-1\right) \left[2 s^2 - 2 (h + \frac{\ld^2}{2}) s+1 \right]=0,
\end{equation}
и имеют тип ``центр'' при $C_2<0$ и тип ``седло'' при $C_2>0$.
\end{proposition}

\begin{proof}
Полагая $F_2=K+\nu H$, найдем из \eqref{eq2_27}, \eqref{eq2_28}:
\begin{equation}\notag
\begin{array}{l}
  \sgrad F_2 = ( 2 \ld^2 s -1 - \nu s) \ds{ \left\{\frac{\Fun \vk  \rho  X}{\sqrt{s}}, -\frac{\rho \, n_1}{2 \vk  s}, \frac{2 \Fun \vk  Y}{\sqrt{s}}, \frac{ \Fun (n_1 + 2 \vk ^3 X Y)}{\vk  \sqrt{s}},\frac{ n_2}{s^2},  \frac{ \Fun (\ld  \rho  s X + \ell  Y)}{\sqrt{s} \vk } \right\}  },
\end{array}
\end{equation}
где
$$
n_1= -\ell  \rho  X + \ld  s Y + 2 \vk ^3 X Y, \qquad n_2=\vk  X + s X Y (3 \ell  \rho  - 4 \vk ^3 Y) + \ld  s^2 (2 - 3 Y^2).
$$
В предположении $2 \ld^2 s -1 - \nu s \ne 0$ этот вектор обращается в нуль лишь при условии
$$
\begin{array}{l}
64 \ell^6 \ld ^2 s^4 + 4 \ell^4 s^2 (-1 + 20 \ld ^2 s - 4 \ld ^4 s^2 + 48 \ld ^4 s^4) \\
\qquad - 4 \ell^2 s (1 - 8 \ld ^2 s + 4 \ld ^4 s^2 + 20 \ld ^2 s^3 - 4 \ld ^4 s^4 + 8 \ld ^6 s^5 - 48 \ld ^6 s^7) \\
\qquad - (1-4 s^2) (1 - 2 \ld ^2 s + 4 \ld ^4 s^4)^2=0.
\end{array}
$$
Его можно получить следующим образом. Заметим, что возможность $X=0, Y=\pm 1$ не приводит к нулевому вектору $\sgrad F_2$ даже в предположении $\rho=0$. Поэтому должно быть $\Fun=0$. Выразим отсюда $X$, подставим в выражения $n_1,n_2$ и найдем результант по $Y$. Получим искомое соотношение на $\ell,s$. Поскольку найденное условие не обращается в тождество на поверхностях $\wsi_{2,3}$, то оно определяет одномерное подмножество. С другой стороны, оно тождественно удовлетворено соответствующими выражениями $\ell$ и $s$ через $r,\ld$ из \eqref{eq4_11}, \eqref{eq4_15} и, следовательно, соответствует кривым $\delta_i$ $(i=1,2,3)$, то есть выполняется в критических точках ранга 0. Таким образом, в критических точках ранга 1 должно быть $2 \ld^2 s -1 - \nu s=0$ и функция \eqref{eq5_3} является искомым интегралом, определяющим тип этих точек.

Прямое вычисление характеристического многочлена приводит к требуемому выражению $\chi_{F_2}(\mu)$. Очевидно, при $s<0$ величина $C_2$ всегда отрицательна, а при $s>0$ может иметь разные знаки, что и определяет соответствующий тип особой точки. Тот факт, что при $C_2=0$ критические точки вырождены (то есть нельзя найти другого элемента в алгебре симплектических операторов, имеющего ненулевые собственные значения), доказывается так же, как в предложении~\ref{propos11} (фактически следует из того, что коэффициент при $K$ удалось выбрать равным 1).
\end{proof}

Мы видим, что для всех критических точек (как ранга 0, так и ранга 1) характеристический многочлен определяющего симплектического оператора зависит только от постоянных первых интегралов. Поэтому имеет место следующее важное свойство.

\begin{theorem}\label{th8}
Все критические точки одного ранга, лежащие на одном совместном уровне первых интегралов $L,H,K$, имеют один и тот же тип.
\end{theorem}

Отсюда, в частности, сразу же следует, что молекула, обозначенная в работах \cite{Gash5,GashDis} как $W_7$, в данной задаче невозможна. Ниже мы покажем, что, даже сделав скидку на погрешности изображения (атомы разных уровней интегралов на одной высоте на рисунке), такая молекула не реализуется по другой, уже существенной, причине.

Введем некоторую терминологию, связанную с детальным исследованием критических подсистем.

Известно \cite{BolFom}, что геометрическим особенностям бифуркационных диаграмм интегрируемых систем с $n$ степенями свободы отвечают критические точки ранга $m < n-1$ и вырожденные критические точки ранга $n-1$. Наличие полного взаимно однозначного соответствия здесь в общем случае не доказано. Дополнительные особенности возникают в семействах гамильтоновых систем, полученных процедурой факторизации из одной системы с симметрией. Примером этого является и рассматриваемая здесь задача. Поверхность $\wsa$ имеет особенность типа самопересечения, проявляющуюся в любых плоских ее сечениях, кроме сечений $\ell=\cons$ (кривая $s=0$ в представлении \eqref{eq3_7} и особая парабола в представлении \eqref{eq3_10}). В самих сечениях фиксированной постоянной площадей оказывается особым сечение $\ell=0$ опять же в силу наличия особой ``кратной'' параболы.

\begin{defin}\label{def9}
Невырожденную критическую точку ранга $1$ назовем кратной, если в окре\-ст\-но\-сти ее образа при отображении момента $J=L{\times}H{\times}K$ бифуркационная диаграмма $\Sigma$ этого отображения не является гладкой двумерной поверхностью.
\end{defin}

Здесь уместно напомнить, что рангом критической точки мы называем ее ранг в приведенной системе с двумя степенями свободы, что на единицу меньше ее ранга по отношению к отображению $J$.

Поскольку окрестность невырожденной критической точки ранга 1 в $\iso$ устроена как один из стандартных атомов и структурно не изменяется при малых изменениях $h$, то никаким локальным анализом особенности в приведенной системе с двумя степенями свободы на $\mPel$ кратную точку выявить нельзя. Можно лишь констатировать, что некоторой дуге бифуркационной диаграммы системы на $\mPel$ отвечает несколько критических окружностей. В этом случае распад кратной точки может произойти при возмущении $\ell$, то есть при переходе к другой системе с двумя степенями свободы. В общих случаях интегрируемости задач динамики твердого тела особым в этом смысле всегда является нулевой уровень циклического интеграла~$L$.

\begin{defin}\label{def10}
Назовем ключевым множеством в $\mP^5$ объединение всех критических точек ранга $0$, вырожденных и кратных критических точек ранга $1$. Обозначим ключевое множество через $\mK$. Пересечение $\mK_j =\mK \cap \mi_j$ назовем \textit{ключевым множеством} критической подсистемы $\mi_j$.
\end{defin}

\begin{defin}\label{def11}
Пусть $\Phi, \Psi$ -- первые интегралы \emph{(}возможно, частные\emph{)} критической подсистемы $\mi_j$, независимые почти всюду. Назовем $(\Phi,\Psi)$-диаграммой подсистемы $\mi_j$ образ ключевого множества $\mK_j$ при отображении $\Phi{\times}\Psi: \mi_j \to \bR^2$.
\end{defin}

\begin{defin}\label{def12}
При наличии некоторого интегрального отображения $\mi_j \to \bR^k$ его образ назовем допустимой областью и обозначим через $\mD_j$. Точки допустимого множества будем называть допустимыми.
\end{defin}

Иначе говоря, значения интегралов допустимы, если им соответствуют некоторые движения (интегральное многообразие не пусто). Из контекста всегда будет ясно, о каком отображении идет речь и в каком пространстве рассматриваются допустимые точки.

В работах \cite{mtt40,RyabHarlUdgu2} изучались $(S,H)$- и $(S,L)$-диаграммы критических подсистем. Здесь $S$~-- частный интеграл, возникающий из представления Лакса, постоянная которого является параметром на поверхностях в уравнениях \eqref{eq3_7}, \eqref{eq3_8}. Он же естественным образом возникает при введении функции с неопределенными множителями Лагранжа для описания критических подсистем как коэффициент при интеграле $K$ -- по формулам \eqref{eq3_18}, \eqref{eq3_19}, \eqref{eq3_20}. Для классификации бифуркаций, происходящих при пересечении точкой $(\ell,h,k)$ поверхностей $\wsi_j$ в $\bbI$, удобно выбрать такую плоскость констант функционально независимых интегралов критических подсистем, при котором допустимые точки этой плоскости находятся во взаимно однозначном соответствии с точками соответствующей бифуркационной поверхности (тогда, например, становится однозначной информация о количестве точек ранга 0 или окружностей ранга 1 в прообразе).

Рассмотрим также задачу классификации бифуркационных диаграмм $\Sigma_{HK}(\ell,\ld)$ отображений $H{\times}K$ приведенных систем на $\mPel$ \cite{RyabRCD,RyabHarlUdgu2} и бифуркационных диаграмм $\mSash$ отображений $L{\times}K$, ограниченных на изоэнергетические уровни $Q_h^4 = H^{-1}(h) \subset \mP^5$ \cite{mtt40}. Эти диаграммы представляют собой сечения плоскостями $\ell=\cons$ и, соответственно, $h=\cons$ бифуркационной диаграммы $\Sigma(\ld)$ общего интегрального отображения $J=L{\times}H{\times}K$. В силу наличия свободного физического параметра $\ld$ в цитированных работах решается вопрос о нахождении так называемого атласа -- разделяющего множества в соответствующей плоскости $(\ld,\ell)$ или $(\ld,h)$, при пересечении которого меняется строение таких сечений. Сформулируем общее утверждение, позволяющее в определенном смысле алгоритмизировать построение подобных атласов.

Пусть $F: \mP^5 \to \bR$ некоторый общий интеграл системы. Фиксируем $f\in\bR$ и рассмотрим отображение
\begin{equation}\notag
    J_{(F,f)} = J|_{F^{-1}(f)} : F^{-1}(f) \to \bR^2,
\end{equation}
где $\bR^2$ --- плоскость значений пары интегралов, дополняющих $F$ до полной инволютивной тройки почти всюду независимых интегралов. Пусть $\sigma_{(F,f)}(\ld)$ --- бифуркационная диаграмма отображения $J_{(F,f)}$. Ясно, что если пара интегралов $U,V$ дополняет $F$ до инволютивной тройки функционально независимых интегралов, то можно естественным образом отождествить $\sigma_{(F,f)}(\ld)$ с бифуркационной диаграммой $\Sigma_{UV}(f,\ld)$ ограничения отображения $U{\times}V$ на подмногообразие $F^{-1}(f)$. В частности, ниже будут рассмотрены варианты $F=L$, и тогда $\sigma_{(L,\ell)}(\ld)=\mSell$, и $F=H$, и в этом случае $\sigma_{(H,h)}(\ld)=\mSash$.

\begin{proposition}[М.П. Харламов \cite{Kh2011,KhVVMSH13}]\label{propos13}
Множество $\Theta_F$ в плоскости $(\ld,f)$, при переходе через которое меняется тип диаграммы $\sigma_{(F,f)}(\ld)$ \emph{(}в гладком смысле\emph{)}, состоит из пар $(\ld,f)$, где $f$ --- критическое значение ограничения функции $F$ на ключевое множество $\mK$ при заданном $\ld$.
\end{proposition}

При желании этому утверждению можно придать достаточную строгость, уточняя понятия однотипных диаграмм и критических значений функции на стратифицированном многообразии. В частности, заведомо предполагается, что интеграл $F$ выбран ``разумно'' и не имеет критических точек на регулярных уровнях отображения момента. Один из возможных подходов к системе соответствующих определений, использующий понятие оснащенных допустимых промежутков, приведен в работе \cite{RyabHarlUdgu2}. Практически же оно означает следующее. Для каждой критической подсистемы $\mi_j$ нужно рассмотреть ее $(G,F)$-диаграмму, где $G$ --- некоторый функционально независимый с $F$, возможно, частный, интеграл на $\mi_j$. Критические значения $F$ на $\mK$ соответствуют узловым точкам этой диаграммы (образам вырожденных критических точек ранга $0$), экстремальным значениям координаты $f$ на гладких участках -- образах множества вырожденных или кратных критических точек ранга $1$, а также всевозможным самопересечениям гладких участков диаграммы.

Множество $\Theta_F$ называют разделяющим множеством при классификации бифуркационных диаграмм $\sigma_{(F,f)}(\ld)$. Будем для краткости называть его $F$-атласом.


\subsection{Детализация. Первая критическая подсистема}

Для первой критической подсистемы ни одна из пар функционально независимых интегралов $(S,H)$ или $(S,L)$ не решает задачи обеспечения взаимно однозначного соответствия выбранной плоскости и поверхности $\wsa$. Действительно, на $(s,h)$-плоскости в {\it любую допустимую точку} с условием
$$
h-\frac{\ld^2}{2}-s>0
$$
отображаются {\it две} точки поверхности $\wsa$, а в $(s,\ell)$-представлении {\it весь допустимый сегмент} параболы
\begin{equation}\label{eq5_4}
k= 1+ (h -\frac{\ld^2}{2})^2, \qquad \ell=0, \qquad h \gs \frac{\ld^2}{2},
\end{equation}
отвечающей на поверхности $\wsa$ значению $s=0$, отображается в {\it одну} точку. Таким образом, соответствие либо двузначно на множестве полной меры, либо имеется ``плохая'' особенность~-- точка, в которую отображается бесконечно много точек поверхности. Заметим, что в силу наличия на поверхности $\wsa$ линии самопересечения задача взаимно однозначной параметризации и не может быть решена. Однако можно обеспечить параметризацию, в которой лишь каждой точке на линии самопересечения отвечала бы естественным образом пара точек на плоскости параметров. Для этого выберем на $\mm$ вместо $S$ частный интеграл
\begin{equation}\notag
P=\omega_1.
\end{equation}
Тогда формулы \eqref{eq2_30} дают взаимно однозначное соответствие допустимых областей на $(p,h)$-плоскости и поверхности $\wsa$, за одним лишь очевидным исключением -- во внутреннюю точку сегмента параболы \eqref{eq5_4}, то есть при $h>\ld^2/2$, отображаются две точки с противоположными знаками $p$. Далее мы рассмотрим диаграммы и ключевые множества подсистемы $\mm$ в терминах $(p,h)$, но, имея в виду сформулированную выше задачу классификацию бифуркационных диаграмм $\mSell$ и $\mSash$, мы также представим и описание $(S,H)$- и $(S,L)$-диаграмм для $\mm$.

Для вычисления топологических инвариантов таких, как графы Фоменко, в невырожденных критических точках ранга 1 нужно знать еще один показатель, а именно, индекс Морса\,--\,Ботта интеграла $K$. Он дает информацию о том, что происходит с критическими окружностями на изоэнергетической поверхности $\iso(\ld)$ при изменении значения дополнительного интеграла в сторону возрастания. В пространстве постоянных всех трех интегралов мы тем самым отслеживаем, какие явления происходят, когда прямая, параллельная $Ok$, протыкает соответствующую бифуркационную поверхность $\wsi_i$. Вид функций \eqref{eq5_1}, \eqref{eq5_3} позволяет явно провести соответствующие вычисления.

В точках $\mm$ найдем пару векторов, определяющих трансверсальную площадку к критической окружности (в критической точке ранга 1). Такая площадка может быть выбрана как ортогональная векторам $\grad \Gamma, \grad L, \grad H, \sgrad H$. Замечая, что на любой траектории \eqref{eq2_24}, \eqref{eq2_25} переменная $r$ осциллирует между корнями многочлена \eqref{eq2_25}, выберем на траектории точку $x_0$, в которой $R(r)=0$. Тогда касательная плоскость к трансверсальной площадке окажется натянутой на векторы
\begin{equation}\notag
\begin{array}{l}
  v_1 = \bigl(0,1,0,0,0,0\bigr), \qquad
  v_2=\bigl(\ld+r,0,-4\po ,2\po (\ld-r),0,2(h- \po ^2-\frac{r^2}{2})\bigr).
\end{array}
\end{equation}
Эти векторы получены из системы уравнений для $u=(u_1,\ldots,u_6)$
$$
u \cdot \grad \Gamma=0,\qquad u \cdot \grad L=0, \qquad u \cdot \grad H=0, \qquad u \cdot \sgrad H=0
$$
при условиях $(u_2,u_3)=(1,0)$  и   $(u_1,u_2)=(0,\ld+r)$.

В точке $x_0$ достигается условный экстремум функции $K$ на совместном уровне $\iso$ функций $\Gamma, L, H$ в $\mP^6$. В частности, $x_0$ есть критическая точка функции с неопределенными множителями Лагранжа
\begin{equation}\notag
K_1=K-2\nu H - \vk_1 L - \vk_2 \Gamma.
\end{equation}
Очевидно, часть этой функции, не содержащая функций Казимира $L,\Gamma$, совпадает с \eqref{eq5_1}, то есть $\nu=\po ^2$. Непосредственно вычисляется, что $\vk_1=4\po$, $\vk_2=1$.

Тип условного экстремума определяется ограничением второго дифференциала функции с множителями Лагранжа на касательное пространство к многообразию ограничений \cite{OptUpr}. Применительно к исследованию фазовой топологии гамильтоновых систем этот факт отметил и использовал А.А.\,Ошемков \cite{Oshem}.

Ограничение $d^2K_1$ на линейную оболочку векторов $v_1,v_2$ в силу выбора $R=0$ оказывается диагональной матрицей с элементами (собственными числами)
\begin{equation}\notag
\begin{array}{rcl}
  \mu_1 &=& 4[h-(\ld^2+\po ^2) + \ld r -\frac{1}{2} r^2] =2[ 2s -(\ld-r)^2], \\[2mm]
  \mu_2 &=& - 8(h-\frac{\ld^2}{2} -3\po ^2)[h-\frac{\ld^2}{2} -3\po ^2-\frac{1}{2}(\ld+r)^2]= \\[2mm]
  {}    &=& \ds{- 32(h-\frac{\ld^2}{2} -\frac{3}{2} s)[h-\frac{\ld^2}{2} -\frac{3}{2} s-(\ld+r)^2]}.
\end{array}
\end{equation}
В частности, произведение
\begin{equation}\notag
\begin{array}{rcl}
  \mu_1 \mu_2 &=& 16 (-2 h + \ld^2 + 6 \po ^2) [1 - (2 h - \ld^2 - 2 \po ^2) (\ld^2 + \po ^2)] =\\[2mm]
   {}         &=& -32 (2 h - \ld^2 - 3 s) [(2 h + \ld^2 - 2 s) s-1]  =\\[2mm]
   {}         &=& \ds{-64 \left[\frac{3}{2}s- (h - \frac{\ld^2}{2})\right] \left[2 s^2 - 2 (h + \frac{\ld^2}{2}) s+1 \right]}
\end{array}
\end{equation}
от $r$ не зависит и определяется расположением точки $(s,h)$ относительно кривых вырождения критических точек ранга 1, заданных уравнением \eqref{eq5_2}.
В силу этого, на невырожденных траекториях величины $\mu_1,\mu_2$, которые будем называть показателями Морса\,--\,Ботта, в ноль не обращаются, а значит, не меняют знака.

\begin{proposition}\label{propos14}
При возрастании интеграла $K$ на изоэнергетическом уровне $\iso$ имеем следующие бифуркации в точках критической подсистемы $\mm$ на невырожденных критических окружностях:

$1)$ для эллиптических траекторий $($тип ``центр''$)$ --- рождение тора при $\mu_1>0,\mu_2>0$ $($атом $A$ с ребром вверх$)$, исчезновение тора при $\mu_1 < 0,\mu_2<0$ $($атом $A$ с ребром вниз$);$

$2)$ для одной гиперболической траектории на критическом уровне $K$ при $\mu_1 > 0,\mu_2<0$ --- атом $B$ с ``внешним'' ребром $($``головой''$)$ вверх и парой ``внутренних'' ребер $($``ног''$)$ вниз, при $\mu_1 < 0,\mu_2 > 0$ --- атом $B$ с ``внешним'' ребром $($``головой''$)$ вниз и парой ``внутренних'' ребер $($``ног''$)$ вверх$;$

$3)$ для двух гиперболических траекторий на критическом уровне $K$ при совпадении знаков двух пар $(\mu_1,\mu_2)$ --- два атома $B$, у которых направление ``внешнего'' ребра определяется по тому же правилу $($обе ``головы'' вверх при $\mu_1 > 0,\mu_2<0$, обе ``головы'' вниз при ${\mu_1 < 0}, {\mu_2>0)};$

$4)$ для двух гиперболических траекторий на критическом уровне $K$ при разных сочетаниях знаков в парах $(\mu_1,\mu_2)$ -- два атома $A^*$.
\end{proposition}

\begin{proof}
Для эллиптических траекторий утверждение очевидно (функция $K$ имеет локальный минимум или максимум на трансверсальной площадке).

Можно показать, что для гиперболических траекторий вектор $v_1$ направлен во внешнюю часть ``восьмерки'': направление оси $O\omega_2$ отвечает за переход с критической поверхности к объемлющему ее тору. Такое понимание легко получить, анализируя, например, проекции интегральных многообразий на плоскость $O\omega_1 \omega_2$ (частично такое исследование выполнено в \cite{GashDis}). Видно, что, как и в классической задаче \cite{Appel, Ipat}, разрывов проекции никогда не происходит в направлении оси $O\omega_2$. Если теперь $\mu_1>0$, то это означает, что на трансверсальной площадке $K$ растет к ``внешней'' окружности и убывает к паре ``внутренних''.

Если на двух гиперболических окружностях наборы знаков в парах $(\mu_1,\mu_2)$ различны, то, предполагая наличие двух атомов $B$ с противоположными направлениями ``голов'', получим бифуркацию трех торов в три. Как отмечено в \cite{Gash2,Gash5}, количество торов на регулярном уровне может быть равным лишь 1, 2 или 4 (для всех камер в $\bbI\setminus \Sigma$ количество регулярных торов будет строго установлено ниже с применением только очевидных атомов типа $A$). Поэтому в данном случае мы имели бы бифуркацию четырех торов в четыре. Однако, как мы увидим ниже,  таких примыкающих друг к другу камер в данной задаче нет. Если предположить, что в рассматриваемой точке имеется атом $C_2$, то в аналитическом решении \eqref{eq2_34}~-- \eqref{eq2_35} существовала бы гетероклиническая траектория, которой также здесь нет. Значит, в такой точке имеется два атома $A^*$. \end{proof}

Критерий существования движений в системе $\mm$ очевиден (предложение 1 работы \cite{RyabHarlUdgu2}). Поскольку в $(p,h)$-плоскости нет ``двойных'' точек, мы легко дополним его информацией о количестве критических окружностей на соответствующем уровне интегралов.
\begin{proposition}\label{propos15}
При заданных $p,h$ вещественные решения \eqref{eq2_25} существуют тогда и только тогда, когда $R(r)\gs 0$ для некоторого $r\in \bR$. Если при этом все корни $R(r)$ простые \emph{(}то есть на заданном интегральном уровне нет критических точек ранга $0$\emph{)}, то количество периодических решений равно количеству промежутков неотрицательности $R(r)$.
\end{proposition}

\begin{remark}\label{rem6}
Напомним еще раз, что при переходе в пространство $\bbI$ в силу формул \eqref{eq2_30} пара точек из прообраза параболы \eqref{eq5_4} с противоположными знаками $p$ переходит в одну точку. В работе {\rm \cite{Gash5}} этот факт отмечен как особый случай.
\end{remark}

\begin{remark}\label{rem7}
Очевидно, предложение $\ref{propos15}$ допускает очевидное обобщение. Если $R(r)$ имеет кратный корень $($а тогда, как показано выше, он единственный$)$, то количество периодических решений равно количеству промежутков неотрицательности $R(r)$ между \emph{простыми} корнями. Теперь, учитывая и предложение $\ref{propos10}$, мы получаем всю информацию о полном составе критических движений на совместном уровне общих интегралов, содержащих критическую точку ранга $0$ соответствующего класса, так как количество и расположение корней многочлена $R(r)$ легко устанавливается для всех классов, исходя из зависимостей $h(r,\ld)$ и $p(r,\ld)$, данных формулами \eqref{eq4_11}, \eqref{eq4_12}. Эта информация приведена в дальнейшем при исследовании топологии полных уровней интегралов для точек ранга 0 в третьем столбце табл.~$\ref{table42}$.
\end{remark}

Введем следующие обозначения:
\begin{equation}\label{eq5_5}
\begin{array}{c}
    \varphi_{\pm}(r)=\ds{\frac{1}{2}\bigl[ r (\ld-r) \pm \frac{2r-\ld}{r-\ld}D\bigr]},\qquad
  \psi_{\pm}(r)=\ds{\frac{r}{2}\bigl[ -r \pm \frac{1}{r-\ld}D\bigr]}, \\[3mm]
  D=|d|=\sqrt{\mathstrut 4+r^2(r-\ld)^2}>0.
\end{array}
\end{equation}
Обозначим через $h_C(\ld)$ значение энергии в вырожденной критической точке $\pi_{31}$, заданной формулой \eqref{eq4_38}. Смысл индекса $C$ будет объяснен в дальнейшем, так как через $C$ будет обозначен образ этой вырожденной точки в диаграммах критических подсистем. При заданном $\ld>0$ значение $h_C$ определяется согласно \eqref{eq4_38}, \eqref{eq4_11} так:
\begin{equation}\label{eq5_6}
\left\{
\begin{array}{l}
3x^4-{2x^3}\ld-4=0 \\
x \in (-\sqrt[4]{4/3},0)
\end{array}\right. \; \Rightarrow \; h_C=\frac{3}{8} x^2+\frac{2}{x^6}.
\end{equation}

\begin{theorem}\label{th9}
$(P,H)$-диаграмма критической подсистемы $\mm$ состоит из следующих множеств
\begin{equation}\label{eq5_7}
    \begin{array}{lll}
      \delta_1: & p^2= \psi_-(r), \quad h= \varphi_-(r), & r \in [0,\ld);\\[3mm]
      \delta_2: & p^2= \psi_+(r), \quad h= \varphi_+(r), & r \in (-\infty,0];\\[3mm]
      \delta_3: & p^2= \psi_+(r), \quad h= \varphi_+(r), & r \in (\ld,+\infty);\\[3mm]
      \gaa: & \ds{h-(3p^2+\frac{\ld^2}{2})=0}, & \ds{\frac{\ld^2}{2} \ls h \ls h_C(\ld)}; \\[3mm]
      \gan: & \ds{h-(p^2+\frac{\ld^2}{2})-\frac{1}{2(\ld^2+p^2)}=0}, & p\in \bR, \\[3mm]
      \ell_0: & \ds{h-(p^2+\frac{\ld^2}{2})=0}, & p \in \bR.
    \end{array}
\end{equation}
При этом границами допустимой области $\mD_1$ служат связная кривая $\delta_1$ и две компоненты связности кривой $\delta_3$. Кривая $\delta_2$ связна и
разбивает допустимую область на две подобласти. Точкам области, лежащей ниже кривой $\delta_2$, соответствует одна критическая окружность, а точкам области, лежащей выше кривой $\delta_2$, соответствуют две критических окружности.
\end{theorem}

\begin{proof}
Поскольку все критические точки ранга 1 принадлежат $\mm$, а значения на них функций $P,H$ даны в формулах \eqref{eq4_8}, \eqref{eq4_11}, то с учетом обозначений \eqref{eq4_20} образы этих точек -- это кривые $\delta_i$ ($i=1,2,3$). Кривые $\gaa$, $\gan$ получим из предложения~\ref{propos11} с учетом \eqref{eq3_13}. Кривая $\ell_0$ соответствует особой параболе \eqref{eq5_4}, которая, как отмечалось, есть образ множества кратных точек. Остается уточнить условия существования движений в системе $\mm$. Они существуют тогда и только тогда, когда многочлен $R(r)$, заданный выражением \eqref{eq2_26}, имеет промежутки неотрицательности. Поскольку его дискриминантным множеством являются в точности кривые $\delta_i$, получаем, что такие промежутки существуют между кривыми $\delta_1, \delta_3$, а ниже $\delta_1$ (по направлению оси $Oh$) и выше $\delta_3$ движений нет. При этом ниже кривой $\delta_2$ многочлен $R$ имеет два вещественных корня, поэтому такой промежуток один (одна критическая окружность), а выше этой кривой корней четыре, поэтому промежутков два (две критические окружности). Часть кривой $\gaa$ вырождения критических точек ранга 1, попадающая в допустимую область, ограничена  снизу естественной границей $h=\ld^2/2$, а сверху -- точкой пересечения с кривой $\delta_3$, ордината которой -- это описанное выше значение $h_C(\ld)$.
\end{proof}

\begin{theorem}\label{th10}
При неотрицательных $\ld$ перестройки $(P,H)$-диаграммы критической подсистемы $\mm$ происходят при следующих значениях параметра:
$$
0, \, \ld_*=1/2^{3/4}, \, 1, \, \ld^*=2\sqrt{2}/3^{3/4},\, \sqrt{2}.
$$
\end{theorem}
\begin{proof}
В силу выбора частных интегралов, обеспечивающих, в частности, взаимно однозначное соответствие точек ключевого множества и их образов (за указанными выше очевидными исключениями), перестройки диаграмм происходят одновременно с перестройками множества вырожденных точек ранга 0 и кратных точек при изменении $\ld$. Поскольку все кривые вырожденных и кратных точек известны (см. рис.~\ref{fig_RazdCrit0}), то эти перестройки таковы: при $\ld=0$ исчезает точка $\vpi_{21}$ (соответствующая кривая уходит в бесконечность) и сливаются точки $\vpi_{22},\vpi_{24}$; при $\ld=\ld_*$ в узловой точке $P_0$ сливаются все кривые вырожденных точек класса $\delta_2$; при $\ld=1$ исчезает точка $\vpi_{21}$ (заканчивается соответствующая кривая); при $\ld=\ld^*$ (минимум $\ld$ на $\ell_0$) возникают кратные точки; при $\ld=\sqrt{2}$ исчезает точка $\vpi_{23}$ (заканчивается соответствующая кривая). Других перестроек нет.
\end{proof}

\begin{figure}[!ht]
\centering
\includegraphics[width=\textwidth,keepaspectratio]{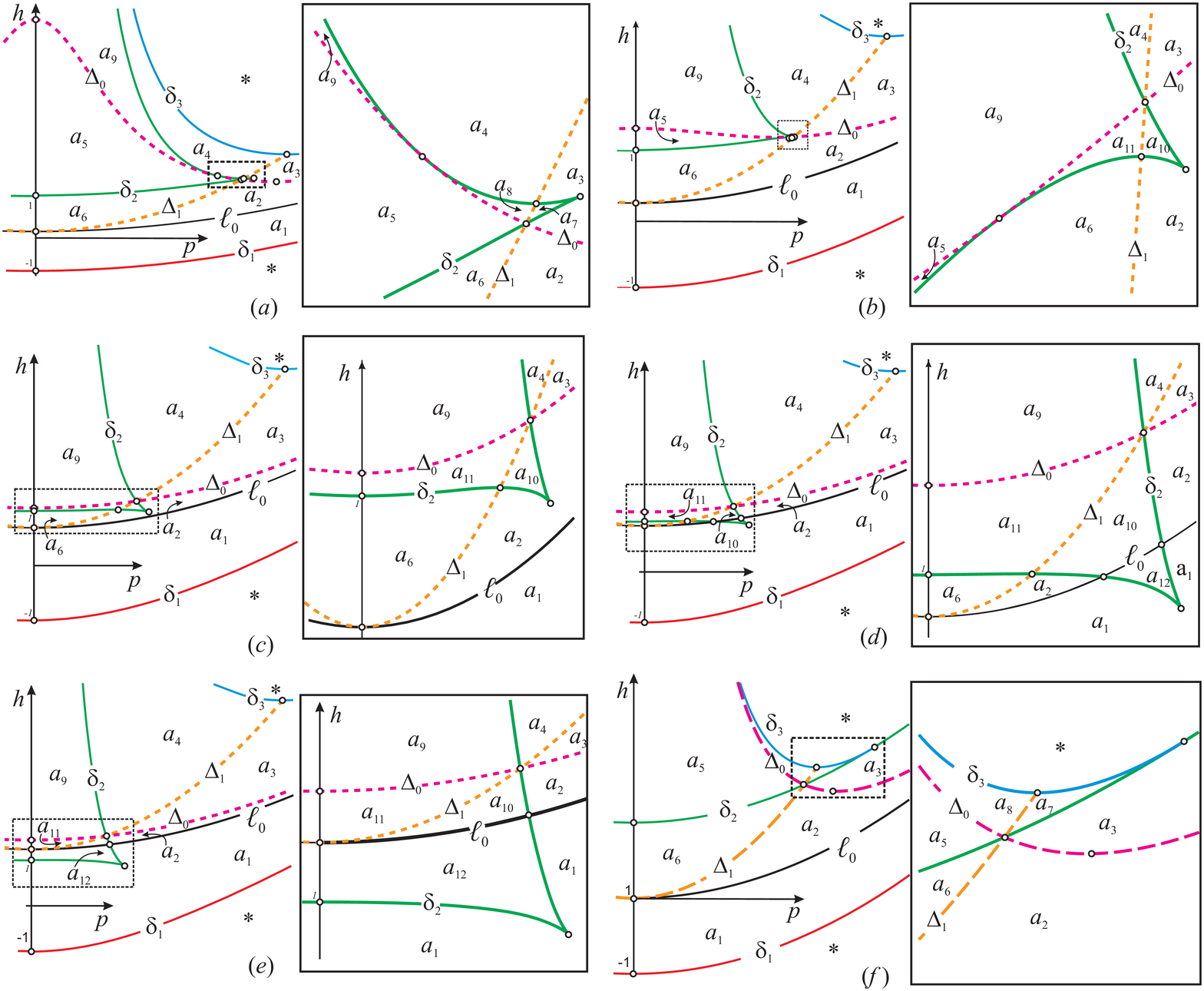}%
\caption{$(P,H)$-диаграммы системы $\mm$ и увеличенные фрагменты.}\label{fig_fig_sys1HP_red}
\end{figure}

На рис.~\ref{fig_fig_sys1HP_red},$(a)-(e)$ показаны $(P,H)$-диаграммы подсистемы $\mm$, устойчивые по па\-ра\-мет\-ру~$\ld$: $(a)$~${0<\ld<\ld_*}$; $(b)$~$\ld_*<\ld<1$; $(c)$~$1 <\ld< \ld^*$; $(d)$~$\ld^* <\ld< \sqrt{2}$; $(e)$~$\ld > \sqrt{2}$.

Символами $a_1 - a_{12}$ обозначены связные компоненты в пространстве $\{(p,h,\ld)\}$, на которые допустимая область делится расширенной $(P,H)$-диаграммой.
Напомним наличие симметрии \eqref{eq4_58}. Области, переходящие друг в друга при таком отображении, обозначены одинаково.

С точки зрения типов критических точек переход через значение $\ld=\ld^*$ никаких качественных изменений не влечет, точки областей $a_{10},a_{12}$ и их общей границы на $\ell_0$ одинаковы. При $\ld=\ld^*$ возникает пересечение кривых $\delta_2$ и $\ell_0$, что демонстрирует момент выхода объединенной области $a_{10} \cup a_{12}$ на многообразие $L=0$ и позволяет сравнить бифуркации и атомы с более подробными исследованиями случая $L=0$, выполненными П.Е.\,Рябовым и П.В.\,Морозовым. Также с локальной точки зрения одинаковы критические точки областей $a_1,a_2$ и их общей границы $\ell_0$. Однако точки на бифуркационной диаграмме $\Sigma$ в $\bbI$, порожденные этими областями, имеют нелокальное отличие, которое в терминах параметров \eqref{eq2_33} установлено в работе \cite{Gash1}, а именно, наличие или отсутствие на полном совместном уровне общих интегралов регулярных торов (двухчастотных движений). Об этом будет сказано ниже.

На плоских сечениях $h=\cons$ или $\ell=\cons$ (изоэнергетических бифуркационных диаграммах или бифуркационных диаграммах приведенных систем с двумя степенями свободы на $\mPel$) поверхностей $\wsi_i$ областям $a_j$ отвечают гладкие дуги, на которых полностью сохраняется как тип и количество особых точек, так и характер бифуркации, то есть атом, возникающий вдоль малого отрезка, трансверсально пересекающего $\wsi_i$ в соответствующей точке.

Для дополнительной информации на рис.~\ref{fig_fig_sys1HP_red},$(f)$ показана диаграмма первой критической подсистемы в классической задаче ($\ld=0$), отвечающая особо замечательным движениям четвертого класса Аппельрота. Она состоит из двух парабол
$$
h=\pm 1+p^2,
$$
кривой, полученной из $\delta_3$ и части $\delta_2$
$$
p=\pm \sqrt{\frac{1}{2}(\sqrt{4+r^4}-r^2)}, \qquad h=\frac{1}{2}(2\sqrt{4+r^4}+r^2),
$$
и предельных кривых
$$
\begin{array}{l}
  \gan: \, \ds{h=\frac{1}{2}(2p^2+\frac{1}{p^2}), \quad p\in \bR\backslash \{0\},} \qquad  \gaa: \, h=3p^2, \quad p\in \left[0, \frac{1}{\sqrt[4]{3}}\right].
\end{array}
$$


Анализируя все признаки, вычисленные для областей $a_j$, сведем информацию в табл.~\ref{table2}. Здесь введены следующие обозначения для атомов $A,B$: индекс ``+'' означает увеличение числа торов при увеличении $K$, а индекс ``$-$'', соответственно, означает уменьшение. Иначе говоря, $A_-$ -- ребро заканчивается вверху, $A_+$ -- ребро заканчивается внизу, $B_-$ -- ``внешняя'' окружность ``восьмерки'' (``голова'') вверху, $B_+$ -- ``внешняя'' окружность ``восьмерки'' (``голова'') внизу. Как видно, все области, кроме одной ($a_4$), имеют выход в классическую задачу Ковалевской ($\ld=0$) и в задачу с гиростатическим моментом при нулевой постоянной площадей ($\ell=0$).
Для этих случаев бифуркации установлены в работах \cite{KhDan83,RyabRCD} (подробное изложение для случая Ковалевской имеется также в \cite{KhPMM83,KhBook88}), метки на соответствующих молекулах вычислены в \cite{BRF,Mor}. Поэтому в графе ``Аналоги'' даны ссылки на обозначения участков \cite{KhBook88,BRF,Mor} или путей, пересекающих соответствующие участки \cite{KhDan83,RyabRCD}, используемые в этих работах. При наличии аналога все атомы описаны в \cite{KhDan83,KhPMM83,RyabRCD}. В двух первых работах был создан и язык для описания атомов, включая учет направленности для несимметричных атомов (само понятие атома тогда, естественно, еще не существовало). Сопоставляя с аналогами, рассмотренными в \cite{BRF,Mor}, можно указать и большинство меченых молекул.

\begin{center}
\small
\begin{longtable}{|c|c| c| c| c|c|}
\multicolumn{6}{r}{\fts{Таблица \myt\label{table2}}}\\
\hline
\begin{tabular}{c}\fts{Область}\\[-3pt]\fts{(время жизни)} \end{tabular}
&\begin{tabular}{c}\fts{К-во}\\[-3pt]\fts{окр-стей}\end{tabular}&\begin{tabular}{c}\fts{Показатели}\\[-3pt]\fts{Морса--Ботта}\end{tabular}
&\begin{tabular}{c}\fts{Выход на}\\[-3pt]\fts{$\ld=0$/$\ell=0$}\end{tabular}
&\begin{tabular}{c}\fts{Атом}\end{tabular} &\begin{tabular}{c}\fts{Аналоги}\end{tabular}\\
\hline\endfirsthead%
\multicolumn{6}{r}{\fts{Таблица \ref{table2} (продолжение)}}\\
\hline
\begin{tabular}{c}\fts{Область}\\[-3pt]\fts{(время жизни)} \end{tabular}
&\begin{tabular}{c}\fts{К-во}\\[-3pt]\fts{окр-стей}\end{tabular}&\begin{tabular}{c}\fts{Показатели}\\[-3pt]\fts{Морса--Ботта}\end{tabular}
&\begin{tabular}{c}\fts{Выход на}\\[-3pt]\fts{$\ld=0$/$\ell=0$}\end{tabular}
&\begin{tabular}{c}\fts{Атом}\end{tabular} &\begin{tabular}{c}\fts{Аналоги}\end{tabular}\\
\hline\endhead
\begin{tabular}{c}$a_1$\\($ 0 \ls \ld <+\infty$) \end{tabular} &{1}& {($-\;-$)} & Да/Да &$A_-$ & \begin{tabular}{l} $2,3$ \cite[Рис.\,6.3]{KhBook88}\\ $a_1, a_2$ \cite[Рис.\,2]{RyabRCD} \\$\gamma_1,\gamma_4$ \cite[Рис.\,11]{BRF}\\$\alpha_2,\alpha_3$ \cite[Рис.\,1]{Mor}\end{tabular}\\
\hline
\begin{tabular}{c}$a_2$\\($ 0 \ls \ld <+\infty$) \end{tabular} &{1}& {($-\;-$)} & Да/Да &$A_-$ & \begin{tabular}{l} $3,3^\prime$ \cite[Рис.\,6.3]{KhBook88}\\ $a_2$ \cite[Рис.\,2]{RyabRCD} \\$\gamma_1,\gamma_4$ \cite[Рис.\,11]{BRF}\\$\alpha_2,\alpha_3$ \cite[Рис.\,1]{Mor}\end{tabular}\\
\hline
\begin{tabular}{c}$a_3$\\($ 0 \ls \ld <+\infty$) \end{tabular} &{1}& {($+\;-$)} & Да/Нет &$B_-$ & \begin{tabular}{l} 9 \cite[Рис.\,6.3]{KhBook88}\\$\gamma_5$ \cite[Рис.\,11]{BRF}\end{tabular}\\
\hline
\begin{tabular}{c}$a_4$\\($ 0 < \ld < +\infty$) \end{tabular} &{1}& {($+\;+$)} & Нет/Нет &$A_+$ & {Отсутств.}\\
\hline
\begin{tabular}{c}$a_5$\\($ 0 \ls \ld < 1 $) \end{tabular} &{2}& {\begin{tabular}{c}($+\;-$),($-\;+$)\end{tabular}} & Да/Да &$2A^*$ & \begin{tabular}{l} 6 \cite[Рис.\,6.3]{KhBook88}\\ $a_4$ \cite[Рис.\,2]{RyabRCD}\\$\gamma_2$ \cite[Рис.\,11]{BRF}\\$\delta_1,\delta_2$ \cite[Рис.\,1]{Mor} \end{tabular}\\
\hline
\begin{tabular}{c}$a_6$\\($ 0 \ls \ld <\sqrt{2} $) \end{tabular} &{1}& {($-\;+$)} & Да/Да &$B_+$& \begin{tabular}{l} 5  \cite[Рис.\,6.3]{KhBook88}\\ $b_2$ \cite[Рис.\,3]{RyabRCD}\\$\gamma_3$ \cite[Рис.\,11]{BRF}\\$\beta_1$ \cite[Рис.\,1]{Mor} \end{tabular}\\
\hline
\begin{tabular}{c}$a_7$\\($ 0 \ls \ld < \ld_*$) \end{tabular} &{2}& {\begin{tabular}{c}($+\;-$),($+\;-$)\end{tabular}} & Да/Нет &$2B_-$& \begin{tabular}{l} Д\,\cite[Рис.\,2]{KhDan83}\\$\gamma_6$ \cite[Рис.\,11]{BRF}\end{tabular}\\
\hline
\begin{tabular}{c}$a_8$\\($ 0 \ls \ld < \ld_*$) \end{tabular} &{2}& {\begin{tabular}{c}($+\;+$),($+\;+$)\end{tabular}} & Да/Нет &$2A_+$& \begin{tabular}{l} E\,\cite[Рис.\,2]{KhDan83}\\$\gamma_7$ \cite[Рис.\,11]{BRF}\end{tabular}\\
\hline
\begin{tabular}{c}$a_9$\\($ 0 < \ld <+\infty$) \end{tabular} &{2}& {\begin{tabular}{c}($+\;+$),($-\;-$)\end{tabular}} & Нет/Да &$A_+, A_-$& \begin{tabular}{l} $a_5$ \cite[Рис.\,2]{RyabRCD}\\$\alpha_5,\alpha_6$ \cite[Рис.\,1]{Mor} \end{tabular}\\
\hline
\begin{tabular}{c}$a_{10}$\\($ \ld_* < \ld <+\infty$) \end{tabular} &{2}& {\begin{tabular}{c}($-\;-$),($-\;-$)\end{tabular}} & Нет/Да &$2A_-$& \begin{tabular}{l} $c_3, c_4$ \cite[Рис.\,4]{RyabRCD}\\$\alpha_3,\alpha_8$\cite[Рис.\,1]{Mor} \end{tabular}\\
\hline
\begin{tabular}{c}$a_{11}$\\($ \ld_*< \ld < +\infty $) \end{tabular} &{2}& {\begin{tabular}{c}($-\;+$),($-\;+$)\end{tabular}} & Нет/Да &$2B_+$& \begin{tabular}{l} $b_4$ \cite[Рис.\,3]{RyabRCD}\\$\beta_5,\beta_6$ \cite[Рис.\,1]{Mor} \end{tabular}\\
\hline
\begin{tabular}{c}$a_{12}$\\($ \ld^* < \ld <+\infty$) \end{tabular} &{2}& {\begin{tabular}{c}($-\;-$),($-\;-$)\end{tabular}} & Нет/Да &$2A_-$& \begin{tabular}{l} $d_2, d_3$ \cite[Рис.\,5]{RyabRCD}\\$\alpha_3,\alpha_8, \alpha_9,\alpha_{10}$\\ \cite[Рис.\,1]{Mor} \end{tabular}\\
\hline
\end{longtable}
\end{center}

Также теперь имеется возможность прояснить связь между параметрами \eqref{eq2_33}, условиями \eqref{eq2_36} и приведенной здесь классификацией точек первой критической подсистемы.
Для экономии обозначений будем под $\ell_0,\gaa,\gan$ понимать значение левых частей уравнений соответствующих кривых в записи \eqref{eq5_7} (в дальнейшем это не приведет к путанице). Тогда  в обозначениях \eqref{eq2_33}
\begin{equation}\label{eq5_8}
\begin{array}{l}
  L_1=\ell_0/8, \qquad L_2=\gaa/2, \qquad L_3=\gan/2.
\end{array}
\end{equation}
Из \eqref{eq5_5}, \eqref{eq5_7} находим, что на кривой $\delta_1$ имеется однозначная зависимость
$$
h=H_1(p), \quad p \in \bR,
$$
а на кривой $\delta_3$ однозначная зависимость определена при ненулевых $p$:
$$
h=H_3(p), \quad p \in \bR \backslash \{0\},
$$
причем для всех $p\ne 0$ имеем $H_3(p)>H_1(p)$. Кривые $\delta_i$ в целом образуют дискриминантное множество многочлена \eqref{eq2_26}. При
\begin{equation}\label{eq5_9}
h \in (-\infty,H_1(p)) \cup (H_3(p), +\infty)
\end{equation}
он вещественных корней не имеет, поэтому и невозможны движения на $\mm$. Покажем, что при соотношениях \eqref{eq2_30} в областях \eqref{eq5_9} движения невозможны и на всем многообразии $G^4$ (напомним, что так обозначен полный прообраз поверхности $\wsa$ в $\mP^5$ при отображении момента~$J$). Для этого достаточно заметить, что в точке зависимости ограничений функций $H,P=\omega_1$ на многообразие $G^4$ ранг $J$ не превышает двух, поэтому такая точка является критической ранга 0 или 1 и обязательно принадлежит $\mm$. Поэтому наибольшее и наименьшее значения $H$ при фиксированном $P$ на $G^4$ совпадают с аналогичными значениями на $\mm$:
$$
\max_{G^4 \cap \{P=p\}} H = \max_{\mm \cap \{P=p\}} H = H_3(p), \quad \min_{G^4 \cap \{P=p\}} H = \min_{\mm \cap \{P=p\}} H = H_1(p).
$$
В связи с этим введем два дополнительных параметра
$$
L_4=h-H_1(p), \qquad L_5=h-H_3(p)
$$
и обозначим области на плоскости $(p,h)$
$$
\overline{a_1} = \{(p,h): L_4<0\}, \qquad \overline{a_2} = \{(p,h): L_5>0\}.
$$

Из доказанного, в частности, вытекает следующее утверждение о существовании вырожденных критических движений ранга 1 в первой подсистеме.

\begin{proposition}\label{propos16}
В допустимую область на бифуркационной диаграмме $\Sigma$ отображения момента $J$ входят следующие сегменты особых точек на поверхности $\wsa$ -- образы вырожденных критических движений ранга $1$:

\emph{1)} кривая касания поверхностей $\wsa,\wsc$ полностью
\begin{eqnarray}\notag
&&    \gan: \quad \left\{ \begin{array}{l}
                        \ds{\ell = \pm \sqrt{\frac{s}{2}(1-2\ld^2s)}} \\
                        \ds{h=\frac{1-\ld^2s+2s^2}{2s}}
                      \end{array} \right. , \quad \ds{0<s\leqslant \frac{1}{2\ld^2}};
\end{eqnarray}

\emph{2)} ребро возврата поверхности $\wsa$ между его пересечениями с кривой $\delta_3$
\begin{eqnarray}\notag
&&    \gaa: \quad \ds{\ell = \pm \frac{2}{3\sqrt{3}}(h-\frac{\ld^2}{2})^{3/2}}, \quad
    \left\{ \begin{array}{l}
    \ds{\frac{\ld^2}{2}\leqslant h \leqslant h_C(\ld)}\\
    |\ell| \ls \ell_C(\ld)
    \end{array}
    \right.,
\end{eqnarray}
или в терминах параметра $s$
\begin{eqnarray}\label{eq5_10}
&&    \gaa: \quad \left\{\begin{array}{l}
    \ds{\ell = \pm \sqrt{\frac{s^3}{2}}} \\
    \ds{h = \frac{3}{2}s + \frac{\ld^2}{2}}
    \end{array}
    \right.,
    \quad
    s \in [0,s_C(\ld)],
\end{eqnarray}
где зависимости $h_C(\ld)$, $\ell_C(\ld)$, определяющие экстремумы соответствующих координат на $\delta_3$ при $\ld \gs 0$, и значение в точке экстремума $s$-ко\-ор\-дина\-ты $s_C(\ld)$ находятся из уравнений \eqref{eq4_11}, \eqref{eq5_6} с использованием вспомогательного параметра
\begin{equation}\notag
    \left\{
    \begin{array}{l}
    3x^4-2\ld x^3-4=0 \\
    x \in [-\sqrt[4]{4/3},0)
    \end{array} \right. \; \Rightarrow \; \left\{ \begin{array}{l}
    \ds{h_C=\frac{3}{8}x^2+\frac{2}{x^6}} \\
    \ds{\ell_C= \left|\frac{(4-x^4)^{3/2}}{4x^3}\right|} \\
    \ds{s_C=\frac{4-x^4}{2x^2}}
    \end{array} \right. .
\end{equation}
\end{proposition}

Следующая таблица~\ref{table3} дополняет результаты работ \cite{Gash5,GashDis} в отношении классификации движений на $G^4$. Наличие асимптотических движений или регулярных торов установлено в \cite{Gash5}, исходя из анализа явных квадратур \eqref{eq2_34} -- \eqref{eq2_35}. Таким образом, видно, что регулярные торы могут находиться на критическом уровне первых интегралов лишь в соседстве с атомами типа $A$, а асимптотические движения, конечно, являются неотъемлемой частью гиперболических атомов. Нелокальное отличие точек в парах областей, разделенных кривой $\ell_0$, то есть точек областей $a_{10}$ и $a_{12}$ и точек областей $a_{2}$ и $a_{1}$ состоит в том, что на уровне общих первых интегралов для области $a_2,a_{11}$ есть регулярные торы, а для областей $a_1,a_{12}$ таких торов нет. Это отличие в терминах параметров аналитического решения установлено в \cite{Gash1}.

\begin{center}
\small
\begin{longtable}{|c|c| c| c| c|c|}
\multicolumn{6}{r}{\fts{Таблица \myt\label{table3}}}\\
\hline
\begin{tabular}{c}\fts{Область}\\[-3pt] \fts{по \cite{Gash5}}\end{tabular} & \begin{tabular}{c}\fts{Условия}\end{tabular} & \begin{tabular}{c}\fts{Области в $\mm$}\end{tabular} &
\begin{tabular}{c}\fts{Периоди-}\\[-3pt]\fts{ческие}\\[-3pt]\fts{движения}\end{tabular}
& \begin{tabular}{c}\fts{Асимпто-}\\[-3pt]\fts{тические}\\[-3pt]\fts{движения}\end{tabular} & \begin{tabular}{c}\fts{Регу-}\\[-3pt]\fts{лярные}\\[-3pt]\fts{торы}\end{tabular}\\
\hline\endfirsthead%
\multicolumn{6}{r}{\fts{Таблица \ref{table3} (продолжение)}}\\
\hline
\begin{tabular}{c}\fts{Область}\\[-3pt] \fts{по \cite{Gash5}}\end{tabular} & \begin{tabular}{c}\fts{Условия}\end{tabular} & \begin{tabular}{c}\fts{Области в $\mm$}\end{tabular} &
\begin{tabular}{c}\fts{Периоди-}\\[-3pt]\fts{ческие}\\[-3pt]\fts{движения}\end{tabular}
& \begin{tabular}{c}\fts{Асимпто-}\\[-3pt]\fts{тические}\\[-3pt]\fts{движения}\end{tabular} & \begin{tabular}{c}\fts{Регу-}\\[-3pt]\fts{лярные}\\[-3pt]\fts{торы}\end{tabular}\\
\hline\endhead
I & \begin{tabular}{l}$L_1  > 0,L_2  > 0,$\\$L_3  > 0,L_5  > 0$\end{tabular} & $\overline {a_2 } $ & Нет & Нет & Нет\\
\hline
I & \begin{tabular}{l}$L_1  > 0,L_2  > 0,$\\$L_3  > 0,L_5  < 0$\end{tabular}
& \begin{tabular}{l}$a_8\,(0\ls\ld < \ld_*)$, \\$a_4,a_9\,(0 < \ld <  + \infty )$\end{tabular}
& ``центр'' & Нет & Да\\ \hline
II,VI & \begin{tabular}{l}$L_1 \gs 0,L_2  < 0,$\\$L_3  < 0$\end{tabular} &
\begin{tabular}{l}$a_2\, (0\ls\ld <  + \infty )$,\\$a_{10}\,(\ld > \ld_*) $\end{tabular}
& ``центр'' & Нет & Да\\
\hline
III & \begin{tabular}{l}$L_1  > 0,L_2  < 0,$\\$L_3  > 0,L_5  > 0$\end{tabular} & $\overline {a_2 }$ & Нет & Нет & Нет\\
\hline
III & \begin{tabular}{l}$L_1  > 0,L_2  < 0,$\\$L_3  > 0,L_5  < 0$\end{tabular} & $a_3$ & ``седло'' & Да & Нет\\
\hline
IV & \begin{tabular}{l}$L_1  > 0,L_2  > 0,$\\$L_3  < 0$\end{tabular}
& \begin{tabular}{l}$a_5\,(0\ls\ld < 1)$,\\$a_6\,(0\ls\ld < \sqrt{2})$,\\$a_{11}\,(\ld_* < \ld <  + \infty)$\end{tabular}
& ``седло'' & Да & Нет\\
\hline
V & \begin{tabular}{l}$L_1  < 0,L_2  < 0,$\\$L_3  < 0,L_4  > 0$\end{tabular}
& \begin{tabular}{l}$a_1\,(0\ls\ld <  + \infty ) $,\\$a_{12}\,(\ld > \ld^*)$ \end{tabular}
& ``центр'' & Нет & Нет\\
\hline
V & \begin{tabular}{l}$L_1  < 0,L_2  < 0,$\\$L_3  < 0,L_4  < 0$\end{tabular} & $\overline {a_1 } $ & Нет & Нет & Нет\\
\hline
\myrul VII & \begin{tabular}{l}$L_2  = 0$ \end{tabular} & $\Delta_1 $ & вырожд. & Да & Нет\\
\hline
\myrul VIII & \begin{tabular}{l}$L_3  = 0$\end{tabular} & $\Delta _0 $ & вырожд. & Да & Нет\\
\hline
\end{longtable}
\end{center}


В силу соотношений \eqref{eq3_13} и \eqref{eq2_30}
$$
\ds \ell=-\po (h-\frac{\ld^2}{2}-\po^2), \qquad \ds s = h -\frac{\ld^2}{2}-\po ^2
$$
легко классифицировать $(S,H)$-диаграммы и $(S,L)$-диаграммы первой критической подсистемы.
Разделяющие значения $\ld$ остаются теми же самыми (определяются по перестройкам сечений ключевого множества). На рис.~\ref{fig_M1_HKeyABC_red}, \ref{fig_M1_HKeyDEF_red} приведены $(S,H)$-диаграммы для случаев: $(a)$~${0<\ld<\ld_*}$; $(b)$~$\ld_*<\ld<1$; $(c)$~$1 <\ld< \ld^*$; $(d)$~$\ld^* <\ld< \sqrt{2}$; $(e)$~$\ld > \sqrt{2}$; $(f)$~$\ld=0$.

\def\wid{0.7}
\begin{figure}[!htp]
\centering
\includegraphics[width=\wid\textwidth,keepaspectratio]{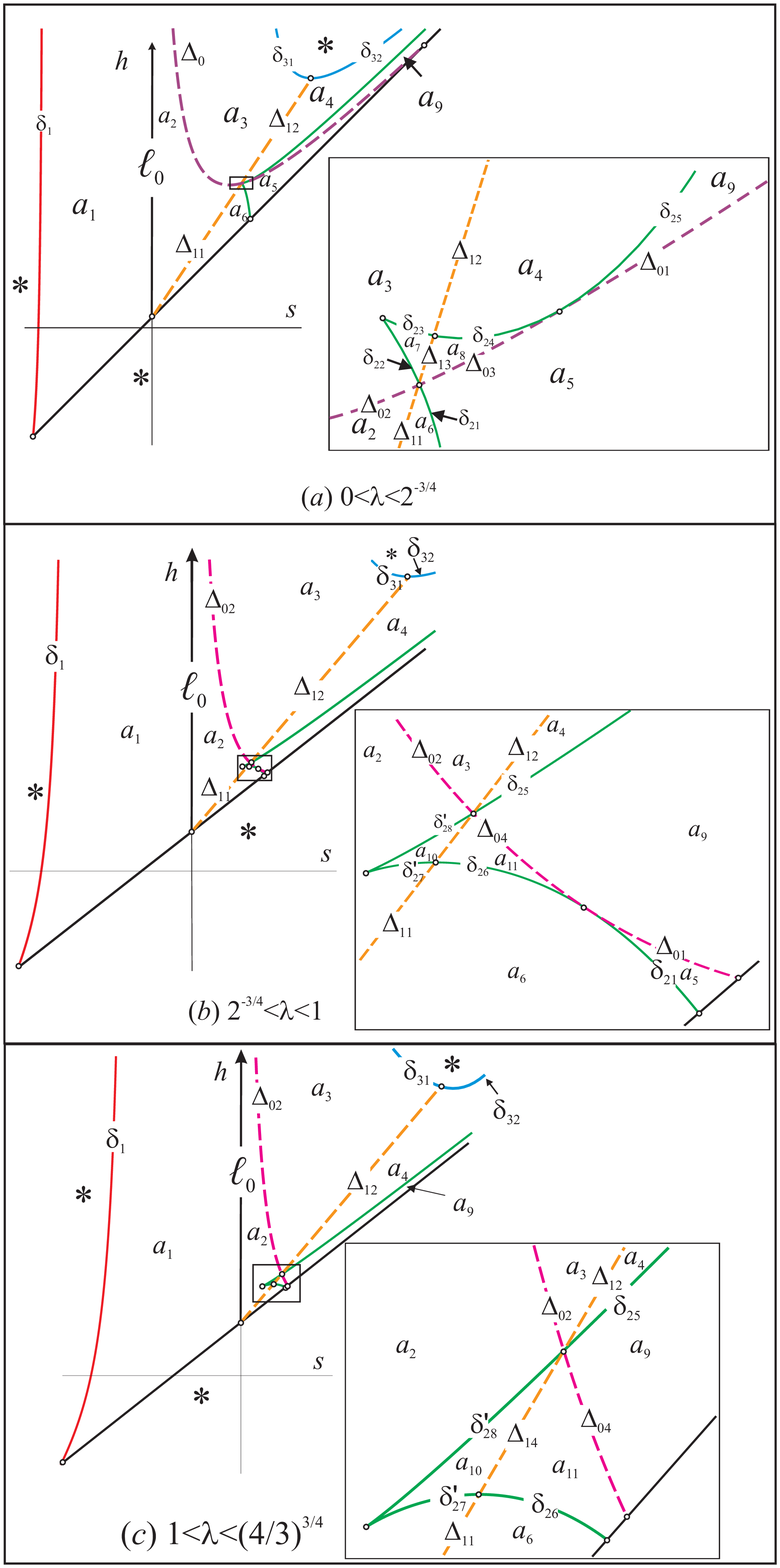}%
\caption{$(S,H)$-диаграммы системы $\mm$ и увеличенные фрагменты.}\label{fig_M1_HKeyABC_red}
\end{figure}

\begin{figure}[!htp]
\centering
\includegraphics[width=\wid\textwidth,keepaspectratio]{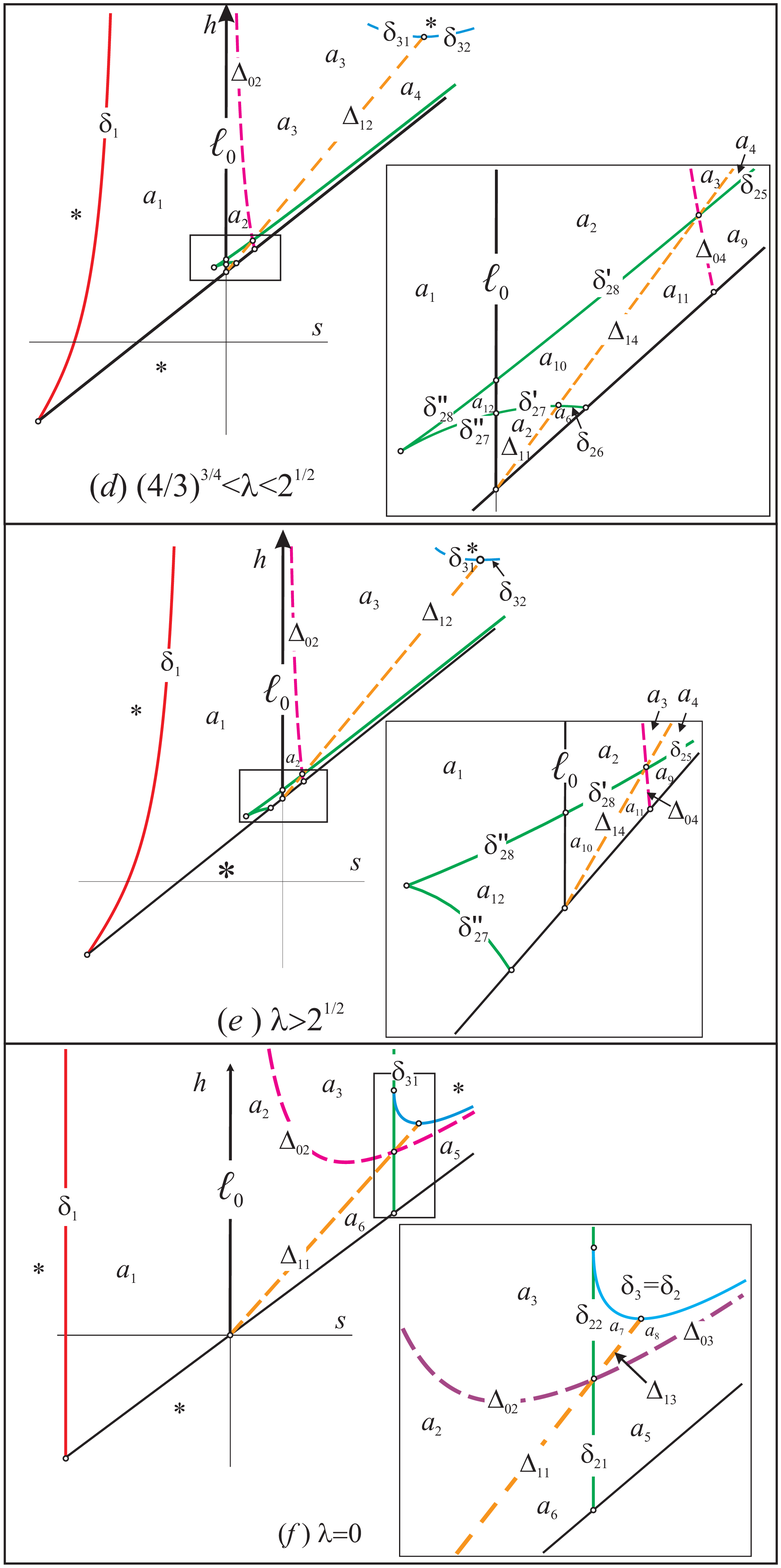}%
\caption{$(S,H)$-диаграммы системы $\mm$ и увеличенные фрагменты (продолжение).}\label{fig_M1_HKeyDEF_red}
\end{figure}

Укажем все точки на плоскостях $(p,h)$ и $(s,h)$, имеющие значение для построения $H$-атласов бифуркационных диаграмм систем на $Q_h^4$ (точки экстремальных значений $h$-ко\-ор\-ди\-на\-ты на ключевых множествах). Соответствующие обозначения приведены на рис.~\ref{fig_fig_sys1HPpoints}, \ref{fig_fig_sys1HSpoints}. Точки, связанные с кривыми $\delta_1, \delta_2, \delta_3$, обозначены соответственно буквами $A,B,C$, снабженными индексами там, где это необходимо. Буква $D$ использована для точек на $\Delta_j$ (здесь $j=0,1$, в третьей системе появляется $\gac$). Для того чтобы увидеть все особые точки, достаточно привести на рисунках случаи $(a)$~${0<\ld<\ld_*}$ и $(b)$~$\ld^* < \ld < \sqrt{2}$.

\begin{figure}[!hb]
\centering
\includegraphics[width=\textwidth,keepaspectratio]{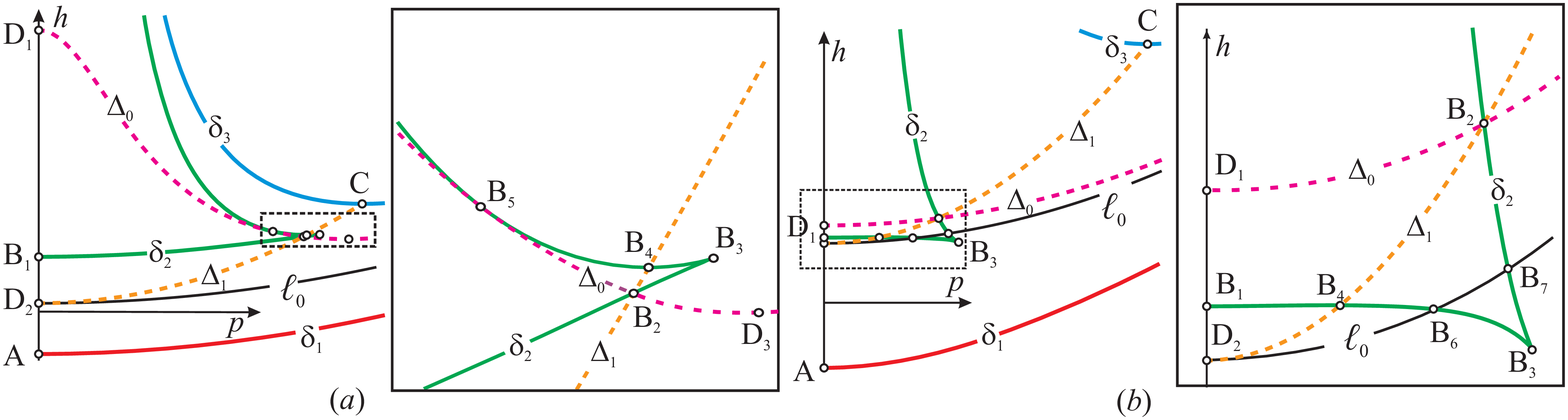}\\
\caption{Особые точки $(P,H)$-диаграммы системы $\mm$.}\label{fig_fig_sys1HPpoints}
\end{figure}
\begin{figure}[!hb]
\centering
\includegraphics[width=\textwidth,keepaspectratio]{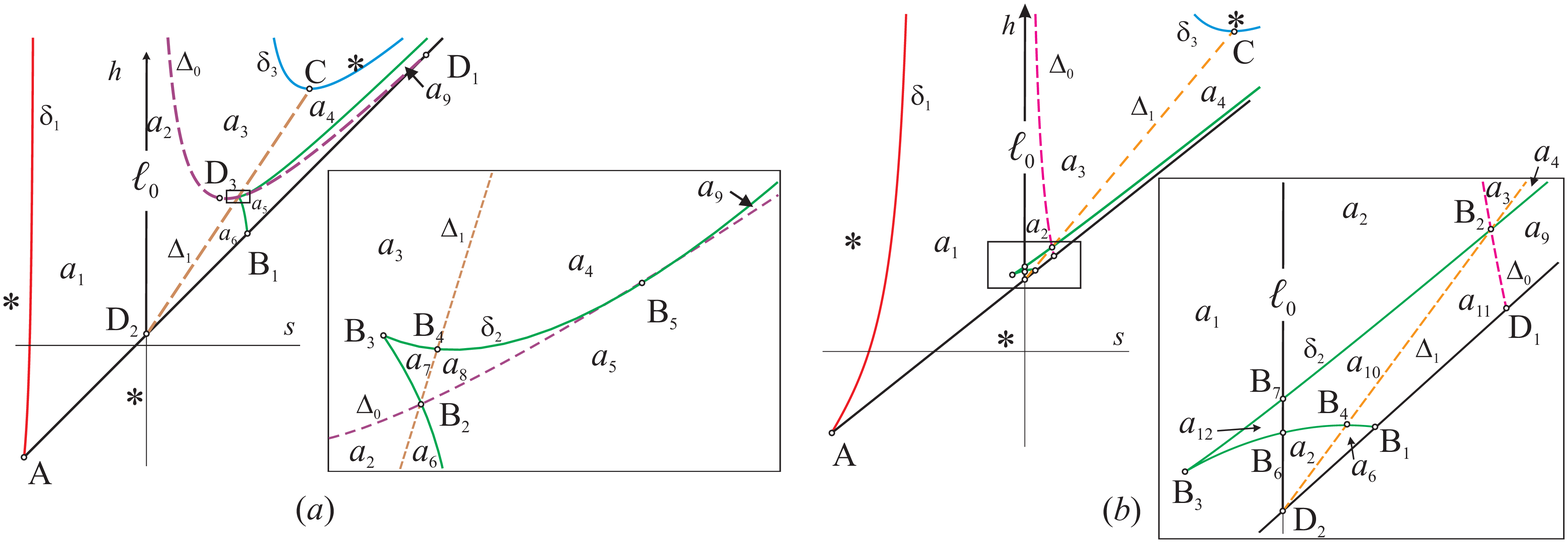}\\
\caption{Особые точки $(S,H)$-диаграммы системы $\mm$.}\label{fig_fig_sys1HSpoints}
\end{figure}

На рис.~\ref{fig_M1_LKeyABC_red}, \ref{fig_M1_LKeyDEF_red} и приведены $(S,L)$-диаграммы подсистемы $\mm$ для неразделяющих случаев: $(a)$~${0<\ld<\ld_*}$; $(b)$~$\ld_*<\ld<1$; $(c)$~$1 <\ld< \ld^*$; $(d)$~$\ld^* <\ld< \sqrt{2}$; $(e)$~$\ld > \sqrt{2}$; $(f)$~$\ld=0$. Здесь сразу отмечены особые точки, имеющие значение для построения $L$-атласа бифуркационных диаграмм систем на $\mPel$. Отметим, что по сравнению с диаграммами, включающими $H$, появилась дополнительная точка $D_4$ -- экстремум $\ell$-координаты на образе вырожденных точек~$\gan$.

\def\wid{0.7}
\begin{figure}[!htp]
\centering
\includegraphics[width=\wid\textwidth,keepaspectratio]{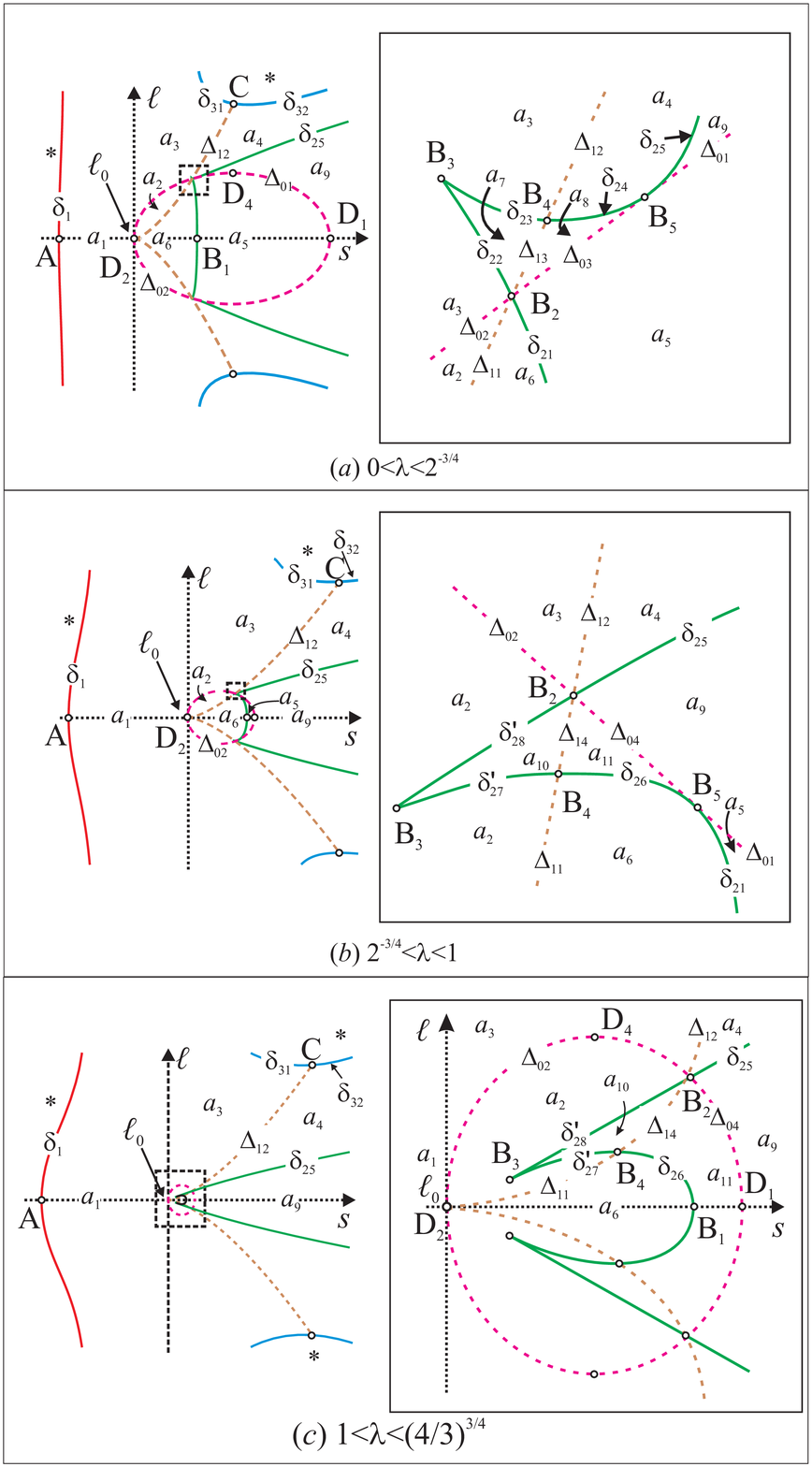}%
\caption{$(S,L)$-диаграммы системы $\mm$ с полной детализацией.}\label{fig_M1_LKeyABC_red}
\end{figure}

\begin{figure}[!htp]
\centering
\includegraphics[width=\wid\textwidth,keepaspectratio]{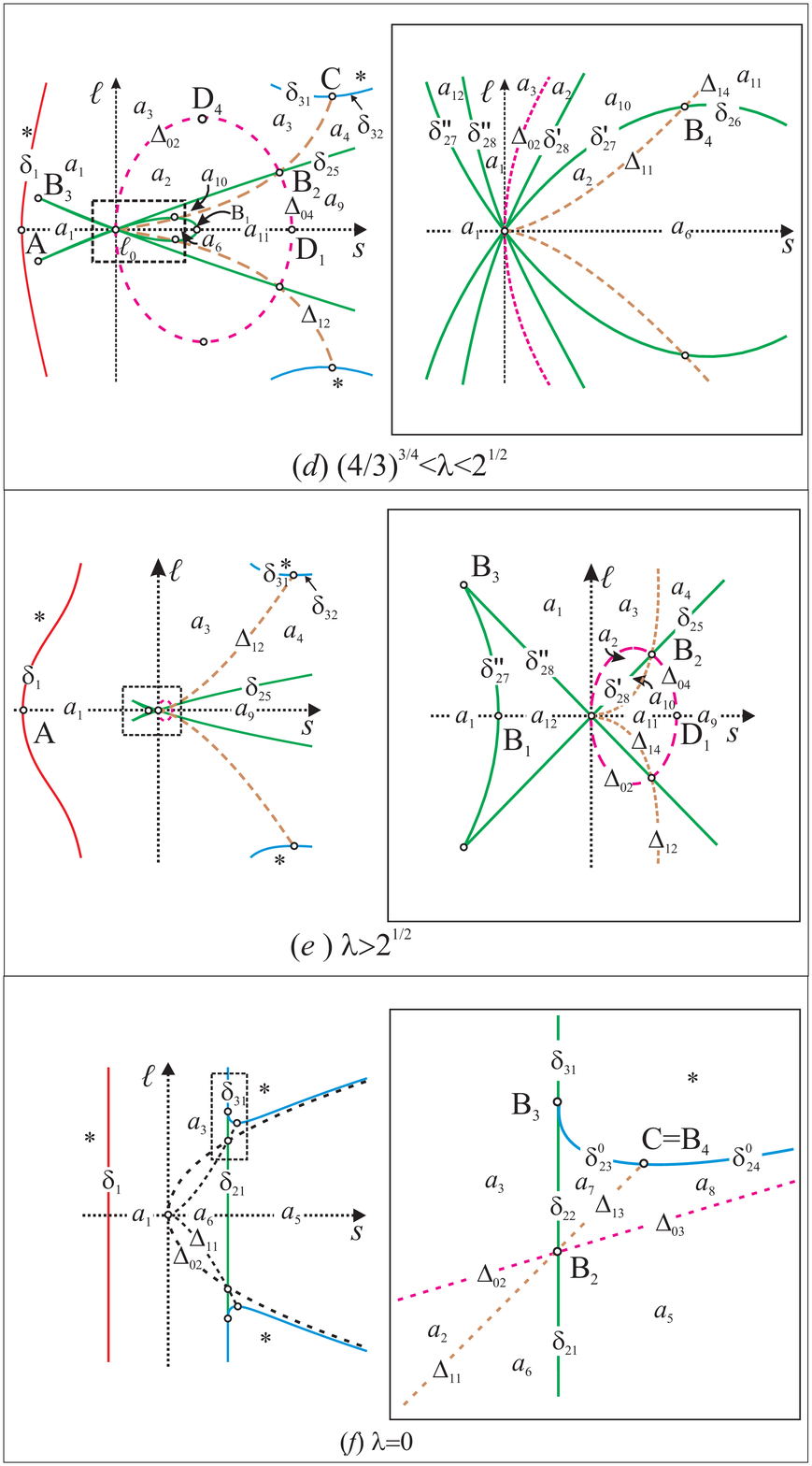}%
\caption{$(S,L)$-диаграммы системы $\mm$ с полной детализацией (продолжение).}\label{fig_M1_LKeyDEF_red}
\end{figure}


Перечислим явно все значения параметров и интегралов (общих и частных) в отмеченных особых точках. Образ вырожденных критических точек ранга $0$ получим, рассекая разделяющее множество на рис.~\ref{fig_RazdCrit0} на заданном уровне $\ld$. При этом точки выхода на ось $r=0$ учитываются только на $\delta_2$. На $\delta_1$ вырожденных точек нет, но есть экстремум $h$ --- глобальное наименьшее значение энергии, достигаемое в точке
\begin{equation}\label{eq5_11}
A: \quad r=0, \quad p=0, \quad s=-1-\ds{\frac{\ld^2}{2}}, \quad h=-1, \quad \ell=0.
\end{equation}

Обозначая
\begin{equation}\notag
\begin{array}{l}
  Q_0 = \ld(r-\ld)+D,\\
  Q_1 = (r-\ld)(2r-\ld)-D, \\
  Q_2 = r(r-\ld)(2r-\ld)+\ld D, \\
  Q_3 = r(\ld-r)+D,
\end{array}
\end{equation}
на кривой $\delta_2$ ($r<0$) найдем
\begin{equation}\notag
\begin{array}{ll}
  \ds \frac{dh}{dr} = \frac{1}{2(\ld-r)^2 D} Q_1 Q_2, & \ds \frac{ds}{dr} = \frac{1}{2 D} Q_3, \\
  \ds \frac{dp^2}{dr} = - \frac{1}{2 (r-\ld)^2 D} Q_2 Q_3, &
  \ds \frac{d\ell^2}{dr} = \frac{1}{8 (r-\ld)^2 D} Q_0 Q_1 Q_2 Q_3.
  \end{array}
\end{equation}
Отсюда, в частности,
\begin{equation}\label{eq5_12}
\begin{array}{l}
  \ds \frac{dh}{dp} = - \sqrt {\frac{2 r}{r-\ld}} \frac{Q_1}{\sqrt{Q_3}}, \qquad
  \ds \frac{d\ell}{ds} = -\frac{1}{2\sqrt{2 r} (r-\ld)^{3/2}} Q_1 \sqrt{Q_3}.
\end{array}
\end{equation}
Поэтому экстремумы $h(r)$ на $\delta_2$ -- это точка $B_1$ ($r=0$), вырожденная точка $B_4$ (кривая $\vpi_{23}$)
\begin{equation}\label{eq5_13}
\begin{array}{l}
(\ld-r)(\ld-2r)-D=0,
\end{array}
\end{equation}
точка возврата $B_3$ (кривая $\vpi_{24}$)
\begin{equation}\label{eq5_14}
\begin{array}{l}
r(\ld-r)(\ld-2r)+\ld D=0.
\end{array}
\end{equation}
Касание кривых $\delta_2$ и $\gan$ дает точку $B_5$ (кривая вырожденных точек $\vpi_{21}$, уравнение \eqref{eq4_21}), а пересечение всех трех кривых $\delta_2$, $\gan$ и $\gaa$ происходит в точке $B_4$ (кривая вырожденных точек $\vpi_{22}$, $r=-\ld$).
Пересечения $\delta_2$ с кривой кратных точек $\ell_0$ определяются, согласно \eqref{eq4_11} уравнением
\begin{equation}\label{eq5_15}
\begin{array}{l}
\ld(r-\ld)+ D =0.
\end{array}
\end{equation}
Решения уравнения \eqref{eq5_13} записаны в параметрическом виде \eqref{eq4_37} с помощью подстановки \eqref{eq4_33}, решения \eqref{eq5_14} находятся явно. Уравнение \eqref{eq5_15} сводится к уравнению
\begin{equation}\notag
(r-\ld)^3(r+\ld)+4 =0
\end{equation}
с условием $r<0$. Из него той же подстановкой \eqref{eq4_33} получим необходимое параметрическое представление координат точек $B_6,B_7$ при $\ld \gs \ld^*$. Согласно \eqref{eq5_12}, особые точки кривой $\delta_2$ на плоскости $(s,\ell)$ дополнительно порождаются условием $Q_3=0$, откуда следует $p=0,\ell=0$, то есть это снова точка $A$. Напомним попутно, что в силу тождества $\ell = -s p$, вся кривая $\ell_0$, отвечающая значению $s=0$, на плоскости $(s,\ell)$ ``схлопывается'' в начало координат, в частности, ось $O\ell$ состоит из недопустимых точек, кроме $(0,0)$. Поэтому на рис.~\ref{fig_M1_LKeyABC_red} и \ref{fig_M1_LKeyDEF_red} она изображена пунктиром. Таким образом, кривая $\delta_2$ порождает следующий набор точек и значений первых интегралов:
\begin{eqnarray}
& & B_1: \quad r=0, \quad p=0, \quad s_1=1-\ds{\frac{\ld^2}{2}}, \quad s_2=\frac{1}{2}, \quad h=1, \quad \ell=0;\label{eq5_16}\\
& & B_2=\vpi_{22}: \, \left\{ \begin{array}{ll}
    r=- \ld, &  p=\pm \ds{ \sqrt{\frac{\sqrt{1+\ld^4}-\ld^2}{2}       }}\\
    s_1=\sqrt{1+\ld^4}-\ld^2, & s_2=s_1\\
    \ds{h=\frac{3}{2}\sqrt{1+\ld^4}-\ld^2}, & \ds{\ell=\mp \frac{1}{\sqrt{2}}(\sqrt{1+\ld^4}-\ld^2)^{3/2}}\end{array}\right. ;\label{eq5_17}\\
& & B_3=\vpi_{24}: \, \left\{ \begin{array}{ll}
    r=-\ds{\frac{1}{2}\sqrt{U}(\sqrt{4+U^2}-U)}, &  p=\pm \ds{\frac{1}{2\sqrt{2}}(\sqrt{4+U^2}-U)^{3/2}}\\
    s_1=\ds{2\frac{\sqrt{4+U^2}-2U}{(\sqrt{4+U^2}-U)^2}}, & s_2=\ds{\frac{1}{2 U}} \qquad (U=\ld^{2/3})\\
    \ds{h=\frac{1}{4}\left[(4+U^2)^{3/2}-U(6+U^2)\right]}, & \ds{\ell=\mp \frac{\sqrt{4+U^2}-2U}{\sqrt{2}(\sqrt{4+U^2}-U )^{1/2}}}\end{array}\right. ; \label{eq5_18}
\end{eqnarray}
\begin{eqnarray}& & B_4=\vpi_{23}: \, \left\{ \begin{array}{ll}
    r=\ds{\frac{x^4-4}{2x^3}}, & \ld=\ds{\frac{3x^4-4}{2x^3}} \\
    p=\pm \ds{\frac{\sqrt{4-x^4}}{2x}}, & s_1=\ds{\frac{4-x^4}{2x^2}}, \quad s_2=\ds{\frac{x^6}{2(3 x^4-4)}}\\
    \ds{h=\frac{3}{8}x^2+\frac{2}{x^6}}, & \ds{\ell=\mp \frac{(4-x^4)^{3/2}}{4x^3}}\end{array}\right. , \, x \in (\sqrt[4]{4/3},\sqrt{2}] ;\label{eq5_19} \\
& & B_5=\vpi_{21}: \, \left\{ \begin{array}{ll}
    r=\ld- \ld^{-1/3}, & p= \pm \ld^{1/3} \sqrt{1-\ld^{4/3}} \\
    s_1=\ds{\frac{1}{2\ld^{2/3}}}, & s_2= s_1\\
    \ds{h=-\frac{\ld^2}{2}+\ld^{2/3}+\frac{1}{2\ld^{2/3}}}, & \ds{\ell=\mp \frac{\sqrt{1-\ld^{4/3}}}{2\ld^{1/3}}}\end{array}\right. , \quad \ld \in (0,1] ;\label{eq5_20}\\
& & B_{6,7}: \, \left\{ \begin{array}{ll}
    r=\ds{-\frac{1}{2}x+\frac{2}{x^3}}, & \ld=\ds{\frac{1}{2}x+\frac{2}{x^3}} \gs \ld^*\\
    p=\pm \ds{\frac{\sqrt{x^4-4}}{x^3}}, & s_1=0, \quad s_2=\ds{\frac{2x^2}{4+x^4}}\\
    \ds{h=\frac{1}{8}x^2+\frac{2}{x^2}-\frac{2}{x^6}}, & \ds{\ell=0}\end{array}\right. ,
    \begin{array}{l}
    B_6: \, x \in [\sqrt{2},\sqrt{2\sqrt{3}}],\\
    B_7: \, x \in [\sqrt{2\sqrt{3}},+\infty).
    \end{array} \label{eq5_21}
\end{eqnarray}
Здесь $s_1,s_2$ --- значения параметров $s$ на поверхностях $\wsa$ и $\wsi_{2,3}$ соответственно.
Формулы \eqref{eq5_18} в точке возврата $B_3$ упрощаются введением параметра $z$:
\begin{equation}\label{eq5_22}
B_3=\vpi_{24}: \, \left\{ \begin{array}{ll}
    \ld=(\ds{\frac{1}{z}} - z)^{3/2}, & 0<z<1\\
    r=-\sqrt{z(1-z^2)}, &  p=\pm z^{3/2}\\
    s_1=\ds{\frac{3z^2-1}{2z^3}}, & s_2=\ds{\frac{z}{2(1-z^2)}}\\
    \ds{h=\frac{1}{2}z(3+z^2)}, & \ds{\ell=\mp \frac{3z^2-1}{2z^{3/2}}}\end{array}\right. .
\end{equation}

На кривой $\delta_3$ точка минимума $h$, соответствующая кривой вырождения $\vpi_{31}$ определяется тем же уравнением \eqref{eq5_13}, и параметризация этой точки дана уравнениями \eqref{eq4_38}. Имеем
\begin{equation}\label{eq5_23}
C=\vpi_{31}: \, \left\{ \begin{array}{ll}
    r=\ds{\frac{x^4-4}{2x^3}}, & \ld=\ds{\frac{3x^4-4}{2x^3}} \\
    p=\mp \ds{\frac{\sqrt{4-x^4}}{2x}}, & s_1=\ds{\frac{4-x^4}{2x^2}}, \quad s_2=\ds{\frac{x^6}{2(3 x^4-4)}}\\
    \ds{h=\frac{3}{8}x^2+\frac{2}{x^6}}, & \ds{\ell=\pm \frac{(4-x^4)^{3/2}}{4x^3}}\end{array}\right. , \, x \in [-\sqrt[4]{4/3},0) .
\end{equation}

Наконец, на кривых $\gan,\gaa$ экстремальные значения $h$ достигаются в точках
\begin{eqnarray}
& & \label{eq5_24}
D_1: \, \left\{ \begin{array}{l}
     p=0, \quad s_1=\ds{\frac{1}{2\ld^2}}, \quad s_2=s_1\\
    \ds{h=\frac{1+\ld^4}{2\ld^2}}, \quad \ds{\ell=0}\end{array}\right., \quad \ld>0 ;
\\
& & \label{eq5_25}
D_2: \, \left\{ \begin{array}{ll}
     p=0, & s_1=0 \\
    \ds{h=\frac{\ld^2}{2}}, & \ds{\ell=0}\end{array}\right., \quad \ld \gs 0 ;
\\
& & \label{eq5_26}
D_3: \, \left\{ \begin{array}{ll}
    \ds{p=\pm \sqrt{\frac{1-\sqrt{2}\ld^2}{\sqrt{2}}}}, & \ds{s_1= \frac{1}{\sqrt{2}}} \\
    \ds{h=\sqrt{2}-\frac{\ld^2}{2}}, & \ds{\ell=\mp \sqrt{\frac{1-\sqrt{2}\ld^2}{2\sqrt{2}}}}\end{array}\right., \quad \ld^2 \ls \frac{1}{\sqrt{2}}.
\end{eqnarray}
Дополнительно, из \eqref{eq5_10} видим, что, как было отмечено ранее, экстремальное значение $\ell$ имеется на $\gan$ в точке (см. рис.~\ref{fig_M1_LKeyABC_red}, \ref{fig_M1_LKeyDEF_red})
\begin{eqnarray}
& & D_4: \quad \left\{ \begin{array}{ll}
      \ds h=\frac{1+6\ld^4}{4\ld^2}, & \ds \ell=\pm \frac{1}{4\ld}, \quad \ds s_1=\frac{1}{4\ld^2}\\
       s_2=s_1, & p=\mp \ld
    \end{array}\right., \quad \ld>0.
\label{eq5_27}
\end{eqnarray}

\subsection{Детализация. Вторая и третья критические подсистемы}
Для второй и третьей критических подсистем перепишем уравнения поверхностей \eqref{eq3_8} в виде
\begin{equation}\label{eq5_28}
\wsb \cup \wsc = \left\{ \displaystyle{h=2 \ell^2 + \frac{1}{2 s} - \frac{\ld^2}{2}(1-4s^2),}\;
k=-4 \ell^2 \ld^2 + \displaystyle{\frac{1}{4 s^2}} - \displaystyle{\frac{\ld^2}{s}(1-\ld^2s)(1-4s^2)}
\right\}.
\end{equation}
По прежнему $s <0 $ для $\wsb$ и $s > 0 $ для $\wsc$.  Отсюда следует, что пара $(s,\ell)$ определяет единственную точку на соответствующей поверхности и это соответствие взаимно однозначно. Поэтому удобно говорить об $(S,L)$-диаграммах подсистем $\mn$ и $\mo$.

Чтобы получить простой критерий существования решений \eqref{eq2_27} -- \eqref{eq2_29}, представим их в алгебраическом виде.
Выполним замену
\begin{equation}\notag
\begin{array}{llll}
  X=\ds{\frac{1-\zeta ^2}{1+\zeta ^2}},& Y=\ds{\frac{2 \zeta }{1+\zeta ^2}}.
\end{array}
\end{equation}
Имеем однозначные зависимости
\begin{equation}\notag
\begin{array}{ll}
  \ds{\omega_1=-\frac{\ell} {s}- \frac{2 \vk \ro  \zeta }{1+\zeta ^2},} & \ds{\omega_3 = \ld+2 \vk \frac{1-\zeta ^2}{1+\zeta ^2},}\\[2mm]
  \ds{\alpha_1=\frac{\ld s(1-\zeta ^4)+2\ell \ro\zeta (1+\zeta ^2)-8\vk^3 \zeta ^2}{\vk(1+\zeta ^2)^2},} & \ds{\alpha_3 = \frac{\ell (1-\zeta ^2)-2\ld \ro s \zeta  }{\vk(1+\zeta ^2)}}
\end{array}
\end{equation}
и выражения с радикалами
\begin{equation}\notag
\begin{array}{l}
  \ds{\omega_2=-\frac{1}{\sqrt{2}(1+\zeta ^2)} \, \sqrt{\frac{\ro^2}{\vk s}Z(\zeta )},} \quad \ds{\alpha_2=-\frac{2\sqrt{2}\vk \zeta }{(1+\zeta ^2)^2} \, \sqrt{\frac{1}{\vk s}Z(\zeta )}.}
\end{array}
\end{equation}
где
\begin{equation}\notag
\begin{array}{l}
  \ds{Z(\zeta )=(\vk -2\ld s^2) \zeta ^4 +4\ell \ro s \zeta  (1+\zeta ^2)+2\vk (1-4\vk^2s)\zeta ^2+(\vk+2\ld s^2).}
\end{array}
\end{equation}
Динамика определяется уравнением
\begin{equation}\notag
\begin{array}{l}
    \ds{\frac{d\zeta }{dt}=\frac{1}{2\sqrt{2}}  \, \sqrt{\frac{1}{\vk s}Z(\zeta )}.}
\end{array}
\end{equation}

Далее в силу \eqref{eq2_27} следует полагать
\begin{equation}\notag
\begin{array}{llll}
  \zeta = z & (z \in \bR), & \ro=\ro_+, & \ro^2 \gs 0; \\
  \zeta = \ri z & (z \in \bR), & \ro=\ri \ro_-, & \ro^2 < 0 \\
\end{array}
\end{equation}
($\ro_+$ и $\ro_-$ считаем неотрицательными). Получим решения в следующем виде.
Для $\ro^2 \gs 0$ имеем однозначные зависимости
\begin{equation}\notag
\begin{array}{ll}
  \ds{\omega_1=-\frac{\ell} {s}- \frac{2 \vk \ro_+  z}{1+z^2},} & \ds{\omega_3 = \ld+2 \vk \frac{1-z^2}{1+z^2},}\\[2mm]
  \ds{\alpha_1=\frac{\ld s(1-z^4)+2\ell \ro_+z(1+z^2)-8\vk^3 z^2}{\vk(1+z^2)^2},} & \ds{\alpha_3 = \frac{\ell (1-z^2)-2\ld \ro_+ s z }{\vk(1+z^2)}}
\end{array}
\end{equation}
и выражения с радикалами
\begin{equation}\notag
\begin{array}{l}
  \ds{\omega_2=-\frac{\ro_+}{\sqrt{2}(1+z^2)} \, \sqrt{\frac{1}{\vk s}Z_+(z)},} \quad \ds{\alpha_2=-\frac{2\sqrt{2}\vk z}{(1+z^2)^2} \, \sqrt{\frac{1}{\vk s}Z_+(z)}.}
\end{array}
\end{equation}
Аналогично для $\ro^2 < 0$ получим
\begin{equation}\notag
\begin{array}{ll}
  \ds{\omega_1=-\frac{\ell} {s}+ \frac{2 \vk \ro_-  z}{1-z^2},} & \ds{\omega_3 = \ld+2 \vk \frac{1+z^2}{1-z^2},}\\[2mm]
  \ds{\alpha_1=\frac{\ld s(1-z^4)-2\ell \ro_- z(1-z^2)+8\vk^3 z^2}{\vk(1-z^2)^2},} & \ds{\alpha_3 = \frac{\ell (1+z^2)+2\ld \ro_- s z }{\vk(1-z^2)}}
\end{array}
\end{equation}
и выражения с радикалами
\begin{equation}\notag
\begin{array}{l}
  \ds{\omega_2=-\frac{\ro_-}{\sqrt{2}(1-z^2)} \, \sqrt{\frac{1}{\vk s}Z_-(z)},} \quad \ds{\alpha_2=-\frac{2\sqrt{2}\vk z}{(1-z^2)^2} \, \sqrt{\frac{1}{\vk s}Z_-(z)}.}
\end{array}
\end{equation}
Здесь
\begin{equation}\notag
\begin{array}{l}
  \ds{Z_+(z)=\phantom{-}(\vk -2\ld s^2) z^4 +4\ell \ro_+ s z (1+z^2)+2\vk (1-4\vk^2s)z^2+(\vk+2\ld s^2),}\\[2mm]
  \ds{Z_-(z)=-(\vk -2\ld s^2) z^4 +4\ell \ro_- s z (1+z^2)+2\vk (1-4\vk^2s)z^2-(\vk+2\ld s^2).}
\end{array}
\end{equation}
Динамика определяется соответствующим уравнением (верхний знак для $\ro^2 \gs 0$, нижний -- для ${\ro^2 < 0}$):
\begin{equation}\label{eq5_29}
\begin{array}{l}
    \ds{\frac{dz}{dt}=\frac{1}{2\sqrt{2}}  \, \sqrt{\frac{1}{\vk s}Z_{\pm}(z)}.}
\end{array}
\end{equation}

Из полученных выражений вытекает следующий критерий.
\begin{proposition}\label{propos17}
При заданных значениях $s,\ell$, отвечающих уровню в $\mno$, не содержащему критических точек ранга $0$, количество критических окружностей в системах $\mno$ равно количеству траекторий в фазовом пространстве $\overline{\bR}{\times}\bR$ соответствующего уравнения \eqref{eq5_29}, где $\overline{\bR}\approx S^1$ --- прямая, дополненная точкой $z=\infty$. Таким образом, если соответствующий многочлен $Z_{\pm}$ имеет $2m$ вещественных корней \emph{(}$m=0,1,2$\emph{)}, то на критическом уровне $s,\ell$ лежит $m$ критических окружностей, за исключением случая, когда $m=0$ и старший коэффициент многочлена положительный. В этом случае имеется две критических окружности \emph{(}$z$ пробегает все $\overline{\bR}$\emph{)}, на каждой из которых сохраняет свой знак переменная $\omega_2$.
\end{proposition}

Отметим последний случай ($m=0$) как \emph{особый}:
\begin{equation}\label{eq5_30}
\begin{array}{l}
    \ro^2>0, \qquad \vk -2\ld s^2 > 0, \qquad Z_{+}(z)> 0 \quad \forall z \in \bR.
\end{array}
\end{equation}

Получить оба показателя Морса\,--\,Ботта в явном виде для систем $\mno$ не удается. Однако оказалось возможным получить несложные выражения для их вычисления и, что самое важное, явно выделить вектор на трансверсальной площадке на траектории в гиперболических точках, по которому, за исключением особого случая \eqref{eq5_30}, разрыв ``восьмерки'' не происходит, что позволяет определить направление атомов типа $B$ при возрастании интеграла $K$.

Как и в предыдущем случае трансверсальную площадку к критической окружности (в критической точке ранга 1) выбираем как ортогональное дополнение к векторам $\grad \Gamma$, $\grad L$, $\grad H$, $\sgrad H$.

Рассмотрим случай, не удовлетворяющий \eqref{eq5_30}. На любой траектории переменная $z$ осциллирует между корнями соответствующего многочлена $Z(z)$, включая, конечно, и возможность прохода через бесконечно удаленную точку. Выберем на траектории точку $x_0$, в которой $Z(z)=0$. Тогда векторы $\grad \Gamma$, $\grad L$, $\grad H$ ортогональны плоскости $O\omega_2 \alpha_2$, а вектор $\sgrad H$ лежит в этой плоскости и имеет вид
\begin{equation}\notag
\begin{array}{l}
  \sgrad H =\bigl(0,\,b_2,\,0,\,0,\,b_5,\,0\bigr).
\end{array}
\end{equation}
Условный экстремум функции $K$ на совместном уровне функций $\Gamma, L, H$ в $\mP^6$ есть критическая точка функции с неопределенными множителями Лагранжа
\begin{equation}\notag
\ds{K_2=K+(2 \ld^2 - \frac{1}{s}) H  + 2 s L^2-\frac{2 \vk ^2}{s} \Gamma}.
\end{equation}
Очевидно, часть этой функции, не содержащая функций Казимира $L,\Gamma$, совпадает с \eqref{eq5_3}.
Пусть $\mk=d^2 K_2 (x_0)$ --- матрица с элементами $\mk_{ij}$. Так как вектор $\sgrad H$ лежит в ядре $\mk$, то
\begin{equation}\notag
\mk_{22} b_2 + \mk_{25} b_5 =0, \qquad \mk_{25} b_2 + \mk_{55} b_5 =0.
\end{equation}
В частности, поскольку $b_2 b_5\ne 0$ и коэффициенты $\mk_{25}, \mk_{55}$ имеют простые выражения
\begin{equation}\notag
    \mk_{25}=-\ds{\frac{8 \vk \ro \zeta}{1+\zeta^2}}, \qquad \mk_{55}=2\ro^2,
\end{equation}
найдем упрощение для достаточно громоздкого элемента
\begin{equation}\notag
    \mk_{22}=\frac{\mk_{25}^2}{\mk_{55}}=\ds{\frac{32\vk^2\zeta^2}{(1+\zeta^2)^2}}
\end{equation}
и выражение
\begin{equation}\notag
    b_{5}=-\ds{\frac{\mk_{25}\, b_2}{\mk_{55}}}=\ds{\frac{4\vk \, \zeta b_2}{\ro(1+\zeta^2)}}.
\end{equation}
Возьмем в качестве первого вектора трансверсальной площадки вектор, лежащий в той же плоскости $O\omega_2 \alpha_2$ и ортогональный $\sgrad H$:
\begin{equation}\notag
\begin{array}{l}
    v_1=(0,\,-b_5,\,0,\,0,\,b_2,\,0).
\end{array}
\end{equation}
Тогда для компьютерных расчетов вектор $v_2$ легко находится как нуль-пространство матрицы из пяти векторов $\grad \Gamma$, $\grad L$, $\grad H, \sgrad H, v_1$. Хотя его аналитическое выражение и достаточно сложно, все необходимые расчеты легко выполняются в динамическом режиме. Кроме того, матрица ограничения квадратичной формы $\mk$ на трансверсальную площадку в базисе $(v_1,v_2)$ оказывается диагональной,так что показатели Морса\,--\,Ботта имеют вид
\begin{equation}\notag
\begin{array}{l}
    \mu_1=\mk v_1 \cdot v_1, \qquad \mu_2=\mk v_2 \cdot v_2.
\end{array}
\end{equation}
В данный момент нам важно, что, как показывает анализ проекций интегральных многообразий на плоскость $O\omega_1 \omega_2$, в окрестности систем $\mno$ так же, как и в случае с системой $\mm$, критическая поверхность гиперболической окружности никогда не рвется в направлении оси $O\omega_2$, в частности, это означает, что для атомов типа $B$ вектор $v_1$ всегда указывает во ``внешность восьмерки''. Отсюда следует, что если $\mu_1>0$, то в направлении ``головы'' атома функция $K$ возрастает (``голова'' вверх), а если $\mu_1<0$, то в направлении ``головы'' атома функция $K$ убывает (``голова'' вниз). Из найденных выражений для различных элементов получаем
\begin{equation}\notag
\begin{array}{l}
    \mu_1=\mk_{22} b_5^2-2 \mk_{25} b_2 b_5+\mk_{55} b_2^2=\ds{2\frac{\left[16\vk^2\zeta^2+\ro^2(1+\zeta^2)^2\right]^2}{\ro^2(1+\zeta^2)^4}b_2^2}.
\end{array}
\end{equation}
Итак, направление ребер у несимметричных атомов определяется знаком величины $\ro^2$. Поскольку в $\mn$ все точки ранга 1 имеют по доказанному эллиптический тип и $\ro^2>0$, то получим следующее утверждение.

\begin{proposition}\label{propos18}
При возрастании интеграла $K$ на изоэнергетическом уровне $\iso$ в точках критической подсистемы $\mn$ на любой невырожденной критической окружности происходит бифуркация $A_+$ рождения тора.
\end{proposition}

\begin{figure}[!ht]
\centering
\includegraphics[width=50mm,keepaspectratio]{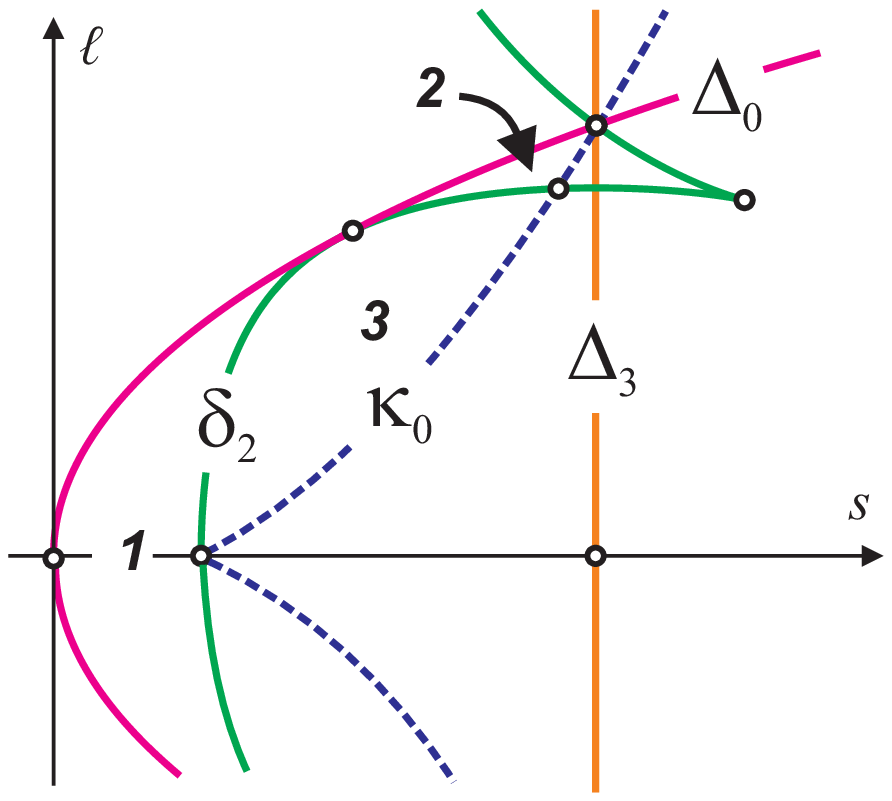}\\
\caption{Особый случай.}\label{fig_spec_case}
\end{figure}

В системе $\mo$ ситуация несколько сложнее, так как известно из классической задачи, что здесь возможно присутствие атомов $C_2$ \cite{KhPMM83}.
Для неособого случая выше доказано, что при наличии двух траекторий на одном и том же уровне интегралов распределение знаков в парах показателей Морса\,--\,Ботта одинаково.
Рассмотрим особый случай \eqref{eq5_30}. На плоскости $(s,\ell)$ условиям положительности указанных констант отвечает область
$$
\ld^2 s^2(2s^2-1)<\ell^2<\frac{s}{2}(1-2\ld^2s).
$$
На рис.~\ref{fig_spec_case} этим условиям удовлетворяют области {\it 1}--{\it 3}. Здесь кривая $\kappa_0$ отвечает граничному случаю обращения в ноль старшего коэффициента многочлена $Z_+$, а кривая $\delta_2$, обозначающая согласно договоренности во всех пространствах интегральных постоянных образ соответствующего множества критических точек ранга 0, является, конечно, дискриминантным множеством многочлена $Z_+$. Обозначения $\gan$, $\gac$ также отвечают введенным ранее. В областях {\it 2} и {\it 3} многочлен $Z_+$ имеет вещественные корни, соответственно, четыре и два. Лишь в области {\it 1} таких корней нет, поэтому здесь и реализуется случай \eqref{eq5_30}. Очевидно, эта область имеет пересечение с осью $\ell=0$ при \begin{equation}\label{eq5_31}
s \in (0,\min \{1/2,1/2\ld^2\}),
\end{equation}
а так как $Z_+>0$ на всей прямой, то удобно взять для упрощения вычислений $z=0$. Тогда векторы, задающие трансверсальную площадку к периодическому решению, легко находятся
\begin{equation*}
\begin{array}{l}
  v_1=\left(-(1+2\ld^2s),\,0,\,0,\,0,\,\frac{\sqrt{2s(1-2\ld^2s)}}{\sqrt{1+2s}}[(1+s)\ld^2s-1],\,\ld s(3-2\ld^2s)\right), \\
  v_2=\left(0,\,\frac{\ld \sqrt{s(1+2s)}}{\sqrt{2(1-2\ld^2s)}},\,1,\,0,\,0,\,0 \right).
\end{array}
\end{equation*}
Квадратичная форма с матрицей $\mathcal{K}$ на этой паре векторов записывается диагональной матрицей, а собственные числа таковы:
\begin{equation*}
\begin{array}{l}
  \mu_1=\ds{\frac{4s}{1+2s}[1+\ld^2s^2(5-2\ld^2s)]^2}, \qquad  \mu_2=\ds{\frac{1}{s}(8\ld^2s^3-1)}.
\end{array}
\end{equation*}
Ясно, что $\mu_1>0$. Обращение в нуль выражения для $\mu_2$ отвечает множеству $\gac$ вырождения точек ранга 1 (вертикальная прямая на рис.~\ref{fig_spec_case}). На промежутке \eqref{eq5_31} всегда ${\mu_2<0}$. Итак, в особом случае показатели Морса\,--\,Ботта разного знака (этот факт был очевиден, так как траектории гиперболические) и, по доказанному, распределение знаков одинаково на обеих траекториях этого уровня, а именно, при рассмотренном выборе базиса это $(+,-)$. Следует констатировать, что вычисление показателей Морса\,--\,Ботта не дает возможности различить с помощью локального анализа атомы $2B$ и $C_2$.

Получаем следующее утверждение.

\begin{proposition}\label{propos19}
При возрастании интеграла $K$ на изоэнергетическом уровне $\iso$ в точках критической подсистемы $\mo$ на невырожденных критических окружностях имеем следующие бифуркации:

$1)$ для эллиптических траекторий $($тип ``центр''$)$ --- рождение тора при $\ro^2>0$ $($атом $A_+)$, исчезновение тора при $\ro^2<0$ $($атом $A_-)$;

$2)$ для одной гиперболической траектории на критическом уровне $K$ при $\ro^2>0$ --- атом $B_-$ $($``внешнее'' ребро вверх и пара ``внутренних'' ребер вниз$)$, при $\ro^2 < 0$ --- атом $B_+$ $($``внеш\-нее'' ребро вниз и пара ``внутренних'' ребер вверх$)$;

$3)$ для двух гиперболических траекторий на критическом уровне $K$ --- два атома $B$, у которых направление ``внешнего'' ребра определяется по тому же правилу $($обе ``головы'' вверх при $\ro^2 > 0$, обе ``головы'' вниз при $\ro^2 < 0)$, или атом $C_2$.
\end{proposition}

Пункт 4 предложения~\ref{propos14} --- для двух гиперболических траекторий на критическом уровне $K$ имеются разные сочетания знаков в парах $(\mu_1,\mu_2)$ ---  в подсистеме $\mo$ невозможен именно в силу того, что знак $\mu_1$ совпадает со знаком $\ro^2$, то есть одинаков для всех критических окружностей на одном уровне интегралов.

Условия существования движений в подсистемах $\mno$ получены в \cite{mtt40,RyabHarlUdgu2} в терминах интегральных констант $s,h$ и $s,\ell$ соответственно.
Из уравнений \eqref{eq2_28}, \eqref{eq2_29} следует, что $\sgn (\ro^2) =\sgn (Y^2)$ и $\sgn (\ro^2) =\sgn (s \,\Fun^2)$. Пусть $Y_*=Y$, если $\ro$ вещественно, и $Y_*=\ri Y$, если $\ro$ чисто мнимое. Тогда в плоскости $(X,Y_*)$ кривая $\Gamma_0$, заданная уравнением $X^2+Y^2=1$, необходимость которого вытекает из \eqref{eq2_27}, представляет собой окружность или гиперболу, а кривая $\Gamma_1$, заданная уравнением ${\Fun^2(X,Y) = 0}$, при всех $\ro^2 \ne 0$ есть эллипс. Получаем следующее утверждение \cite{mtt40}.
\begin{proposition}\label{propos20}
Для существования вещественных решений \eqref{eq2_27} при заданных $s,h,\ell$, связанных уравнениями поверхностей $\wsi_{2,3}$, необходимо и достаточно выполнение следующих условий:

$1)$ при $s<0$ окружность $\Gamma_0$ и эллипс $\Gamma_1$ имеют общую точку$;$

$2)$ при $s>0, \; \ro^2 \gs 0$ окружность $\Gamma_0$ не лежит целиком строго внутри области, ограниченной эллипсом $\Gamma_1$$;$

$3)$ при $s>0, \; \ro^2 < 0$ гипербола $\Gamma_0$ и эллипс $\Gamma_1$ имеют общую точку.
\end{proposition}

Действительно, непосредственно проверяется, что при $s<0$ точка $X=1,Y=0$ лежит вне эллипса. Движение на $\Gamma_0$ происходит по такому сегменту, где $\Fun^2<0$, то есть внутри эллипса. Поэтому для наличия сегментов окружности внутри эллипса необходимо и достаточно наличие точки пересечения $\Gamma_0\cap \Gamma_1$. При $s>0,\ro^2<0$ в действительном движении $\Fun^2<0$, то есть хотя бы одна точка гиперболы лежит внутри эллипса, что снова означает $\Gamma_0\cap \Gamma_1\ne \varnothing$. Наконец, при $s>0,\ro^2>0$ на окружности должна существовать точка, в которой $\Fun^2>0$, то есть точка вне эллипса. Это не выполняется лишь в том случае, когда окружность $\Gamma_0$ лежит целиком внутри эллипса $\Gamma_1$.

Очевидно соответствие этого утверждения и предложения~\ref{propos17}. При этом количество критических окружностей на фиксированном уровне первых интегралов равно количеству сегментов окружности $\Gamma_0$, лежащих внутри эллипса $\Gamma_1$, при $s<0$, и количеству сегментов гиперболы $\Gamma_0$, лежащих внутри эллипса $\Gamma_1$, при $s>0,\ro^2<0$. При $s>0,\ro^2>0$ количество критических окружностей равно количеству сегментов окружности вне эллипса, за исключением особого случая \eqref{eq5_30}. В этом случае вся окружность лежит вне эллипса, то есть точек пресечения нет, сегмент один, но траекторий две в силу возможности выбора двух знаков перед не обращающимся в нуль радикалом $\Fun$.

Напомним обозначение $D=\sqrt{4+r^2(r-\ld)^2}>0$ и в дополнение к~\eqref{eq5_5} положим
\begin{equation}\notag
    \begin{array}{l}
    \theta_{\pm}(r)=\ds{\frac{r-\ld}{4\ld}\bigl[r(r-\ld) \mp D\bigr]},\\
    \eta_{\pm}(r)= \ds{\frac{1}{2}\bigl[  \ld (r-\ld) \pm D\bigr]} \sqrt{\psi_{\pm}(r)},
    \end{array}
\end{equation}
Как следует из \eqref{eq4_11},\eqref{eq4_15} функции $\varphi_{\pm}(r)$, $\theta_{\pm}(r)$, $\eta_{\pm}(r)$ являются выражениями интегралов $H,S,L$ в критических точках ранга 0, принадлежащих подсистемам $\mn,\mo$.

Следующая теорема соответствует предложению 4 работы \cite{RyabHarlUdgu2}.
\begin{theorem}\label{th11}
$(S,L)$-диаграмма критической системы $\mn$ состоит из следующих множеств:
\begin{equation*}
    \begin{array}{ll}
      \delta_1: \quad s=\theta_-(r), & \ell =\pm \eta_-(r), \quad r \in [0,\ld);\\
      \delta_3: \quad s=\theta_+(r), & \ell =\pm \eta_+(r), \quad r \in (\ld,+\infty).
    \end{array}
\end{equation*}
Внешней границей допустимой области $\mD_2$ служит связная кривая $\delta_1$. Кривая $\delta_3$ состоит из двух компонент и разбивает $\mD_2$ на три подобласти. Точкам подобласти, содержащей значения $\ell=0$, отвечает одна критическая окружность, точкам двух других подобластей, ограниченных кривой $\delta_3$, отвечают две критических окружности. Качественных перестроек диаграммы по параметру $\ld$ не происходит, кроме предельного случая $\ld=0$.
\end{theorem}

\begin{figure}[htp]
\centering
\includegraphics[width=0.5\textwidth,keepaspectratio]{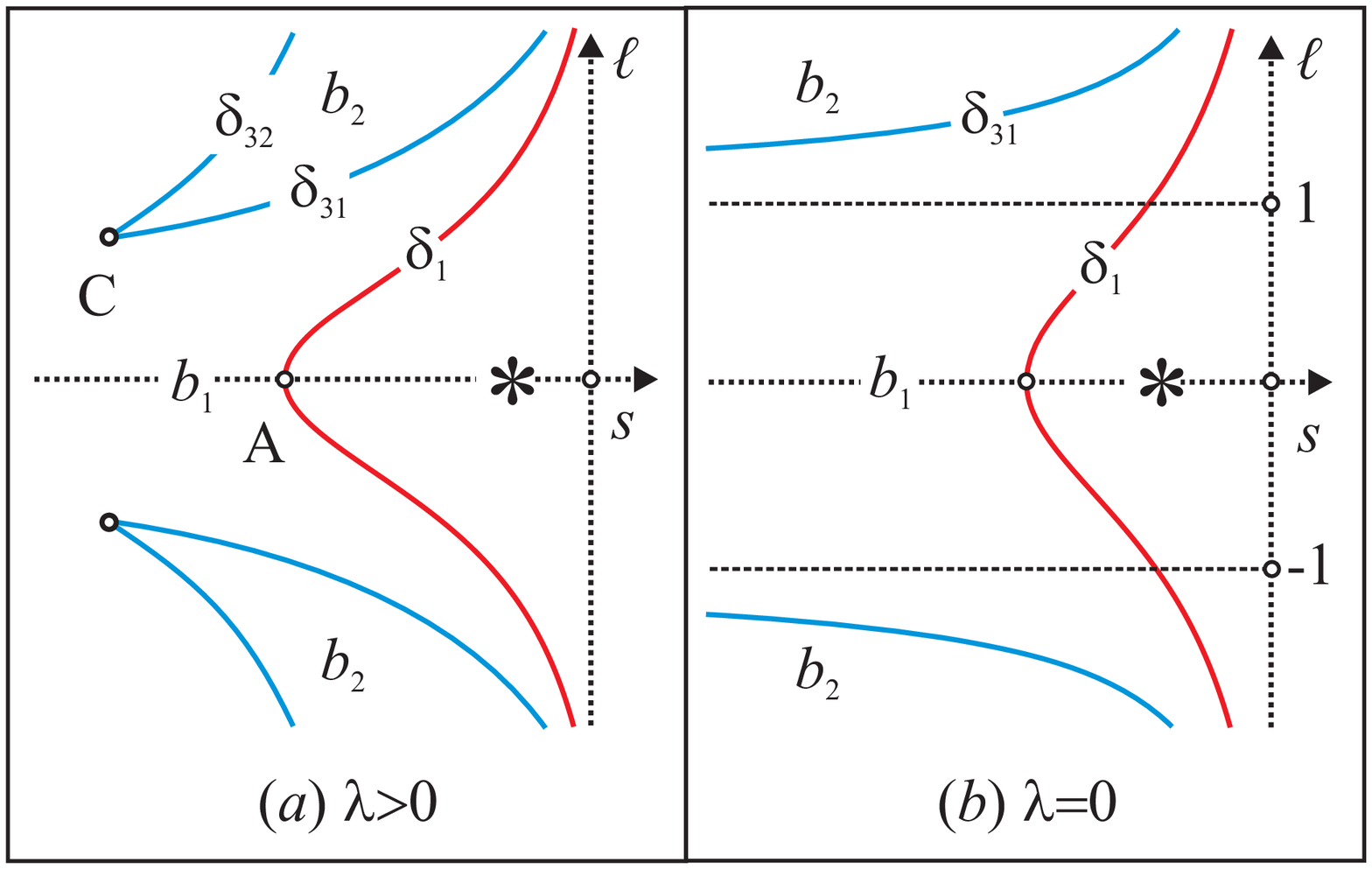}\\
\caption{$(S,L)$-диаграмма системы $\mn$ с полной детализацией.}\label{fig_M2_LKeyAB_red}
\end{figure}

На рис.~\ref{fig_M2_LKeyAB_red} вместе с $(S,L)$-диаграммой второй критической подсистемы (для общего случая $(a)$ $\ld>0$ и для предельного случая $(b)$ $\ld=0$) показаны область $b_1$ с одной критической окружностью  и две симметричных относительно $\ell=0$ области $b_2$ с двумя критическими окружностями. В области, помеченной звездочкой, $\Gamma_0\cap \Gamma_1\ne \varnothing$, поэтому движений нет. Согласно предложению~\ref{propos12} все точки ранга 1 имеют тип ``центр''. Ключевое множество здесь --- только критические точки ранга $0$. Как обычно, образы особых точек, порожденных экстремальными значениями первых интегралов на ключевом множестве, для одних и тех же точек в прообразе на $\mm$ и $\mn$ обозначены одинаково (здесь это точки $A$ и $C$). Свойства соответствующих атомов собраны в табл.~\ref{table4}.

\begin{table}[ht]
\centering
\small
\begin{tabular}{|c|c| c| c| c|c|}
\multicolumn{6}{r}{\fts{Таблица \myt\label{table4}}}\\
\hline
\begin{tabular}{c}\fts{Область}\\[-3pt]\fts{(время жизни)} \end{tabular}
&\begin{tabular}{c}\fts{К-во}\\[-3pt]\fts{окр-стей}\end{tabular}&\begin{tabular}{c}\fts{Показатели}\\[-3pt]\fts{Морса--Ботта}\end{tabular}
&\begin{tabular}{c}\fts{Выход на}\\[-3pt]\fts{$\ld=0$/$\ell=0$}\end{tabular}
&\begin{tabular}{c}\fts{Атом}\end{tabular} &\begin{tabular}{c}\fts{Аналоги}\end{tabular}\\
\hline
\begin{tabular}{c}$b_1$\\($ 0 \ls \ld <+\infty$) \end{tabular} &{1}& {($+\;+$)} & Да/Да &$A_+$ & \begin{tabular}{l} 1 \cite[Рис.\,6.3]{KhBook88}\\ $a_1$ \cite[Рис.\,2]{RyabRCD}\\$\alpha_1$ \cite[Рис.\,11]{BRF}\\$\alpha_1$ \cite[Рис.\,1]{Mor} \end{tabular}\\
\hline
\begin{tabular}{c}$b_2$\\($ 0 \ls \ld <+\infty$) \end{tabular} &{2}& {\begin{tabular}{c}($+\;+$),($+\;+$)\end{tabular}} & Да/Нет &$2A_+$ & \begin{tabular}{l} Переход \ts{III$\to$VI}\\ \cite[Рис.\,6.1{\it д}]{KhBook88}\\$\alpha_2$ \cite[Рис.\,11]{BRF}\end{tabular}\\
\hline
\end{tabular}
\end{table}

Из полученных результатов по критическим подсистемам $\mm,\mn$ вытекает следующее простое описание допустимой области $\mD$ в пространстве $\bbI$ констант общих интегралов~-- образа фазового пространства $\mP^5$ под действием отображения момента \eqref{eq3_6}. Считаем $\ld$ заданным.

\begin{theorem}\label{th12}
Допустимая область $\mD=J(\mP^5)$ есть односвязное множество, внешней границей которого служат образы областей $a_1,a_{12}$ первой критической подсистемы и областей $b_1,b_2$ второй критической подсистемы.
\end{theorem}
\begin{proof}
Фиксируем точку $(\ell,h)$, для которой $\iso \ne \varnothing$. Множество таких точек --- односвязная область на плоскости $O\ell h$, ограниченная снизу кривой $\delta_1$ (см. рис.~\ref{fig_sm1}), гомеоморфная замкнутой полуплоскости. Все $\iso$ --- компактны и связны, поэтому $K(\iso)$ --- отрезок, стягивающийся в точку над $\delta_1$. В частности, $\mD$ --- приведенное расслоение отрезков над замкнутой полуплоскостью, гомеоморфное (но не диффеоморфное) замкнутому полупространству.

Фиксируем $h$, рассмотрим компактное инвариантное подмножество $Q_h^4=H^{-1}(h) \subset \mP^5$ и  бифуркационную диаграмму $\Sigma_{LK}(h)$ отображения
\begin{equation}\notag
    L{\times}K|_{Q_h^4}: Q_h^4 \to \bR^2.
\end{equation}
Очевидно, что $\Sigma_{LK}(h)$ есть образ при отображении момента пересечений $Q_h^4$ с критическими подсистемами $\mi_i$.
Рассмотрим замкнутую область, ограниченную образом полосы между $\delta_1,\ell_0$ плоскости $(p,h)$ на поверхности $\wsa$ (то есть образом областей $a_1,a_{12}$) и образом всей подсистемы $\mn$ на $\wsb$ (то есть образом областей $b_1,b_2$). Образ подсистемы $\mo$ заведомо имеет точки внутри этой кривой (например, при $\ell=0$). В пространстве $\bbI$ образ $\mo$ связен, поэтому если бы существовали точки в $J(\mo)$ за пределами указанной области, то имелось бы трансверсальное пересечение образа $J(\mo \cap Q_h^4)$ и замыкания объединения кривых $a_1 \cup a_{12}$ (легко понять просто из геометрии $\wsa,\wsc$, что общая часть, служащая кривой касания, за пределы этой области поверхность $\wsc$ не выводит). Такое пересечение происходит по части кривой $\delta_2$ на границе между $a_1,a_{12}$, однако, по доказанному, на площадке, трансверсальной к $\mm$ в прообразе соответствующих точек диаграммы системы $\mm$, все критические точки ранга 0 и 1 имеют эллиптический тип, поэтому выход за пределы указанной ограниченной области невозможен.
\end{proof}

Различные типы оболочки сечений множества $\Sigma \subset \bbI$ плоскостями $h=\cons$ приведены на рис.~\ref{fig_LK_admreg}: $(a)$~$-1<h<\ld^2/2$; $(b)$~$\ld^2/2<h<h_C(\ld)$; $(c)$~$h>h_C(\ld)$ (без сечений поверхности $\wsc$, целиком лежащей внутри). При $h=-1$ сечение допустимой области стягивается в точку. В случае $(c)$ пунктиром показан образ области $\overline{a_2}$, для точек которой, как показано выше, не существует вещественных движений. Таким образом, допустимая область имеет форму ``носа ладьи'', причем вначале ``палуба'' гладкая, а затем в середине возникает ``ребро''.

\begin{figure}[htp]
\centering
\includegraphics[width=110mm,keepaspectratio]{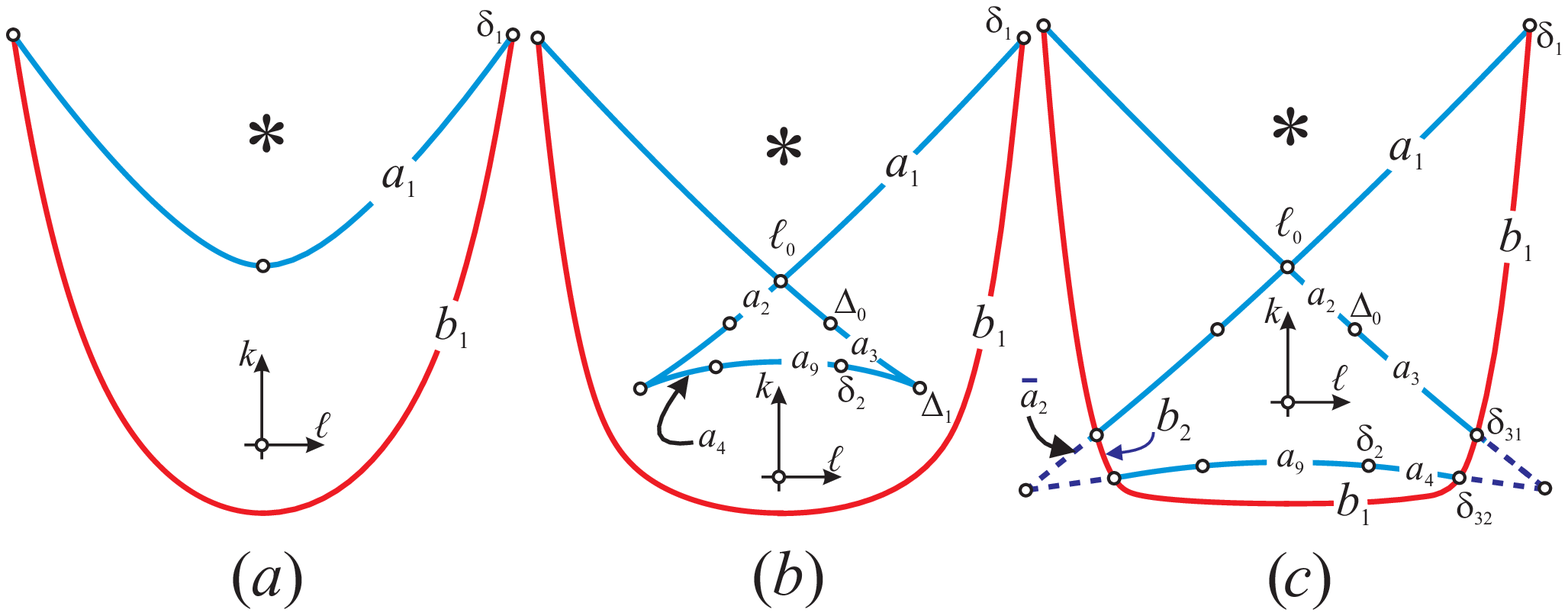}\\
\caption{Сечения допустимой области плоскостями $h=\cons$.}\label{fig_LK_admreg}
\end{figure}

\begin{figure}[htp]
\centering
\includegraphics[width=60mm,keepaspectratio]{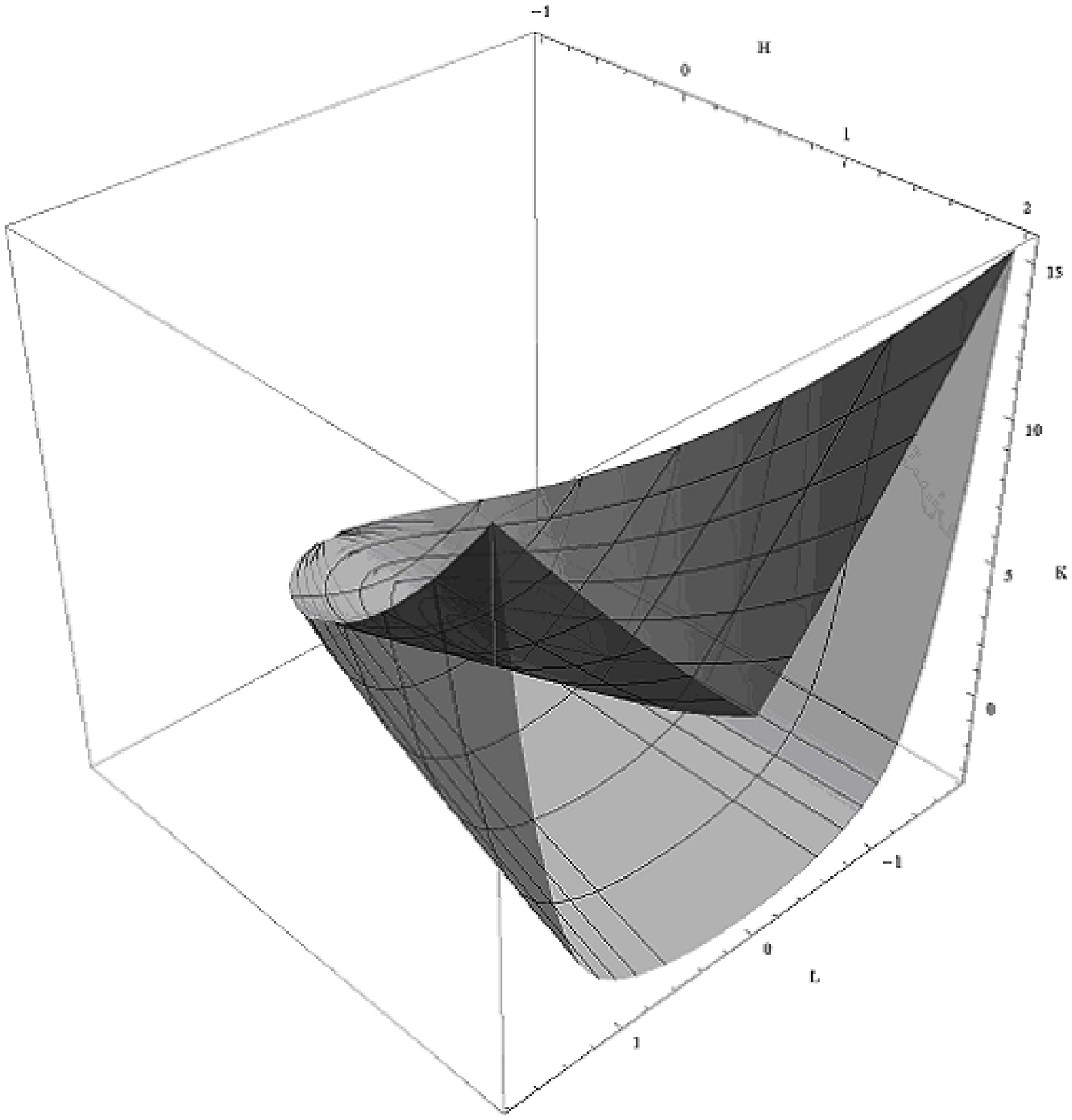}\\
\caption{Носовая часть допустимой области.}\label{fig_boat}
\end{figure}

Перейдем к классификации областей в образе подсистемы $\mo$. Полное описание допустимой области на плоскостях констант интегралов $S,H$ и $S,L$ дано в работах \cite{mtt40,RyabHarlUdgu2}. Следующая теорема соответствует предложению 5 работы \cite{RyabHarlUdgu2}.

\begin{theorem}\label{th13}
$(S,L)$-диаграмма критической системы $\mo$ состоит из следующих множеств:
\begin{equation*}
    \begin{array}{ll}
      \delta_2: & s=\theta_+(r), \quad \ell=\pm \eta_+(r),  \quad r \in (-\infty,0];\\[2mm]
      \gan: & \ds{\ell= \pm \sqrt{\frac{s}{2}(1-2\ld^2s)}}, \quad 0< s \ls \ds{\frac{1}{2\ld^2}};\\[2mm]
      \gac: &  s=\ds{\frac{1}{2\ld^{2/3}}}, \quad
      \left\{
      \begin{array}{ll}  \ell \in \bR, & \ld \ls \ld^*\\
                      |\ell| \gs  \ds{\frac{2\ld^{2/3}-\sqrt{4+\ld^{4/3}}}{\sqrt{2}(\sqrt{4+\ld^{4/3}}-\ld^{2/3} )^{1/2}}}, & \ld > \ld^*
      \end{array}
      \right. .
      \end{array}
\end{equation*}
В допустимую область $\mD_3$ не входят следующие компоненты дополнения к диаграмме, в которых не существует критических движений: при всех $\ld$ --- область, прилегающая к оси $s=0$ и ограниченная ветвями кривых $\gan, \delta_2$, при $\ld>\ld^*$ --- область, ограниченная кривой $\delta_2$ между двумя ее точками пересечения с осью $\ell=0$ при $r \neq 0$.

Перестройки типов диаграмм в области $\ld \gs 0$ происходят при следующих значениях параметра: $0, \ld_*, \, 1,\,\ld^*, \, \sqrt{2}$.
\end{theorem}

Неравенства для $r,s,\ell$, определяющие допустимую область для точек ключевого множества, мгновенно следуют из критерия, данного в предложении \ref{propos17}.

Для критических точек подсистем $\mno$ отдельно сформулируем утверждение о существовании вырожденных периодических решений, аналогичное предложению~\ref{propos16}.

\begin{proposition}\label{propos21}
В допустимую область на бифуркационной диаграмме $\Sigma$ отображения момента $J$ входит следующий сегмент ребра возврата $\wsc$ -- образ вырожденных критических движений ранга $1$:
\begin{equation}\notag
    \gac: \quad \ds{\ell = \pm \sqrt{\frac{h-h^*}{2}}}, \quad h \in \left\{
    \begin{array}{ll}
    \ [h^*,+\infty ), & \textrm{при} \quad \ld \leqslant \ld^* \\
    \ [h^{**},+\infty), & \textrm{при} \quad \ld > \ld^*
    \end{array}      \right.,
\end{equation}
где
\begin{equation}\notag
\begin{array}{ll}
  h^* & = \ds{\frac{\ld^{2/3}}{2}(3-\ld^{4/3})}, \\[3mm]
  h^{**} & = \ds{\frac{1}{4}\left[(4+\ld^{4/3})^{3/2}-\ld^{2/3}(6+\ld^{4/3})\right]}=\\[3mm]
  {} & = \ds{\frac{1}{8}\left(\sqrt{4+\ld^{4/3}}-\ld^{2/3}\right)^2\left(2\sqrt{4+\ld^{4/3}}+\ld^{2/3} \right)}.
\end{array}
\end{equation}
или, в терминах констант $s,\ell$
\begin{equation}\label{eq5_32}
    \gac: \quad s=\frac{1}{2\ld^{1/3}}, \qquad \left\{ \begin{array}{ll}
    \ell \in \bR, & \ld \ls \ld^* \\
    \ds{|\ell| \gs \frac{2 \ld^{2/3} - \sqrt{4+\ld^{4/3}}}{\sqrt{2}(\sqrt{4+\ld^{4/3}}-\ld^{2/3} )^{1/2}}}, & \ld >\ld^*
    \end{array} \right. .
\end{equation}
\end{proposition}

Утверждение в координатах $(s,\ell)$ --- часть теоремы~\ref{th13}. Пересчет на плоскость $(\ell,h)$ выполнен по формулам \eqref{eq3_8}. Из них следует, что на $\wsc$
\begin{equation}\label{eq5_33}
  h=2\ell^2+h_0(s), \qquad h_0(s)=\frac{1}{2s}-\frac{\ld^2}{2}+2\ld^2s^2.
\end{equation}
Нетрудно видеть, что наименьшее значение $h_0(s)$ достигается в точности при $s=\frac{1}{2\ld^{1/3}}$, то есть в точке пересечения кривой минимума $h$ на $\mo$ с образом ребра возврата $\gac$, и это значение есть
$$
\min_{s>0} h_0(s) = h^*.
$$
Тогда при $\ld \ls \ld^*$ и $\ell\in \bR$ имеем $h\in [h^*,+\infty)$. Граничная точка -- новая особая точка на ключевых множествах, имеющая значение при рассмотрении $H$-атласов, так она является точкой экстремума $h$-координаты на образе $\gac$ семейства вырожденных точек ранга 1:
\begin{equation}\label{eq5_34}
D_5: \, \ell=0, \quad s_2=\frac{1}{2\ld^{1/3}}, \quad h=\frac{\ld^{2/3}}{2}(3-\ld^{4/3}), \quad \ld \in [0,\ld^*].
\end{equation}

В случае же $\ld >\ld^*$, вычисляя значение \eqref{eq5_33} на $\gac$ при наименьшем допустимом значении $|\ell|$ из \eqref{eq5_32}, приходим к величине $h^{**}$. Эти значения отвечают ранее полученной особой точке $B_3$ в составе ключевых множеств.

На следующих рисунках~\ref{fig_M3_LKeyABC_red} и \ref{fig_M3_LKeyDEF_red} представлены $(S,L)$-диаграммы критической подсистемы $\mo$: $(a)$~$0<\ld<\ld_*$; $(b)$~$\ld_*<\ld<1$; $(c)$~$1<\ld<\ld^*$; $(d)$~$\ld^*<\ld<\sqrt{2}$; $(e)$~$\ld>\sqrt{2}$; $(f)$~предельный случай $\ld=0$. Здесь уместно напомнить обозначения
$$
\ld_*=1/2^{3/4}, \qquad \ld^*=2\sqrt{2}/3^{3/4}.
$$
Пунктиром показаны кривые вырождения $\gan,\gac$ за исключением той части кривой $\gan$, которая является внешней границей допустимой области $\mD_3$ --- она показана сплошной линией. Как и ранее, звездочкой отмечены области (здесь это $\overline{c_1},\overline{c_2}$), в которых критические движения отсутствуют. На диаграммах нанесена детализация ключевых множеств в соответствии с данной выше классификацией критических точек ранга 0, а также очевидным разбиением на гладкие участки кривых $\gan$ (разбиение уже получено ранее в подсистеме $\mm$) и $\gac$. Напомним, что это разбиение порождено общими точками этих кривых с образом точек ранга 0 и точками взаимного пересечения. Соответствующая аналитика будет предъявлена в следующем разделе.

\def\wid{0.7}
\begin{figure}[!htp]
\centering
\includegraphics[width=\wid\textwidth, keepaspectratio]{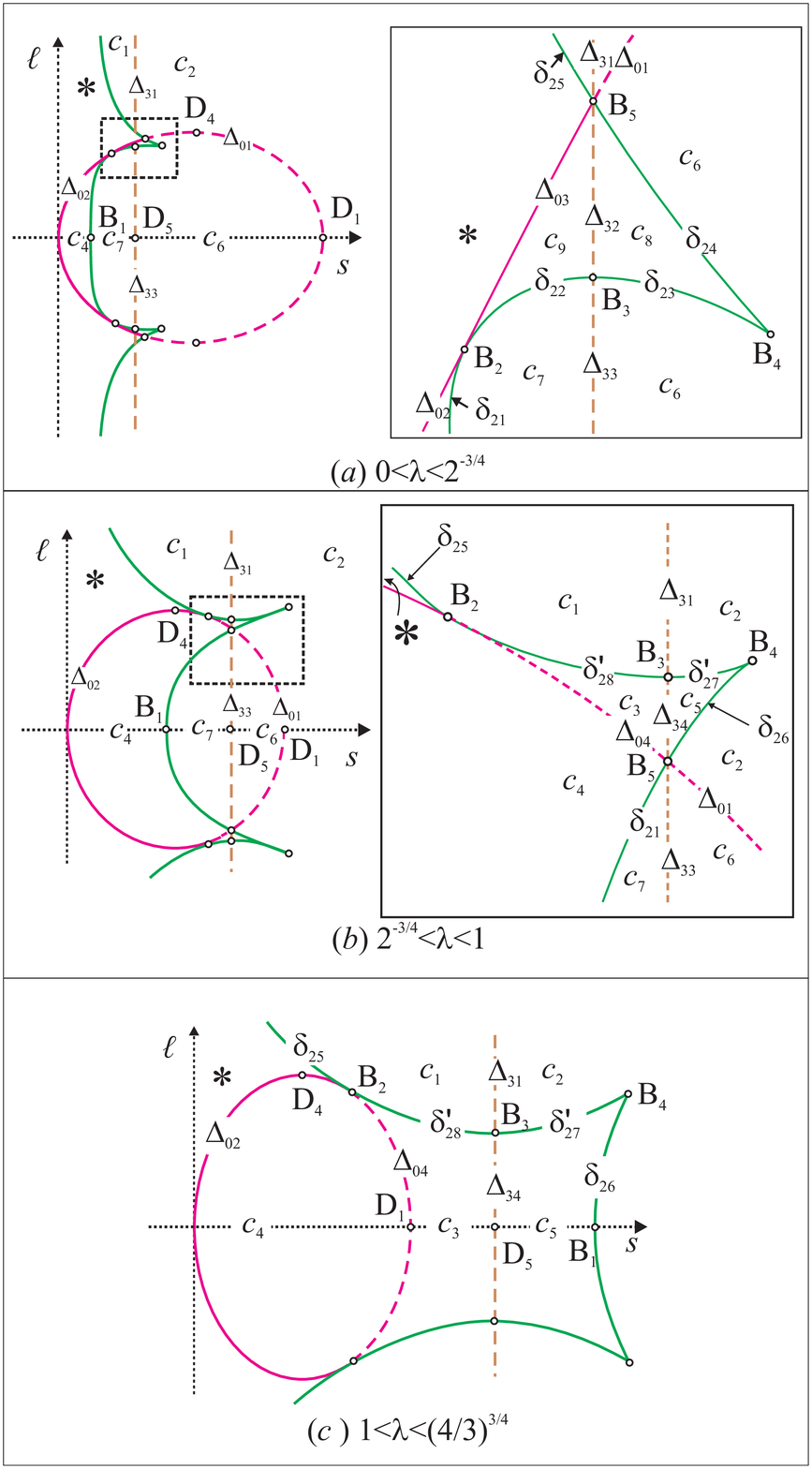}\\
\caption{$(S,L)$-диаграммы подсистемы $\mo(\ld)$ с полной детализацией.}\label{fig_M3_LKeyABC_red}
\end{figure}

\begin{figure}[!htp]
\centering
\includegraphics[width=\wid\textwidth, keepaspectratio]{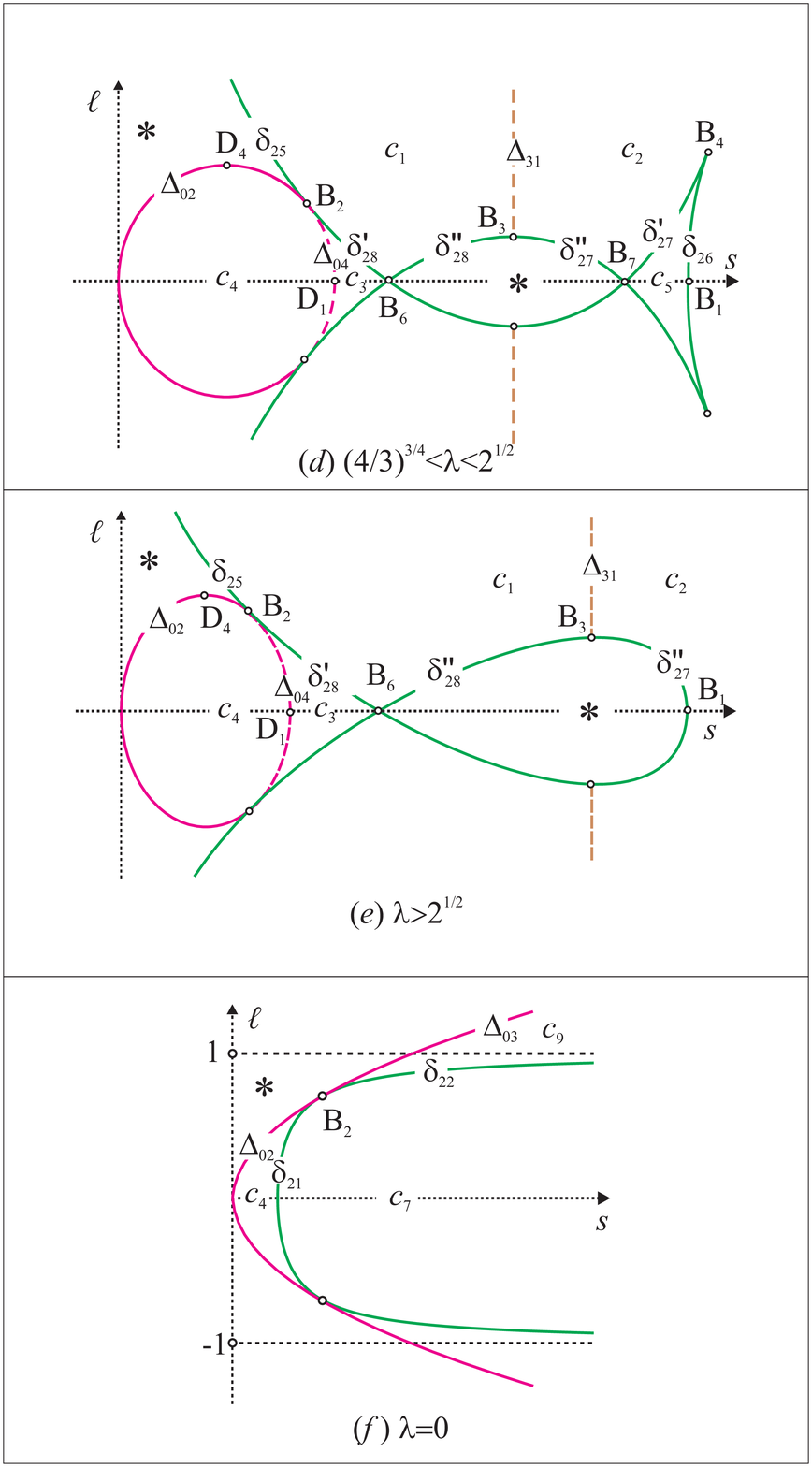}
\caption{$(S,L)$-диаграммы подсистемы $\mo(\ld)$ с полной детализацией (продолжение).}\label{fig_M3_LKeyDEF_red}
\end{figure}

Теорема \ref{th13} в утверждениях относительно кривых в составе диаграммы является следствием определения диаграммы. Перестройки диаграммы по $\ld$ вычисляются непосредственно (детали см. в \cite{RyabHarlUdgu2}) или же по общему свойству: это разделяющие значения $\ld$ для сечений $\ld=\cons$ семейства разделяющих кривых для точек ранга 0 в плоскости $(r,\ld)$ и совокупности экстремумов $\ell$ на образах вырожденных точек ранга 1. Для доказательства теоремы в части существования движений достаточно проанализировать распределение корней и знаков старших коэффициентов соответствующих многочленов $Z_{\pm}$ и воспользоваться предложением~\ref{propos17}. Наглядное доказательство дает компьютерная визуализация кривых $\Gamma_0,\Gamma_1$ в соответствии с предложением~\ref{propos20}. Тот факт, что в области $\overline{c_1}$ (см. рис.~\ref{fig_M3_LKeyABC_red},\,$(a)$), примыкающей к оси $O\ell$, движения отсутствуют, практически очевиден, поскольку $s \ne 0$. Отсутствие критических движений в области $\overline{c_2}$ (рис.~\ref{fig_M3_LKeyDEF_red},\,$(d),(e)$) ранее было доказано в \cite{RyabDis} путем исследования точек оси $\ell=0$.

На рис.~\ref{fig_M3_LKeyABC_red}, \ref{fig_M3_LKeyDEF_red} также приведены и особые точки, возникающие на $(S,L)$-диаграммах подсистемы $\mo$. Непосредственно проверяем, что все они уже встречались в подсистеме $\mm$, за исключением введенной для использования в $H$-атласе точки $D_5$, описанной выше и заданной уравнениями \eqref{eq5_34}.

Переформулируем теорему \ref{th13} в терминах интегралов $S,H$ \cite{mtt40}.
\begin{theorem}\label{th13h}
$(S,H)$-диаграмма критической системы $\mo$ состоит из кривых
\begin{equation}\notag
    \begin{array}{lll}
      \gan: & h=\htan(s), & 0< s \ls \ds{\frac{1}{2\ld^2}};\\[2mm]
      \delta_2: & h=\varphi_+(r), & s=\theta_+(r), \quad r \in (-\infty,0]; \\[2mm]
     \gac: &  s=\ds{\frac{1}{2\ld^{2/3}}}, &
      \left\{
      \begin{array}{ll}  h \gs h^*, & \ld \ls \ld^*\\
                      h \gs h^{**}, & \ld > \ld^*
      \end{array}
      \right. ;\\
      h_{\rm min}: & h=h_{0}(s) , & s \in I(\ld),
    \end{array}
\end{equation}
где
\begin{equation}\notag
  \htan(s)=\ds{\frac{1-\ld^2 s+2 s^2}{2s}},
\end{equation}
а область изменения $s$ на кривой $h_{\rm min}$ определяется так:
\begin{equation}\notag
    I(\ld)= \left\{
    \begin{array}{ll}
      (0,+\infty), & \ld \ls \ld^* \\
      (0, s_0] \cup [s^0, +\infty), & \ld^* \ls \ld \ls \sqrt{2}\\
      (0, s_0] \cup [1/2, +\infty), & \ld \gs \sqrt{2}
    \end{array}\right.,
\end{equation}
$s_0(\ld),s^0(\ld)$ -- абсциссы точек касания кривых $\delta_2$ и $h_{\rm min}$, существующих при $\ld \gs \ld^*$.  Они удовлетворяют неравенствам $s_0(\ld)<s^0(\ld)$ при $\ld > \ld^*$ и  $s^0(\ld)<1/2$ при $\ld^* \ls \ld <\sqrt{2}$.

Внешними границами допустимой области служат:

$1)$ кривая $h_{\rm min}$ в пределах $s\in I(\ld);$

$2)$ кривая $\gan$ в пределах
\begin{equation}\notag
    s \in  \left\{
    \begin{array}{ll}
      (0, 1/(2 \ld^{2/3})] , & \ld \ls \ld_* \\
      (0, \sqrt{1+\ld^4}-\ld^2], & \ld \gs \ld_*
    \end{array}
    \right. ;
\end{equation}

$3)$ кривая $\delta_2$ в пределах
\begin{equation}\notag
   s\in
   \left\{
   \begin{array}{ll}
      \, [s_0, s^0]  , & \ld^* \ls \ld \ls \sqrt{2} \ \\
      \, [s_0, 1/2]    , & \ld \gs \sqrt{2}
   \end{array}
   \right. .
\end{equation}

Перестройки типов диаграмм в области $\ld \gs 0$ происходят при значениях параметра $0, \ld_*, \, 1,\,\ld^*, \, \sqrt{2}$.
\end{theorem}

Кривая $h_{\rm min}$ отвечает значению $\ell=0$ и ее появление естественно связано со складкой при накрытии поверхностью $\wsc$ плоскости $(s,h)$, поскольку $h$ выражается через $\ell^2$. Полное исследование $(S,H)$-диаграмм подсистемы $\mo$ и их компьютерная визуализация выполнены в работах \cite{mtt40,SaKhSh2012}. Значения $s_0,s^0$ соответствуют особым точкам $B_6,B_7$, отмеченных на $(S,L)$-диаграммах на рис.~\ref{fig_M3_LKeyDEF_red},$(d),(e)$. Эти точки существуют при $\ld \gs \ld^*$ (совпадая при $\ld=\ld^*$), причем в момент $\ld=\sqrt{2}$ точка $B_7$ сливается с $B_1$ и затем переходит в недопустимую область (на $\delta_2$ соответствующее значение $r$ было бы положительным, что не так). Доказательства и пояснения имеются в работе \cite{mtt40}. Они легко следуют из уравнения связи $s$ и $\ld$ в этих точках, полученных исключением вспомогательного параметра $x$ из соответствующего параметрического представления \eqref{eq5_21} значений $s_2,\ld$ в этих точках (см. рис.~\ref{fig_s00}):
\begin{equation}\notag
\ld = \frac{\sqrt{1+2s}-\sqrt{1-2s}}{(2s)^{3/2}}, \qquad s \in (0,\frac{1}{2}].
\end{equation}

\begin{figure}[!ht]
\centering
\includegraphics[width=0.4\textwidth,keepaspectratio]{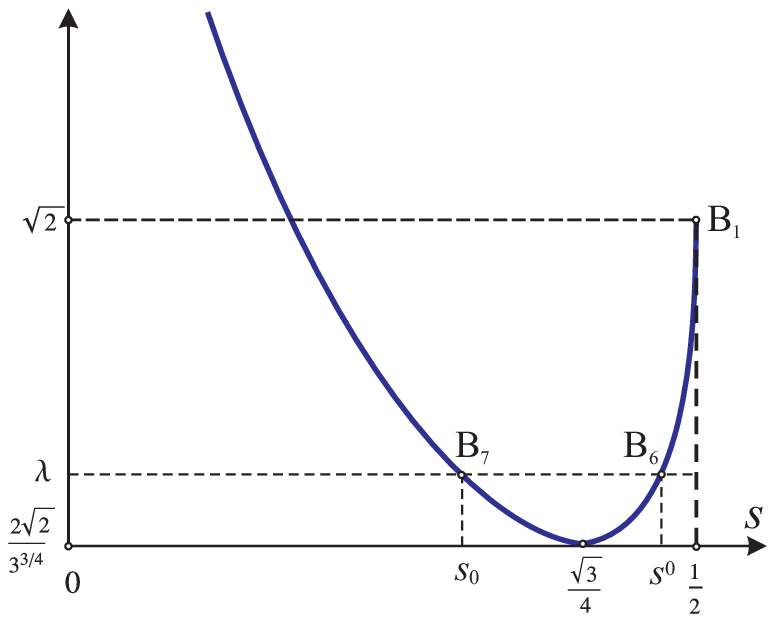}
\caption{Связь $s$ и $\ld$ в точках $B_{6,7}$.}\label{fig_s00}
\end{figure}

Применение полученных результатов к точкам областей $c_1-c_9$ в $(s,\ell)$-образе подсистемы $\mo$ приводит к описанию характеристик и атомов, собранных в табл.~\ref{table5}.
Как видим, все области, кроме $c_1,c_8$, при рассмотрении расширенных диаграмм в пространстве $(s,\ell,\ld)$ имеют выход на соответствующие области исследованных ранее задач ($\ld=0$ или $\ell=0$), поэтому для атомов здесь добавлена лишь их направленность. В частности, наличие атома $C_2$ в области $c_4$ и двух атомов $B$ в области $c_9$ обосновано в работах \cite{KhPMM83, KhDan83, KhBook88} (в других обозначениях). Наличие двух атомов $B$ в области $c_5$ следует из результатов \cite{RyabRCD}.
В новых областях $c_1,c_8$ по доказанному выше критические окружности имеют эллиптический тип, их количество вычисляется по приведенным критериям, а направленность определяется показателями Морса\,--\,Ботта.

\begin{center}
\small

\begin{longtable}{|c|c| c| c| c|c|}
\multicolumn{6}{r}{\fts{Таблица \myt\label{table5}}}\\
\hline
\begin{tabular}{c}\fts{Область}\\[-3pt]\fts{(время жизни)} \end{tabular}
&\begin{tabular}{c}\fts{К-во}\\[-3pt]\fts{окр-стей}\end{tabular}&\begin{tabular}{c}\fts{Показатели}\\[-3pt]\fts{Морса--Ботта}\end{tabular}
&\begin{tabular}{c}\fts{Выход на}\\[-3pt]\fts{$\ld=0$/$\ell=0$}\end{tabular}
&\begin{tabular}{c}\fts{Атом}\end{tabular} &\begin{tabular}{c}\fts{Аналоги}\end{tabular}\\
\hline\endfirsthead%
\multicolumn{6}{r}{\fts{Таблица \ref{table5} (продолжение)}}\\
\hline
\begin{tabular}{c}\fts{Область}\\[-3pt]\fts{(время жизни)} \end{tabular}
&\begin{tabular}{c}\fts{К-во}\\[-3pt]\fts{окр-стей}\end{tabular}&\begin{tabular}{c}\fts{Показатели}\\[-3pt]\fts{Морса--Ботта}\end{tabular}
&\begin{tabular}{c}\fts{Выход на}\\[-3pt]\fts{$\ld=0$/$\ell=0$}\end{tabular}
&\begin{tabular}{c}\fts{Атом}\end{tabular} &\begin{tabular}{c}\fts{Аналоги}\end{tabular}\\
\hline\endhead
\myrulm\begin{tabular}{c}$c_1$\\($ 0 < \ld <+\infty$) \end{tabular} &{1}& {($-\;-$)} & Нет/Нет &$A_-$ & {Отсутств.}\\
\hline
\myrulm\begin{tabular}{c}$c_2$\\($ 0 < \ld <+\infty$) \end{tabular} &{1}& {($-\;+$)} & Нет/Да &$B_+$ & \begin{tabular}{l} $a_5$ \cite[Рис.\,2]{RyabRCD}\\$\beta_3$ \cite[Рис.\,1]{Mor}\end{tabular}\\
\hline
\myrulm\begin{tabular}{c}$c_3$\\($ \ld_* < \ld < +\infty$) \end{tabular} &{2}& {\begin{tabular}{c}($-\;-$),($-\;-$)\end{tabular}} & Нет/Да &$2A_-$ & \begin{tabular}{l} $b_4$ \cite[Рис.\,3]{RyabRCD}\\$\alpha_7$ \cite[Рис.\,1]{Mor}\end{tabular}\\
\hline
\begin{tabular}{c}$c_4$\\($ 0 \ls \ld < +\infty $) \end{tabular} &{2}& {\begin{tabular}{c}($+\;-$),($+\;-$)\end{tabular}} & Да/Да &$C_2$ & \begin{tabular}{l} 8 \cite[Рис.\,6.3]{KhBook88}\\ $a_4, b_5$ \cite[Рис.\,2,3]{RyabRCD}\\$\beta_2$ \cite[Рис.\,11]{BRF}\\$\gamma$ \cite[Рис.\,1]{Mor} \end{tabular}\\
\hline
\myrulm\begin{tabular}{c}$c_5$\\($ \ld_* < \ld <\sqrt{2} $) \end{tabular} &{2}& {\begin{tabular}{c}($-\;+$),($-\;+$)\end{tabular}} & Нет/Да &$2B_+$& \begin{tabular}{l} $b_3$ \cite[Рис.\,3]{RyabRCD}\\$\beta_4$ \cite[Рис.\,1]{Mor} \end{tabular}\\
\hline
\myrulm\begin{tabular}{c}$c_6$\\($ 0 < \ld < 1$) \end{tabular} &{1}& {($+\;+$)} & Нет/Да &$A_+$& \begin{tabular}{l} $a_3, a_4$ \cite[Рис.\,2]{RyabRCD}\\$\alpha_4$ \cite[Рис.\,1]{Mor}\end{tabular}\\
\hline
\begin{tabular}{c}$c_7$\\($ 0 \ls \ld < 1$) \end{tabular} &{1}& {($+\;-$)} & Да/Да &$B_-$& \begin{tabular}{l} 7 \cite[Рис.\,6.3]{KhBook88}\\$a_3$ \cite[Рис.\,2]{RyabRCD}\\$\beta_1$ \cite[Рис.\,11]{BRF}\\$\beta_2$ \cite[Рис.\,1]{Mor}\end{tabular}\\
\hline
\myrulm\begin{tabular}{c}$c_8$\\($ 0 < \ld < \ld_*$) \end{tabular} &{2}& {\begin{tabular}{c}($+\;+$),($+\;+$)\end{tabular}} & Нет/Нет &$2A_+$& {Отсутств.} \\
\hline
\myrulm\begin{tabular}{c}$c_9$\\($ 0\ls \ld < \ld_*$) \end{tabular} &{2}& {\begin{tabular}{c}($+\;-$),($+\;-$)\end{tabular}} & Да/Нет &$2B_-$& \begin{tabular}{l} E\,\cite[Рис.\,2]{KhDan83}\\$\beta_3$ \cite[Рис.\,11]{BRF} \end{tabular}\\
\hline
\end{longtable}

\end{center}


\clearpage

\subsection{Классы вырожденных точек ранга 1}

\subsubsection{Множество $\gan$}

Первый класс вырожденных точек ранга 1 -- это точки, лежащие в прообразе линии касания бифуркационных поверхностей. Образ этого множества в пространстве интегральных констант обозначается через $\gan$. Это -- линия касания поверхностей $\wsa$ и $\wsc$. Ее уравнение в параметрах $s,h$ имеет вид
\begin{equation}\label{eq5_35}
  2s^2-2(h+\frac{\ld^2}{2})s+1=0,
\end{equation}
а в параметрах $s,\ell$ таково:
\begin{equation}\label{eq5_36}
  2\ld^2s^2-s+2\ell^2=0.
\end{equation}
В совокупности с условиями вещественности получаем следующее представление $\gan$:
\begin{equation}\notag
  \ell =\displaystyle{\pm \sqrt{\frac{s}{2}(1-2\ld^2 s)}}, \quad \displaystyle{h=s+\frac{1}{2s}-\frac{\ld^2}{2}}, \quad  \displaystyle{s\in (0,\frac{1}{2\ld^2}]}.
\end{equation}
Это множество делится на качественно различные части его пересечением с образом множества критических точек ранга 0, то есть с подмножествами $\delta_i$ ($i=1,2,3$) пространства интегральных констант. Действительно, вне таких пересечений окрестность точки множества $\gan$ подвергается диффеоморфизму (как стратифицированное многообразие) и топология соответствующего прообраза сохраняется.

Напомним уравнения множеств $\delta_i$:
\begin{equation}\label{eq5_37}
  \begin{array}{l}
    \ell = \mp \displaystyle{\frac{1}{2}[\ld(r-\ld)+d]}\sqrt{\displaystyle{\frac{r}{2}\left[-r+\frac{1}{r-\ld}d\right]}},\quad
    h= -\displaystyle{\frac{1}{2}r(r-\ld)+\frac{2r-\ld}{2(r-\ld)}d}, \\
    d^2 = 4+r^2(r-\ld)^2, \qquad  r \in (-\infty,0] \cup [0,\ld) \cup (\ld,+\infty),
\end{array}
\end{equation}
при этом соответствие номеров подмножеств, промежутков изменения константы $r$, определяющей относительное равновесие, и знаков величины $d$ таково:
\begin{equation*}
  \begin{array}{lll}
    \delta_1: & r \in [0,\ld), & d <0; \\
    \delta_2: & r \in (-\infty,0], & d > 0; \\
    \delta_3: & r \in (\ld,+\infty), & d > 0.
  \end{array}
\end{equation*}
Параметр $s$ поверхностей $\wsi_i$, служащий частным интегралом для критических подсистем $\mi_i$, в точках \eqref{eq5_37} принимает следующие значения: в подсистеме $\mm$ (все три кривые $\delta_i$)
\begin{equation}\label{eq5_38}
  s=\frac{1}{2}\left[\ld (r-\ld)+d\right],
\end{equation}
в подсистемах $\mn$ (кривые $\delta_1,\delta_3$) и $\mo$ (кривая $\delta_2$)
\begin{equation}\label{eq5_39}
  s=\frac{r-\ld}{4\ld}\left[r (r-\ld)-d\right].
\end{equation}
Воспользуемся выражением на $\mm$, тогда в точках перестройки $\gan$ уравнения \eqref{eq5_37} должны быть совместны, например, с \eqref{eq5_35}. Получим
\begin{equation*}
  r^2-2\ld^2+\frac{2}{\ld(r-\ld)+d}-\frac{r}{r-\ld}d =0.
\end{equation*}
Отсюда либо
\begin{equation}\label{eq5_40}
  r=-\ld,
\end{equation}
либо
\begin{equation*}
  2+(r-\ld)^2(r^2-2\ld r +2\ld^2)-(r-\ld)(r-2\ld)d=0.
\end{equation*}
Последнее равенство можно записать в виде
\begin{equation}\label{eq5_41}
  \frac{1}{2}\left[(r-\ld)(r-2\ld)-d\right]^2=0,
\end{equation}
что, в частности, означает, что оно отвечает за точку касания $\gan$ с одним из $\delta_i$. Выясним, каким $\delta_i$ соответствуют значения \eqref{eq5_40}, \eqref{eq5_41}. Для \eqref{eq5_40} очевидно, что $r<0$, поэтому это -- точка пересечения $\gan$ с $\delta_2$. Следствием уравнения \eqref{eq5_41} является $\ld (r-\ld)^3+1=0$, откуда
\begin{equation}\label{eq5_42}
  r=\ld-\frac{1}{\ld^{1/3}}.
\end{equation}
Но при этом, согласно \eqref{eq5_41}, должно быть $\sgn d =\sgn (r-\ld)(r-2\ld)$, что при очевидном неравенстве $r<\ld$ имеет место лишь при $r \ls 0$, то есть на кривой $\delta_2$. Итак, точка \eqref{eq5_42} касания $\gan$ и $\delta_2$ существует тогда и только тогда, когда $\ld \ls 1$.

Таким образом, принимая на $\gan$ в качестве независимых параметров $s$ и $\ld$ (если заменить $s$ на $h$ или $\ell$, то соответствие в области существования не будет взаимно однозначным), получим в качестве разделяющих кривых в квадранте $\{(\ld,s): \ld \gs 0, \, s > 0\}$ верхнюю границу допустимых значений, отвечающую точке $D_1$ в ключевых множествах,
\begin{equation}\notag
  \varphi_0: s = \frac {1}{2\ld^2}, \qquad \ld >0,
\end{equation}
кривую, полученную из \eqref{eq5_38}, \eqref{eq5_40} и отвечающую точке $B_2$ в ключевых множествах,
\begin{equation}\notag
  \varphi_1: s = \sqrt{1+\ld^4}-\ld^2, \qquad \ld \gs 0,
\end{equation}
и кривую, полученную из \eqref{eq5_38}, \eqref{eq5_42} и отвечающую точке $B_5$ в ключевых множествах,
\begin{equation}\notag
  \varphi_2: s = \frac{1}{2\ld ^{2/3}}, \qquad \ld \in (0,1].
\end{equation}
Соответствующая область показана на рис.~\ref{fig_sepdelta0}, где нанесены обозначения возникающих подобластей в соответствии с рисунками \ref{fig_M1_LKeyABC_red}, \ref{fig_M1_LKeyDEF_red}, \ref{fig_M3_LKeyABC_red}, \ref{fig_M3_LKeyDEF_red}, на которых каждая такая подобласть порождает сегмент ключевого множества. Далее в бифуркационных диаграммах $\mSell$ эти подобласти выступают как узловые точки. Отметим, что разделяющие кривые для $\gan$ отвечают структурно неустойчивым диаграммам $\mSell$: на плоскости $(\ld,\ell)$ вычисление $\ell(\ld)$ из формулы \eqref{eq5_36} переводит кривую $\varphi_0$ в ось $\ell=0$, кривую $\varphi_1$ -- в кривую $\vpi_{22}$ разделяющего множества $\Theta_L$, кривую $\varphi_3$ -- в кривую $\vpi_{21}$ в соответствии с формулами \eqref{eq4_48}.

\begin{figure}[htp]
\centering
\includegraphics[width=0.3\textwidth, keepaspectratio]{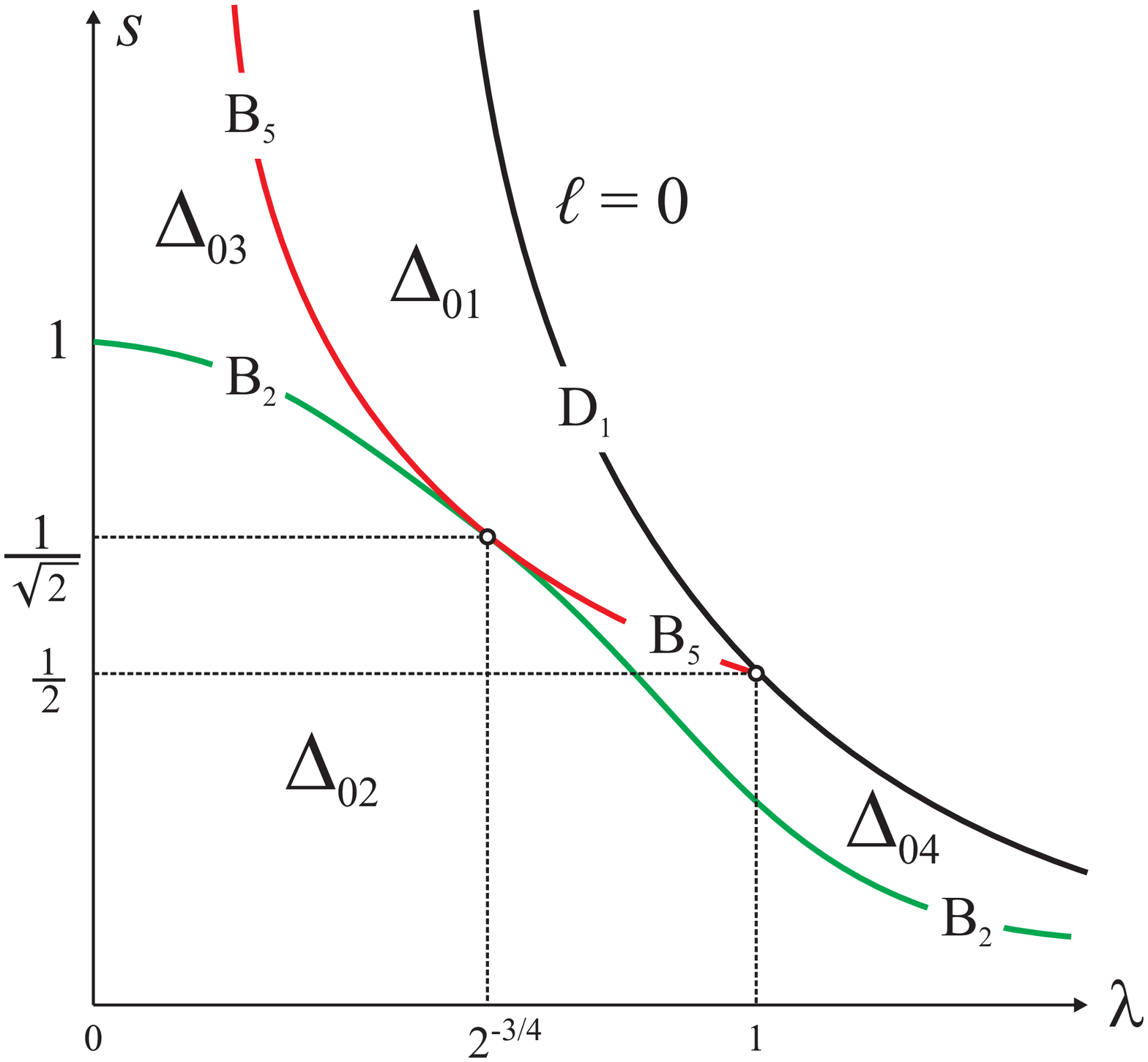}
\caption{Разделяющие кривые для точек $\gan$.}\label{fig_sepdelta0}
\end{figure}

Как видно из рис.~\ref{fig_sepdelta0}, точки $\Delta_{02}, \Delta_{03}$ имеют выход на классическую задачу Ковалевской ($\ld=0$), а точки $\Delta_{01}, \Delta_{04}$ продолжают существовать и в случае Рябова\,--\,Морозова ($\ell=0$). Поэтому на первый взгляд может показаться, что для описания соответствующих круговых молекул можно воспользоваться результатами работ \cite{BRF,Mor}. Однако это не совсем так. Топология всех уровней первых интегралов в проколотых окрестностях этих точек была установлена еще в работах \cite{KhPMM83, RyabRCD}. Но это не дает ответа на вопрос о том, как соединить семейства при обходе исследуемых точек. В зависимости от этого, многообразие в прообразе малой окружности с центром в точке может иметь разное количество компонент связности. Обозначим через $J_{0i}$ объединение тех связных компонент прообраза точки $\Delta_{0i}$, которые содержат критические точки отображения момента (то есть критические окружности, составленные из вырожденных точек ранга 1).

Заметим, что сегменты $\Delta_{0i}$ \textit{внутри} критической подсистемы $\mm$ бифуркационными не являются -- они отвечают лишь вырождению критических движений по отношению к полной системе. В частности, количество критических окружностей в прообразе $\Delta_{0i}$ такое же, как и в прообразе прилегающих с обеих сторон областей $\aaa_j$. По диаграммам подсистемы $\mm$ и данным из табл.~\ref{table2} устанавливаем, что $\Delta_{01}$ разделяет $\aaa_5$ и $\aaa_9$ (две критических окружности в $J_{01}$), $\Delta_{02}$ разделяет $\aaa_2$ и $\aaa_3$ (одна критическая окружность), $\Delta_{03}$ разделяет $\aaa_5$ и $\aaa_8$ (две критические окружности), $\Delta_{04}$ разделяет $\aaa_{9}$ и $\aaa_{11}$ (две критические окружности). Таким образом, поверхность $J_{02}$ связна. Для остальных же точек на уровне $J_{0i}$ имеется по две критические окружности, поэтому  круговые молекулы по имеющимся данным однозначно не восстанавливаются. В работах \cite{BRF,Mor} неявно используется гипотеза о том, что круговые молекулы этих точек имеют по две связных компоненты, но не представлено никаких мотивировок для этой гипотезы и никакого аппарата для соответствующих доказательств. Фактически использован некоторый принцип ``максимального правдоподобия'', согласно которому молекулы следует искать среди уже известных. Как мы сейчас докажем, ответ получился правильный.

\begin{theorem}\label{theo_twocomps0}
Поверхности $J_{01}, J_{03}, J_{04}$ состоят из двух связных компонент.
\end{theorem}

\begin{proof}
При $\ld \ne 0$ воспользуемся квадратурами, найденными в работе \cite{Gash1}. В обозначениях настоящей работы множество $\gan$ отвечает следующим значениям параметров
\begin{equation}\label{eq5_43}
  L_3=0, \qquad L_1 = \frac{s}{8} >0.
\end{equation}
Поэтому система сводится к двум уравнениям
\begin{equation}\label{eq5_44}
  {\dot \eta}^2 = \frac{L_1}{F^2(z)}, \qquad \dot z = \sqrt{f(z)},
\end{equation}
где
\begin{equation*}
  F(z)=z(z-2\ld)-\frac{1}{2}\left(s-\frac{1}{s}\right), \qquad f(z)=\left(\frac{s}{2}-z^2\right)F(z),
\end{equation*}
переменная $z$ вещественна, а вспомогательная переменная $\eta$, не равная бесконечности лишь на асимптотических движениях, введена равенством
\begin{equation*}
  \eta=\frac{y}{\sqrt{F(z)}},
\end{equation*}
где переменная $y$ -- вещественна. Из \eqref{eq5_43}, \eqref{eq5_44} следует, что ${\dot \eta}^2>0$, поэтому вещественна и переменная $\eta$, но тогда к условию существования осцилляции переменной $z$
\begin{equation}\notag
  f(z) \gs 0
\end{equation}
добавляется условие существования асимптотических движений
\begin{equation}\notag
  F(z) \gs 0.
\end{equation}
Полагая
\begin{equation*}
  \ell=0, \qquad s=\frac{1}{2\ld^2},
\end{equation*}
получим
\begin{equation}\notag
  F(z)=\left[ z-\left(\ld+\frac{1}{2\ld}\right)\right]\left[ z-\left(\ld-\frac{1}{2\ld}\right)\right], \qquad f(z)=\left(\frac{1}{4\ld^2}-z^2\right)F(z).
\end{equation}
Переменная $z$ осциллирует на отрезках
\begin{equation}\notag
\left[-\frac{1}{2\ld},\ld-\frac{1}{2\ld}\right], \quad \left[\frac{1}{2\ld},\ld+\frac{1}{2\ld}\right]
\end{equation}
при условии $\ld<1$  и на отрезках
\begin{equation}\notag
\left[-\frac{1}{2\ld},\frac{1}{2\ld}\right], \quad \left[\ld-\frac{1}{2\ld},\ld+\frac{1}{2\ld}\right]
\end{equation}
при условии $\ld>1$. В обоих случаях $F(z) \gs 0$ на первом отрезке и $F(z) \ls 0$ -- на втором. Следовательно, для первой критической окружности существуют асимптотические к ней движения, а для второй -- таких движений нет. Следовательно, вторая критическая окружность исчерпывает свою компоненту связности критической интегральной поверхности. Таким образом, для точек $\Delta_{01}$, $\Delta_{04}$ в прообразе имеется две связных компоненты с критическими окружностями.

Для точки $\Delta_{03}$ можно положить $\ld=0$. Тогда
\begin{equation}\notag
  F(z)= z^2-\frac{1}{2}\left(s-\frac{1}{s}\right), \qquad f(z)=\left(\frac{s}{2}-z^2\right)F(z), \qquad s > 1.
\end{equation}
Таким образом, $z$ осциллирует в симметричных промежутках
$$
\left[-\sqrt{\frac{s}{2}},-\sqrt{\frac{1}{2}(s-\frac{1}{s})}\right], \quad \left[\sqrt{\frac{1}{2}(s-\frac{1}{s})},\sqrt{\frac{s}{2}}\right],
$$
и для обеих критических окружностей имеются асимптотические движения. Достаточно ясно, что при стремлении к предельным критическим окружностям разные промежутки осцилляции $z$ не могут дать общий предел. Однако и это все же необходимо строго доказывать, исходя из приведенных в \cite{Gash1} квадратур, что требует определенных технических выкладок. В результате получим, что на связной компоненте интегральной поверхности переменная $z$ не может сменить промежуток осцилляции, поэтому таких компонент две (столько же, сколько и критических окружностей). Дадим и другое доказательство наличия двух связных компонент. Отметим, что, несмотря на сильную степень вырождения, выражения, полученные из уравнений первых интегралов, не дают какого-либо обозримого решения данного вопроса. Воспользуемся уравнениями, полученными С.В.\,Ковалевской, и результатами работы \cite{KhND10}. Напомним, что в переменных разделения Ковалевской $s_1,s_2$ уравнения движения имеют вид
\begin{equation}\notag
(s_2  - s_1 )\frac{ds_1 } {dt} = \ri \sqrt {2S(s_1 )} ,\quad (s_2  -
s_1 ) \frac{ds_2 } {dt} =  - \ri \sqrt {2S(s_2 )} ,
\end{equation}
где
\begin{equation*}
\begin{array}{l}
S(x) = (x - h + \sqrt k )(x - h - \sqrt k )\varphi (x),\qquad
\varphi (x) = x(x - h)^2  + (1 - k)x - 2\ell^2
\end{array}
\end{equation*}
и $s_1  \geqslant s_2$. С.В.\,Ковалевская также указала выражения всех фазовых переменных через переменные разделения в виде однозначных функций от $s_1,s_2$ и набора алгебраических радикалов
\begin{equation}\label{eq5_45}
R_{i\gamma}=\sqrt{s_i - e_\gamma} \qquad (i=1,2, \quad \gamma=1,\ldots,5),
\end{equation}
где $e_\gamma$ -- корни многочлена $S(x)$. Принято обозначать через $e_1,e_2,e_3$ корни $\varphi(x)$ в порядке возрастания, если все они вещественны, и полагать $e_4=h-\sqrt{k}, e_5=h+\sqrt{k}$.
Известно \cite{Appel,Ipat}, что в случае трех вещественных корней у $\varphi(x)$ переменные разделения изменяются в областях
\begin{equation}\notag
s_1  \in [\max \{ e_1 ,e_2 ,e_3 ,e_4 \} ,e_5 ], \qquad s_2  \in [ - \infty ,\min \{ e_1 ,e_2 ,e_3 ,e_4 ,e_5 \} ],
\end{equation}
и, как показано в \cite{KhND10} методом булевых функций, выражения для фазовых переменных в общем случае сводятся к однозначным зависимостям от следующих произведений алгебраических радикалов, приведенных к вещественной форме:
\begin{equation}\label{eq5_46}
\begin{array}{l}
\ds  \sqrt{s_1-e_1} \sqrt{\frac{e_2-s_2}{e_5-s_2}}\sqrt{\frac{e_3-s_2}{e_5-s_2}}\sqrt{\frac{e_4-s_2}{e_5-s_2}}, \\[3mm]
\ds  \sqrt{s_1-e_2} \sqrt{\frac{e_1-s_2}{e_5-s_2}}\sqrt{\frac{e_3-s_2}{e_5-s_2}}\sqrt{\frac{e_4-s_2}{e_5-s_2}}, \\[3mm]
\ds  \sqrt{s_1-e_3} \sqrt{\frac{e_2-s_2}{e_5-s_2}}\sqrt{\frac{e_2-s_2}{e_5-s_2}}\sqrt{\frac{e_4-s_2}{e_5-s_2}}
\end{array}
\end{equation}
(радикалы, содержащие отношения разностей, не меняют знак при $s_2=-\infty$). В случае, когда переменная разделения отражается от границы своей осцилляции, соответствующий радикал следует считать меняющим знак (радикал, отвечающий значению $s_2=-\infty$, уже исключен процедурой редукции к выражениям \eqref{eq5_46}). Если какое-либо из граничных значений является кратным корнем, то соответствующий радикал тоже считается меняющим знак, т.к. топологически асимптотическая траектория вместе с предельными точками есть связное множество.

Обратимся к рассматриваемой точке $\Delta_{03}$ при $\ld=0$. Имеем
$$
\ds e_1=\frac{1}{s} < e_2=e_3=e_4 = s < e_5=s+\frac{1}{s}
$$
(напомним, что $s>1$). Тогда, обозначая для наглядности трехкратный корень через $e_*$, имеем
\begin{equation}\notag
  s_1 \in [e_*, e_5], \qquad s_2 \in [-\infty, e_1].
\end{equation}
В выражениях \eqref{eq5_46} второе и третье совпадают, произведение пары одинаковых радикалов становится однозначной функцией и на знаки алгебраических выражений не влияет. Остаются следующие алгебраические выражения
\begin{equation}\notag
\begin{array}{l}
\ds  \sqrt{s_1-e_1} \sqrt{\frac{e_*-s_2}{e_5-s_2}}, \qquad \ds  \sqrt{s_1-e_*} \sqrt{\frac{e_1-s_2}{e_5-s_2}},
\end{array}
\end{equation}
в которых второе меняет знак на связной компоненте интегральной поверхности и, следовательно, на количество связных компонент не влияет, а первое вдоль замыкания любой траектории сохраняет фиксированный знак, выбранный в начальный момент. Вследствие этого, поверхность $J_{03}$ имеет две связные компоненты. Теорема доказана.
\end{proof}


\subsubsection{Множество $\gaa$}

Следующий класс вырожденных точек ранга 1 -- это точки, лежащие в прообразе ребра возврата поверхности $\wsa$. Для интегральных параметров на этом множестве, обозначаемом через $\gaa$, с учетом условий существования движений имеем
\begin{equation}\label{eq5_47}
  \ell =\displaystyle{\pm \frac{1}{\sqrt{2}}s^{3/2}}, \quad \displaystyle{h=\frac{3}{2}s+\frac{\ld^2}{2}}, \quad  \displaystyle{s\in [0,s_*]},
\end{equation}
где $s_*=s_*(\ld)$ -- наибольший (вещественный) корень уравнения
\begin{equation}\label{eq5_48}
  9s^4+2\ld^2 s^3-24 s^2-24 \ld^2 s+4(4-\ld^4)=0,
\end{equation}
существующий и положительный для всех $\ld \gs 0$. Множество $\gaa$ делится на качественно различные части его пересечением с образом множества критических точек ранга 0, то есть с подмножествами $\delta_i$ пространства интегральных констант, а также пересечениями с исследованным выше множеством $\gan$. Очевидно, что кривая $\delta_1$ в этих пересечениях не участвует. Непосредственно проверяется, что все точки пересечения  $\gaa \cap \gan$ содержатся в $\gaa \cap \delta_2$. Условие пересечения $\gaa \cap (\delta_2\cup \delta_3)$ представляется либо уравнением \eqref{eq5_48}, либо уравнением
\begin{equation}\label{eq5_49}
  s^2+2\ld^2 s-1=0.
\end{equation}
На плоскости $(\ld,s)$ получаем разделяющие кривые
\begin{eqnarray}
& & \psi_1: \ds \ld =\left| \frac{\sqrt{s^2+4}-2s}{(\sqrt{s^2+4}-s)^{1/2}}\right|, \qquad s>0; \label{eq5_50}\\
& & \psi_2: \ds \ld =\sqrt{\frac{1}{2}(\frac{1}{s}-s)}, \quad s \in (0,1] \quad \Leftrightarrow \quad s=\sqrt{1+\ld^4}-\ld^2, \quad \ld \gs 0.\label{eq5_51}
\end{eqnarray}
С учетом \eqref{eq5_47} образ кривой $\psi_1$ на плоскости $(\ell,\ld)$ -- это объединение кривых $\vpi_{23}$ и $\vpi_{31}$, которые вместе могут быть представлены графиком однозначной функции
\begin{equation}\notag
  \ds \ld =\left| \frac{\sqrt{s^2+4}-2s}{(\sqrt{s^2+4}-s)^{1/2}}\right|, \qquad s = (\sqrt{2}\ell)^{2/3}, \qquad \ell \gs 0,
\end{equation}
а образ кривой $\psi_2$ -- это кривая $\vpi_{22}$ в представлении \eqref{eq4_48}.

Верхняя граница $s_*$ значений $s$ -- это ветвь кривой \eqref{eq5_50} при $s \gs 2/\sqrt{3}$, полученная в результате пересечения $\gaa$ с $\delta_3$ (точка $C$ на ключевых множествах). Другая ветвь кривой \eqref{eq5_50} для значений $s \in (0,2/\sqrt{3}]$ и кривая \eqref{eq5_51} -- это пересечения $\gaa$ с $\delta_2$ (соответственно точки $B_4$ и $B_2$ на ключевых множествах, причем $B_4$ сливается с $C$ при $\ld=0$). Разделяющее множество и обозначения соответствующих точек на диаграммах подсистемы $\mm$ и на бифуркационных диаграммах приведены на рис.~\ref{fig_sepdelta1}. Точки $\Delta_{11}, \Delta_{13}$ имеют выход на частный случай $\ld=0$, точка $\Delta_{14}$ существует в частном случае $\ell=0$, точка $\Delta_{12}$ таких аналогов не имеет.

\begin{figure}[htp]
\centering
\includegraphics[width=0.3\textwidth, keepaspectratio]{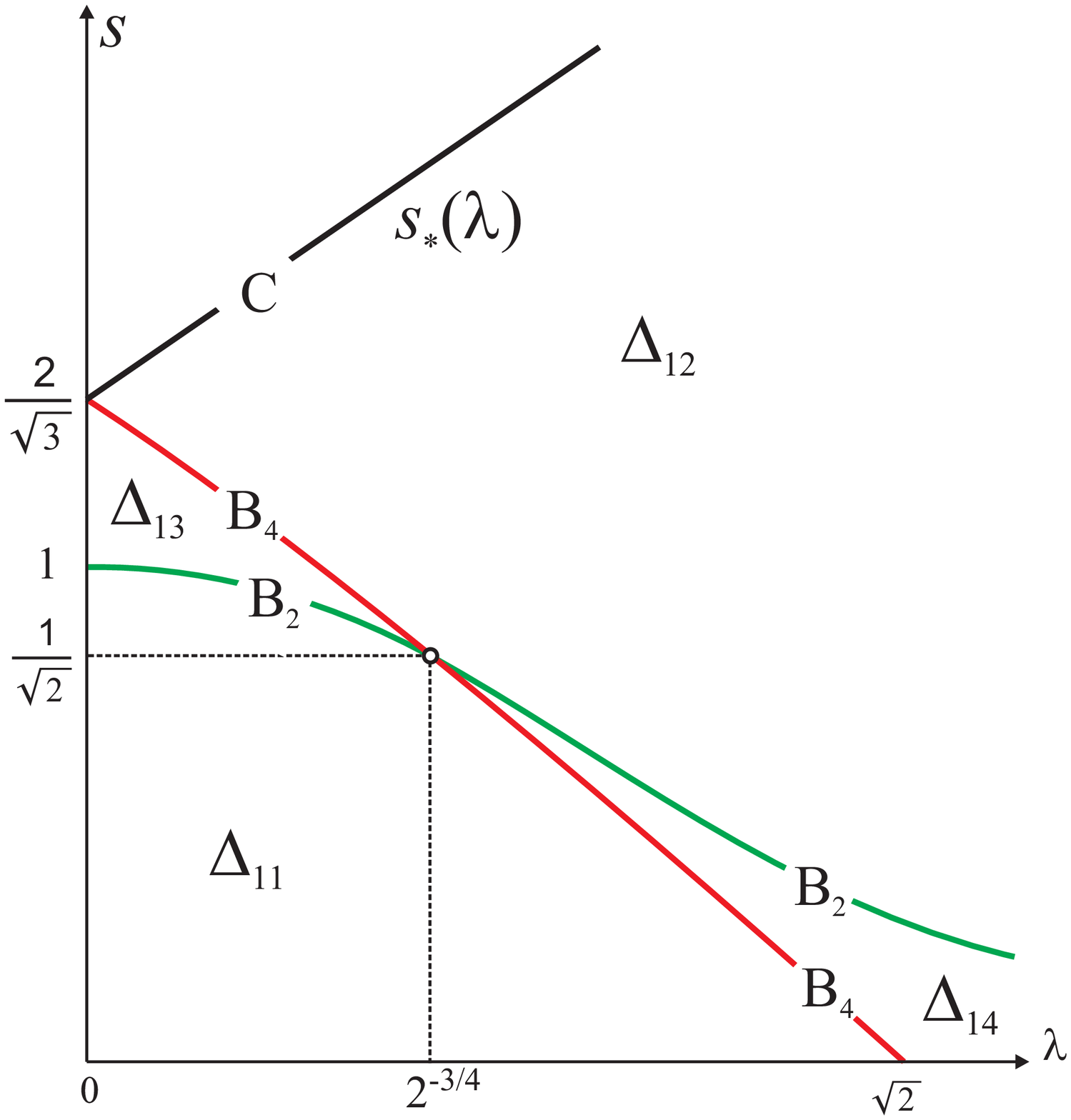}
\caption{Разделяющие кривые для точек $\gaa$.}\label{fig_sepdelta1}
\end{figure}

Как и ранее, количество критических окружностей в прообразе $\Delta_{1i}$ устанавливается по диаграммам подсистемы $\mm$ и данным для примыкающих областей из табл.~\ref{table2}.
Видим, что $\Delta_{11}$ разделяет $\aaa_2$ и $\aaa_6$ (одна критическая окружность в прообразе), $\Delta_{12}$ разделяет $\aaa_3$ и $\aaa_4$ (одна критическая окружность), $\Delta_{13}$ разделяет $\aaa_7$ и $\aaa_8$ (две критические окружности), $\Delta_{14}$ разделяет $\aaa_{10}$ и $\aaa_{11}$ (две критические окружности). В случае, когда в прообразе точки $\Delta_{1i}$ имеется лишь одна вырожденная окружность, то есть соответствующая критическая интегральная поверхность и круговая молекула связны, для точного описания топологии окрестности точек ниже устанавливаются номера семейств, участвующих в бифуркациях. Это легко следует из результатов о прилегающих камерах и о номерах семейств в этих камерах. Если же окружностей две, то для обоснования структуры круговых молекул нужно доказательно установить число компонент связности критической поверхности-прообраза $\Delta_{1i}$. По умолчанию, в \cite{BRF} предполагалось, что при $\ld=0$ круговая молекула точки $\Delta_{13}$ состоит из двух связных компонент, и также в \cite{Mor} считалось, что из двух компонент состоит круговая молекула точки $\Delta_{14}$ при $\ell=0$. Докажем общее утверждение, воспользовавшись квадратурами И.Н.\,Гашененко. Аналогично предыдущему случаю, обозначим через $J_{1i}$ объединение тех связных компонент прообраза точки $\Delta_{1i}$, которые содержат критические точки отображения момента.

\begin{theorem}\label{theDelta134}
Поверхности $J_{13},J_{14}$ состоят из двух компонент связности.
\end{theorem}
\begin{proof} Естественно, по определению классов точек считаем, что $(\ld,s)$ принадлежит соответствующей открытой подобласти на рис.~\ref{fig_sepdelta1}. В терминах $(\ld,s)$ многочлен $f(z)$ из \eqref{eq2_32} примет вид
\begin{equation*}
f(z)=-z^2(z-\ld)^2+\frac{1}{2}s z(z-2\ld)-\frac{1}{4}(s^2-1).
\end{equation*}
Конечно, его дискриминант есть произведение многочленов в левых частях \eqref{eq5_48}, \eqref{eq5_49}. Очевидно, $f(z)$ имеет четыре вещественных корня в областях  $\Delta_{13},\Delta_{14}$ (достаточно рассмотреть выходы на оси координат, когда соответствующее уравнение явно решается). Как отмечалось выше, на связной компоненте интегрального многообразия переменная $z$ не может сменить промежуток осцилляции, поэтому таких компонент в этих случаях две.
\end{proof}

В целом, согласно \eqref{eq5_8}, точки $\Delta_1$ отвечают случаю $L_2=0$ в решении Гашененко, тогда из \eqref{eq2_33}, \eqref{eq2_34} получаем
\begin{equation*}
L_1 >0, \qquad F(z)=(z-\ld)^2 \gs 0,
\end{equation*}
поэтому для обеих критических окружностей существуют асимптотические движения.

\subsubsection{Множество $\gac$}

Последний класс вырожденных точек ранга 1 -- это точки, лежащие в прообразе ребра возврата поверхности $\wsc$. Для интегральных параметров на этом множестве, обозначаемом через $\gac$, с учетом условий существования движений имеем
\begin{equation}\notag
    \begin{array}{l}
      \displaystyle{h=h^*+2\ell^2,}
      \quad \left\{
      \begin{array}{ll}  \ell \in \bR, & \ld \leqslant \ld^*\\
                      |\ell| \geqslant  \ell^*, & \ld > \ld^*
      \end{array}
      \right. , \qquad s=\displaystyle{\frac{1}{2\ld^{2/3}}}, \quad k=-4\ld^2\ell^2+k^*,
    \end{array}
\end{equation}
где
\begin{equation}\notag
\begin{array}{l}
  \ds \ld^*=\frac{2\sqrt{2}}{3^{3/4}}, \quad \ell^*=\displaystyle{\frac{2\ld^{2/3}-\sqrt{4+\ld^{4/3}}}{\sqrt{2}(\sqrt{4+\ld^{4/3}}-\ld^{2/3} )^{1/2}}}>0  \qquad (\ld > \ld^*),\\
\ds  h^*=\frac{1}{2}\ld^{2/3}\left(3-\ld^{4/3}\right), \quad \quad k^*=(\ld^{4/3}-1)^3+1.
\end{array}
\end{equation}

Множество $\gac$ делится на качественно различные части его пересечением с образом множества критических точек ранга 0, то есть с подмножествами $\delta_i$ пространства интегральных констант, а также пересечениями с исследованными выше множествами $\gan, \gaa$. Очевидно, что кривые $\delta_1, \delta_3$ в этих пересечениях не участвуют, как не имеющие общих точек с $\wsc$. Непосредственно проверяется, что все точки пересечений  $\gac \cap \gan$ и $\gac \cap \gaa$ содержатся в $\gac \cap \delta_2$. Поскольку на $\gac$ при фиксированном $\ld$ постоянно и $s$, то в качестве параметров, определяющих точку $\gac$ удобно взять $\ld$ и $\ell$. Тогда условия пресечения $\gac \cap \delta_2$ (точка $B_3$ ключевых множеств для всех $\ld$ и точка $B_5$ для $\ld\in [0,1]$) сводятся к уравнениям двух разделяющих кривых $\theta_1, \theta_2$ для определения областей классификации диаграмм $\mSell$:
\begin{equation*}
\begin{array}{l}
\theta_1:   \ell= \ds{\frac{1}{2\ld^{1/3}}\sqrt{1-\ld^{4/3}}}, \quad 0 < \ld \ls 1, \\
\theta_2:  \ell=
\ds{ \frac{|2\ld^{2/3}-\sqrt{4 + \ld^{4/3}}|}{\sqrt{2}(\sqrt{4 + \ld^{4/3}}-\ld^{2/3})^{1/2}}},\quad \ld \gs 0.
\end{array}
\end{equation*}
Ясно, что на $\theta_2$ $\ell(\ld^*)=0$, а при $\ld > \ld^*$
\begin{equation*}
  \ell(\ld)=\ds{ \frac{2\ld^{2/3}-\sqrt{4 + \ld^{4/3}}}{\sqrt{2}(\sqrt{4 + \ld^{4/3}}-\ld^{2/3})^{1/2}}}=\ell^*,
\end{equation*}
поэтому ветвь кривой $\theta_2$ для $\ld \gs \ld^*$ является границей области существования точек множества $\gac$. Разделяющие кривые и обозначения точек $\Delta_{3i}$ в соответствии с рис.~\ref{fig_M3_LKeyABC_red}, \ref{fig_M3_LKeyDEF_red} показаны на рис.~\ref{fig_sepdelta3}.
Сопоставляя с \eqref{eq4_48}, видим, что кривая $\theta_1$ -- это $\vpi_{21}$, а кривая $\theta_2$ -- это $\vpi_{24}$.

\begin{figure}[htp]
\centering
\includegraphics[width=0.3\textwidth, keepaspectratio]{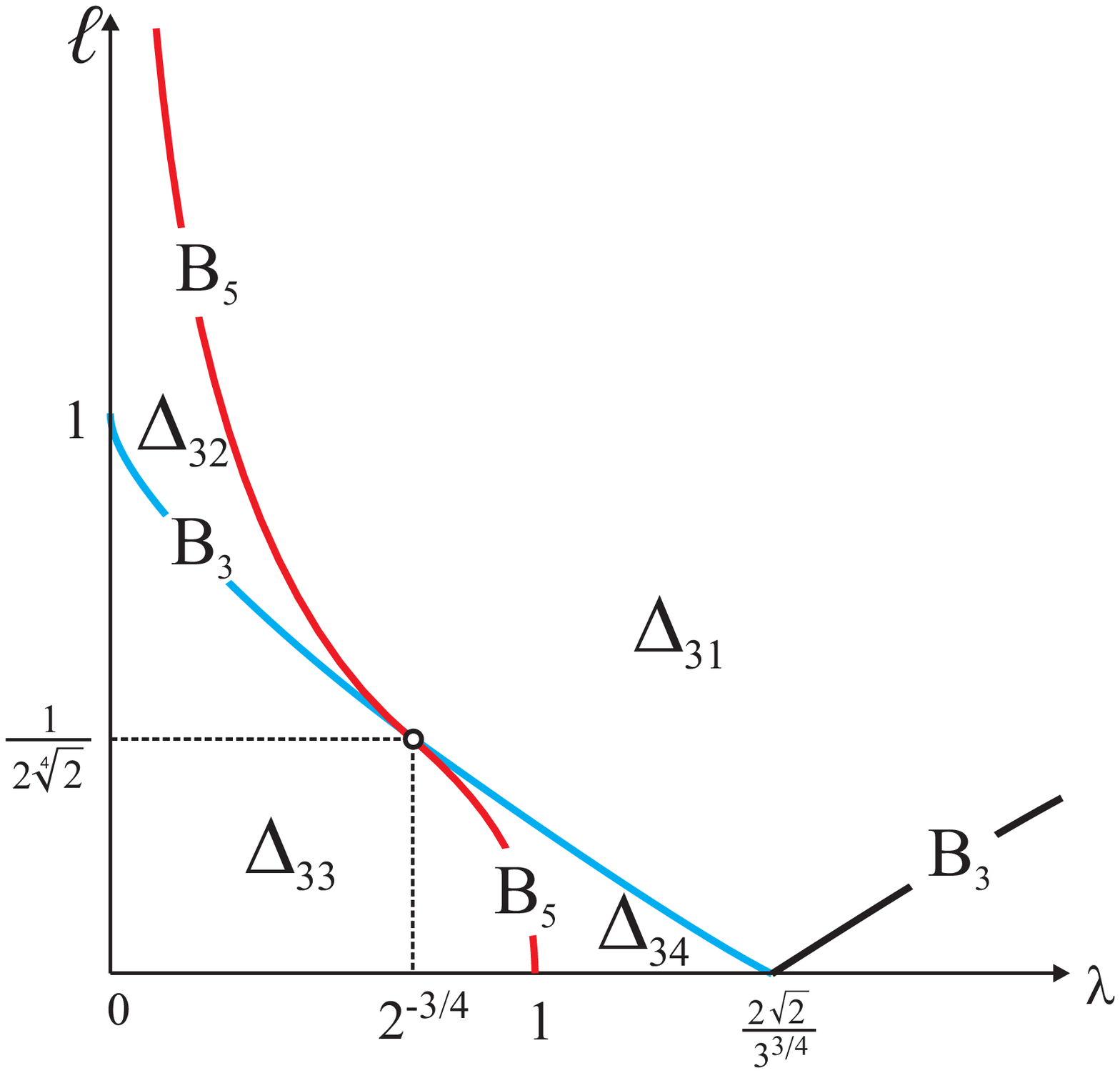}
\caption{Разделяющие кривые для точек $\gac$.}\label{fig_sepdelta3}
\end{figure}

Количество критических окружностей в прообразе $\Delta_{3i}$ устанавливается по диаграммам подсистемы $\mo$ и данным для примыкающих областей из табл.~\ref{table5}.
Видим, что $\Delta_{31}$ разделяет $\ccc_1$ и $\ccc_2$ (одна критическая окружность в прообразе), $\Delta_{32}$ разделяет $\ccc_8$ и $\ccc_9$ (две критических окружности), $\Delta_{33}$ разделяет $\ccc_6$ и $\ccc_7$ (одна критическая окружность), $\Delta_{34}$ разделяет $\ccc_{3}$ и $\ccc_{5}$ (две критические окружности).

Как и для множеств $\gan,\gaa$, если в прообразе точки $\Delta_{3i}$ имеется лишь одна вырожденная окружность, то соответствующая критическая интегральная поверхность $J_{3i}$ связна. Приведем легко доказываемое утверждение, важное для дальнейшего.

\begin{theorem}\label{theDelta32}
Критическое многообразие $J_{32}$ состоит из двух связных компонент.
\end{theorem}
\begin{proof}
Как видно из рис.~\ref{fig_sepdelta3}, точки класса $\Delta_{32}$ имеют выход на классическую задачу Ковалевской $\ld=0$. Предельный переход в образе получим из \eqref{eq5_28}, полагая  $s=1/(2\ld^{2/3}) \to \infty$. В результате $h=2\ell^2$, $k=0$, и $\Delta_{32}$ отвечают значениям $\ell>1.$ Предельное критическое многообразие известно (см. \cite{KhBook88}), а именно, $J_{32}(0)=2S^1$. Поэтому и при малых $\ld$ это многообразие несвязно, а так как критических окружностей в прообразе две, то и компонент связности две.
\end{proof}

Отметим, что точки класса $\Delta_{34}$ имеют выход на случай $\ell=0$, для которого в работах \cite{Mor,Mo2008} были изучены круговые молекулы, но при этом наличие двух компонент связности в молекулах для этих точек принималось без доказательства. Ниже мы это \textit{докажем}, откуда будет следовать, что и критическое многообразие $J_{34}$ состоит из двух связных компонент.

\clearpage

\section{Топология приведенных систем}\label{sec6}
\subsection{Разделяющее множество и бифуркационные диаграммы}
В соответствии с предложением \ref{propos13} множество $\Theta_L$ в плоскости параметров $(\ld,\ell)$, классифицирующее бифуркационные диаграммы $\mSell$ приведенных систем на $\mPel$,
состоит из пар $(\ld,\ell)$, где $\ell$ --- критическое значение ограничения функции $L$ на одно из ключевых множеств $\mK_i$ критических подсистем $\mi_i$ ($i=1,2,3$) при заданном $\ld$. Очевидно, это значения $\ell$ в точках \eqref{eq5_11}, \eqref{eq5_16} -- \eqref{eq5_21}, \eqref{eq5_23} -- \eqref{eq5_27}. Как отмечалось ранее, при $\ld \ell \ne 0$ из них только последняя не соответствует вырожденным критическим точкам ранга 0. Напомним, что в силу имеющихся очевидных симметрий по параметрам $\ld,\ell$, это множество $\Theta_L$ рассматривается в первом квадранте
\begin{equation}\notag
  \ld \gs 0, \qquad \ell \gs 0,
\end{equation}
при этом полуоси $\ld=0$ и $\ell=0$ включаются в разделяющее множество по умолчанию.
Действительно, точки $A,B_1,D_1,D_2$ дают значения $\ell=0,\ld \in \bR$, а при $\ld=0$ и любом $\ell$ имеются кратные критические точки ранга 1, отвечающие, например, слиянию подобластей $c_2,c_6,c_8$ с $b_1$ (в подсистеме $\mo$ надо рассмотреть возможность $s=\infty$). Кратные участки, очевидно, возникают и в подсистеме $\mm$. Полный анализ соответствующих диаграмм для $\ld=0$ имеется в \cite{KhBook88}.

\begin{theorem}[П.Е.\,Рябов]\label{theClassEll}
Разделяющее множество $\Theta_L$ при классификации бифуркационных диаграмм $\mSell$ состоит из  точек координатных осей $\{\ld=0\}\cup\{\ell=0\}$ и пяти кривых, заданных явными однозначными функциями:
\begin{equation*}
\begin{array}{ll}
\gmm_1: & \ds \ld=\ld_1(\ell)=\frac{|\sqrt{s^2+4}-2s|}{(\sqrt{s^2+4}-s)^{1/2}}, \quad s=(2\ell^2)^{1/3}, \quad \ell \gs 0;\\[3mm]
\gmm_2: & \ell=\ell_2(\ld)=
\ds{ \frac{|2\ld^{2/3}-\sqrt{4 + \ld^{4/3}}|}{\sqrt{2}(\sqrt{4 + \ld^{4/3}}-\ld^{2/3})^{1/2}}},\quad \ld \gs 0;\\[3mm]
\gmm_3:  & \ell=\ell_3(\ld)= \ds{\frac{1}{\sqrt{2}}(\sqrt{1+\ld^4}-\ld^2)^{3/2}}, \quad  \ld \gs 0;\\[3mm]
\gmm_4:  & \ell=\ell_4(\ld)= \ds{\frac{1}{2\ld^{1/3}}\sqrt{1-\ld^{4/3}}}, \quad 0 < \ld \ls 1;\\[3mm]
\gmm_5:  & \ell=\ell_5(\ld)= \ds{\frac{1}{4\ld}}, \quad \ld \gs 0.
\end{array}
\end{equation*}
\end{theorem}

Подчеркнем снова, что, кроме новой кривой $\gmm_5$, все остальные отвечают кривым \eqref{eq4_48}:
$\gmm_1 = \vpi_{23}\cup\vpi_{31}$, $\gmm_2=\vpi_{24}$, $\gmm_3=\vpi_{22}$, $\gmm_4=\vpi_{21}$.
Первый квадрант разбивается на 18 областей, в которых диаграммы $\mSell$, рассматриваемые как стратифицированное одномерное многообразие, одинаковы. Разделяющее множество и нумерация возникающих областей представлены на рис.~\ref{fig_sepplotLregnum}. Всюду ниже слово ``область'' применяется именно к связной компоненте дополнения к разделяющему множеству на плоскости параметров, а связные компоненты дополнений к диаграммам называются \textit{камерами}.

\begin{figure}[ht]
\centering
\includegraphics[width=0.8\textwidth, keepaspectratio]{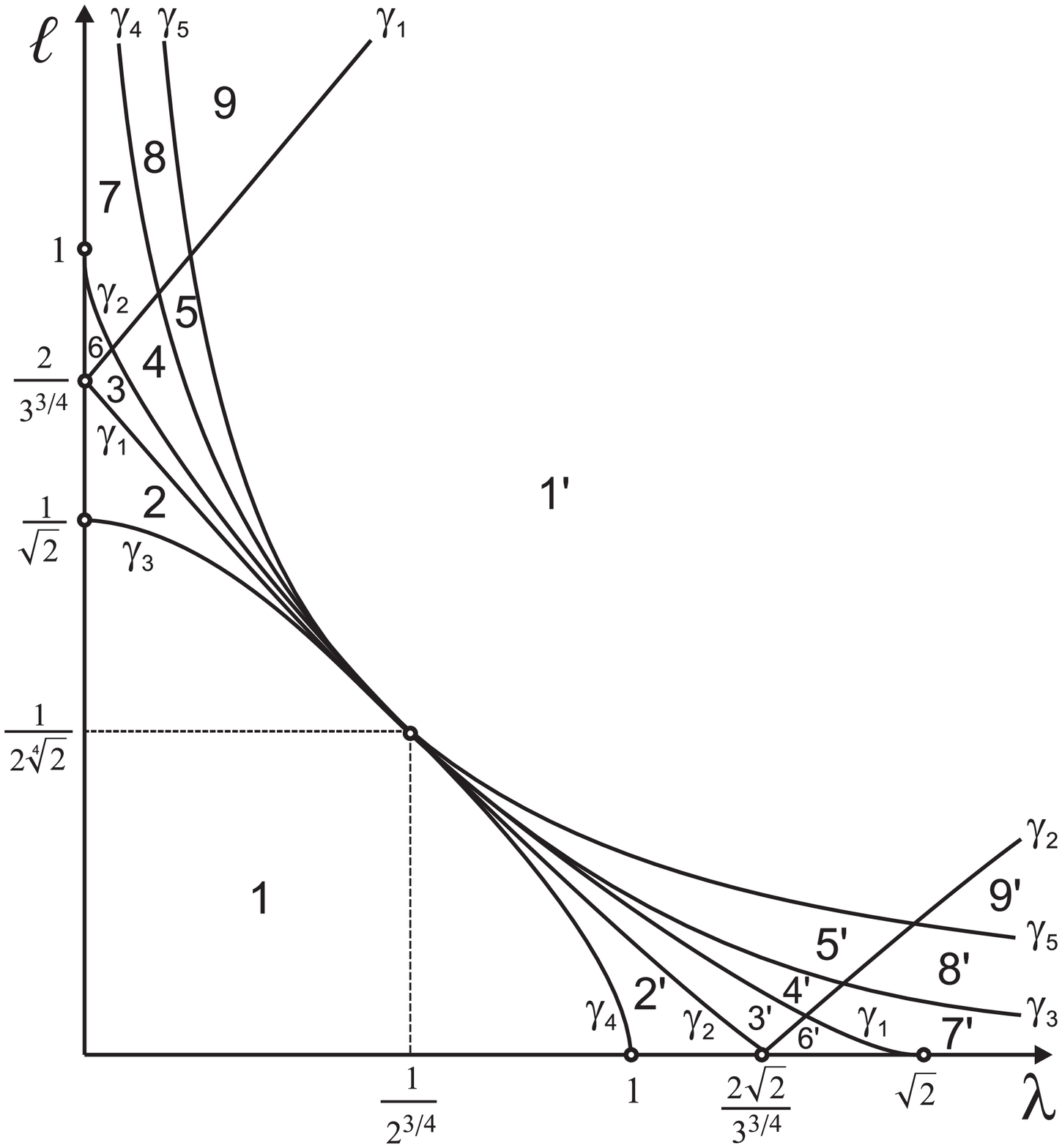}
\caption{Разделяющее множество и области параметров.}\label{fig_sepplotLregnum}
\end{figure}

На рисунках \ref{fig_reg01both}\,--\,\ref{fig_reg07sboth} диаграмма $\mSell$ для каждой из 18 областей плоскости $(\ld,\ell)$ построена в двух представлениях. На первом нанесены обозначения гладких участков и особых точек в соответствии с диаграммами критических подсистем. На втором -- на гладких участках диаграммы указаны атомы бифуркаций, происходящих в системе с двумя степенями свободы на $\mPel$. Стрелки на несимметричных атомах указывают в сторону возрастания количества связных компонент регулярных интегральных многообразий.

Объединяя диаграммы $\mSell$ в расширенном пространстве $\mwide{\bbI}=\bR^4(\ld,\ell,k,h)$ (см. \eqref{eq4_19}), проследим эволюцию камер -- открытых областей, на которые $\mSell$ делит плоскость $(k,h)$. Для этого введем расширенную диаграмму
\begin{equation}\notag
  \mwide{\Sigma}= \bigcup_{(\ld,\ell)} \{(\ld,\ell,k,h): (k,h) \in \mSell \}
\end{equation}
и назовем расширенной камерой открытую связную компоненту дополнения $\mwide{\Sigma}$ в $\bR^4$. Расширенные камеры, отличающиеся лишь знаками $\ld,\ell$, не различаем. Интегральные многообразия, отвечающие двум точкам одной расширенной камеры, естественным образом диффеоморфны и могут быть получены одно из другого гладкой изотопией, возможно, дополненной симметрией, меняющей знаки $\ld,\ell$. В дальнейшем оговорок о знаках $\ld,\ell$ делать более не будем.
Из сказанного вытекает, что две камеры при заданных $\ld,\ell$ следует считать одинаковыми, если они получены сечением одной и той же расширенной камеры. В силу этого легко проследить по рисункам, что различных камер имеется всего восемь, причем одна из них, содержащая сколь угодно большие по модулю отрицательные значения $h$, недопустима, то есть в прообразе ее точек интегральные многообразия пусты. Таким образом, интересующих нас допустимых камер всего семь и они занумерованы на рисунках римскими цифрами $\ts{I} - \ts{VII}$.

Суммируя представленную выше информацию, строим 18 устойчивых по параметрам бифуркационных диаграмм $\mSell$ систем на $\mPel$. При их построении принципиальной становится проблема визуализации малых областей. В связи с этим диаграммы $\mSell$ на рисунках сильно искажены, причем не удается соблюсти сохранение структуры уровней $h=\cons$ с тем, чтобы впоследствии легко увидеть грубый изоэнергетический инвариант Фоменко. Для этого более удобны изоэнергетические диаграммы \cite{mtt40,SaKhSh2012}, также классифицированные ниже.

\def\uuu{0.9}

\begin{figure}[!ht]
\centering
\includegraphics[width=\uuu\textwidth, keepaspectratio]{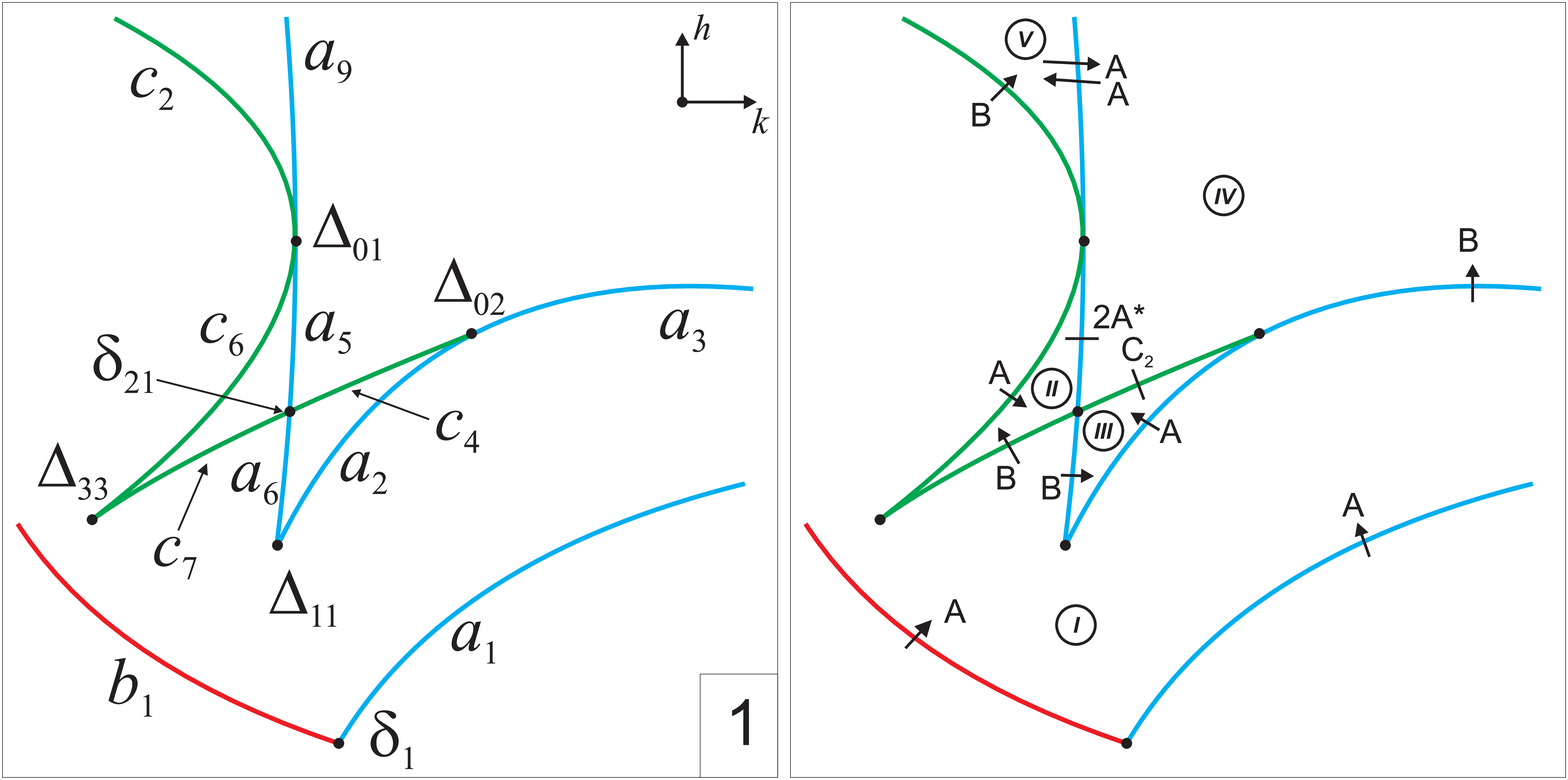}
\includegraphics[width=\uuu\textwidth, keepaspectratio]{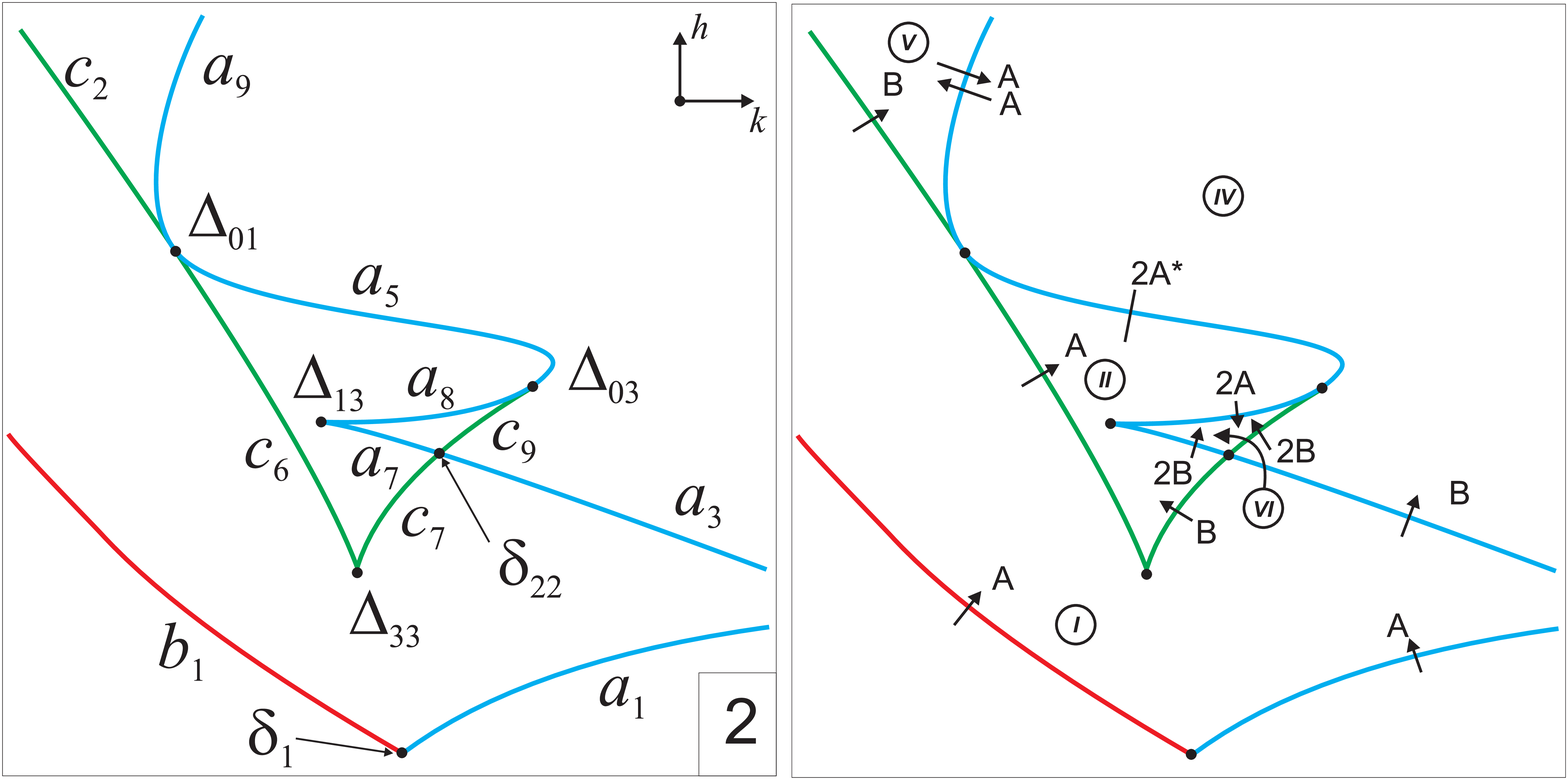}
\includegraphics[width=\uuu\textwidth, keepaspectratio]{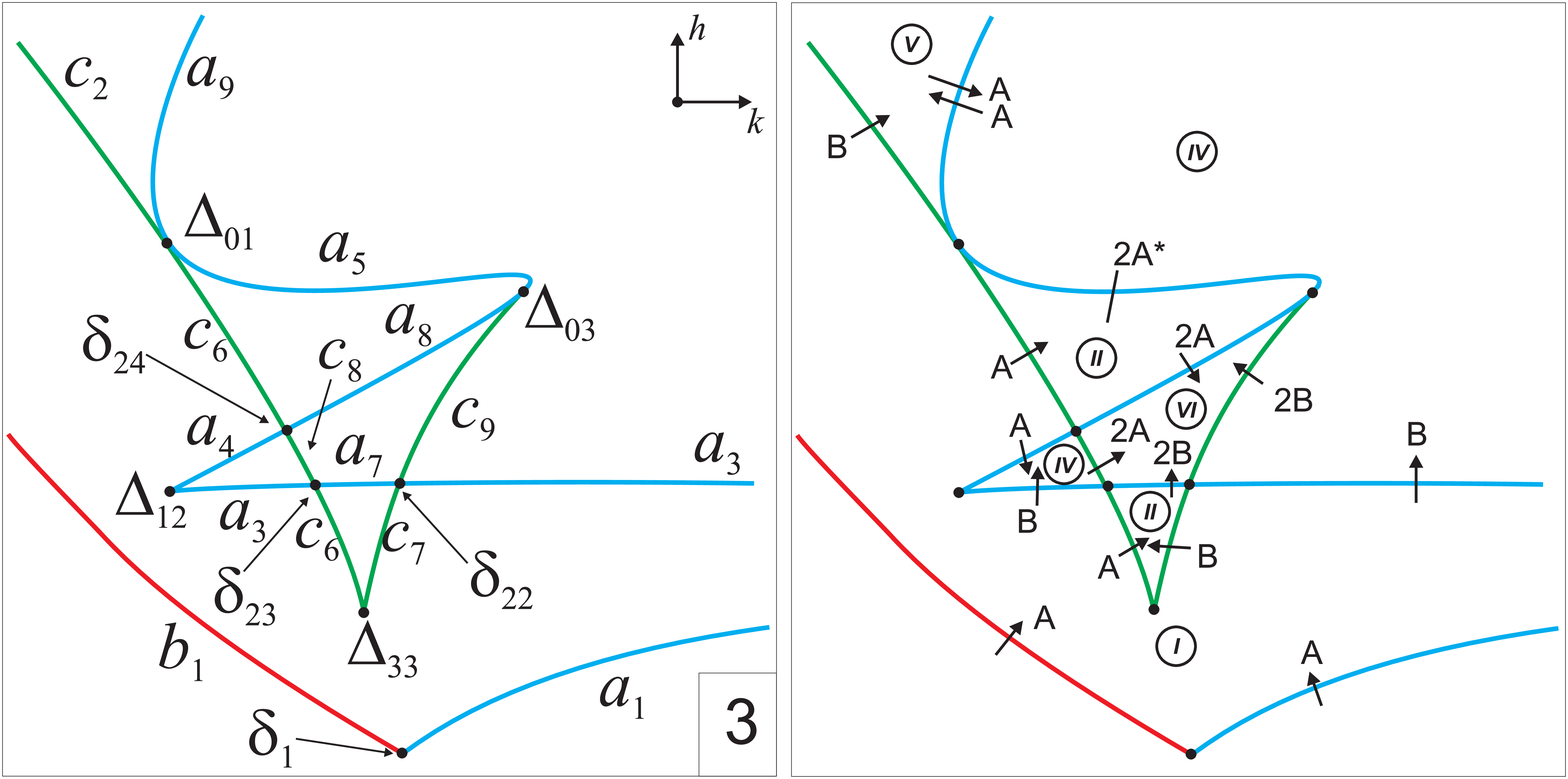}
\caption{Диаграммы $\mSell$ для областей 1 -- 3.}\label{fig_reg01both}
\end{figure}

\begin{figure}[!ht]
\centering
\includegraphics[width=\uuu\textwidth, keepaspectratio]{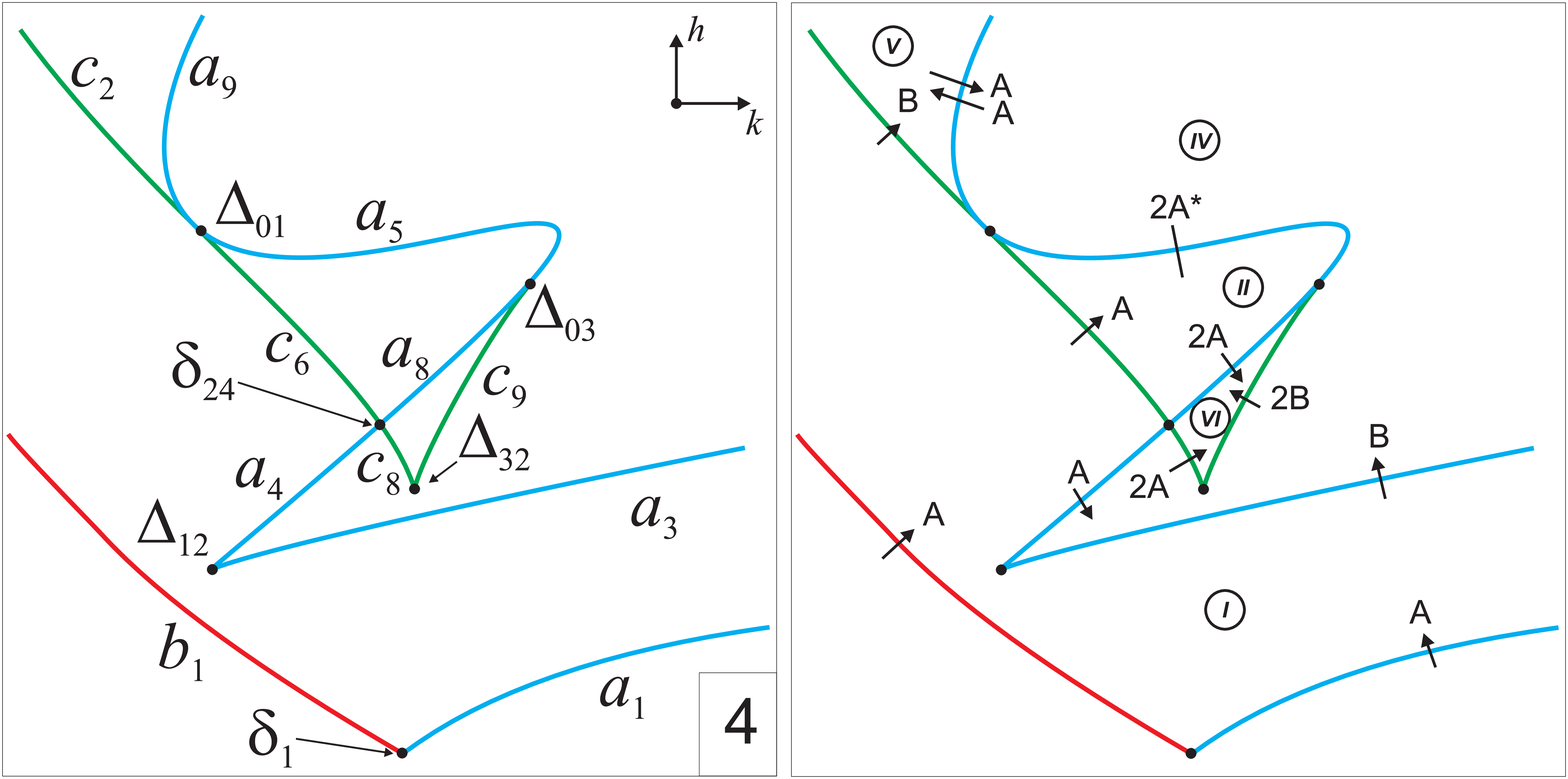}
\includegraphics[width=\uuu\textwidth, keepaspectratio]{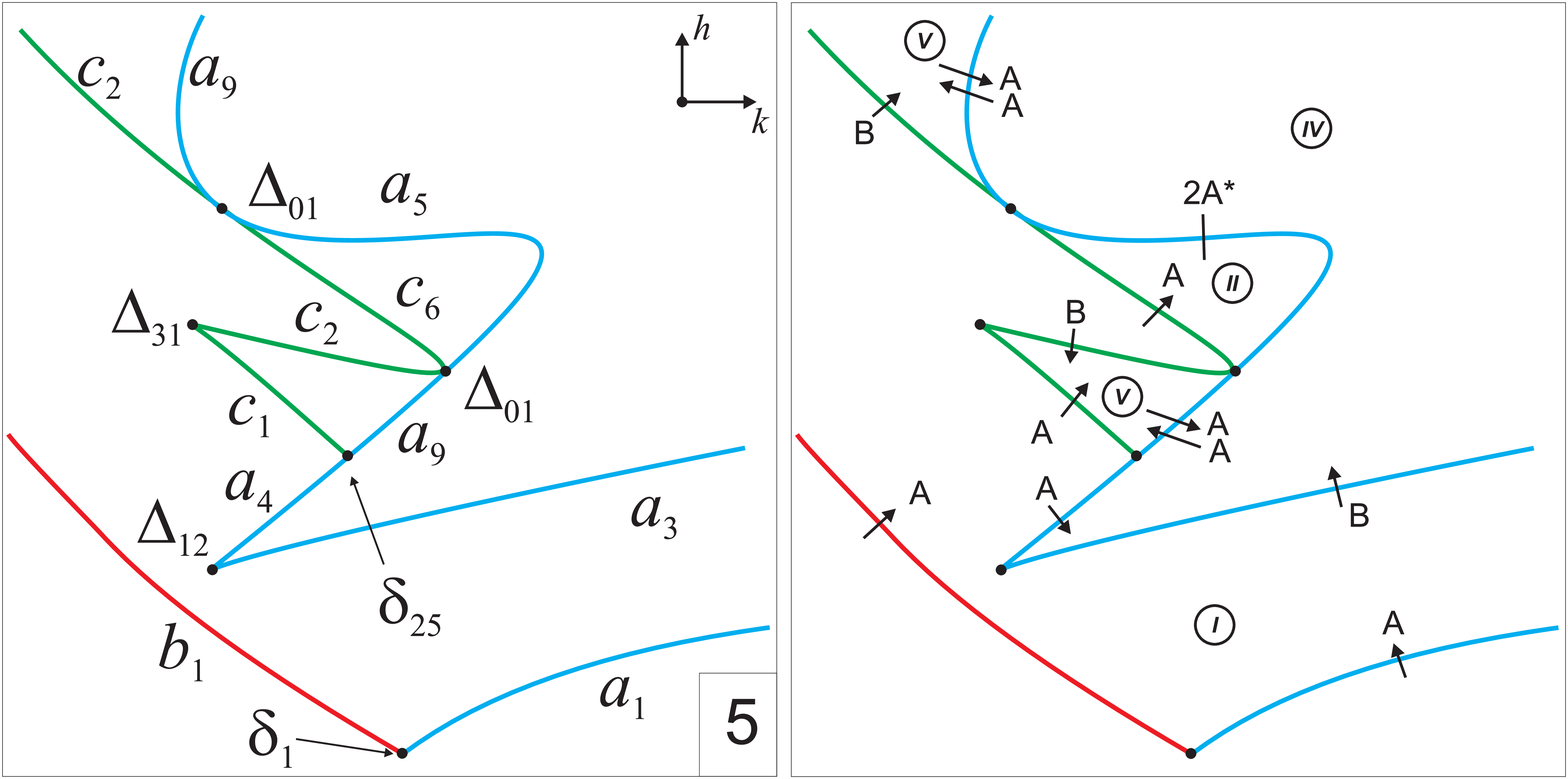}
\includegraphics[width=\uuu\textwidth, keepaspectratio]{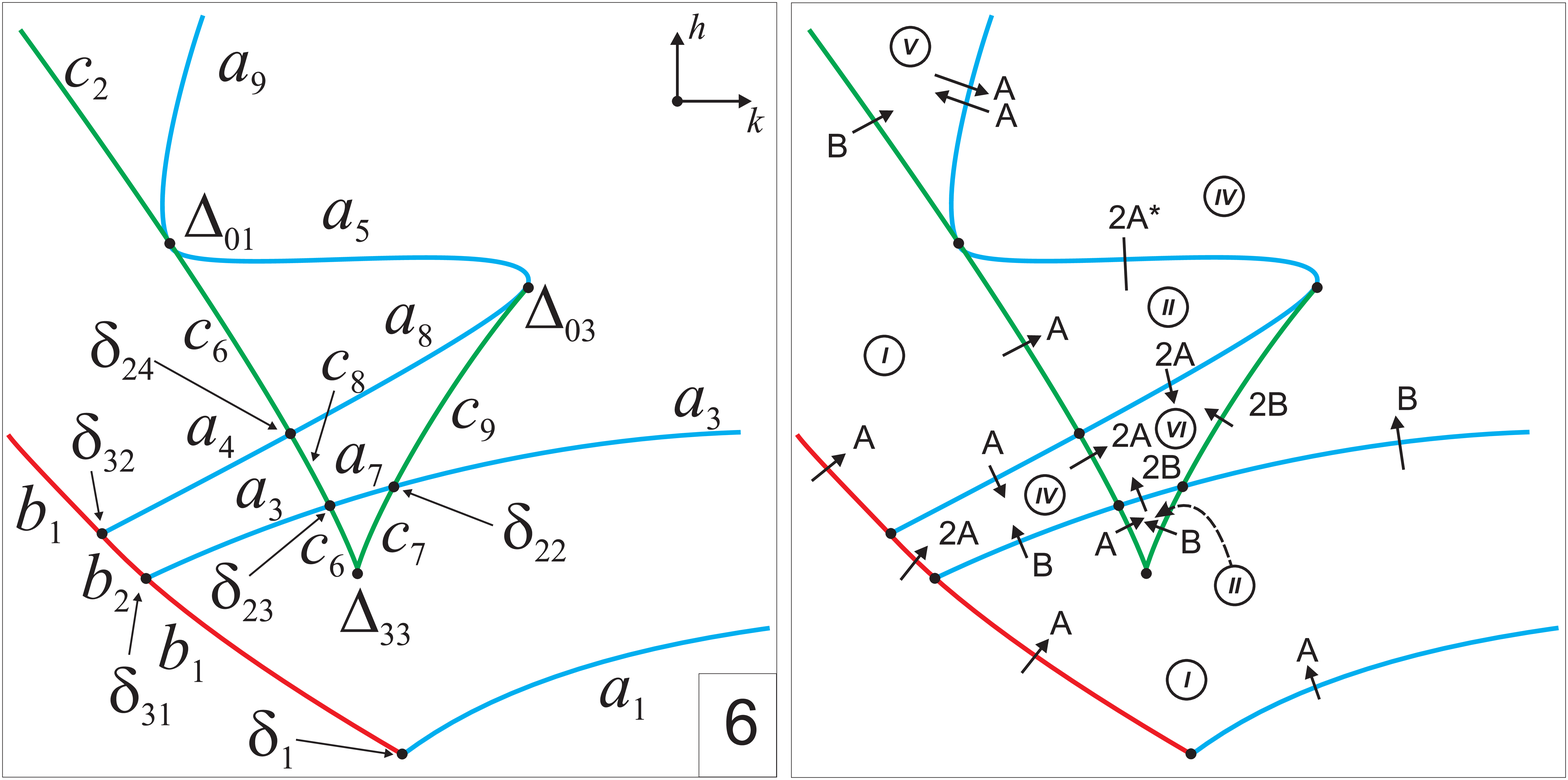}
\caption{Диаграммы $\mSell$ для областей 4 -- 6.}\label{fig_reg04both}
\end{figure}

\begin{figure}[!ht]
\centering
\includegraphics[width=\uuu\textwidth, keepaspectratio]{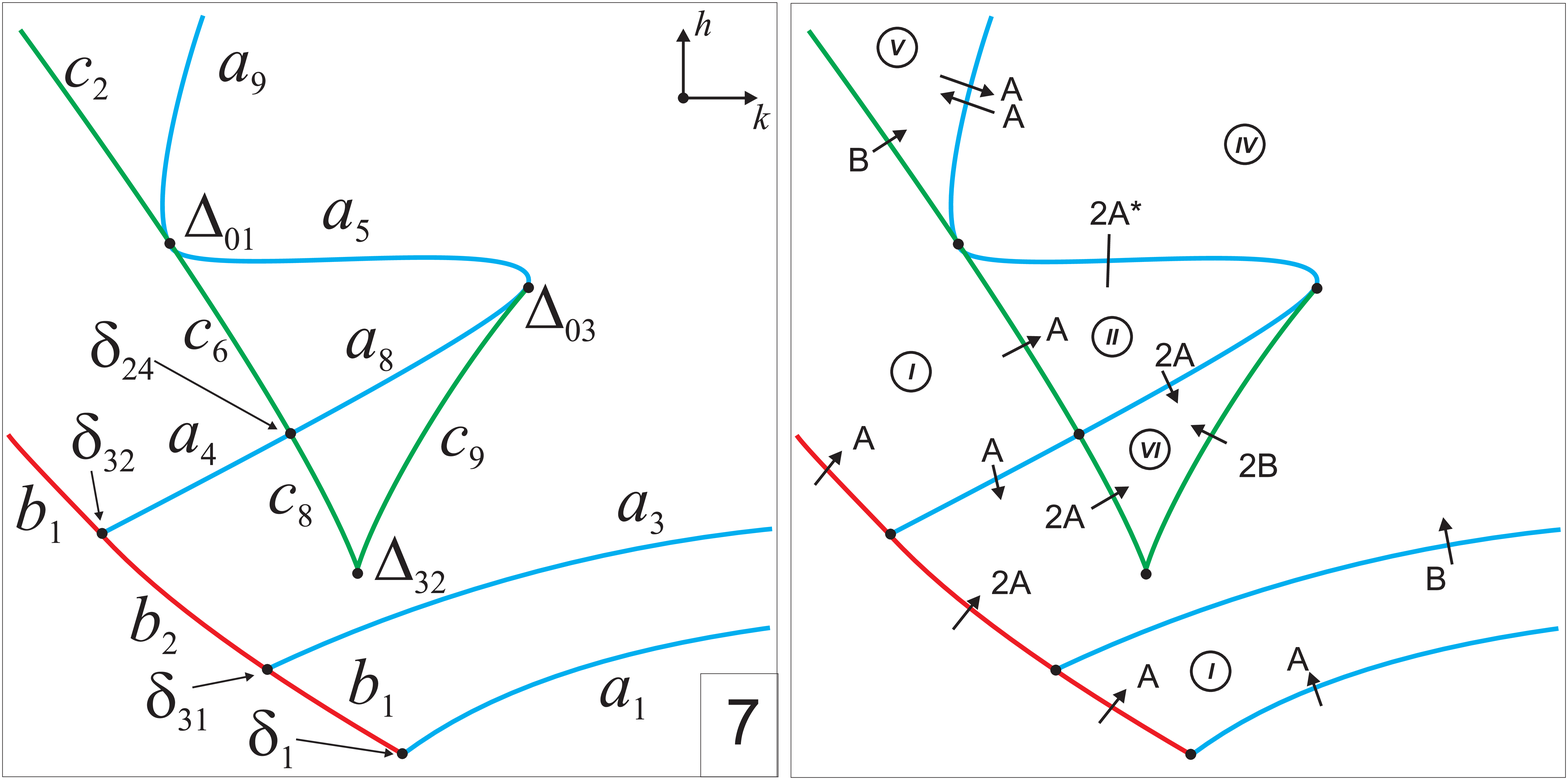}
\includegraphics[width=\uuu\textwidth, keepaspectratio]{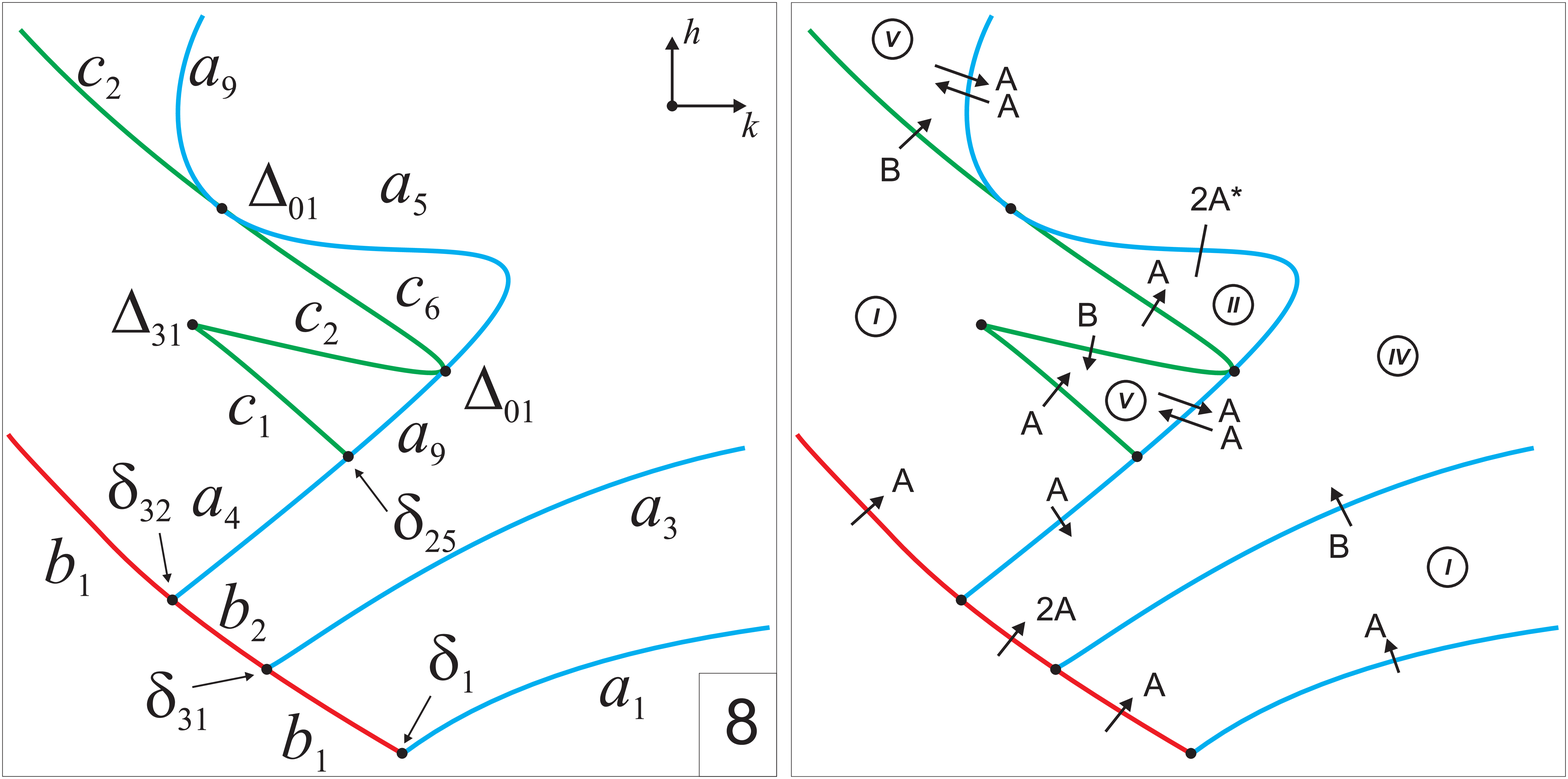}
\includegraphics[width=\uuu\textwidth, keepaspectratio]{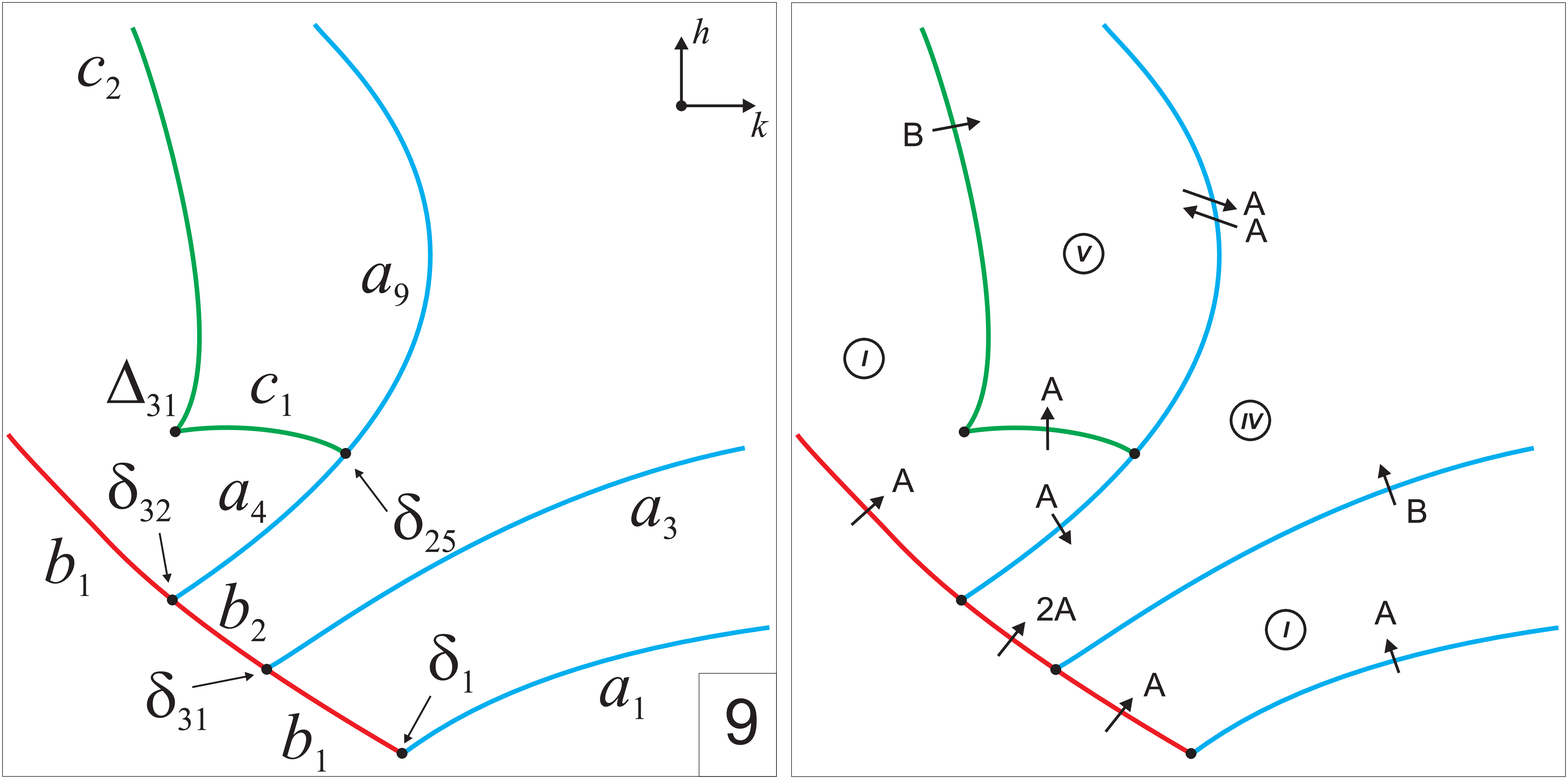}
\caption{Диаграммы $\mSell$ для областей 7 -- 9.}\label{fig_reg07both}
\end{figure}

\begin{figure}[!ht]
\centering
\includegraphics[width=\uuu\textwidth, keepaspectratio]{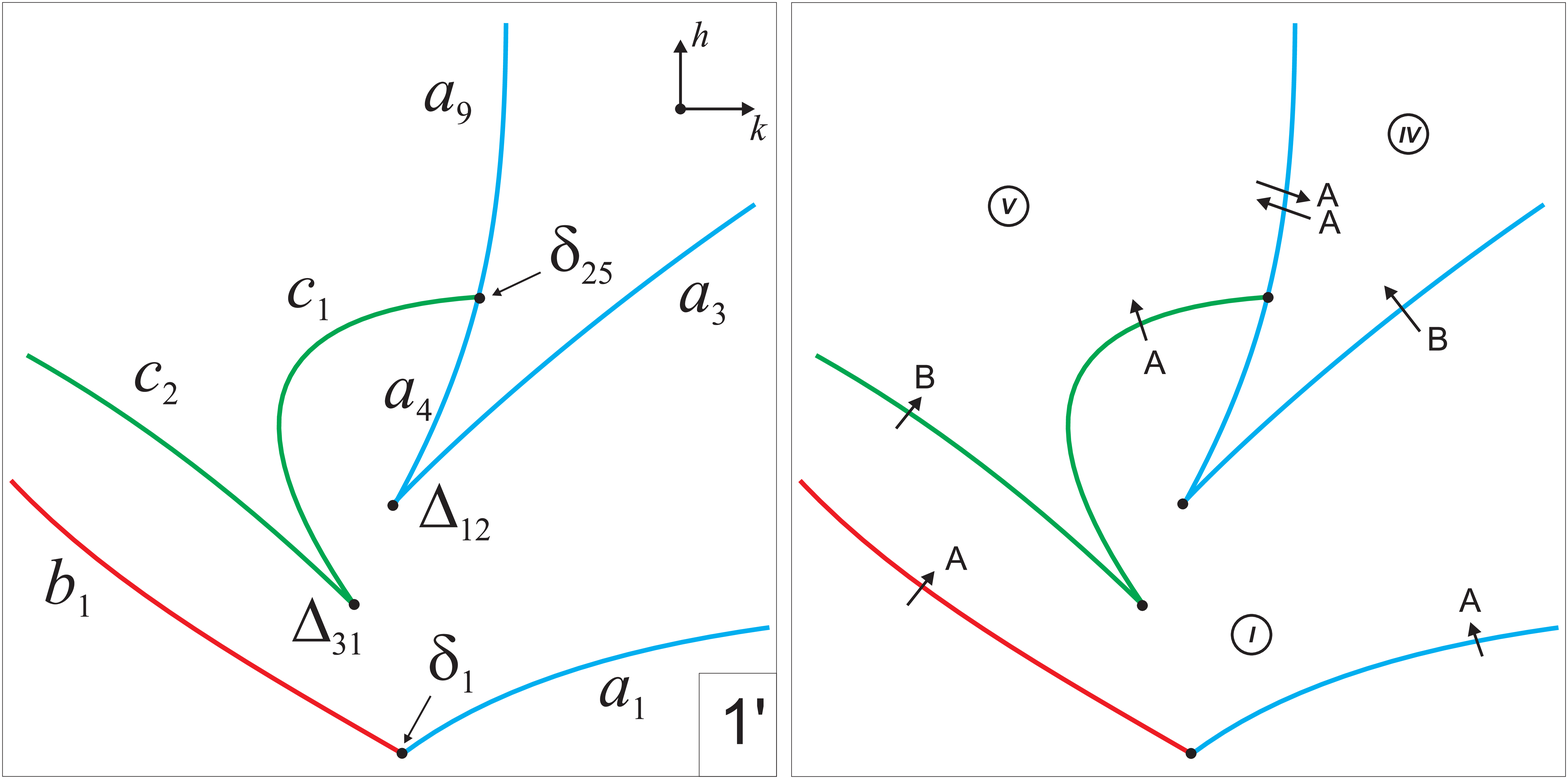}
\includegraphics[width=\uuu\textwidth, keepaspectratio]{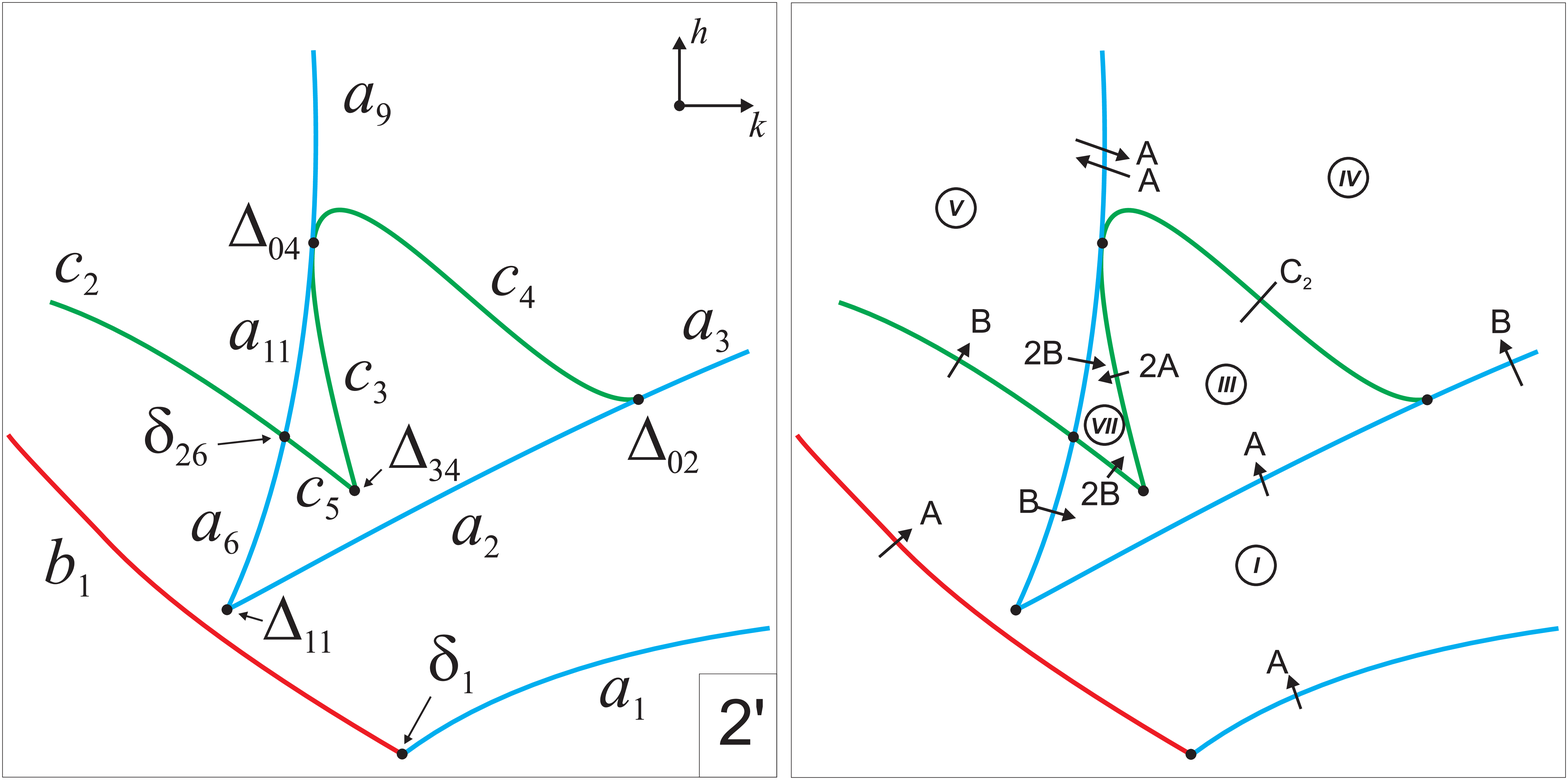}
\includegraphics[width=\uuu\textwidth, keepaspectratio]{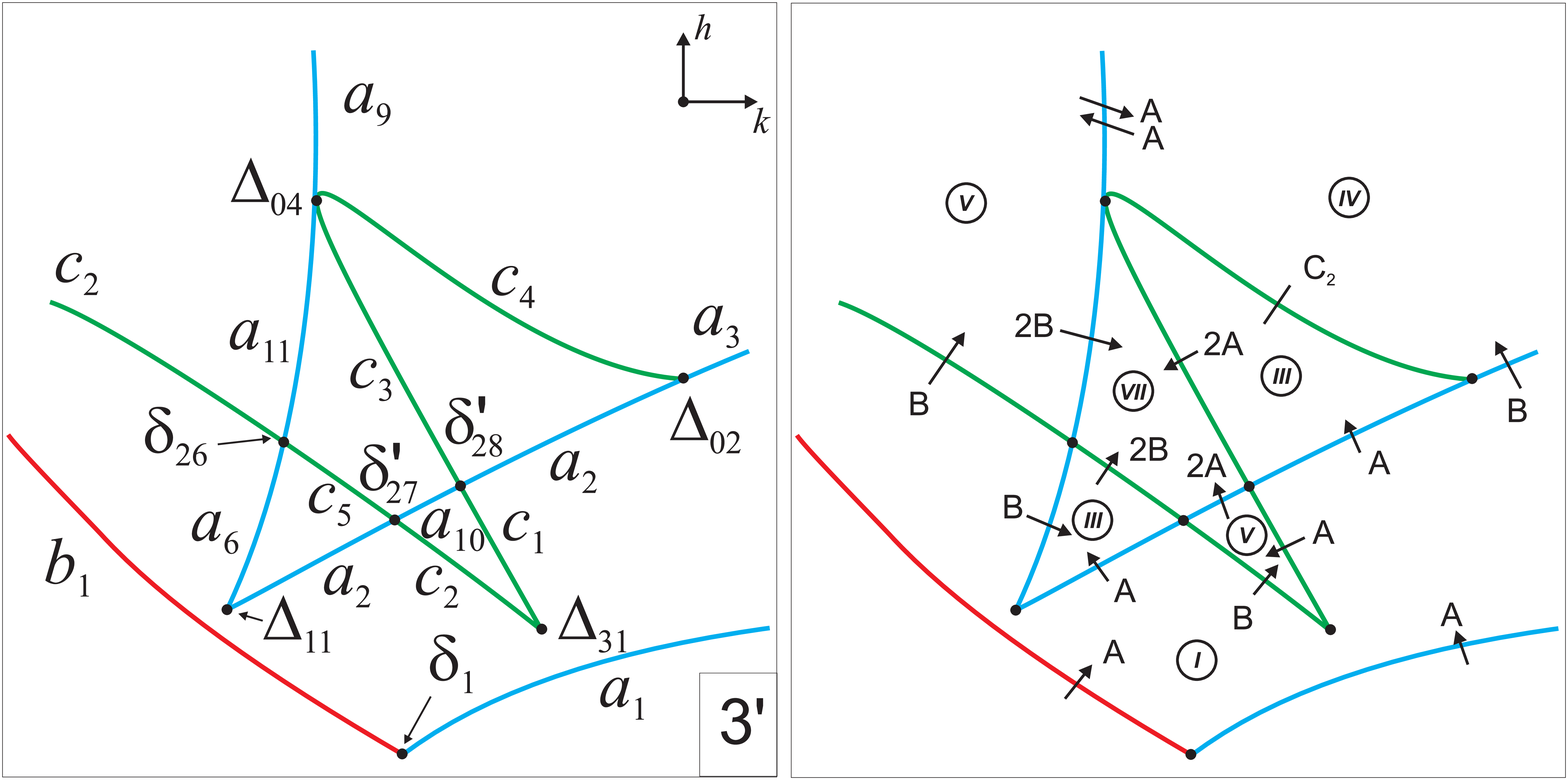}
\caption{Диаграммы $\mSell$ для областей $1'$ -- $3'$.}\label{fig_reg01sboth}
\end{figure}

\begin{figure}[!ht]
\centering
\includegraphics[width=\uuu\textwidth, keepaspectratio]{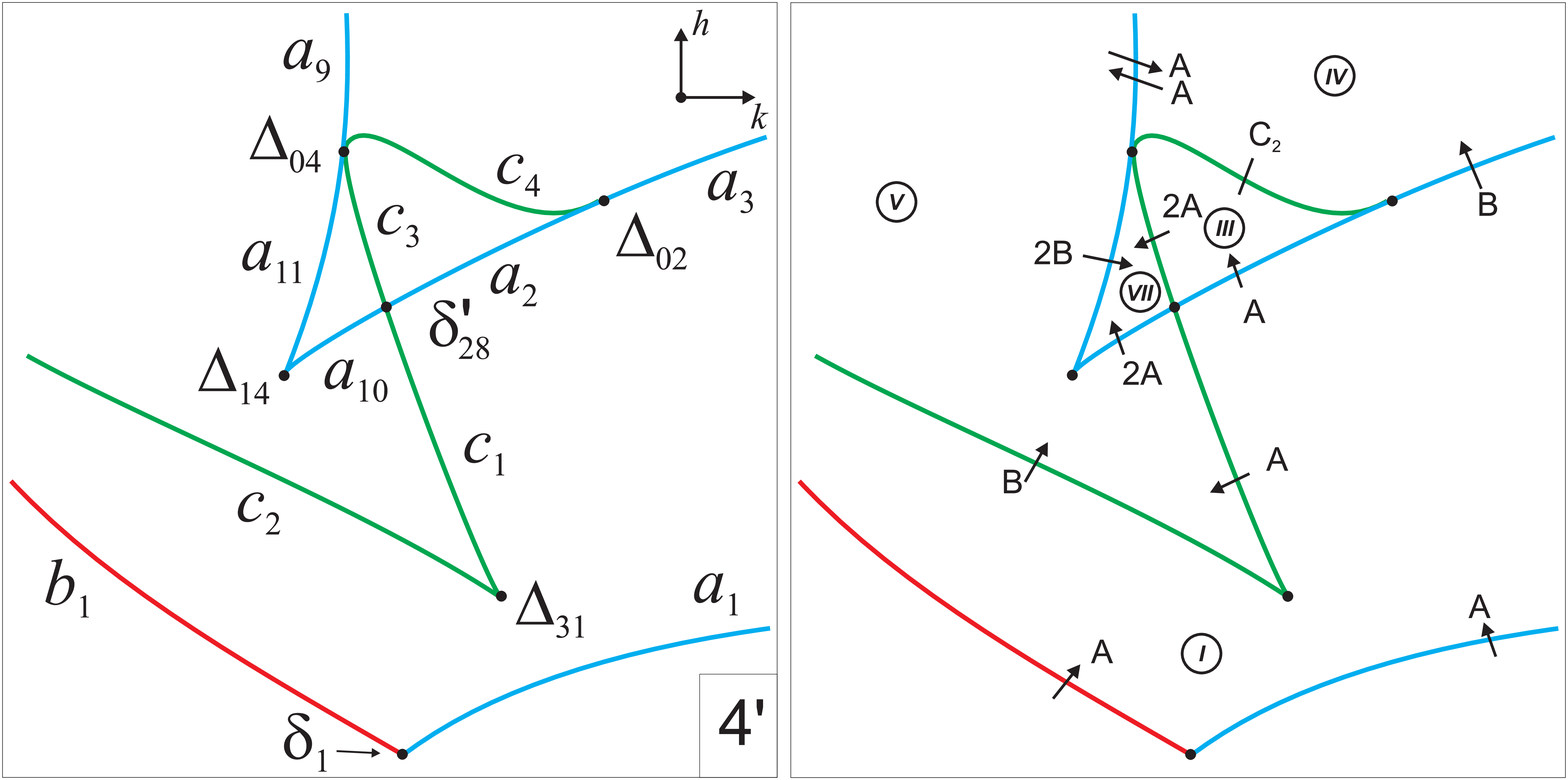}
\includegraphics[width=\uuu\textwidth, keepaspectratio]{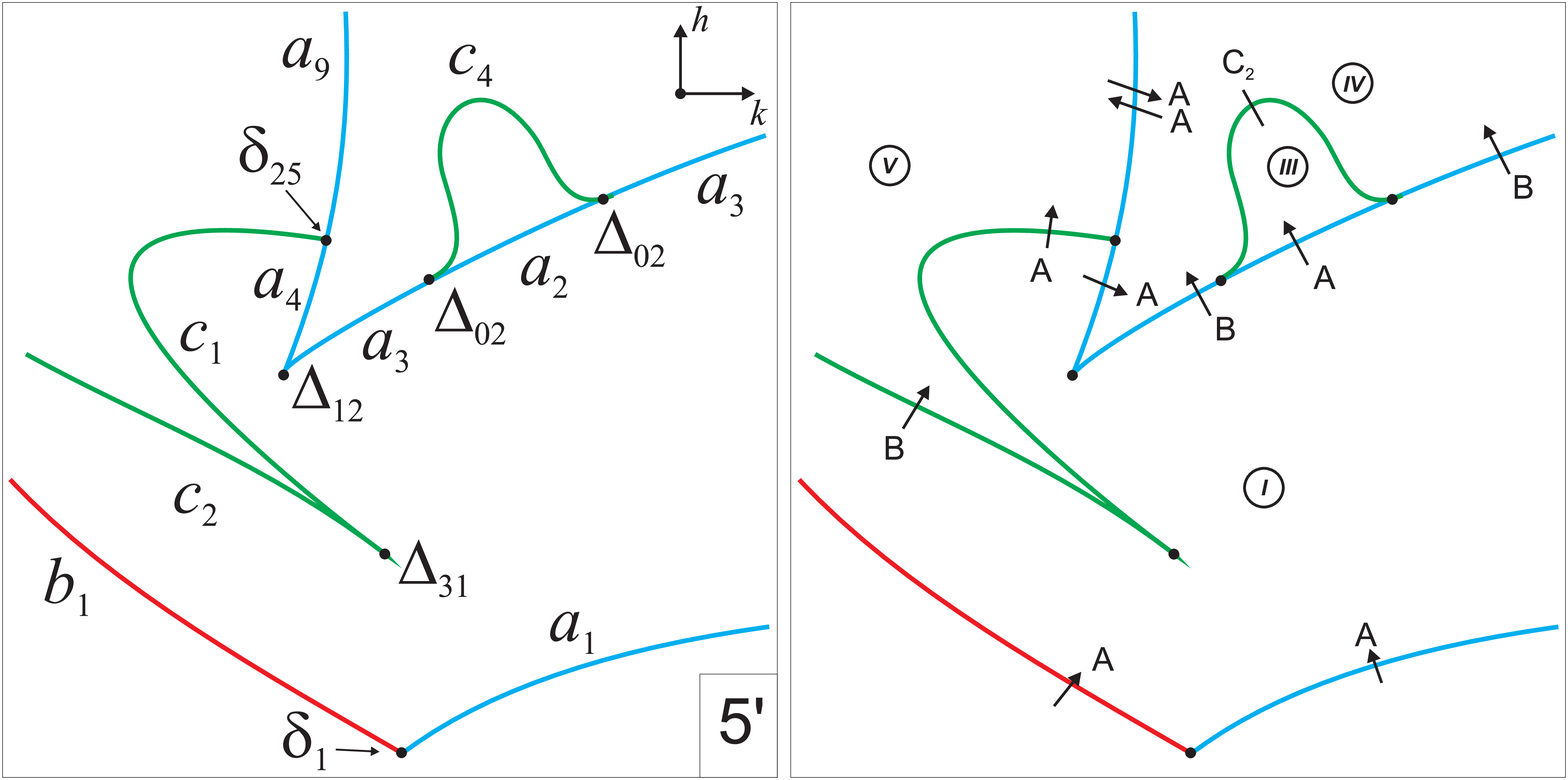}
\includegraphics[width=\uuu\textwidth, keepaspectratio]{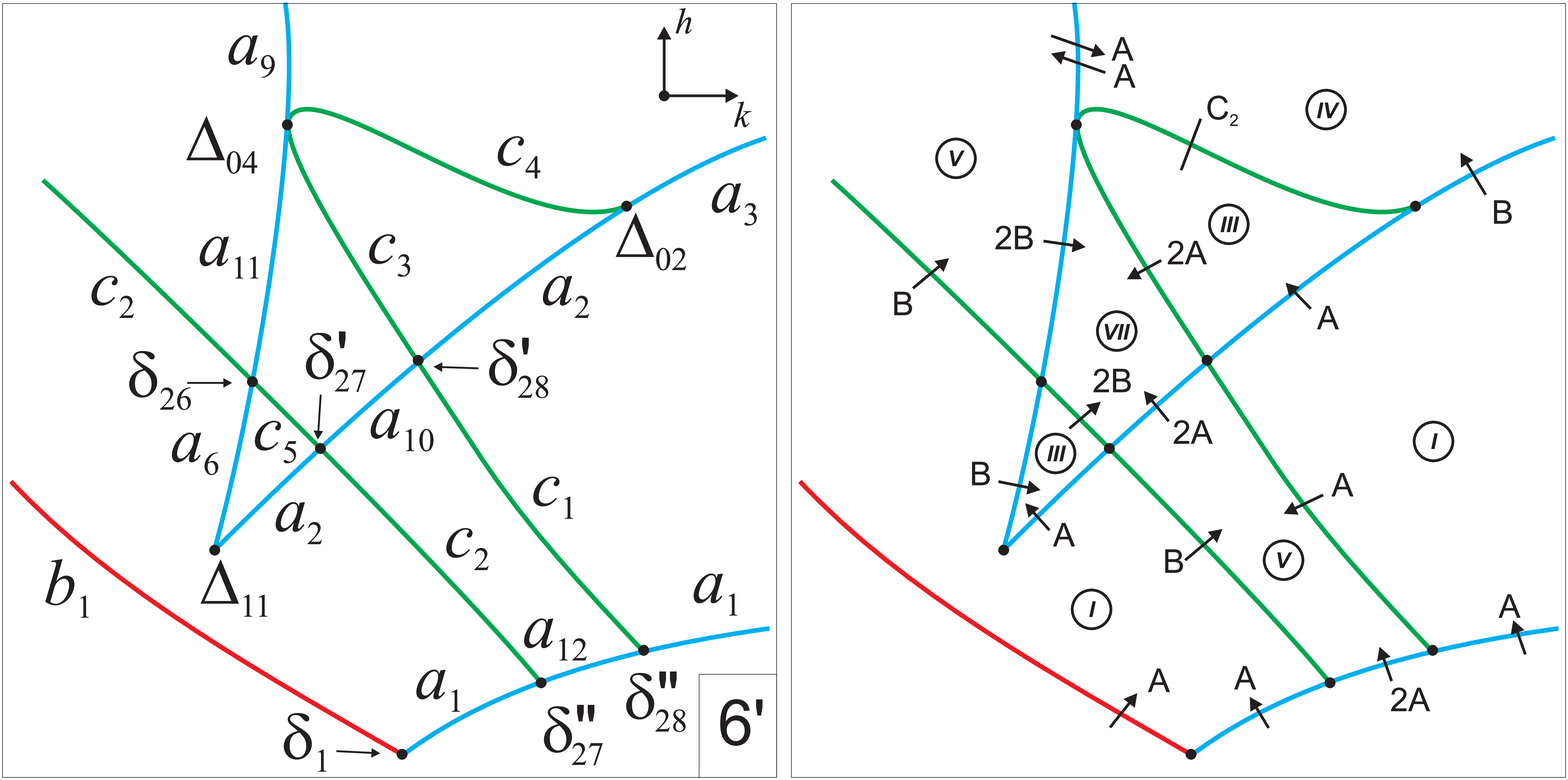}
\caption{Диаграммы $\mSell$ для областей $4'$ -- $6'$.}\label{fig_reg04sboth}
\end{figure}

\begin{figure}[!ht]
\centering
\includegraphics[width=\uuu\textwidth, keepaspectratio]{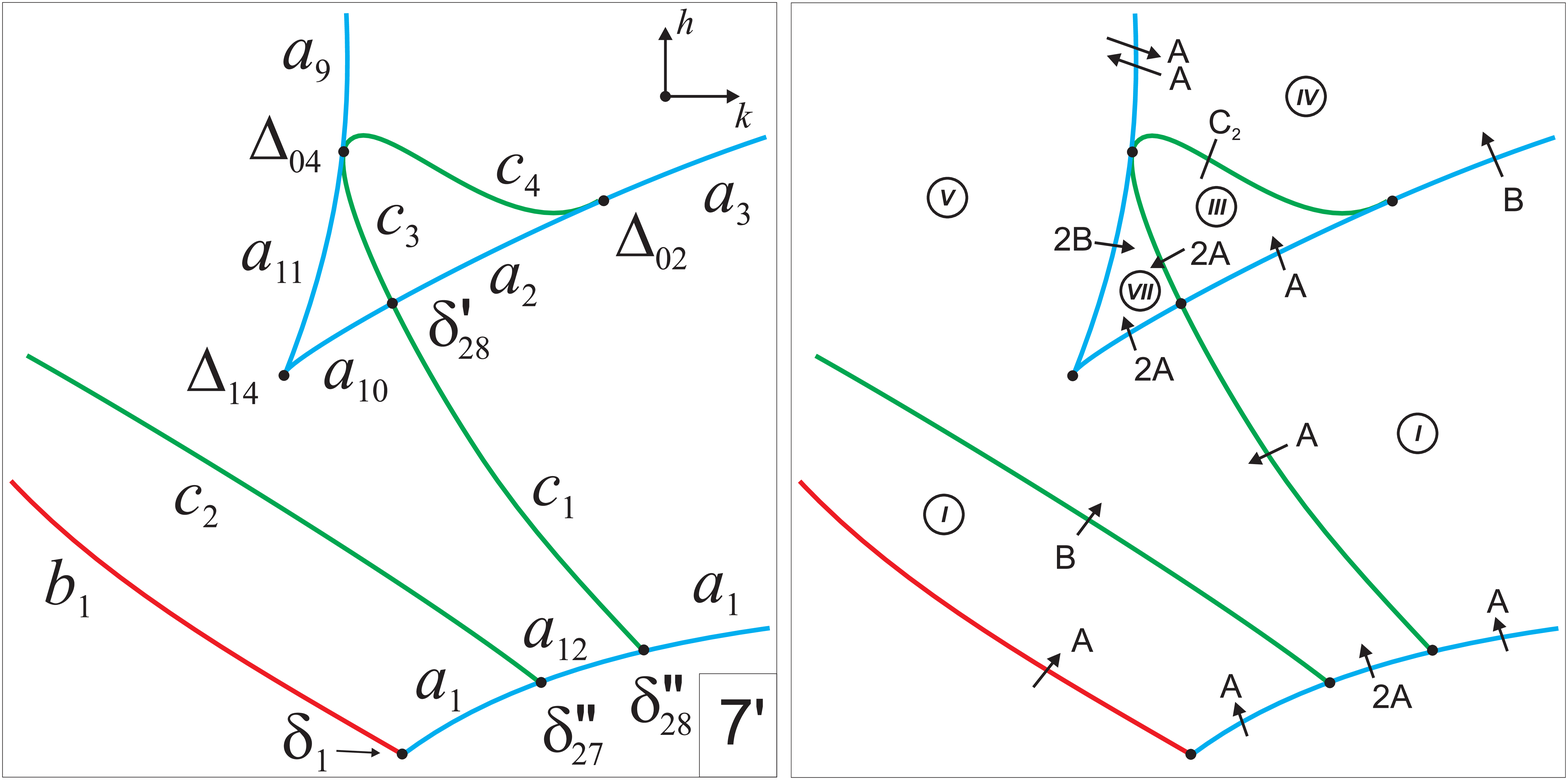}
\includegraphics[width=\uuu\textwidth, keepaspectratio]{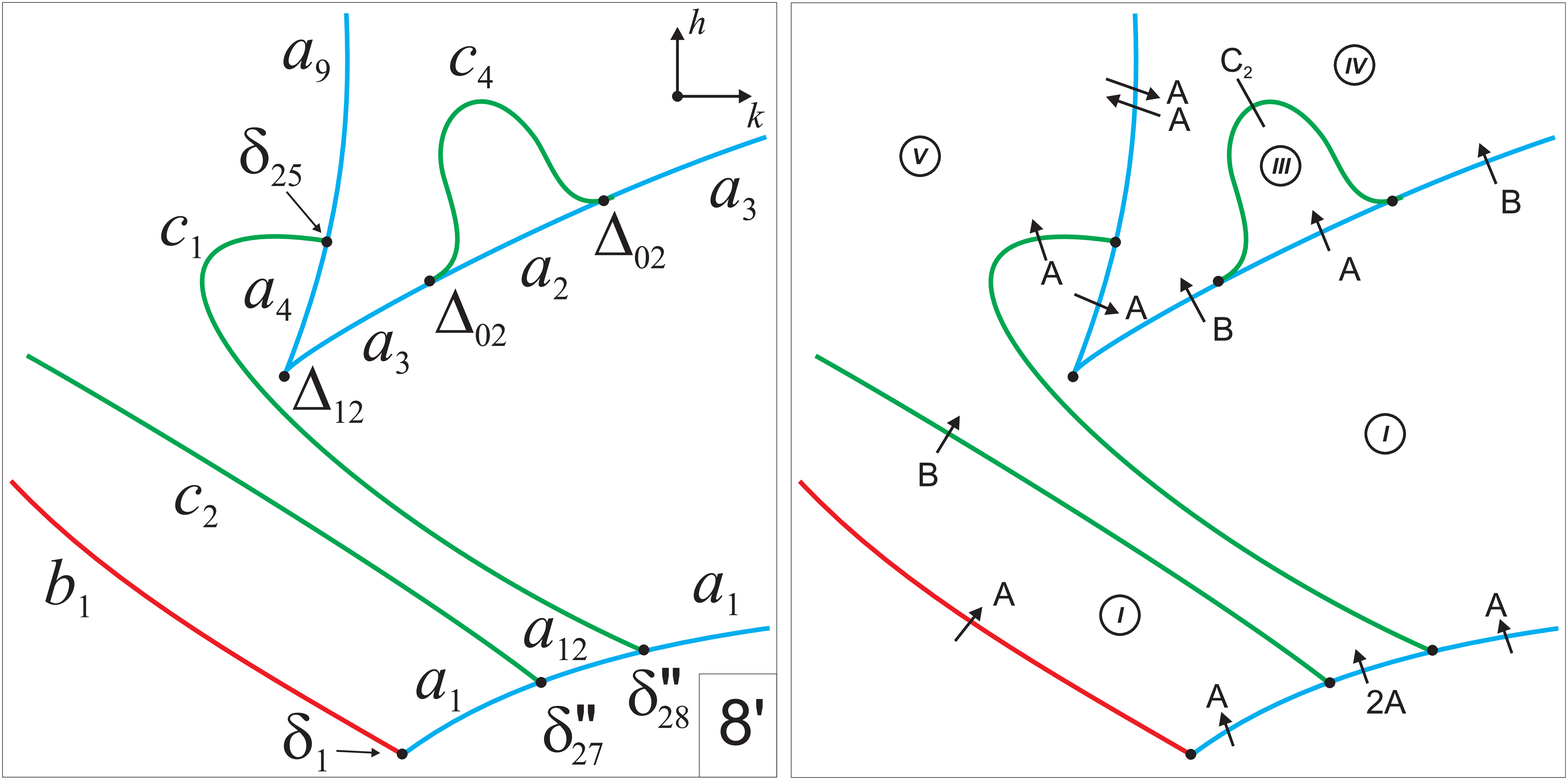}
\includegraphics[width=\uuu\textwidth, keepaspectratio]{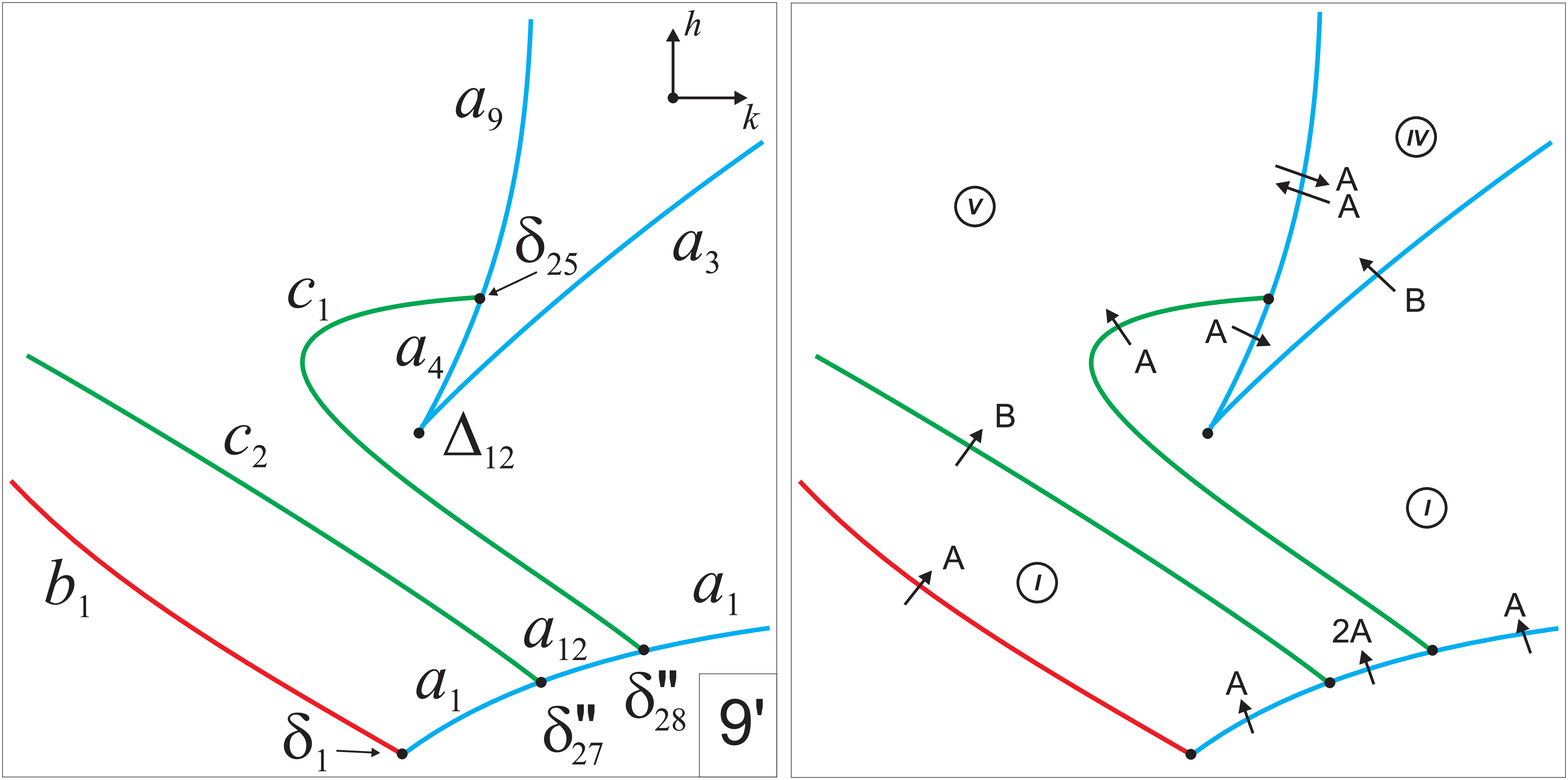}
\caption{Диаграммы $\mSell$ для областей $7'$ -- $9'$.}\label{fig_reg07sboth}
\end{figure}

\clearpage

\subsection{Топологический анализ}
Первый вопрос топологического анализа -- сколько торов Лиувилля содержит в себе интегральное многообразие в прообразе регулярной точки? Не зная точно топологический тип критического интегрального многообразия с гиперболическими особенностями, ответить на этот вопрос в одной отдельно взятой системе на $\mPel$ достаточно сложно. В системе с двумя степенями свободы характер бифуркации, происходящей при пересечении критических точек ранга~1, однозначно определен типом точки лишь в случае эллиптических особенностей~-- атомов~$A$. В данной задаче доказательно исследованы лишь два частных случая $\ld=0$ \cite{KhDan83,KhPMM83,KhBook88} и $\ell=0$ \cite {RyabDis,RyabRCD}. В общем случае утверждение \cite{Gash2,Gash5} о встречающемся в этой задаче количестве торов в составе регулярных многообразий основано на численном исследовании их проекций на плоскость первых двух компонент угловой скорости (аналог областей Жуковского). Рассматривая, однако, всю картину в четырехмерном пространстве интегральных параметров $\bR^4(\ld,\ell,k,h)$ (расширенном пространстве), получим ответ сразу же, не прибегая ни к каким дополнительным исследованиям. Действительно, отбросим все гиперболические атомы, считая их точный вид пока неизвестным. Для расширенной камеры гладкое ребро диаграммы $\mSell$ с фиксированным обозначением порождает (трехмерную) стенку. Для выбранной камеры атом $A$ на ребре или стенке назовем \textit{входом}, если его стрелка направлена внутрь камеры (то есть тор Лиувилля рождается при входе в камеру) и \textit{выходом} -- в противоположном случае. Рассмотрим три сечения в расширенном пространстве, отвечающие областям $1$, $2$ и $2'$ (рис.~\ref{fig_define_tori}). Самая нижняя (относительно направления оси $h$) точка диаграммы отвечает одной невырожденной критической точке ранга~0 типа ``центр-центр'' -- самому нижнему относительному равновесию при заданных $(\ld,\ell)$, поэтому, как отмечено выше, в незанумерованной камере ниже диаграммы движений нет. Следовательно, в камере $\ts{I}$ имеем ровно один тор (на стенках $\aaa_1$, $\bbb_1$ по одному входу). Из диаграммы для области $1$ видно, что из камеры $\ts{I}$ имеется ровно по одному входу в камеры $\ts{II}$ и $\ts{III}$ (соответственно через стенки $\ccc_6$ и $\ccc_4$). Аналогично, из диаграммы для области $1'$ видно, что из камеры $\ts{I}$ имеется ровно по одному входу и в камеры $\ts{IV}$ и $\ts{V}$ (соответственно через стенки $\aaa_4$ и $\ccc_1$). Следовательно, в камерах $\ts{II} - \ts{V}$ каждой точке отвечают два тора Лиувилля. Диаграмма для области $2$ демонстрирует ровно два входа из камеры $\ts{II}$ в камеру $\ts{VI}$ (стенка $\aaa_8$), а диаграмма для области $2'$ дает ровно два входа из камеры $\ts{III}$ в камеру $\ts{VII}$ (стенка $\ccc_3$). Следовательно, в камерах $\ts{VI}-\ts{VII}$ имеем четыре тора Лиувилля. Следующая теорема уточняет результаты работ \cite{Gash2,Gash5,GashDis}, в которых ошибочно утверждается наличие шести, а не семи различных камер. В частности, при малых отличных от нуля значениях $\ld$ констатируется наличие лишь пяти камер, хотя, как видно уже из диаграмм для областей $1, 2$ на рис.~\ref{fig_reg01both}, для сколь угодно малых $\ld$ камера $\ts{III}$ области~$1$ меняется на камеру~$\ts{VI}$ области~$2$. Это означает, что в пространстве $\bR^3(\ell,h,k)$ при сколь угодно малых фиксированных $\ld$ диаграмма полного отображения момента делит это пространство на шесть, а не на пять связных компонент с непустыми интегральными многообразиями. При увеличении $\ld$ после прохождения этим параметром особого значения $2^{-3/4}$ шестая камера исчезает, а появляется седьмая.

\begin{theorem}\label{theo_tori}
Интегральные многообразия в прообразах точек камер пространства интегральных постоянных таковы: $\bT^2$ для камеры $\ts{I}$; $2\bT^2$ для камер $\ts{II} - \ts{V}$; $4\bT^2$ для камер $\ts{VI} - \ts{VII}$.
\end{theorem}

\def\uuu{0.9}
\begin{figure}[!ht]
\centering
\includegraphics[width=\uuu\textwidth, keepaspectratio]{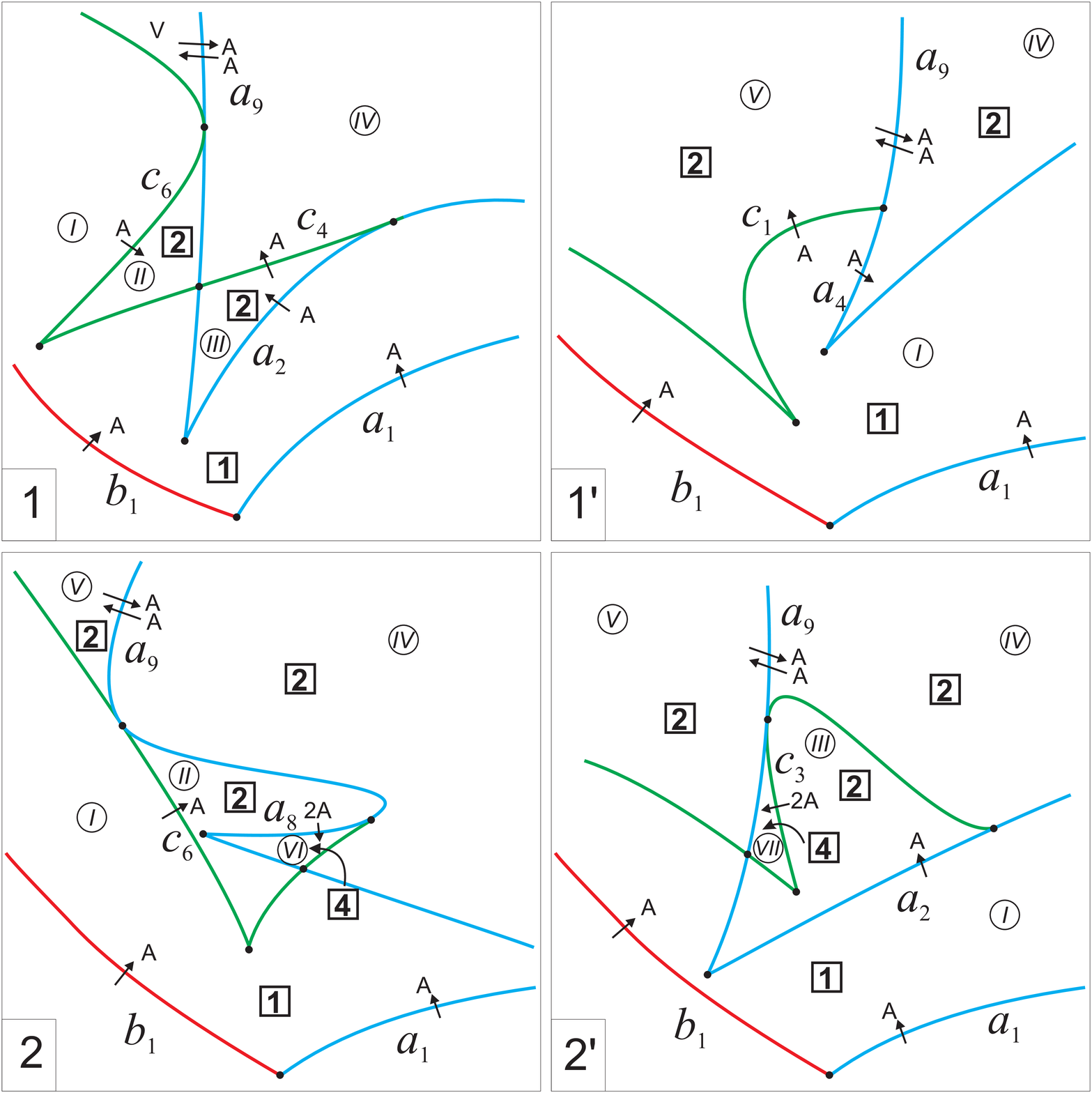}
\caption{Количество торов в камерах.}\label{fig_define_tori}
\end{figure}

Для изучения тонкой топологии системы одним из основных объектов исследования, наряду с камерами и количеством торов в них, являются так называемые семейства торов Лиувилля.
Обозначим для краткости через $\mPr$ расширенное фазовое пространство:
$$
\mPr=\mwide{\mP^5}=\bigcup_{(\ld,\ell)} \mPel{\times}\{(\ld,\ell)\} = \bigcup_{\ld} S^2(\ba){\times}\bR^3(\bo){\times}\{\ld\} = S^2(\ba){\times}\bR^3(\bo){\times}\bR(\ld).
$$
Все интегральные многообразия $J_{\ell,h,k}(\ld)$ вкладываются в $\mPr$ как $J_{\ell,h,k}(\ld){\times}\{\ld\}$. Естественно возникает следующее определение.

\begin{defin}
Исключим из $\mPr$ все связные компоненты критических уровней, содержащие критические точки отображения момента. Связная компонента оставшегося множества называется \emph{семейством} торов Лиувилля.
\end{defin}

Если в системе нет минимальных (максимальных) регулярных торов (а в рассматриваемой задаче их нет, что видно из классификации диаграмм, оснащенных атомами), то пересечение прообраза точки любой камеры с любым семейством может состоять не более чем из одного тора. В некоторых задачах говорят о нескольких торах одного семейства в прообразе точек камеры \cite{BRF} или, что то же самое, о семействах, состоящих из нескольких торов \cite{Mor}. В этом случае определение семейства принимается таким: в одно семейство относятся торы, которые испытывают одинаковые бифуркации на границах камеры \cite{Mor}. Подобные семейства встречаются в классическом случае Ковалевской $\ld=0$ \cite{BRF} и в частном случае $\ell=0$ \cite{Mor}, однако, как легко видеть, в общем случае $\ld \ell \ne 0$ таких семейств просто нет. Действительно, в камерах $\ts{II}, \ts{III}$ имеется два тора и стенка с одним входом, на которой, следовательно, подвергается бифуркации только один тор. Между камерами $\ts{IV}, \ts{V}$ имеется стенка, содержащая один вход и один выход, поэтому два тора этих камер испытывают заведомо различные бифуркации. У камер $\ts{VI},\ts{VII}$, точкам которых отвечают четыре тора, не только в расширенном пространстве, но и при любых фиксированных $\ld,\ell$, при которых эти камеры существуют, имеются одна либо две стенки с двумя входами из соседних камер, причем, если стенок две, то на них испытывают бифуркации разные пары торов. Поэтому и здесь нет сразу четырех или даже двух торов, испытывающих одинаковые бифуркации. В частности, отсюда следует, что заведомо неверна теорема 6 работы \cite{LogRus}, в которой утверждается наличие семи семейств, два из которых имеют в своем составе по два тора. Семейства в работе \cite{LogRus} понимаются в смысле определения работы \cite{Mor}. Еще раз  подчеркнем, что при любых фиксированных $\ld,\ell$, таких что $\ld \ell \ne 0$, не существует камер, на границах которых все торы этих камер испытывали бы одинаковые бифуркации.

Установим количество семейств и их распределение по камерам. Отметим одно свойство случая Ковалевской\,--\,Яхья: в прообразе точки бифуркационной диаграммы регулярные торы не могут присутствовать одновременно с гиперболическими атомами. Иначе говоря, если ребру диаграммы отвечают гиперболические атомы, то в бифуркациях участвуют все торы прилегающих камер\footnote{Это свойство не имеет места, например, в случае Горячева\,--\,Чаплыгина.}. Поэтому для того, чтобы в двух камерах присутствовало одно и то же семейство, необходимо (но, конечно, не достаточно), чтобы из одной из них существовал путь в другую, состоящий из входов или выходов. Если в некоторой камере семейство только одно, то наличие из нее выхода в любую камеру является и достаточным условием того, что данное семейство сохранится во второй камере. Рассмотрим диаграммы для областей $1$, $1'$, $3$ и $3'$ (рис.~\ref{fig_define_fams}). Семейства нумеруем арабскими цифрами и, как множества, помещаем в фигурные скобки (в отличие от номеров областей на $(\ld,\ell)$-плоскости). Обозначим через $\tfs{1}$ семейство, начинающееся в камере~$\ts{I}$. Из диаграммы для области $1$ видно, что из камеры $\ts{I}$ имеется ровно по одному входу в камеры $\ts{II},\ts{III}$. Поэтому семейство $\tfs{1}$ продолжает жить и в этих камерах. Могут ли торы, рождающиеся на входах в эти камеры, принадлежать одному семейству? На диаграммах c нанесенными атомами только типа $A$ прослеживаем, что между камерами $\ts{II},\ts{III}$ не существует пути только с атомами $A$, не заходящего в камеру $\ts{I}$ (при таком заходе новые семейства умирают). Поэтому вторые семейства в этих камерах различны. Обозначим их через $\tfs{2}$ и $\tfs{3}$. Из диаграммы для области $1'$ также вытекает, что семейство $\tfs{1}$ продолжает жить и в камерах $\ts{IV},\ts{V}$. Могут ли торы, рождающиеся на входах в эти камеры, принадлежать одному семейству? Очевидно, нет, поскольку общая стенка имеет вход и выход, поэтому при ее пересечении сохраняется лишь семейство $\tfs{1}$, а другого пути только с атомами $A$ между камерами $\ts{IV},\ts{V}$, не заходящего в камеру $\ts{I}$, не существует. Аналогично устанавливаем, что семейства, рождающиеся на общей стенке камер $\ts{IV},\ts{V}$, не могут совпадать ни с одним из уже введенных. Поэтому обозначим эти новые семейства через $\tfs{4},\tfs{5}$. Из диаграммы для области~$3$ видно, что в камеру $\ts{VI}$ с четырьмя торами имеется по два входа из камер $\ts{II}$, $\ts{IV}$. Поэтому там сохранятся семейства $\tfs{1},\tfs{2},\tfs{4}$, а родится лишь одно новое семейство, которое обозначим $\tfs{6}$. Аналогично, из диаграммы для области $3'$ видно, что в камеру $\ts{VII}$ с четырьмя торами имеется по два входа из камер $\ts{III}$, $\ts{V}$. Поэтому там сохранятся семейства $\tfs{1},\tfs{3},\tfs{5}$, а родится лишь одно новое семейство, которое обозначим~$\tfs{7}$. Номера семейств, присутствующих одновременно в камере, соединяем знаком ``плюс''.

\begin{figure}[ht]
\centering
\includegraphics[width=\uuu\textwidth, keepaspectratio]{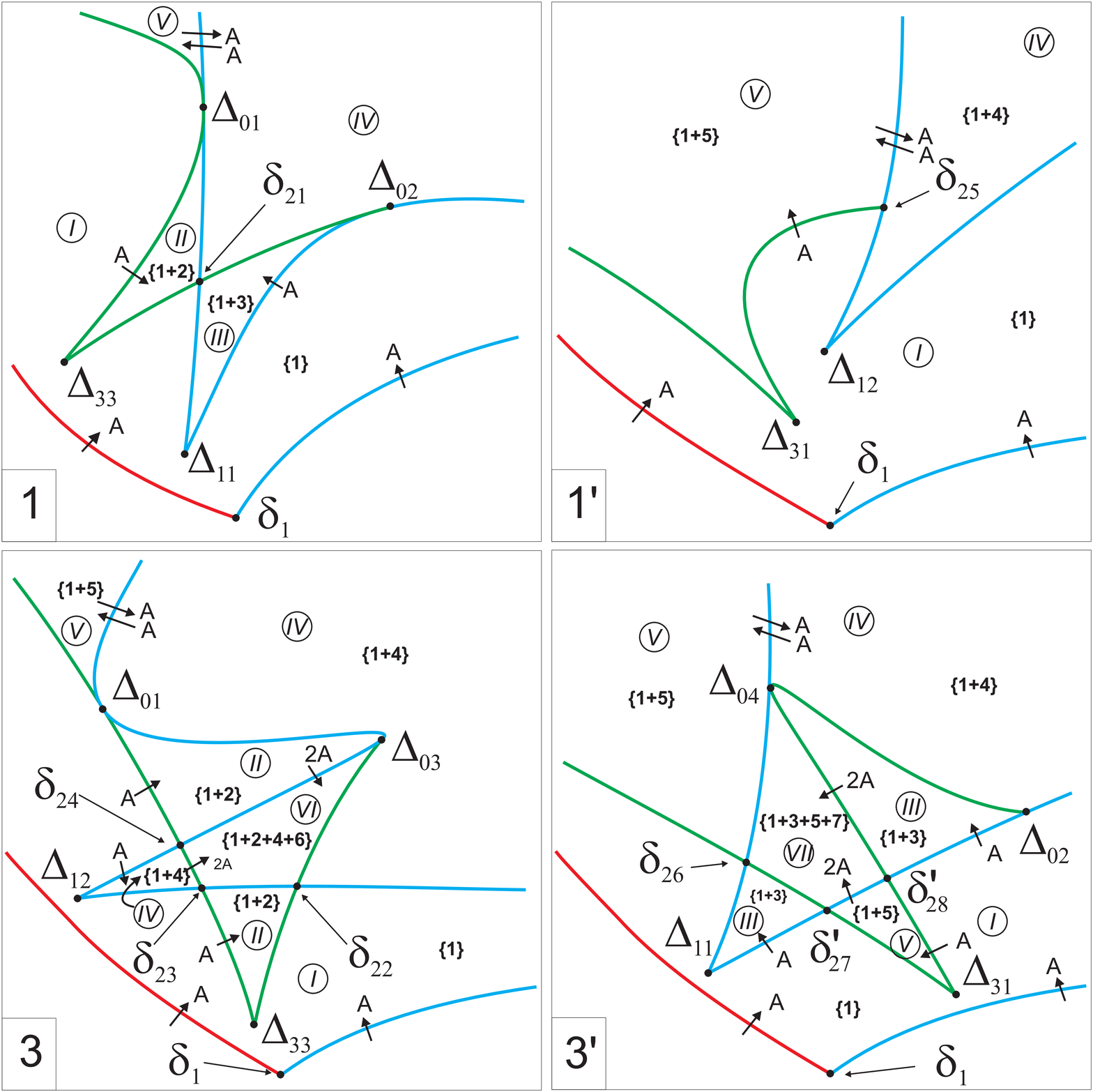}
\caption{Семейства торов в камерах.}\label{fig_define_fams}
\end{figure}

\begin{theorem}\label{theo_fams}
В случае Ковалевской\,--\,Яхья существует ровно семь семейств регулярных торов и их распределение по камерам указано в таблице~{\rm \ref{tabfam}}.
\begin{center}
\small
\begin{tabular}{|c|c|}
\multicolumn{2}{r}{{Таблица \myt\label{tabfam}}}\\
\hline
\begin{tabular}{c}{Камера}\end{tabular} &\begin{tabular}{c}{Семейства}\end{tabular}\\
\hline
\begin{tabular}{c}$\ts{I}$\end{tabular} &$\tfs{1}$\\
\hline
\begin{tabular}{c}$\ts{II}$\end{tabular} &$\tfs{1+2}$\\
\hline
\begin{tabular}{c}$\ts{III}$\end{tabular} &$\tfs{1+3}$\\
\hline
\begin{tabular}{c}$\ts{IV}$\end{tabular} &$\tfs{1+4}$\\
\hline
\begin{tabular}{c}$\ts{V}$\end{tabular} &$\tfs{1+5}$\\
\hline
\begin{tabular}{c}$\ts{VI}$\end{tabular} &$\tfs{1+2+4+6}$\\
\hline
\begin{tabular}{c}$\ts{VII}$\end{tabular} &$\tfs{1+3+5+7}$\\
\hline
\end{tabular}
\end{center}
\end{theorem}

Далее мы устанавливаем грубые и меченые круговые молекулы узловых точек бифуркационных диаграмм. Под грубой молекулой понимается граф, вершины которого отвечают бифуркациям, а ребра -- семействам регулярных торов, и при этом вершины и ребра оснащены обозначениями соответствующих бифуркаций и семейств. Меченая молекула содержит информацию о метках на ребрах графа и, при необходимости, метки так называемых семей (подробности см. в \cite{BolFom}). Основная проблема -- правильно отследить семейства, смыкающиеся в точках бифуркации. Знание грубой топологии, то есть топологии совместных уровней первых интегралов, как регулярных, так и критических, не всегда позволяет однозначно расставить обозначения семейств. Но, как будет видно ниже, после того как эта задача решена, возможно не только построить грубую молекулу, но и воспользоваться результатами общей теории, исследованиями частных случаев для доказательного построения всех меченых круговых молекул.

\clearpage

\section{Тонкая топология узловых точек}\label{sec7}
\subsection{Круговые молекулы точек ранга $0$}

Приведем описание топологии полных интегральных уровней относительных равновесий (критических точек ранга 0). Эта совокупность результатов представлена в работе \cite{KhMTT42} и опирается на работу \cite{KhRy2011}, в которой, в частности, построены и используемые здесь многократно таблицы \ref{table2}, \ref{table4} и \ref{table5}.

Пусть $x$ -- точка невырожденного относительного равновесия, $\ell=L(x)$. Рассмотрим $\jpriv_\ell=H{\times}K: \mP^4_\ell \to \bR^2$ -- интегральное отображение \eqref{eq3_10} приведенной системы. Обозначим через $\mathfrak{J}(x)$ полный прообраз точки $\jpriv_\ell(x)$ -- критическую интегральную поверхность, а через $U(x)$~-- достаточно малую насыщенную окрестность поверхности $\mathfrak{J}(x)$, не содержащую относительных равновесий с другими значениями отображения $\jpriv_\ell$. Поверхность $\mathfrak{J}(x)$ может состоять из нескольких связных компонент. Как отмечено выше, компонента точки $x$ всегда содержит ровно одно относительное равновесие. Как следует из исследования критических подсистем, одновременно могут существовать компоненты, содержащие ровно одну критическую окружность. Количество и топология этих компонент  устанавливаются уже по результатам для критической подсистемы $\mm$ согласно замечанию~\ref{rem7}. Далее, рассматриваем точку $x$ в каждой из двух содержащих ее критических подсистем и анализируем информацию по прилегающим областям в образе критической подсистемы. Эти области порождают дуги бифуркационной диаграммы отображения $\jpriv_\ell$ в окрестности точки $x$. Структура критического множества в прообразах этих дуг и перестройки в $U(x)$ при их пересечении находятся по таблицам \ref{table2}, \ref{table4} и \ref{table5}. После этого тип круговых молекул самих относительных равновесий и лежащих на том же уровне критических периодических траекторий вместе с метками однозначно устанавливается исходя из исчерпывающего описания круговых молекул невырожденных особенностей малой сложности \cite{BolFom}. Кроме компонент, содержащих критические точки, в $\mathfrak{J}(x)$ могут входить и регулярные торы, заполненные двояко-периодическими траекториями. Их количество однозначно устанавливается по виду бифуркационной диаграммы отображения $\jpriv_\ell$, дополненной указанием атомов на дугах.

\def\wid{18mm}

{\small

\renewcommand{\arraystretch}{0}

\begin{center}

\begin{longtable}{|m{12mm}|m{20mm}|m{20mm}|m{30mm}|m{27mm}|m{12mm}|}
\multicolumn{6}{r}{\fts{Таблица \myt\label{table42}}}\\[2mm]
\hline
\hspace*{-2mm}{\renewcommand{\arraystretch}{0.8}\fns{\begin{tabular}{c} Класс\\точек\end{tabular}} }
&
\hspace*{-2.5mm}{ \renewcommand{\arraystretch}{0.8}\fns{\begin{tabular}{c} \ru{10} Компо-\\
нент\\ в прообразе
\end{tabular}}  }
&
\hspace*{-3mm}{ \fns{\renewcommand{\arraystretch}{0.8}\begin{tabular}{c} \ru{10} Особые\\
траектории\\ в прообразе
\end{tabular}}  } & \hspace*{-3mm}\fns{\renewcommand{\arraystretch}{0.8}\begin{tabular}{c}Фрагмент\\ $(H,K)$-диаграммы\end{tabular}} &
{\renewcommand{\arraystretch}{0.8} \fns{\begin{tabular}{c} Молекулы\\в прообразе\end{tabular}} }
&
{
\hspace*{-2mm}\fns{\renewcommand{\arraystretch}{0.8}\begin{tabular}{c}
 Регул.\\ торы
\end{tabular}} }\\
\hline\endfirsthead%
\multicolumn{6}{r}{\fts{Таблица \ref{table42} (продолжение)}}\\[2mm]
\hline
\hspace*{-2mm}{\renewcommand{\arraystretch}{0.8}\fns{\begin{tabular}{c} Класс\\точек\end{tabular}} }
&
\hspace*{-2.5mm}{ \renewcommand{\arraystretch}{0.8}\fns{\begin{tabular}{c} \ru{10} Компо-\\
нент\\ в прообразе
\end{tabular}}  }
&
\hspace*{-3mm}{ \fns{\renewcommand{\arraystretch}{0.8}\begin{tabular}{c} \ru{10} Особые\\
траектории\\ в прообразе
\end{tabular}}  } & \hspace*{-3mm}\fns{\renewcommand{\arraystretch}{0.8}\begin{tabular}{c}Фрагмент\\ $(H,K)$-диаграммы\end{tabular}} &
{\renewcommand{\arraystretch}{0.8} \fns{\begin{tabular}{c} Молекулы\\в прообразе\end{tabular}} }
&
{
\hspace*{-2mm}\fns{\renewcommand{\arraystretch}{0.8}\begin{tabular}{c}
 Регул.\\ торы
\end{tabular}} }\\
\hline\endhead%
\hf{$\delta_1$} & \hf{1} & \hf{$ {p_{CC}} $} & \hf{\includegraphics[width=\wid,keepaspectratio]{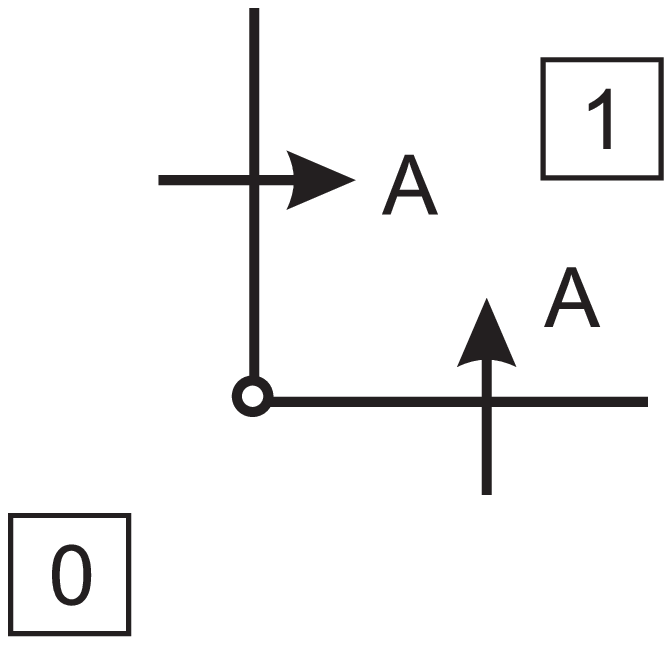}} & \hf{\includegraphics[width=\wid,keepaspectratio]{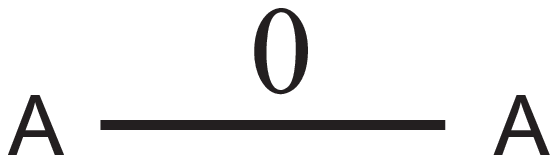}} & \hf{0}\\
\hline
 \hf{$\delta_{25}$} & \hf{3} & \hf{${p_{CC}} \cup S_E^1$}& \hf{\includegraphics[width=\wid,keepaspectratio]{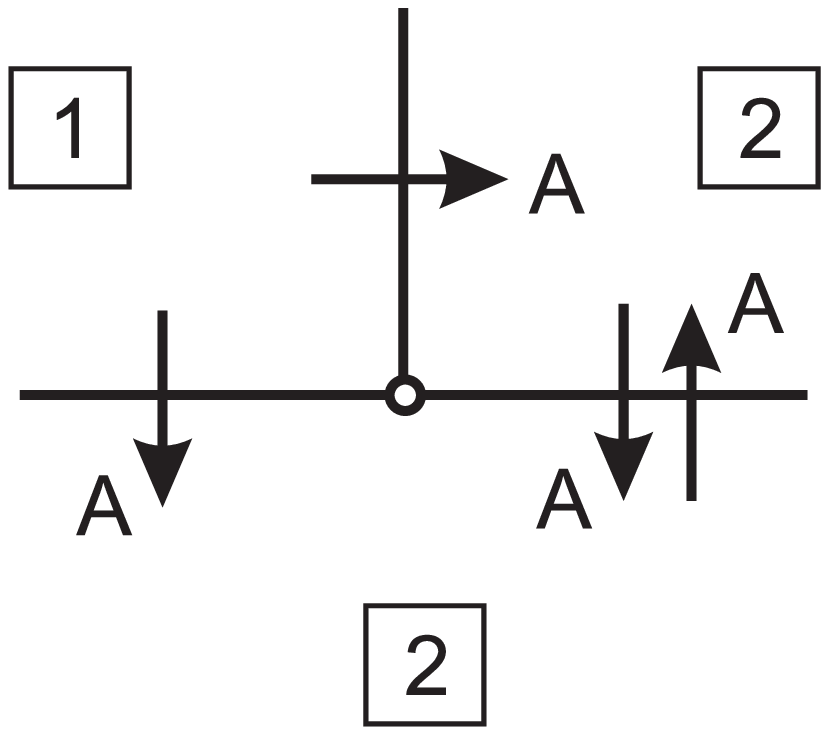}} & \hf{\includegraphics[width=\wid,keepaspectratio]{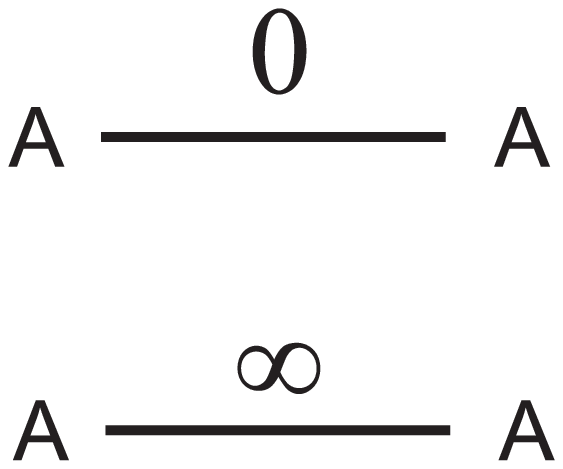}} & \hf{1} \\
\hline
\hf{$\delta^{\prime\prime}_{28}, \delta_{32}$} & \hf{2} & \hf{${p_{CC}} \cup S_E^1$}& \hf{\includegraphics[width=\wid,keepaspectratio]{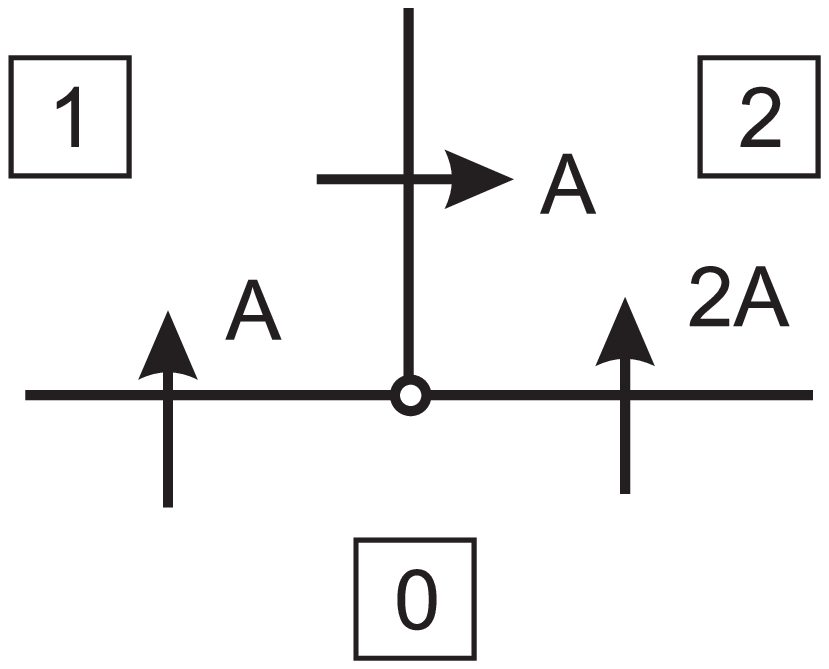}} & \hf{\includegraphics[width=\wid,keepaspectratio]{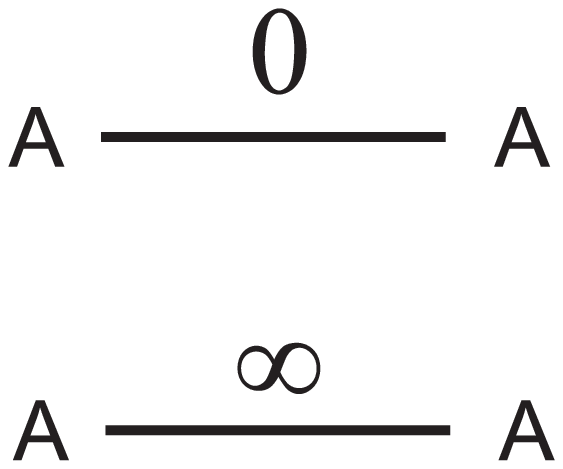}} & \hf{0} \\
\hline
\hf{$\delta_{24},\delta^{\prime}_{28}$} & \hf{4} & \hf{${p_{CC}}\cup 2S_E^1$} & \hf{\includegraphics[width=\wid,keepaspectratio]{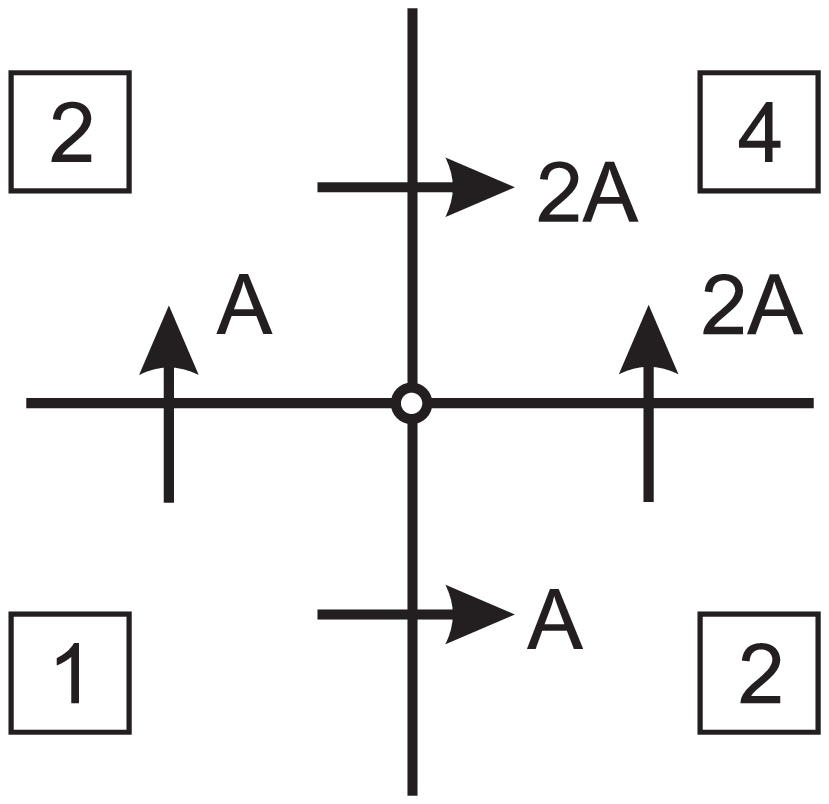}} & \hf{\includegraphics[width=\wid,keepaspectratio]{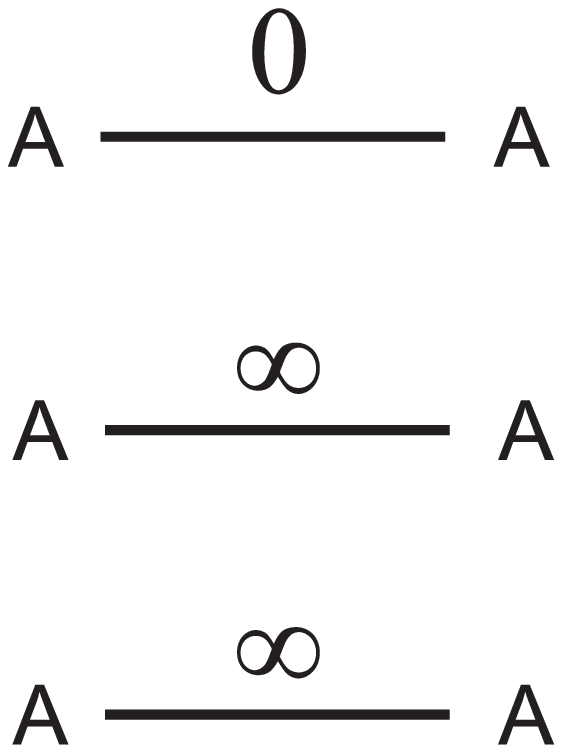}} & \hf{1}\\
\hline
\hf{$\delta_{23},\delta^{\prime}_{27}$} & \hf{2} & \hf{${p_{CS}}\cup S_H^1$} & \hf{\includegraphics[width=\wid,keepaspectratio]{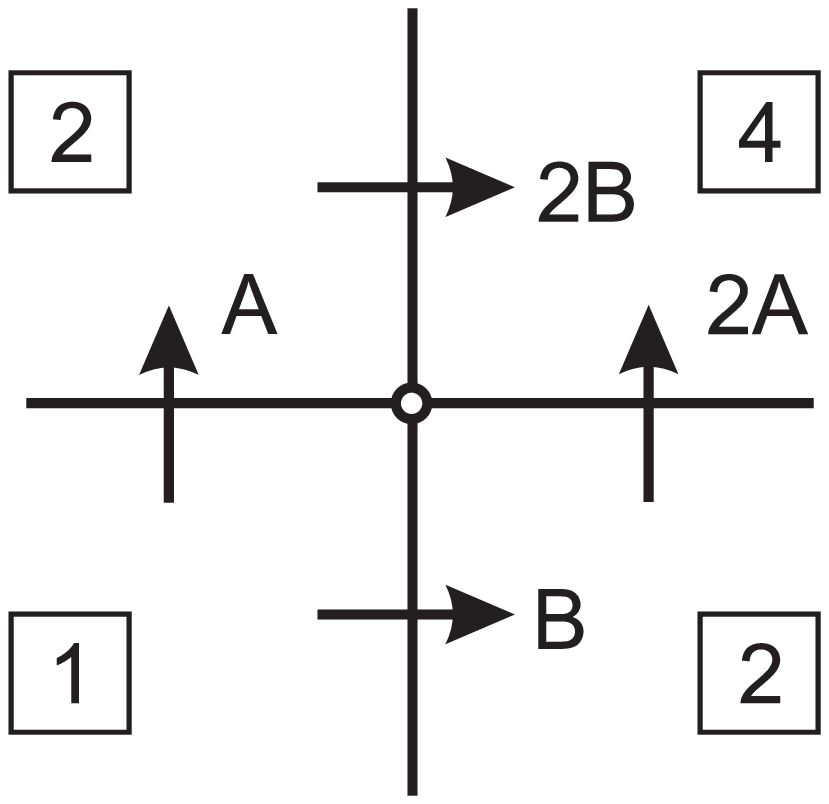}} & \hf{\includegraphics[width=\wid,keepaspectratio]{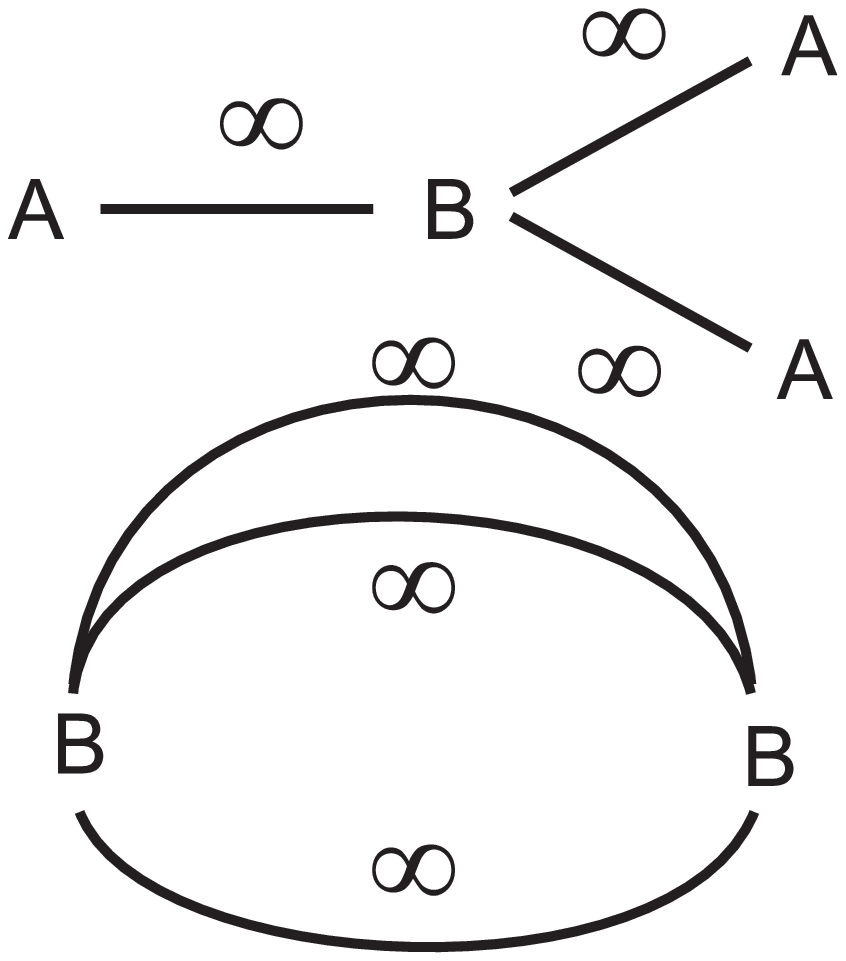}} & \hf{0}\\
\hline
\hf{$\delta_{31},\delta^{\prime\prime}_{27}$} & \hf{1} & \hf{${p_{CS}}$} & \hf{\includegraphics[width=\wid,keepaspectratio]{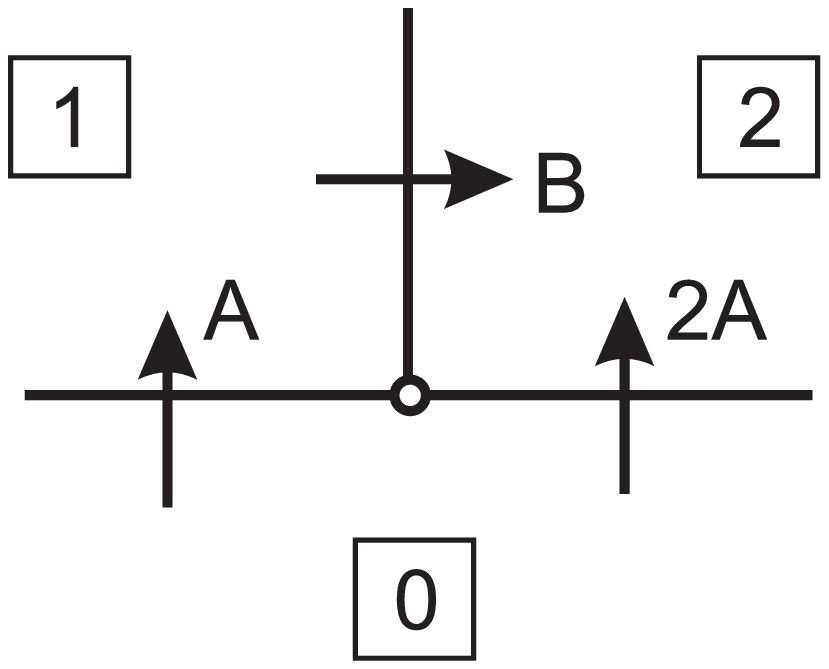}} & \hf{\includegraphics[width=\wid,keepaspectratio]{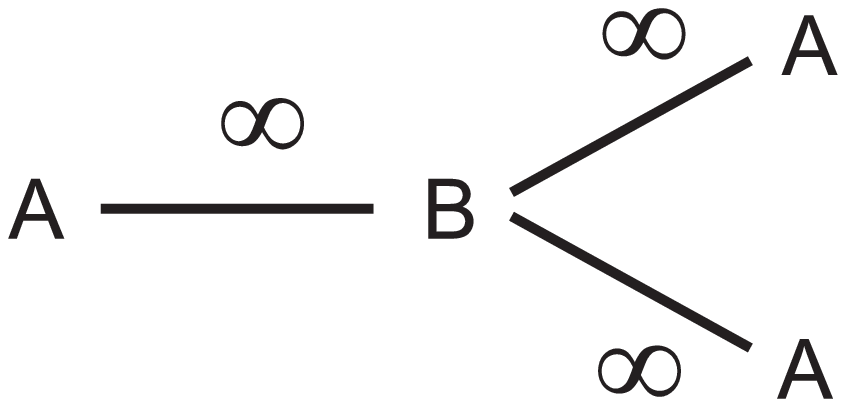}} & \hf{0}\\
\hline
\hf{$\delta_{22},\delta_{26}$} & \hf{1} & \hf{${p_{SS}}$} & \hf{\includegraphics[width=\wid,keepaspectratio]{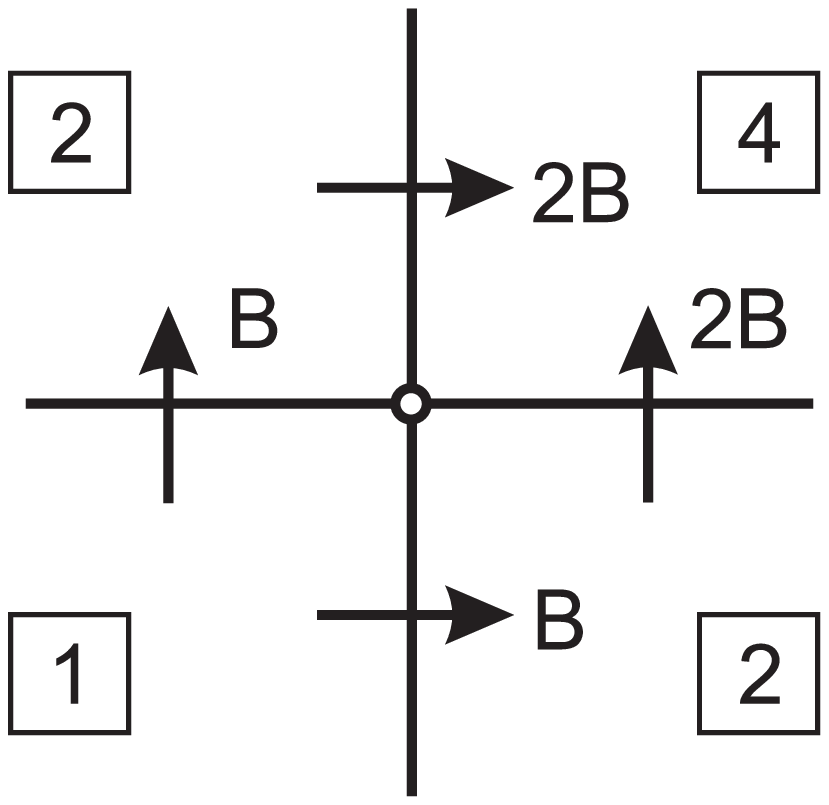}} &
\hf{\includegraphics[width=\wid,keepaspectratio]{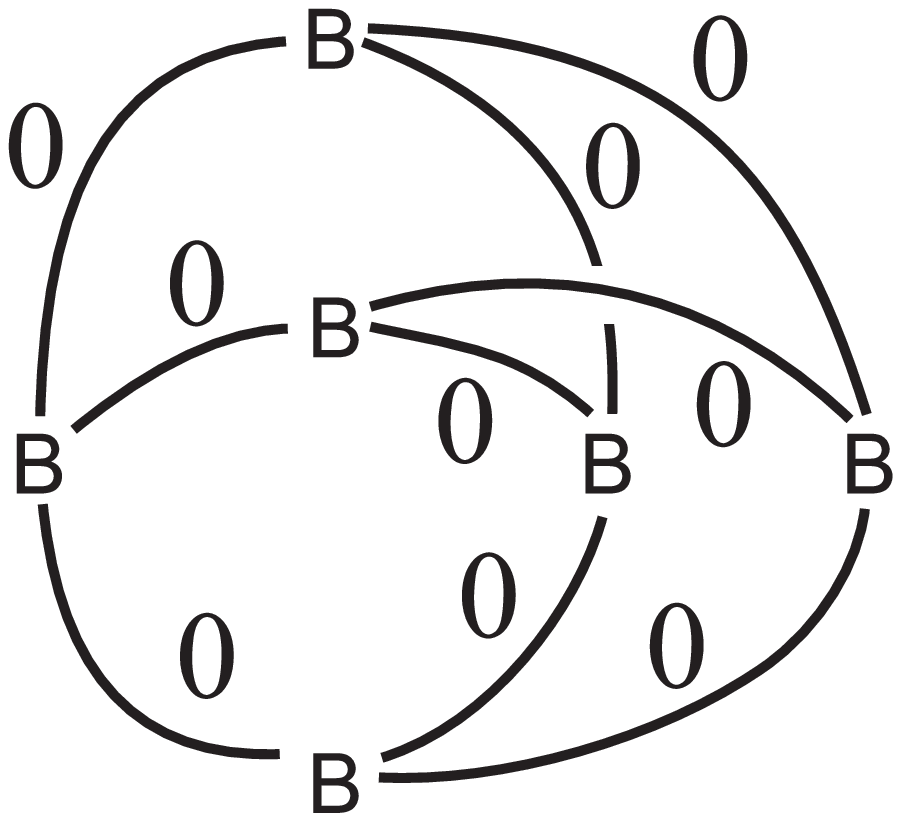}} & \hf{0}\\
\hline

\hf{$\delta_{21}$} & \hf{1} & \hf{${p_{SS}}$} & \hf{\includegraphics[width=\wid,keepaspectratio]{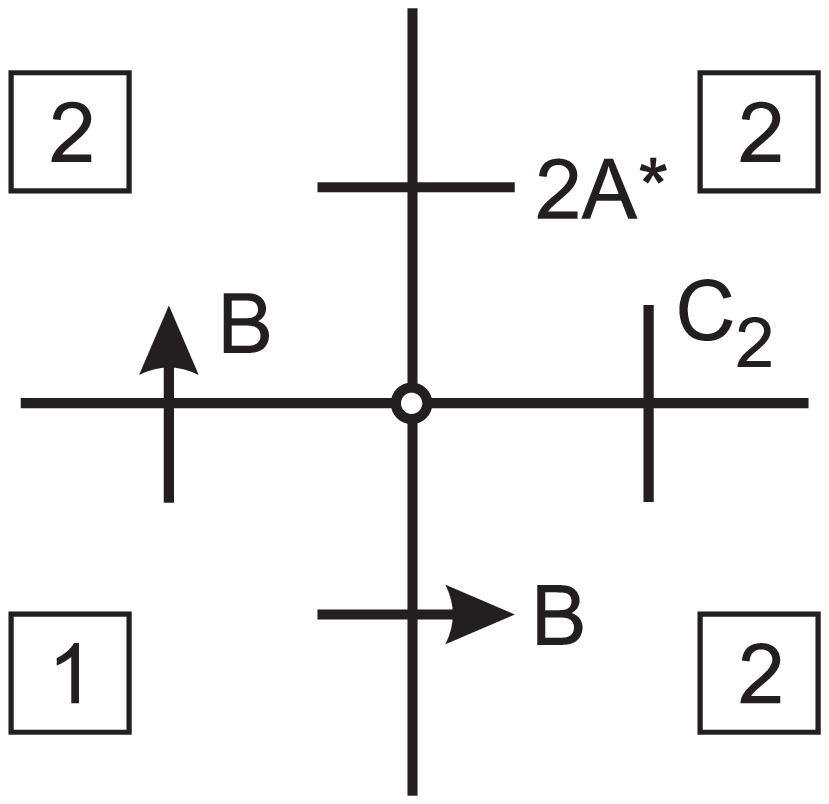}} &
\hf{\includegraphics[width=\wid,keepaspectratio]{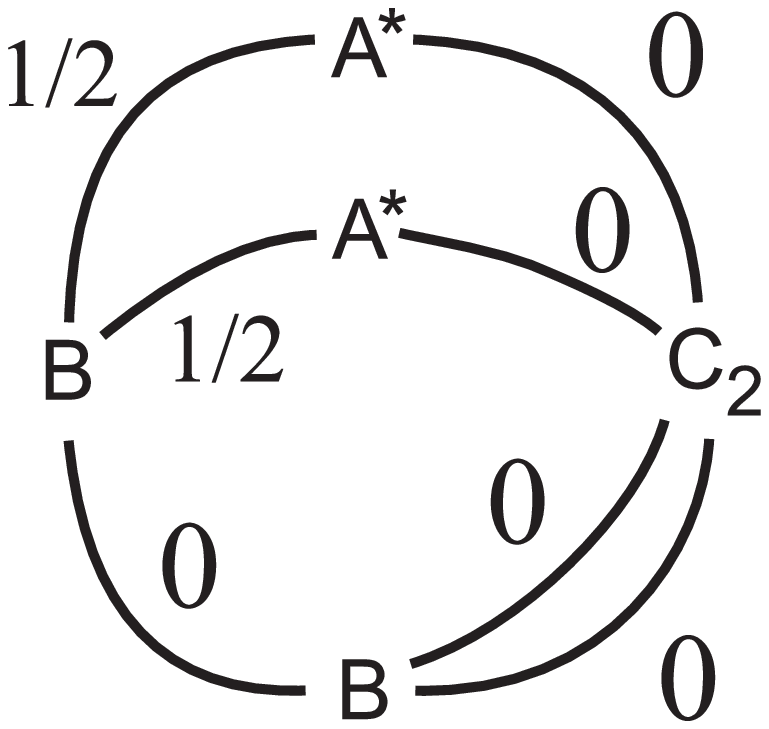}} & \hf{0}\\
\hline
\end{longtable}
\end{center}

}


Результат топологической классификации приведен в табл.~\ref{table42}. Для круговых молекул указаны только $r$-метки с единственной целью -- отличить молекулы относительных равновесий типа ``центр-центр'' (метка $r=0$) от молекул лежащих на том же уровне эллиптических периодических траекторий (метка $r=\infty$). В действительности же, здесь все метки (включая $\varepsilon, n$-метки) выставляются автоматически в соответствии с \cite{BolFom}.

На фрагментах бифуркационной диаграммы отображения $\jpriv_\ell$ в окрестности относительного равновесия указаны атомы, возникающие на критических окружностях при пересечении дуг диаграммы. Здесь встречаются лишь атомы типов $A,B,A^*,C_2$. Для несимметричных атомов стрелкой указано направление возрастания числа торов. Само это число в регулярных областях указано в рамке. При описании особых траекторий (сингулярной компоненты интегрального многообразия) ${p_{CC}}$ -- это неподвижная точка типа ``центр-центр'', ${p_{CS}}$~-- критическое многообразие неподвижной точки типа ``центр-седло'' -- восьмерка, ${p_{SS}}$~-- критическое многообразие неподвижной точки типа ``седло-седло'' -- две восьмерки с общей центральной точкой и приклеенные к ним четыре прямоугольника, заполненных асимптотическими траекториями из регулярных точек (правило склейки полностью определено соответствующей круговой молекулой). Через $S^1_E$ обозначена периодическая траектория эллиптического типа, исчерпывающая соответствующую связную компоненту, а через $S^1_H$ -- поверхность периодической траектории гиперболического типа, отвечающая атому типа $B$ (прямое произведение восьмерки на окружность). Полученная классификация по количеству классов и по виду круговых молекул полного прообраза значения отображения момента в относительных равновесиях отличается от результатов, представленных в \cite{LogRus}. Ввиду отсутствия в цитируемой работе точного определения принципа классификации мы не приводим сопоставления результатов. Достаточно сказать, что после выхода работы \cite{KhMTT42} описание классов точек ранга 0 и их круговых молекул (по-прежнему, без точного определения принципа классификации и без необходимых аналитических обоснований) было приведено в порядок в работе \cite{Slav2Rus}.

В соответствии с информацией из табл.~\ref{table42} и замечанием относительно меток на круговых молекулах, для завершения тонкой классификации всех круговых молекул невырожденных точек ранга 0 необходимо лишь расставить на ребрах найденных молекул номера семейств регулярных торов. Для молекул точек типов ``центр-центр'' и ``центр-седло'' (с учетом и молекул присутствующих на том же интегральном уровне критических окружностей, составленных из невырожденных точек ранга 1) эта задача решается непосредственно из вида прилегающего фрагмента оснащенной бифуркационной диаграммы $\mSell$.

Рассмотрим, для примера, точку из класса $\delta_{23}$, которую можно найти на диаграмме для области 6. На прилегающих сегментах атомы таковы: $B_{[\aaa_3]}$, $2B_{[\aaa_7]}$ и $A_{[\ccc_6]}$, $2A_{[\ccc_8]}$. Здесь и далее используется обозначение атома, к которому приписано в квадратных скобках обозначение ребра диаграммы, на котором происходит данная бифуркация. Так как сама точка ранга 0 имеет тип ``центр-седло'', то собственно ее круговая молекула (то есть связная молекула, отвечающая связной компоненте, содержащей рассматриваемую критическую точку ранга 0) описывает бифуркации
$$
A_{[\ccc_6]}\rightarrow B_{[\aaa_7]} \rightarrow 2A_{[\ccc_8]},
$$
причем в атоме $A_{[\ccc_6]}$ рождается семейство $\tfs{2}$, а в атомах $2A_{[\ccc_8]}$ умирают семейства $\tfs{2}$ и $\tfs{6}$. Таким образом, в атоме $B_{[\aaa_7]}$, принадлежащем связной молекуле точки $\delta_{23}$, семейство $\tfs{2}$ перестраивается в семейства $\tfs{2}$ и $\tfs{6}$. На соответствующую диаграмму накладывается гладкий сегмент $B_{[\aaa_3]} - B_{[\aaa_7]}$, который отвечает невырожденной критической седловой окружности. Его левый атом $B_{[\aaa_3]}$, как видно из фрагмента диаграммы, разделяет камеры $\ts{I}$ и $\ts{IV}$, поэтому в атоме $B_{[\aaa_3]}$ и, следовательно, в том атоме $B_{[\aaa_7]}$, который принадлежит второй компоненте прообраза круговой молекулы (компоненте, не содержащей точку ранга 0), семейство $\tfs{1}$ через критическую окружность перестраивается в семейства $\tfs{1}$ и $\tfs{4}$.

В дальнейшем (для молекулы точек $\delta_{22}$ и $\Delta_{13}$) нам особо понадобится информация о перестройке семейств в атомах $B_{[\aaa_7]}$, в связи с чем выделим это отдельным предложением.

\begin{proposition}\label{propfora7}
В двух атомах $B$ на ребре $\aaa_7$ имеем следующие бифуркации: на одном атоме тор из семейства $\tfs{1}$ перестраивается в два тора из семейств $\tfs{1}$, $\tfs{4}$, на другом -- тор из семейства $\tfs{2}$ перестраивается в два тора из семейств $\tfs{2}$, $\tfs{6}$.
\end{proposition}

Информация по круговым молекулам в прообразах невырожденных точек типов ``центр-центр'' и ``центр-седло'' собрана в табл.~\ref{table_nev1}. Здесь в первом столбце, для удобства, вместе с обозначением класса точек указан номер такой области на плоскости $(\ld,\ell)$, на диаграммах которой присутствует точка этого класса (но, конечно, такая точка может присутствовать и на диаграммах для других областей -- например, точка класса $\delta_1$ имеется на всех диаграммах).

\def\dwid{40mm}
\begin{center}
\renewcommand{\arraystretch}{0}
\begin{longtable}{|m{10mm}|m{\dwid}||m{10mm}|m{\dwid}|}
\multicolumn{4}{r}{Таблица \myt\label{table_nev1}}\\[2mm]
\hline
\hspace*{-2.5mm}\hf{\renewcommand{\arraystretch}{0.8}{\small \begin{tabular}{c} \ru{10}Класс\\ точек\end{tabular}} }
&
\hf{ \renewcommand{\arraystretch}{0.8}{\small  \begin{tabular}{c} \ru{10}Молекула
\end{tabular}}}
&
\hspace*{-2.5mm}\hf{\renewcommand{\arraystretch}{0.8}{\small \begin{tabular}{c} \ru{10}Класс\\ точек\end{tabular}} }
&
\hf{ \renewcommand{\arraystretch}{0.8}{\small \begin{tabular}{c} \ru{10}Молекула
\end{tabular}}}\\
\hline\endfirsthead%
\multicolumn{4}{r}{Таблица \ref{table_nev1} (продолжение)}\\[2mm]
\hline
\hspace*{-2.5mm}\hf{\renewcommand{\arraystretch}{0.8}{\small \begin{tabular}{c} \ru{10}Класс\\ точек\end{tabular}} }
&
\hf{ \renewcommand{\arraystretch}{0.8}{\small  \begin{tabular}{c} \ru{10}Молекула
\end{tabular}}}
&
\hspace*{-2.5mm}\hf{\renewcommand{\arraystretch}{0.8}{\small \begin{tabular}{c} \ru{10}Класс\\ точек\end{tabular}} }
&
\hf{ \renewcommand{\arraystretch}{0.8}{\small \begin{tabular}{c} \ru{10}Молекула
\end{tabular}}}\\
\hline\endhead
\hf{\begin{tabular}{c}$\delta_1$\\[5mm]$(1)$\end{tabular}} & \hf{\includegraphics[width=\dwid,keepaspectratio]{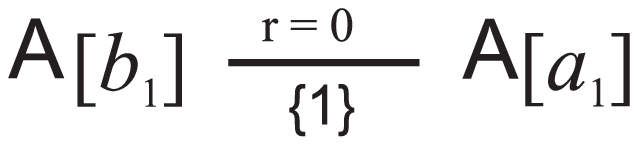}} & \hf{\begin{tabular}{c}$\delta_{25}$\\[5mm]$(1')$\end{tabular}} & \hf{\includegraphics[width=\dwid,keepaspectratio]{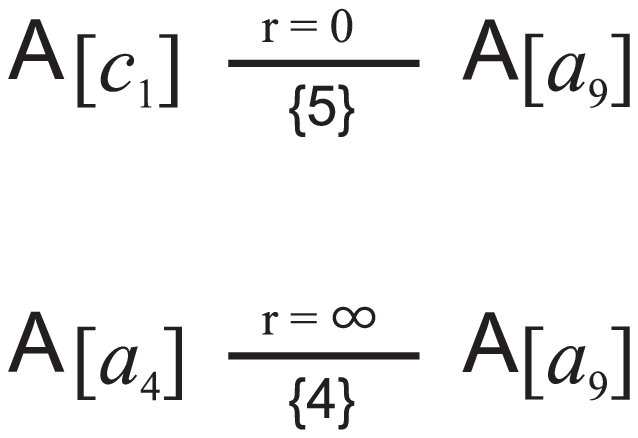}}
\\
\hline
\hf{\begin{tabular}{c}$\delta''_{28}$\\[5mm]$(6')$\end{tabular}} & \hf{\includegraphics[width=\dwid,keepaspectratio]{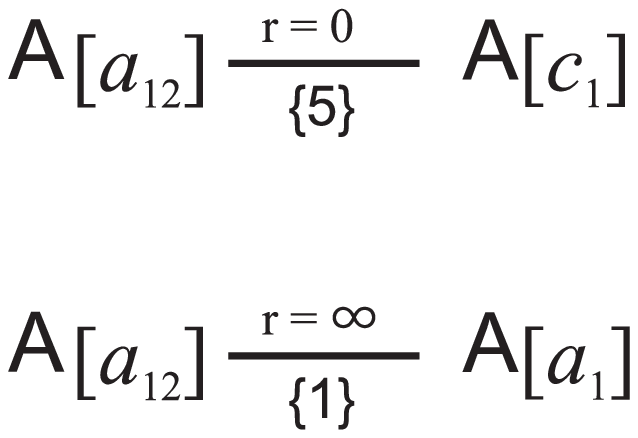}} & \hf{\begin{tabular}{c}$\delta_{32}$\\[5mm]$(6)$\end{tabular}} & \hf{\includegraphics[width=\dwid,keepaspectratio]{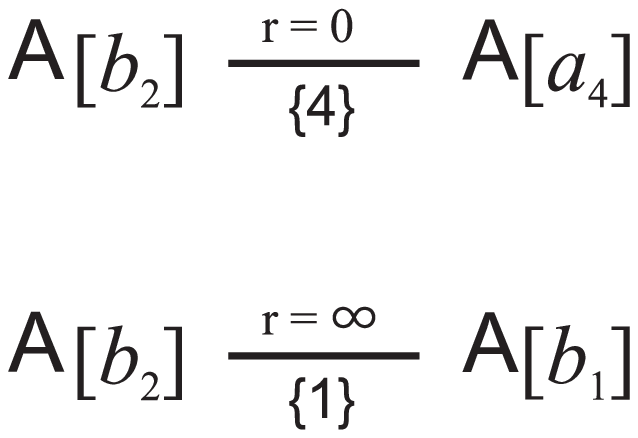}}
\\
\hline
\hf{\begin{tabular}{c}$\delta_{24}$\\[5mm]$(3)$\end{tabular}} & \hf{\includegraphics[width=\dwid,keepaspectratio]{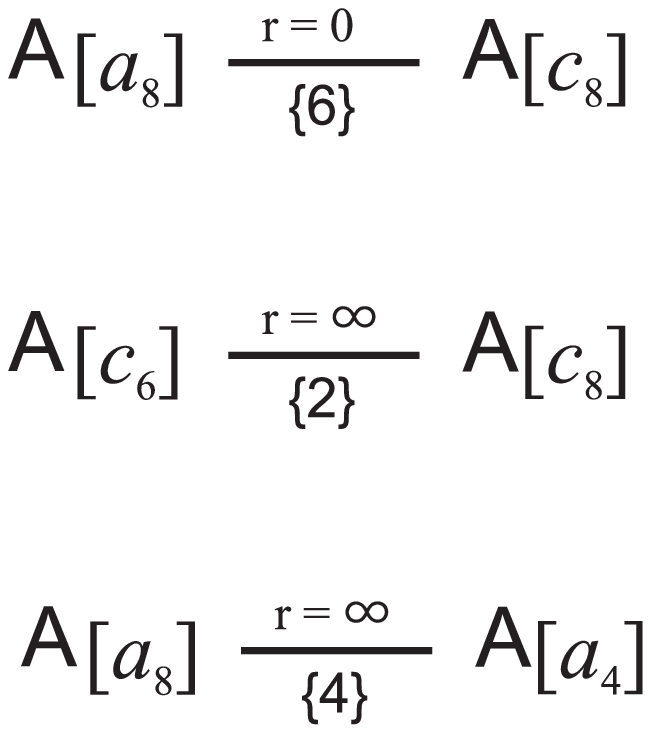}} & \hf{\begin{tabular}{c}$\delta'_{28}$\\[5mm]$(6')$\end{tabular}} & \hf{\includegraphics[width=\dwid,keepaspectratio]{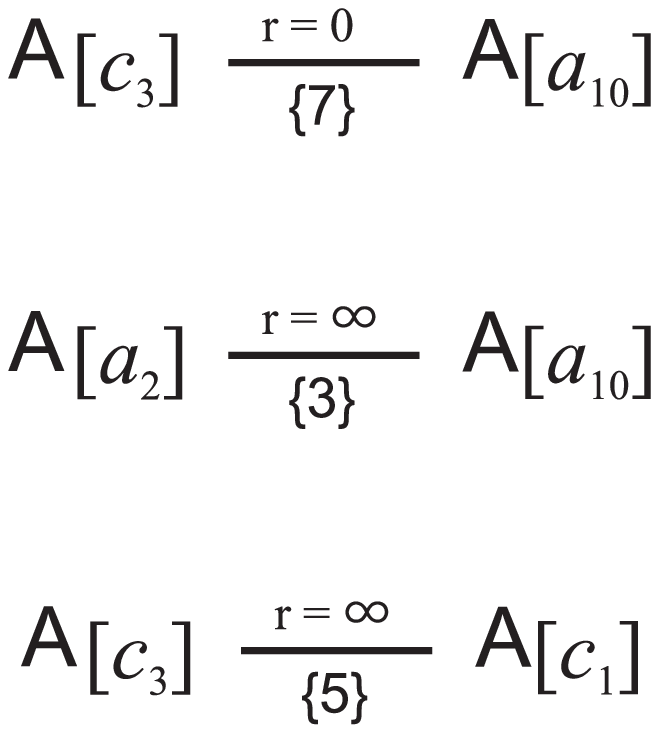}}
\\
\hline
\hf{\begin{tabular}{c}$\delta_{31}$\\[5mm]$(9)$\end{tabular}} & \hf{\includegraphics[width=\dwid,keepaspectratio]{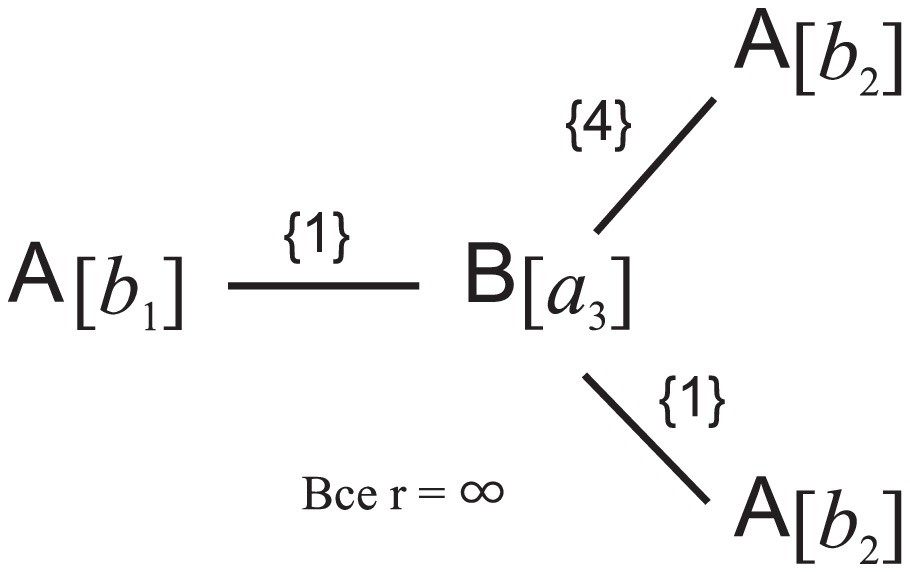}} & \hf{\begin{tabular}{c}$\delta''_{27}$\\[5mm]$(6')$\end{tabular}} & \hf{\includegraphics[width=\dwid,keepaspectratio]{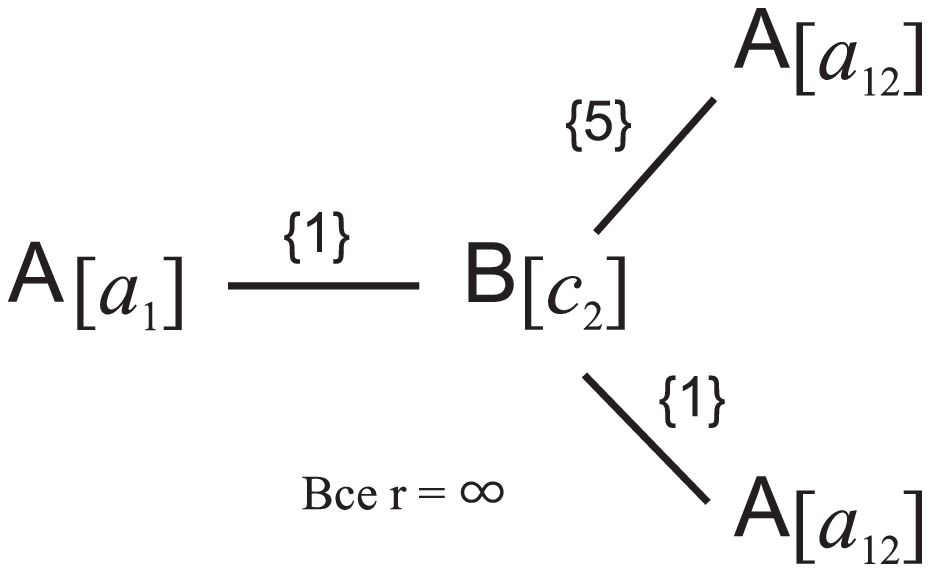}}
\\
\hline
\hf{\begin{tabular}{c}$\delta_{23}$\\[5mm]$(6)$\end{tabular}} & \hf{\includegraphics[width=\dwid,keepaspectratio]{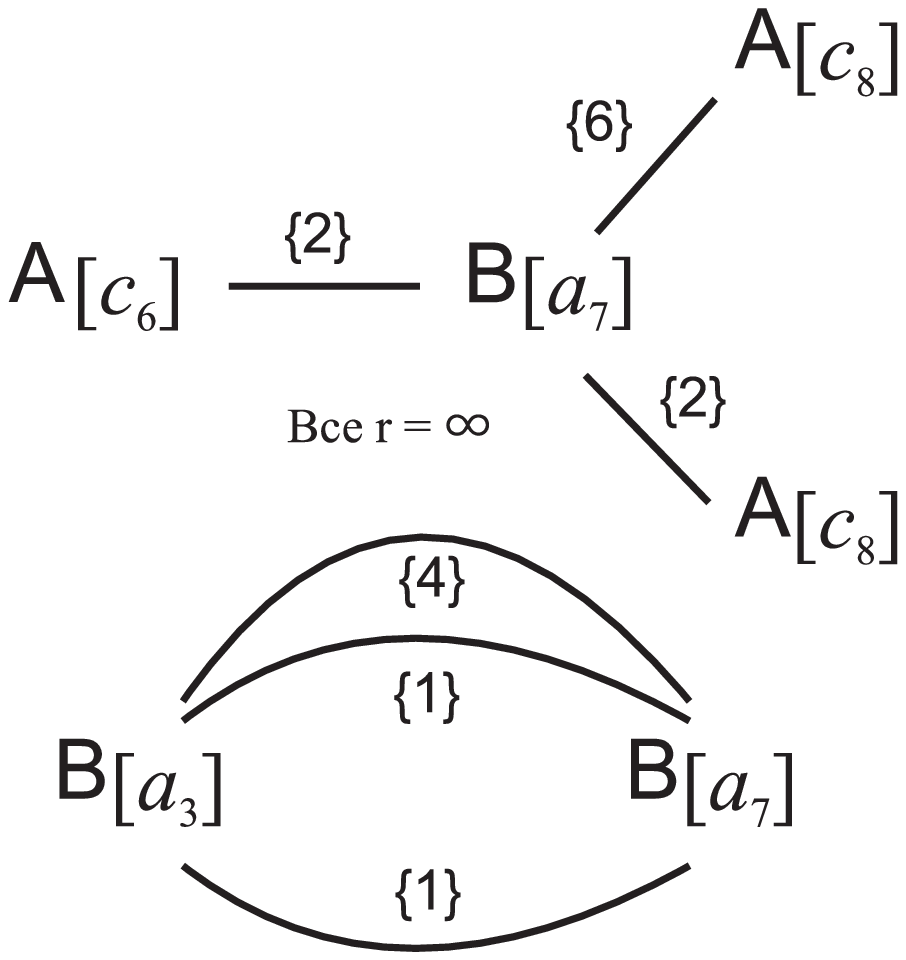}} & \hf{\begin{tabular}{c}$\delta'_{27}$\\[5mm]$(6')$\end{tabular}} & \hf{\includegraphics[width=\dwid,keepaspectratio]{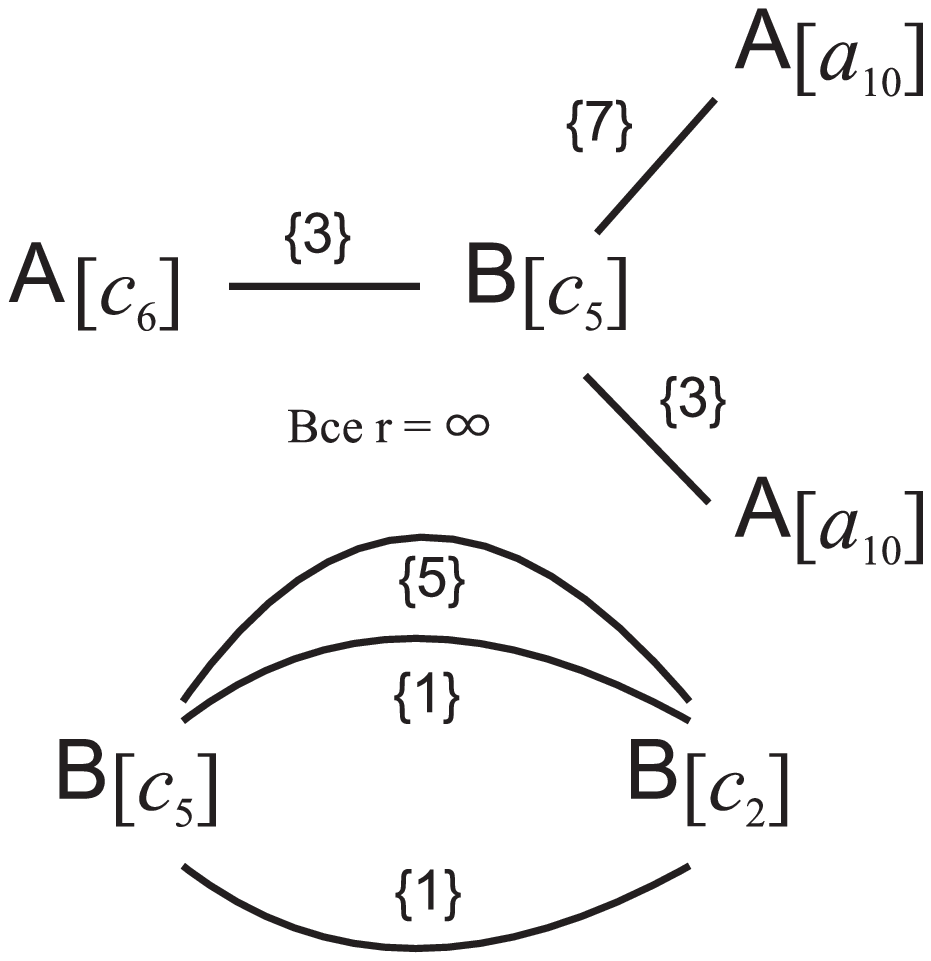}}
\\
\hline
\end{longtable}
\end{center}

Рассмотрим круговые молекулы оставшихся трех классов невырожденных критических точек ранга 0 типа ``седло-седло''.

Возьмем точку $\delta_{22}$. Она присутствует, например, на диаграммах для областей 2 и 3. Уже известно, что на примыкающих ребрах бифуркационной диаграммы перестройки семейств таковы:
$$
\begin{array}{ll}
B_{[\ccc_7]}: & \tfs{1} \mapsto \tfs{1}+\tfs{2};\\
B_{[\aaa_3]}: & \tfs{1} \mapsto \tfs{1}+\tfs{4};\\
2B_{[\aaa_7]}: & \tfs{1} \mapsto \tfs{1}+\tfs{4}, \quad \tfs{2} \mapsto \tfs{2}+\tfs{6}.
\end{array}
$$
Оставшееся ребро $\ccc_9$ разделяет камеры $\ts{IV}$ (с семействами $\tfs{1}+\tfs{4}$) и $\ts{VI}$ (с семействами $\tfs{1}+\tfs{2}+\tfs{4}+\tfs{6}$). Закон расстановки семейств на четырех ребрах, соединяющих $2B_{[\aaa_7]}$ и $2B_{[\ccc_9]}$, без обращения к вырожденным критическим точкам ранга 1 установить не удается. Однако, согласно теореме \ref{theDelta32} мы знаем, что круговая молекула точки $\Delta_{32}$ имеет две связных компоненты, в силу чего каждая из этих компонент имеет известный вид, при котором семейство в ``голове'' вновь приходит в одну из ``ног''. Но точка $\Delta_{32}$ является точкой возврата, в которую приходит гиперболическое ребро $\ccc_9$ и эллиптическое ребро $\ccc_8$ на диаграмме области~$7$. В силу этого сразу получаем, что следует, что в атомах $B_{[\ccc_9]}$ перестройки семейств таковы:
$$
\tfs{1} \mapsto \tfs{1}+\tfs{4}\vee\tfs{6}, \qquad \tfs{4} \mapsto \tfs{2}+\tfs{4}\vee\tfs{6}
$$
(запись $\tfs{4}\vee\tfs{6}$ означает неопределенную пока альтернативу). Таким образом, возникает следующая логическая задача. Имеются два мужчины $M_1, M_2$ (семейства в ``головах'' атомов $2B_{[\aaa_7]}$) и две женщины $F_1, F_4$ (семейства в ``головах'' атомов $2B_{[\ccc_9]}$). Известно, что у каждого мужчины от каждой женщины имеется ровно по одному ребенку и эти дети $D_1,D_2,D_4,D_6$ (семейства на ребрах круговой молекулы точки $\delta_{22}$ между $\aaa_7$ и $\ccc_9$). Известно также, что дети $M_1$ -- это $D_1,D_4$, а дети $M_2$ -- это $D_2,D_6$. Также известно, что матерью $D_1$ является $F_1$, а матерью $D_2$ является $F_4$. Необходимо установить пару родителей (а именно, мать) каждого ребенка. Теперь ответ очевиден: $D_1$ имеет родителей $M_1,F_1$, $D_2$ имеет родителей $M_2,F_4$, $D_4$ имеет родителей $M_1,F_2$, $D_6$ имеет родителей $M_2,F_1$. Бифуркация в атомах $B_{[\ccc_9]}$ важна для дальнейшего, поэтому выделим ее отдельным утверждением.

\begin{proposition}\label{propforc9}
В двух атомах $B_{[\ccc_9]}$ перестройки семейств таковы:
$$
\tfs{1} \mapsto \tfs{1}+\tfs{6}, \qquad \tfs{4} \mapsto \tfs{2}+\tfs{4}.
$$
\end{proposition}
В итоге для точки $\delta_{22}$ получаем первую молекулу в табл.~\ref{table_nev2}.

\def\dwid{35mm}

\begin{center}
\renewcommand{\arraystretch}{0}
\begin{tabular}{|m{10mm}|m{\dwid}||m{10mm}|m{\dwid}||m{10mm}|m{\dwid}|}
\multicolumn{6}{r}{Таблица \myt\label{table_nev2}}\\[2mm]
\hline
\hspace*{-2.5mm}\hf{\renewcommand{\arraystretch}{0.8}{\small \begin{tabular}{c} \ru{10}Класс\\ точек\end{tabular}} }
&
\hf{ \renewcommand{\arraystretch}{0.8}{\small  \begin{tabular}{c} \ru{10} Молекула
\end{tabular}}}
&
\hspace*{-2.5mm}\hf{\renewcommand{\arraystretch}{0.8}{\small \begin{tabular}{c} \ru{10}Класс\\ точек\end{tabular}} }
&
\hf{ \renewcommand{\arraystretch}{0.8}{\small \begin{tabular}{c} \ru{10} Молекула
\end{tabular}}}
&
\hspace*{-2.5mm}\hf{\renewcommand{\arraystretch}{0.8}{\small \begin{tabular}{c} \ru{10}Класс\\ точек\end{tabular}} }
&
\hf{ \renewcommand{\arraystretch}{0.8}{\small \begin{tabular}{c} \ru{10} Молекула
\end{tabular}}}
\\
\hline
\hf{\begin{tabular}{c}$\delta_{22}$\\[5mm]$(2)$\end{tabular}} & \hf{\includegraphics[width=\dwid,keepaspectratio]{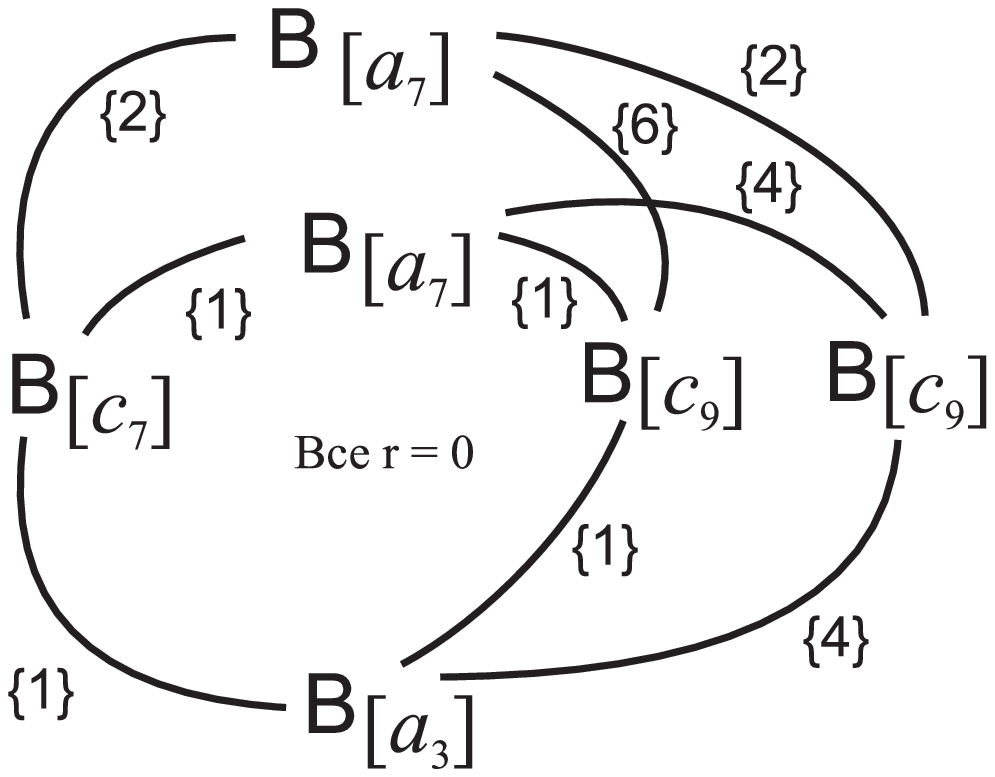}} & \hf{\begin{tabular}{c}$\delta_{26}$\\[5mm]$(2')$\end{tabular}} & \hf{\includegraphics[width=\dwid,keepaspectratio]{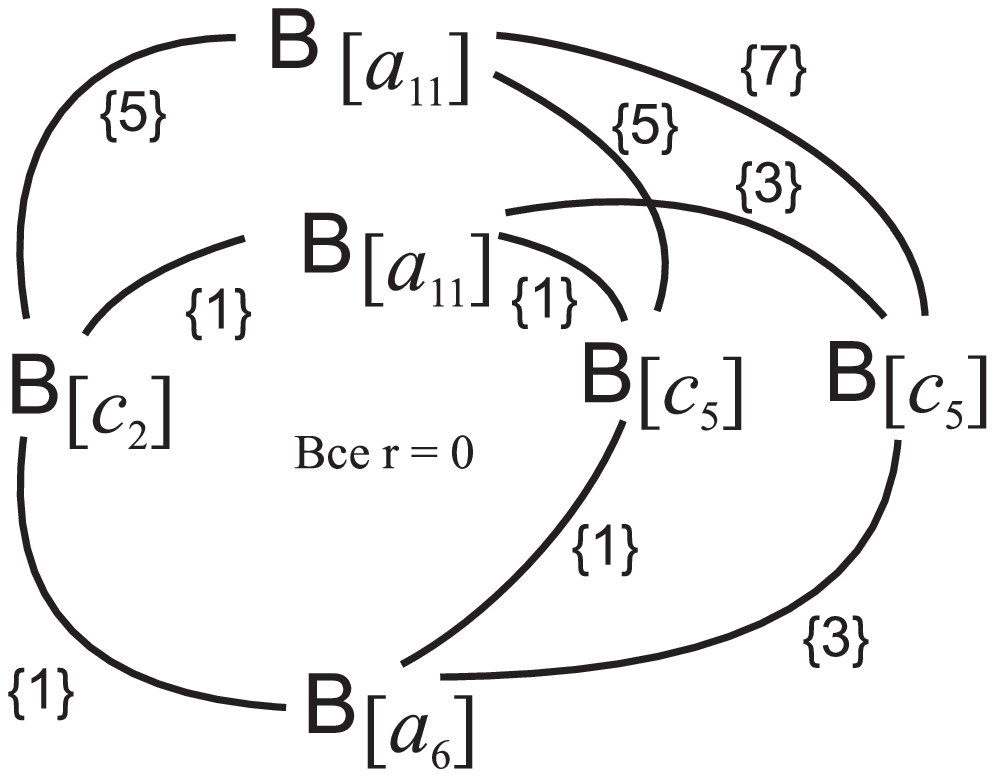}}
& \hf{\begin{tabular}{c}$\delta_{21}$\\[5mm]$(1)$\end{tabular}} & \hf{\includegraphics[width=\dwid,keepaspectratio]{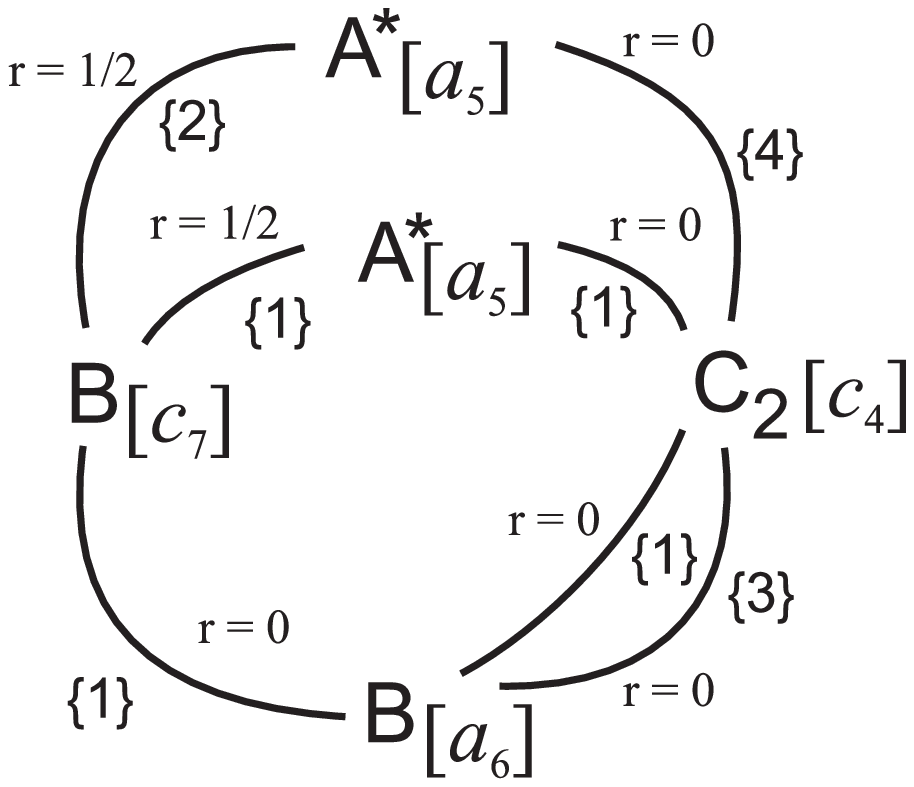}}
\\
\hline
\end{tabular}
\end{center}

Рассмотрим точку $\delta_{26}$. Здесь неопределен лишь закон расстановки семейств на четырех ребрах, соединяющих $2B_{[\aaa_{11}]}$ и $2B_{[\ccc_5]}$. Ребро $\aaa_{11}$ разделяет камеры
$\ts{V}$ и $\ts{VII}$, то есть семейства $\tfs{1}+\tfs{5}$ перестраиваются в семейства $\tfs{1}+\tfs{3}+\tfs{5}+\tfs{7}$. Заметим, что ребро $\aaa_{11}$ имеет при значениях параметров из области $4'$ выход на точку $\Delta_{14}$, в которой согласно теореме \ref{theDelta134} круговая молекула имеет две связных компоненты, так что в ней семейства в ``головах'' атомов $B_{[\aaa_{11}]}$ приходят в ``ногу'' того же атома. В частности, бифуркации семейств в атомах $B_{[\aaa_{11}]}$ таковы:
\begin{equation*}
\tfs{1} \mapsto \tfs{1}+\tfs{3}\vee\tfs{7}, \qquad \tfs{5} \mapsto \tfs{5}+\tfs{3}\vee\tfs{7}.
\end{equation*}
С другой стороны, из круговой молекулы точки $\delta'_{27}$ (см. табл.~\ref{table_nev1}) мы знаем, что на атомах $B_{[\ccc_5]}$ бифуркации таковы:
\begin{equation*}
\tfs{1} \mapsto \tfs{1}+\tfs{5}, \qquad \tfs{3} \mapsto \tfs{3}+\tfs{7}.
\end{equation*}
Так как, в соответствии c установленными атомами на ребрах бифуркационной диаграммы, количеством регулярных торов в примыкающих камерах и общей классификацией \cite{BolFom}, окрестность точки $\delta_{26}$ имеет тип $B{\times}B$, то альтернатива в атомах $B_{[\aaa_{11}]}$ однозначно разрешается.
\begin{proposition}\label{propfora11}
В двух атомах $B_{[\aaa_{11}]}$ перестройки семейств таковы:
$$
\tfs{1} \mapsto \tfs{1}+\tfs{3}, \qquad \tfs{5} \mapsto \tfs{5}+\tfs{7}.
$$
\end{proposition}
В итоге для точки $\delta_{26}$ получаем вторую молекулу в табл.~\ref{table_nev2}.

Для последней невырожденной точки $\delta_{21}$ (диаграмма для области $1$) вопрос состоит в том, как правильно соединить семейства в бифуркации $2A^*_{[\aaa_5]}$. Ребро $\aaa_5$ разделяет камеры $\ts{II}$ и $\ts{IV}$, то есть пара семейства $\tfs{1}+\tfs{2}$ перестраивается в пару $\tfs{1}+\tfs{4}$. Закон соответствия семейств однозначно устанавливается ниже при исследовании круговой молекулы точки $\Delta_{01}$, которая всегда выступает одним из концов ребра $\aaa_5$ (диаграммы для областей $1 - 6$). Сформулируем его пока без доказательства. Оно будет дано на стр.~\pageref{propfora5page}.

\begin{proposition}\label{propfora5}
В двух атомах $A^*_{[\aaa_5]}$ перестройки семейств таковы:
$$
\tfs{1} \mapsto \tfs{1}, \qquad \tfs{2} \mapsto \tfs{4}.
$$
\end{proposition}

Снова используем установленные атомы на ребрах бифуркационной диаграммы, количество регулярных торов в примыкающих камерах и общую классификацию \cite{BolFom}, чтобы однозначно установить для точки $\delta_{21}$ третью молекулу в табл.~\ref{table_nev2}.

Этим заканчивается исследование круговых молекул невырожденных критических точек ранга 0 и структуры их полного прообраза при отображении момента. Еще раз подчеркнем, что требует доказательства предложение \ref{propfora5}.


\subsection{Круговые молекулы вырожденных точек ранга 1}
\subsubsection{Точки $\gan$}

\def\pp{0.55}
\def\pph{0.34}
\def\mpp{0.7}

Построим круговые молекулы для вырожденных замкнутых траекторий в прообразе множества $\gan$.

Рассмотрим точку $\Delta_{01}$, взяв, например, диаграмму для области 1. Атомы на ребрах диаграммы в окрестности этой точки указаны на рис.~\ref{fig_reg01both}, семейства торов в камерах -- на рис.~\ref{fig_define_fams}. Вся эта информация собрана на рис.~\ref{fig_delta01},$(a)$.

\begin{figure}[htp]
\centering
\includegraphics[width=\mpp\textwidth, keepaspectratio]{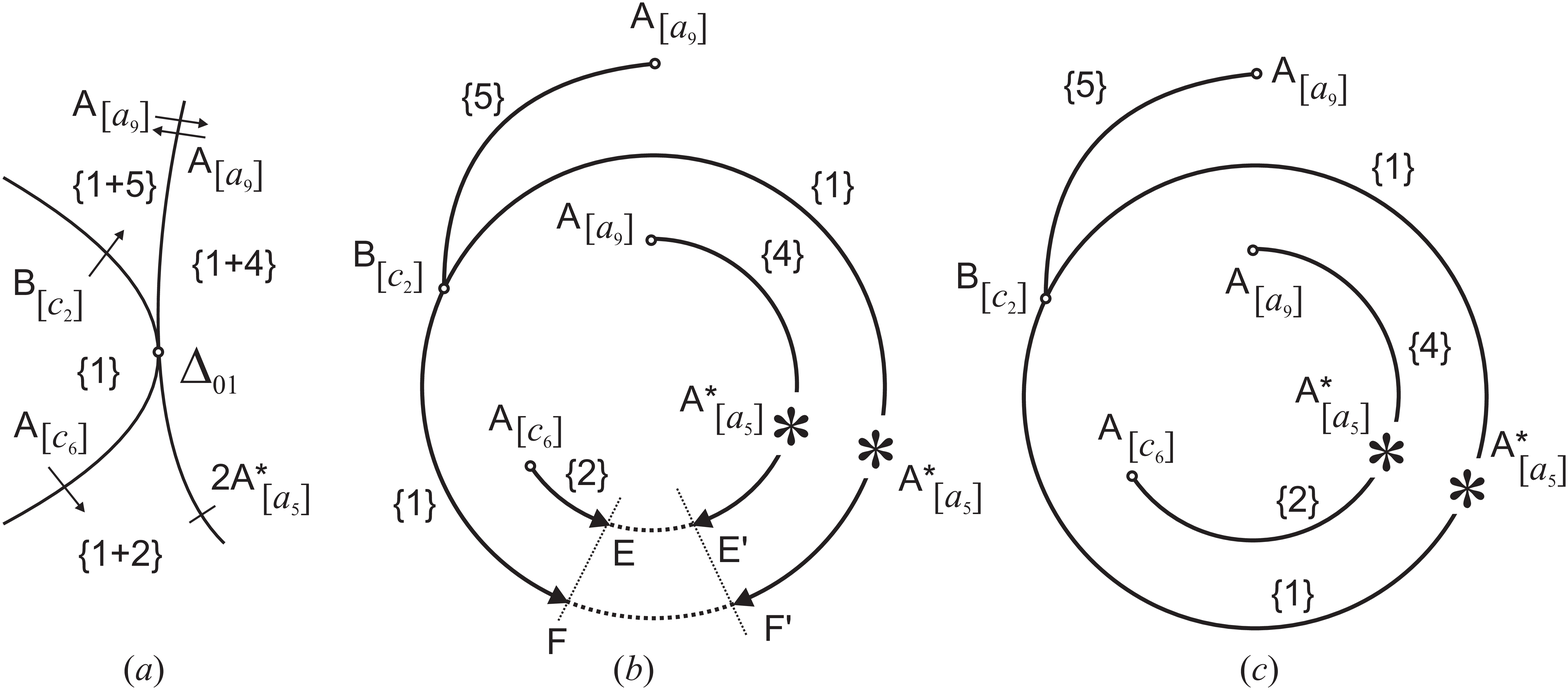}
\caption{Построение круговой молекулы для $\Delta_{01}$.}\label{fig_delta01}
\end{figure}

\begin{figure}[htp]
\centering
\includegraphics[width=\pp\textwidth, keepaspectratio]{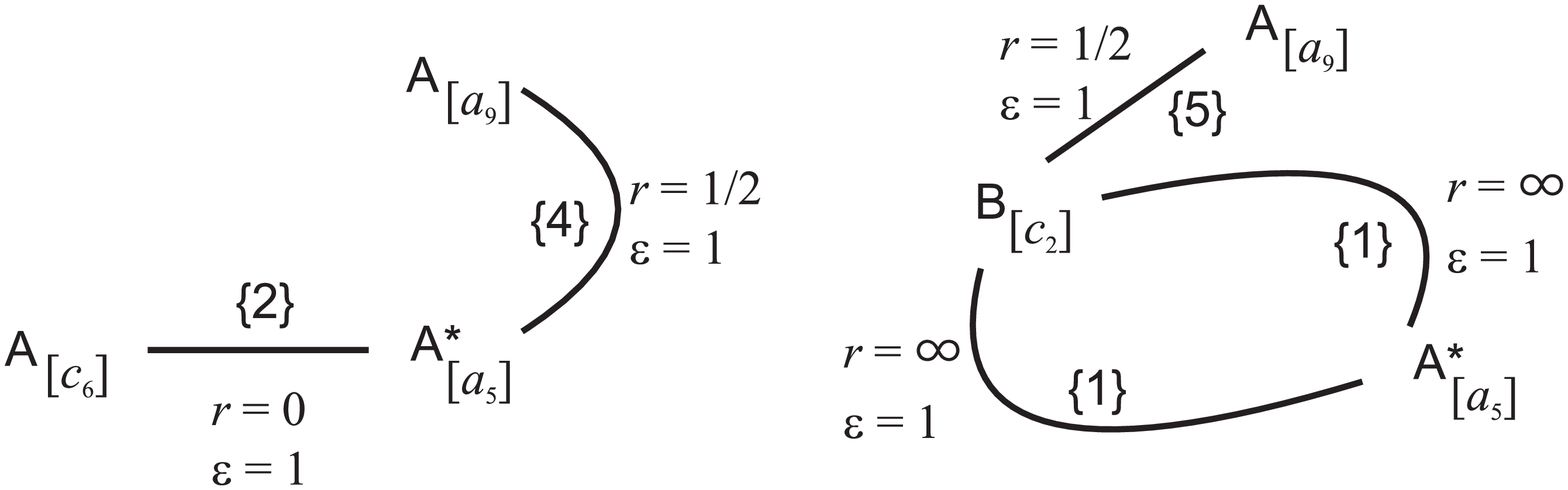}
\caption{Круговая молекула для $\Delta_{01}$.}\label{fig_delta01mols}
\end{figure}

Граф молекулы с вершинами-атомами и ребрами-семействами восстанавливается по этой информации лишь до последнего момента -- правила склейки семейств в двух атомах $A^*$. Это показано на рис.~\ref{fig_delta01},$(b)$. Однако далее видно, что есть лишь два варианта склейки, и если склеить точки $E$ с $F'$, а $F$ с $E'$, то молекула получится связной, что невозможно, так как по теореме \ref{theo_twocomps0} поверхность $J_{01}$ состоит из двух компонент. Поэтому необходимо склеить точку $E$ с $E'$, а точку $F$ с $F'$. Тогда получится молекула из двух компонент, показанная на рис.~\ref{fig_delta01},$(c)$. Этим, в частности, доказано анонсированное ранее предложение \ref{propfora5}.\label{propfora5page}

Поскольку теперь утверждение работы \cite{Mor} о грубой структуре молекулы строго доказано, можно воспользоваться и вычисленными в \cite{Mor} метками (точка $\Delta_{01}$ соответствует точке $z_3$ работы \cite{Mor}). Окончательный результат для точки $\Delta_{01}$ показан на рис.~\ref{fig_delta01mols}. Имеется различие с работой \cite{Mor} в указанных склейках семейств. Так, во второй компоненте в атоме $A^*$ сходятся торы одного и того же семейства. Это нельзя было утверждать, оставаясь в пределах многообразия $\ell=0$, которое изучалось в \cite{Mor}. Здесь же семейства трактуются в смысле расширенного фазового пространства. Поэтому два различных семейства торов в однопараметрическом множестве систем с двумя степенями свободы $\ell=0, \ld \in \bR$ оказались одним и тем же семейством в двухпараметрическом множестве таких систем.

Рассмотрим точку $\Delta_{02}$ на диаграмме для области 1. Как и в предыдущем случае, атомы на ребрах диаграммы в окрестности этой точки указаны на рис.~\ref{fig_reg01both}, семейства торов в камерах -- на рис.~\ref{fig_define_fams}. Вся эта информация собрана на рис.~\ref{fig_delta02},$(a)$. Здесь грубая молекула и семейства на ребрах восстанавливаются однозначно (см. рис.~\ref{fig_delta02},$(b)$). Метки на ребрах можно найти в работе \cite{BRF} и книге \cite{BolFom} (в этих источниках приводятся только $r$-метки, однако $\varepsilon$-метки все равны $+1$, согласно соответствующему результату работы \cite{Mo2004}). Полученная молекула показана на рис.~\ref{fig_delta02mols}. Несовпадение номеров семейств с работой \cite{BRF} объясняется той же причиной, что и в предыдущем случае -- невозможностью отождествления некоторых семейств в рамках частного случая (здесь частный случай -- это $\ld=0$). Далее аналогичные замечания о номерах семейств при сопоставлении молекул уже делать не будем.

\begin{figure}[ht]
\centering
\includegraphics[width=\pp\textwidth, keepaspectratio]{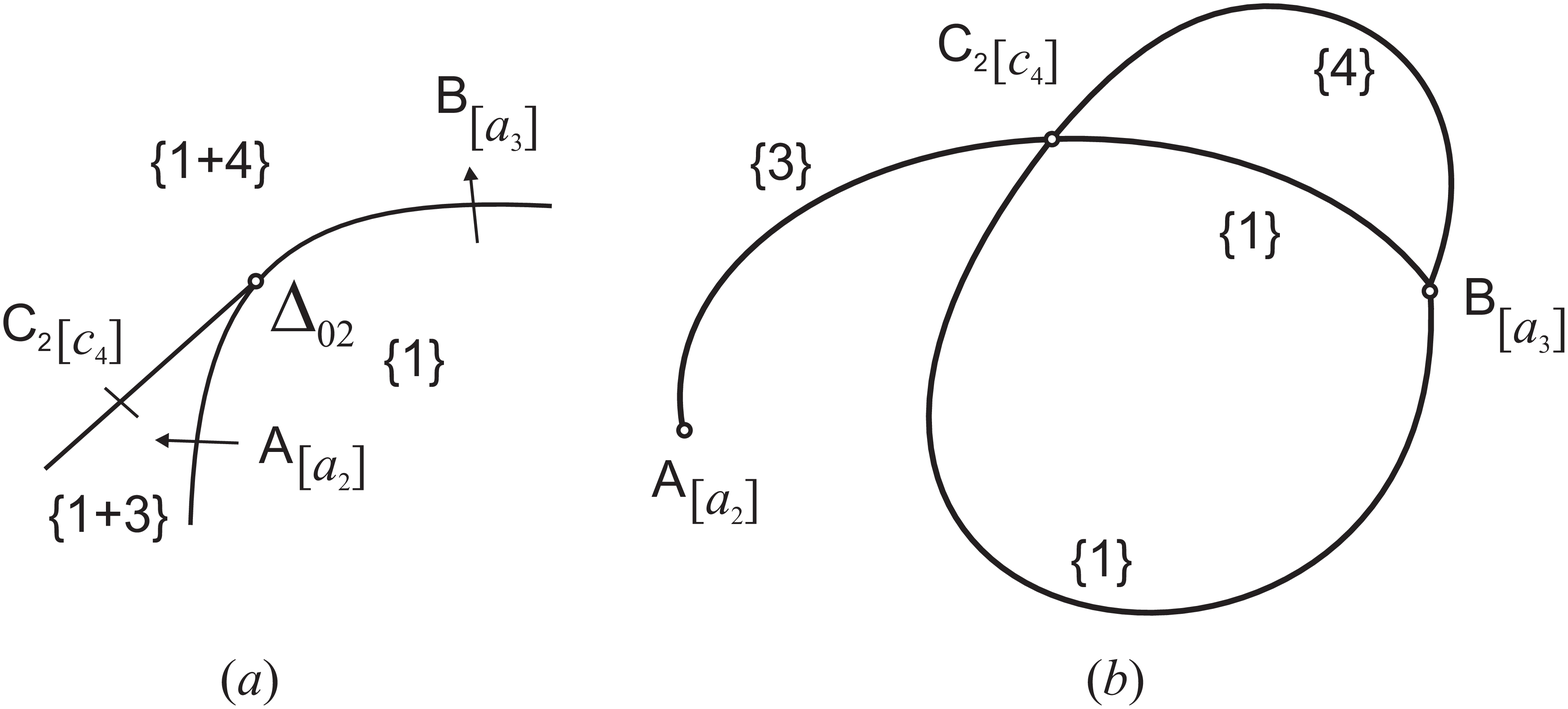}
\caption{Построение круговой молекулы для $\Delta_{02}$.}\label{fig_delta02}
\end{figure}

\begin{figure}[ht]
\centering
\includegraphics[width=\pph\textwidth, keepaspectratio]{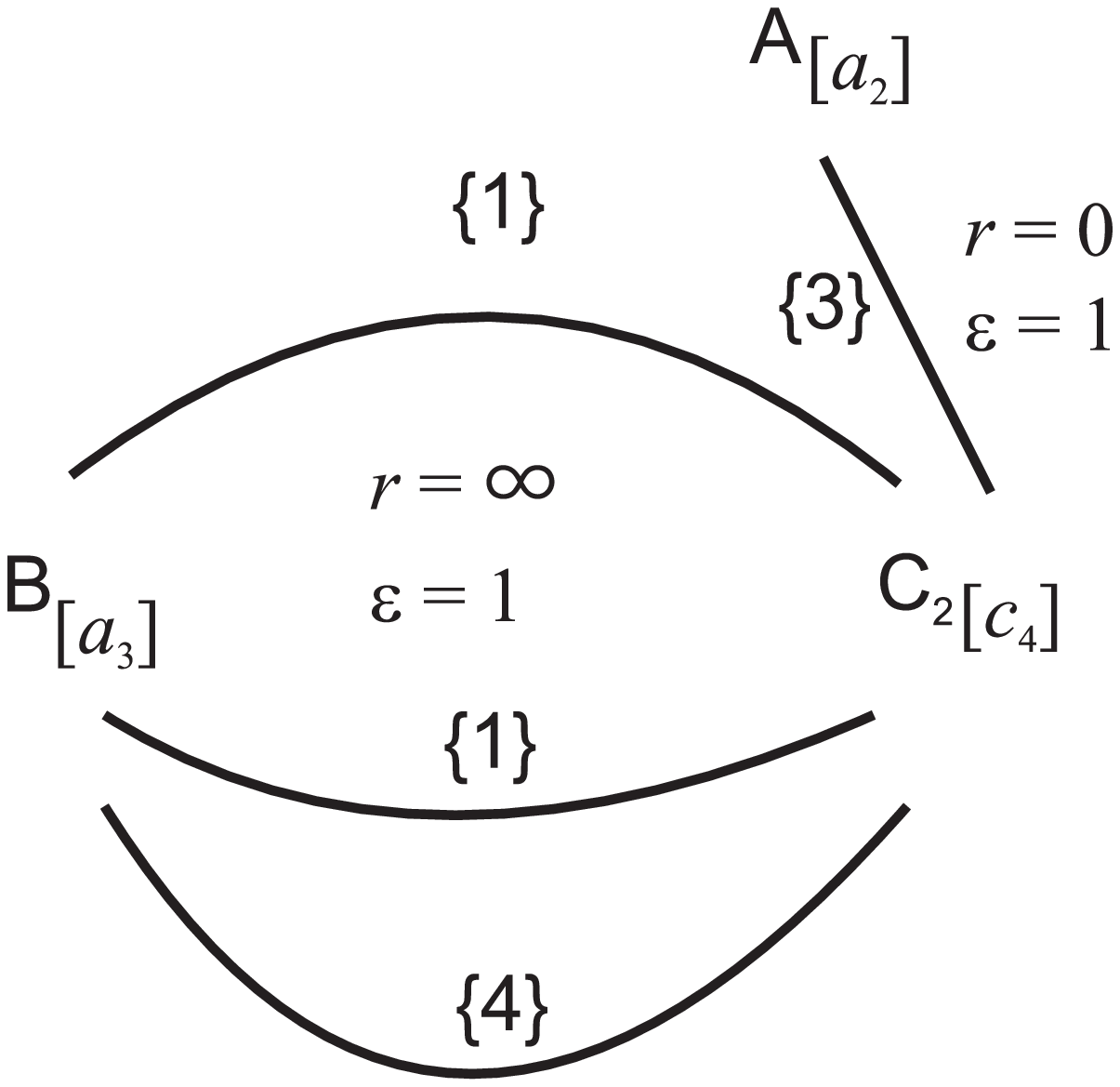}
\caption{Круговая молекула для $\Delta_{02}$.}\label{fig_delta02mols}
\end{figure}

Рассмотрим точку $\Delta_{03}$. Ее можно найти на диаграмме для области 2.
Атомы на ребрах диаграммы в окрестности этой точки указаны на рис.~\ref{fig_reg01both}, камеры и количество торов в них  -- на рис.~\ref{fig_define_tori}, и, соответственно, семейства представлены в табл.~\ref{tabfam}. Вся информация собрана на рис.~\ref{fig_delta03},$(a)$. Склейки семейств в атомах $A^*$ ребра $\aaa_5$ описываются предложением \ref{propfora5}, доказанным при построении молекулы для точки $\Delta_{01}$. Семейства $\tfs{4}$ и $\tfs{6}$, возникающие на атомах $A$ ребра $\aaa_8$, приходят в атомы $B$ ребра $\ccc_9$ в соответствии с предложением \ref{propforc9}. Выбор в альтернативе склейки пар точек $E,F$ и $E',F'$ определяется однозначно: точка $E$ не может быть соединена с $F'$, поскольку при этом возникает связная молекула, а по теореме \ref{theo_twocomps0} она состоит из двух компонент. Таким образом, на этом этапе мы имеем грубую молекулу, показанную на рис.~\ref{fig_delta03},$(c)$. Меченая молекула показана на рис.~\ref{fig_delta03mols}. Поскольку точка $\Delta_{03}$ существует и при $\ld=0$, метки на ребрах молекулы могут быть взяты из работы \cite{BRF} с учетом утверждения \cite{Mo2004} об $\varepsilon$-метках.

\begin{figure}[!ht]
\centering
\includegraphics[width=\mpp\textwidth, keepaspectratio]{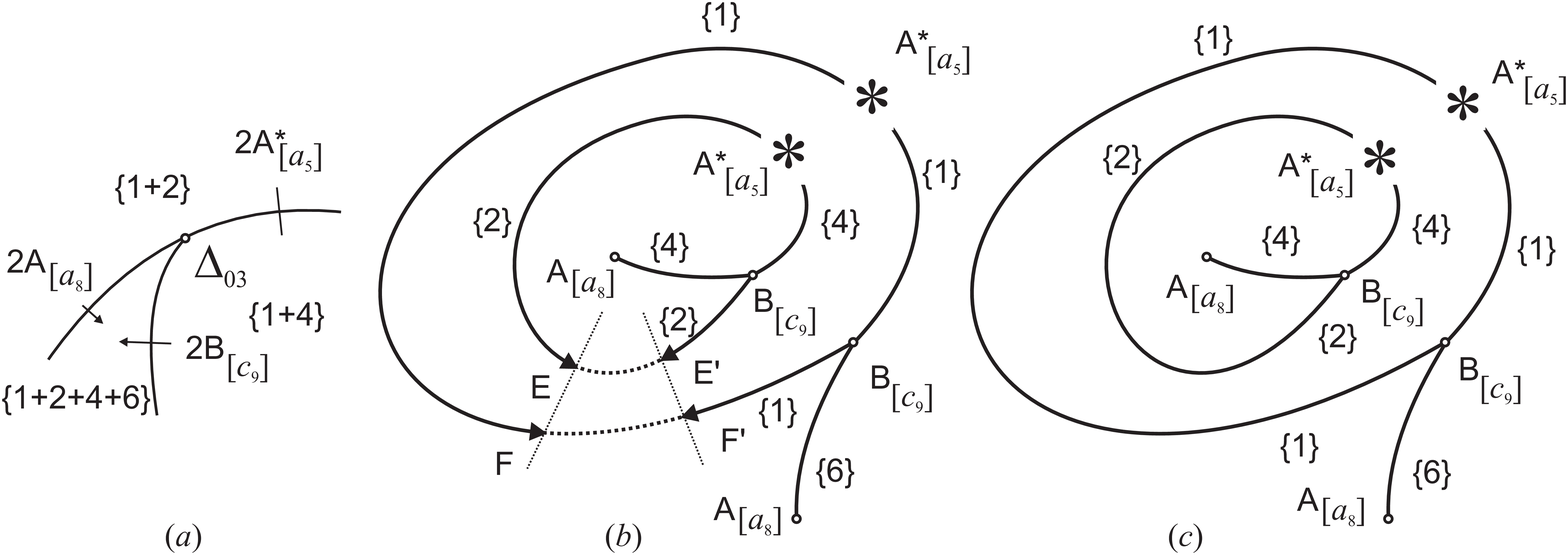}
\caption{Построение круговой молекулы для $\Delta_{03}$.}\label{fig_delta03}
\end{figure}

\begin{figure}[!ht]
\centering
\includegraphics[width=\pp\textwidth, keepaspectratio]{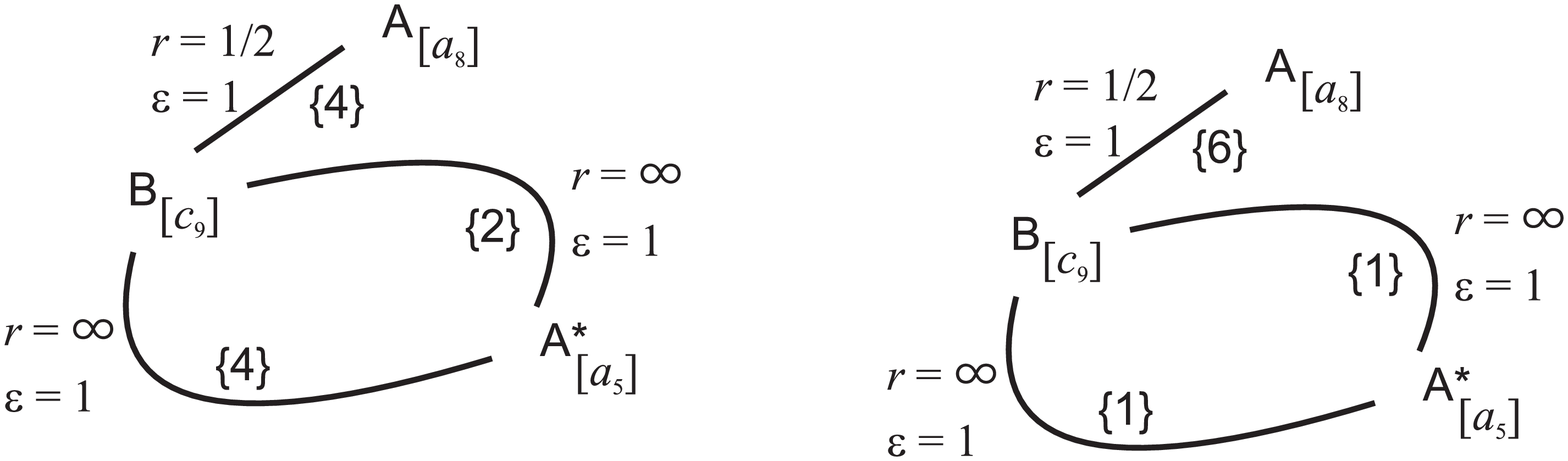}
\caption{Круговая молекула для $\Delta_{03}$.}\label{fig_delta03mols}
\end{figure}

Рассмотрим последнюю из точек на $\gan$ -- точку $\Delta_{04}$. Ее можно найти на диаграмме для области $2'$. Атомы на ребрах диаграммы в окрестности этой точки указаны на рис.~\ref{fig_reg01sboth}, камеры и количество торов в них  -- на рис.~\ref{fig_define_tori}, а семейства, отвечающие камерам, представлены в табл.~\ref{tabfam}. Вся информация собрана на рис.~\ref{fig_delta04},$(a)$. Здесь единственная возможная неоднозначность связана с перестройкой семейств на двух атомах $B$ ребра $\aaa_{11}$. Ранее мы доказали предложение \ref{propfora11} о бифуркациях семейств на ребре $\aaa_{11}$, прибегнув к прогнозу достаточно простой молекулы точки $\Delta_{14}$. Еще раз докажем это же из других соображений. Если предположить, что хотя бы одно из семейств $\tfs{5}$, $\tfs{7}$, рождающихся на ребре $\ccc_3$ приходит в ``ногу'' того атома $B_{[\aaa_{11}]}$, который перестраивает два тора в один тор семейства $\tfs{1}$, то полученная молекула будет связна, что не так по теореме \ref{theo_twocomps0}. Тогда ни одна из точек $G'$, $H'$ не может склеиться с точкой $E$ или с точкой $F$ (рис.~\ref{fig_delta04},$(b)$). Поэтому единственно возможный результат показан на рис.~\ref{fig_delta04},$(c)$.  Молекула для данной точки существует и при $\ell=0$, поэтому метки на ее ребрах могут быть взяты из работы \cite{Mor} (соответствующая точка в \cite{Mor} обозначена через $z_4$). В результате для точки $\Delta_{04}$ получаем круговую молекулу, представленную на рис.~\ref{fig_delta04mols}.

\begin{figure}[ht]
\centering
\includegraphics[width=\mpp\textwidth, keepaspectratio]{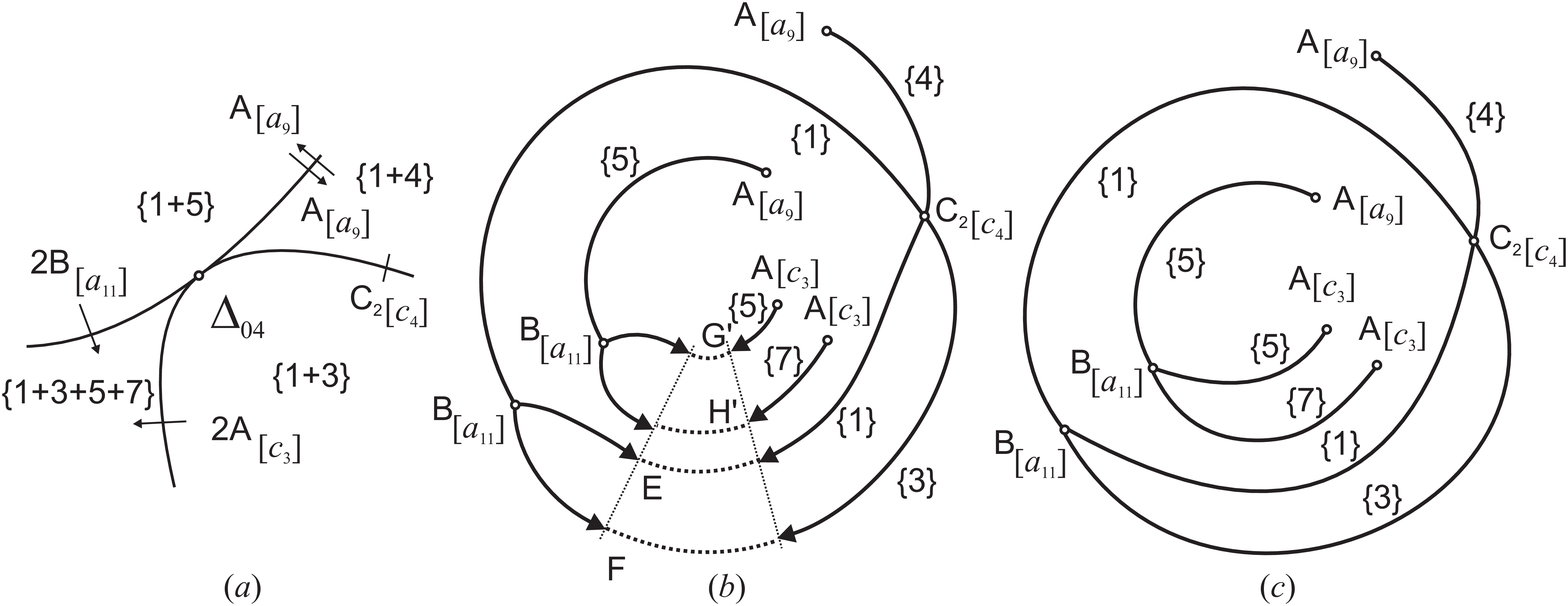}
\caption{Построение круговой молекулы для $\Delta_{04}$.}\label{fig_delta04}
\end{figure}

\begin{figure}[ht]
\centering
\includegraphics[width=\pp\textwidth, keepaspectratio]{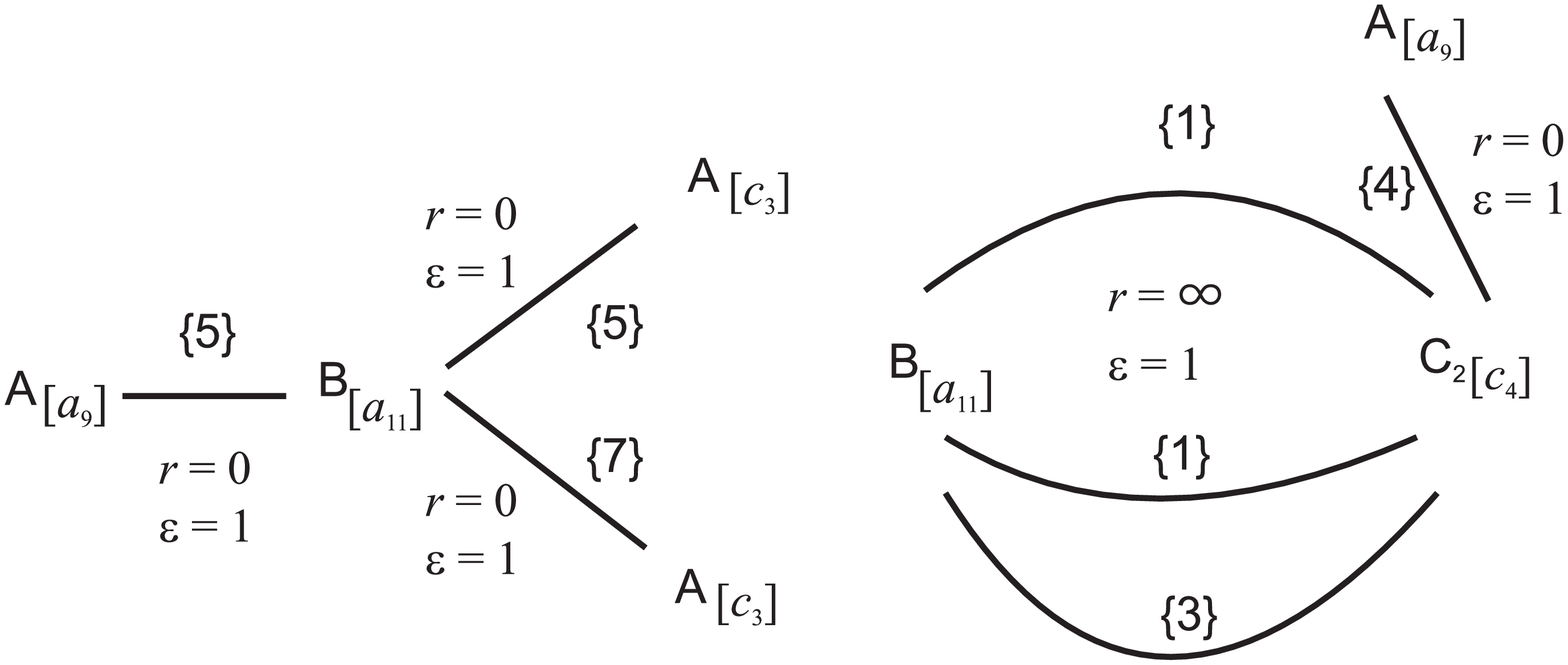}
\caption{Круговая молекула для $\Delta_{04}$.}\label{fig_delta04mols}
\end{figure}

\subsubsection{Точки $\gaa$}
Следующий класс вырожденных точек ранга 1 -- это точки, лежащие в прообразе ребра возврата поверхности $\wsa$. Метки круговой молекулы такого типа для точек возврата на диаграммах описываются утверждением работы \cite{Mor}, а именно, $r$-метки равны $\infty$ на ребре $B - B$ и конечны на ребре $A - B$, все $\varepsilon$-метки равны $+1$. Поэтому для точки $\Delta_{11}$ (см., например, диаграмму для области 1) имеем результат, показанный на рис.~\ref{fig_delta11}. Здесь $r$-метка на ребре $A - B$ равна нулю по доказанному в \cite{BolFom}. Равенство нулю $r$-метки на ребре $A_{[\aaa_4]} - B_{[\aaa_3]}$ для молекулы точки $\Delta_{12}$, показанной на рис.~\ref{fig_delta12} и присутствующей на диаграмме для области~3, необходимо обосновать.
Воспользуемся методом работы \cite{Mor}.

\begin{proposition}\label{propDelta12}
В круговой молекуле точки $\Delta_{12}$ на ребре, соединяющем атом $A$ с седловым, $r$-метка равна нулю.
\end{proposition}

\begin{proof} Используем правило сложения меток \cite{BolFom}: пусть некоторое семейство участвует в бифуркациях на трех стенках $u_1,u_2,u_3$ одной камеры и при этом для перехода $u_1 \to u_2$ метка $r=\infty$, тогда $r$-метки переходов $u_1 \to u_3$ и $u_2 \to u_3$ совпадают.

Для точки $\Delta_{12}$ рассмотрим переход по семейству $\tfs{4}$ от $B_{[\aaa_3]}$ к $A_{[\aaa_4]}$. Обращаясь к бифуркационным диаграммам видим, что эти две стенки могут соседствовать с разными другими стенками. Например, в области 3 имеем переходы $B_{[\aaa_3]} \to 2A_{[\ccc_8]} \to A_{[\aaa_4]}$. Однако в бифуркации на ребре $\ccc_8$ семейство $\tfs{4}$ не участвует. Точка $\Delta_{12}$ присутствует и в области $1'$, однако, здесь нельзя указать разложение нужного перехода в сумму двух. Обратимся к диаграмме для области 9. Здесь самой точки $\Delta_{12}$ нет, но есть переходы по семейству $\tfs{4}$ следующего вида: переход $B_{[\aaa_3]} \to A_{[\bbb_2]}$ -- принадлежит молекуле невырожденной точки $\delta_{31}$ с меткой $r=\infty$; переход $ A_{[\bbb_2] \to  B_{[\aaa_4]}}$ -- принадлежит молекуле невырожденной точки $\delta_{32}$ с меткой $r=0$. Следовательно, недостающая $r$-метка равна 0.
\end{proof}

\begin{figure}[ht]
\centering
\includegraphics[width=0.5\textwidth, keepaspectratio]{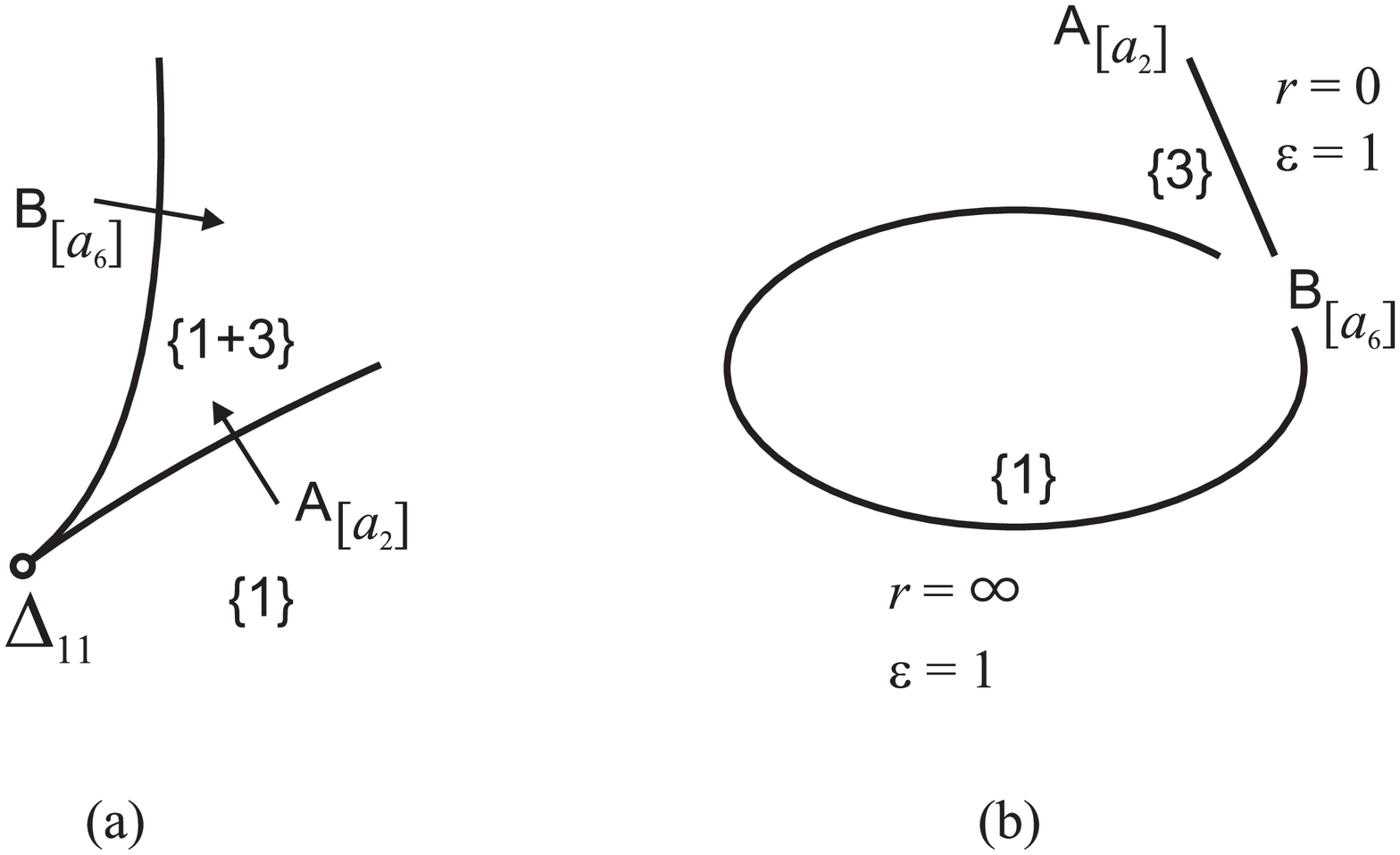}
\caption{Диаграмма и круговая молекула для $\Delta_{11}$.}\label{fig_delta11}
\end{figure}

\begin{figure}[ht]
\centering
\includegraphics[width=0.5\textwidth, keepaspectratio]{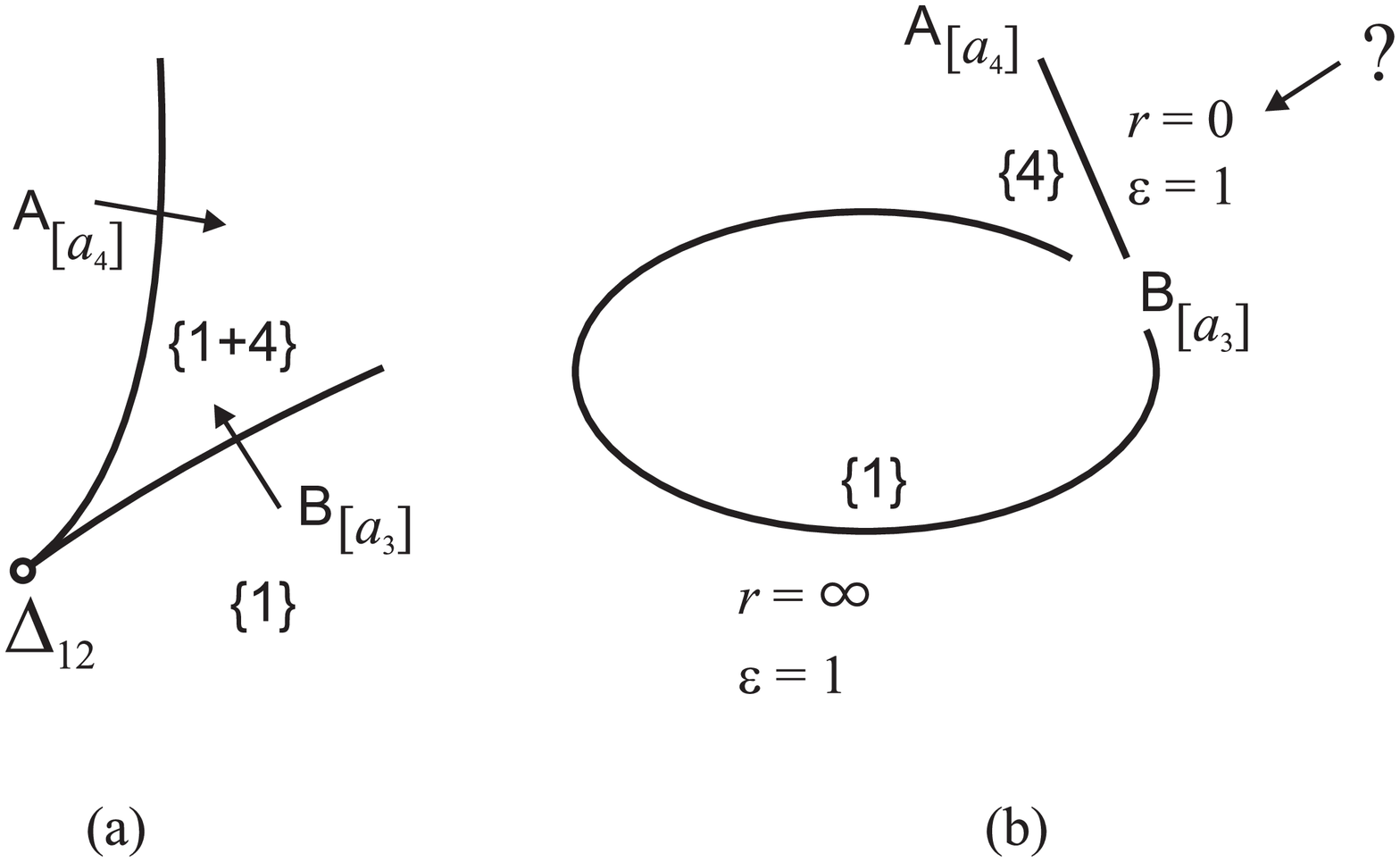}
\caption{Диаграмма и круговая молекула для $\Delta_{12}$.}\label{fig_delta12}
\end{figure}

Рассмотрим точку $\Delta_{13}$. В ней сходятся ребра диаграммы (область 2): ребро $\aaa_8$ -- с двумя атомами $A$, в которых рождаются торы семейств $\tfs{4}, \tfs{6}$, ребро $\aaa_7$ -- с двумя атомами $B$, в которых каждый из торов семейств $\tfs{1}, \tfs{2}$ перестраивается в пару торов, выбранную из семейств $\tfs{1}, \tfs{2}, \tfs{4}, \tfs{6}$. Это ребро заканчивается невырожденной точкой $\delta_{23}$, при исследовании которой было доказано предложение \ref{propfora7} о перестройках семейств в двух атомах $B_{[\aaa_7]}$.  Как следствие получаем круговую молекулу из двух компонент для $\Delta_{13}$, показанную вместе с фрагментом диаграммы на рис.~\ref{fig_delta13}.

Отметим, что факт наличия двух компонент в составе круговой молекулы вырожденной круговой орбиты, даже при наличии точки возврата на диаграмме, должен доказываться, так как полной классификации таких молекул нет. Ранее мы доказали это с использованием аналитического решения в теореме \ref{theDelta134}. Но теперь оказывается, что молекула из одной компоненты могла возникнуть лишь в том случае, если бы два тора семейств, рожденных на ребре $\aaa_8$, пришли в один и тот же атом $B$ ребра $\aaa_7$, что не так в силу предложения~\ref{propfora7}. Так что наличие двух компонент в $J_{13}$ доказано и топологически. Поэтому метки могут быть взяты из работы \cite{BolFom}.

\begin{figure}[ht]
\centering
\includegraphics[width=\pp\textwidth, keepaspectratio]{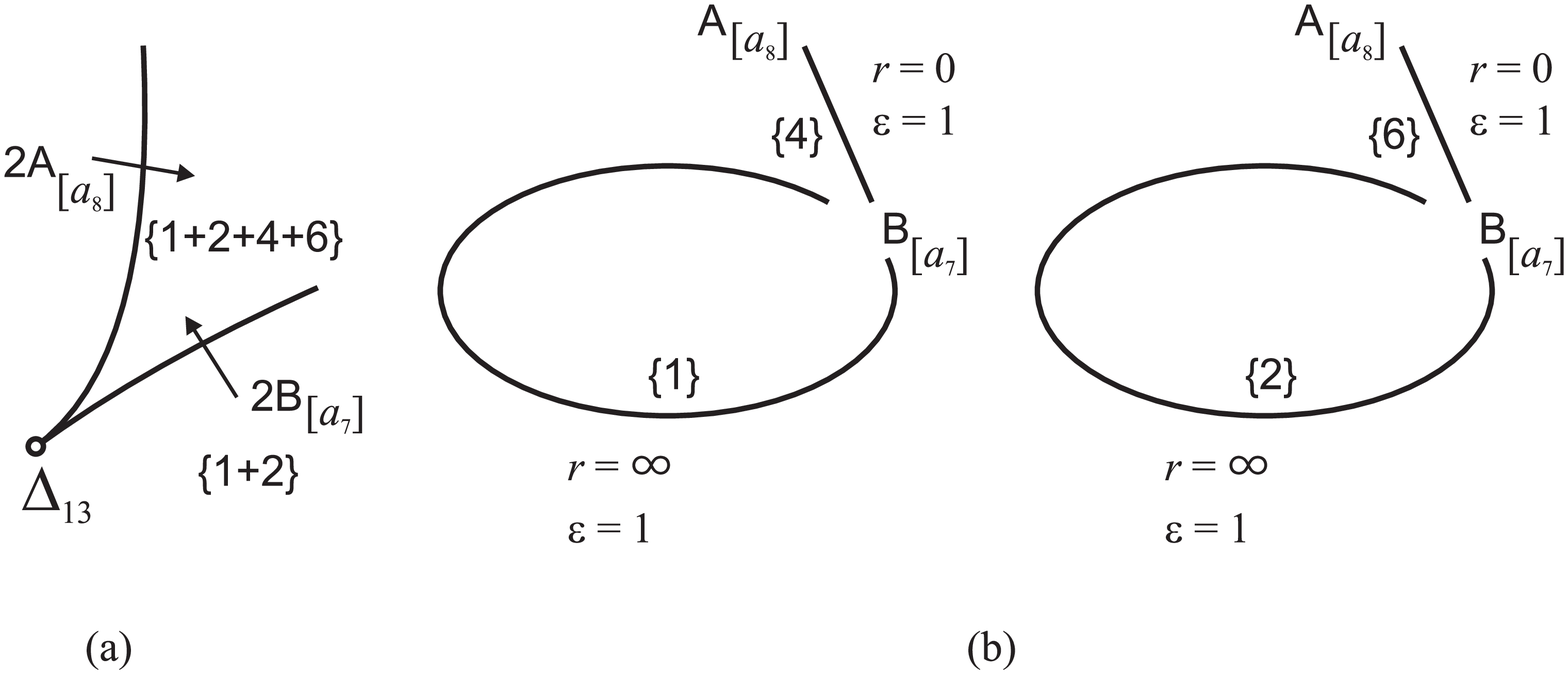}
\caption{Диаграмма и круговая молекула для $\Delta_{13}$.}\label{fig_delta13}
\end{figure}

\begin{figure}[ht]
\centering
\includegraphics[width=\pp\textwidth, keepaspectratio]{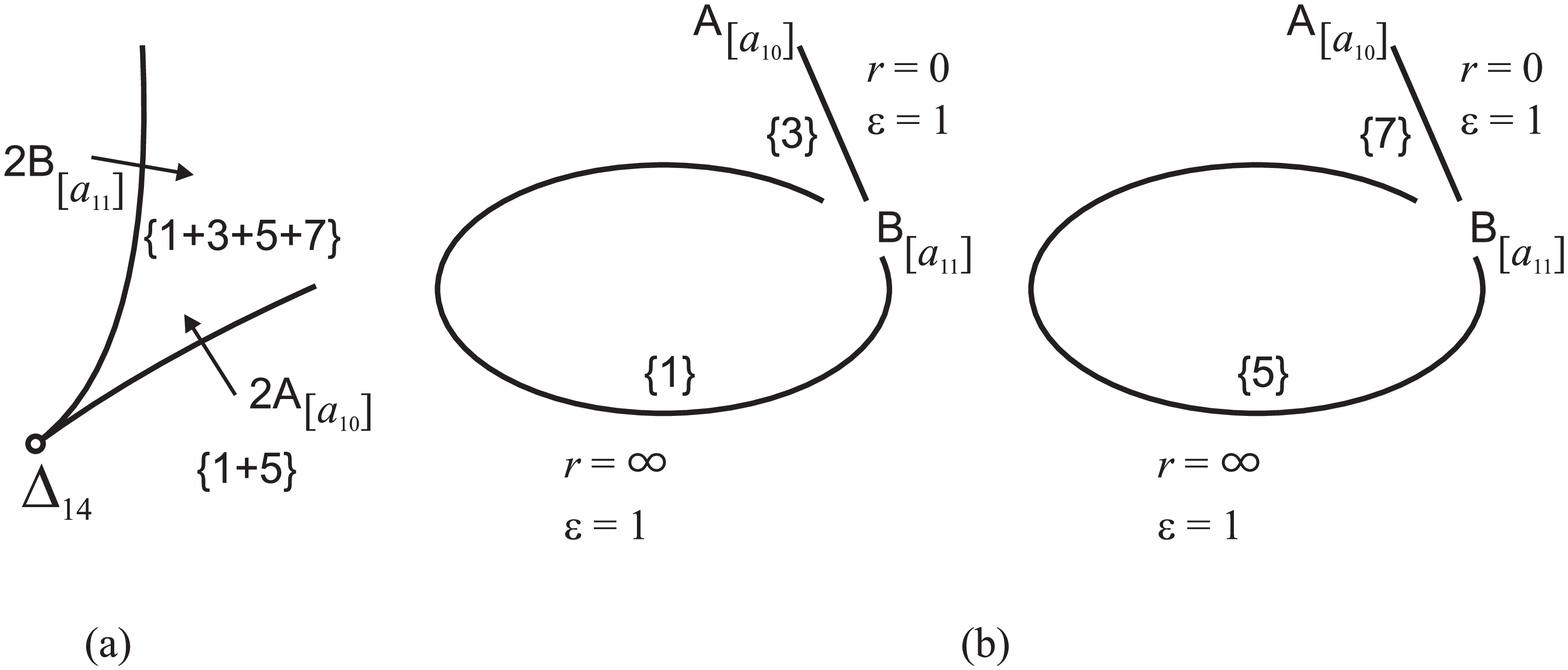}
\caption{Диаграмма и круговая молекула для $\Delta_{14}$.}\label{fig_delta14}
\end{figure}

Для точки $\Delta_{14}$ возникает тот же вопрос о количестве связных компонент в круговой молекуле. Наличие двух компонент также доказано ранее с использованием аналитического решения в теореме \ref{theDelta134}. С другой стороны, здесь два атома $B$ реализуются на ребре $\aaa_{11}$ (см. диаграмму для области $4'$). Это ребро заканчивается в уже изученной точке $\Delta_{04}$, поэтому известно, что на одном атоме $B$ тор из семейства $\tfs{1}$ перестраивается в два тора из семейств $\tfs{1}$, $\tfs{3}$, на другом -- тор из семейства $\tfs{5}$ перестраивается в два тора из семейств $\tfs{5}$, $\tfs{7}$, что показано в компонентах круговой молекулы на рис.~\ref{fig_delta04mols}. Результат для точки $\Delta_{14}$ -- фрагмент диаграммы и две компоненты круговой молекулы, обоснованные и топологически, -- представлен на рис.~\ref{fig_delta14}. Метки здесь взяты из работы \cite{Mor}.

\subsubsection{Точки $\gac$}

Последний класс вырожденных точек ранга 1 -- это точки, лежащие в прообразе ребра возврата поверхности $\wsc$.

Как и для множества $\gaa$, если в прообразе точки $\Delta_{3i}$ имеется лишь одна вырожденная окружность, то есть соответствующая критическая интегральная поверхность и круговая молекула связны, достаточно установить лишь номера семейств, участвующих в бифуркациях. Поэтому для точек $\Delta_{31},\Delta_{33}$ (см., например, диаграммы для областей $1$ и $1'$) имеем результат, показанный на рис.~\ref{fig_delta31}, \ref{fig_delta33}. Метки круговой молекулы точки $\Delta_{33}$ известны из \cite{BRF}, \cite{Mor}. Для точки $\Delta_{31}$ бесконечная $r$-метка является следствием общего утверждения \cite{Mor}, цитированного выше, равенство нулю $r$-метки на ребре $A_{[\ccc_1]} - B_{[\ccc_2]}$ требует дополнительного обоснования.

\begin{figure}[!ht]
\centering
\includegraphics[width=0.5\textwidth, keepaspectratio]{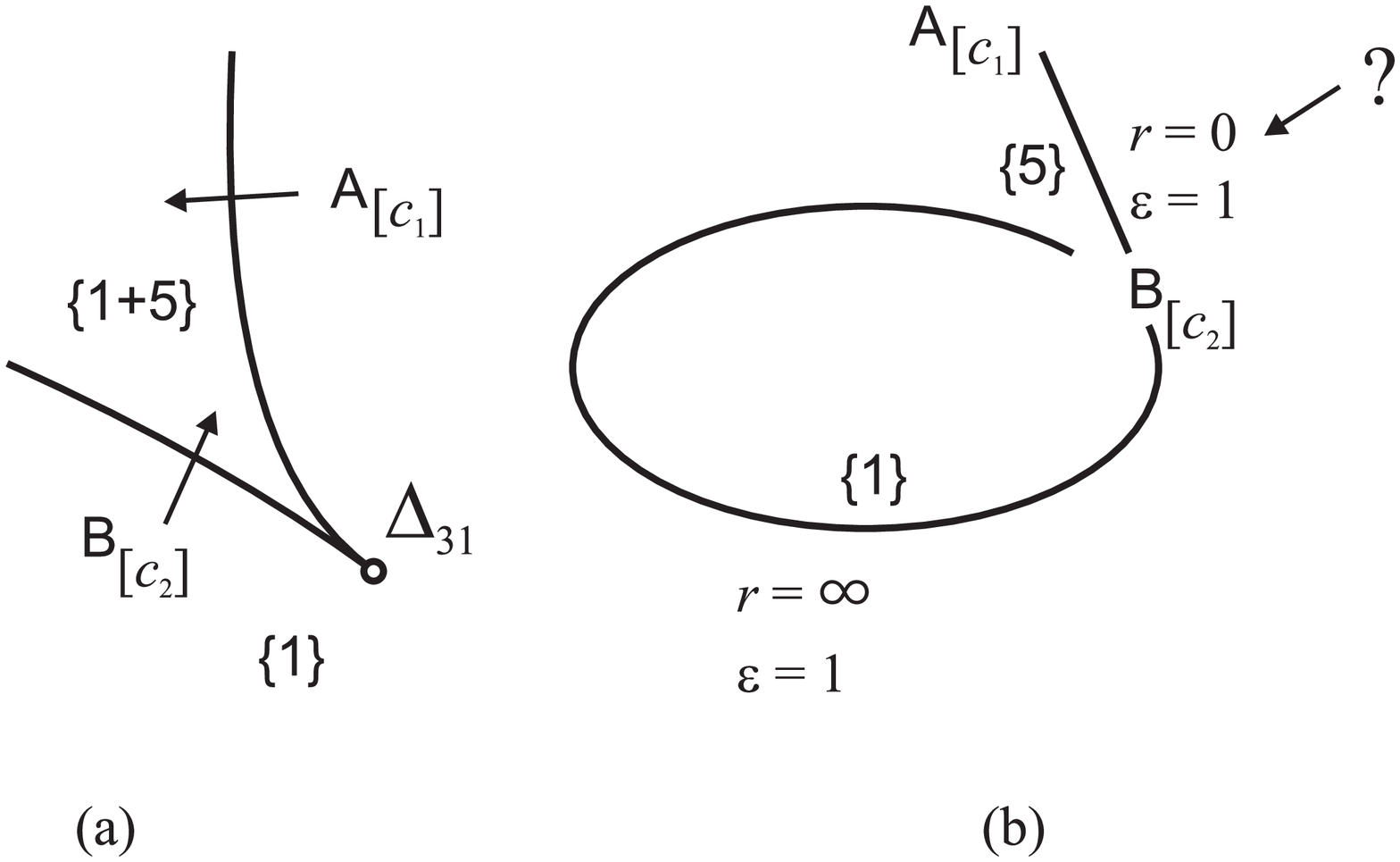}
\caption{Диаграмма и круговая молекула для $\Delta_{31}$.}\label{fig_delta31}
\end{figure}

\begin{figure}[!ht]
\centering
\includegraphics[width=0.5\textwidth, keepaspectratio]{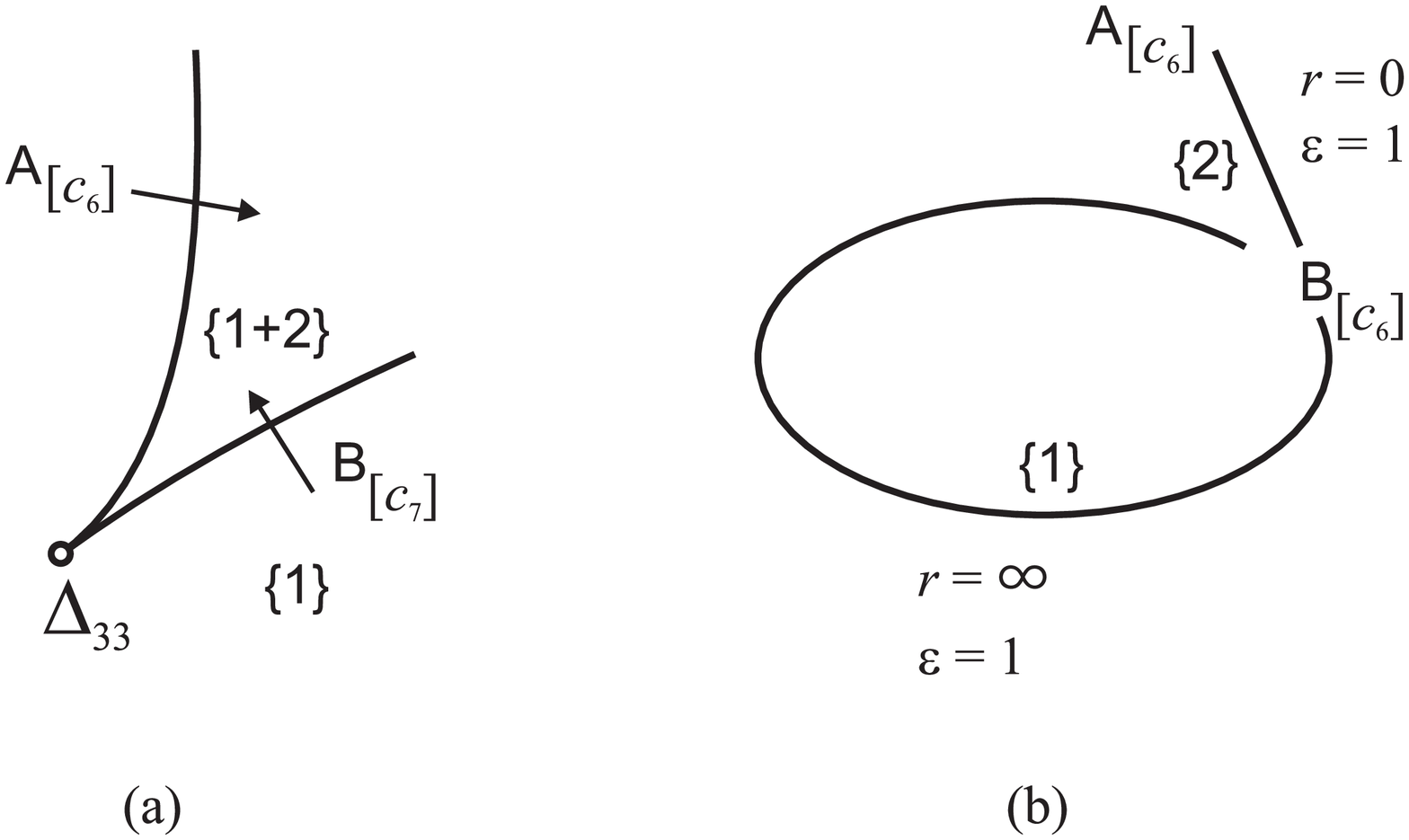}
\caption{Диаграмма и круговая молекула для $\Delta_{33}$.}\label{fig_delta33}
\end{figure}

\begin{proposition}\label{propDelta31}
В круговой молекуле точки $\Delta_{31}$ на ребре, соединяющем атом $A$ с седловым, $r$-метка равна нулю.
\end{proposition}

\begin{proof} Вновь используем правило сложения меток при одной бесконечной $r$-метке, теперь для точки $\Delta_{31}$ и перехода по семейству $\tfs{5}$ от $B_{[\ccc_2]}$ к $A_{[\ccc_1]}$. Здесь также можно рассмотреть различные диаграммы. Например, в области 5 имеем переходы $A_{[\ccc_1]} \to A_{[\aaa_9]} \to B_{[\ccc_2]}$. Однако обе метки здесь конечны: для первого из молекулы точки $\delta_{25}$ имеем $r=0$, а для второго -- из молекулы точки $\Delta_{03}$ следует, что $r=1/2$. Поэтому ``сумма'' меток определена неоднозначно. Точка $\Delta_{31}$ присутствует и в области $1'$, но, как и для точки $\Delta_{12}$, здесь нельзя указать разложение нужного перехода в сумму двух. Обратимся к диаграмме для области $9'$. Здесь самой точки $\Delta_{31}$ уже нет, но есть переходы по семейству $\tfs{5}$ следующего вида: переход $B_{[\ccc_2]} \to A_{[\aaa_{12}]}$ принадлежит молекуле невырожденной точки $\delta''_{27}$ с меткой $r=\infty$; переход $ A_{[\aaa_{12}]} \to  A_{[\ccc_1]}$ принадлежит молекуле невырожденной точки $\delta''_{28}$ с меткой $r=0$. Следовательно, недостающая $r$-метка равна 0.
\end{proof}

Рассмотрим точку $\Delta_{32}$. Фрагмент диаграммы с обозначениями ребер и соответствующими атомами возьмем из диаграмм для области $4$. Для этой точки нужно решить вопрос о количестве связных компонент молекулы (аналитическое решение для полной интегральной поверхности отсутствует). Однако для того, чтобы молекула оказалась связной, необходимо, чтобы семейства $\tfs{2}, \tfs{6}$, рожденные на ребре $\ccc_8$ пришли в один атом $B$ ребра $\ccc_9$. Но, как следует из предложения~\ref{propforc9}, это не так. Таким образом, молекула точки $\Delta_{32}$ состоит из двух компонент, а предложение~\ref{propforc9} определяет правила склейки семейств. Метки вычислены в работе \cite{BolFom}. Результат показан на рис.~\ref{fig_delta32}.

\begin{figure}[ht]
\centering
\includegraphics[width=\pp\textwidth, keepaspectratio]{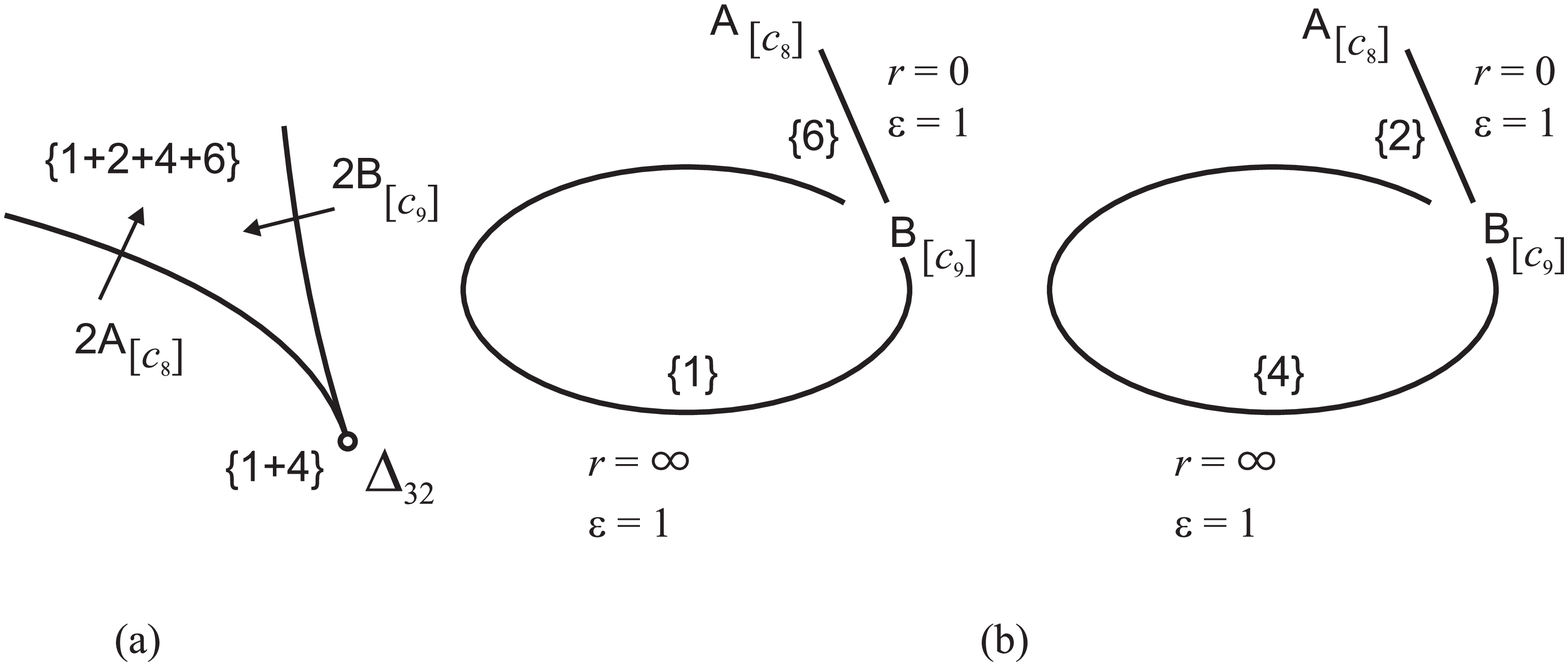}
\caption{Диаграмма и круговая молекула для $\Delta_{32}$.}\label{fig_delta32}
\end{figure}

\begin{figure}[ht]
\centering
\includegraphics[width=\pp\textwidth, keepaspectratio]{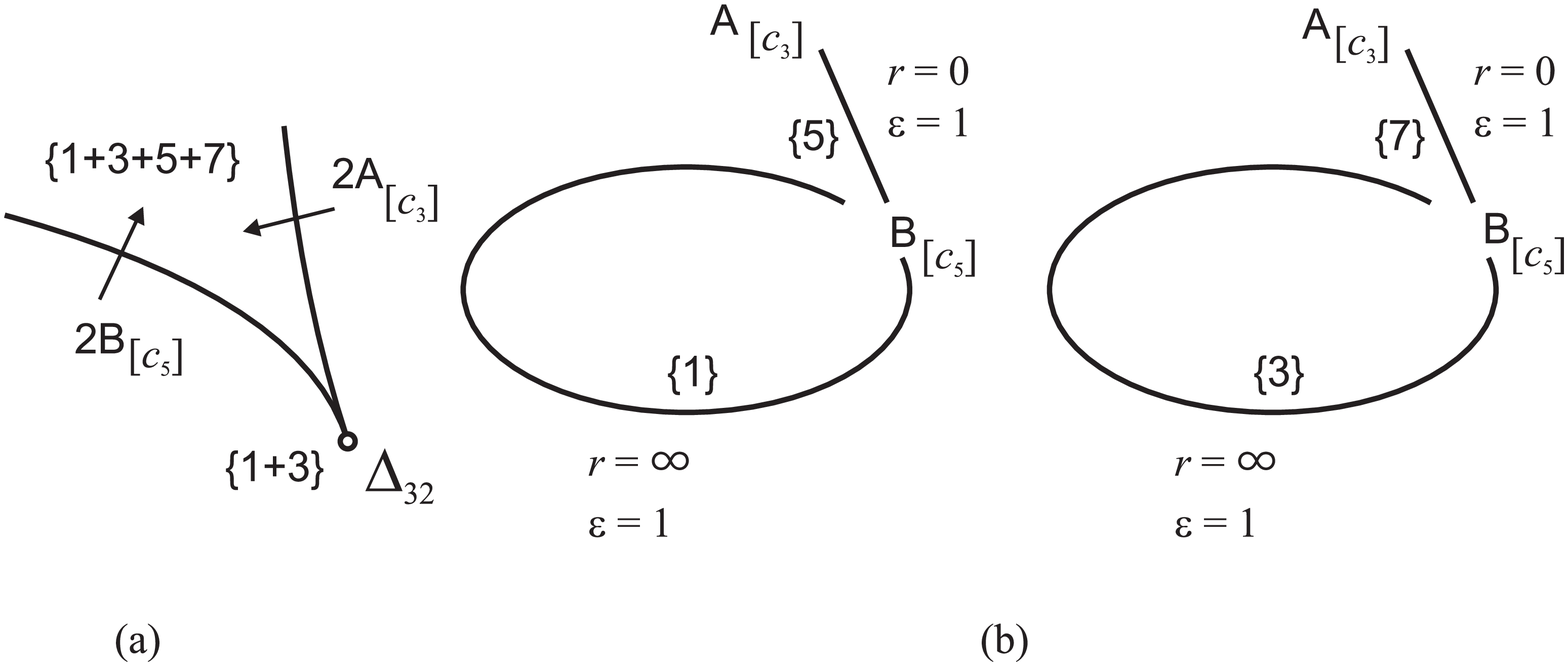}
\caption{Диаграмма и круговая молекула для $\Delta_{34}$.}\label{fig_delta34}
\end{figure}

Для точки $\Delta_{34}$ (диаграмма $\mSell$ для области $2'$) вопрос о количестве компонент и правилах склейки семейств решается на ребре $\ccc_5$, где имеется два атома $B$. Это ребро с одной стороны заканчивается точкой $\delta_{26}$ -- образом невырожденной критической точки ранга 0. В то же время, если рассмотреть диаграмму $\mSell$ для области $3'$, то видно, что ребро $\ccc_5$ также имеет выход к точке $\delta'_{27}$, образу другой невырожденной критической точки ранга 0, откуда, как показано ранее, при бифуркации семейств в двух атомах $B_{[\ccc_5]}$ тор из семейства $\tfs{1}$ перестраивается в два тора из семейств $\tfs{1}$, $\tfs{5}$, а тор из семейства $\tfs{3}$ перестраивается в два тора из семейств $\tfs{3}$, $\tfs{7}$. А поскольку на ребре $\ccc_3$ рождаются торы семейств $\tfs{5},\tfs{7}$ и эти торы приходят в разные атомы $B$ ребра $\ccc_5$, то молекула точки $\Delta_{34}$ имеет две связных компоненты. Отсюда же устанавливаются однозначно и правила склейки семейств. Имея доказательство наличия двух компонент, метки возьмем из работы \cite{Mor}. Результат показан на рис.~\ref{fig_delta34}.
\clearpage

\clearpage

\section{Изоэнергетические бифуркационные диаграммы}\label{sec8}
\subsection{Разделяющее множество}
В этом разделе излагается классификация бифуркационных диаграмм бифуркационных диаграмм $\mSash$ отображений $L{\times}K$, ограниченных на изоэнергетические уровни $Q_h^4 = H^{-1}(h)$ в полном фазовом пространстве $\mP^5$ \cite{mtt40}. Топология самих этих уровней устанавливается тривиально. Поскольку $H$ на $\mP^5$ есть функция Морса, имеющая ровно две критические точки
$$
\bo=0, \qquad \ba = (\pm 1,0,0),
$$
то
$$
Q_h^4 = \left\{\begin{array}{ll}
\varnothing, & h<-1\\
S^4, & -1 <h <1\\
S^2{\times}S^2, & h>1
\end{array} \right. .
$$

Разделяющее множество $\Theta_H$, или $H$-атлас, получим сразу же из предложения \ref{propos13}, учитывая результаты всех накопленных фактов о критических подсистемах: это множество есть образ на $(\ld,h)$-плоскости отмеченных выше особых точек \eqref{eq5_11}, \eqref{eq5_16} -- \eqref{eq5_21}, \eqref{eq5_23} -- \eqref{eq5_27}, \eqref{eq5_34} ключевых множеств критических подсистем (напомним, что точка $D_4$ введена при классификации диаграмм $\mSell$ и перестроек диаграмм на $Q_h^4$ не вызывает, зато дополнительно появляется влияющая точка $D_5$). Как и ранее, достаточно рассматривать $\ld \gs 0$, а иллюстрации удобно привести в плоскости $(h,\ld^2)$.

\begin{theorem}[М.П.\,Харламов, И.И.\,Харламова, Е.Г.\,Шведов]\label{theClassH}
Разделяющее множество $\Theta_H$ при классификации бифуркационных диаграмм $\mSash$ состоит из  следующих 13 кривых в плоскости $(\ld,h)$:
\begin{eqnarray}
& & \rz_0: h=-1, \quad \ld \gs 0; \nonumber \\
& & \rz_1: h=1,  \quad \ld \gs 0; \nonumber \\
& & \rz_2: h=\frac{3}{2}\sqrt{1+\ld^4}-\ld^2, \quad \ld \gs 0; \nonumber \\
& & \rz_3: h= \ds \frac{1}{8} \left(\sqrt{4+\ld^{4/3}}-\ld^{2/3}\right)^2\left(2\sqrt{4+\ld^{4/3}}+\ld^{2/3} \right), \quad \ld \gs 0;\nonumber \\
& & \rz_4: \ds \ld=\frac{3x^4-4}{2x^3}, \quad \ds h=\frac{3}{8}x^2+\frac{2}{x^6}, \quad x \in (\sqrt[4]{4/3},\sqrt{2}]; \nonumber \\
& & \rz_5: \ds{h=-\frac{\ld^2}{2}+\ld^{2/3}+\frac{1}{2\ld^{2/3}}}, \quad \ld \in (0,1]; \nonumber \\
& & \rz_6: \ld=\ds{\frac{1}{2}x+\frac{2}{x^3}}, \quad
    \ds{h=\frac{1}{8}x^2+\frac{2}{x^2}-\frac{2}{x^6}}, \quad x \in [\sqrt{2},\sqrt{2\sqrt{3}}];\nonumber \\
& & \rz_7: \ld=\ds{\frac{1}{2}x+\frac{2}{x^3}}, \quad
    \ds{h=\frac{1}{8}x^2+\frac{2}{x^2}-\frac{2}{x^6}}, \quad x \in [\sqrt{2\sqrt{3}},+\infty);\nonumber \\
& & \rz_8: \ld=\ds{\frac{3x^4-4}{2x^3}}, \quad \ds{h=\frac{3}{8}x^2+\frac{2}{x^6}}, \quad x \in [-\sqrt[4]{4/3},0); \nonumber \\
& & \rz_9: \ds{h=\frac{1+\ld^4}{2\ld^2}}, \quad \ld>0 ;\nonumber \\
& & \rz_{10}: \ds{h=\frac{\ld^2}{2}}, \quad \ld \gs 0 ;\nonumber \\
& & \rz_{11}: \ds{h=\sqrt{2}-\frac{\ld^2}{2}}, \quad \ld^2 \ls \frac{1}{\sqrt{2}}; \nonumber \\
& & \rz_{12}: h=\frac{\ld^{2/3}}{2}(3-\ld^{4/3}), \quad \ld^2 \in [0,\frac{8}{3\sqrt{3}}].\nonumber \end{eqnarray}
\end{theorem}

В этой формулировке разделяющие кривые записаны в последовательности порождающих точек $A$, $B_1 -B_7$, $C$, $D_1 - D_3$, $D_5$. Нетрудно видеть, что $\rz_4,\rz_8$ объединяются в одну гладкую кривую, а $\rz_6,\rz_7$ -- в одну кривую с точкой возврата.

На рис.~\ref{fig_atlas_all_xg} показаны разделяющие кривые с их номерами в квадратиках. В целом они разбивают полуплоскость $\ld \gs 0$ на 34 области, из которых одна, а именно та, где $h<-1$, отвечает пустым изоэнергетическим многообразиям $Q_h^4$. Остальные области помечены номерами 1 -- 33. В силу достаточно сложной структуры, на рис.~\ref{fig_fragm1numreg} представлены фрагменты допустимого множества с полной нумерацией областей. Для удобства это также сделано на плоскости $(h,\ld^2)$.

\begin{figure}[!ht]
\centering
\includegraphics[width=0.7\textwidth, keepaspectratio]{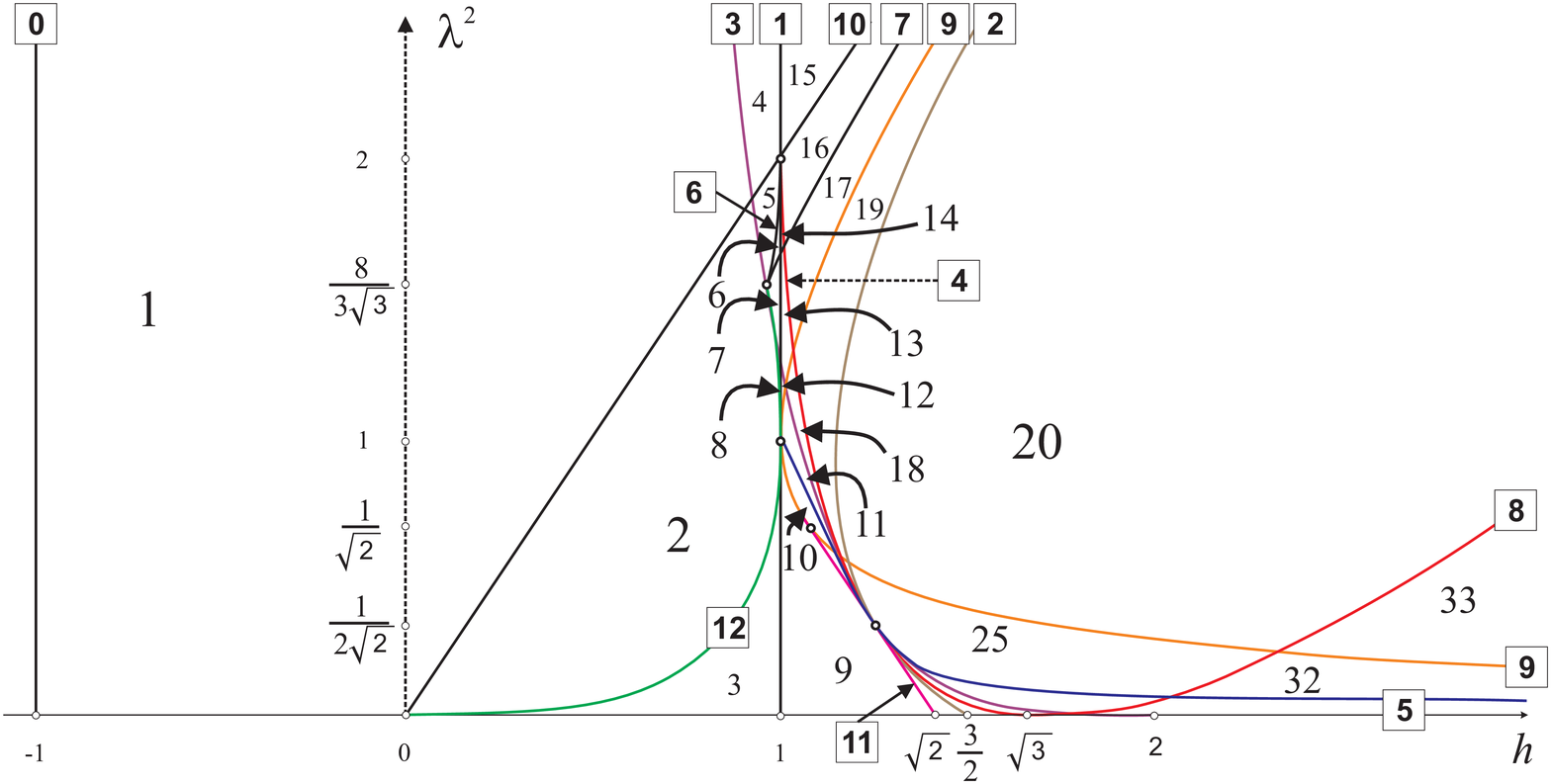}
\caption{Разделяющее множество $\Theta_H$ и порожденные области.}\label{fig_atlas_all_xg}
\end{figure}

\begin{figure}[!ht]
\centering
\includegraphics[width=0.45\textwidth, keepaspectratio]{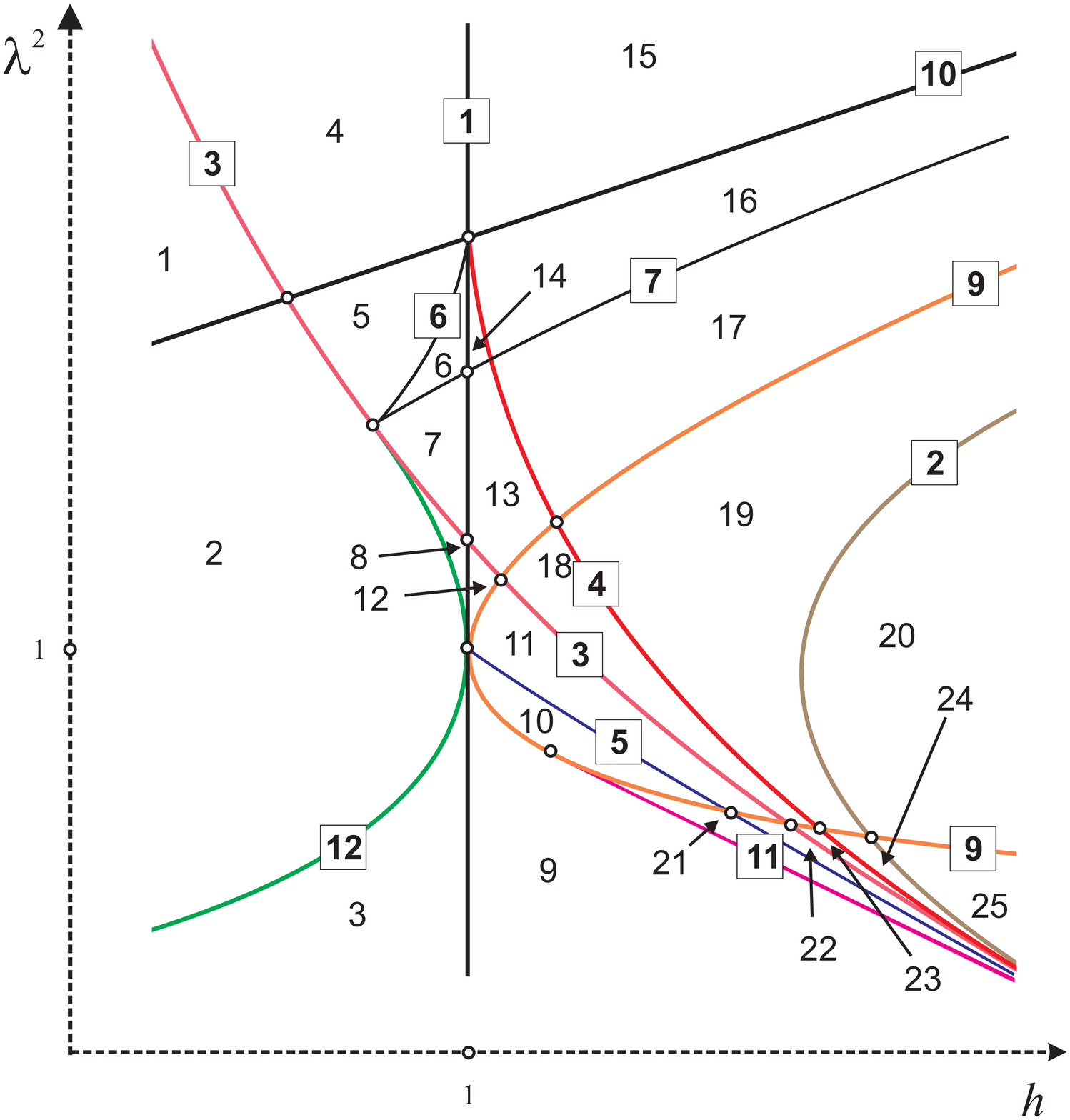}\ \includegraphics[width=0.45\textwidth, keepaspectratio]{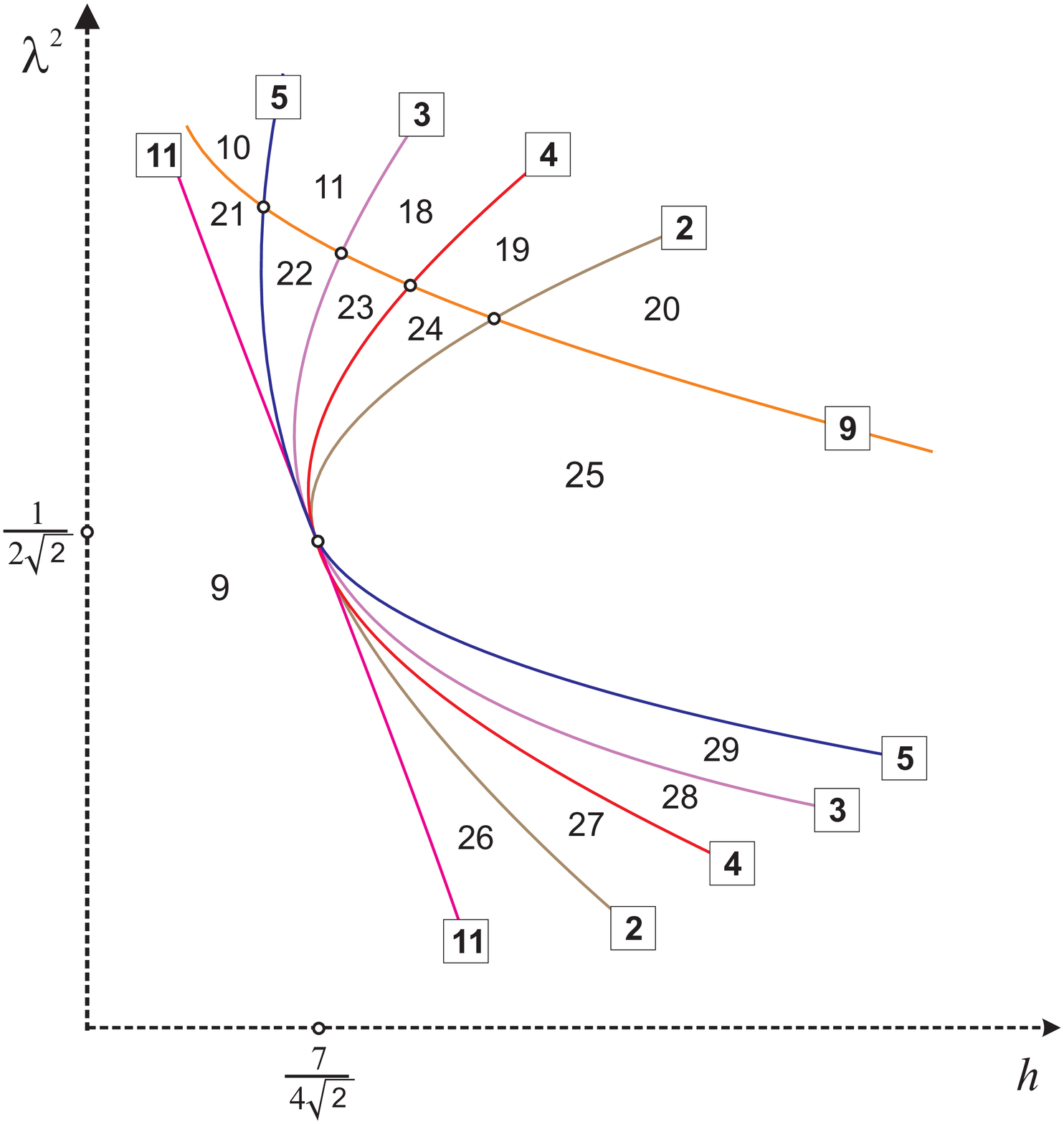}\\
\includegraphics[width=0.45\textwidth, keepaspectratio]{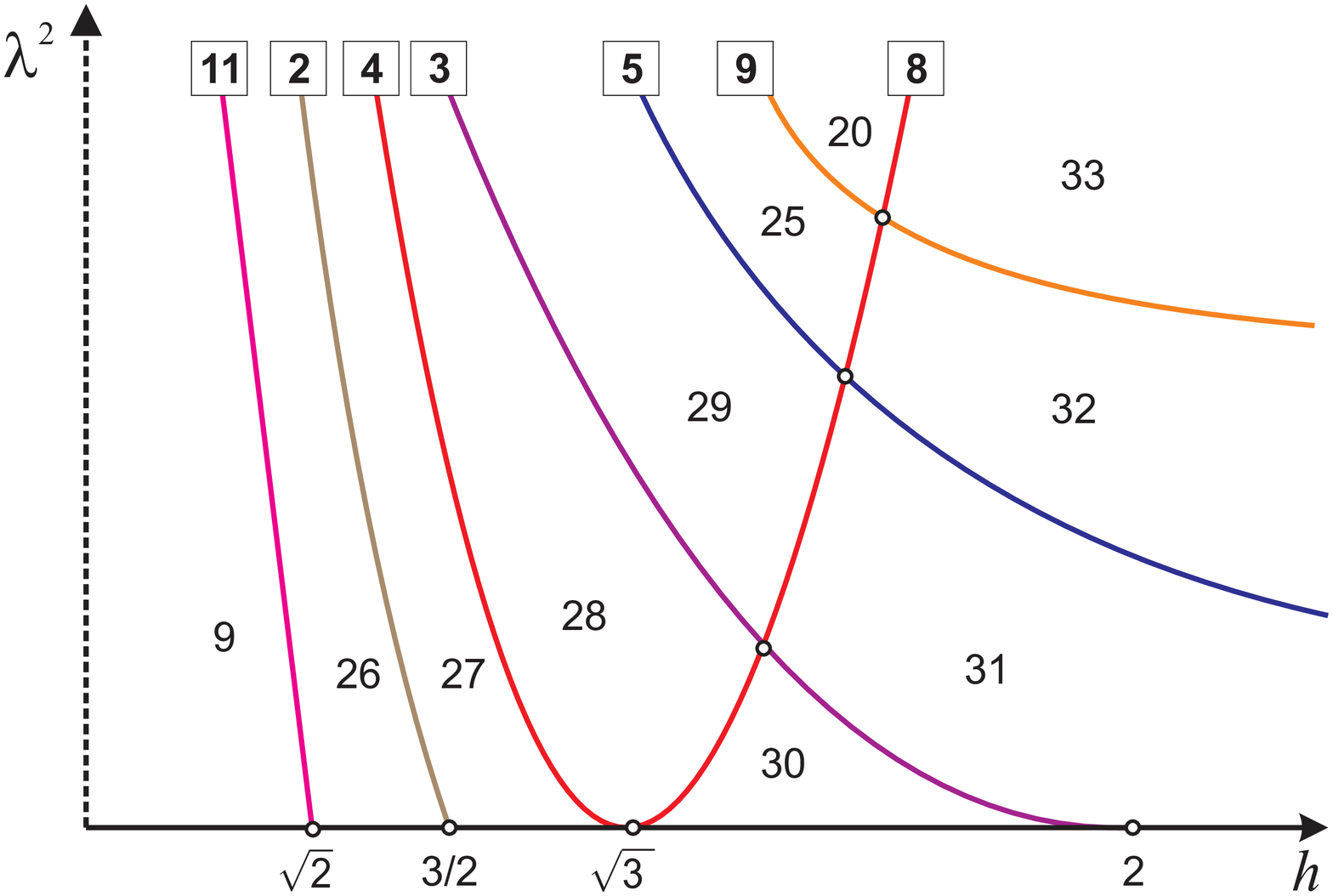}
\caption{Фрагменты множества $\Theta_H$ и порожденные области.}\label{fig_fragm1numreg}
\end{figure}

\subsection{Изоэнергетические диаграммы}
На последующих рисунках \ref{fig_reg01} -- \ref{fig_reg31} показаны оснащенные изоэнергетические диаграммы. Они симметричны относительно оси $\ell=0$, поэтому есть возможность одновременно нанести все обозначения гладких сегментов (для которых ранее установлены бифуркации и правила склейки семейств) в правой половине и явно указать камеры и бифуркации в левой половине. Это дает полное представление об изоэнергетическом двумерном инварианте ограничения системы на соответствующее $Q_h^4$, а также оказывается очень удобным для построения всех изоэнергетических одномерных инвариантов (графов Фоменко) приведенных систем на $\mPel$. Хорошо известна проблема визуализации ``сверхмалых'' областей при компьютерном построении бифуркационных диаграмм. Имея в виду использование диаграмм $\mSash$ для построения графов Фоменко (а с этой целью необходимо правильно увидеть все различные по структуре сечения допустимой области вертикальными прямыми $\ell=\cons$), мы постарались на рисунках сохранить, по возможности, такие сечения, то есть пользоваться искажениями диаграмм ``послойными'' диффеоморфизмами. Конечно, это оказалось не всегда возможно без потери наглядности. В таких случаях требование наглядности выдвигалось на первый план.
Основное отступление от послойных (сохраняющих сечения $\ell=\cons$) диффеоморфизмов области диаграммы можно наблюдать в иллюстрациях для областей $20, 25, 27 - 33$. Оно состоит в том, что ребро $\ccc_4$ в действительности является продолжением ребра $\ccc_9$ и потому должно быть ``короче'', чем $\aaa_5$. В результате $\ell$-координата точек $\Delta_{02}$ по абсолютной величине меньше, чем $\ell$-координата точек $\Delta_{03}$, что не так на рисунках. Кроме того, может нарушаться соотношение $\ell$-координат для пар ``хвостов'' $(\Delta_{31},\Delta_{11})$, $(\Delta_{31},\Delta_{14})$ и т.п. При этом, конечно, сама оснащенная изоэнергетическая диаграмма как грубый топологический инвариант остается верной.

Для точного построения и анализа всех сечений диаграммы $\Sigma$ фиксированными $h,\ell$ были разработаны пакеты в системе Wolfram Mathematica \cite{SaKhSh2012}, которые дают возможность в режиме диалога с достаточной точностью позиционировать точку $(h,\ld)$ относительно разделяющего множества $\Theta_H$, выбирать нужное значение $\ell$, одновременно наблюдать все соответствующие явления в ключевых множествах критических подсистем, и, в соответствии с приведенными выше формулами, для всех особых точек на сечении $\ell=\cons$ диаграммы $\mSash$ выводить значения показателей Морса\,--\,Ботта, определяющих бифуркацию по информации, собранной выше в таблицы. Аналогичный пакет для диаграмм $\mSell$ описан в работе \cite{RyabHarlUdgu2}. Конечно, результаты подобного компьютерного моделирования могут считаться лишь вспомогательными, дающими некоторый наглядный ``конструктор топологических инвариантов''. Предварительно все результаты по классификации графов Фоменко должны быть строго обоснованы аналитически. Это выполнено в следующем разделе.

\def\wid{0.6}

\begin{figure}[!htp]
\centering
\includegraphics[width=\wid\textwidth, keepaspectratio]{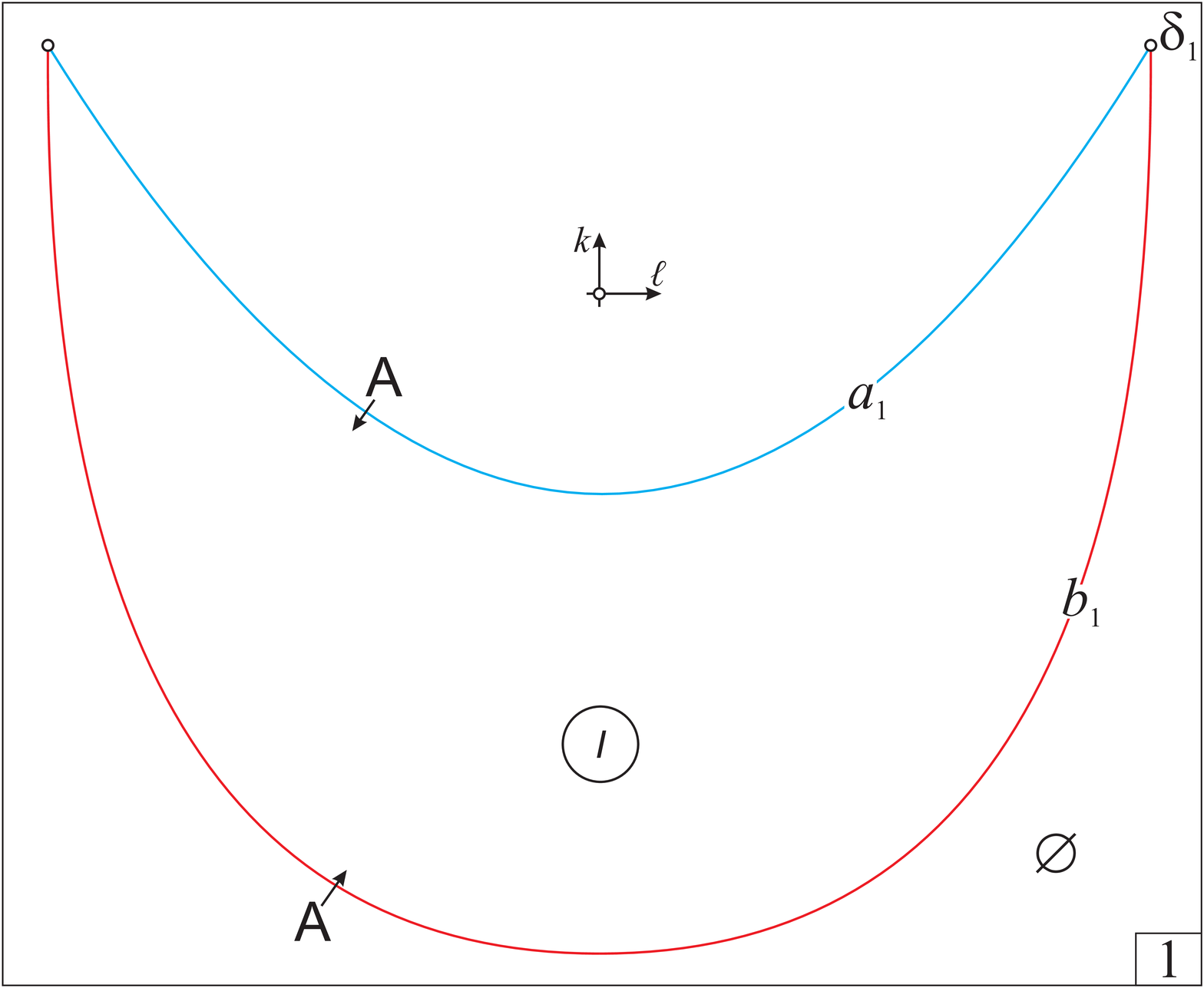}\\ \includegraphics[width=\wid\textwidth, keepaspectratio]{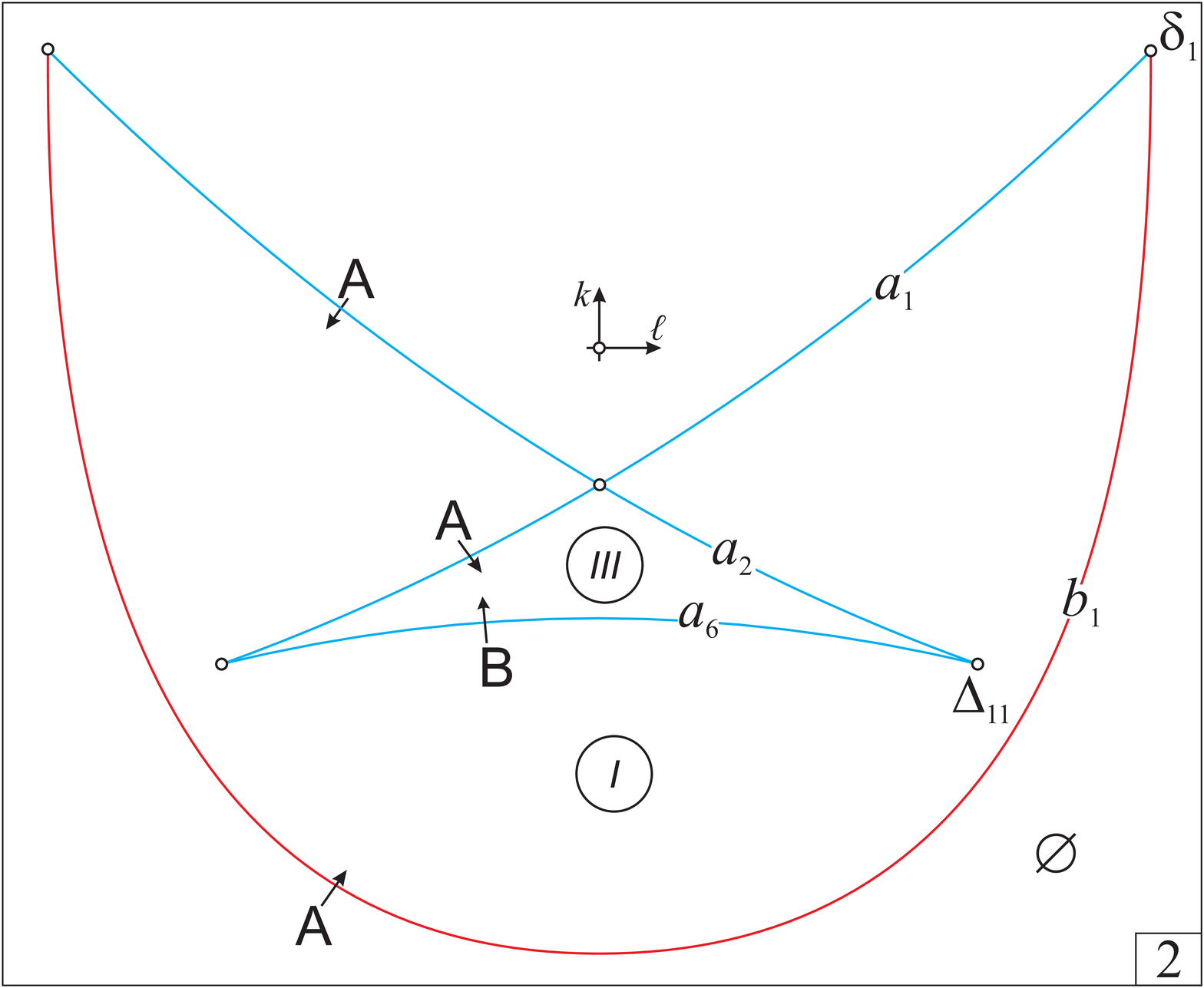}\\
\includegraphics[width=\wid\textwidth, keepaspectratio]{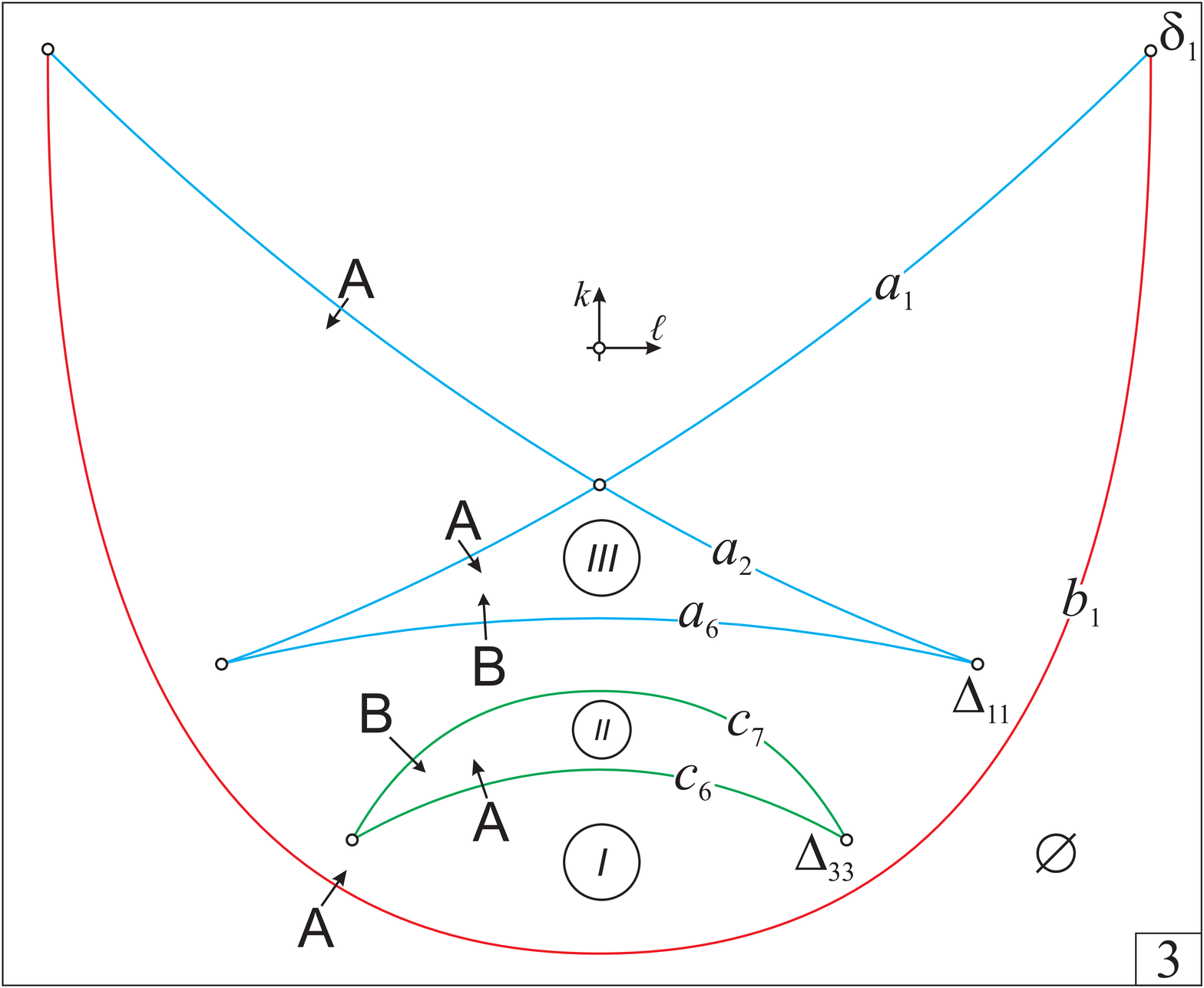}
\caption{Оснащенные диаграммы $\mSash$.}\label{fig_reg01}
\end{figure}

\begin{figure}[!htp]
\centering
\includegraphics[width=\wid\textwidth, keepaspectratio]{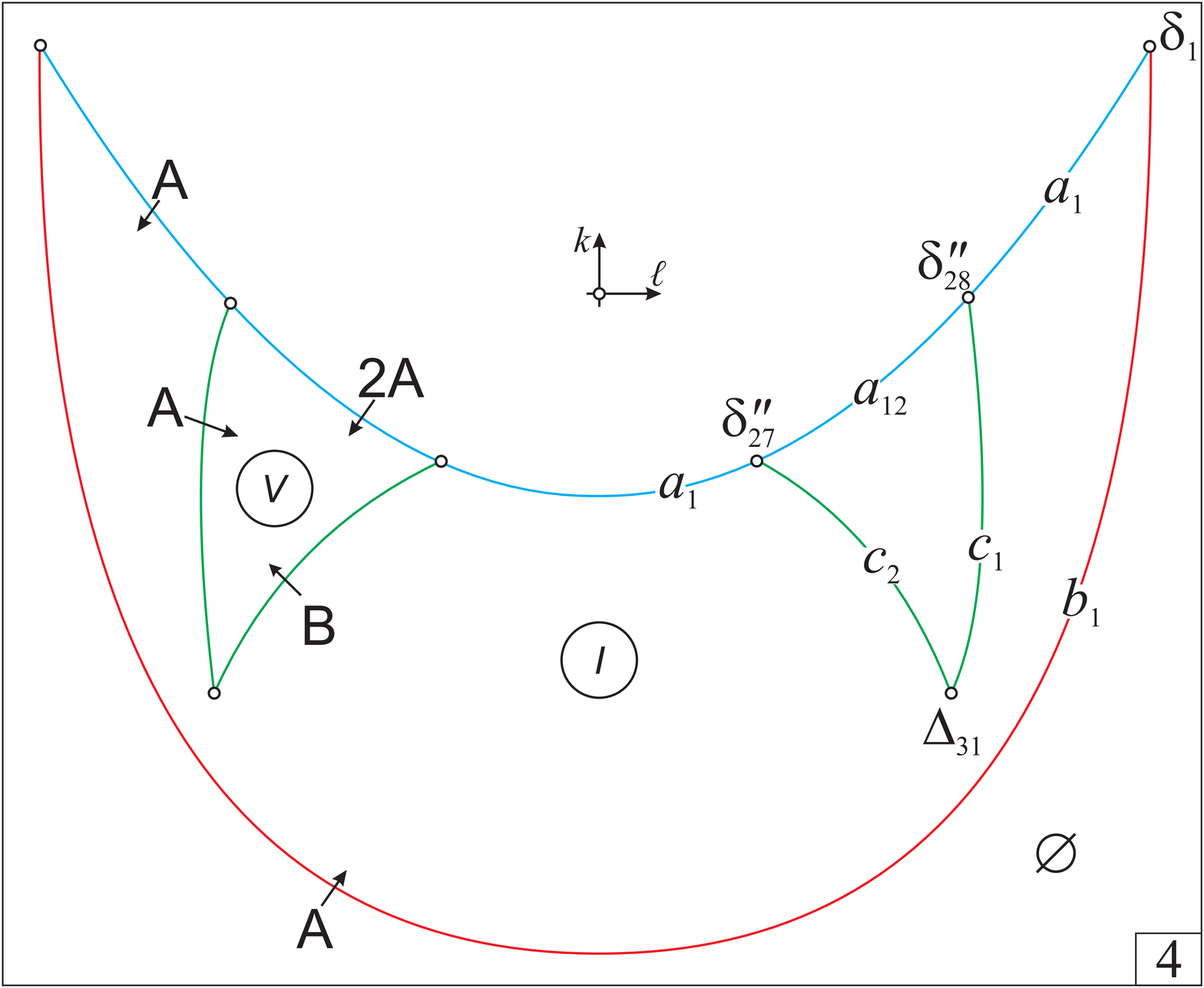}\\ \includegraphics[width=\wid\textwidth, keepaspectratio]{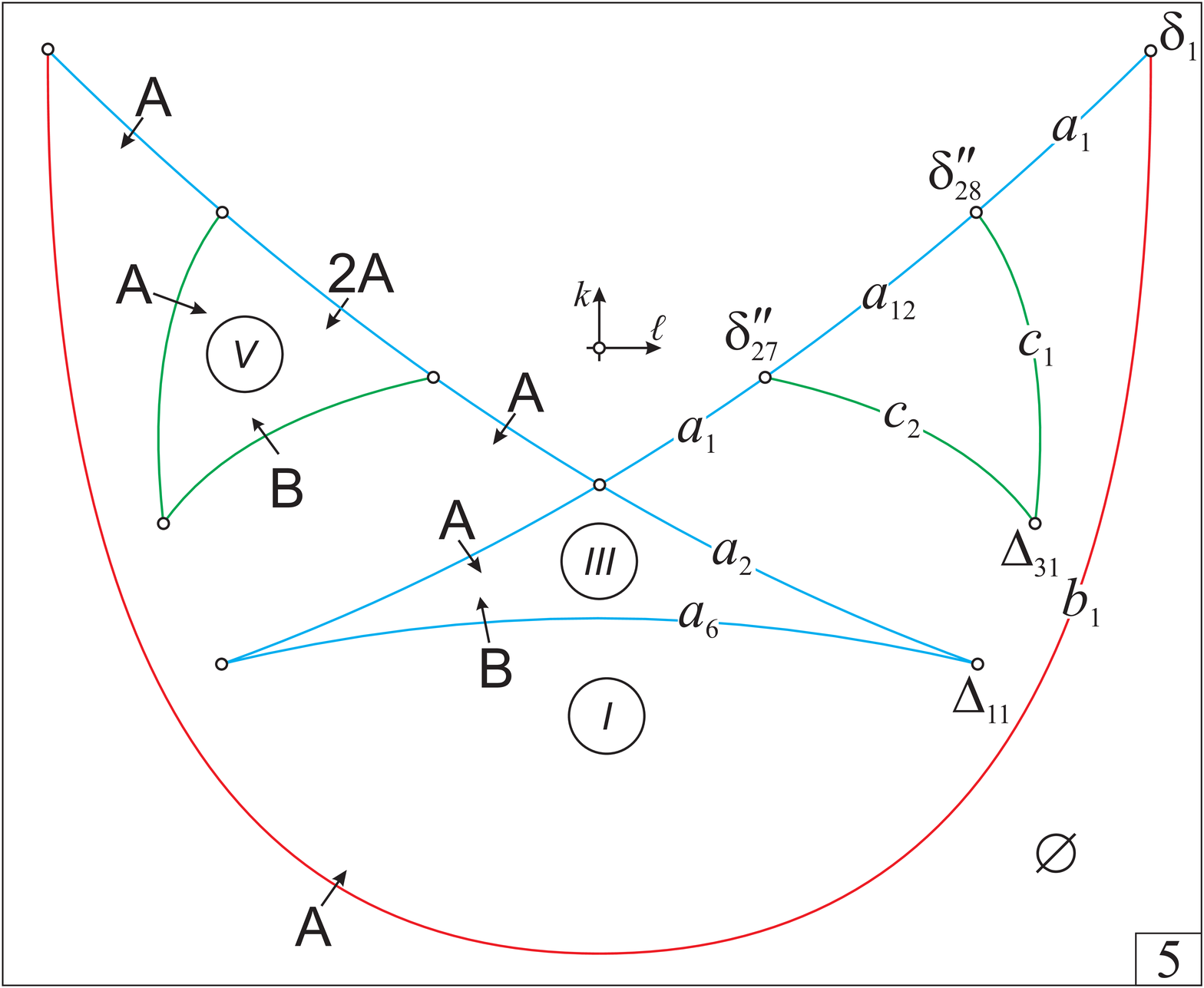}\\
\includegraphics[width=\wid\textwidth, keepaspectratio]{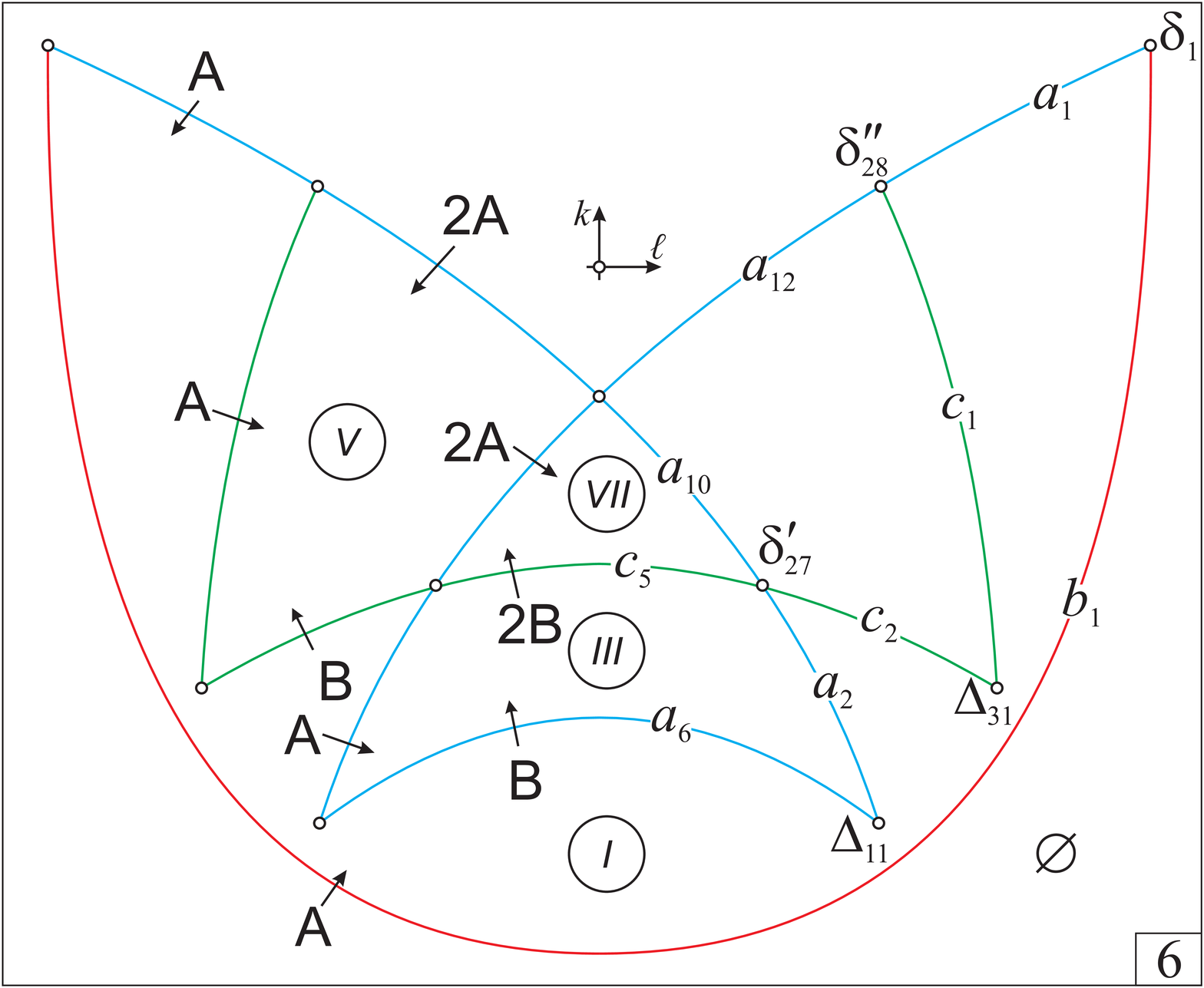}
\caption{Оснащенные диаграммы $\mSash$ (продолжение).}\label{fig_reg04}
\end{figure}

\begin{figure}[!htp]
\centering
\includegraphics[width=\wid\textwidth, keepaspectratio]{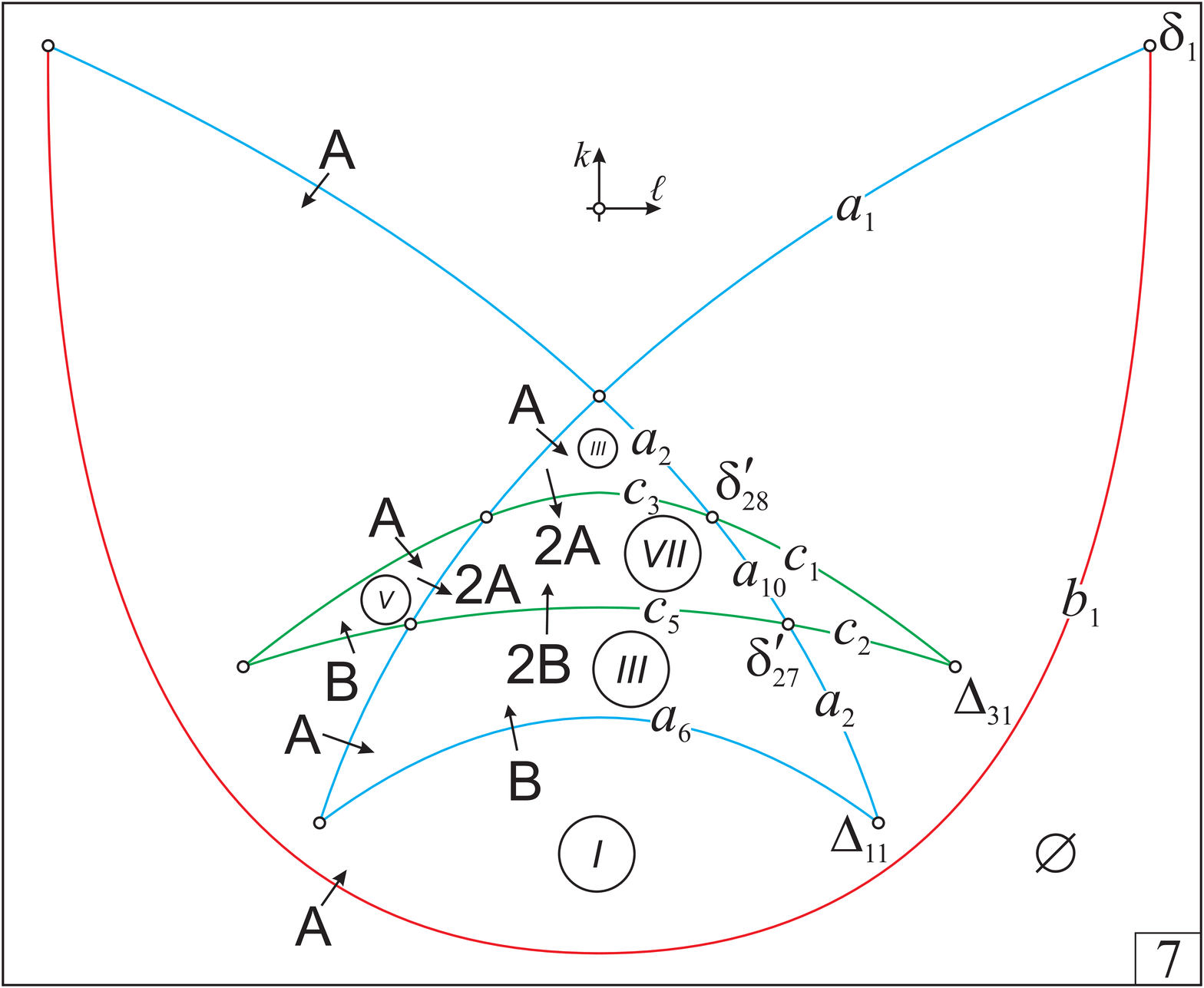}\\ \includegraphics[width=\wid\textwidth, keepaspectratio]{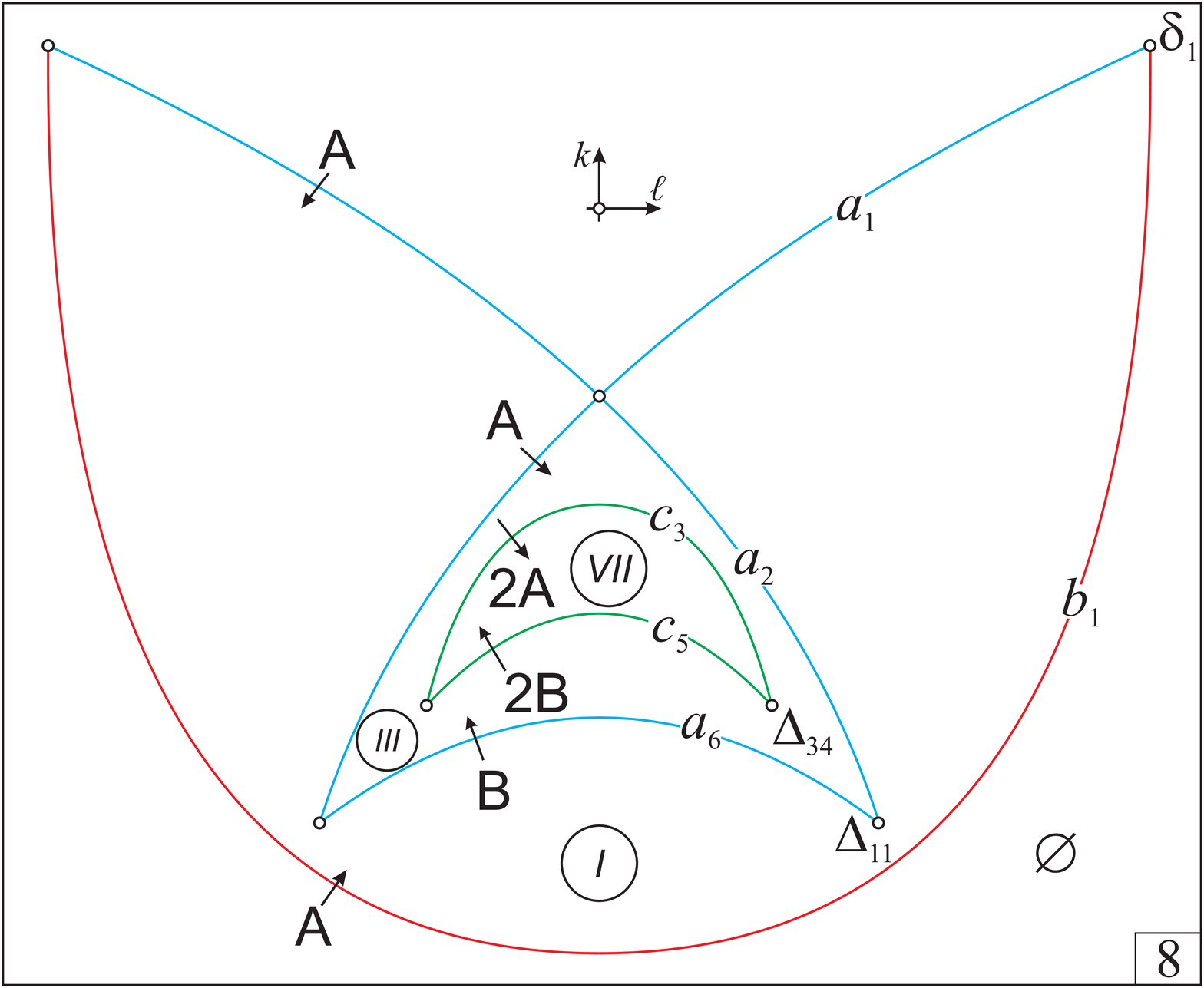}\\
\includegraphics[width=\wid\textwidth, keepaspectratio]{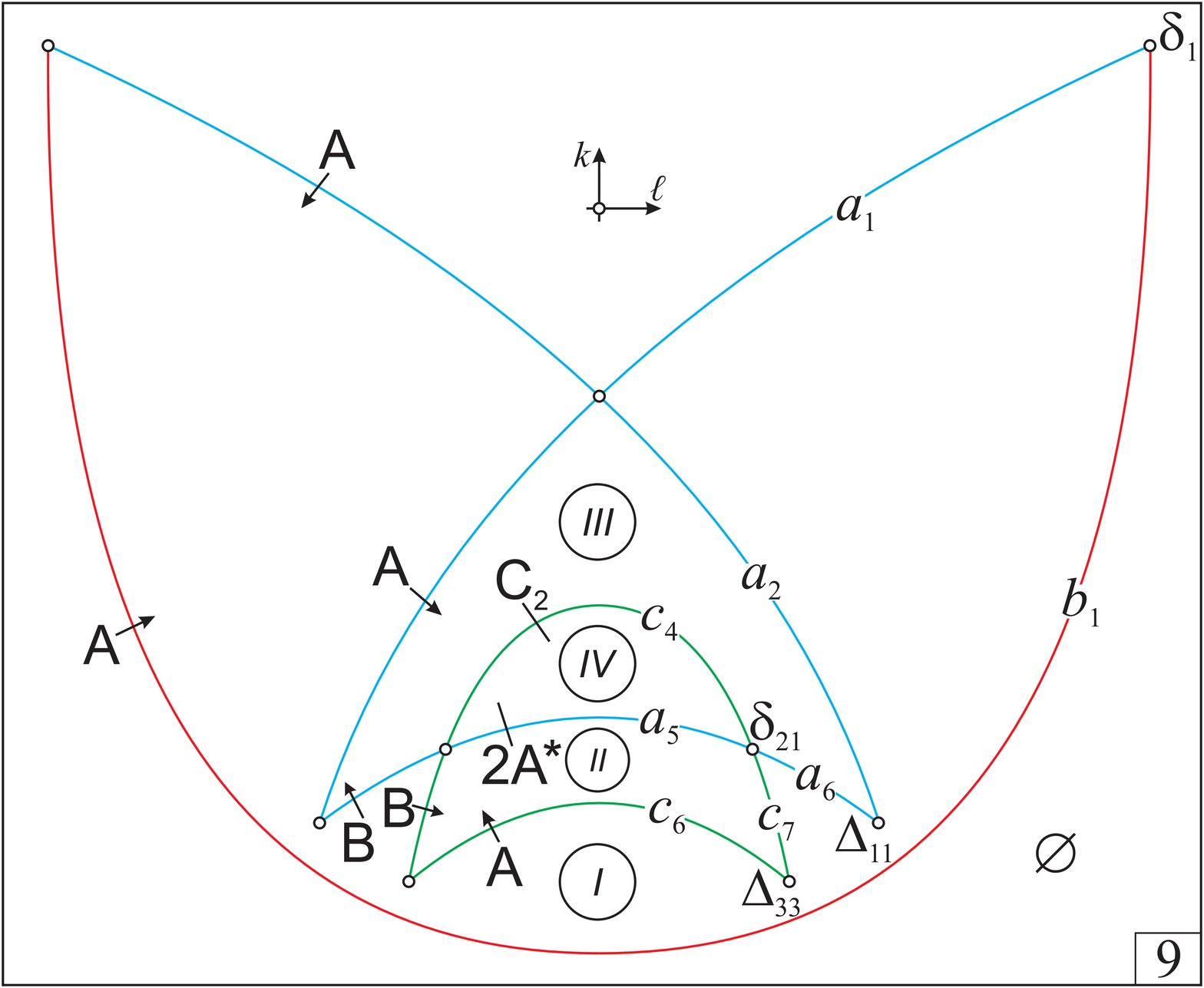}
\caption{Оснащенные диаграммы $\mSash$ (продолжение).}\label{fig_reg07}
\end{figure}

\begin{figure}[!htp]
\centering
\includegraphics[width=\wid\textwidth, keepaspectratio]{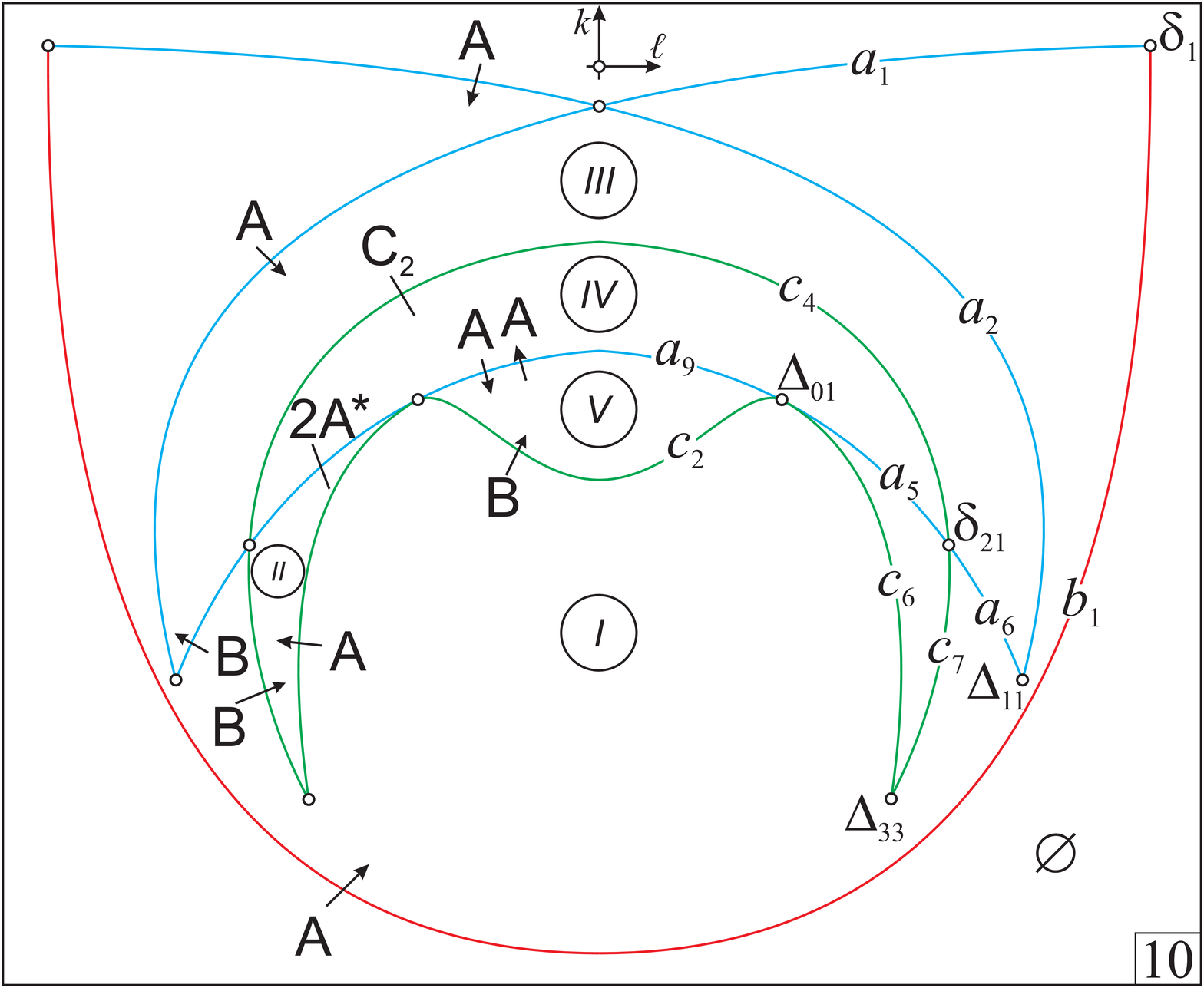}\\ \includegraphics[width=\wid\textwidth, keepaspectratio]{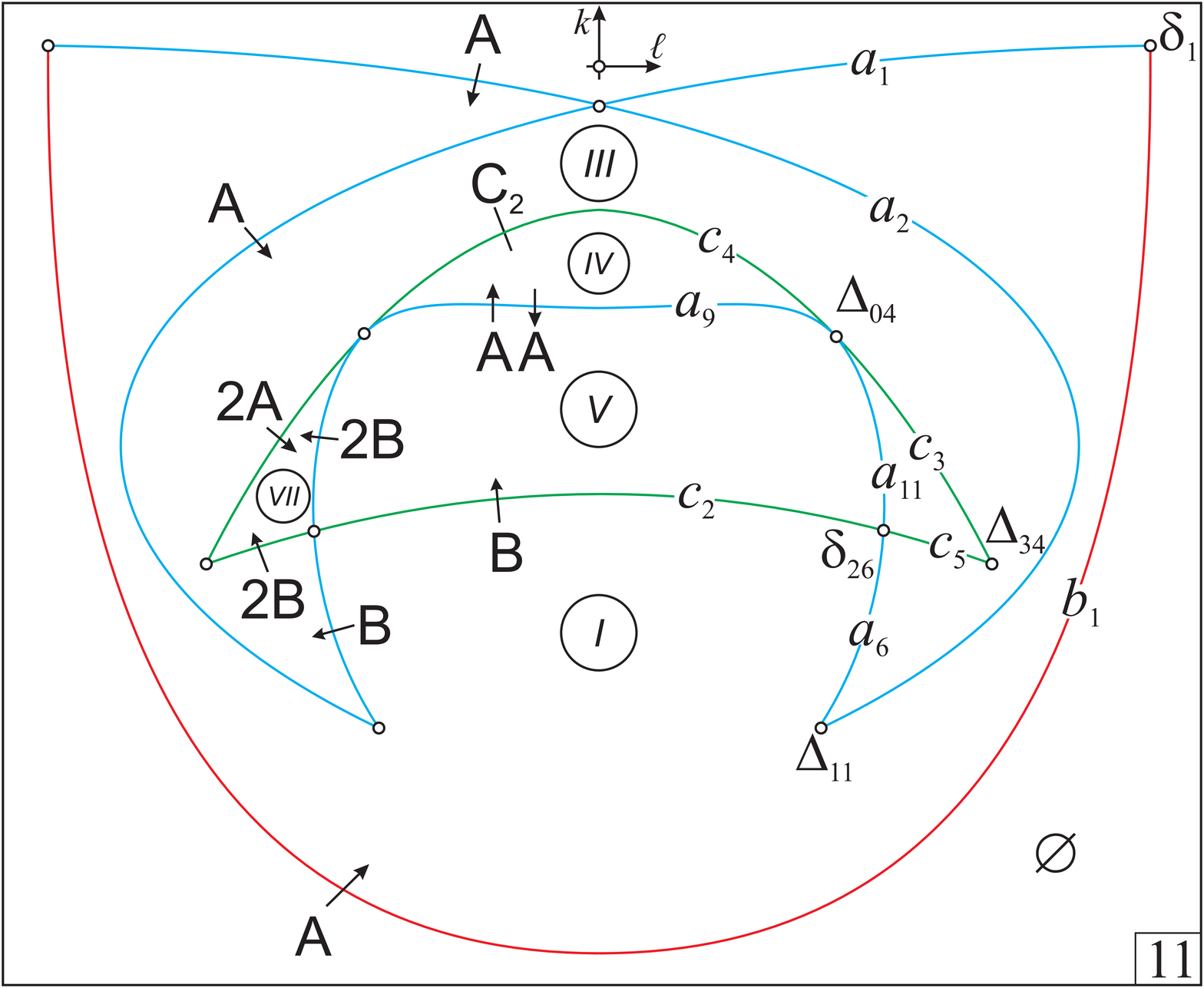}\\
\includegraphics[width=\wid\textwidth, keepaspectratio]{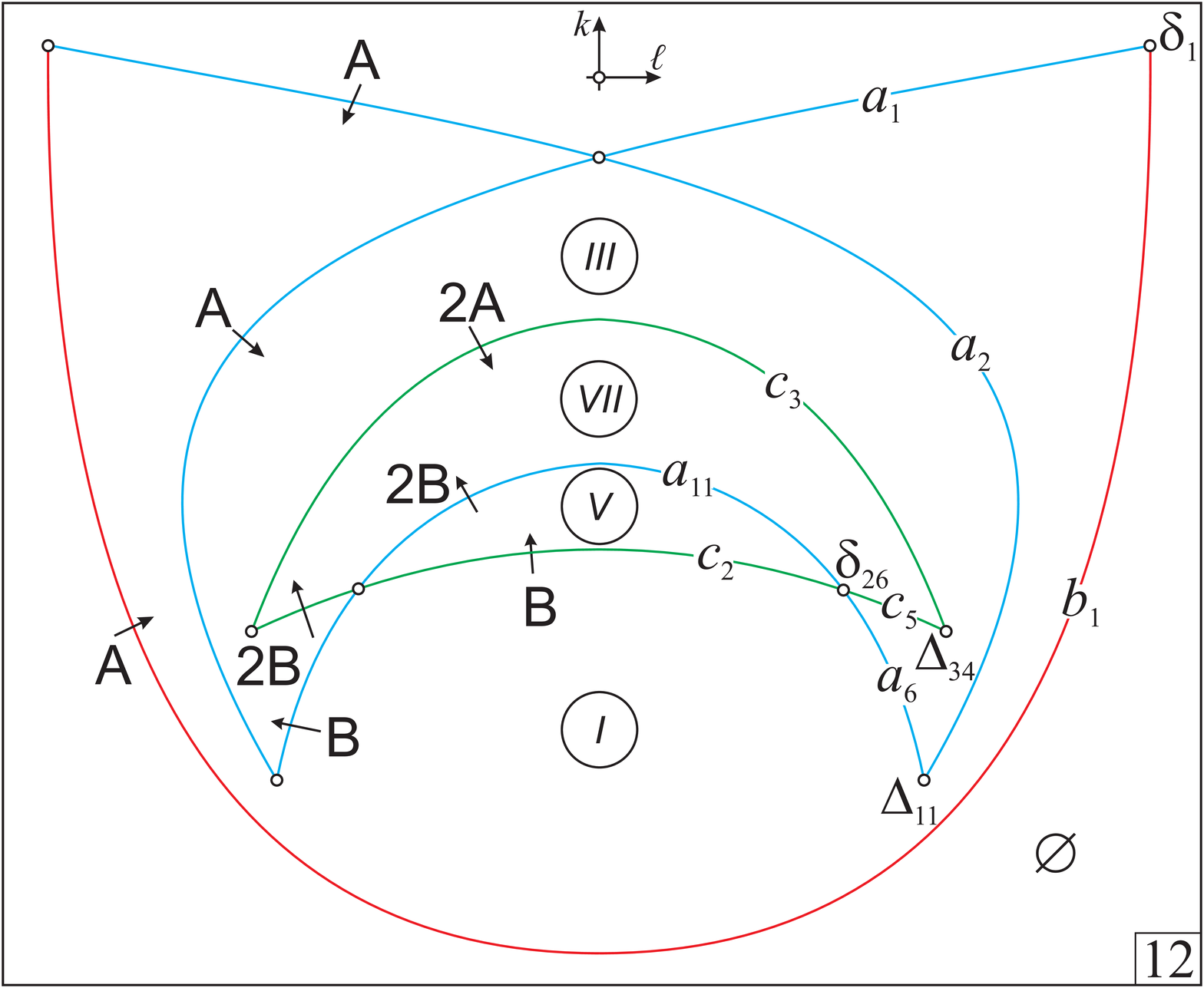}
\caption{Оснащенные диаграммы $\mSash$ (продолжение).}\label{fig_reg10}
\end{figure}

\begin{figure}[!htp]
\centering
\includegraphics[width=\wid\textwidth, keepaspectratio]{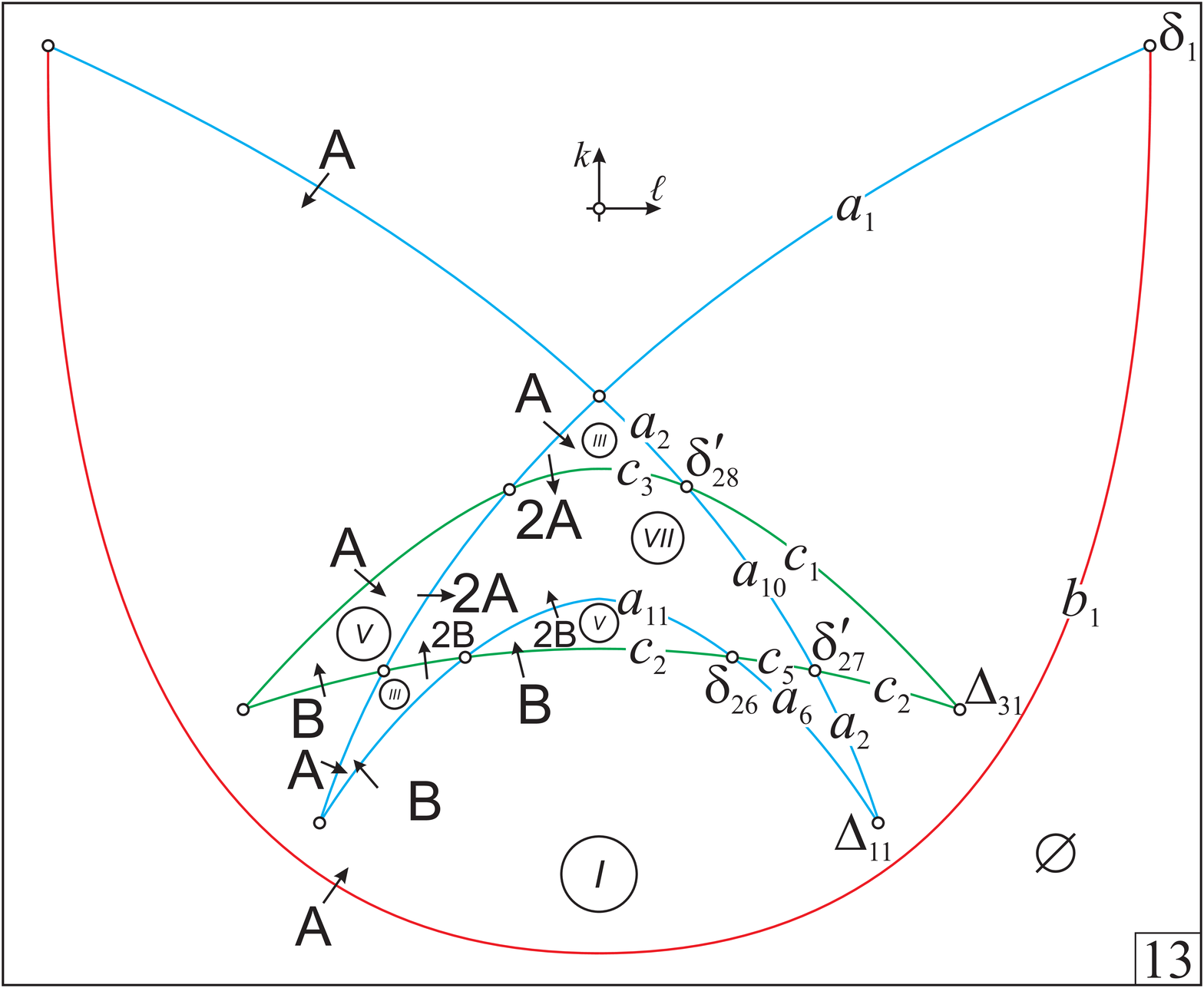}\\ \includegraphics[width=\wid\textwidth, keepaspectratio]{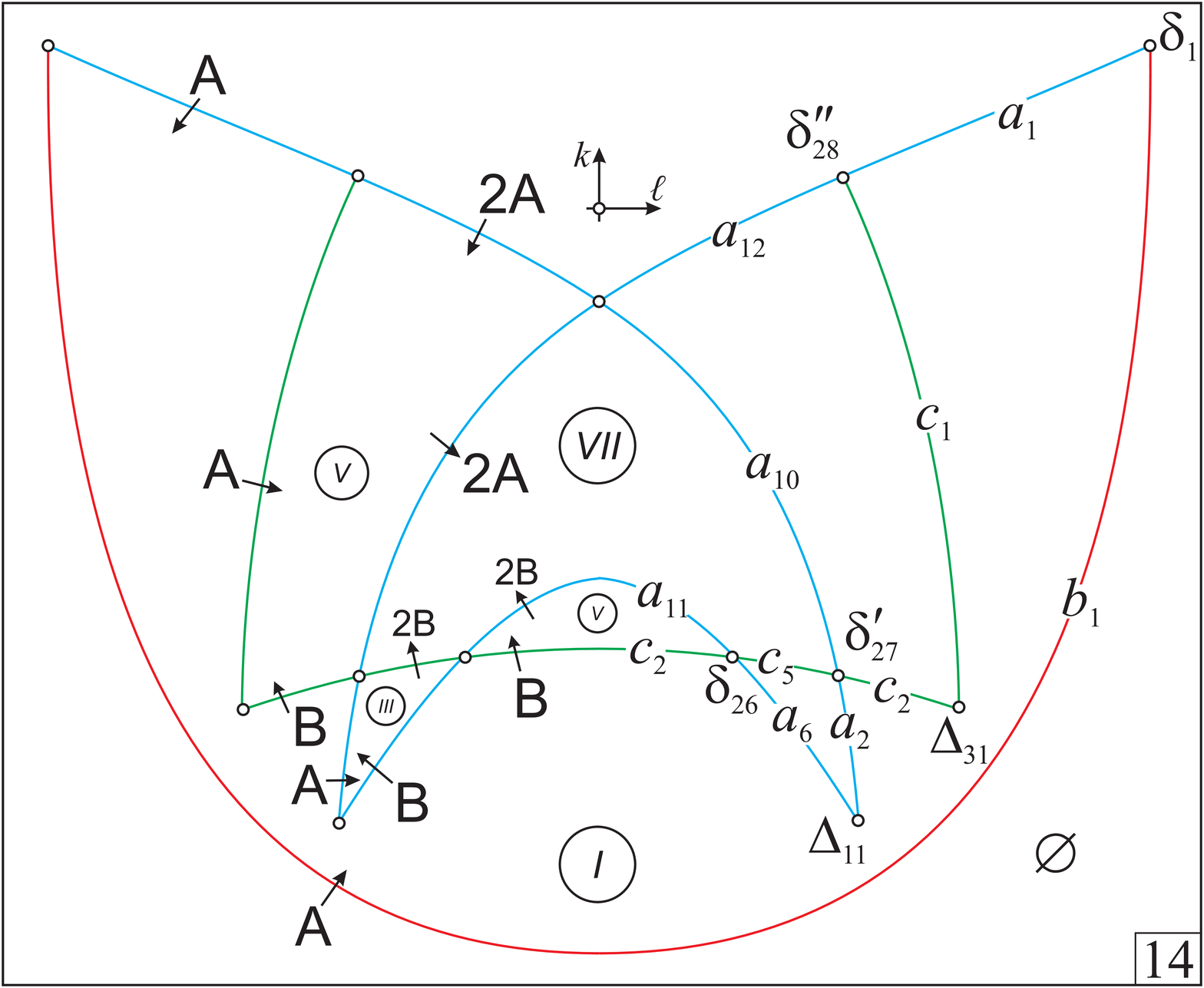}\\
\includegraphics[width=\wid\textwidth, keepaspectratio]{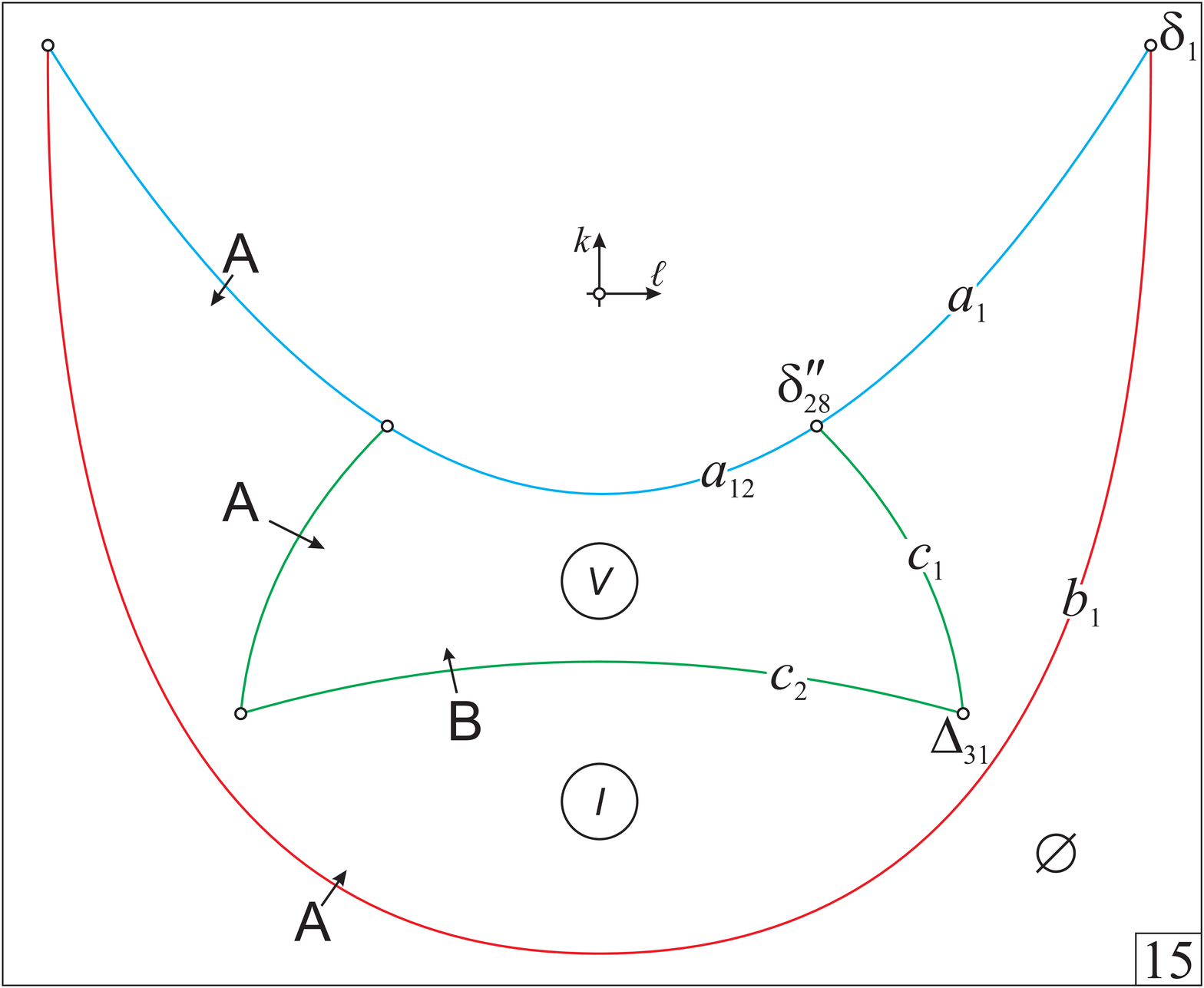}
\caption{Оснащенные диаграммы $\mSash$ (продолжение).}\label{fig_reg13}
\end{figure}

\begin{figure}[!htp]
\centering
\includegraphics[width=\wid\textwidth, keepaspectratio]{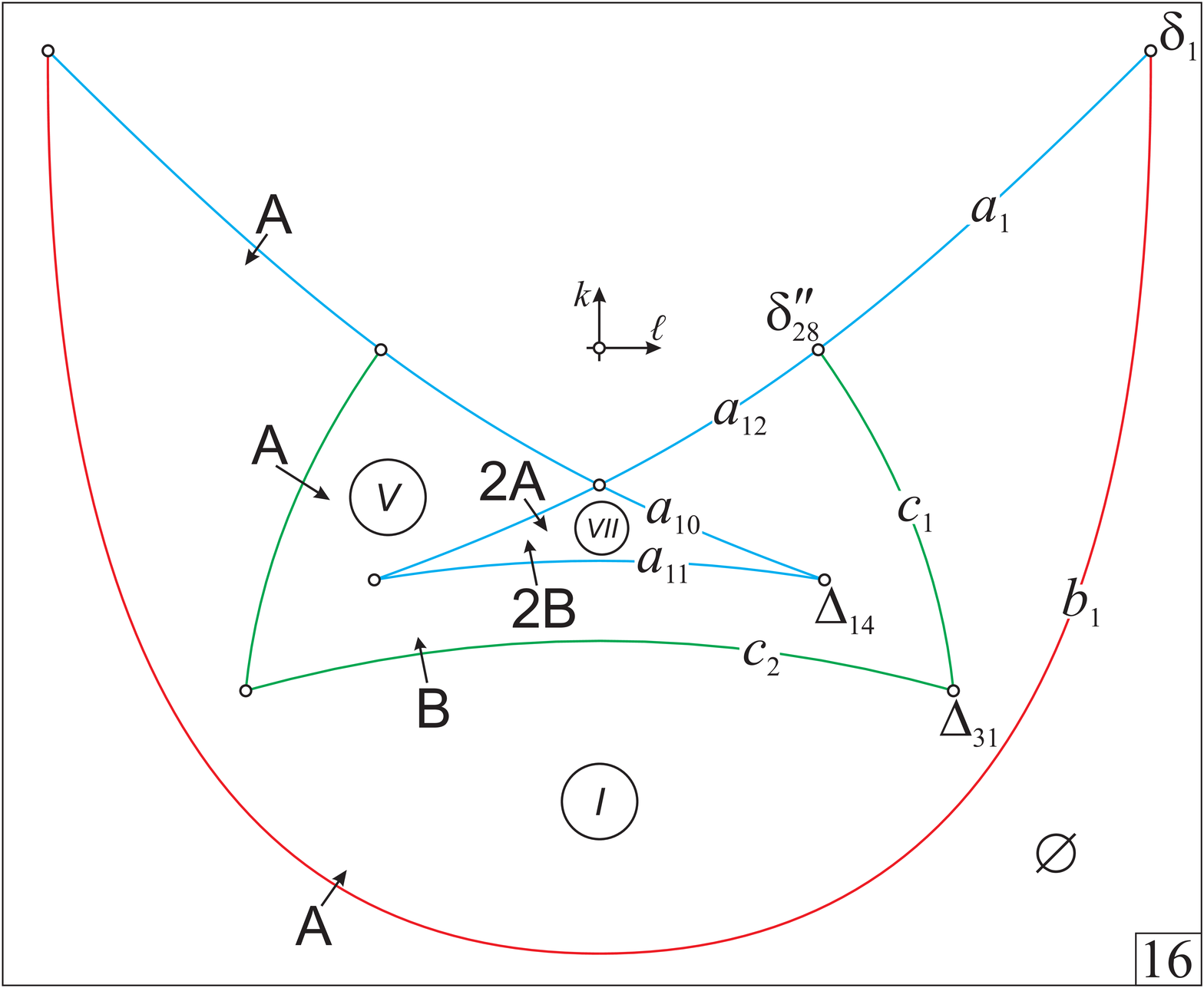}\\ \includegraphics[width=\wid\textwidth, keepaspectratio]{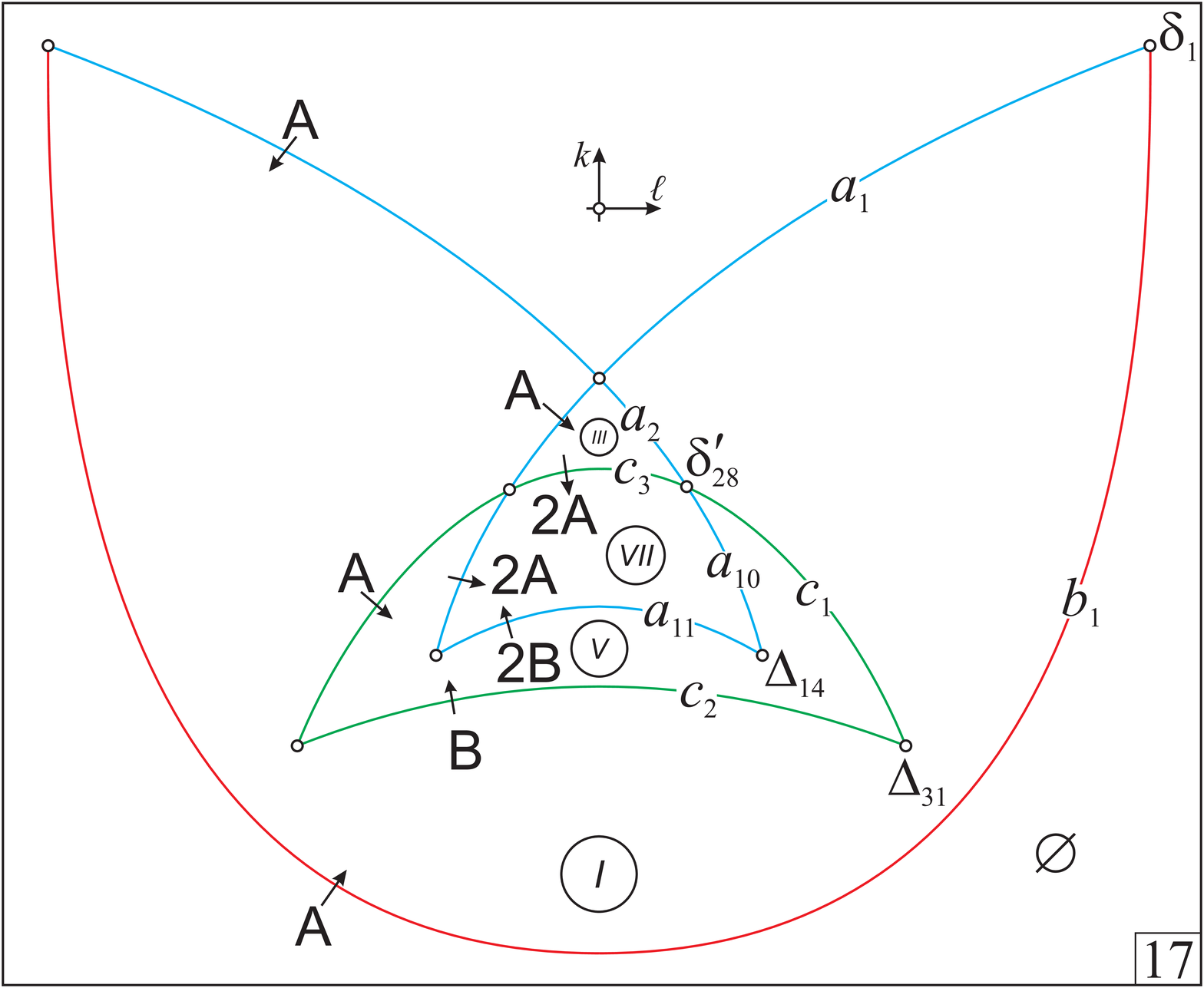}\\
\includegraphics[width=\wid\textwidth, keepaspectratio]{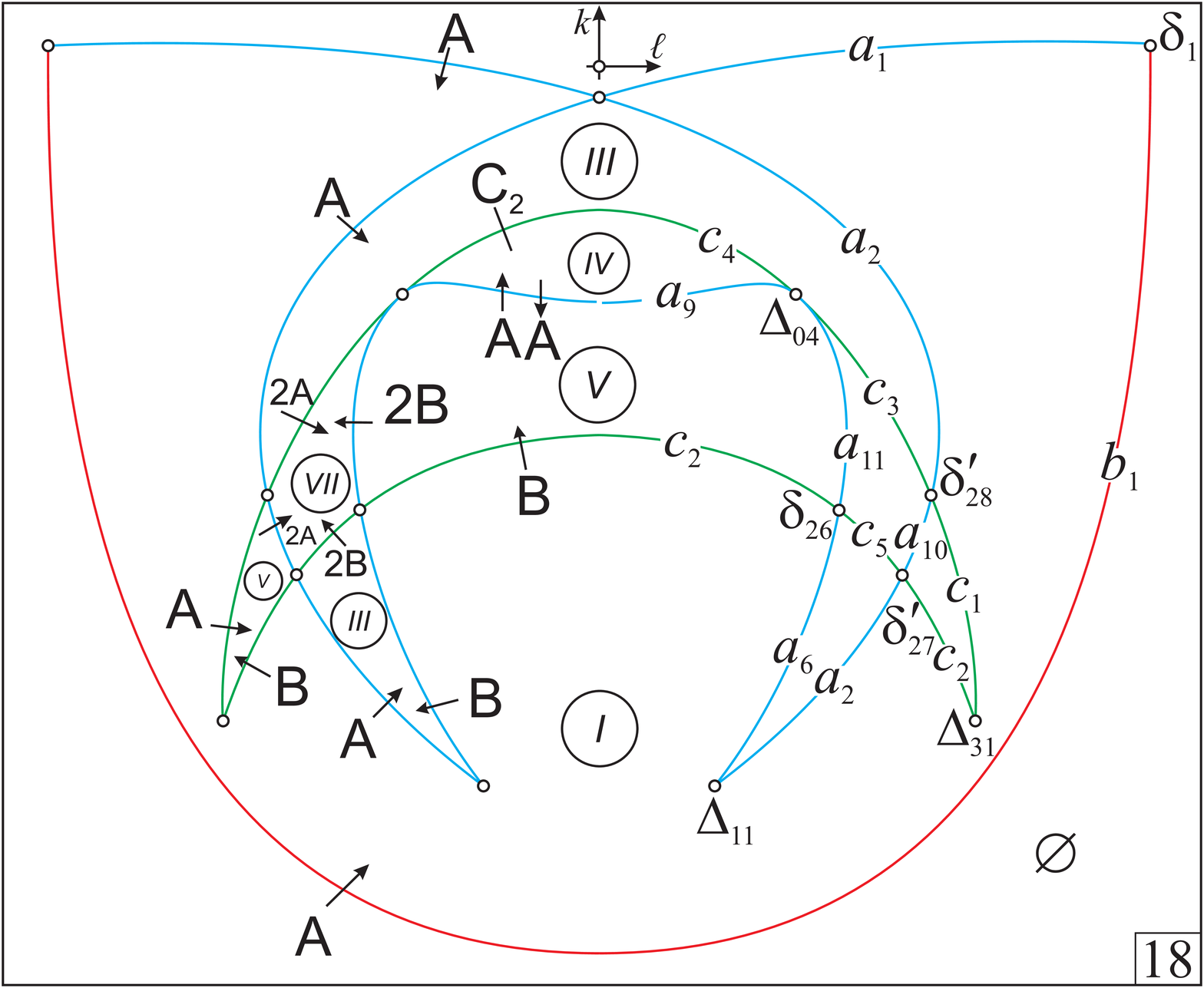}
\caption{Оснащенные диаграммы $\mSash$ (продолжение).}\label{fig_reg16}
\end{figure}

\begin{figure}[!htp]
\centering
\includegraphics[width=\wid\textwidth, keepaspectratio]{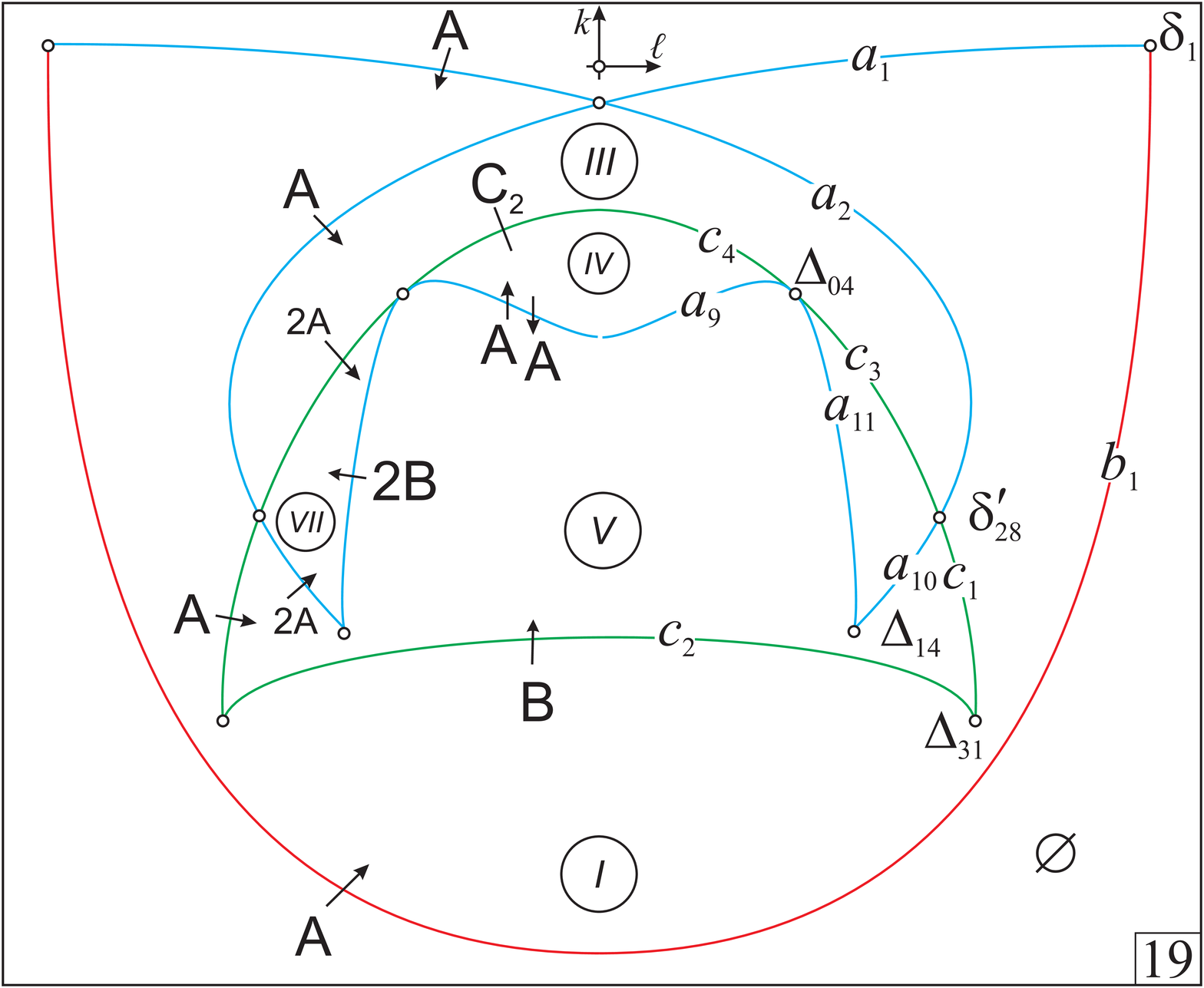}\\ \includegraphics[width=\wid\textwidth, keepaspectratio]{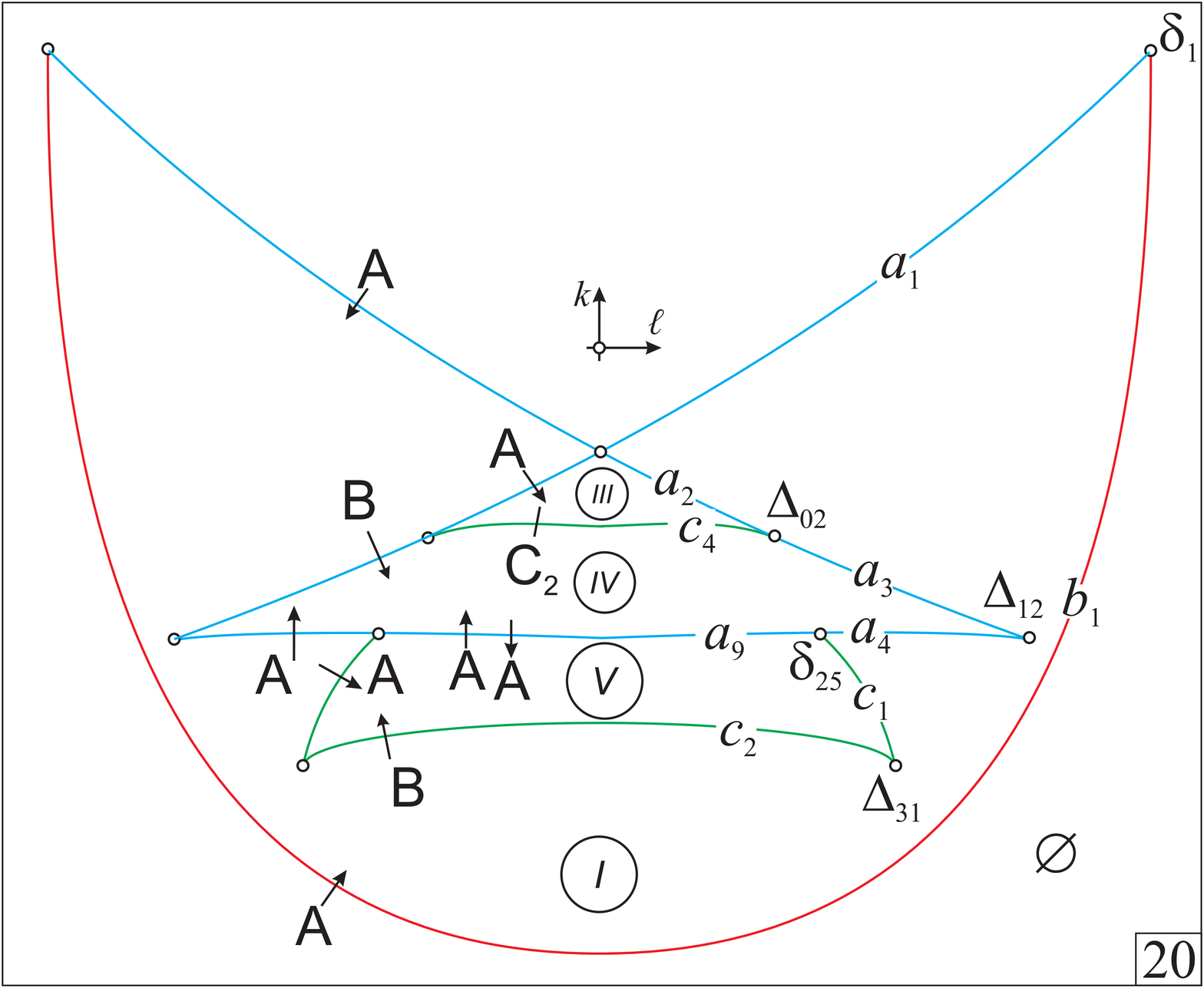}\\
\includegraphics[width=\wid\textwidth, keepaspectratio]{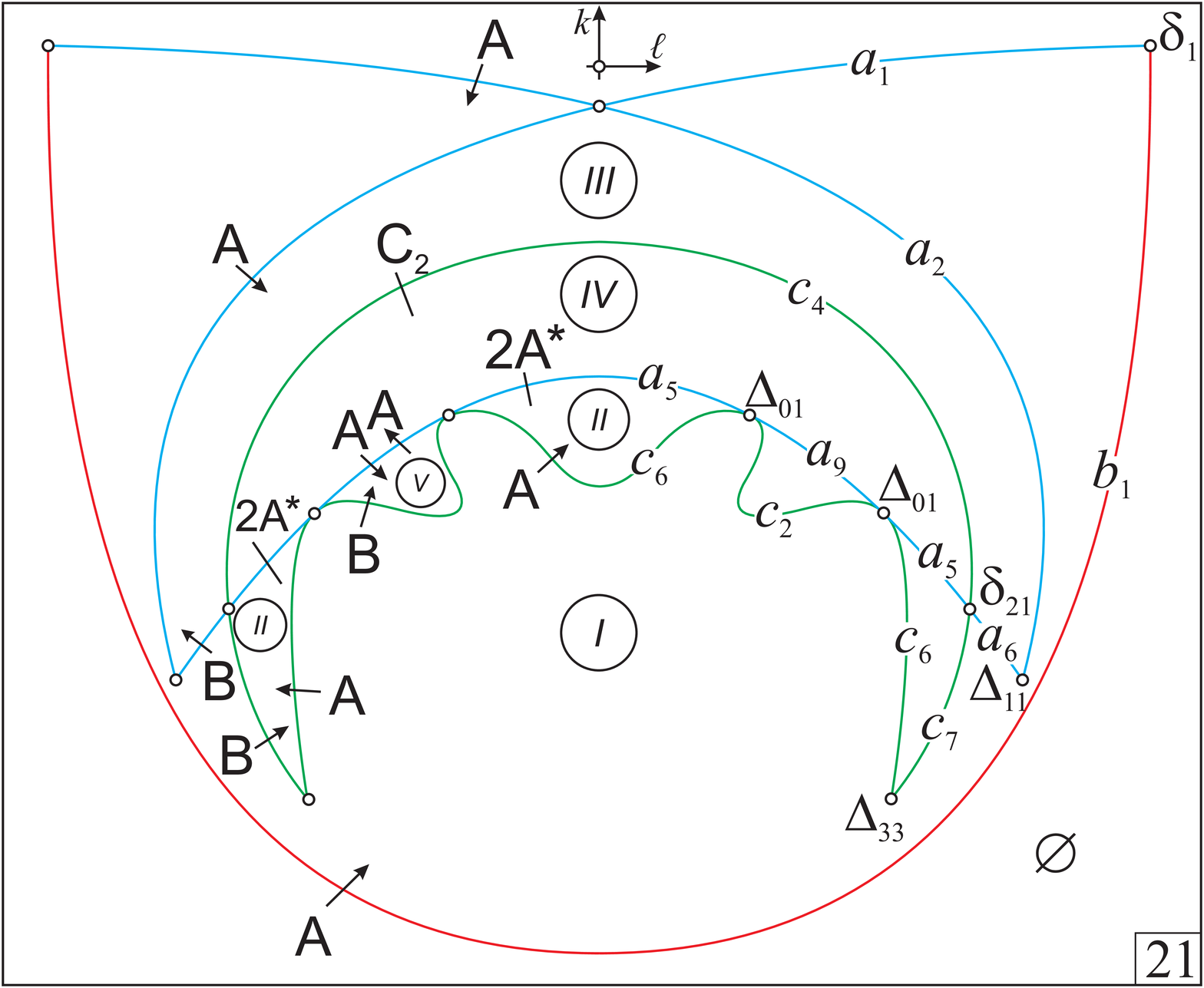}
\caption{Оснащенные диаграммы $\mSash$ (продолжение).}\label{fig_reg19}
\end{figure}

\begin{figure}[!htp]
\centering
\includegraphics[width=\wid\textwidth, keepaspectratio]{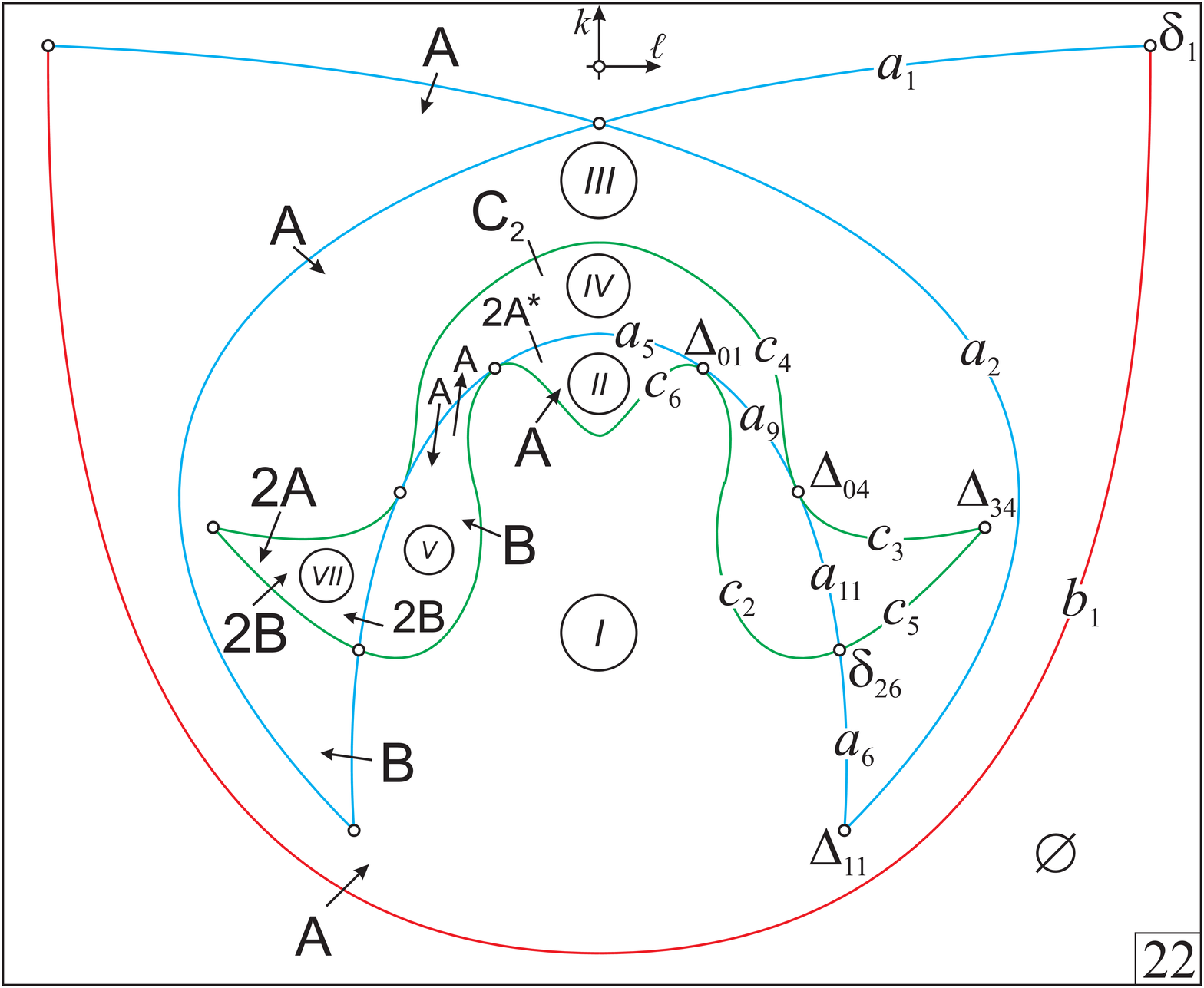}\\ \includegraphics[width=\wid\textwidth, keepaspectratio]{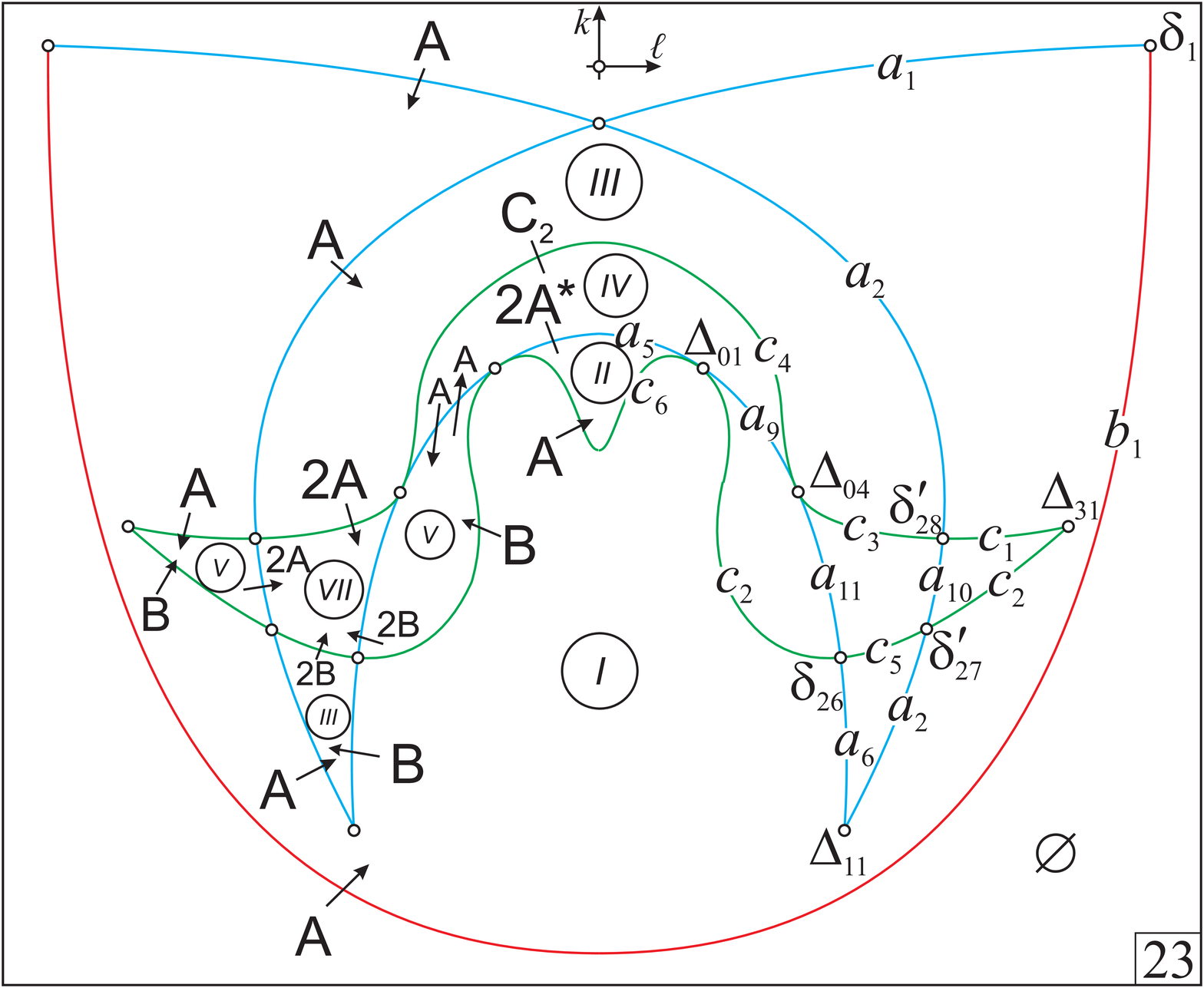}\\
\includegraphics[width=\wid\textwidth, keepaspectratio]{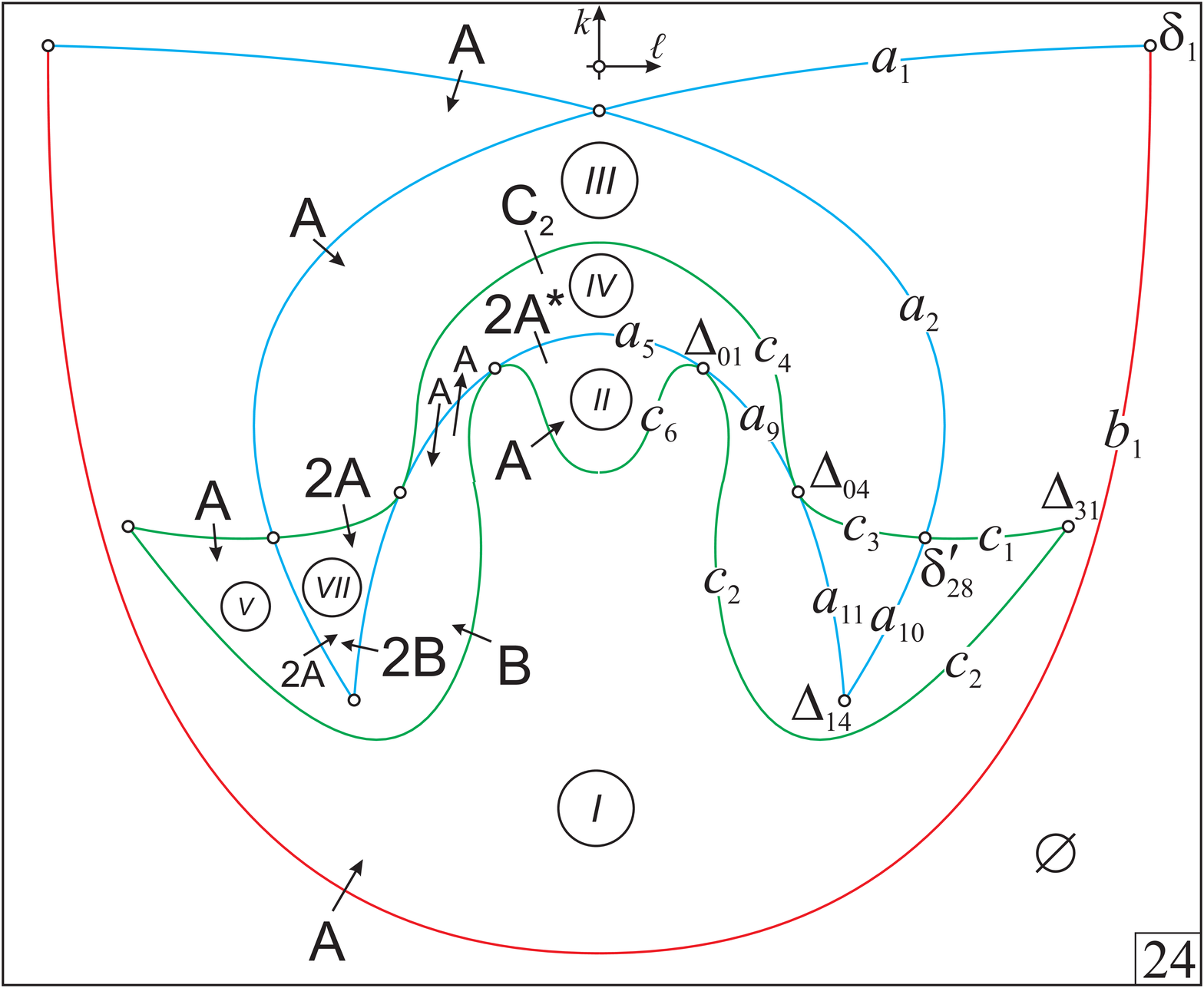}
\caption{Оснащенные диаграммы $\mSash$ (продолжение).}\label{fig_reg22}
\end{figure}

\begin{figure}[!htp]
\centering
\includegraphics[width=\wid\textwidth, keepaspectratio]{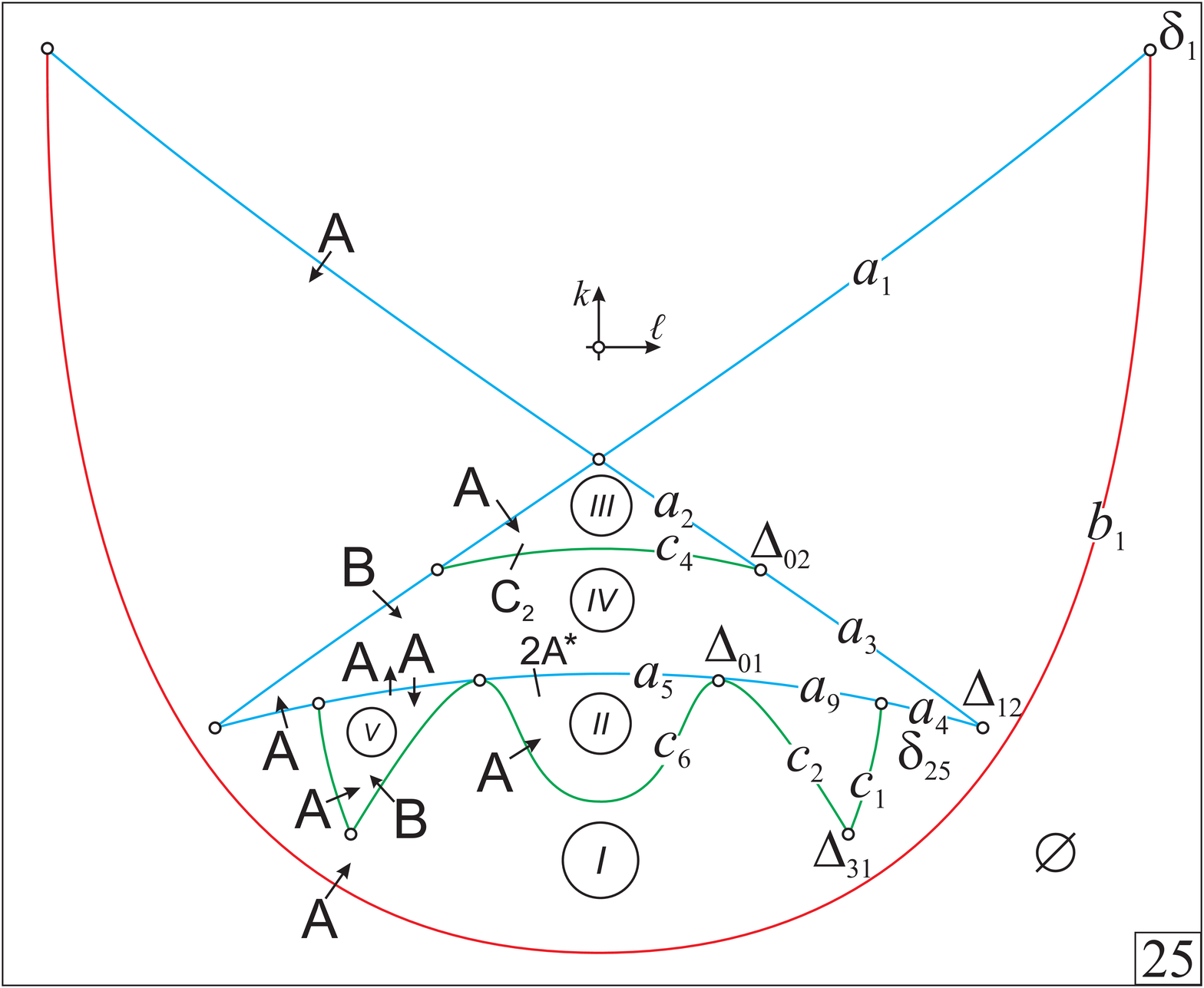}\\ \includegraphics[width=\wid\textwidth, keepaspectratio]{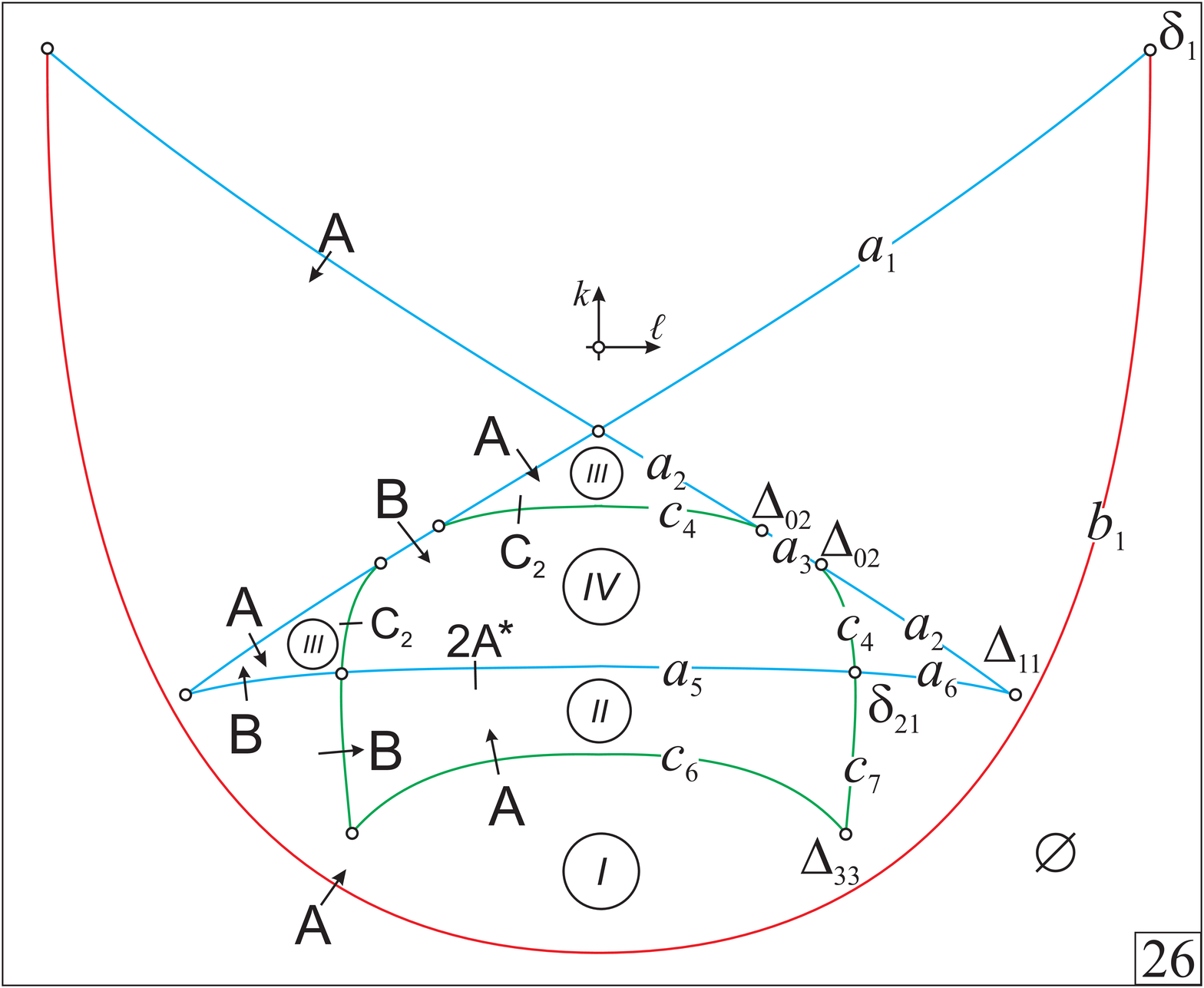}\\
\includegraphics[width=\wid\textwidth, keepaspectratio]{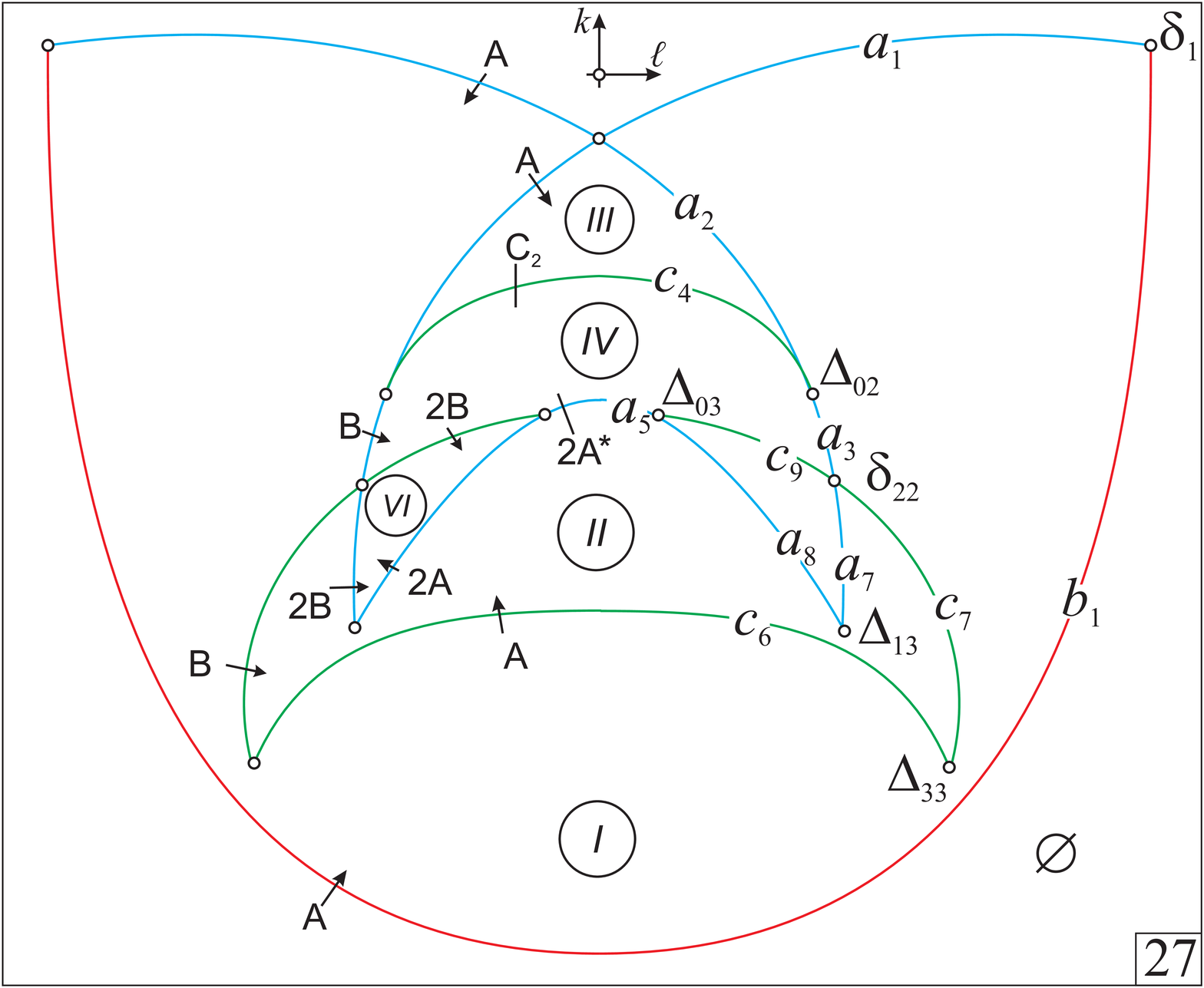}
\caption{Оснащенные диаграммы $\mSash$ (продолжение).}\label{fig_reg25}
\end{figure}

\begin{figure}[!htp]
\centering
\includegraphics[width=\wid\textwidth, keepaspectratio]{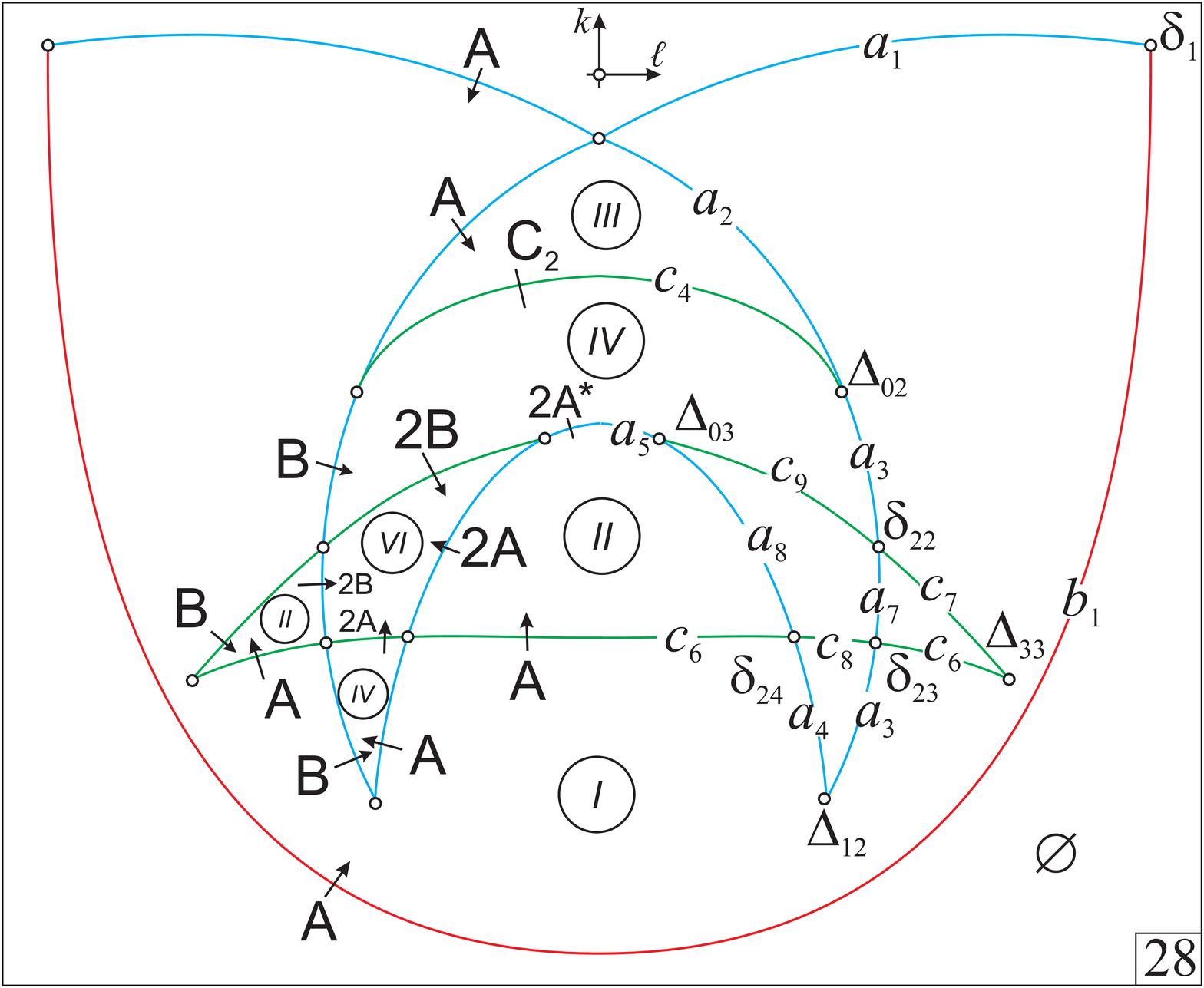}\\ \includegraphics[width=\wid\textwidth, keepaspectratio]{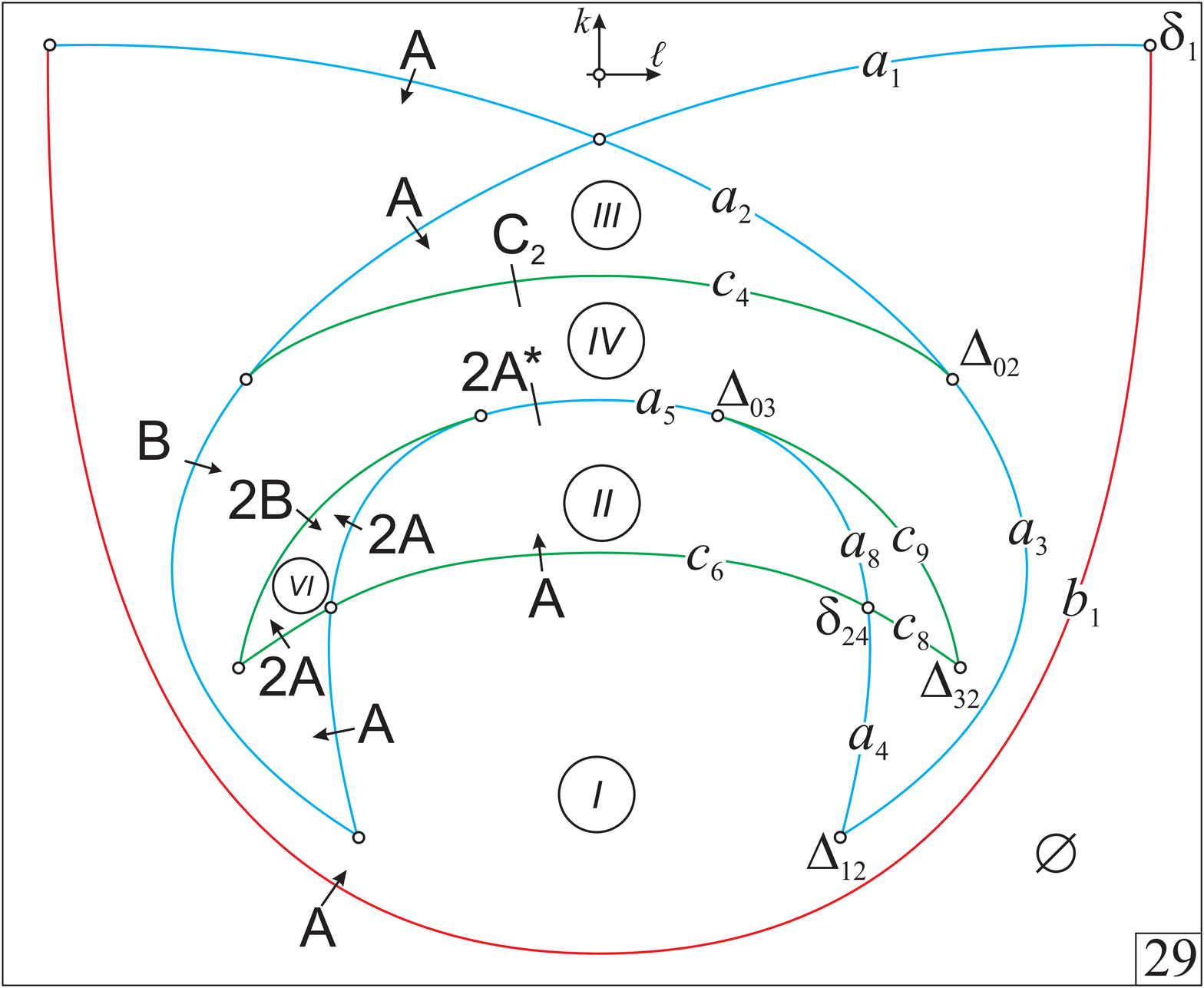}\\
\includegraphics[width=\wid\textwidth, keepaspectratio]{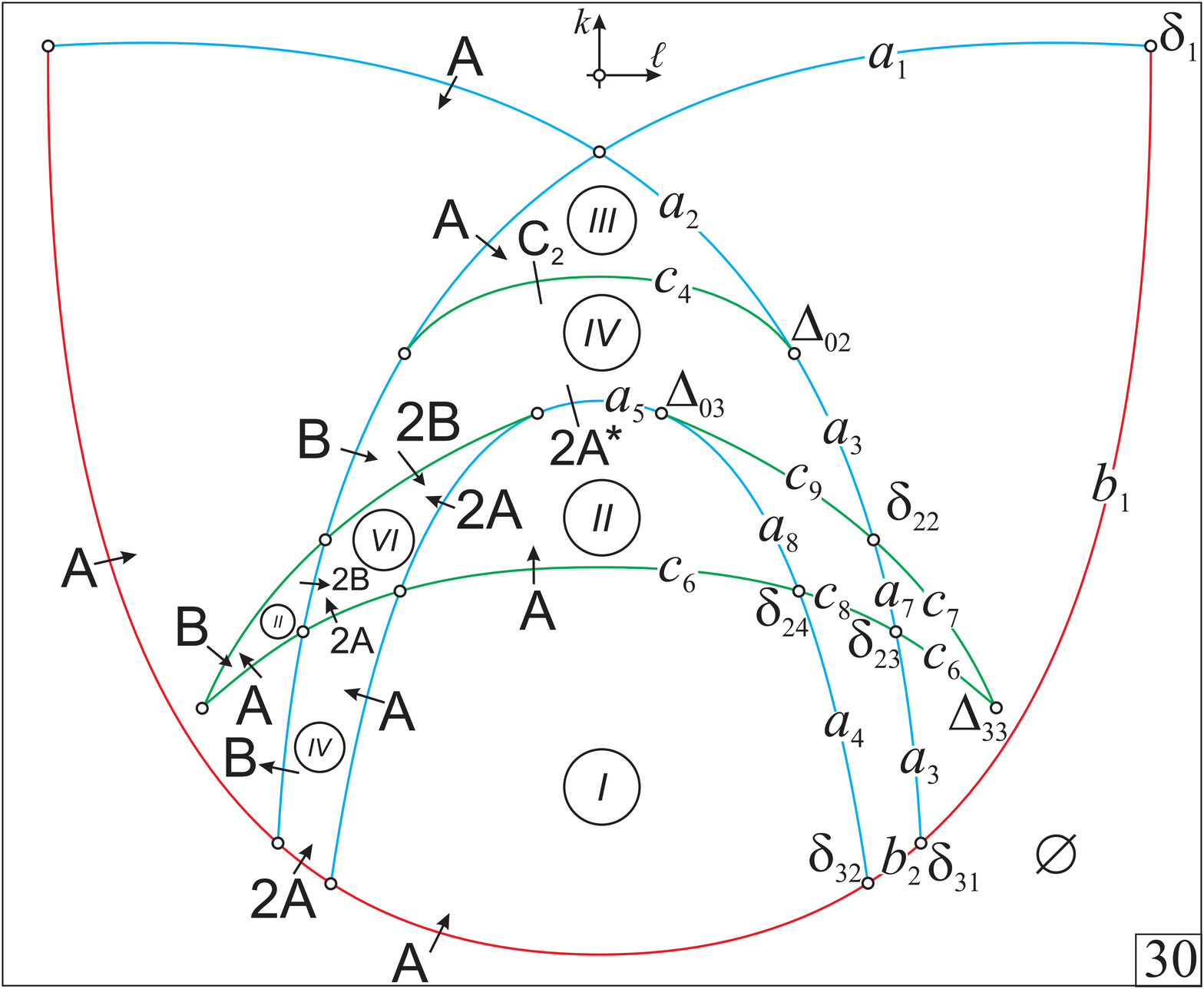}
\caption{Оснащенные диаграммы $\mSash$ (продолжение).}\label{fig_reg28}
\end{figure}

\begin{figure}[!htp]
\centering
\includegraphics[width=\wid\textwidth, keepaspectratio]{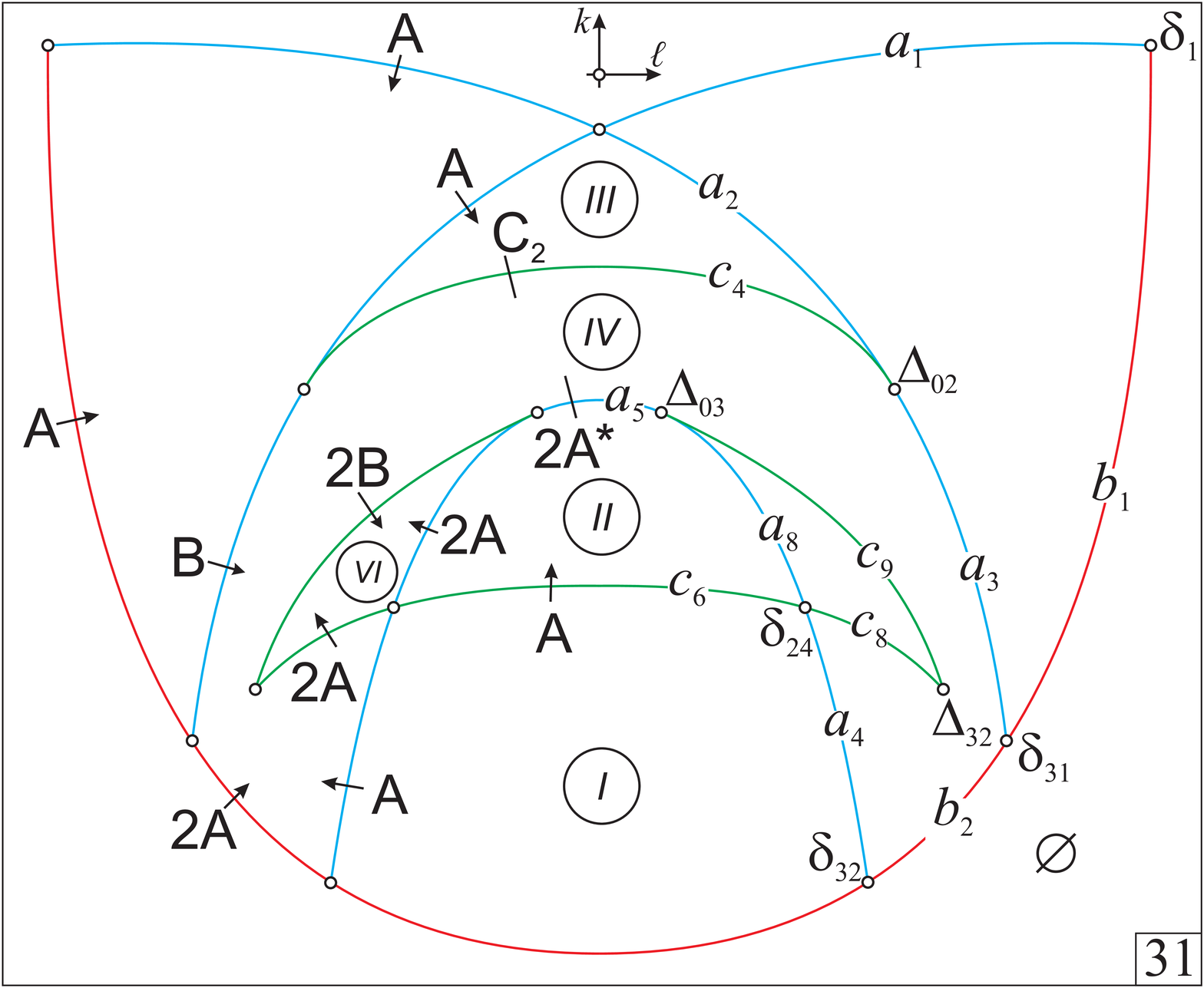}\\ \includegraphics[width=\wid\textwidth, keepaspectratio]{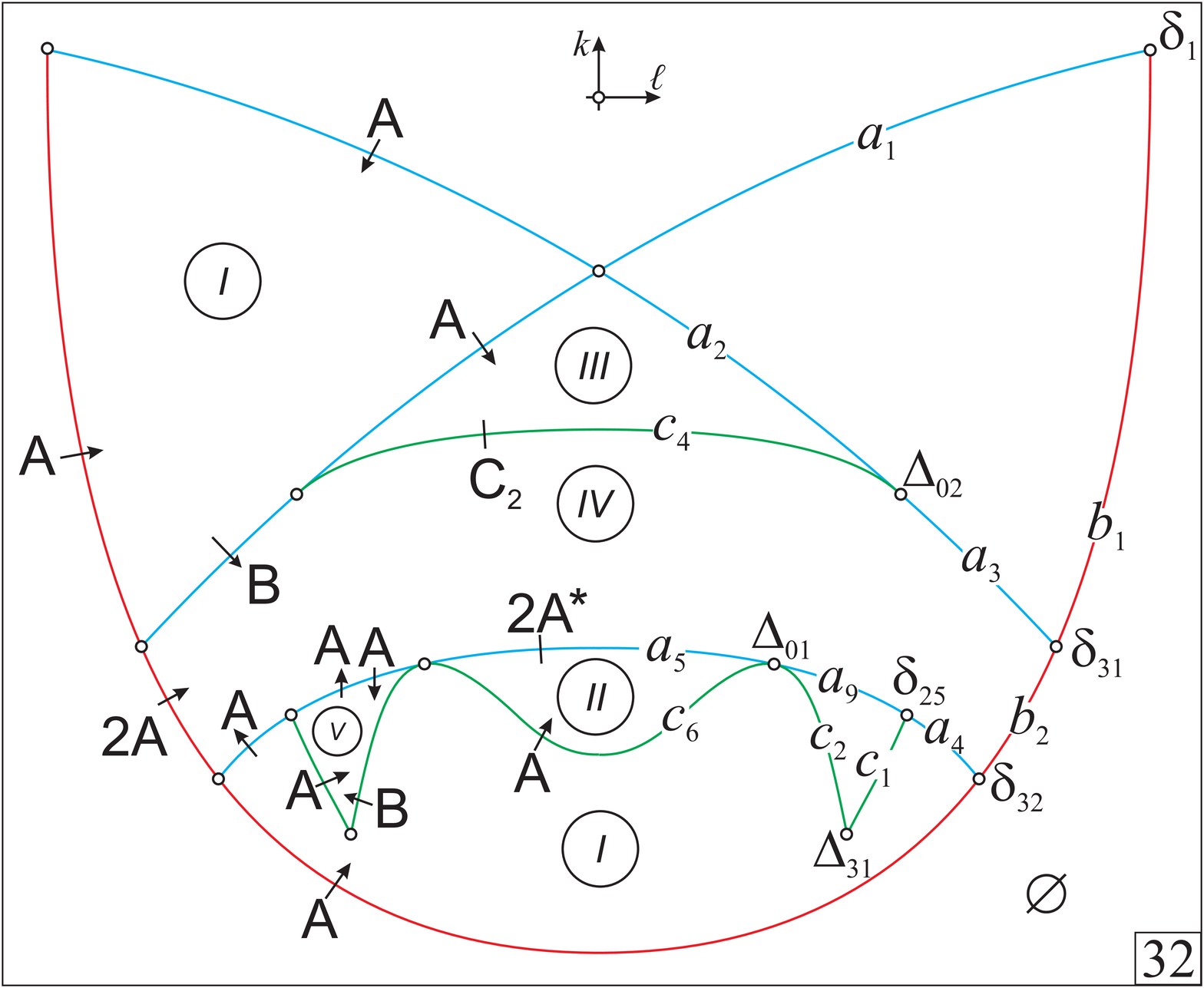}\\
\includegraphics[width=\wid\textwidth, keepaspectratio]{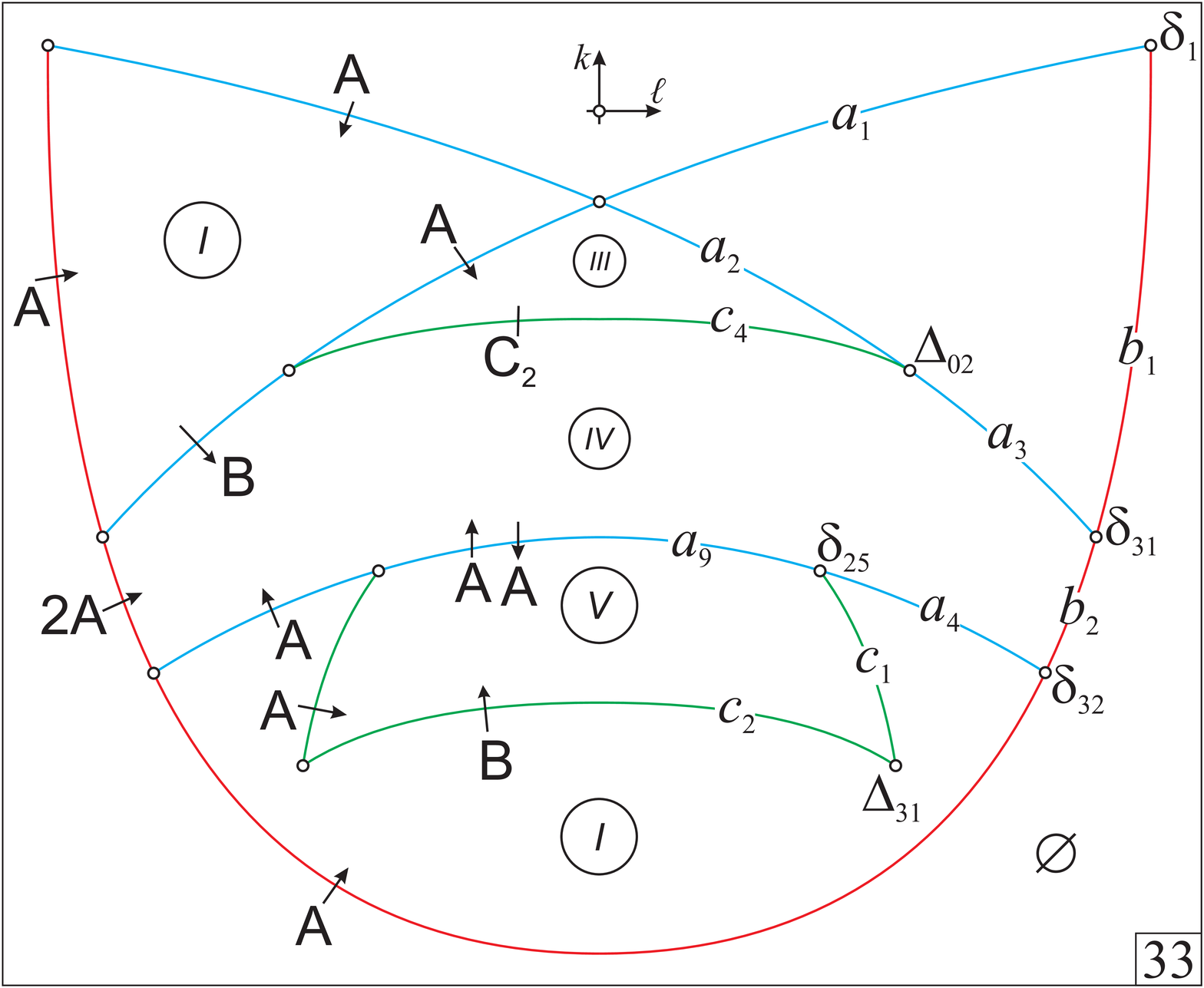}
\caption{Оснащенные диаграммы $\mSash$ (окончание).}\label{fig_reg31}
\end{figure}


\subsection{Классический случай Ковалевской}
Остановимся отдельно на предельном случае $\ld=0$ (классический случай Ковалевской).
Как видно из рис.~\ref{fig_atlas_all_xg} и \ref{fig_fragm1numreg}, разделяющими значениями $h$ при $\ld=0$ служат
\begin{equation*}
 - 1;\;0;\;1;\;\sqrt 2 ;\;\frac{3}{2};\;\sqrt 3 ;\;2.
\end{equation*}
По видимому, впервые $(L,K)$-диаграммы для случая Ковалевской рассматривались в школе П.\,Рихтера, где уже в 90-е годы для визуализации диаграмм, бифуркаций и графов Фоменко использовались мощные графические процессоры. В результате был снят фильм, описывающий все эти явления и демонстрирующий типичные движения тела для различных семейств двумерных торов приведенной системы. Математическое описание оснований для фильма, базирующееся на результатах \cite{KhPMM83,KhBook88,BRF}, представлено в работе \cite{RiDuWi1997}. В ней, в частности, приведены с детализацией точные изоэнергетические диаграммы, которые, для решения проблемы малых областей, рассчитывались на плоскости $(\ell,\sqrt{k})$. Разделяющие значения $h$ в работе явно не перечислены, но вытекают из приведенных там ранее формул. Эти же значения получены как предельные при исследовании изоэнергетических диаграмм волчка Ковалевской в двойном поле \cite{KhSh2004}.
Диаграммы $\mSash$, соответствующие неразделяющим значениям $h$ (выход на ось $\ld=0$ из областей $1,3,9,26,27,30,31$) показаны на рис.~\ref{fig_reg010} -- \ref{fig_reg270}. Здесь, за счет больших искажений по вертикали, мы сохранили послойный (сохраняющий сечения $\ell=\cons$) диффеоморфизм области диаграммы и в достаточно сложных случаях $27_0,30_0,31_0$.

Отметим еще одну интересную особенность. В работе \cite{RiDuWi1997} для области $30_0$ приведены две диаграммы. Очевидно, соответствующие системы на $Q_h^4$ лиувиллево эквивалентны, однако, они имеют разные наборы графов Фоменко, так как точка $\delta_{22}$ может проецироваться на ось $Oh$ как на ребро $\ccc_6+\bbb_1$, так и на ребро $\ccc_8+\bbb_2$.

\def\wid{0.3}

\begin{figure}[!htp]
\centering
\includegraphics[width=\wid\textwidth, keepaspectratio]{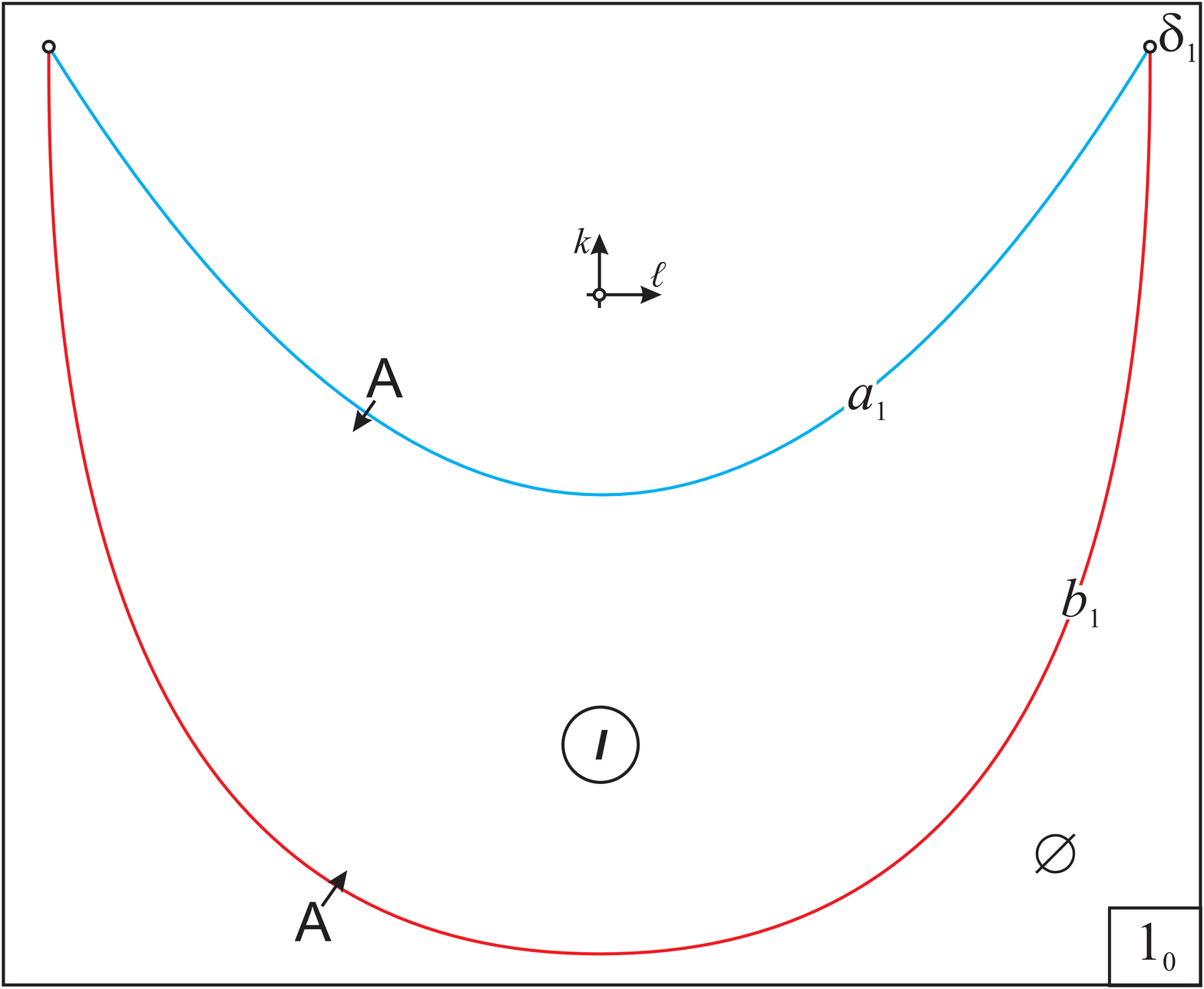}\  \includegraphics[width=\wid\textwidth, keepaspectratio]{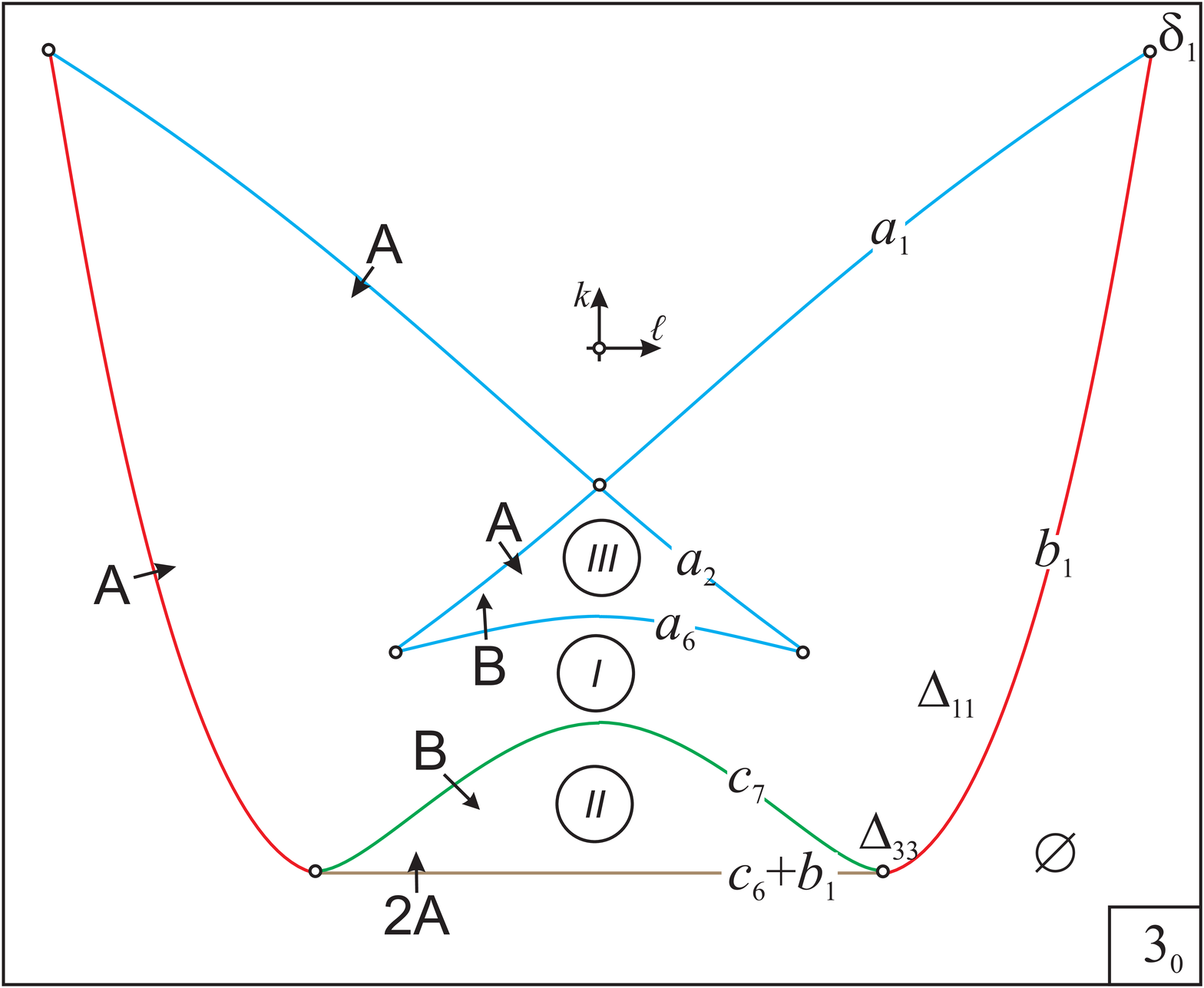}
\caption{Оснащенные диаграммы $\Sigma_{LK}(h,0)$ (случай Ковалевской).}\label{fig_reg010}
\end{figure}

\begin{figure}[!htp]
\centering
\includegraphics[width=\wid\textwidth, keepaspectratio]{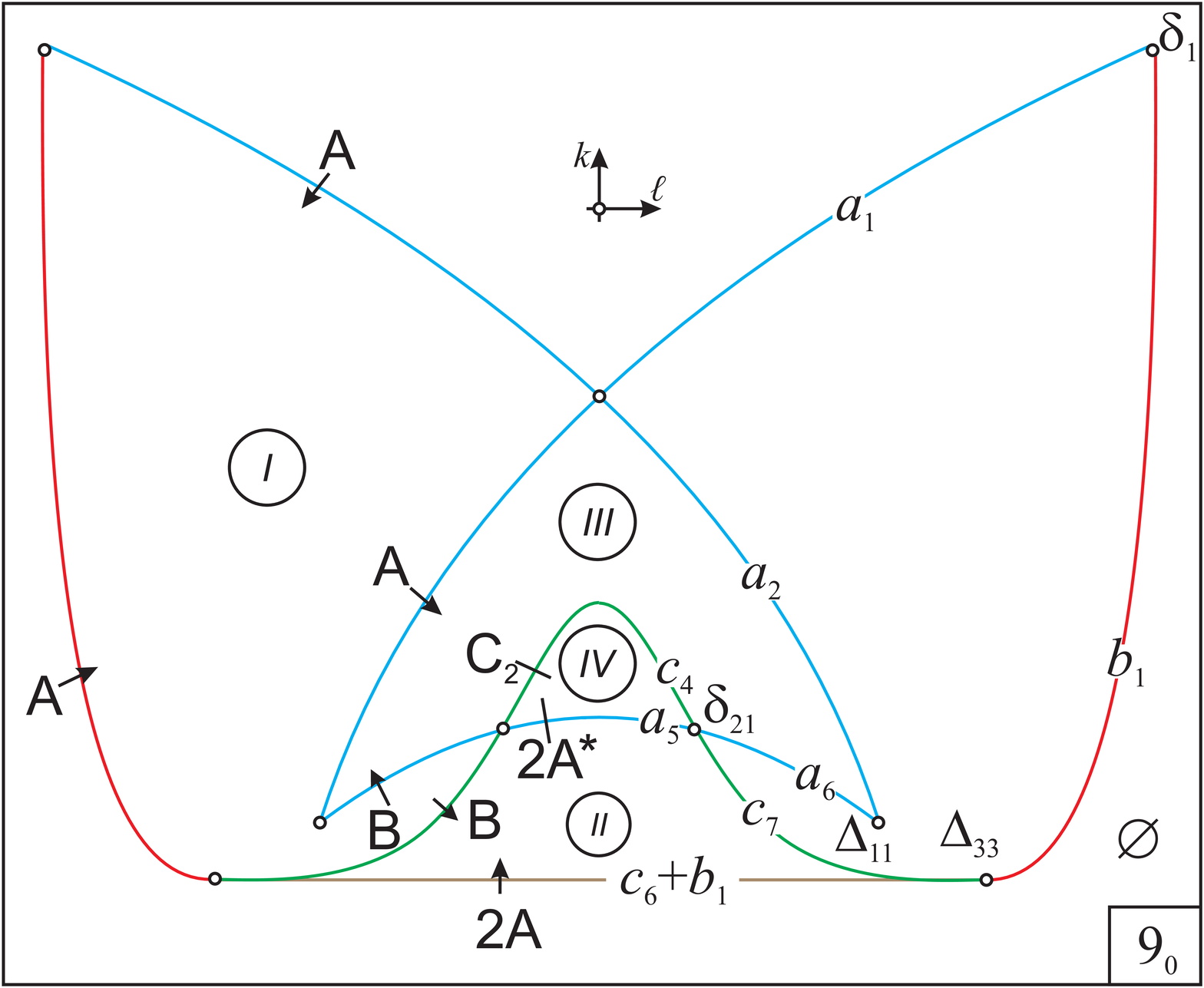}\
\includegraphics[width=\wid\textwidth, keepaspectratio]{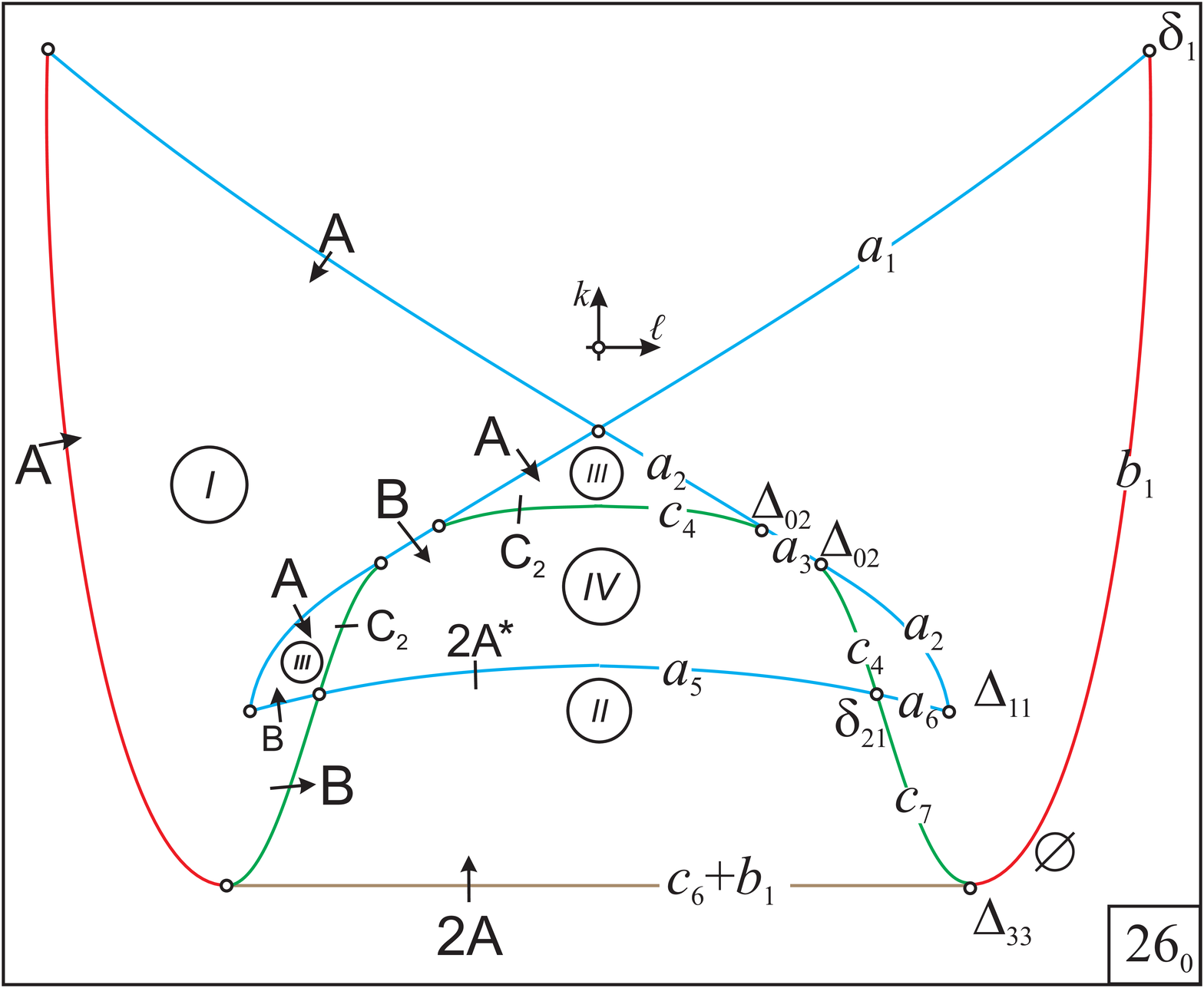}
\caption{Оснащенные диаграммы $\Sigma_{LK}(h,0)$ (продолжение).}\label{fig_reg090}
\end{figure}

\begin{figure}[!htp]
\centering
\includegraphics[width=\wid\textwidth, keepaspectratio]{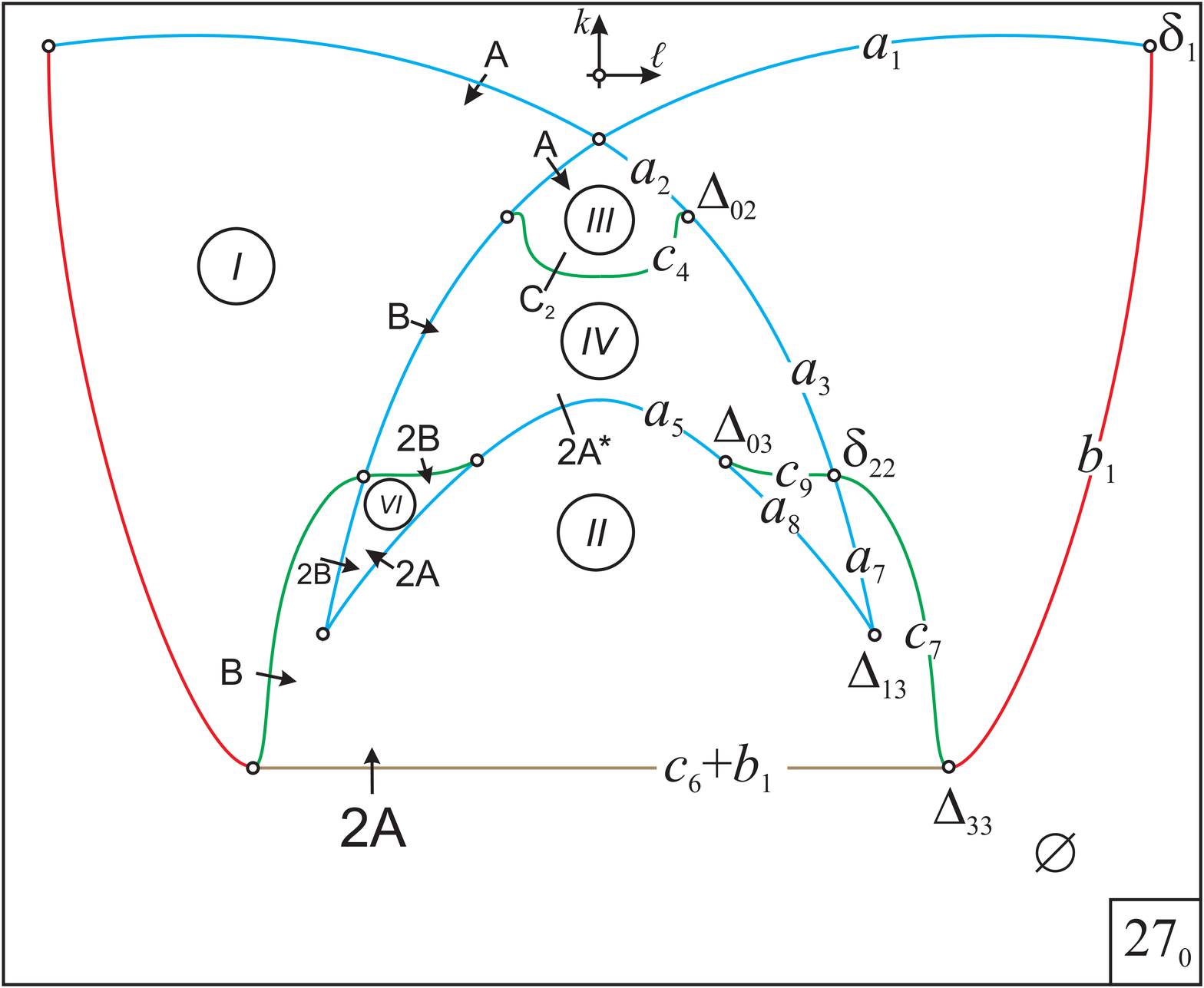}\
\includegraphics[width=\wid\textwidth, keepaspectratio]{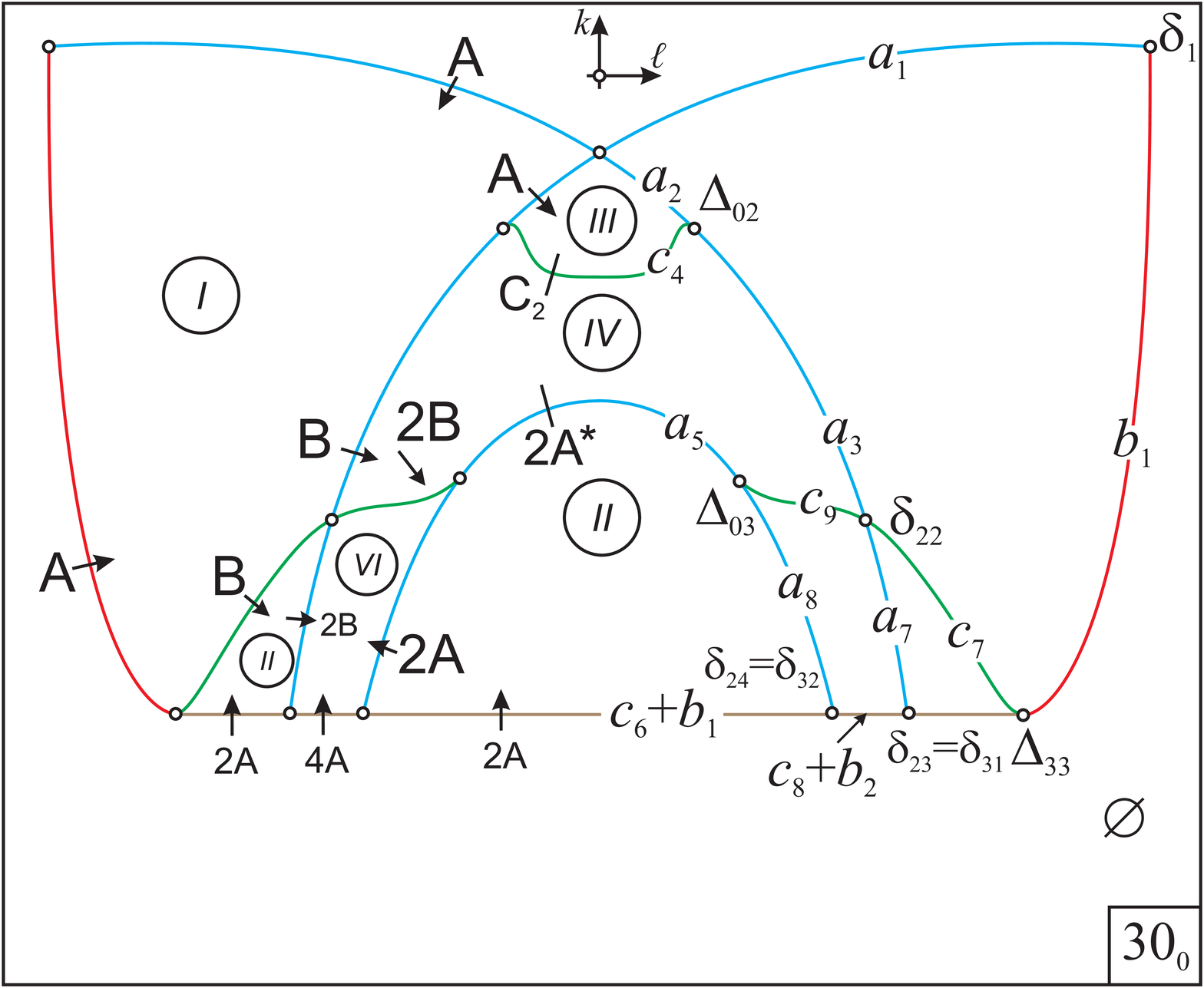}\\
\includegraphics[width=\wid\textwidth, keepaspectratio]{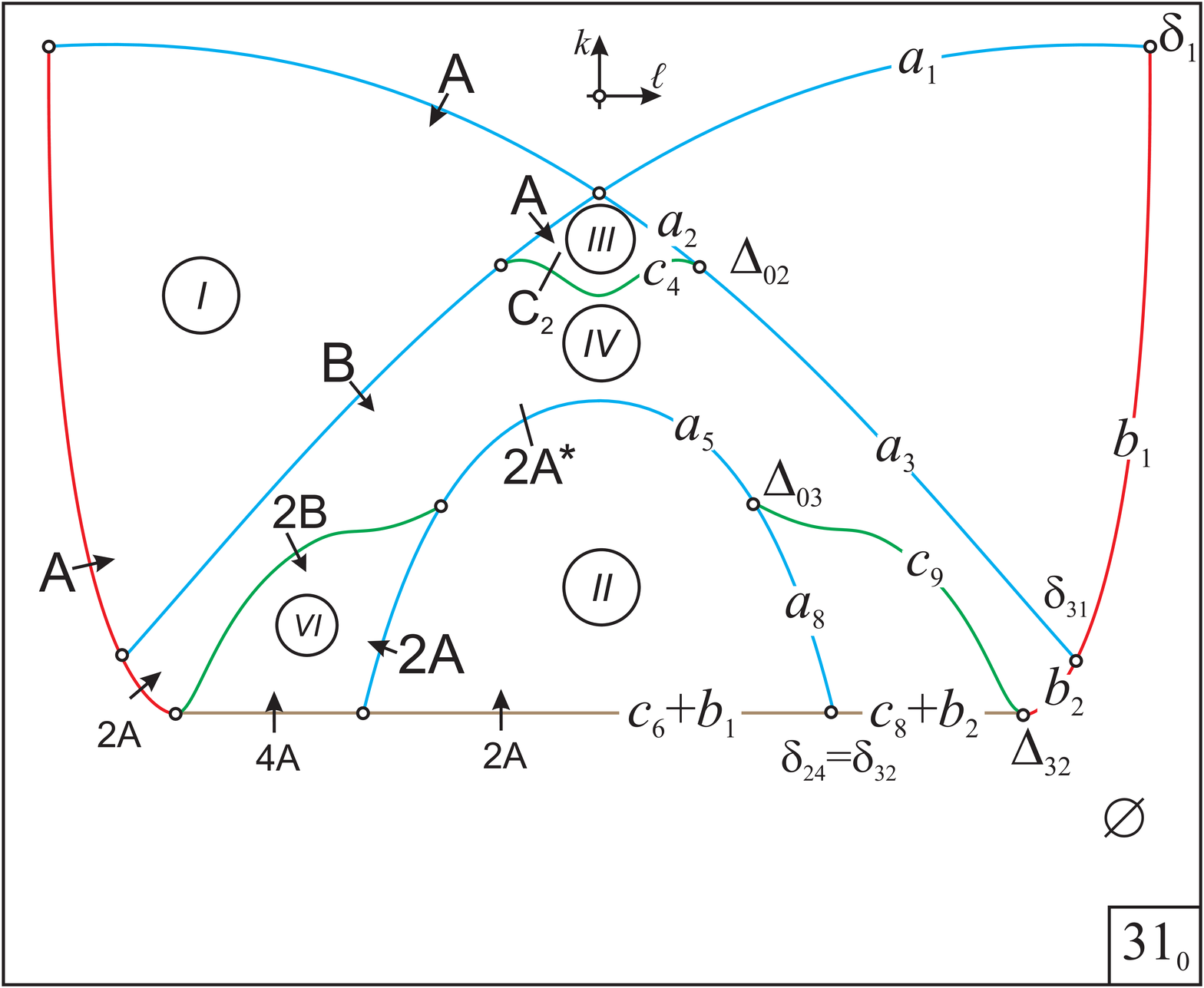}
\caption{Оснащенные диаграммы $\Sigma_{LK}(h,0)$ (продолжение).}\label{fig_reg270}
\end{figure}

\clearpage

\section{Топологические инварианты}\label{sec9}
\subsection{Диаграммы Смейла\,--\,Фоменко}

Для того чтобы классифицировать графы Фоменко на $\iso(\ld)$, нужно на плоскость с диаграммой Смейла наложить проекцию образа множества вырожденных критических точек ранга~$1$, так как именно переход через такие точки определяет перестройки в графах, не связанные с изменением топологии $\iso(\ld)$. Теоретически также возможно явление расщепления некоторых типов атомов без возникновения вырожденных точек (см. \cite{BolFom}), но, как показано выше путем перечисления всех возникающих в этой задаче атомов, здесь мы с этим явлением не встречаемся. Выпишем уравнения образа множества вырожденных критических точек ранга~$1$ с учетом условий существования движений нужного типа в соответствии с предложениями~\ref{propos16}~и~\ref{propos21}.

В состав кривых, классифицирующих графы Фоменко, входят следующие (см. теорему \ref{th13} и предложение \ref{propos21}):

1) кривая касания поверхностей $\wsa,\wsc$ полностью
\begin{equation}\label{eq9_1}
    \gan: \quad \left\{ \begin{array}{l}
                        \ds{\ell = \pm \sqrt{\frac{s}{2}(1-2\ld^2s)}} \\
                        \ds{h=\frac{1-\ld^2s+2s^2}{2s}}
                      \end{array} \right. , \quad \ds{0<s\leqslant \frac{1}{2\ld^2}};
\end{equation}

2) ребро возврата поверхности $\wsa$ между его пересечениями с кривой $\delta_3$ (то есть между точками -- образами класса $\vpi_{31}$ вырожденных точек ранга 0, переходящими в точки $B_3$ на ключевых множествах)
\begin{equation}\label{eq9_2}
    \gaa: \quad \ds{\ell = \pm \frac{2}{3\sqrt{3}}(h-\frac{\ld^2}{2})^{3/2}}, \quad \ds{\frac{\ld^2}{2}\leqslant h \leqslant h_C(\ld)}
\end{equation}
где зависимость $h_C(\ld)$ находится согласно \eqref{eq5_6};

3) ребро возврата поверхности $\wsc$ в следующих пределах
\begin{equation}\label{eq9_3}
    \gac: \quad \ds{\ell = \pm \sqrt{\frac{h-h^*}{2}}}, \quad h \in \left\{
    \begin{array}{l}
    \ [h^*,+\infty ), \quad \textrm{при} \quad \ld \leqslant \ld^* \\
    \ [h_{**},+\infty), \quad \textrm{при} \quad \ld > \ld^*
    \end{array}      \right.,
\end{equation}
где
\begin{equation}\notag
\begin{array}{ll}
  h^* & = \ds{\frac{\ld^{2/3}}{2}(3-\ld^{4/3})}, \\[3mm]
  h^{**} & = \ds{\frac{1}{4}\left[(4+\ld^{4/3})^{3/2}-\ld^{2/3}(6+\ld^{4/3})\right]}=\\[3mm]
  {} & = \ds{\frac{1}{16}(\sqrt{4+\ld^{4/3}}-\ld^{2/3})\left[(\sqrt{4+\ld^{4/3}}-\ld^{2/3})^2+12 \right]}.
\end{array}
\end{equation}

\begin{defin}\label{def13}
Назовем диаграммой Смейла\,--\,Фоменко $\smales$ объединение $\smale$ с образом множества вырожденных точек ранга $1$ под действием отображения $L{\times}H$.
\end{defin}

Такие диаграммы впервые строил А.А.\,Ошем\-ков \cite{Oshem} для классических задач динамики твердого тела.

Поскольку множество точек пересечения кривых $\Delta_i$ с основной диаграммой состоит из вырожденных точек ранга $0$, то его трансформации по $\ld$ уже учтены. Для классификации расширенных диаграмм нужно добавить значения $\ld$, при которых перестраивается множество $\gan\cup \gaa\cup \gac$. Рассмотрим эволюцию самопересечений и двойных пересечений этих кривых (тройных пересечений нет, $\gan\cap \gaa\cap \gac=\varnothing$). Кривые $\gaa$ и $\gac$ определяются однозначными зависимостями $h(\ell)$ и поэтому самопересечений не имеют. Кривая $\gan$ имеет самопересечение, отвечающее значениям
\begin{equation}\label{eq9_4}
s_{\pm}=\frac{1\pm \sqrt{1-8 \ld^4}}{4\ld^2},\quad h=\frac{1-\ld^4}{2\ld^2}, \quad \ell = \pm \frac{\ld}{\sqrt{2}},
\end{equation}
где $s_{\pm}$, при условии их существования $\ld \ls \ld_*=1/2^{3/4}$, всегда лежат в интервале $(0,1/2\ld^2)$.

Условие пересечения $\gan \cap \gac$ представим в виде $2s(h-h^*-2\ell^2)=0$, где $h,\ell$ --- значения \eqref{eq9_1}. В подстановке $Z=\ld^{2/3}s$ получим уравнение $(Z+1)(2Z-1)^2=0$,
имеющее единственный положительный корень $Z=1/2$, причем этот корень кратный, поэтому имеем точку касания с координатами
\begin{equation}\label{eq9_5}
    h=\frac{1}{2}\left[\ld^{-2/3}+2\ld^{2/3}-\ld^2\right], \qquad \ell=\pm \frac{1}{2}\sqrt{\ld^{-2/3}-\ld^{2/3}},
\end{equation}
существующую при $\ld \ls 1$. Так как $1<\ld^*$, то нужно еще убедиться, что в найденной точке $h\gs h^*$. При $\ld \ls 1$ для значения $h$ из \eqref{eq9_5} имеем
\begin{equation}\notag
    h-h^* = \frac{1-\ld^{4/3}}{2\ld^{2/3}}\gs 0.
\end{equation}
Непосредственно проверяется, что точки \eqref{eq9_5} -- это образ вырожденных точек ранга $0$, отвечающих кривой $\vpi_{21}$ (см. предложение \ref{propos3}).

Условие пересечения $\gan \cap \gaa$ представим в виде $2s^3[4(h-\frac{\ld^2}{2})^3-27\ell^2]=0$, где $h,\ell$ --- значения \eqref{eq9_1}. Отсюда имеем
\begin{equation}\label{eq9_6}
(8s^2-2\ld^2 s+1)(s^2-2\ld^2s+1)^2=0.
\end{equation}
Первый сомножитель имеет корни
$$
s=\frac{1}{8}\left[\ld^2\pm \sqrt{\ld^4-8}\right],
$$
но тогда
$$
\ld^2= - \frac{1}{128}\left[\ld^2 \pm \sqrt{\ld^4-8}\right]^3<0,
$$
то есть эти решения --- посторонние. Второй сомножитель в \eqref{eq9_6} отвечает за точку касания и дает положительный корень
\begin{equation}\notag
s = \sqrt{\ld^4+1}-\ld^2,
\end{equation}
для которого точки
\begin{equation}\notag
    h=\frac{3}{2}\sqrt{\ld^4+1} -\ld^2, \qquad \ell=\pm \sqrt{\frac{1}{2}(\sqrt{\ld^4+1} -\ld^2)^3}
\end{equation}
существуют при всех $\ld$ и являются образом вырожденных точек ранга $0$, отвечающих кривой $\vpi_{22}$ (см. предложение \ref{propos6}).

Рассмотрим пересечения $\gaa \cap \gac$. Система \eqref{eq9_2}, \eqref{eq9_3} без учета ограничений имеет три решения
\begin{eqnarray}
& & h=\frac{1}{2} \left(3 \ld^{2/3} + \ld^2\right),\quad \ell^2=\frac{\ld^2}{2}; \label{eq9_7}\\
& &
\left\{\begin{array}{l}
h=\frac{1}{4} \left(-3 \ld^{2/3} + 2 \ld^2 + 3 \sqrt{3(2 - \ld^{4/3})}\right) \\ \ell^2=\frac{1}{8}\left(-9 \ld^{2/3} + 4 \ld^2 + 3 \sqrt{3(2 - \ld^{4/3})}\right)
\end{array}\right. ; \label{eq9_8}\\
& & \left\{\begin{array}{l}
h=\frac{1}{4} \left(-3 \ld^{2/3} + 2 \ld^2 - 3 \sqrt{3(2 - \ld^{4/3})}\right) \\ \ell^2=\frac{1}{8}\left(-9 \ld^{2/3} + 4 \ld^2 - 3 \sqrt{3(2 - \ld^{4/3})}\right)
\end{array}\right. . \label{eq9_9}
\end{eqnarray}
Пара точек \eqref{eq9_7} существует всегда.
Очевидно, в этих точках
$$
h-h^*=\ld^2>0, \quad h-h^{**}=\frac{1}{4}(4+\ld^{4/3})\left[3\ld^{2/3}-\sqrt{4+\ld^{4/3}}\right].
$$
Последняя разность должна быть положительна при $\ld>\ld^*$, но она отрицательна лишь при $\ld<1/2^{3/4}=\ld_*<\ld^*$, поэтому эти точки удовлетворяют всем ограничения. Интересно отметить, что эти точки лежат на тех же уровнях $\ell$, что и точки самопересечения кривой $\gan$ в соответствии с \eqref{eq9_4} (если последние существуют).

В точках \eqref{eq9_9} при условии вещественности $\ld<2^{4/3}$ значение $\ell^2$ оказывается отрицательным, поэтому такое решение --- постороннее.

Рассмотрим решение \eqref{eq9_8}, вещественное при всех $\ld<2^{4/3}$.
Проверим выполнение условий по $h$ на кривой $\gaa$. Исключая $x$ из \eqref{eq5_6} и полагая в \eqref{eq9_8} $h=h_C$, придем к уравнению
\begin{equation}\label{eq9_10}
    1+48 \ld^{4/3}+12 \ld^{8/3}-8 \ld^{4}-6\sqrt{3(2-\ld^{4/3})}(\ld^{2/3}+4 \ld^{2})=0.
\end{equation}
Подстановка $X=\ld^{4/3}$ сводит его к уравнению
$$
(2X-1 )^3 ( 2 X-1 - 3 \cdot 2^{1/3} + 3\cdot 2^{2/3}) \left[\left( 2 X  + \frac{3}{2^{2/3}} - \frac{3}{2^{1/3}}-1\right)^2 + \frac{27}{2^{2/3}} \left(1 + \frac{1}{2^{2/3}} + 2^{2/3}\right)\right]=0.
$$
Отсюда видно, что \eqref{eq9_10} имеет ровно два положительных решения: трехкратный корень $\ld=\ld_*$ и простой корень
\begin{equation}\label{eq9_11}
    \ld_3 = \left( \frac{45}{2^{1/3}}-\frac{99}{4 \cdot 2^{2/3}}-\frac{161}{8} \right)^{1/4} \approx 0.0287, \qquad \ld_1<\ld_3<\ld_*.
\end{equation}

Условие $h<h_C$ выполнено при $\ld>\ld_3$. Корень $\ld=\ld_*$ на это неравенство не  влияет, его появление связано с особой точкой \eqref{eq4_37}, определенной теми же уравнениями \eqref{eq5_6}, но с положительным $x$. Проверим выполнение в точках \eqref{eq9_8} условий по $h$ на кривой $\gac$. Имеем
$$
h-h^*=\frac{1}{4}\left[\ld^2-9\ld^{2/3}+3\sqrt{3(2-\ld^{4/3})} \right].
$$
Эта величина неотрицательна при $\ld<(3/2)^{3/4} \approx 1.3554$, и тем более это имеет место при $\ld<\ld^* \approx 1.2408$.
Пусть $\ld>\ld^*$. Тогда должна быть неотрицательна величина
$$
h-h^{**}=\frac{1}{4}\left[3\ld^2+3\ld^{2/3}+3\sqrt{3(2-\ld^{4/3})}-(4+\ld^{4/3})^{3/2} \right].
$$
Подстановка $X=(\ld^{2/3}+\sqrt{4+\ld^{4/3}})^2$ сводит условие $h\gs h^{**}$ к неравенству
$$
(X-8)^3(X^3-12X^2-32) \ls 0,
$$
решением которого является промежуток $X\in [8,2(2+2^{2/3}+2^{4/3})]$, тогда $\ld \in [\ld_*, \ld_4]$, где обозначено
\begin{equation}\label{eq9_12}
    \ld_4 = \left[\frac{1}{2}(25-27\cdot 2^{1/3}+9\cdot 2^{2/3})\right]^{1/4} \approx 1.2740.
\end{equation}
Итак, пара точек пересечения кривых $\gaa$ и $\gac$, определяемых уравнениями \eqref{eq9_8}, существует при $\ld_3<\ld \ls \ld_4$.

В результате получаем следующее утверждение.
\begin{theorem}\label{th14}
В случае Ковалевской\,--\,Яхья имеется десять структурно устойчивых диаграмм Смей\-ла--Фо\-мен\-ко $\smales(\ld)$. Разделяющими значениями параметра $\ld$ служат
\begin{equation}\label{eq9_13}
0,\quad \ld_1,\quad \ld_3, \quad \ld_*= 1/2^{3/4},\quad 1,\quad \ld^*=2 \sqrt{2}/3^{3/4},\quad \ld_4,\quad \ld_5=2\sqrt{\sqrt{2}-1},\quad \ld_2, \quad \sqrt{2},
\end{equation}
где $\ld_1,\ld_2,\ld_3,\ld_4$ определены равенствами {\rm \eqref{eq4_56}, \eqref{eq4_57}, \eqref{eq9_11}, \eqref{eq9_12}}.

В расширенном пространстве $\bR^3(\ell,h,\ld)$ расширенная диаграмма
$$
\mwide{\smales} = \bigcup_{\ld} (\smales(\ld){\times}{\{\ld\}}
$$
порождает 29 камер. При этом камеры Смейла $\mtC,\mtD,\mtF,\mtG$ не испытывают дополнительного разбиения, а из камер Смейла $\mtA,\mtB,\mtE,\mtH$ возникают камеры Смейла\,--\,Фоменко $\mtA_1 - \mtA_{13}$, $\mtB_1 - \mtB_{3}$, $\mtE_1 - \mtE_{6}$ и $\mtH_1 - \mtH_{3}$.
\end{theorem}

\def\wid{1}
\begin{figure}[!ht]
\centering
\includegraphics[width=\wid\textwidth,keepaspectratio]{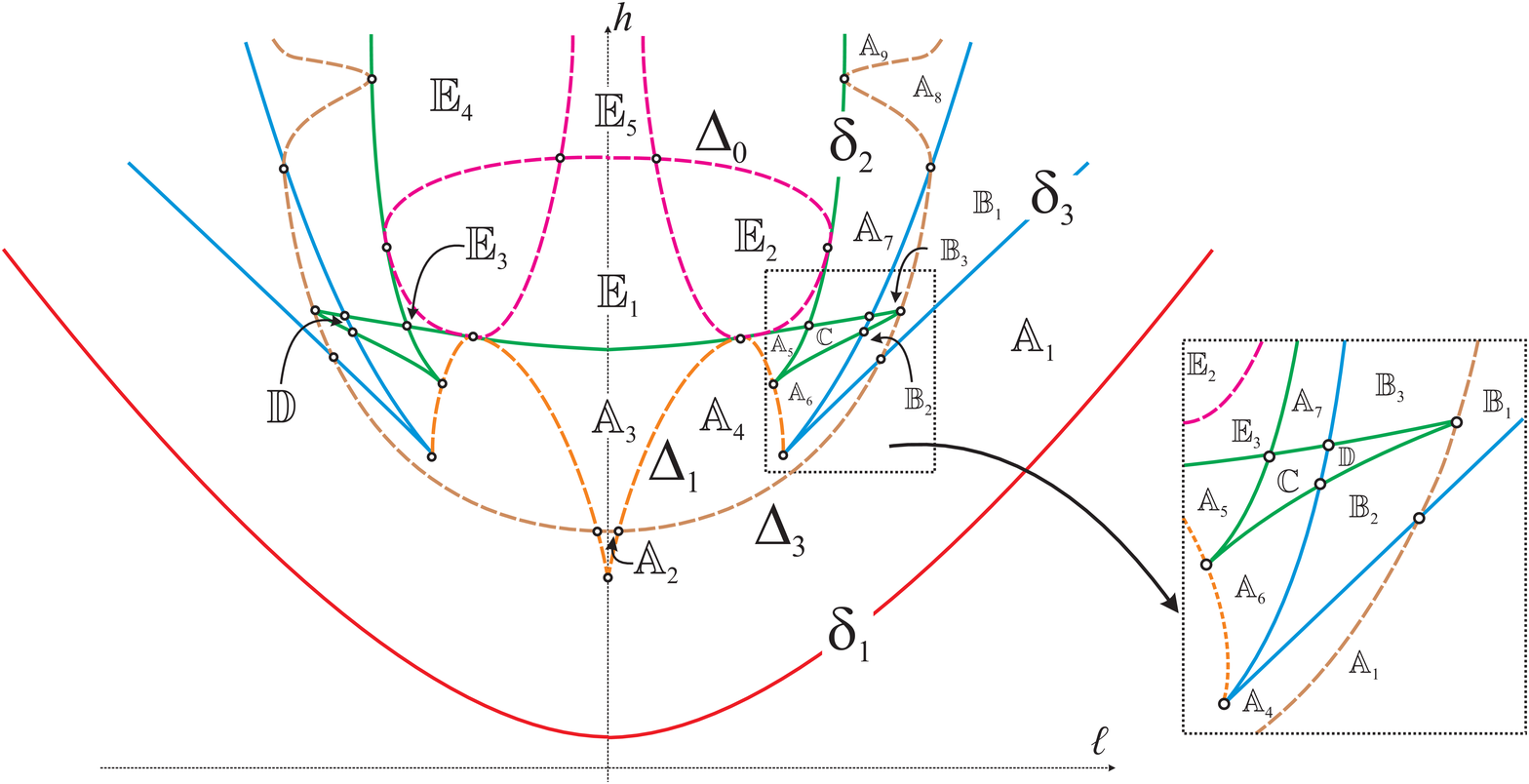}\\
\caption{Диаграмма Смейла--Фоменко при $\ld<\ld_1$.}\label{fig_smalefom01}
\end{figure}

На рис.~\ref{fig_smalefom01} показана диаграмма Смейла\,--\,Фоменко при малых $\ld$ ($0<\ld<\ld_1$). Допущены в целом небольшие гладкие искажения общей картины. Напомним, что, как это видно из диаграмм третьей критической подсистемы (см., например, рис.~\ref{fig_M3_LKeyABC_red}), множества $\Delta_3$ и $\delta_2$ имеют две общих точки в области $\ell>0$ при малых $\ld$ -- это точки $B_3,B_5$. В диаграмме Смейла\,--\,Фоменко кривая $\Delta_3$ проходит через точку возврата $\delta_2$. Это образ особой точки $B_3$ и он имеет конечный предел при $\ld \to 0$. Образ же второй точки $B_5$ имеет, согласно \eqref{eq5_20}, $h$-координату, стремящуюся к $\infty$. Поэтому
значительное искажение сделано на кривой $\Delta_3$ при больших $h$ с тем, чтобы показать ``далекую'' общую точку $\Delta_3$ с $\delta_2$ (точку касания) и границу между камерами $\mtA_8$ и $\mtA_9$.

Далее на рис.~\ref{fig_smalefom02} и \ref{fig_smalefom03} показаны изменения, связанные с переходами через $\ld=\ld_1$ (исчезают камеры $\mtB_3,\mtD$) и через $\ld=\ld_3$ (исчезает камера $\mtB_2$). Новых камер здесь не появляется.

\def\wid{1}
\begin{figure}[htp]
\centering
\includegraphics[width=\wid\textwidth,keepaspectratio]{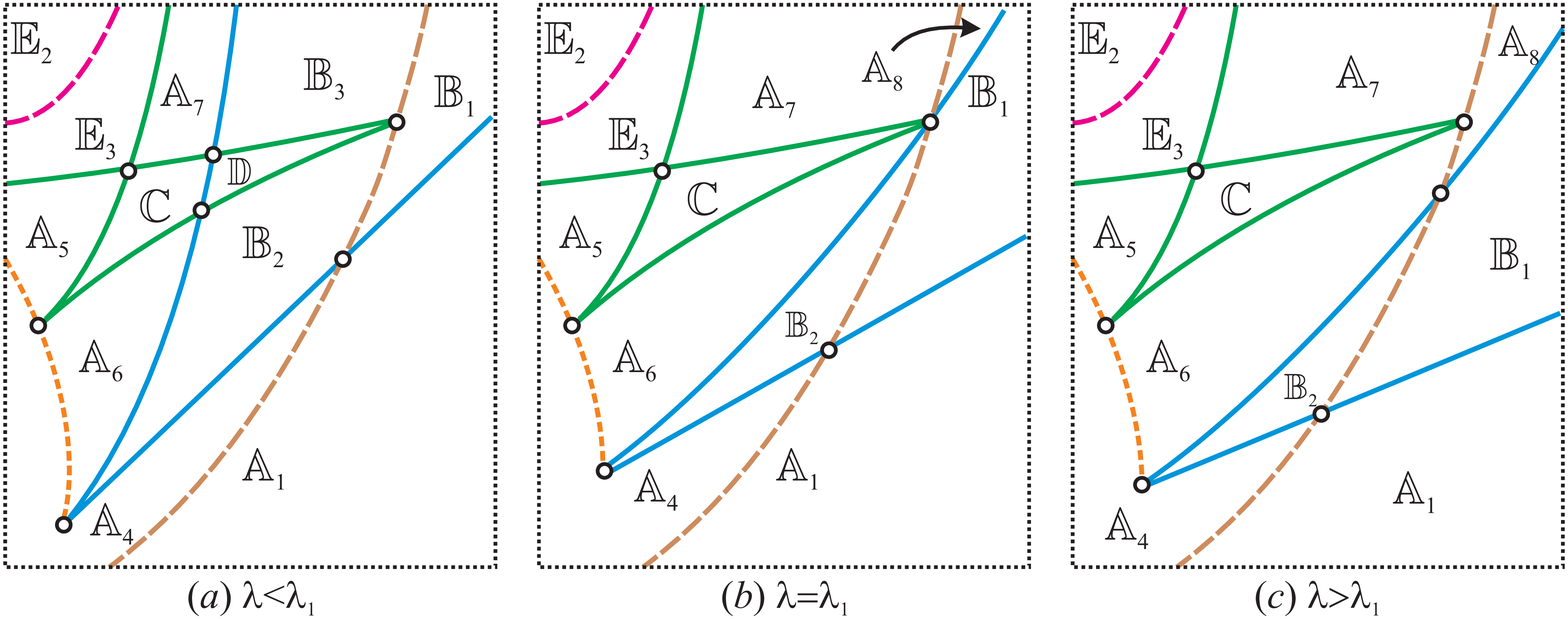}\\
\caption{Переход через $\ld_1$.}\label{fig_smalefom02}
\end{figure}

\def\wid{1}
\begin{figure}[htp]
\centering
\includegraphics[width=\wid\textwidth,keepaspectratio]{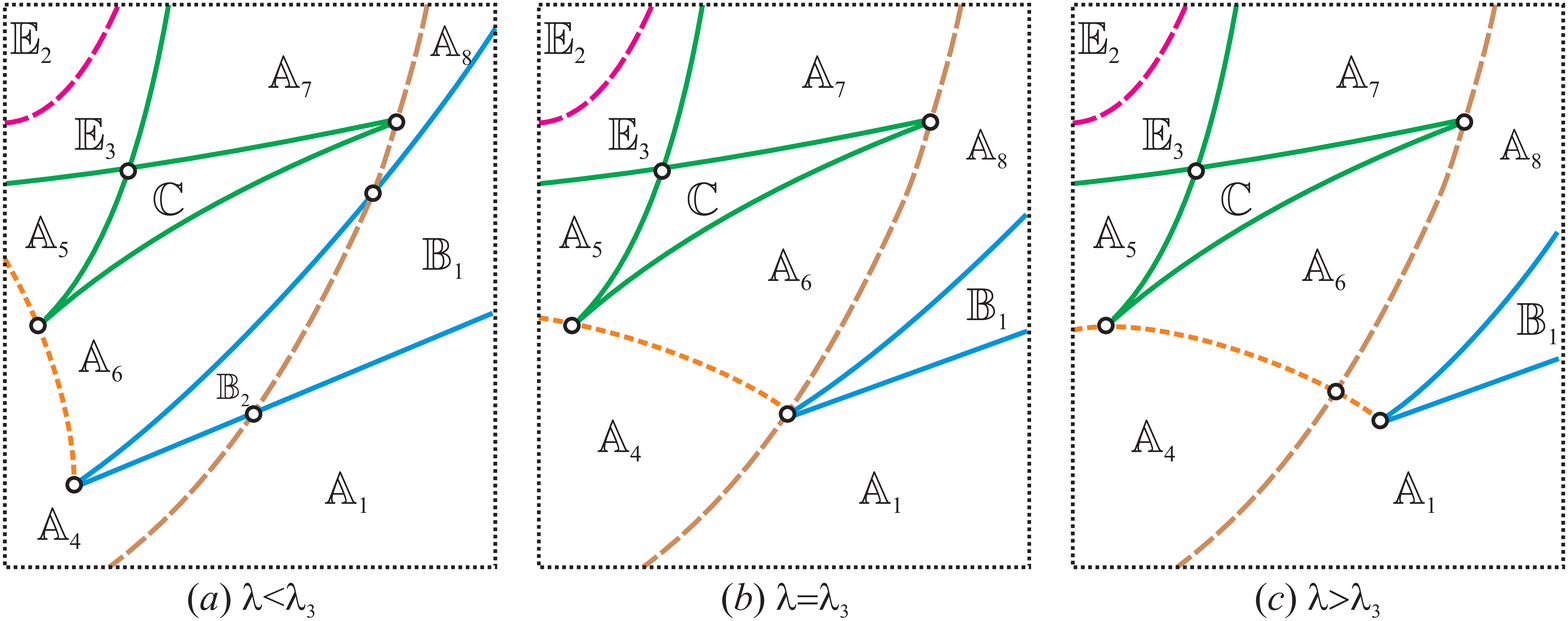}\\
\caption{Переход через $\ld_3$.}\label{fig_smalefom03}
\end{figure}

\def\wid{0.6}
\begin{figure}[htp]
\centering
\includegraphics[width=\wid\textwidth,keepaspectratio]{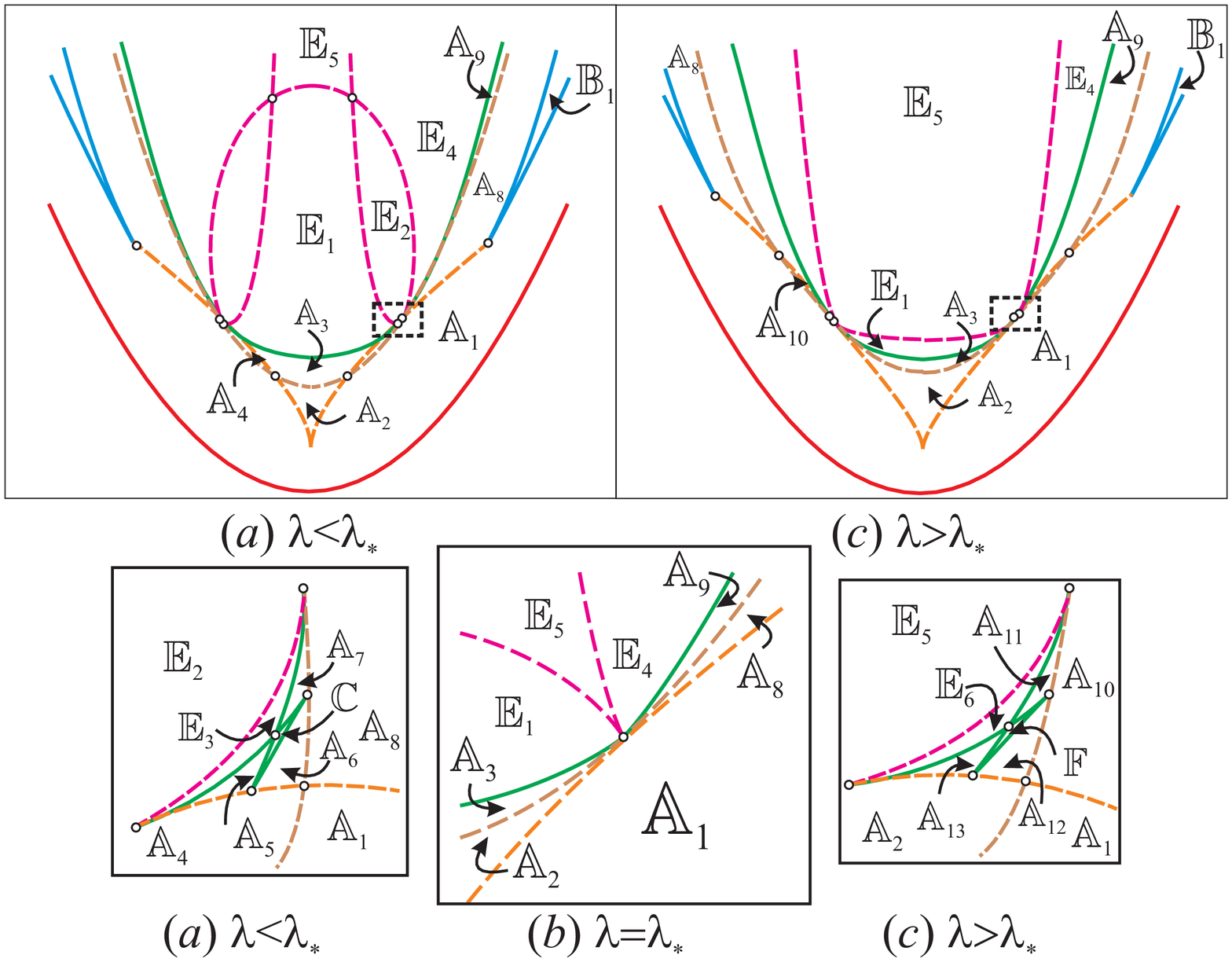}\\
\caption{Переход через $\ld_*$.}\label{fig_smalefom04}
\end{figure}


\def\wid{0.7}
\begin{figure}[htp]
\centering
\includegraphics[width=\wid\textwidth,keepaspectratio]{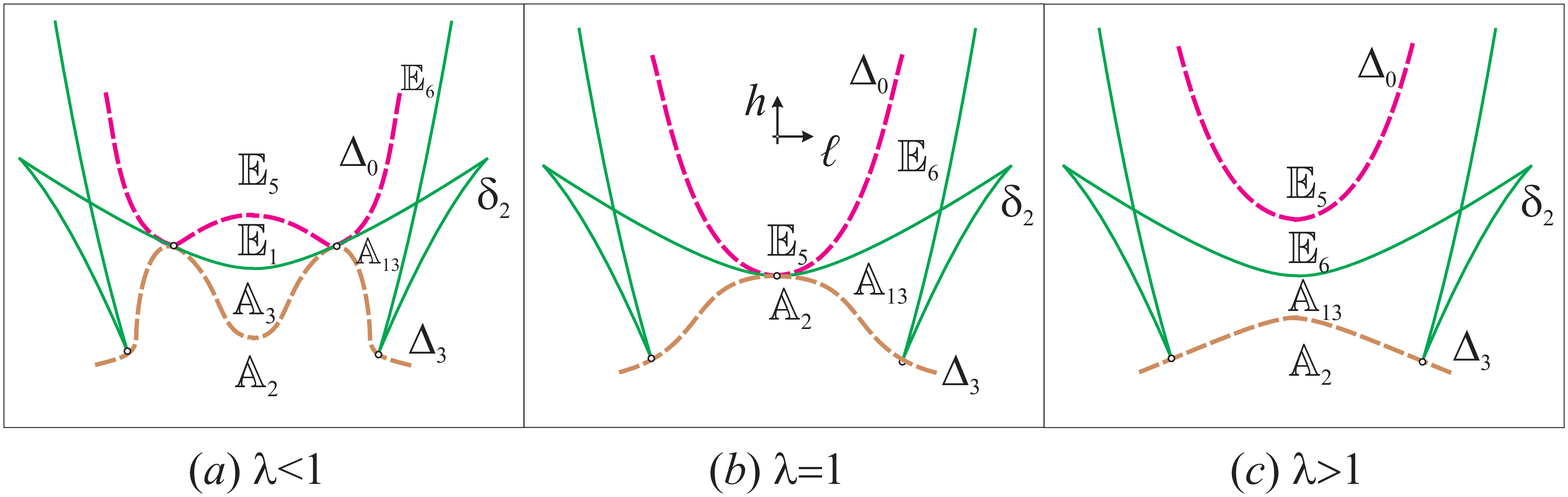}\\
\caption{Переход через $\ld=1$.}\label{fig_smalefom05}
\end{figure}

\def\wid{0.55}
\begin{figure}[htp]
\centering
\includegraphics[width=\wid\textwidth,keepaspectratio]{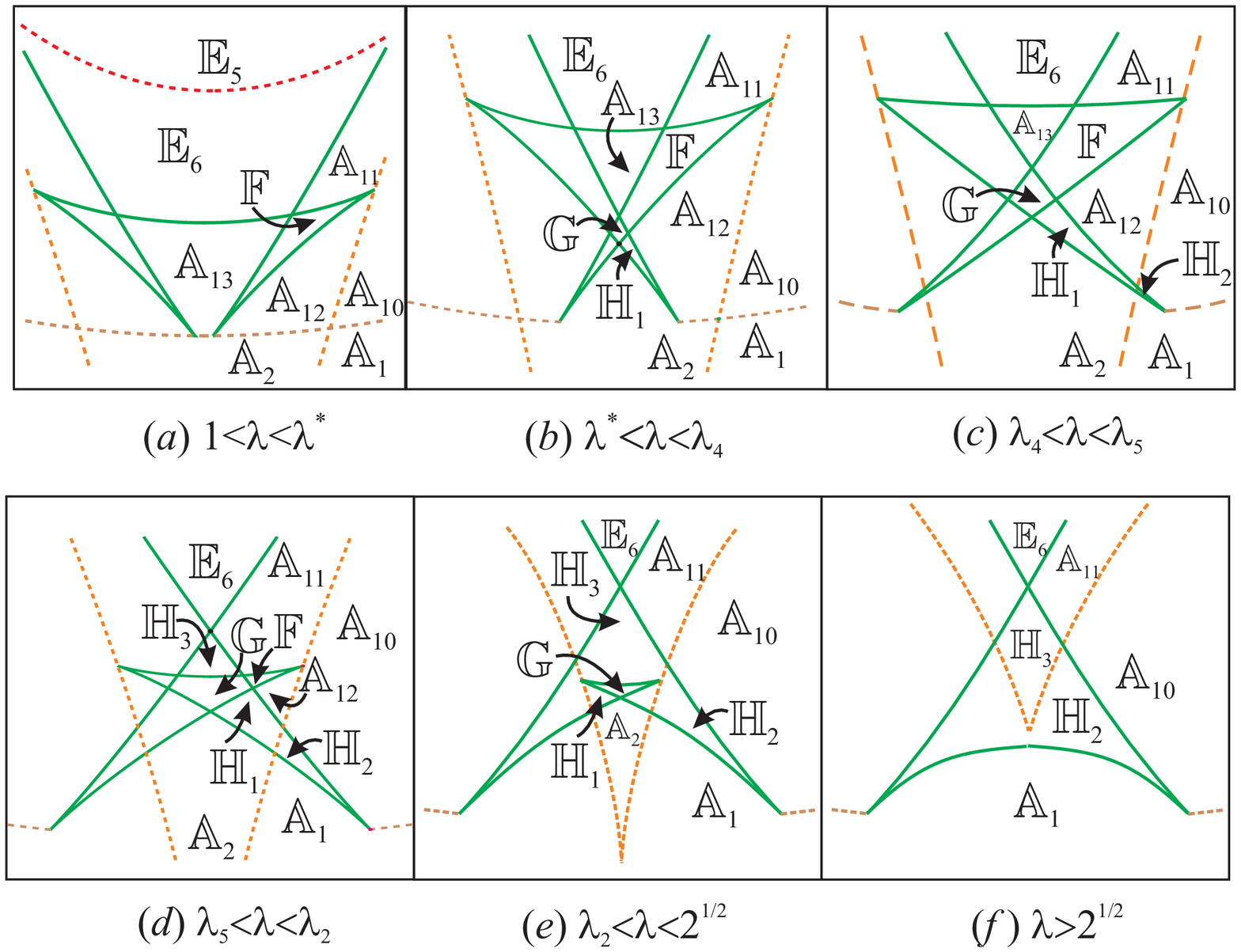}\\
\caption{Перестройки при больших $\ld$.}\label{fig_smalefom06}
\end{figure}

Значительные изменения в диаграмме Смейла\,--\,Фоменко (как и во всех ранее рассмотренных разделяющих множествах), а также в составе камер происходят при переходе через значение $\ld=\ld_*$. При этом значении стягиваются в точку (и затем исчезают) петля, окружающая камеру $\mtE_2$, камера $\mtE_3$, ласточкин хвост кривой $\delta_2$ вместе с камерами $\mtA_5,\mtA_6,\mtA_7,\mtC$, отрезки границ камеры $\mtA_4$. При $\ld>\ld_*$ снова появляется хвост на кривой $\delta_2$, но уже с другими новыми камерами в его окрестности. Это камеры $\mtA_{10}, \mtA_{11},\mtA_{12},\mtA_{13},\mtE_6,\mtF$. Переход показан на рис.~\ref{fig_smalefom04}.

Переход через $\ld=1$ достаточно прост -- общая точка касания $\delta_2,\Delta_0,\Delta_3$ попадает на ось симметрии $\ell=0$ и затем исчезает. При этом исчезают камеры $\mtA_3,\mtE_1$. Переход показан на рис.~\ref{fig_smalefom05}.

Все дальнейшие перестройки диаграмм Смейла\,--\,Фоменко связаны с точками на оси $\ell=0$. Область, содержащая все такие перестройки, и соответствующие камеры показаны на рис.~\ref{fig_smalefom06}.

При переходе через $\ld^*$ рождаются камеры $\mtG,\mtH_1$, при переходе через $\ld_4$ появляется камера $\mtH_2$. Переход через $\ld_5$ сопровождается исчезновением камеры $\mtA_{13}$ и появлением $\mtH_3$. При переходе через $\ld_2$ исчезает камера $\mtF$, а при переходе через последнее разделяющее значение $\sqrt{2}$ исчезают камеры $\mtG,\mtH_1,\mtA_2$. Ниже информация по времени существования камер собрана в таблицу (см. табл.~\ref{table6}).


\subsection{Графы Фоменко}
Для каждой камеры в пространстве $\bR^3(\ell,h,\ld)$, вырезанной диаграммой Смейла\,--\,Фоменко, определен грубый топологический инвариант --- граф Фоменко или молекула $W_{\ell,h}(\ld)$, то есть граф, полученный стягиванием в точку каждой связной компоненты интегрального многообразия с указанием для каждого критического уровня дополнительного интеграла $K$ типа и, при необходимости, ориентации соответствующего атома. Два графа Фоменко считаются совпадающими (индентичными), если существует гомеоморфизм графов, продолжающийся на атомы (подробности см. в  \cite{BolFom}).

\begin{figure}[ht]
\centering
\includegraphics[width=120mm,keepaspectratio]{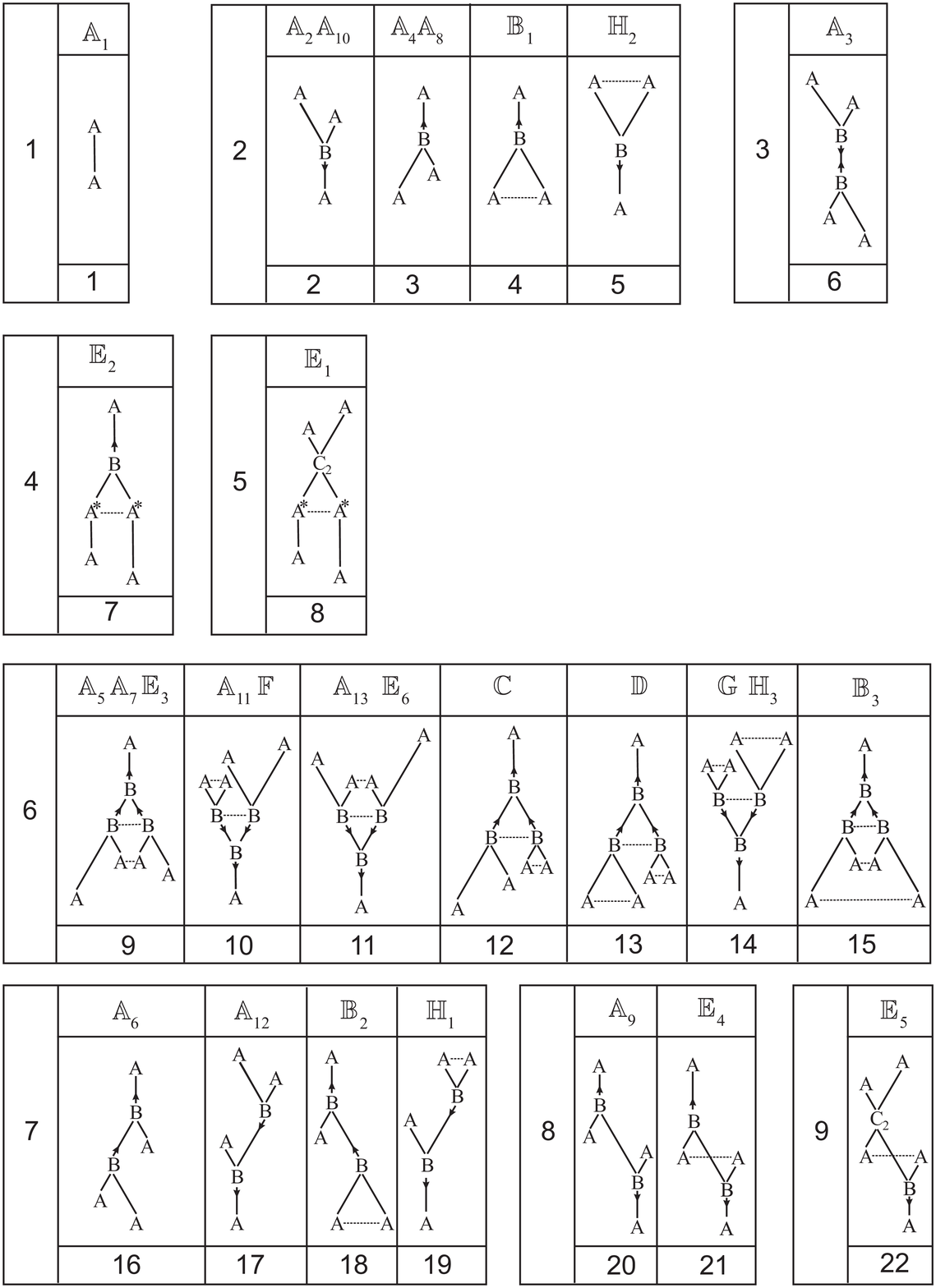}\\
\caption{Графы Фоменко}\label{fig_graph_fom_xn}
\end{figure}

Все возникающие в этой задаче графы Фоменко на гладких изоэнергетических уровнях $\iso$, не содержащих критических точек ранга 0 и вырожденных критических точек ранга 1, описаны в работе \cite{KhRy2011} и представлены на рис.~\ref{fig_graph_fom_xn}. Здесь в группы собраны графы, формально идентичные, но в рамках данной системы они не могут быть переведены один в другой. Группы 1--6 соответствуют типам графов $W_1-W_6$, найденным в работе \cite{Gash5}. Группы 7--9 являются новыми.  В группах 7,\,8 атомы $B$ не соединены ``голова в голову'', как в группе 3 и в типе $W_3$ \cite{Gash5}, а в группе 9, в отличие от графа $W_7$ \cite{Gash5}, ребро из атома $C_2$ идет в ``ногу'', а не в ``голову'' атома $B$, что порождает иное слоение Лиувилля.

Естественно, что имеются совпадающие графы Фоменко даже для различных типов изоэнергетических поверхностей. Для того, чтобы различать такие графы, необходимо применять тонкую классификацию \cite{BolFom}. Собирая информацию по всем возникающим камерам в табл.~\ref{table6}, видим, что для большинства графов метки могут быть получены непосредственно из аналогов для случаев $\ld=0$ или $\ell=0$. Отметим, что в работе \cite{KhRy2011} в аналогичной таблице допущена неточность -- время жизни камеры $\mtA_2$ указано неограниченным, хотя ее и соответствующего ей графа не существует при $\ld^2>2$.

\begin{center}
\small
\begin{longtable}{| c| c| c|c|c|}
\multicolumn{5}{r}{\fts{Таблица \myt\label{table6}}}\\[2mm]
\hline
\begin{tabular}{c}\fts{Камера}\end{tabular} &\begin{tabular}{c}\fts{Группа графа}\\[-3pt]\fts{(номер)}\end{tabular}
&\begin{tabular}{c}\fts{Время жизни}\\[-3pt]\fts{по $\ld$}\end{tabular}&\begin{tabular}{c}\fts{Выход на}\\[-3pt]\fts{$\ld=0$}/\fts{$\ell=0$}
\end{tabular} &
\begin{tabular}{c}\fts{Меченая}\\[-3pt]\fts{молекула}\end{tabular}\\
\hline
$\mtA_1$ & 1(1) & $ 0 \ls \ld <+\infty$ & Да/Да & \begin{tabular}{c}$A$\,\cite[Табл.~3]{BRF},\\ \,\cite[Табл.~8]{Mor}\end{tabular}\\
\hline\endfirsthead%
\multicolumn{5}{r}{\fts{Таблица \ref{table6} (продолжение)}}\\[2mm]
\hline
\begin{tabular}{c}\fts{Камера}\end{tabular} &\begin{tabular}{c}\fts{Группа графа}\\[-3pt]\fts{(номер)}\end{tabular}
&\begin{tabular}{c}\fts{Время жизни}\\[-3pt]\fts{по $\ld$}\end{tabular}&\begin{tabular}{c}\fts{Выход на}\\[-3pt]\fts{$\ld=0$}/\fts{$\ell=0$}
\end{tabular} &
\begin{tabular}{c}\fts{Меченая}\\[-3pt]\fts{молекула}\end{tabular}\\
\hline
$\mtA_1$ & 1(1) & $ 0 \ls \ld <+\infty$ & Да/Да & \begin{tabular}{c}$A$\,\cite[Табл.~3]{BRF},\\ \,\cite[Табл.~8]{Mor}\end{tabular}\\
\hline\endhead%

$\mtA_2$ & 2(2) & $ 0 < \ld <\sqrt{2}$ & Нет/Да & $B$\,\cite[Табл.~8]{Mor}\myrul\\

\hline
$\mtA_3$ & 3(6) & $ 0 \ls \ld < 1$ & Да/Да & \begin{tabular}{c}$C$\,\cite[Табл.~3]{BRF},\\ \,\cite[Табл.~8]{Mor}\end{tabular}\\

\hline
$\mtA_4$ & 2(3) & $ 0 \ls \ld <\ld_*$ & Да/Нет & $B$\,\cite[Табл.~3]{BRF}\myrul\\

\hline
$\mtA_5$ & 6(9) & $ 0 \ls \ld <\ld_*$ & Да/Нет & $J$\,\cite[Табл.~3]{BRF}\myrul\\

\hline
$\mtA_6$ & 7(16) & $ 0 < \ld < \ld_*$ & Нет/Нет & \myrul\\

\hline
$\mtA_7$ & 6(9) & $ 0 < \ld < \ld_*$ & Нет/Нет & \myrul\\

\hline
$\mtA_8$ & 2(3) & $ 0 < \ld <+\infty$ & Нет/Нет & \myrul\\

\hline
$\mtA_9$ & 8(20) & $ 0 < \ld <+\infty$ & Нет/Нет & \myrul\\

\hline
$\mtA_{10}$ & 2(2) & $ \ld > \ld_* $ & Нет/Нет & \myrul\\

\hline
$\mtA_{11}$ & 6(10) & $ \ld > \ld_* $ & Нет/Нет & \myrul\\

\hline
$\mtA_{12}$ & 7(17) & $ \ld_* < \ld < \ld_2 $ & Нет/Нет & \myrul\\

\hline
$\mtA_{13}$ & 6(11) & $ \ld_* < \ld < \ld_5 $ & Нет/Да & $F$\,\cite[Табл.~8]{Mor}\myrul\\

\hline

$\mtB_{1}$ & 2(4) & $ 0 \ls \ld < +\infty $ & Да/Нет & $F$\,\cite[Табл.~3]{BRF}\myrul\\

\hline
$\mtB_{2}$ & 7(18) & $ 0 < \ld < \ld_3 $ & Нет/Нет & \myrul\\

\hline

$\mtB_{3}$ & 6(15) & $ 0 \ls \ld < \ld_1 $ & Да/Нет & $G$\,\cite[Табл.~3]{BRF}\myrul\\

\hline
$\mtC$ & 6(12) & $ 0 <\ld < \ld_* $ & Нет/Нет & \myrul\\

\hline
$\mtD$ & 6(13) & $ 0 \ls\ld < \ld_1 $ & Да/Нет & $I$\,\cite[Табл.~3]{BRF}\myrul\\

\hline
$\mtE_1$ & 5(8) & $ 0 \ls\ld < 1 $ & Да/Да & \begin{tabular}{c}$D$\,\cite[Табл.~3]{BRF},\\ \,\cite[Табл.~8]{Mor}\end{tabular}\\

\hline
$\mtE_2$ & 4(7) & $ 0 \ls\ld < \ld_* $ & Да/Нет & $E$\,\cite[Табл.~3]{BRF}\myrul\\

\hline
$\mtE_3$ & 6(9) & $ 0 \ls\ld < \ld_* $ & Да/Нет & $H$\,\cite[Табл.~3]{BRF}\myrul\\

\hline
$\mtE_4$ & 8(21) & $ 0 <\ld < +\infty $ & Нет/Нет & \myrul\\

\hline
$\mtE_5$ & 9(22) & $ 0 <\ld < +\infty $ & Нет/Да & $E$\,\cite[Табл.~8]{Mor}\myrul\\

\hline
$\mtE_6$ & 6(11) & $ \ld_* <\ld < +\infty $ & Нет/Да & $G$\,\cite[Табл.~8]{Mor}\myrul\\

\hline
$\mtF$ & 6(10) & $ \ld_* <\ld < \ld_2$ & Нет/Нет & \myrul\\

\hline
$\mtG$ & 6(14) & $ \ld^* <\ld < \sqrt{2}$ & Нет/Да & $H$\,\cite[Табл.~8]{Mor}\myrul\\

\hline
$\mtH_1$ & 7(19) & $ \ld^* <\ld < \sqrt{2}$ & Нет/Нет & \myrul\\

\hline
$\mtH_2$ & 2(5) & $ \ld_4 <\ld < +\infty$ & Нет/Да & $J$\,\cite[Табл.~8]{Mor}\myrul\\

\hline
$\mtH_3$ & 6(14) & $ \ld_5 <\ld < +\infty$ & Нет/Да & $I$\,\cite[Табл.~8]{Mor}\myrul\\

\hline
\end{longtable}

\end{center}

\begin{figure}[!ht]
\centering
\includegraphics[width=60mm,keepaspectratio]{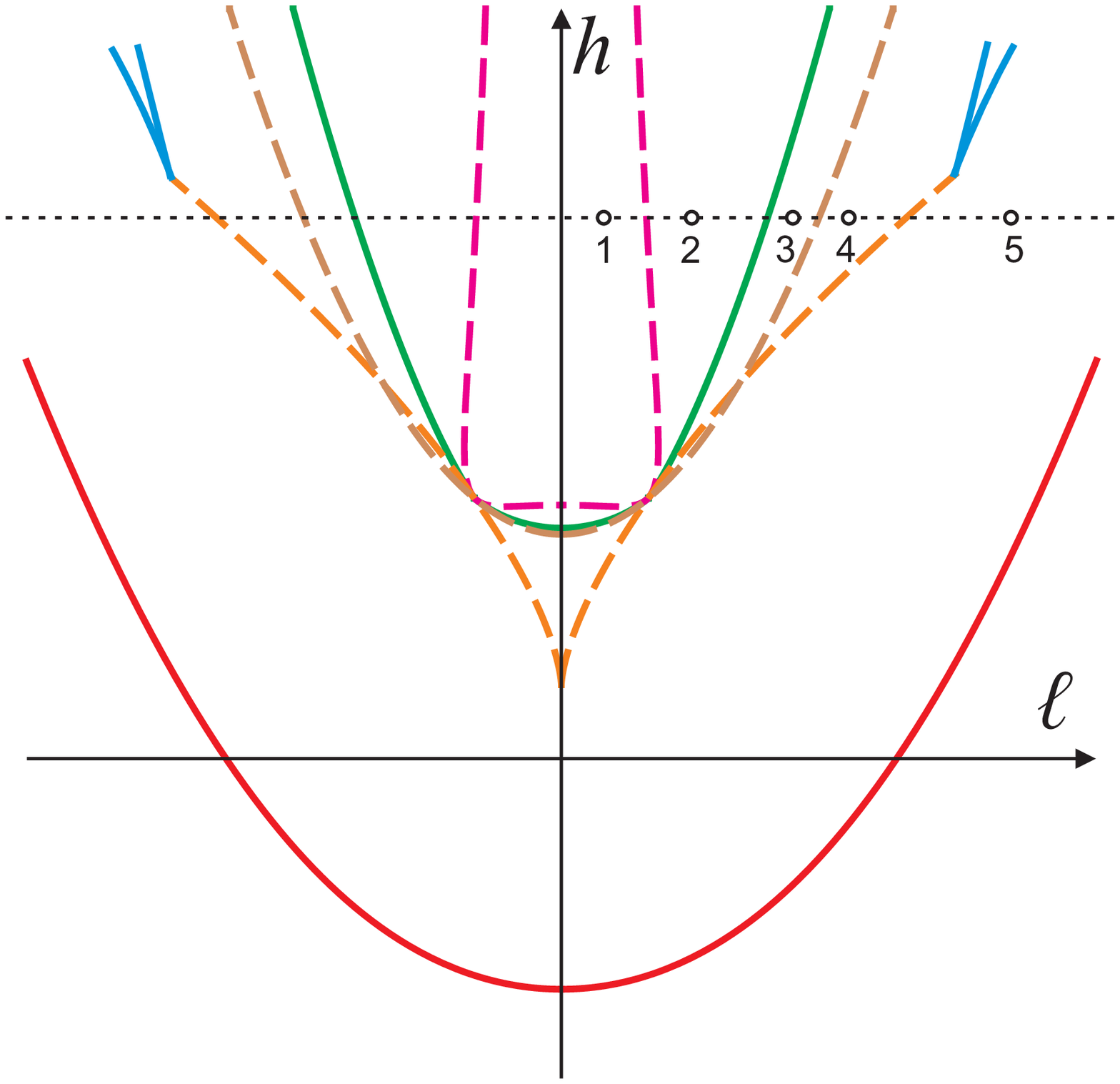}\\
\caption{Диаграмма Смейла--Фоменко для $\ld=0.8$.}\label{fig_smalefomd}
\end{figure}

\begin{figure}[!ht]
\centering
\includegraphics[width=120mm,keepaspectratio]{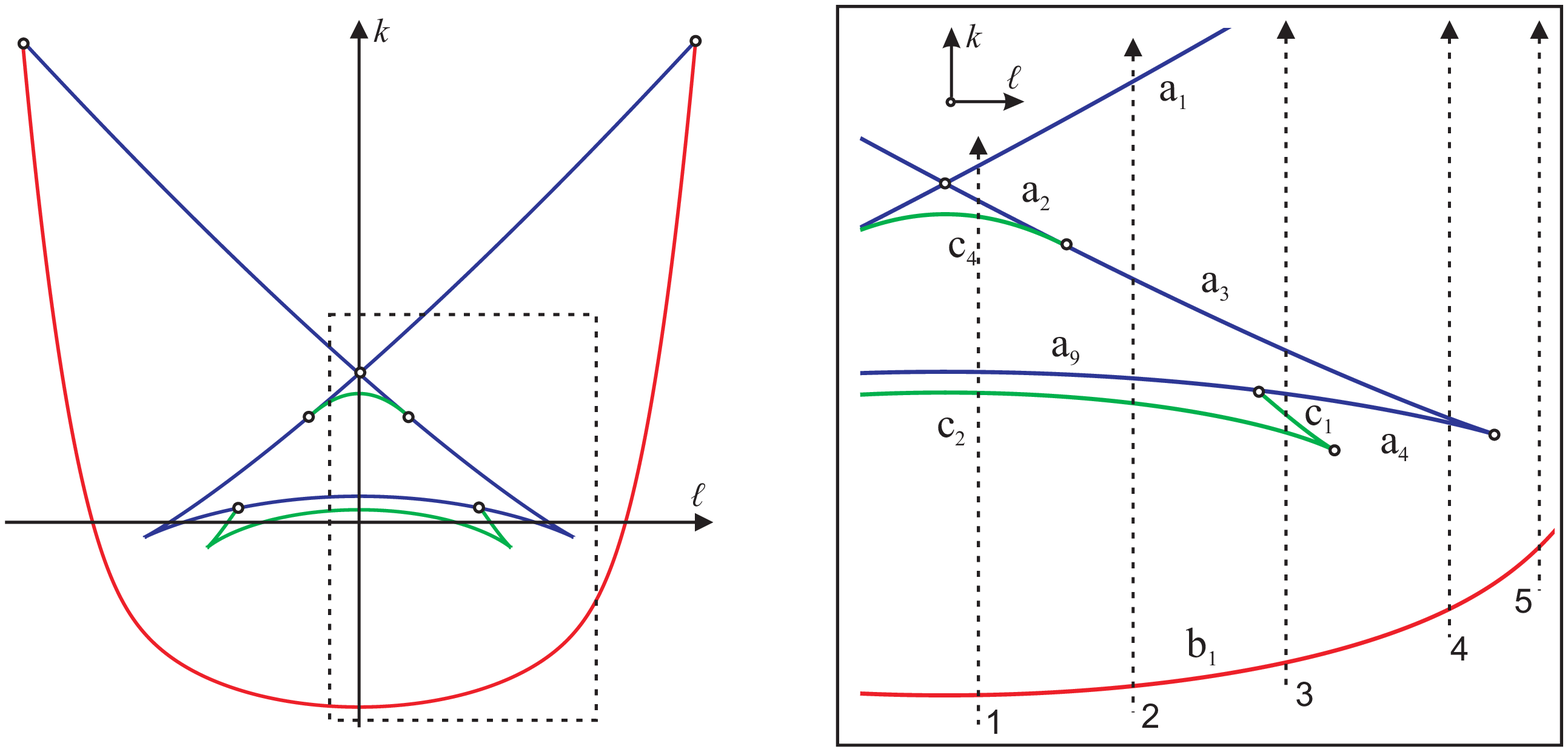}\\
\caption{Бифуркационная диаграмма в $(\ell,k)$-плоскости для $\ld=0.8,h=2.5$.}\label{fig_bifdiag}
\end{figure}

Метки на графах, не имеющих аналогов, можно расставить с помощью найденных выше круговых молекул узловых точек бифуркационных диаграмм. Один из вариантов представления результата (с матрицами склейки) приведен в работе~\cite{Slav2Rus}.

Проиллюстрируем работу ``конструктора'' графов Фоменко на примерах. Рассмотрим диаграммы Смейла\,--\,Фоменко для средних значений $\ld$ (выбрано $\ld=0.8$) и пусть $h=2.5$ (рис.~\ref{fig_smalefomd}). Этот уровень при возрастании $\ell$ от нуля пересекает пять камер $\mtE_5,\mtE_4,\mtA_9,\mtA_8,\mtA_1$ (пути $\ell=\cons$ на рис.~\ref{fig_smalefomd} занумерованы цифрами 1,...,5).
В соответствующем $h$-сечении диаграммы $\Sigma(\ld)$ графы Фоменко определяются бифуркациями вдоль прямых $\ell=\cons$ при возрастании $k$ (пять пунктирных стрелок на рис.~\ref{fig_bifdiag}). Еще 8 камер можно увидеть на одном уровне $h$ при малых $\ld$. Диаграмма Смейла\,--\,Фоменко для $\ld=0.1$ и уровень $h=1.8$ показаны на рис.~\ref{fig_smalefomd_1}. Соответствующие пути $\ell=\cons$ на диаграмме $\Sigma(\ld)$ с обозначением пересекаемых дуг показаны на рис.~\ref{fig_bifdiag_1}.
Пересекаемые дуги отвечают областям, определенным диаграммами критических подсистем. Последовательности этих пересечений и соответствующих атомов для путей с номерами 1,...,13 приведены в табл.~\ref{table7}.

\begin{center}
\small
\begin{longtable}{| c| c| c|c|c|}
\multicolumn{5}{r}{\fts{Таблица \myt\label{table7}}}\\[2mm]
\hline
\begin{tabular}{c}\fts{Номер}\\[-3pt]\fts{пути}\end{tabular} &\begin{tabular}{c}\fts{Камера}\end{tabular} &\begin{tabular}{c}\fts{Последовательность}\\[-3pt]\fts{дуг}\end{tabular}
&\begin{tabular}{c}\fts{Последовательность}\\[-3pt]\fts{атомов}\end{tabular}&\begin{tabular}{c}\fts{Группа графа}\\[-3pt]\fts{(номер)}
\end{tabular} \\
\hline\endfirsthead%
\multicolumn{5}{r}{\fts{Таблица \ref{table7} (продолжение)}}\\[2mm]
\hline
\begin{tabular}{c}\fts{Номер}\\[-3pt]\fts{пути}\end{tabular} &\begin{tabular}{c}\fts{Камера}\end{tabular} &\begin{tabular}{c}\fts{Последовательность}\\[-3pt]\fts{дуг}\end{tabular}
&\begin{tabular}{c}\fts{Последовательность}\\[-3pt]\fts{атомов}\end{tabular}&\begin{tabular}{c}\fts{Группа графа}\\[-3pt]\fts{(номер)}
\end{tabular} \\
\hline\endhead%
{1}&{$\mtE_5$}&\begin{tabular}{c}\fts{$b_1\to c_2\to a_9\to$}\\[-3pt] \fts{$\to c_4\to a_2\to a_1$}\end{tabular}&
\begin{tabular}{c}\fts{$A_+ \to B_+ \to (A_+,A_-) \to$}\\[-3pt] \fts{$\to C_2 \to A_- \to A_-$}\end{tabular}
&{9(22)}\\
\hline
{2}&{$\mtE_4$}&
\begin{tabular}{c}\fts{$b_1\to c_2\to a_9\to$}\\[-3pt] \fts{$\to a_3\to a_1$}\end{tabular}
&\begin{tabular}{c}\fts{$A_+\to B_+ \to (A_+,A_-)\to$}\\[-3pt] \fts{$\to  B_-\to A_-$}\end{tabular}&{8(21)}\\
\hline
{3}&{$\mtA_9$}&
\begin{tabular}{c}\fts{$b_1\to c_2\to c_1\to$}\\[-3pt] \fts{$\to a_4\to a_3\to a_1$}\end{tabular}
&\begin{tabular}{c}\fts{$A_+\to B_+ \to A_- \to A_+\to$}\\[-3pt] \fts{$\to  B_-\to A_-$}\end{tabular}&{8(20)}\\
\hline
{4}&{$\mtA_8$}&
\begin{tabular}{c}\fts{$b_1\to a_4\to a_3\to a_1$}\end{tabular}
&\fts{$A_+ \to A_+ \to B_- \to A_-$}&{2(3)}\myrul\\
\hline
{5}&{$\mtA_1$}&
\fts{$b_1\to a_1$}&\fts{$A_+ \to A_-$}&{1(1)}\myrul\\
\hline
{6}&{$\mtB_2$}&
\begin{tabular}{c}\fts{$b_2\to a_3\to c_6\to$}\\[-3pt] \fts{$\to c_7\to a_1$}\end{tabular}
&\begin{tabular}{c}\fts{$2A_+ \to B_- \to A_+\to$}\\[-3pt] \fts{$\to B_-\to A_-$}\end{tabular}&{7(18)}\\
\hline
{7}&{$\mtA_6$}&
\begin{tabular}{c}\fts{$b_1\to a_4\to a_3\to$}\\[-3pt] \fts{$\to c_6\to c_7\to a_1$}\end{tabular}
&\begin{tabular}{c}\fts{$A_+ \to A_+ \to B_-\to$}\\[-3pt] \fts{$\to A_+ \to B_-\to A_-$}\end{tabular}&{7(16)}\\
\hline
{8}&{$\mtC$}&
\begin{tabular}{c}\fts{$b_1\to a_4\to c_8\to$}\\[-3pt] \fts{$\to a_7\to c_7\to a_1$}\end{tabular}
&\begin{tabular}{c}\fts{$A_+ \to A_+ \to 2A_+ \to 2B_-\to$}\\[-3pt] \fts{$\to B_- \to A_-$}\end{tabular}&{6(12)}\\
\hline
{9}&{$\mtA_5$}&
\begin{tabular}{c}\fts{$b_1\to c_6\to a_8\to$}\\[-3pt] \fts{$\to a_7\to c_7\to a_1$}\end{tabular}
&\begin{tabular}{c}\fts{$A_+ \to A_+ \to 2A_+ \to 2B_-\to$}\\[-3pt] \fts{$\to B_- \to A_-$}\end{tabular}&{6(9)}\\
\hline
{10}&{$\mtE_3$}&
\begin{tabular}{c}\fts{$b_1\to c_6\to a_8\to$}\\[-3pt] \fts{$\to c_9\to a_3\to a_1$}\end{tabular}
&\begin{tabular}{c}\fts{$A_+\to A_+ \to 2A_+ \to 2B_-\to$}\\[-3pt] \fts{$\to  B_-\to A_-$}\end{tabular}&{6(9)}\\
\hline
{11}&{$\mtE_2$}&
\begin{tabular}{c}\fts{$b_1\to c_6\to a_5\to$}\\[-3pt] \fts{$\to a_3\to a_1$}\end{tabular}
&\begin{tabular}{c}\fts{$A_+\to A_+ \to 2A^*\to$}\\[-3pt] \fts{$\to  B_-\to A_-$}\end{tabular}&{4(7)}\\
\hline
{12}&{$\mtE_1$}&\begin{tabular}{c}\fts{$b_1\to c_6\to a_5\to$}\\[-3pt] \fts{$\to c_4\to a_2\to a_1$}\end{tabular}&
\begin{tabular}{c}\fts{$A_+ \to A_+ \to 2A^* \to$}\\[-3pt] \fts{$\to C_2 \to A_- \to A_-$}\end{tabular}
&{5(8)}\\
\hline
{13}&{$\mtA_4$}&
\fts{$b_1\to c_6\to c_7\to a_1$}&\fts{$A_+ \to A_+\to B_-\to A_-$}&{2(3)}\myrul\\
\hline
\end{longtable}
\end{center}


\begin{figure}[!ht]
\centering
\includegraphics[width=120mm,keepaspectratio]{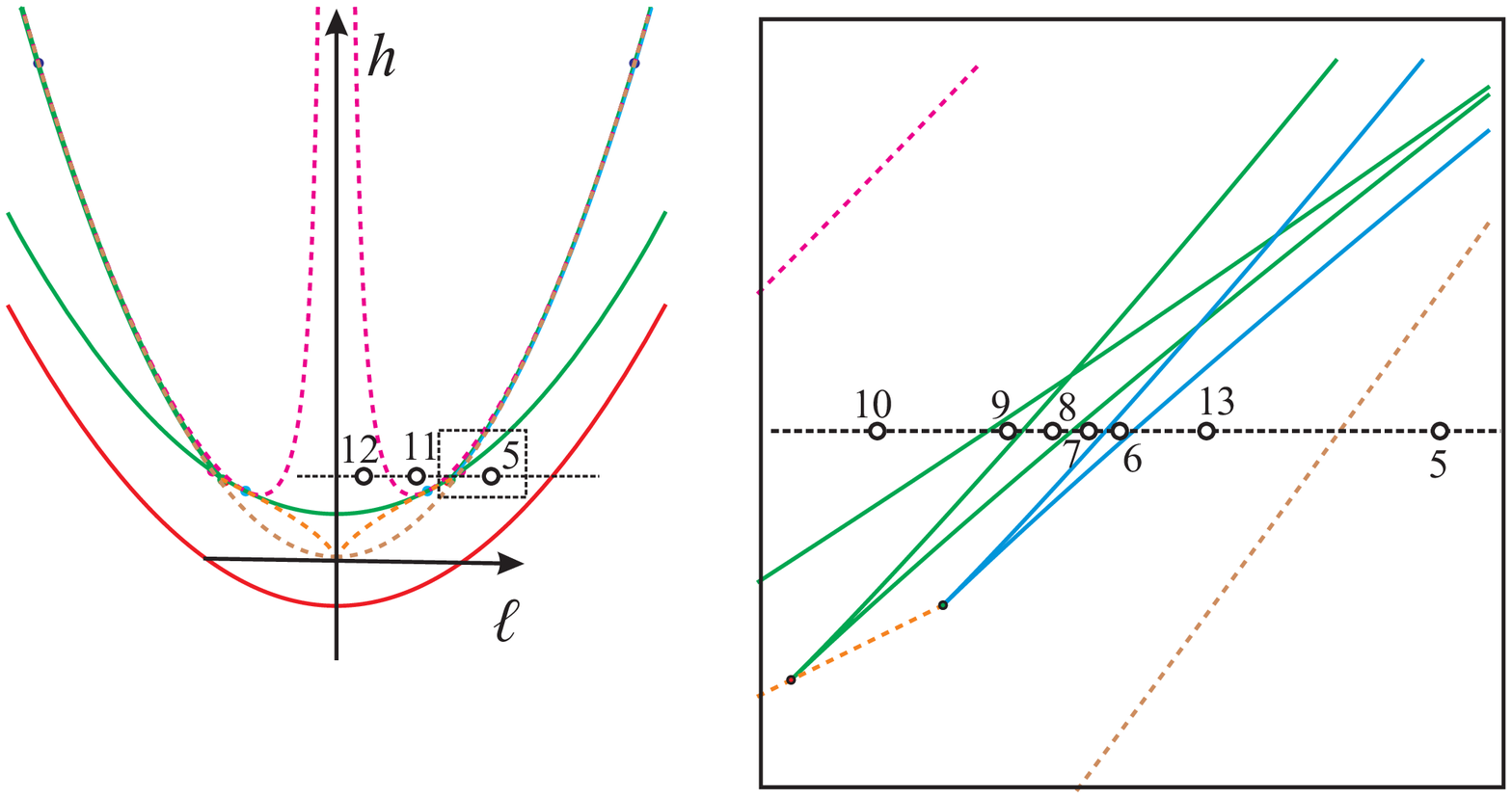}\\
\caption{Диаграмма Смейла--Фоменко для $\ld=0.1$.}\label{fig_smalefomd_1}
\end{figure}

\begin{figure}[!ht]
\centering
\includegraphics[width=120mm,keepaspectratio]{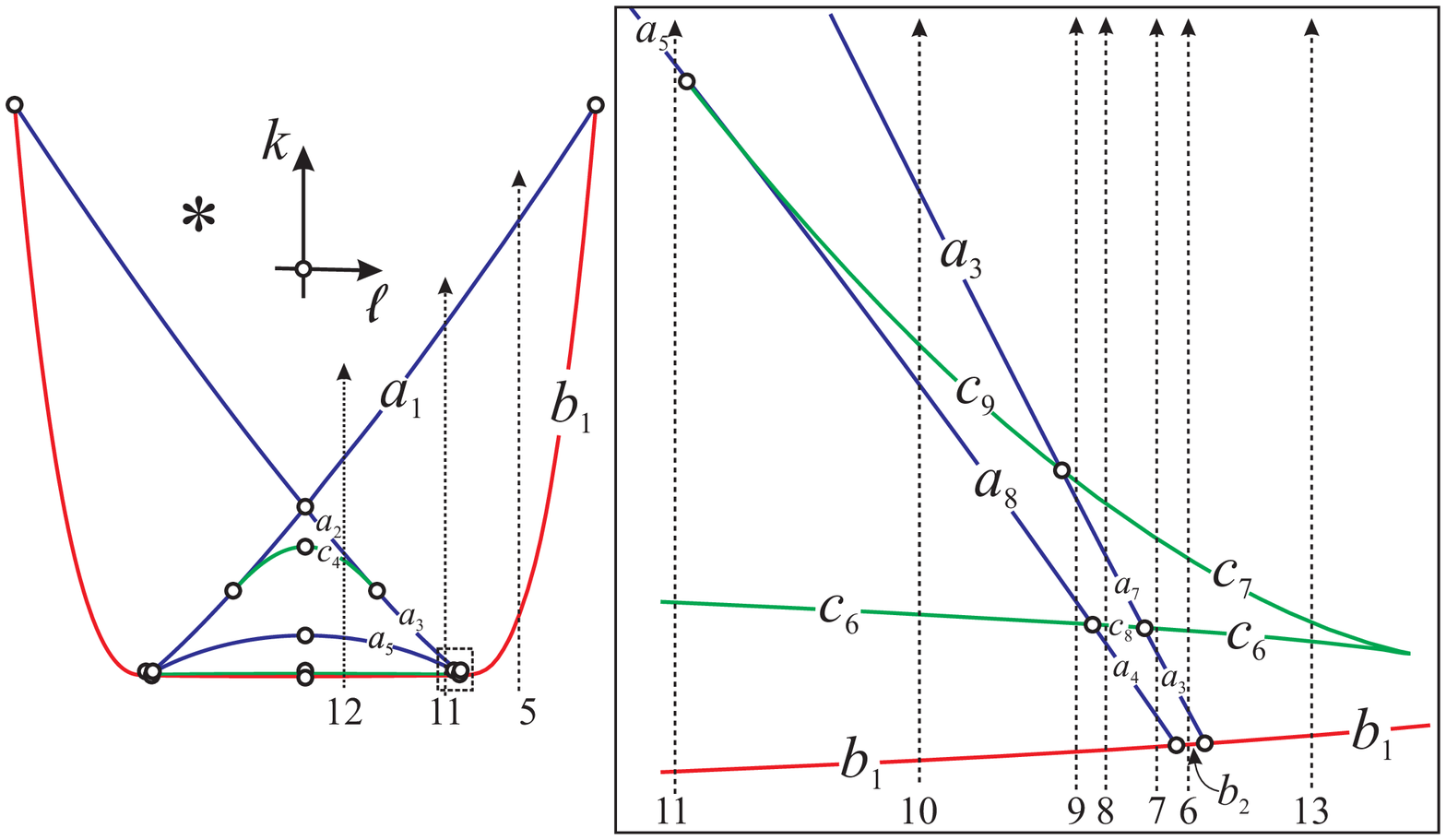}\\
\caption{Бифуркационная диаграмма в $(\ell,k)$-плоскости для $\ld=0.1,h=1.8$.}\label{fig_bifdiag_1}
\end{figure}


\begin{center}
\small

\begin{longtable}{| c| c| c|c|}
\multicolumn{4}{r}{\fts{Таблица \myt\label{table8}}}\\[2mm]
\hline
\begin{tabular}{c}\fts{Камера}\end{tabular} &\begin{tabular}{c}\fts{Последовательность}\\[-3pt]\fts{дуг}\end{tabular}
&\begin{tabular}{c}\fts{Последовательность}\\[-3pt]\fts{атомов}\end{tabular}&\begin{tabular}{c}\fts{Группа графа}\\[-3pt]\fts{(номер)}
\end{tabular} \\
\hline\endfirsthead%
\multicolumn{4}{r}{\fts{Таблица \ref{table8}}}\\[2mm]
\hline
\begin{tabular}{c}\fts{Камера}\end{tabular} &\begin{tabular}{c}\fts{Последовательность}\\[-3pt]\fts{дуг}\end{tabular}
&\begin{tabular}{c}\fts{Последовательность}\\[-3pt]\fts{атомов}\end{tabular}&\begin{tabular}{c}\fts{Группа графа}\\[-3pt]\fts{(номер)}
\end{tabular} \\
\hline\endhead%
\myrul{$\mtA_1 (f_1)$}&
\fts{$b_1\to a_1$}&\fts{$A_+ \to A_-$}&{1(1)}\\
\hline
\myrul{$\mtA_2 (f_2)$}&
\begin{tabular}{c}\fts{$b_1\to a_6\to a_2\to a_1$}\end{tabular}&
\begin{tabular}{c}\fts{$A_+ \to B_+ \to A_- \to A_-$}\end{tabular}
&{2(2)}\\
\hline
{$\mtA_3 (f_3)$}&
\begin{tabular}{c}\fts{$b_1\to c_6\to c_7\to$}\\[-3pt] \fts{$\to a_6\to a_2\to a_1$}\end{tabular}
&\begin{tabular}{c}\fts{$A_+\to A_+ \to B_-\to$}\\[-3pt] \fts{$\to  B_+\to A_-\to A_-$}\end{tabular}&{3(6)}\\
\hline
\myrul{$\mtA_4 (f_4)$}&
\fts{$b_1\to c_6\to c_7\to a_1$}&\fts{$A_+ \to A_+\to B_-\to A_-$}&{2(3)}\\
\hline
{$\mtA_5 (f_5)$}&
\begin{tabular}{c}\fts{$b_1\to c_6\to a_8\to$}\\[-3pt] \fts{$\to a_7\to c_7\to a_1$}\end{tabular}
&\begin{tabular}{c}\fts{$A_+ \to A_+ \to 2A_+ \to$}\\[-3pt] \fts{$\to 2B_-\to B_- \to A_-$}\end{tabular}&{6(9)}\\
\hline
{$\mtA_6 (f_6)$}&
\begin{tabular}{c}\fts{$b_1\to a_4\to a_3\to$}\\[-3pt] \fts{$\to c_6\to c_7\to a_1$}\end{tabular}
&\begin{tabular}{c}\fts{$A_+ \to A_+ \to B_-\to$}\\[-3pt] \fts{$\to A_+ \to B_-\to A_-$}\end{tabular}&{7(16)}\\
\hline
{$\mtA_7 (f_7)$}&
\begin{tabular}{c}\fts{$b_1\to a_4\to c_8\to$}\\[-3pt] \fts{$\to c_9\to a_3\to a_1$}\end{tabular}
&\begin{tabular}{c}\fts{$A_+ \to A_+ \to 2A_+\to$}\\[-3pt] \fts{$\to 2B_- \to B_-\to A_-$}\end{tabular}&{6(9)}\\
\hline
\myrul{$\mtA_8 (f_8)$}&
\begin{tabular}{c}\fts{$b_1\to a_4\to a_3\to a_1$}\end{tabular}
&\fts{$A_+ \to A_+ \to B_- \to A_-$}&{2(3)}\\
\hline
{$\mtA_9 (f_9)$}&
\begin{tabular}{c}\fts{$b_1\to c_2\to c_1\to$}\\[-3pt] \fts{$\to a_4\to a_3\to a_1$}\end{tabular}
&\begin{tabular}{c}\fts{$A_+\to B_+ \to A_- \to$}\\[-3pt] \fts{$ \to A_+\to  B_-\to A_-$}\end{tabular}&{8(20)}\\
\hline
{$\mtA_{10} (f_{10})$}&
\begin{tabular}{c}\fts{$b_1\to c_2\to c_1\to a_1$}\end{tabular}
&\begin{tabular}{c}\fts{$A_+\to B_+ \to A_- \to A_-$}\end{tabular}&{2(2)}\\
\hline
{$\mtA_{11} (f_{11})$}&
\begin{tabular}{c}\fts{$b_1\to c_2\to a_{11}\to$}\\[-3pt] \fts{$\to a_{10}\to c_1\to a_1$}\end{tabular}
&\begin{tabular}{c}\fts{$A_+\to B_+ \to 2B_+ \to$}\\[-3pt] \fts{$ \to 2A_-\to  A_-\to A_-$}\end{tabular}&{6(10)}\\
\hline
{$\mtA_{12} (f_{12})$}&
\begin{tabular}{c}\fts{$b_1\to a_6\to a_2\to$}\\[-3pt] \fts{$\to c_2\to c_1\to a_1$}\end{tabular}
&\begin{tabular}{c}\fts{$A_+\to B_+ \to A_- \to$}\\[-3pt] \fts{$ \to B_+\to  A_-\to A_-$}\end{tabular}&{7(17)}\\
\hline
{$\mtA_{13} (f_{13})$}&
\begin{tabular}{c}\fts{$b_1\to a_6\to c_5\to$}\\[-3pt] \fts{$\to c_3\to a_2\to a_1$}\end{tabular}
&\begin{tabular}{c}\fts{$A_+\to B_+ \to 2B_+ \to$}\\[-3pt] \fts{$ \to 2A_-\to  A_-\to A_-$}\end{tabular}&{6(11)}\\
\hline
\myrul{$\mtB_1 (f_{14})$}&
\begin{tabular}{c}\fts{$b_2\to a_3\to a_1$}\end{tabular}
&\begin{tabular}{c}\fts{$2A_+ \to B_- \to A_-$}\end{tabular}&{2(4)}\\
\hline
{$\mtB_2 (f_{15})$}&
\begin{tabular}{c}\fts{$b_2\to a_3\to c_6\to$}\\[-3pt] \fts{$\to c_7\to a_1$}\end{tabular}
&\begin{tabular}{c}\fts{$2A_+ \to B_- \to A_+\to$}\\[-3pt] \fts{$\to B_-\to A_-$}\end{tabular}&{7(18)}\\
\hline
{$\mtB_3 (f_{16})$}&
\begin{tabular}{c}\fts{$b_2\to c_8\to c_9\to$}\\[-3pt] \fts{$\to a_3\to a_1$}\end{tabular}
&\begin{tabular}{c}\fts{$2A_+ \to 2A_+ \to 2B_-\to$}\\[-3pt] \fts{$\to B_-\to A_-$}\end{tabular}&{6(15)}\\
\hline
{$\mtC (f_{17})$}&
\begin{tabular}{c}\fts{$b_1\to a_4\to c_8\to$}\\[-3pt] \fts{$\to a_7\to c_7\to a_1$}\end{tabular}
&\begin{tabular}{c}\fts{$A_+ \to A_+ \to 2A_+ \to$}\\[-3pt] \fts{$\to 2B_-\to B_- \to A_-$}\end{tabular}&{6(12)}\\
\hline
{$\mtD (f_{18})$}&
\begin{tabular}{c}\fts{$b_2\to c_8\to a_7\to$}\\[-3pt] \fts{$\to c_7\to a_1$}\end{tabular}
&\begin{tabular}{c}\fts{$2A_+ \to 2A_+ \to 2B_- \to$}\\[-3pt] \fts{$\to B_- \to A_-$}\end{tabular}&{6(13)}\\
\hline
{$\mtE_1 (f_{19})$}&
\begin{tabular}{c}\fts{$b_1\to c_6\to a_5\to$}\\[-3pt] \fts{$\to c_4\to a_2\to a_1$}\end{tabular}&
\begin{tabular}{c}\fts{$A_+ \to A_+ \to 2A^* \to$}\\[-3pt] \fts{$\to C_2 \to A_- \to A_-$}\end{tabular}&{5(8)}\\
\hline
{$\mtE_2 (f_{20})$}&
\begin{tabular}{c}\fts{$b_1\to c_6\to a_5\to$}\\[-3pt] \fts{$\to a_3\to a_1$}\end{tabular}
&\begin{tabular}{c}\fts{$A_+\to A_+ \to 2A^*\to$}\\[-3pt] \fts{$\to  B_-\to A_-$}\end{tabular}&{4(7)}\\
\hline
{$\mtE_3(f_{21})$}&
\begin{tabular}{c}\fts{$b_1\to c_6\to a_8\to$}\\[-3pt] \fts{$\to c_9\to a_3\to a_1$}\end{tabular}
&\begin{tabular}{c}\fts{$A_+\to A_+ \to 2A_+ \to$}\\[-3pt] \fts{$\to 2B_-\to  B_-\to A_-$}\end{tabular}&{6(9)}\\
\hline
{$\mtE_4 (f_{22})$}&
\begin{tabular}{c}\fts{$b_1\to c_2\to a_9\to$}\\[-3pt] \fts{$\to a_3\to a_1$}\end{tabular}
&\begin{tabular}{c}\fts{$A_+\to B_+ \to (A_+,A_-)\to$}\\[-3pt] \fts{$\to  B_-\to A_-$}\end{tabular}&{8(21)}\\
\hline
{$\mtE_5 (f_{23})$}&\begin{tabular}{c}\fts{$b_1\to c_2\to a_9\to$}\\[-3pt] \fts{$\to c_4\to a_2\to a_1$}\end{tabular}&
\begin{tabular}{c}\fts{$A_+ \to B_+ \to (A_+,A_-) \to$}\\[-3pt] \fts{$\to C_2 \to A_- \to A_-$}\end{tabular}
&{9(22)}\\
\hline
{$\mtE_6 (f_{24})$}&\begin{tabular}{c}\fts{$b_1\to c_2\to a_{11}\to$}\\[-3pt] \fts{$\to c_3\to a_2\to a_1$}\end{tabular}&
\begin{tabular}{c}\fts{$A_+ \to B_+ \to 2B_+ \to$}\\[-3pt] \fts{$\to 2A_- \to A_- \to A_-$}\end{tabular}
&{6(11)}\\
\hline
{$\mtF (f_{25})$}&\begin{tabular}{c}\fts{$b_1\to a_6\to c_5\to$}\\[-3pt] \fts{$\to a_{10}\to c_1\to a_1$}\end{tabular}&
\begin{tabular}{c}\fts{$A_+ \to B_+ \to 2B_+ \to$}\\[-3pt] \fts{$\to 2A_- \to A_- \to A_-$}\end{tabular}
&{6(10)}\\
\hline
{$\mtG (f_{26})$}&\begin{tabular}{c}\fts{$b_1\to a_6\to c_5\to$}\\[-3pt] \fts{$\to a_{10}\to a_{12}$}\end{tabular}&
\begin{tabular}{c}\fts{$A_+ \to B_+ \to 2B_+ \to$}\\[-3pt] \fts{$\to 2A_- \to A_- $}\end{tabular}
&{6(14)}\\
\hline
{$\mtH_1 (f_{27})$}&\begin{tabular}{c}\fts{$b_1\to a_6\to a_2\to$}\\[-3pt] \fts{$\to c_2\to a_{12}$}\end{tabular}&
\begin{tabular}{c}\fts{$A_+ \to B_+ \to A_- \to$}\\[-3pt] \fts{$\to B_+ \to 2A_-$}\end{tabular}
&{7(19)}\\
\hline
\myrul{$\mtH_2(f_{28})$}&
\begin{tabular}{c}\fts{$b_1\to
c_2\to a_{12}$}\end{tabular}&
\begin{tabular}{c}\fts{$A_+ \to B_+ \to 2A_- $}\end{tabular}
&{2(5)}\\
\hline
{$\mtH_3 (f_{29})$}&\begin{tabular}{c}\fts{$b_1\to c_2\to a_{11}\to$}\\[-3pt] \fts{$\to a_{10}\to a_{12}$}\end{tabular}&
\begin{tabular}{c}\fts{$A_+ \to B_+ \to 2B_+ \to$}\\[-3pt] \fts{$\to 2A_- \to A_- $}\end{tabular}
&{6(14)}\\
\hline
\end{longtable}
\end{center}

Отметим одно возникающее при этом явление, которое обычно при трактовке совпадения графов Фоменко не отмечается.
Как видно на рис.~\ref{fig_graph_fom_xn} ряд графов в некоторых группах отличается тем, что пара атомов попадает или не попадает на один и тот же критический уровень $K$. Все приведенные на этом рисунке уровни, содержащие две критические окружности обладают следующим свойством ``устойчивости'': при любом достаточно малом возмущении $(\ell,h)$ количество окружностей на таком критическом уровне не изменяется. Однако это не так, если критический уровень содержит кратные точки. Таковыми являются уровни $\ell=0, k=1+(h-\ld^2/2), h \gs \ld^2/2$, то есть все уровни, образ которых попадает на особую параболу. Такой уровень содержится в любом графе Фоменко вида $W_{0,h}$ с $h > \ld^2/2$ (напомним, что граничное значение $h = \ld^2/2$ задает изоэнергетическое многообразие с вырожденной точкой и здесь не рассматривается.
Как легко видеть, неустойчивой уровень при малом возмущении $\ell$ от нулевого значения без выхода за пределы камеры распадается на два, лежащие на нем атомы расходятся на разную высоту по $k$. Это явление имеет место во всех камерах, имеющих выход на ось $\ell=0$ при $h > \ld^2/2$. Неустойчивые (в указанном смысле) графы Фоменко представлены на рис.~\ref{fig_graph_fom_xn0}. Здесь обозначение камеры, снабженное индексом 0, означает пересечение камеры с осью $\ell=0$.

\begin{figure}[!ht]
\centering
\includegraphics[width=80mm,keepaspectratio]{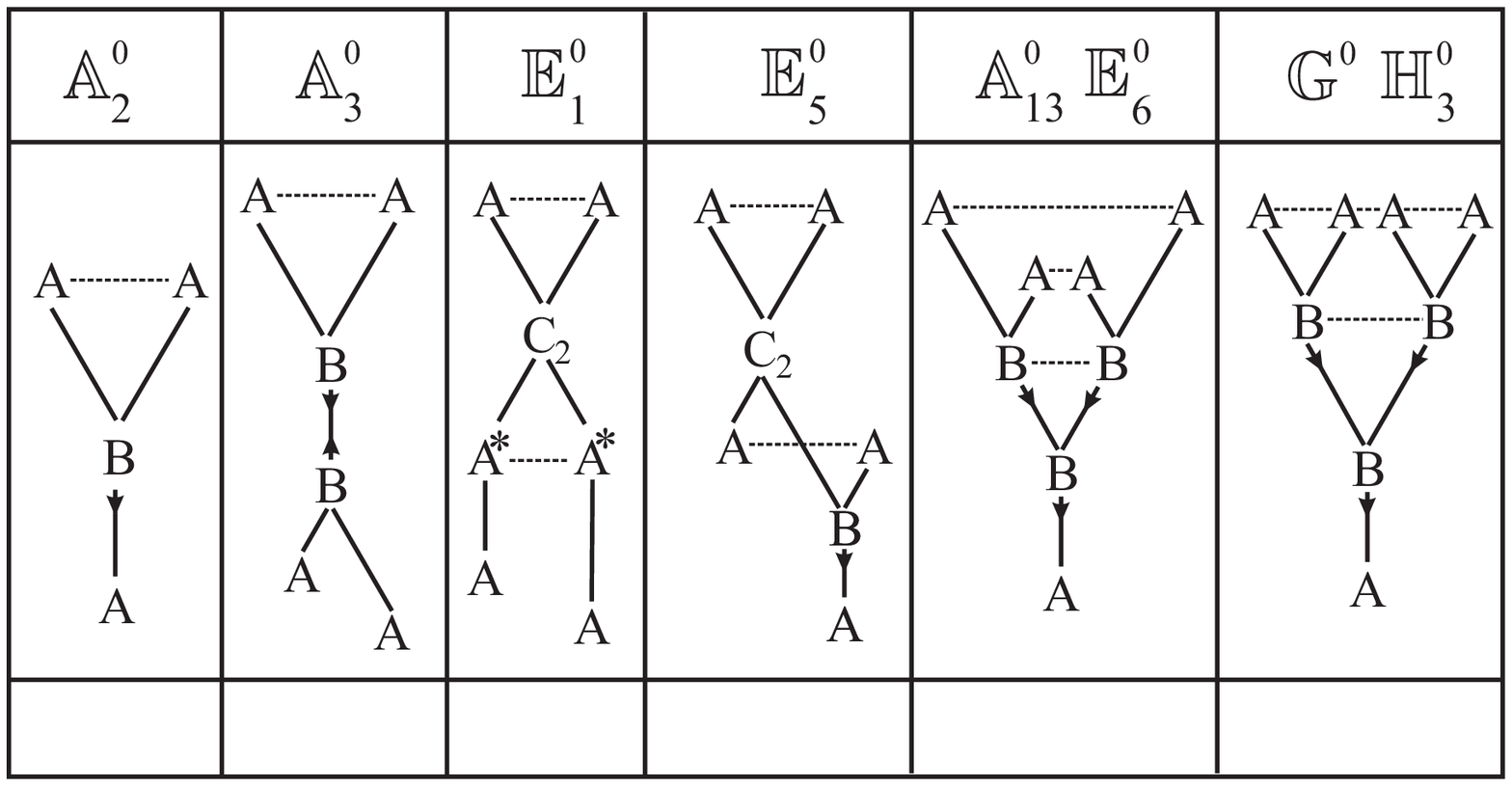}\\
\caption{Неустойчивые графы Фоменко.}\label{fig_graph_fom_xn0}
\end{figure}

\FloatBarrier

\section{Заключение}
Мы изложили имеющиеся на сегодня строго обоснованные результаты по аналитическим решениям и топологическому анализу случая Ковалевской\,--\,Яхья. Как отмечалось выше, эту работу следует дополнить результатами по предельным случаям, с которых и начинались подобные исследования, а именно, по классической задаче Ковалевской и по случаю нулевой постоянной площадей.

Мы сознательно не включали в этот текст никаких (кроме совсем необходимого минимума) фактов из общей теории интегрируемых гамильтоновых систем, поскольку тогда объем работы был бы несоразмерно увеличен. Последний раздел по топологическим инвариантам также следует расширить и включить в него строгие обоснования всех меченых топологических инвариантов. Авторы надеются в ближайшем будущем завершить работу над расширенным текстом. Она была бы полезна для дальнейших обобщений на гиростат в двойном поле \cite{ReySemRus} и гиростат с гироскопическими силами \cite{SoTs2002}. Значительный ряд результатов в этих задачах уже получен и заслуживает отдельного обсуждения.

\end{document}